\documentclass[a4paper,10pt]{article}
\pdfoutput=1 

\usepackage{jheppub} 

\usepackage[T1]{fontenc} 
\usepackage{macros_JHEP}
\usepackage{ytableau}
\Yboxdim{6pt}

\numberwithin{equation}{subsection}

\graphicspath{{./Figures/}}

\title{\boldmath Gauge origami and quiver W-algebras IV: Pandharipande--Thomas $qq$-characters}

\author[a]{Taro Kimura,}
\author[b]{Go Noshita}

\affiliation[a]{ Institut de Mathématiques de Bourgogne,\it Université Bourgogne Europe, CNRS, Dijon, France }
\affiliation[b]{Department of Physics, Institute of Science Tokyo, Tokyo, Japan }

\emailAdd{taro.kimura@ube.fr }
\emailAdd{gnoshita969hep@gmail.com}

\vfill

\abstract{We develop a contour integral formalism for computing the K-theoretic equivariant 3-vertex. Within the Jeffrey--Kirwan (JK) residue framework, we show that, by an appropriate choice of the reference vector, both the equivariant Donaldson--Thomas (DT) and Pandharipande--Thomas (PT) 3-vertices can be extracted from the same integrand. We analyze three distinct limits of the PT 3-vertex, recovering the unrefined topological vertex, the refined topological vertex, and the Macdonald refined topological vertex. Higher-rank extensions of PT counting and the DT/PT correspondence are also explored. From a quantum algebraic perspective, we construct an operator version of the equivariant PT 3-vertex and term it the Pandharipande--Thomas $qq$-character. We then discuss its connection with the quantum toroidal $\mathfrak{gl}_{1}$.}

\makeatletter
\gdef\@fpheader{\phantom{}}
\makeatother

\preprint{TIT/HEP-706}

\begin{document} 
\maketitle
\flushbottom

\section{Introduction and summary}
Gauge origami, introduced by Nekrasov \cite{Nekrasov:2015wsu,Nekrasov:2016qym,Nekrasov:2016ydq,Nekrasov:2017rqy,Nekrasov:2017gzb}, describes a new class of supersymmetric gauge theories engineered from intersecting D-branes in type II string theory. In this framework, multiple stacks of branes wrap distinct but intersecting complex surfaces inside a higher-dimensional ambient space, giving rise to coupled gauge sectors localized on each brane worldvolume. The low-energy effective theory admits an $\Omega$-background deformation, which localizes the path integral to a finite-dimensional contour integral over the moduli space of Bogomolny-Prasad-Sommerfield (BPS) configurations, which eventually gives the \textit{gauge origami partition functions}.

In particular, if one starts from type IIA string theory on $\mathbb{R}\times\mathbb{S}^{1}\times\mathbb{C}^{4}$, one can place D$(2p)$-branes wrapping $\mathbb{S}^{1}\times \mathcal{C}$, where $\mathcal{C}$ are complex $p$-cycles in $\mathbb{C}^{4}$. The twisted Witten index of the D0-branes wrapping $\mathbb{S}^{1}$ and probing the D$(2p)$-branes reduces to a finite dimensional contour integral and the so-called Jeffrey--Kirwan (JK) residue formalism \cite{Jeffrey1993LocalizationFN,Szenes2003ToricRA,Brion1999ArrangementOH} (see also \cite{Benini:2013xpa,Benini:2013nda,Hori:2014tda,Hwang:2014uwa,Cordova:2014oxa}) provides a way to evaluate this contour integral. Spiked instantons \cite{Nekrasov:2016gud,Nekrasov:2016qym,Rapcak:2018nsl,Rapcak:2020ueh}, tetrahedron instantons \cite{Pomoni:2021hkn,Pomoni:2023nlf,Fasola:2023ypx}, and the magnificent four \cite{Nekrasov:2017cih,Nekrasov:2018xsb,Billo:2021xzh,Kool:2025qou} are examples of this setup.
Moreover, a particularly rich structure emerges from toric Calabi--Yau fourfolds (CY4) replacing the $\mathbb{C}^4$ in the aforementioned setup\footnote{See also \cite{Bao:2025hfu,Bao:2024ygr,Kimura:2023bxy,Szabo:2023ixw,Bonelli:2020gku,Benini:2018hjy} for papers discussing the quiver formalism.} \cite{Cao:2014bca,Borisov:2017GT,Cao:2017swr,Cao:2019tnw,Cao:2019tvv,Oh:2020rnj,Monavari:2022rtf}. In this case, the gauge theory partition function encodes enumerative invariants of the underlying toric geometry, generalizing the well-known correspondence between instanton sums and Donaldson--Thomas (DT) invariants in the Calabi--Yau threefold (CY3) case \cite{Maulik:2003rzb,Maulik:2004txy,Okounkov:2003sp,Ooguri:2009ijd}. The CY4 background naturally organizes the computation into local building blocks associated with toric vertices, leading to an extended DT vertex formalism \cite{Monavari:2022rtf,Nekrasov:2023nai,Piazzalunga:2023qik,DelZotto:2021gzy}. This formalism generalizes the topological vertex \cite{Iqbal:2007ii,Aganagic:2003db,Awata:2008ed,Nekrasov:2014nea} familiar from the study of topological string theory on CY3, providing a combinatorial framework to assemble partition functions from local contributions at each toric vertex.

From a physical perspective, equivariant DT4 invariants count BPS bound states of D8-D6-D4-D2-D0 branes and the vertex formalism encodes their interaction structure in a way that is both computable and manifestly compatible with the toric symmetries. The traditional DT3 invariants are obtained by considering the geometry to be toric CY3 $\times $ $\mathbb{C}$ and wrapping D-branes inside the CY3. Physically, they count the D6-D2-D0 bound states of the CY3. The complete vertex formalism accompanies factors coming from the edges and faces of the associated toric diagram, but in this paper, we will be focusing only on the vertex contribution. Moreover, we will be only considering vertex contributions coming from D6-branes and D8-branes will be not discussed.

Combinatorially, the equivariant DT3 vertex is a generating function for plane partitions with Young diagrams as boundary conditions at the three-legs (see for example Fig.~\ref{fig:minimal-pp}). However, it is not just a generating function but we also have nontrivial weights. This already appears when the boundary conditions are all trivial. For example, the D6-D0 partition function with a D6-brane wrapping a $\mathbb{C}^{3}$ subspace of $\mathbb{C}^{4}$ takes the form as
\bea
M[\fq,q_{1,2,3,4}]=\sum_{\pi}\fq^{|\pi|}\mathcal{Z}^{\D6}[\pi],
\eea
where the sum is taken over possible plane partitions, and $\mathcal{Z}^{\D6}[\pi]$ is some function depending on the $\Omega$-deformation parameters $q_{1,2,3,4}$.  After taking a special limit, the nontrivial weight becomes $1$ and we simply obtain the MacMahon function, which we denote as $M[\fq]$. For the equivariant DT3 vertex, we schematically have
\bea
\mathsf{DT}_{\lambda\mu\nu}[\fq,q_{1,2,3,4}]=\sum_{\pi}\fq^{|\pi|}\mathcal{Z}^{\DT}_{\lambda\mu\nu}[\pi]
\eea
where the sum is taken over plane partitions with boundary Young diagrams $(\lambda,\mu,\nu)$ and the nontrivial weight $\mathcal{Z}_{\lambda\mu\nu}[\pi]$ depends on the parameters $q_{1,2,3,4}$. Under a suitable limit of the parameters, we have $\mathcal{Z}^{\DT}_{\lambda\mu\nu}[\pi]\rightarrow 1$, which reduces to the generating function of the possible plane partitions with nontrivial Young diagrams, which we denote as $\mathcal{Z}_{\lambda\mu\nu}[\fq]$. Such kind of counting problem is usually dubbed as DT counting.

In parallel, Pandharipande--Thomas (PT) theory \cite{Pandharipande:2007sq,Pandharipande:2007kc} provides an alternative stability condition for counting essentially the same class of BPS states. Such kind of counting problem is dubbed as PT counting, and the corresponding vertex contribution is called the equivariant PT3 vertex, taking the form as
\bea
\mathsf{PT}_{\lambda\mu\nu}[\fq,q_{1,2,3,4}]=\sum_{\pi}\fq^{|\pi|}\mathcal{Z}^{\PT}_{\lambda\mu\nu}[\pi]
\eea
where the sum is taken over all possible PT configurations in this case. Although the combinatorial rule is much more complicated and determining each weight $\mathcal{Z}^{\PT}_{\lambda\mu\nu}[\pi]$ is a nontrivial task compared to the DT count, the PT invariant shows better analytic behaviors and numerical efficiency. Similar to the DT invariant, under a special limit, the PT vertex reproduces the topological vertex \cite{Aganagic:2003db,Nekrasov:2014nea}.

In the CY3 context, DT and PT generating functions are related by a universal wall-crossing factor, reflecting the appearance or disappearance of ``free'' D0-brane contributions as one moves between stability chambers. This kind of correspondence is called the DT/PT correspondence \cite{Pandharipande:2007sq,Pandharipande:2007kc,Toda:2008ASPM,Toda:2010JAMS,Bridgeland:2011JAMS,Stoppa:2011BSMF,Toda2016HallAI,Kononov:2019fni,Jenne2020TheCP,Jenne:2021irh,Kuhn:2023koa}:
\bea
\mathcal{Z}_{\lambda\mu\nu}[\fq]=M[\fq]C_{\lambda\mu\nu}(\fq)
\eea
where $C_{\lambda\mu\nu}(\fq)$ is the topological vertex.\footnote{Strictly speaking, one needs to normalize some factors of the topological vertex to get the correspondence, but here we are only being schematic.} Furthermore, we also have
\bea
\mathsf{DT}_{\lambda\mu\nu}[\fq,q_{1,2,3,4}]=M[\fq,q_{1,2,3,4}]\mathsf{PT}_{\lambda\mu\nu}[\fq,q_{1,2,3,4}]
\eea
where the equality holds even with the nontrivial weight factors.

An interesting property is that the equivariant DT vertex can be computed by considering contour integral formulas and applying the JK-residue formalism. However, to the author's knowledge, derivation of the equivariant PT vertex from the JK-residue formalism is not known. In the first part of this paper, we will show that once the contour integrand is determined properly, both the equivariant DT and PT vertices are obtained from the same integrand by choosing appropriate poles. Recall that, within the JK-residue formalism, we need to specify an additional data called the reference vector $\eta$ to classify the poles. We will see that choosing the standard reference vector $\eta=(1,\ldots,1)$ gives the DT vertex, while choosing $\eta=(-1,\ldots,-1)$ gives the PT vertex. The dependence of this reference vector already implies a wall-crossing phenomenon. In this paper, we will not discuss the wall-crossing from the JK-residue viewpoint but rather focus on how to determine the equivariant PT vertex using the JK formalism. We also discuss the relations with the well-known topological vertices. An advantage of using this JK-residue formalism is that it gives a natural prescription to discuss higher rank generalizations. We discuss such higher rank versions of the PT vertex and examine the DT/PT correspondence in the corresponding setups.

Another interesting aspect of the gauge origami is the non-perturbative Dyson--Schwinger equations arising from the compatibility of creating and annihilating the BPS particles \cite{Nekrasov:2012xe,Nekrasov:2013xda,Nekrasov:2015wsu,Kim:2016qqs,Kanno:2012hk,Kanno:2013aha,Bourgine:2015szm,Bourgine:2016vsq,Kimura:2015rgi}. Such a property implies an alternative dual description based on an infinite-dimensional quantum algebra, which is called the BPS/CFT correspondence \cite{Alday:2009aq,Nekrasov:2015wsu,Awata:2009ur,Awata:2010yy}. The fundamental objects in this BPS/CFT correspondence are the $qq$-characters \cite{Kimura:2015rgi,Nekrasov:2016ydq,Kim:2016qqs}. Roughly speaking, they are operator lift-ups of BPS partition functions and they are physical observables characterizing the non-perturbative Dyson--Schwinger equations.

The $qq$-characters associated with the gauge origami system of $\mathbb{C}^{4}$ were systematically constructed in \cite{Kimura:2023bxy}. Moreover, the $qq$-characters of the equivariant DT vertices were also constructed in \cite{Kimura:2024osv}. From this point of view, it is natural to consider an operator lift-up of the equivariant PT vertices. In the second half of this paper, we will show that one can indeed introduce such kind of $qq$-characters and we call them the Pandharipande--Thomas (PT) $qq$-characters. We will see new properties specific to the PT theory compared with the DT theory. We study the $qq$-characters from various perspectives based on the contour integral formulas, the screening charges, and the quantum toroidal $\mathfrak{gl}_{1}$ \cite{Ginzburg,ding1997generalization,miki2007q,FFJMM1,Feigin2011,Feigin2011plane}.




\paragraph{Organization of the paper}

The organization and the summary of this paper is as follows. In section~\ref{sec:SQM-Wittenindex}, we review how to determine the Witten index of the system from the effective field theory of the D0-branes, which is a one-dimensional $\mathcal{N}=2$ supersymmetric quiver quantum mechanics. We also review how to evaluate the contour integral using the Jeffrey--Kirwan (JK) residue formalism. Depending on the reference vector, the integral will be evaluated with different poles, and such a difference will plays an essential role in understanding the DT and PT counting. The review of the gauge origami system and its partition functions are also provided. In particular, the tetrahedron instanton setup coming from D6-D0 bound states are mentioned here. 

Section~\ref{sec:DT-PT-counting} is devoted to introducing contour integral formulas to obtain the K-theoretic equivariant DT and PT vertices. 
In section~\ref{sec:DT3-JK}, we review how to engineer the contour integral formulas for the DT3 vertex or the DT3 partition function. From the quiver quantum mechanics point of view, the choice of the framing node is the nontrivial part and we will give two approaches to determine them. The framed quiver and superpotential are briefly reviewed in section~\ref{sec:DT-framedquiver} and the infinite product realization is given in section~\ref{sec:infinite-product-reg}. We then will see that by choosing the typical reference vector $\eta=(1,\ldots,1)$, the poles are indeed classified by the boxes that one can add to the minimal plane partition with nontrivial boundary conditions. In section~\ref{sec:PTrule-coord}, we review the box-counting rules of the PT3 counting for the boundary conditions with one, two, and three-legs, respectively. We review both the counting rules of the original paper \cite{Pandharipande:2007kc,Pandharipande:2007sq} and the recent ones proposed by Gaiotto--Rap\v{c}\'{a}k \cite{Gaiotto:2020dsq}. We will see that, for the case with the three-legs of the plane partition, the counting rules are involved compared to the one-leg and two-legs cases. Section~\ref{sec:PT3counting-etavector} is one of the main parts of this section. We will see that flipping the sign of the reference vector to $\eta=(-1,\ldots,-1)$, while keeping the contour integrand to be the same, the poles picked up from the JK-residue instead correspond to the PT configurations. We explicitly perform the JK-residue formalism for the one-leg in section~\ref{sec:PToneleg}, two-legs in section~\ref{sec:PTtwolegs}, and three-legs in section~\ref{sec:PTthreelegs}, with at most one box for each leg. For the one-leg and two-legs cases, during the JK-residue process, the poles are always of first order. On the other hand, for the three-legs case, \textit{second order poles} appear in the integrand. Such a property makes the counting rule much more complicated. We also give a different realization of the PT3 counting by changing the counting integrand and keeping the reference vector $\eta=(1,\ldots,1)$ in section~\ref{sec:PT3counting-conjugate}. The relations between the PT vertex and the conventional topological vertices are discussed in section~\ref{sec:top-vertex-JK}. We first study three particular limits of the D6-D0 partition function in section~\ref{sec:general-MacMahon}: unrefined limit, refined limit, and the Macdonald refined limit. Under these limits, we obtain the MacMahon function, the refined MacMahon function, and the Macdonald refined MacMahon function, respectively. In this sense, the D6-D0 partition function is a \textit{mother function} producing generalized MacMahon functions. Because of this property, it is natural to expect that the PT3 vertex obtained from the JK-residue also gives the known topological vertices after taking special limits. We will see that this is indeed the case and that under the three limits, we reproduce the unrefined topological vertex (section~\ref{sec:unrefined-vertex}), refined topological vertex (section~\ref{sec:refined-vertex}), and the Macdonald refined topological vertex (section~\ref{sec:Macdonald-vertex}). We explicitly consider the limits and check this correspondence for the cases with one box at most for each leg. The DT/PT correspondence is discussed in section~\ref{sec:DTPTcorrespondence}. We also give a short discussion on the wall-crossing phenomenon in terms of contour integrals at level one. Higher rank generalizations of DT and PT counting are discussed in section~\ref{sec:higherrankDTPT}. We will demonstrate three generalizations: multiple parallel D6-branes (section~\ref{sec:rankNDTPT}), tetrahedron instantons generalizations (section~\ref{sec:tetrahedronDTPT}), and the supergroup generalizations (section~\ref{sec:supergroupDTPT}). We propose new DT/PT correspondences and explicitly check them for low levels.

Section~\ref{sec:BPS/CFT-PT3qq} is the second main part of this paper. We initiate the study of the PT $qq$-character, which was conjectured in \cite{Kimura:2024osv}. We follow the strategy discussed in \cite{Kimura:2023bxy,Kimura:2024xpr,Kimura:2024osv}. We first introduce vertex operators associated with D-branes in section~\ref{sec:vertexop-def}. We then discuss the free field realizations of the two contour integral formulas of section~\ref{sec:PT3counting-etavector} and section~\ref{sec:PT3counting-conjugate} in section~\ref{sec:freefield-contourintegral}. The contour integral formula of section~\ref{sec:PT3counting-etavector} has the same free field realization with the DT $qq$-characters in \cite{Kimura:2024osv}, while the contour integral formula of section~\ref{sec:PT3counting-conjugate} has a different free field realization. Such free field realizations lead to operator lift ups of the PT3 vertex, which is the PT3 $qq$-characters. For the three-legs case, the existence of the second-order pole introduces derivatives of the vertex operators, which is a new feature of this $qq$-character. In section~\ref{sec:PT3qq-screeningcharge}, we discuss the algebraic properties of the PT $qq$-characters. In particular, we study the commutativity with the screening charges. Since the integrand in section~\ref{sec:PT3counting-etavector} is the same with the one for the DT counting, the conventional screening charge does not give the PT $qq$-character but instead the DT $qq$-character. Thus, we need to consider a new type of screening charge, which will be discussed in this section. In other words, keeping the same highest weight, we can produce both the DT and PT $qq$-characters by using different screening charges. For the one-leg and two-legs cases, we explicitly check the commutativity, whereas for the three-legs case, the commutativity is difficult to confirm due to the second-order poles. We instead use the fusion procedure and collision limit to obtain the three-legs PT $qq$-character. Finally, in section~\ref{sec:QTgl1}, we discuss the relation with the PT module of the quantum toroidal $\mathfrak{gl}_1$, which was initiated in \cite{Gaiotto:2020dsq}.

Supplementary materials are given in Appendix. In Appendix~\ref{app:top-vertex-symm-funct}, we summarize formulas for the topological vertices and symmetric functions. Explicit examples of the PT3 vertices are given in Appendix~\ref{app:sec-PT3vertex-examples}. The $qq$-character of $A_{1}$ theory and its quantum algebraic properties are reviewed in Appendix~\ref{app-sec:A1qqcharacter-collisionlimit}.


\section{Supersymmetric quantum mechanics and Witten index}\label{sec:SQM-Wittenindex}
In this section, we review the supersymmetric quantum mechanics (SQM) and Witten index of the gauge origami system. We review the formulas for the 1d $\mathcal{N}=2, 4$ SQM in section~\ref{sec:SQM}. The Jeffrey--Kirwan (JK) residue formalism to evaluate the contour integral formulas is reviewed in section~\ref{sec:JK-residue}. In section~\ref{sec:gaugeorigami-index}, we apply it to the gauge origami system.

\subsection{Supersymmetric quantum mechanics}\label{sec:SQM}
Let us review the formulas for the supersymmetric quantum mechanics and the Witten index we will use in this paper. The Witten index of the D$(2p)$--D0 branes system in type IIA string theory is obtained by performing the supersymmetric localization of the effective $\mathcal{N}=4$ or $\mathcal{N}=2$ supersymmetric quantum mechanics (SQM) of the D0-branes \cite{Benini:2013xpa,Benini:2013nda,Hori:2014tda,Ohta:2014ria,Cordova:2014oxa}. Let us first review the fundamental multiplets appearing in the supersymmetric quantum mechanics.
\paragraph{$\mathcal{N}=4$ SQM}
The 1d $\mathcal{N}=4$ SQM is obtained from the dimensional reduction of the 4d $\mathcal{N}=1$ gauge theory. In this paper, the gauge group will always be unitary gauge groups $U(N)$. We have two basic multiplets: vector multiplets $V$ and chiral multiplets $\Phi$. The superpotential $W$ is a holomorphic function of the chiral superfields:
\bea
\delta L=\int d^{2}\theta  W(\Phi)+\text{c.c.}
\eea
Integrating out the auxiliary fields, the $F$-term condition is
\bea
\frac{\partial W(\phi)}{\partial \phi}=0.
\eea

\paragraph{$\mathcal{N}=2$ SQM}
The 1d $\mathcal{N}=2$ SQM is obtained from the dimensional reduction of the 2d $\mathcal{N}=(0,2)$ gauge theory. We have three basic multiplets: vector multiplet $V$, chiral multiplet $\Phi$, and Fermi multiplet $\Lambda$. We additionally have superpotential terms which are usually called the $J$-term and the $E$-term. The $E$-term is a holomorphic function of the chiral superfields included in the Fermi superfield as $\overline{\mathcal{D}}\Lambda=E(\Phi)$, where $\overline{\mathcal{D}}$ is the gauge covariant superderivative. Similarly, the $J$-term is a holomorphic function of the chiral fields and the corresponding interaction term is added to the Lagrangian as
\bea
\delta L=\int d\theta \,\,\Lambda\, J(\Phi)|_{\bar{\theta}=0}+\text{c.c}.
\eea
For the system to have two supersymmetries, we need the traceless condition
\bea\label{eq:2susytraceless}
\Tr(J(\Phi)\cdot E(\Phi))=0.
\eea
The vacuum moduli space comes from
\bea
E_{\alpha}(\phi)=J_{\alpha}(\phi)=0.
\eea
which are usually called the $E,J$-term conditions.


The 1d $\mathcal{N}=4$ multiplets can be decomposed into the 1d $\mathcal{N}=2$ multiplets as follows:
\bea\label{eq:4susy2susycorrespondence}
\renewcommand{\arraystretch}{1.05}
    \begin{tabular}{|c|c|}\hline
     $\mathcal{N}=4$    & $\mathcal{N}=2$  \\\hline
      vector $V_{\mathcal{N}=4}$  &  vector $V_{\mathcal{N}=2}$ \\
        & chiral $\Sigma_{\mathcal{N}=2}$ \\ \hline
        chiral $\Phi_{\mathcal{N}=4}$  &  chiral $\Phi_{\mathcal{N}=2}$\\
         & Fermi $\Lambda_{\mathcal{N}=2}$\\
          & $E$-term $E=\Sigma_{\mathcal{N}=2}\cdot \Phi_{\mathcal{N}=2}$\\ \hline
          superpotential $W$  & $J$-term  $J=\partial W$\\\hline
     \end{tabular}
\eea

\paragraph{Quiver diagram}The quantum mechanics have a quiver theory generalization. Let $\overbar{Q}=(\overbar{Q}_{0},\overbar{Q}_{1})$ be a quiver where $\overbar{Q}_{0}$ is a set of nodes and $\overbar{Q}_{1}$ is a set of edges, the information of the gauge symmetry and matter fields are read as follows.
\begin{itemize}[topsep=0pt, partopsep=0pt, itemsep=0pt]
    \item For each node $a$, a group $\U(N_{a})$ is assigned.
    \item Edges connect the nodes and there are two sets of oriented edges colored in black and red. 
    \item A black oriented edge from a source node $a$ to a target node $b$ represents a chiral superfield $\Phi_{ba}=\Phi_{b\leftarrow a}$ transforming in the bifundamental representation $(\overline{N}_{a},N_{b})$ of $\U(N_{a})\times \U(N_{b})$.
    \item A red oriented edge connecting two nodes $a,b$ represents\footnote{Usually, the Fermi superfields are denoted as \textit{unoriented} red edges due to the equivalence of the $J$-term and $E$-term. In this paper, we assume that the $E$-term and $J$-term are determined and use the $E$-terms to determine the orientations.} a Fermi superfield $\Lambda_{ba}=\Lambda_{b\leftarrow a}$ transforming in the bifundamental representation $(\overline{N}_{a},N_{b})$ of $\U(N_{a})\times \U(N_{b})$.
\end{itemize}
We note that as in usual quiver gauge theories, circle nodes give gauge symmetries while square nodes give the flavor symmetries. Moreover, fields corresponding to self-loop arrows transform in the adjoint representation. 

We can further consider quiver gauge theories with four supersymmetries. Let $Q=(Q_{0},Q_{1})$ be a quiver, where $Q_{0}$ is a set of nodes and $Q_{1}$ is a set of edges. The matter fields are described as follows.
\begin{itemize}[topsep=0pt, partopsep=0pt, itemsep=0pt]
    \item For each node $a$, a group $\U(N_{a})$ is assigned.
    \item There is only one type of edges connecting the nodes. 
    \item An oriented edge from a source node $a$ to a target node $b$ represents a $\mathcal{N}=4$ chiral superfield $\Phi_{ba}=\Phi_{b\leftarrow a}$ transforming in the bifundamental representation $(\overline{N}_{a},N_{b})$ of $\U(N_{a})\times \U(N_{b})$.
\end{itemize}
Similar to the $\mathcal{N}=2$ quiver gauge theory, circle nodes give gauge symmetries while square nodes give flavor symmetries.
The $\mathcal{N}=4\rightarrow \mathcal{N}=2$ decomposition \eqref{eq:4susy2susycorrespondence} is then described as
\bea
\begin{tikzpicture}[decoration={markings,mark=at position \arrowHeadPosition with {\arrow{latex}}}]
 \tikzset{
        box/.style={draw, minimum width=0.6cm, minimum height=0.6cm, text centered,thick},
        ->-/.style={decoration={
        markings,mark=at position #1 with {\arrow[scale=1.5]{>}}},postaction={decorate},line width=0.5mm},
        -<-/.style={decoration={
        markings,
        mark=at position #1 with {\arrow[scale=1.5]{<}}},postaction={decorate},line width=0.5mm}    
    }
\begin{scope}{xshift=0cm}
    \draw[fill=black!20!white,thick](0,0) circle(0.3cm);
    \node[left]at (-0.3,0){$V_{\mathcal{N}=4}$};
    \node[scale=2.0] at (1.4,0){$\rightsquigarrow$};
\end{scope}
\begin{scope}[xshift=3cm]
   
    \draw[fill=black!20!white,thick](0,0) circle(0.3cm);
    \node[right] at (0.3,0.0){$V_{\mathcal{N}=2}$};
    \node[right] at (0.3,1.0){$\Sigma_{\mathcal{N}=2}$};
    \chiralarc[postaction={decorate},thick](0,0.5)(-45:225:0.3:0.8)
    
\end{scope}

\begin{scope}{}
\draw[fill=black!20!white,thick](0.3,-2) circle(0.3cm);
\draw[fill=black!20!white,thick](-1.7,-2) circle(0.3cm);
\draw[postaction={decorate}, thick](-1.4,-2)--(0,-2);
\node[above] at (-0.6,-1.8){$\Phi_{\mathcal{N}=4}$};
    \node[scale=2.0] at (1.4,-2){$\rightsquigarrow$};
\end{scope}

\begin{scope}{xshift=3cm}
\draw[fill=black!20!white,thick](4.7,-2) circle(0.3cm);
\draw[fill=black!20!white,thick](2.7,-2) circle(0.3cm);
\node[above] at (3.8,-1.8){$\Phi_{\mathcal{N}=2}$};
\node[below] at (3.8,-2.2){$\Lambda_{\mathcal{N}=2}$};
\draw[postaction={decorate}, thick](3,-1.9)--(4.4,-1.9);
\draw[postaction={decorate}, thick, red](3,-2.1)--(4.4,-2.1);
\end{scope}
\end{tikzpicture}
\eea


\paragraph{Witten index and contour integral formula}
Given the $\mathcal{N}=2$ supersymmetric quantum mechanics, the Witten index is given by the following contour integral formula
\bea\label{eq:Wittenindex-integral}
\mathcal{Z}(y)=\frac{1}{|W_{G}|}\oint _{\text{JK}}Z_{\text{V}}\prod_{i}Z_{\Phi_{i}}\prod_{\alpha}Z_{\Lambda_{\alpha}}
\eea
where
\bea\label{eq:2susylocalization}
Z_{\text{V}}&=\prod_{I}\frac{d\phi_{I}}{2\pi i}\prod_{\alpha\in G}\sh\left( \alpha\cdot \phi\right),\quad Z_{\Phi_{i}}=\prod_{\rho\in V_{\text{chiral}}}\frac{1}{\sh(\rho\cdot \mu)},\quad Z_{\Lambda_{\alpha}}=\prod_{\rho\in V_{\text{Fermi}}}\sh(\rho\cdot \mu),
\eea
$G$ is the set of roots associated with the gauge group, $V_{\text{chiral}}, V_{\text{Fermi}}$ are the sets of the corresponding weights and $|W_{G}|$ is the order of the Weyl group of $G$. We also denoted
\bea
\sh(x)=2\sinh\left(\frac{x}{2}\right)=e^{x/2}-e^{-x/2}=[e^{x}]
\eea
and for later use, we also introduce 
\bea\label{eq:cosh-def}
\ch(x)=2\cosh\left(\frac{x}{2}\right)=e^{x/2}+e^{-x/2}.
\eea
The parameter $\phi$ parametrizes the Cartan subalgebra of $G$ and $\mu$ contains $\phi$ and the fugacities of the flavor nodes. Note that if the system has $\mathcal{N}=4$ supersymmetries, we can rewrite the index in a way that the number of numerators and the denominators are the same. 

In particular, we are interested in the Witten index of 1d $\mathcal{N}=2$ quiver gauge theories with unitary groups. Let $\overline{Q}=(\overline{Q}_{0},\overline{Q}_{1})$ be the 2 SUSY quiver diagram and $\overline{Q}_{0},\overline{Q}_{1}$ denoting the set of nodes and edges. For use in the gauge origami setup, there is a flavor $\U(1)^{3}$ symmetry and we turn on the fugacities $q_{i}=e^{\epsilon_{i}}$. We denote the $\U(1)^{3}$ charges of the chiral, Fermi superfields as $q(\Phi)=e^{\epsilon(\Phi)}$, $q(\Lambda)=e^{\epsilon(\Lambda)}$. From the D$(2p)$-branes viewpoint, these parameters play the roles of the $\Omega$-deformation parameters.


The basic factors appearing in the Witten index of an $\mathcal{N}=2$ quiver quantum mechanics are then given as follows.
\begin{itemize}[topsep=0pt, partopsep=0pt, itemsep=0pt]
    \item For each gauge group (circle node) of the theory, we have 
    \bea
\adjustbox{valign=c}{\begin{tikzpicture}[decoration={markings,mark=at position \arrowHeadPosition with {\arrow{latex}}}]
 \tikzset{
        box/.style={draw, minimum width=0.6cm, minimum height=0.6cm, text centered,thick},
        ->-/.style={decoration={
        markings,mark=at position #1 with {\arrow[scale=1.5]{>}}},postaction={decorate},line width=0.5mm},
        -<-/.style={decoration={
        markings,
        mark=at position #1 with {\arrow[scale=1.5]{<}}},postaction={decorate},line width=0.5mm}    
    }
\begin{scope}{xshift=0cm}
    \draw[fill=black!10!white,thick](0,0) circle(0.4cm);
    \node at (0,0){$N_{a}$};
\end{scope}
\end{tikzpicture}}\quad {\rightsquigarrow} \quad \prod_{I=1}^{N_{a}}\frac{d\phi_{I}^{(a)}}{2\pi i } \prod_{I\neq J}\sh(\phi_{I}^{(a)}-\phi_{J}^{(a)}).
    \eea
\item The chiral superfield corresponding to an arrow from $\U(N_{a})$ to $\U(N_{b})$ gives the contribution
    \bea
   \adjustbox{valign=c}{ \begin{tikzpicture}[decoration={markings,mark=at position \arrowHeadPosition with {\arrow{latex}}}]
 \tikzset{
        box/.style={draw, minimum width=0.6cm, minimum height=0.6cm, text centered,thick},
        ->-/.style={decoration={
        markings,mark=at position #1 with {\arrow[scale=1.5]{>}}},postaction={decorate},line width=0.5mm},
        -<-/.style={decoration={
        markings,
        mark=at position #1 with {\arrow[scale=1.5]{<}}},postaction={decorate},line width=0.5mm}    
    }

\begin{scope}{}
\draw[fill=black!10!white,thick](4.7,-2) circle(0.4cm);
\node at (4.7,-2){$N_{b}$};
\draw[fill=black!10!white,thick](2.1,-2) circle(0.4cm);
\node at (2.1,-2){$N_{a}$};
\draw[postaction={decorate}, thick](2.5,-2)--(4.3,-2);
\node[above] at (3.4,-2){$q(\Phi_{a\rightarrow b})$};
\end{scope}
\end{tikzpicture}}\quad \rightsquigarrow\quad \prod_{I=1}^{N_{a}}\prod_{J=1}^{N_{b}}\frac{1}{\sh(\phi_{J}^{(b)}-\phi_{I}^{(a)}-\epsilon(\Phi_{a\rightarrow b}))}.
    \eea
    For chiral superfields in the adjoint representation, we set $a=b$. For square nodes, we simply replace $\{\phi_{I}^{(a)}\}_{I=1}^{N_{a}}$ with the flavor fugacities of the associated flavor symmetry, which we denote $\{\mathfrak{a}_{\alpha}^{(a)}\}_{\alpha=1}^{N_{a}}$, and do not take the integral over them. 
    \item The Fermi superfield corresponding to an arrow from $\U(N_{a})$ to $\U(N_{b})$ gives the contribution 
    \bea
   \adjustbox{valign=c}{ \begin{tikzpicture}[decoration={markings,mark=at position \arrowHeadPosition with {\arrow{latex}}}]
 \tikzset{
        box/.style={draw, minimum width=0.6cm, minimum height=0.6cm, text centered,thick},
        ->-/.style={decoration={
        markings,mark=at position #1 with {\arrow[scale=1.5]{>}}},postaction={decorate},line width=0.5mm},
        -<-/.style={decoration={
        markings,
        mark=at position #1 with {\arrow[scale=1.5]{<}}},postaction={decorate},line width=0.5mm}    
    }

\begin{scope}{}
\draw[fill=black!10!white,thick](4.7,-2) circle(0.4cm);
\node at (4.7,-2){$N_{b}$};
\draw[fill=black!10!white,thick](2.1,-2) circle(0.4cm);
\node at (2.1,-2){$N_{a}$};
\draw[red, postaction={decorate}, thick](2.5,-2)--(4.3,-2);
\node[above] at (3.4,-2){\textcolor{red}{$q(\Lambda_{a\rightarrow b})$}};
\end{scope}
\end{tikzpicture}}\quad \rightsquigarrow\quad \prod_{I=1}^{N_{a}}\prod_{J=1}^{N_{b}}\sh(\phi_{J}^{(b)}-\phi_{I}^{(a)}-\epsilon(\Lambda_{a\rightarrow b})).
    \eea
    Similarly, we set $a=b$ for Fermi superfields in the adjoint representation, and for square nodes we simply replace $\{\phi_{I}^{(a)}\}_{I=1}^{N_{a}}$ with $\{\mathfrak{a}_{\alpha}^{(a)}\}_{\alpha=1}^{N_{a}}$ and do not perform the contour integral. We note that since the $J$-term and the $E$-term are conjugate with each other, the JK-residue contribution differs only up to extra $-1$ factors. However, such extra sign factors can be included in the redefinition of topological terms in instanton computations and thus they are non-essential.
\end{itemize}


\subsection{Jeffrey--Kirwan (JK) residue formalism}\label{sec:JK-residue}
Given the Witten index \eqref{eq:Wittenindex-integral}, we need to evaluate the contour integral. The correct way to evaluate it is known to be the Jeffrey--Kirwan residue formalism. In this section, we review how to perform it. See \cite{Jeffrey1993LocalizationFN,Szenes2003ToricRA,Brion1999ArrangementOH,Benini:2013xpa,Benini:2013nda,Hori:2014tda,Hwang:2014uwa} for references. We basically follow the Appendices of \cite{Nawata:2023aoq,Kim:2024vci}.

Let us consider the following contour integral
\bea
\mathcal{Z}=\oint_{\JK}\prod_{i=1}^{k}\frac{d\phi_i}{2\pi i }\mathcal{Z}(\phi)
\eea
where $\phi_i$ takes values in the Cartan subalgebra $\mathfrak{h}$ of the gauge group. The integrand takes form of ratios in hyperbolic sine functions and the poles come from the zeros of them in the denominator given by the hyperplanes in $\mathfrak{h}$ characterized by the charge vector $Q_{i}\in\mathfrak{h}^{\ast}$:
\bea
H_i : \ Q_{i}(\phi)+f_{i}(\fra,\eps_{a})=0,\quad i=1,\ldots n
\eea
where $f_{i}$ is some function dependent on the parameters of the theory. 
The integrand is singular at $\mathcal{M}_{\text{sing}}=\cup_i H_{i}$ and we denote $\mathcal{M}^{\ast}_{\text{sing}}$ to be the set of isolated points where $n\geq k$ linearly independent singular hyperplanes meet. When the number of hyperplanes meeting a point $\phi_{\ast}\in\mathcal{M}^{\ast}_{\text{sing}}$ is exactly $k$, the intersection is called \textit{non-degenerate}, while when $n\geq k$ hyperplanes meet it, the intersection is called \textit{degenerate}.

The JK-prescription picks up poles in $\mathcal{M}^{\ast}_{\text{sing}}$ with an additional ingredient $\eta\in\mathfrak{h}^{\ast}$, which is called the reference vector. To explicitly show the reference vector dependence, we denote the contour integral as
\bea
\oint_{\eta}\prod_{i=1}^{k}\frac{d\phi_i}{2\pi i }\mathcal{Z}(\phi).
\eea
If $\phi_{\ast}$ is a non-degenerate pole associated with the charge vectors $Q_{1},\ldots,Q_{k}$, the JK-residue is defined as the following residue:
\bea
\underset{\phi=\phi_{\ast}}{\text{JK-}\text{Res}}(Q(\phi_{\ast}),\eta)\mathcal{Z}(\phi)=\delta(Q,\eta)\frac{1}{|\det Q|}\underset{\delta_{k}=0}{\Res}\cdots \underset{\delta_{1}=0}{\Res}\left.\mathcal{Z}(\phi)\right|_{Q_{i}(\phi)+f_{i}(\fra,\eps_{a})=\delta_{i}}
\eea
where
\bea
\delta(Q,\eta)=\begin{dcases}
    1,\,\,\eta\in \text{Cone}(Q_{1},\ldots,Q_{k})\\
    0,\,\, \text{otherwise}
\end{dcases},\quad 
\text{Cone}(Q_{1},\ldots,Q_{k})=\left\{\sum_{i=1}^{k}\lambda_{i}Q_{i}=\eta\mid \lambda_{i}>0\right\}.
\eea

For a degenerate pole, we can associate a set of charge vectors as $Q_{\ast}=\{Q_{1},\ldots,Q_{n}\}$ with $n>k$. We then choose a sequence of $k$ linearly independent charge vectors $Q_{j_{1}},\ldots Q_{j_{k}}$ from $Q_{\ast}$ and construct a flag $F$:
\bea
\{0\}\subset F_{1}\subset \cdots \subset F_{k}=\mathbb{R}^{k},\quad F_{i}=\text{span}\{Q_{j_{1}},\ldots,Q_{j_{k}}\}
\eea
When different sequences give the same flag, we only consider one of them. The sequence $Q_{j_{1}},\ldots,Q_{j_{k}}$ is called the basis $\mathcal{B}(F,Q_{\ast})$ of the flag. The basis is not unique generally, but for a given flag we choose one of them. From the flag $F$ and its basis $\mathcal{B}(F,Q_{\ast})$, we construct a sequence of vectors defined as
\bea
\kappa(F,Q_{\ast})=(\kappa_1,\ldots,\kappa_{k}),\quad \kappa_{i}=\sum_{\substack{Q\in Q_{\ast}\\ Q\in F_{i}}} Q.
\eea
The JK residue is then given as
\bea\label{eq:JKresidue-iterative}
\underset{\phi=\phi_{\ast}}{\text{JK-}\text{Res}}(Q(\phi_{\ast}),\eta)\mathcal{Z}(\phi)=\sum_{F}\delta(F,\eta)\frac{\text{sign}\det(\kappa(F,Q_{\ast}))}{\det \mathcal{B}(F,Q_{\ast})}\underset{\delta_{k}=0}{\Res}\cdots\underset{\delta_{1}=0}{\Res}\left.\mathcal{Z}(\phi)\right|_{Q_{i}(\phi)+f_{i}(\fra,\eps_a)=\delta_{i}}
\eea
where $\delta(F,\eta)$ is $1$ when the reference vector $\eta$ is included in the closed cone spanned by $\kappa(F,Q_{\ast})$ and $0$ otherwise and the residue is performed in the order $\delta_{1},\ldots,\delta_{k}$. Note also that the sum is taken over all possible flags. One can also confirm that the iterative residue given here reduces back to the definition for the non-degenerate poles discussed above. 

Therefore, for a generic reference vector, the contour integral is evaluated as
\bea
\oint _{\eta}\prod_{i=1}^{k}\frac{d\phi_i}{2\pi i}\mathcal{Z}(\phi)=\sum_{\phi_{\ast}\in \mathcal{M}_{\text{sing}}^{\ast}}\underset{\phi=\phi_{\ast}}{\text{JK-Res}}(Q(\phi_{\ast}),\eta)\mathcal{Z}(\phi).
\eea
In this paper, we will mainly use two choices of reference vectors, which we define as
\bea
\eta_{0}=(1,\ldots,1),\quad \tilde{\eta}_{0}=(-1,\ldots,-1).
\eea

Let us explicitly perform the JK-residue formalism for some examples.
\paragraph{Example 1}
We have
\bea
\mathcal{Z}=\oint_{\eta_{0}} \frac{d\phi}{2\pi i}\frac{1}{\sh(\phi-\fra)}=\underset{\phi=\fra}{\Res}\frac{1}{\sh(\phi-\fra)}=+1
\eea
where the middle term is evaluated by the residue.
\paragraph{Example 2}
Similarly, for 
\bea\label{eq:JKresidue-sign}
\mathcal{Z}=\oint_{\tilde{\eta}_{0}} \frac{d\phi}{2\pi i}\frac{1}{\sh(-\phi+\fra)}&=\underset{\delta=0}{\Res}\frac{1}{\sh(\delta)}=-\underset{\phi=\fra}{\Res}\frac{1}{\sh(-\phi+\fra)}=+1
\eea
Note that when performing the residue using the variable $\phi$, an extra sign appears.

\paragraph{Example 3}
In this paper, we also need to deal with cases when higher order poles appear:
\bea
\mathcal{Z}=\oint_{\eta_{0}} \frac{d\phi}{2\pi i}\frac{\sh(\phi-\mathfrak{b})}{\sh(\phi-\fra)^{2}}=\,\underset{\phi=\fra}{\Res}\frac{\sh(\phi-\mathfrak{b})}{\sh(\phi-\fra)^{2}}=\left.\frac{\partial}{\partial \phi}(\sh(\phi-\mathfrak{b}))\right|_{\phi=\fra}=\frac{1}{2}\ch(\mathfrak{b}-\fra)
\eea

\paragraph{Example 4}Let us also introduce an example with degenerate poles: 
\bea
\mathcal{Z}=\oint_{\eta_{0}} \frac{d\phi_{1}}{2\pi i} \frac{d\phi_{2}}{2\pi i}\frac{f(\phi_{1},\phi_{2})}{\sh(\phi_{1}-\fra)\sh(\phi_{2}-\fra-\eps)\sh(-\phi_2+\phi_1+\eps)},
\eea
where $\eps$ is some parameter. We are also assuming that $f(\phi_1,\phi_2)$ further contains no poles nor zeros at $(\phi_{1},\phi_{2})=(\fra,\fra+\eps)$. Both of the charge vectors give the pole $(\phi_{1},\phi_{2})=(\fra,\fra+\eps)$. The charge vectors are
\bea
Q_{\ast}=\left\{\vectxy{1}{0},\vectxy{0}{1},\vectxy{1}{-1}\right\}
\eea
and it is a \textit{degenerate pole}. We can choose three sets of linear independent charge vectors
\bea
\{(1,0),(0,1)\},\quad \{(1,0),(1,-1)\},\quad \{(0,1),(1,-1)\}.
\eea

For the first set of charge vectors, we can construct two flags
\bea
\{0\}\subset \text{span}\{(1,0)\} \subset \text{span}\{(1,0),(0,1)\}= \mathbb{R}^{2},\\
\{0\}\subset\text{span}\{(0,1)\} \subset \text{span}\{(1,0),(0,1)\}= \mathbb{R}^{2},
\eea
where the basis is the ordered sequence $\{(1,0),(0,1)\}$ and $\{(0,1),(1,0)\}$, respectively. The corresponding $\kappa$ vectors for each basis is
\bea
\kappa_{F}=\left(\vectxy{1}{0},\vectxy{2}{0}\right),\quad \kappa_F=\left(\vectxy{0}{1},\vectxy{2}{0}\right).
\eea
Note that since the $\kappa$ vectors are the sum of the charge vectors contained in the span of the flag, the extra charge vector not used in the construction of the flag also appears:
\bea
\vectxy{1}{0}+\vectxy{0}{1}+\vectxy{1}{-1}=\vectxy{2}{0}.
\eea

For the second choice of linear independent charge vectors, we have
\bea
\kappa_{F}=\left(\vectxy{1}{0},\vectxy{2}{0}\right),\quad \kappa_{F}=\left(\vectxy{1}{-1},\vectxy{2}{0}\right).
\eea

For the third choice, we have
\bea
\kappa_{F}=\left(\vectxy{0}{1},\vectxy{2}{0}\right),\quad \kappa_{F}=\left(\vectxy{1}{-1},\vectxy{2}{0}\right).
\eea

Only when $\kappa_{F}=((0,1),(2,0))$, the reference vector is contained in the cone. Thus, the allowed flags are
\bea
\text{span}\{(0,1)\} \subset \text{span}\{(1,0),(0,1)\}=\text{span}\{(1,-1),(0,1)\}.
\eea
Since that if there are identical representatives for a flag, we simply choose one of them, the residue is performed as 
\bea
\underset{\phi_{1}=\fra}{\Res}\underset{\phi_{2}=\fra+\eps}{\Res}\frac{f(\phi_{1},\phi_{2})}{\sh(\phi_{1}-\fra)\sh(\phi_{2}-\fra-\eps)\sh(-\phi_2+\phi_1+\eps)}.
\eea
Namely, to take the residue we always have to take the residue at the pole $\phi_{2}=\fra+\eps$ first and then take the residue at the pole $\phi_{1}=\fra$. The opposite way to take the iterative residue is forbidden from the above construction. In other words, the flag construction give a natural ordering in how the poles are evaluated.

\paragraph{Example 5}Assume that the integrand $\mathcal{Z}(\phi)$ have poles coming from the hyperplanes whose corresponding charge vectors are $\{\pm e_{I}\}_{I=1}^{k},$ and $\{e_{I}-e_{J}\}_{I,J=1}^{k}$, where $e_{i}$ is the unit vector pointing the $i$-th direction in $\mathbb{R}^{k}$.

Let us choose the reference vector to be $\eta=\eta_{0}=(1,\ldots,1)$ and classify the possible choices of the charge vectors such that $\eta\in \text{Cone}(Q_{1},\ldots,Q_{k})$. Let us focus on the case when the pole is non-degenerate. Generalization to the case when we have degenerate poles can be done similarly but we omit it. The charge vectors then obey the following properties. See \cite[section 3.1]{Hwang:2014uwa} for a derivation.
\begin{itemize}
    \item Elements of $\{-e_{i}\}$ are not contained in the charge vectors.
    \item The charge vectors form an oriented tree structure, which is determined as follows. 
    \begin{itemize}
    \item We have $k$ vertices labeled by $\{e_{I}\}_{I=1}^{k}$ and arrows connecting them. For each $\{e_{I}\}_{I=1}^{k}$, there is only one vertex.
    \item The root vertex labeled by $e_{J}$ correspond to the charge vector $+e_{J}$. 
    \item The tree grows by adding vertices $e_{K}$ to $e_{J}$ with an arrow $J\rightarrow K$. For such kind of arrow, we associate the charge vector $e_{K}-e_{J}$. 
    
    \end{itemize}
\end{itemize}

Similarly, if we choose the reference vector to be $\tilde{\eta}_0$, the charge vectors associated with the vertices and the arrows of the tree are the ones whose signs are flipped compared to the $\eta_{0}$ case. Namely, for a root vertex labeled by $e_{J}$, the charge vector is $-e_{J}$, and for an arrow $J\rightarrow K$, the charge vector is $e_{J}-e_{K}$.

The iterative residue is performed by taking the residue for each hyperplane along the oriented tree. Let the reference vector to be $\eta=\eta_0$ and an example for the oriented tree and the corresponding charge vectors is 
\bea
\adjustbox{valign=c}{\includegraphics[width=4cm]{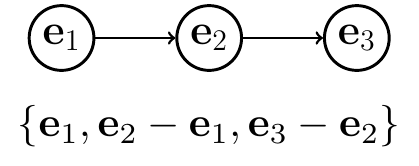}}
\eea
For this case, the iterative residue is taken in the order $\underset{\phi_{3}}{\Res}\,\underset{\phi_{2}}{\Res}\,\underset{\phi_{1}}{\Res}$. 

This can be also seen using the flag construction. We have six possible ways to order the charge vectors and define the corresponding $\kappa$ matrix. Only when the charge vectors are ordered in the way the oriented tree is formed, the reference vector is contained in the cone obtained from the $\kappa$-matrix. For example, for the order $\{e_{1},e_{2}-e_{1},e_{3}-e_{2}\}$, we have
\bea
\kappa_{F}=\left(\vectxyz{1}{0}{0},\vectxyz{0}{1}{0},\vectxyz{0}{0}{1}\right)
\eea
while for the order $\{e_{2}-e_{1},e_{1},e_{3}-e_{2}\}$, we have
\bea
\kappa_F=\left(\vectxyz{-1}{1}{0},\vectxyz{0}{1}{0},\vectxyz{0}{0}{1}\right)
\eea
and only the first one includes the $\eta_0$.

Generalization to cases when we have degenerate poles can be done similarly but we omit it.

\subsection{Gauge Origami and instanton partition functions}\label{sec:gaugeorigami-index}
\paragraph{Notations}
Let us review the gauge origami setup and the notations we use in this paper. The ten-dimensional space-time is $\mathbb{C}^{4}\times\mathbb{R}^{1,1}$. We denote the real coordinates as $\{x^{i}\}_{i=0}^{9}$ and the four complex coordinates as $\{z_{a}\}_{a=1}^{4}$ with $z_{a}=x^{2a-1}+ix^{2a}$. There are four complex one-planes and three-planes which we write as $\mathbb{C}_{a},\,\, \mathbb{C}_{\bar{a}}^{3}$ for $a\in\four,\bar{a}\in\four^{\vee}$, where
\bea
 \four=\{1,2,3,4\},\quad \four^{\vee}=\{123,124,134,234\}.
\eea
In later sections, we frequently use the quadrality symmetry between the four variables $1\leftrightarrow 2 \leftrightarrow 3 \leftrightarrow 4$. The complement $\bar{A}$ of $A\in\six=\{12,13,14,23,24,34\}$ is denoted for example as $A=12,\,\, \bar{A}=34$. Note that we have $\four\simeq \four^{\vee}$ under the map $a\in\four\leftrightarrow\bar{a}\in\four^{\vee} $. 

\paragraph{Tetrahedron instanton}
The tetrahedron instanton setup \cite{Pomoni:2021hkn,Pomoni:2023nlf,Fasola:2023ypx} is a gauge origami setup where the D6-branes wrap $\mathbb{C}^{3}_{\bar{a}}$ for $a\in\four$. The brane configuration is summarized as
\bea\label{eq:2Btetrahedroninstanton}
\renewcommand{\arraystretch}{1.05}
\begin{tabular}{|c|c|c|c|c|c|c|c|c|c|c|}
\hline
& \multicolumn{2}{c|}{$\mathbb{C}_{1}$} & \multicolumn{2}{c|}{$\mathbb{C}_{2}$} & \multicolumn{2}{c|}{$\mathbb{C}_{3}$} & \multicolumn{2}{c|}{$\mathbb{C}_{4}$} & \multicolumn{2}{c|}{$\mathbb{R}^{1,1}$} \\
\cline{2-11}  & 1 & 2 & 3 & 4& 5 & 6 & 7 & 8 & 9& 0\\
\hline D0& $\bullet$ & $\bullet$  & $\bullet$  & $\bullet$  & $\bullet$  & $\bullet$   & $\bullet$  & $\bullet$  & $\bullet$   & $-$\\
\hline
$\D6_{123} $& $-$ & $-$ & $-$ & $-$ & $-$ & $-$ & $\bullet$ & $\bullet$ & $\bullet$ & $-$ \\
\hline
$\D6_{124} $& $-$ & $-$& $-$ & $-$  & $\bullet$ & $\bullet$ & $-$ & $-$ & $\bullet$ & $-$ \\
\hline
$\D6_{134} $& $-$ & $-$  & $\bullet$ & $\bullet$ & $-$ & $-$& $-$ & $-$ & $\bullet$ & $-$ \\
\hline $\D6_{234}$ & $\bullet$ & $\bullet$ & $-$ & $-$ & $-$ & $-$ & $-$ & $-$ & $\bullet$ & $-$ \\
\hline
\end{tabular}
\eea

With suitable $B$-field to preserve supersymmetry \cite{Witten:2000mf}, the low energy theory on the D0-brane is described by a 1d $\mathcal{N}=2$ SQM and the instanton partition function is the Witten index of the D0-branes. Let us focus on the case when there are only $n_{\bar{4}}$ D6$_{\bar{4}}$-branes. Note that this system preserves $\mathcal{N}=4$ supersymmetries. The corresponding quiver diagram is
\bea\label{eq:4SUSYD7setup}
\adjustbox{valign=c}{
\begin{tikzpicture}[decoration={markings,mark=at position \arrowHeadPosition with {\arrow{latex}}}]
 \tikzset{
        box/.style={draw, minimum width=0.7cm, minimum height=0.7cm, text centered,thick},
        ->-/.style={decoration={
        markings,mark=at position #1 with {\arrow[scale=1.5]{>}}},postaction={decorate},line width=0.5mm},
        -<-/.style={decoration={
        markings,
        mark=at position #1 with {\arrow[scale=1.5]{<}}},postaction={decorate},line width=0.5mm}    
    }
\begin{scope}[xshift=4cm]
    \node[box,fill=black!10!white] at (0,1.6) {$n_{\bar{4}}$};
    \draw[postaction={decorate},thick] (0,1.25)--(0,0.4);
    \foreach \ang in {125,180,235} {
    \begin{scope}[rotate=\ang]
        \chiralarc[postaction={decorate},thick](0,0.5)(-45:225:0.22:0.65)
    \end{scope}
    }
    \node[right] at (0,0.8) {$\mathsf{I}$};
    \node[] at (-1.5,-0.9) {$\mathsf{B}_{1}$};
    \node[below] at (0,-1.6) {$\mathsf{B}_{2}$};
    \node[ ] at (1.5,-0.9){$\mathsf{B}_{3}$};
    \draw[fill=black!10!white,thick](0,0) circle(0.4cm);
    \node at (0,0) {$k$};
\end{scope}
\end{tikzpicture}}\qquad \adjustbox{valign=c}{
\begin{tikzpicture}[decoration={markings,mark=at position \arrowHeadPosition with {\arrow{latex}}}]
 \tikzset{
        box/.style={draw, minimum width=0.7cm, minimum height=0.7cm, text centered,thick},
        ->-/.style={decoration={
        markings,mark=at position #1 with {\arrow[scale=1.5]{>}}},postaction={decorate},line width=0.5mm},
        -<-/.style={decoration={
        markings,
        mark=at position #1 with {\arrow[scale=1.5]{<}}},postaction={decorate},line width=0.5mm}    
    }
\begin{scope}[xshift=4cm]
    \draw[fill=black!10!white,thick](0,0) circle(0.4cm);
    \node at (0,0) {$k$};
    \node[box,fill=black!10!white] at (0,1.6) {$n_{\bar{4}}$};
    \draw[postaction={decorate},thick] (-0.1,1.25)--(-0.1,0.4);
    \draw[postaction={decorate},red,thick] (0.1,1.25)--(0.1,0.4);
    \foreach \ang in {90,145,215,270} {
    \begin{scope}[rotate=\ang]
        \chiralarc[postaction={decorate},thick](0,0.5)(-45:225:0.22:0.65)
    \end{scope}
    }
    \foreach \ang in {90,145,270} {
    \begin{scope}[rotate=\ang]
    \fermiarc[postaction={decorate},thick](0,0.5)(-45:225:0.1:0.5)
    \end{scope}
    \node[left] at (-0.1,0.8) {$\mathsf{I}$};
    \node[right] at (0.1,0.8) {$\textcolor{red}{\Lambda_{\mathsf{I}}}$};
    \node[left] at (-1.5,0) {$\mathsf{B}_{2},\textcolor{red}{\Lambda_{2}}$};
    \node[right] at (1.6,0) {$\mathsf{B}_{1},\textcolor{red}{\Lambda_{1}}$};
    \node[below left] at (-0.9,-1){$\mathsf{B}_{3},\textcolor{red}{\Lambda_{3}}$};
    \node[below right] at (0.9,-1){$\mathsf{B}_{4}$};
    \draw[fill=black!10!white,thick](0,0) circle(0.4cm);
    \node at (0,0) {$k$};
    
    }
\end{scope}
\end{tikzpicture}}
\eea
where $\mathsf{B}_{a}\,(a\in\four),\mathsf{I}$ are the chiral superfields and $\Lambda_{\mathsf{I}},\Lambda_{1,2,3}$ are the Fermi superfields. We identify the chiral superfields with the first component field of the superfield. The superpotential $W$ in the 4 susy notation and the $J, E$-terms in the 2 susy notation are given as
\bea\label{eq:2susyD7onestack}
\qquad  
\mathsf{W}_{0}=\Tr\left(\mathsf{B}_{1}[\mathsf{B}_{2},\mathsf{B}_{3}]\right)
\quad E_{i}=[\mathsf{B}_{4},\mathsf{B}_{i}],\quad E_{\mathsf{I}}=\mathsf{B}_{4}\,\mathsf{I}   \quad J_{i}=\partial \mathsf{W}_{0}/\partial \mathsf{B}_{i}=\frac{1}{2}\varepsilon_{ijk4}[\mathsf{B}_{j},\mathsf{B}_{k}]
\eea

Assuming the following $U(1)$ action on the chiral fields
\bea
\mathsf{B}_{a}\rightarrow q_{a}\mathsf{B}_{a},\quad q_{a}=e^{\epsilon_{a}},\quad a\in\four
\eea
the chiral superfields and Fermi superfields have the following global symmetry charges:
\bea
\begin{tabular}{|c||c|}\hline
    Fields & $\U(1)^{3}$-charges  \\
\hline \hline $\mathsf{I}_{\bar{a}},\,(a\in\four)$     &  1\\
 \hline $\mathsf{B}_{a}$, $(a\in\four)$     &  $q_{a}$\\
\hline $\mathsf{\Lambda}_{1,2,3}$     &  $q_{4}q_{1,2,3}$ \\
\hline $\mathsf{\Lambda}_{\bar{a}}$ &  $q_{a}$  \\\hline
\end{tabular}
\eea
with the condition 
\bea
q_{1}q_{2}q_{3}q_{4}=1.
\eea
For later use, for a set $\mathcal{S}$ of elements of $\four$, we shortly write
\bea
q_{\mathcal{S}}=\prod_{a\in\mathcal{S}}q_{a},\quad \eps_{\mathcal{S}}=\sum_{a\in\mathcal{S}}\eps_{a}.
\eea

\begin{theorem}\label{thm:D6vacuum}
The Witten index of this setup is then given as
\bea
\mathcal{Z}^{\D6}_{\text{inst.}}=\sum_{k=0}^{\infty}\mathcal{Z}^{\D6}[k],\quad \mathcal{Z}^{\D6}[k]=\frac{1}{k!}\left(\frac{\sh(-\epsilon_{14,24,34})}{\sh(-\epsilon_{1,2,3,4})}\right)^{k}\oint_{\eta_0}\prod_{I=1}^{k}\frac{d\phi_{I}}{2\pi i }\mathcal{Z}^{\D6\tbar\D0}(\{\fra_{\alpha}\},\{\phi_{I}\})\prod_{I<J}\mathcal{Z}^{\D0\tbar\D0}(\phi_{I},\phi_{J})
\eea
where
\bea
\mathcal{Z}^{\D6\tbar\D0}(\{\fra_{\alpha}\},\phi_{I})&=\prod_{\alpha=1}^{n_{\bar{4}}}\mathcal{Z}^{\D6_{\bar{4}}\tbar\D0}(\fra_{\alpha},\phi_{I}),\quad \mathcal{Z}^{\D6_{\bar{4}}\tbar\D0}(\fra,\phi_{I})=\prod_{I=1}^{k}\frac{\sh(\phi_{I}-\fra-\epsilon_{4})}{\sh(\phi_{I}-\fra)},\\
\mathcal{Z}^{\D0\tbar\D0}(\phi_{I},\phi_{J})&=\frac{\sh(\phi_{I}-\phi_{J})\sh(\phi_{I}-\phi_{J}-\epsilon_{14,24,34})}{\sh(\phi_{I}-\phi_{J}-\epsilon_{1,2,3,4})}\frac{\sh(\phi_{J}-\phi_{I})\sh(\phi_{J}-\phi_{I}-\epsilon_{14,24,34})}{\sh(\phi_{J}-\phi_{I}-\epsilon_{1,2,3,4})}.
\eea
We chose the typical reference vector $\eta=\eta_0$ to specify the contour integral.
\end{theorem}
From a geometric point of view, the D6-D0 moduli space described by the quiver shown in \eqref{eq:4SUSYD7setup} is given by the Quot scheme of degree $k$ of rank $n_{\bar{4}}$ for $\mathbb{C}^3$ denoted by $\operatorname{Quot}_{\mathbb{C}^3}(\mathcal{O}^{\oplus n_{\bar{4}}},k)$. The $k$ \D0-brane contribution to the partition function $\mathcal{Z}^{\D6}[k]$ is then given by the equivariant Euler characteristics $\chi_{\mathsf{T}}(\operatorname{Quot}_{\mathbb{C}^3}(\mathcal{O}^{\oplus n_{\bar{4}}},k),\widehat{\mathcal{O}}^\text{vir})$ with the twisted virtual structure sheaf $\widehat{\mathcal{O}}^\text{vir}$~\cite{Nekrasov:2014nea}, and the total partition function $\mathcal{Z}^{\D6}_{\text{inst.}}$ is the corresponding generating function.

Note that since this system is an $\mathcal{N}=4$ system, the $\mathcal{N}=4$ vector multiplet comes from the $\mathcal{N}=2$ vector multiplet and the $\mathcal{N}=2$ adjoint chiral multiplet $\mathsf{B}_{4}$, which gives
\bea
\frac{\sh(\phi_I-\phi_J)}{\sh(\phi_I-\phi_J-\eps_4)}
\frac{\sh(\phi_J-\phi_I)}{\sh(\phi_J-\phi_I-\eps_4)}.
\eea
Moreover, the $\mathcal{N}=4$ chiral multiplets $\mathsf{B}_{i}\,(i=1,2,3)$ come from the $\mathcal{N}=2$ chiral multiplet $\mathsf{B}_{i}$ and the $\mathcal{N}=2$ Fermi multiplet $\Lambda_{i}$, which gives
\bea
\frac{\sh(\phi_I-\phi_J-\eps_{i}-\eps_{4})}{\sh(\phi_I-\phi_J-\eps_{i})}\frac{\sh(\phi_J-\phi_I-\eps_{i}-\eps_{4})}{\sh(\phi_J-\phi_I-\eps_{i})}
\eea
while the $\mathcal{N}=4$ chiral multiplet $\mathsf{I}$ comes from the $\mathcal{N}=2$ chiral multiplet $\mathsf{I}$ and the $\mathcal{N}=2$ Fermi multiplet, which gives
\bea
\frac{\sh(\phi_I-\fra_{\alpha}-\eps_4)}{\sh(\phi_I-\fra_{\alpha})}.
\eea
An observation is that for the $\mathcal{N}=4$ vector multiplet, the denominator is shifted by $-\eps_4$ compared to the numerator, while for the $\mathcal{N}=4$ chiral multiplet, the numerator is shifted by $-\eps_4$ compared to the denominator. The parameter $\eps_{4}$ is identified with the R-symmetry charge and such shift is a general phenomenon for the $\mathcal{N}=4$ system.

Using the quadrality symmetry, we also define
\bea
\mathcal{Z}^{\D6_{\bar{a}}\tbar\D0}(\fra,\phi_{I})=\prod_{I=1}^{k}\frac{\sh(\phi_{I}-\fra-\epsilon_{a})}{\sh(\phi_{I}-\fra)}.
\eea
To include other types of D6-branes wrapping different $\mathbb{C}^{3}_{\bar{a}}$-planes, we just need to modify the framing node contribution as
\bea
\mathcal{Z}^{\D6\tbar\D0}(\{\fra_{\bar{a},\alpha}\},\phi_{I})=\prod_{a\in\four}\prod_{\alpha=1}^{n_{\bar{a}}}\mathcal{Z}^{\D6_{\bar{a}}\tbar\D0}(\fra_{\bar{a},\alpha},\phi_{I}).
\eea

We also can include anti D6-brane contributions, which play the roles of adding 4 susy anti-fundamental chiral superfields to the quantum mechanics:
\bea\label{eq:D6antidef}
\mathcal{Z}^{\overline{\D6}_{\bar{a}}\tbar\D0}(\fra,\phi_{I})=\frac{\sh(\fra-\phi_{I})}{\sh(\fra+\epsilon_{a}-\phi_{I})}.
\eea
Although the physical interpretation of such term is yet to be understood, inclusion of such terms and application to $qq$-characters were discussed in \cite{Kimura:2023bxy}.

\paragraph{JK-residue and plane partitions}
Let us focus on the case with a single D6$_{\bar{4}}$-brane. The integrand takes the form as
\bea
\prod_{I=1}^{k}\frac{\sh(\phi_{I}-\fra-\epsilon_{4})}{\sh(\phi_{I}-\fra)}\prod_{I\neq J}\frac{\sh(\phi_{I}-\phi_{J})\sh(\phi_{I}-\phi_{J}-\epsilon_{14,24,34})}{\sh(\phi_{I}-\phi_{J}-\epsilon_{1,2,3,4})}.
\eea
The charge vectors are $\{e_{i},\,e_{i}-e_{j}\}$ where $e_{i}\in \mathbb{R}^{k}$ is the unit vector pointing the $i$-th direction. We associate the charge vectors as
\bea\label{eq:chargevector-corresp}
e_{I}\longleftrightarrow \phi_I-\fra,\quad e_{I}-e_{J} \longleftrightarrow \phi_{I}-\phi_{J}-\eps_{1,2,3,4}
\eea

We choose the reference vector to be $\eta=\eta_{0}=\sum_{i=1}^{k}e_{i}$ and explicitly evaluate low levels. Although, the JK-residue formalism gives the poles one needs to take the residues, not all of them give non-zero JK-residue. 

For level one, we only have the pole at $\phi_1=\fra$ corresponding to the charge vector $e_{1}$. For level two, using the Weyl invariance, the possible charge vectors are $\{e_{1},e_{2}\}$ and $\{e_{1},e_{2}-e_{1}\}$ giving
\bea
\{e_{1},e_{2}\},&\quad \phi_1-\fra=0,\,\,\phi_2-\fra=0,\\
\{e_{1},e_{2}-e_{1}\},&\quad \phi_1-\fra=0,\,\,\phi_{2}-\phi_1-\eps_{1,2,3,4}=0.
\eea
Note that we also have $\{e_{2},e_{1}-e_{2}\}$ from the Weyl invariance.

The first candidate is canceled by the numerator $\sh(\phi_{1}-\phi_{2})$ and the residue is zero. Similarly, the pole $\phi_{2}=\phi_{1}+\epsilon_{4}$ is canceled by the numerator $\sh(\phi_{2}-\fra-\epsilon_{4})$ and thus only the poles $\phi_{2}=\phi_{1}+\epsilon_{1,2,3}$ remain. Note that as mentioned in Example~5 in section~\ref{sec:JK-residue}, the iterative residue needs to be performed as
\bea
\underset{\phi_2=\fra+\eps_{i}}{\Res}
\underset{\phi_1=\fra}{\Res},\qquad (i=1,2,3).
\eea

Further studying higher levels, one can see that the poles giving the non-zero JK-residues are classified by a plane partition, which is a stack of cubes in a three-dimensional corner. Moreover, the iterative residue corresponds to how one can stack boxes to obtain the plane partition. 

We denote a plane partition $\pi$ as a two-dimensional sequence of non-negative integers with a non-increasing condition:
\bea
\pi=\{\pi_{i,j}\},\quad \pi_{i,j}\geq \pi_{i+1,j},\quad \pi_{i,j}\geq \pi_{i,j+1}.
\eea
The size of the plane partition is the number of boxes $|\pi|=\sum_{i,j}\pi_{i,j}$ and the set of them is denoted as $\mathcal{PP}$. The coordinates of the boxes in the plane partition are assigned $(i,j,k)\in\mathbb{Z}_{\geq 1}^{3}$ and 
\bea
(i,j,k)\in \pi \Leftrightarrow 1\leq k\leq \pi_{i,j}.
\eea
An interesting property of the plane partition is that boxes in it obey a \textit{melting rule}.
\begin{condition}\label{cond:DTmeltingrule}
     If any of 
    \bea
    (i+1,j,k),\quad (i,j+1,k),\quad (i,j,k+1)
    \eea
    is contained in $\pi$, then $(i,j,k)\in\pi$.
\end{condition}
We assign two types of coordinates: additive and multiplicative
\bea\label{eq:qcoordinates}
\text{additive:}&\quad c_{abc,\fra}(\cube)=\fra+(i-1)\epsilon_{a}+(j-1)\epsilon_{b}+(k-1)\epsilon_{c}\\
\text{multiplicative:}&\quad \chi_{abc,v}(\cube)=e^{c_{abc,\fra}(\scube)}=vq_{a}^{i-1}q_{b}^{j-1}q_{c}^{k-1},\quad v=e^{\fra}
\eea
which we call $\epsilon$-coordinates and $q$-coordinates, respectively.

\begin{theorem}\label{thm:tetraJKpoles}
Choosing the reference vector as $\eta=\eta_{0}=(1,\ldots,1)$, the poles are classified by a plane partition $\pi$ 
    \bea
    \{\phi_{I}\}_{I=1}^{k}\rightarrow \{c_{\bar{4},\fra}(\cube)\mid \cube=(x_{1},x_{2},x_{3})\in \pi\}
    \eea
   with $k=|\pi|$. The D6 partition function is then given as
   \bea
\mathcal{Z}^{\D6}_{\text{inst.}}=\sum_{\pi\in\mathcal{PP}}\mathfrak{q}^{|\pi|}\mathcal{Z}^{\D6}_{\bar{4}}[\pi],\quad \mathcal{Z}^{\D6}_{\bar{4}}[\pi]=\underset{\phi=\phi_{\pi}}{\Res} \mathcal{Z}^{\D6_{\bar{4}}\tbar\D0}(\fra,\phi_{I})\prod_{I<J}\mathcal{Z}^{\D0\tbar\D0}(\phi_{I},\phi_{J})   ,\quad k=|\pi|
\eea
where
\bea
\underset{\phi=\phi_{\pi}}{\Res}=\underset{\phi_{k}=\phi_{k\ast}}{\Res}\cdots \underset{\phi_{1}=\phi_{1\ast}}{\Res}.
\eea
The set $\{\phi_{I\ast}\}$ is the set of $\eps$-coordinates of boxes in the plane partition and the iterative residue is performed in a way such that for each step the melting rule in Cond.~\ref{cond:DTmeltingrule} is obeyed.
\end{theorem}

Generalizations to the higher rank case and the tetrahedron instantons are straightforward \cite{Pomoni:2021hkn,Awata:2009dd,Benini:2018hjy}. The poles will be simply classified by the coordinates of the boxes in multiple plane partitions with different orientations. 

\paragraph{Plethystic exponential (PE) formula}
 The tetrahedron instanton partition function with $n_{\bar{a}}$ D6$_{\bar{a}}$-branes has a plethystic exponential (PE) formula \cite{Pomoni:2021hkn,Pomoni:2023nlf,Fasola:2023ypx,Nekrasov:2017cih,Nekrasov:2018xsb}:
\bea\label{eq:tetrahedronPE}
\mathcal{Z}^{\D6}_{n_{\bar{1}},n_{\bar{2}},n_{\bar{3}},n_{\bar{4}}}[\fq,q_{1,2,3,4}]=\PE\left[\frac{[q_{14}][q_{24}][q_{34}]}{[q_{1}][q_{2}][q_{3}][q_{4}]}\frac{\left[\prod_{a\in\four}q_{a}^{n_{\bar{a}}}\right]}{\left[\fq \prod_{a\in\four}q_{a}^{-n_{\bar{a}}/2}\right]\left[\fq \prod_{a\in\four}q_{a}^{n_{\bar{a}}/2} \right]}\right]
\eea
where the PE is defined as
\bea
\PE\left[f(x_{1},\ldots,x_{n})\right]=\exp\left(\sum_{l=1}^{\infty}\frac{1}{l}f(x_{1}^{l},\ldots ,x_{n}^{l})\right).
\eea
For the $n_{\bar{4}}=1,n_{\bar{b}}=0\,(b\neq 4)$ case, we have
\bea\label{eq:D6U1PEformula}
\mathcal{Z}^{\D6}_{\bar{4}}[\fq,q_{1,2,3,4}]=\PE\left[\frac{[q_{14}][q_{24}][q_{34}]}{[q_{1}][q_{2}][q_{3}]}\frac{1}{[\fq q_{4}^{-1/2}][\fq q_{4}^{1/2}]}\right].
\eea

\section{DT3 and PT3 counting}\label{sec:DT-PT-counting}
In this section, we discuss the contour integral formalism of DT3 and PT3 counting. We first review the supersymmetric quantum mechanics associated with the DT3 counting in section~\ref{sec:DT3-JK} and then introduce the contour integral fomulas. We then review the box counting rules for PT3 counting and then propose contour integral formulas reproducing them. We show that the JK-residue of them indeed produces the PT3 invariants. We further discuss higher rank versions of them. We also discuss the DT/PT correspondence and relation with the refined topological vertices.

\subsection{DT3 counting and JK-residue}\label{sec:DT3-JK}
Let us review the supersymmetric quantum mechanics associated with the DT3 counting. See \cite[Sec.~6.1]{Galakhov:2021xum} for a reference. Physically, we are considering a setup where we have a D6-brane with D2--D0 states bound to it. Using the triality, we can consider the counting of D2--D0 bound states on only one D6$_{\bar{4}}$-brane. Generalizations to the cases with multiple D6-branes are given in section~\ref{sec:higherrankDTPT}.

\begin{figure}[t]
    \centering
    \includegraphics[width=0.3\linewidth]{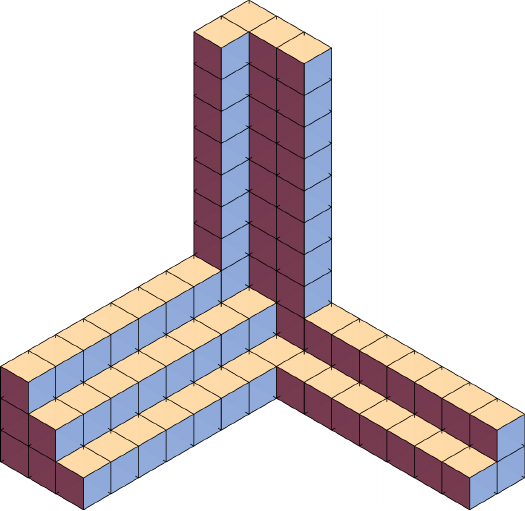}
    \hspace{2.5cm}
    \includegraphics[width=0.35\linewidth]{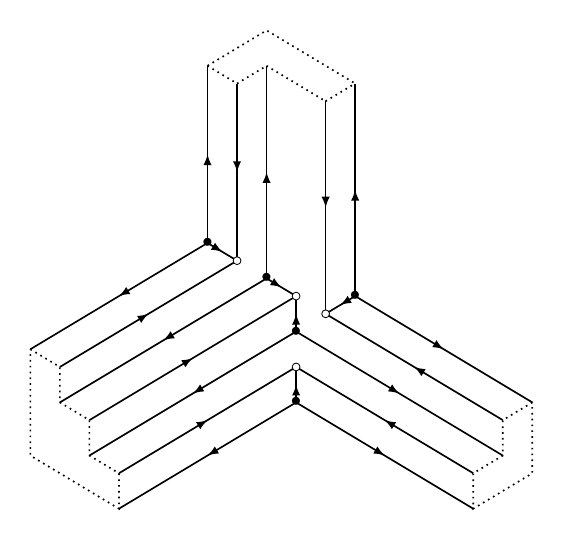}
    \caption{Left: Minimal plane partition with three boundary Young diagrams $\lambda=\{3,2,1\}$, $\mu=\{2,1\}$, and $\nu=\{3,1\}$. Right: The corresponding 2d lattice of the minimal plane partition.}
    \label{fig:minimal-pp}
\end{figure}
Combinatorially, we are considering plane partitions with nontrivial Young diagrams at the three axes. Let us derive the corresponding superpotential and contour integral formulas. We denote the three Young diagrams at the three axes $1,2,3$ as $\overrightarrow{Y}=(\lambda,\mu,\nu)$, respectively. The minimal plane partition is the smallest plane partition obeying this condition (see Fig.~\ref{fig:minimal-pp}). The orientation of the Young diagrams are chosen as in Fig.~\ref{fig:minimal-pp} and we denote the set of the boxes of the three-legs as
\bea
\mathcal{B}_{\lambda}&=\left\{(a,b,c)\mid a=1,\ldots,\infty,\,b=1,\ldots, \ell(\lambda),\,c=1,\ldots,\lambda_{b}\right\},\\
\mathcal{B}_{\mu}&=\left\{ (a,b,c) \mid b=1,\ldots,\infty,\,a=1,\ldots, \mu_{c},\,c=1,\ldots,\ell(\nu)  \right\},\\
\mathcal{B}_{\nu}&=\left\{(a,b,c) \mid c=1,\ldots,\infty,\, a=1,\ldots, \ell(\nu),\, b=1,\ldots, \nu_{a}   \right\}.
\eea
The length of the Young diagram is denoted as $\ell(\lambda)$. For later use, we also introduce
\bea
\,&\mathcal{B}_{\lambda\cap \mu}=\mathcal{B}_{\lambda}\cap \mathcal{B}_{\mu},\quad \mathcal{B}_{\lambda\cap \nu}=\mathcal{B}_{\lambda}\cap \mathcal{B}_{\nu},\quad \mathcal{B}_{\mu\cap \nu}=\mathcal{B}_{\mu}\cap \mathcal{B}_{\nu},\quad \mathcal{B}_{\lambda\cap\mu\cap\nu}=\mathcal{B}_{\lambda}\cap\mathcal{B}_{\mu}\cap \mathcal{B}_{\nu},\\
\,&\mathcal{B}_{\lambda\mu\nu}\coloneqq (\mathcal{B}_{\lambda}+\mathcal{B}_{\mu}+\mathcal{B}_{\nu})- (\mathcal{B}_{\lambda\cap \mu}+\mathcal{B}_{\lambda\cap \nu}+\mathcal{B}_{\mu\cap \nu})+\mathcal{B}_{\lambda\cap\mu\cap\nu}.
\eea

\subsubsection{Framed quiver and superpotential}\label{sec:DT-framedquiver}
The framed quiver and the superpotential corresponding with the plane partition with $\vec{Y}$ were given in \cite{Galakhov:2021xum}. Let us briefly review them.

We first project the minimal plane partition to a 2d lattice, which we denote by $\mathcal{C}_{\vec{Y}}$. After this projection, the vertices of the 2d lattice correspond to the convex and concave corners of the minimal plane partition. The procedure to determine the graph $\mathcal{C}_{\vec{Y}}$ is given as follows.

\begin{enumerate}
    \item We first define a coordinate system on the 2d lattice using three vectors $\vec{e}_{1,2,3}$ obeying $\vec{e}_{1}+\vec{e}_{2}+\vec{e}_{3}=0$, whose orientations are given as
    \bea
    \begin{tikzpicture}
        \draw[->, thick] (0,0)--(0,0.5);
        \node[above] at (0,0.5){$\vec{e}_{3}$};
        \draw[->, thick] (0,0)--(-30:0.5);
        \node[right] at (-30:0.5){$\vec{e}_{2}$};
        \draw[->,thick](0,0)--(210:0.5);
        \node[left] at (210:0.5){$\vec{e}_{1}$};
    \end{tikzpicture}
    \eea
    The coordinate of each vertex in the 2d lattice is defined by 
    \bea
    \vec{x}(\mathsf{v})=\sum_{i=1}^{3}n_{i}\vec{e}_{i},\quad n_{i}>0.
    \eea
    \item For a given set of Young diagrams $\vec{Y}$, we can convert the shape of the Young diagrams to three groups of semi-infinite arrows along the three directions $\vec{e}_{1,2,3}$ in the 2d lattice. Withing each group, the directions of the external lines alternate pointing the positive and negative directions of $\vec{e}_{i}$, while both ends are outgoing arrows. For example, the Young diagram $\nu=\{3,1\}$ is converted to
    \bea\label{eq:minimal-pp-1-leg}
    \adjustbox{valign=c}{\includegraphics[width=1.5cm]{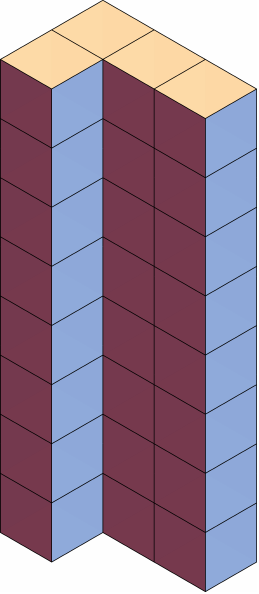}}\hspace{3cm} \adjustbox{valign=c,scale=0.9}{\begin{tikzpicture}[decoration={markings,mark=at position 0.5 with {\arrow{latex}}}]
 \tikzset{
        box/.style={draw, minimum width=0.6cm, minimum height=0.6cm, text centered,thick},
        ->-/.style={decoration={
        markings,mark=at position #1 with {\arrow[scale=1.5]{>}}},postaction={decorate},line width=0.5mm},
        -<-/.style={decoration={
        markings,
        mark=at position #1 with {\arrow[scale=1.5]{<}}},postaction={decorate},line width=0.5mm}    
    }
        \draw[postaction={decorate},thick] (0,0)--(0,3);
        \draw[thick,dotted] (0,3)--(0.5,2.7);
        \draw[postaction={decorate},thick] (0.5,2.7)--(0.5,0);
        \draw[thick,dotted] (0.5,2.7)--(1,3);
        \draw[postaction={decorate},thick] (1,0)--(1,3);
        \draw[thick,dotted] (1,3)--(2,2.4);
        \draw[postaction={decorate},thick] (2,2.4)--(2,0);
        \draw[thick,dotted] (2,2.4)--(2.5,2.7);
        \draw[postaction={decorate},thick] (2.5,0)--(2.5,2.7);
        \draw[thick,dotted](0,3)--(1,3.6)--(2.5,2.7);
    \end{tikzpicture}}
    \eea
    The distance between the external lines depends on the lengths of the edges of the boundaries of the Young diagrams. 

    \item The three semi-infinite lines in the 2d lattice intersect and the intersection determines the minimal plane partition configuration. So that the vertices in the 2d lattice matches with the convex and concave corners of the 3d plane partition, the lines in the 2d lattice need to obey the following two conditions. First, the arrows are either all incoming or all outgoing for each vertex. Second, all vertices are either valence-3 or valence-6. Additional internal lines are added to the external lines so that these conditions are satisfied.\\
    The corners and the vertices in $\mathcal{C}_{\vec{Y}}$ have the following correspondences.
    \begin{itemize}
        \item Concave corner $\Longleftrightarrow$ outgoing valence 3-vertex (black vertex)
        \bea
    \begin{tikzpicture}
        \draw[->, thick] (0,0)--(0,0.5);
        \draw[->, thick] (0,0)--(-30:0.5);
        \draw[->,thick](0,0)--(210:0.5);
        \node at (0,0){\textbullet};
    \end{tikzpicture}
    \eea
    We denote the set of boxes (concave corner boxes) corresponding to these black vertices in the three minimal plane partition as $s(\vec{Y})$.
    \item Half-convex corner $\Longleftrightarrow$ incoming valence 3-vertex (white vertex)\\
    There are six possible configurations:
    \bea
    \begin{tikzpicture}[decoration={markings,mark=at position 0.7 with {\arrow{latex}}}]
        \draw[postaction={decorate}] (90:0.5)--(0,0);
        \draw[postaction={decorate}] (150:0.5)--(0,0);
        \draw[postaction={decorate}](210:0.5)--(0,0);
        \node[circle,draw,fill=white,scale=0.4] at (0,0){};
    \end{tikzpicture}\qquad \quad 
    \begin{tikzpicture}[decoration={markings,mark=at position 0.7 with {\arrow{latex}}}]
        \draw[postaction={decorate}] (150:0.5)--(0,0);
        \draw[postaction={decorate}] (210:0.5)--(0,0);
        \draw[postaction={decorate}](270:0.5)--(0,0);
        \node[circle,draw,fill=white,scale=0.4] at (0,0){};
    \end{tikzpicture}\qquad \quad 
    \begin{tikzpicture}[decoration={markings,mark=at position 0.7 with {\arrow{latex}}}]
        \draw[postaction={decorate}] (210:0.5)--(0,0);
        \draw[postaction={decorate}] (270:0.5)--(0,0);
        \draw[postaction={decorate}](330:0.5)--(0,0);
        \node[circle,draw,fill=white,scale=0.4] at (0,0){};
    \end{tikzpicture}\qquad \quad 
    \begin{tikzpicture}[decoration={markings,mark=at position 0.7 with {\arrow{latex}}}]
        \draw[postaction={decorate}] (270:0.5)--(0,0);
        \draw[postaction={decorate}] (330:0.5)--(0,0);
        \draw[postaction={decorate}](390:0.5)--(0,0);
        \node[circle,draw,fill=white,scale=0.4] at (0,0){};
    \end{tikzpicture}\qquad \quad 
    \begin{tikzpicture}[decoration={markings,mark=at position 0.7 with {\arrow{latex}}}]
        \draw[postaction={decorate}] (330:0.5)--(0,0);
        \draw[postaction={decorate}] (390:0.5)--(0,0);
        \draw[postaction={decorate}](450:0.5)--(0,0);
        \node[circle,draw,fill=white,scale=0.4] at (0,0){};
    \end{tikzpicture}\qquad \quad 
    \begin{tikzpicture}[decoration={markings,mark=at position 0.7 with {\arrow{latex}}}]
        \draw[postaction={decorate}] (30:0.5)--(0,0);
        \draw[postaction={decorate}] (90:0.5)--(0,0);
        \draw[postaction={decorate}](150:0.5)--(0,0);
        \node[circle,draw,fill=white,scale=0.4] at (0,0){};
    \end{tikzpicture}
    \eea
    We denote the set of half-convex corners corresponding to these incoming valence 3-vertices as $p_{1}(\vec{Y})$.
    \item Full-convex corner $\Leftrightarrow$ incoming valence 6-vertex (white vertex)
    \bea
    \begin{tikzpicture}[decoration={markings,mark=at position 0.7 with {\arrow{latex}}}]
        \draw[postaction={decorate}] (30:0.5)--(0,0);
        \draw[postaction={decorate}] (90:0.5)--(0,0);
        \draw[postaction={decorate}](150:0.5)--(0,0);
        \draw[postaction={decorate}](210:0.5)--(0,0);
        \draw[postaction={decorate}](270:0.5)--(0,0);
        \draw[postaction={decorate}](330:0.5)--(0,0);
        \node[circle,draw,fill=white,scale=0.4] at (0,0){};
    \end{tikzpicture}
    \eea
    We denote the set of full-convex corners corresponding to the incoming valence 6-vertices as $p_{2}(\vec{Y})$.
    \end{itemize}
    For the example in Fig.~\ref{fig:minimal-pp}, the vertices are
    \bea
    s(\vec{Y})&=\{(3,4,1),(2,3,2),(2,2,3),(1,4,3),(3,1,4)\},\\
    p_{1}(\vec{Y})&=\{(3,4,2),(2,3,3),(2,4,3),(3,2,4)\},\\
    p_{2}(\vec{Y})&=\{\}.
    \eea

\end{enumerate}

\paragraph{Framed quiver and superpotential}
Using the information of the 2d lattice and the corners of the minimal plane partition, we can construct the framed quiver and superpotential. For each element $z\in s(\vec{Y})$, we introduce a fundamental chiral multiplet $\mathsf{I}_{z}$. For each element $w\in p_{1}(\vec{Y})$, we introduce an anti-fundamental chiral multiplet $\mathsf{J}_{w}$ and for each $w\in p_{2}(\vec{Y})$, we introduce a pair of anti-fundamental chiral multiplets $\tilde{\mathsf{J}}_{w,\alpha}$\,($\alpha=1,2$). The framed quiver is 
\bea\label{eq:4SUSYD6D2D0setup}
\adjustbox{valign=c}{
\begin{tikzpicture}[decoration={markings,mark=at position 0.9 with {\arrow{latex}}}]
 \tikzset{
        box/.style={draw, minimum width=0.7cm, minimum height=0.7cm, text centered,thick},
        ->-/.style={decoration={
        markings,mark=at position #1 with {\arrow[scale=1.5]{>}}},postaction={decorate},line width=0.5mm},
        -<-/.style={decoration={
        markings,
        mark=at position #1 with {\arrow[scale=1.5]{<}}},postaction={decorate},line width=0.5mm}    
    }
\begin{scope}[xshift=4cm]
    \node[box,fill=black!10!white] at (0,1.6) {$$};
    \draw[] (-0.1,1.25)--(-0.1,0.3);
     \draw[] (0.1,1.25)--(0.1,0.3);
    \draw[->,>=Stealth] (-0.1,1.25)--(-0.1,0.4);
    \draw[->,>=Stealth] (-0.1,1.25)--(-0.1,0.6);
    \draw[->,>=Stealth] (0.1,0.3)--(0.1,1.25);
    \draw[->,>=Stealth] (0.1,0.3)--(0.1,1.);
    \foreach \ang in {125,180,235} {
    \begin{scope}[rotate=\ang]
        \chiralarc[postaction={decorate},thick](0,0.5)(-45:225:0.22:0.65)
    \end{scope}
    }
    \node[left] at (-0.1,0.8) {$\mathsf{I}_{z}$};
    \node[right] at (0.1,0.8) {$\mathsf{J}_{w},\tilde{\mathsf{J}}_{w,\alpha}$};
    \node[] at (-1.5,-0.9) {$\mathsf{B}_{1}$};
    \node[below] at (0,-1.6) {$\mathsf{B}_{2}$};
    \node[ ] at (1.5,-0.9){$\mathsf{B}_{3}$};
    \draw[fill=black!10!white,thick](0,0) circle(0.4cm);
    \node at (0,0) {$$};
\end{scope}
\end{tikzpicture}}
\eea
The superpotential is modified as
\bea
\mathsf{W}=\mathsf{W}_{0}+\sum_{w\in p_{1,2}(\vec{Y})}\delta \mathsf{W}_{w},\quad \mathsf{W}_{0}=\Tr\mathsf{B}_{1}[\mathsf{B}_{2},\mathsf{B}_{3}],
\eea
and the $\U(1)^{3}$ charges are
\bea\label{eq:DTflavorcharges}
[\,\mathsf{I}_{z}\,]&=q_{1}^{i-1}q_{2}^{j-1}q_{3}^{k-1},\quad z=(i,j,k)\in s(\vec{Y}),\\
[\,\mathsf{J}_{w}\,]&=[\,\tilde{\mathsf{J}}_{w,\alpha}\,]=q_{4}^{-1}q_{1}^{-i+1}q_{2}^{-j+1}q_{3}^{-k+1},\quad w=(i,j,k)\in p_{1,2}(\vec{Y}).
\eea

The modification of the superpotential is determined as follows.
\begin{itemize}
    \item For $w\in p_{1}(\vec{Y})$, there are two situations. Each superpotential contribution from the white vertex $w\in p_{1}(\vec{Y})$ comes from the nearest two black vertices $z_{1,2}\in s(\vec{Y})$.
    \bea
    \adjustbox{valign=c}{\begin{tikzpicture}[decoration={markings,mark=at position 0.7 with {\arrow{latex}}}]
        \node at (30:1){\textbullet};
        \node at (150:1){\textbullet};
        \node[below] at (0,0) {$w$};
        \draw[postaction={decorate}] (90:1)--(0,0);
        \draw[postaction={decorate}] (150:1)--(0,0);
        \draw[postaction={decorate}] (30:1)--(0,0);
          \node[circle,draw,fill=white,scale=0.4] at (0,0){};
        \node[right] at (30:1){$z_{1}$};
        \node[left] at (150:1){$z_{2}$};
        \node[below,right] at (10:0.4){$i \vec{e}_{1}$};
        
        \node[below,left] at (170:0.4){$j \vec{e}_{2}$};
    \end{tikzpicture}}\qquad\qquad   &\delta \mathsf{W}_{w}=\mathsf{J}_{w}\left(\mathsf{B}_{1}^{i}\mathsf{I}_{z_{1}}-\mathsf{B}_{2}^{j}\mathsf{I}_{z_{2}}\right)\\
   \adjustbox{valign=c}{ \begin{tikzpicture}[decoration={markings,mark=at position 0.7 with {\arrow{latex}}}]
   \draw[postaction={decorate}] (90:1)--(0,0){};
   \draw[postaction={decorate}](150:1)--(0,0){};
   \draw [postaction={decorate}] (210:1)--(0,0){};
   \node at (150:1){\textbullet};
   \node[left] at (150:1){$z_{1}$};
   \node[below] at (150:0.7){$j \vec{e}_{2}$};

   \node at (350:1){\textbullet};
    \node[right] at (350:1){$z_{2}$};
    \draw[dashed,postaction={decorate}] (350:1)--(1,1.732/3);
    \node [right] at (1,1.732/3){$k \vec{e}_{3}$};
    \node[above] at (30:0.6){$i\vec{e}_{1}$};
    \draw [dashed,postaction={decorate}] (1,1.732/3)--(0,0);
   \node[circle,draw,fill=white,scale=0.4] at (0,0){};
   \node[below] at (0,0) {$w$};
    \end{tikzpicture}}\qquad \qquad &\delta \mathsf{W}_{w}=\mathsf{J}_{w}\left(\mathsf{B}_{2}^{j}\mathsf{I}_{z_{1}}-\mathsf{B}_{1}^{i}\mathsf{B}_{3}^{k}\mathsf{I}_{z_{2}}\right)
    \eea
    Other cases are simply obtained by the triality symmetry.
    \item For $w\in p_{2}(\vec{Y})$, the superpotential term looks like 
    \bea
    \adjustbox{valign=c}{\begin{tikzpicture}[decoration={markings,mark=at position 0.7 with {\arrow{latex}}}]
        \draw[postaction={decorate}] (30:1)--(0,0);
        \draw[postaction={decorate}] (90:1)--(0,0);
        \draw[postaction={decorate}](150:1)--(0,0);
        \draw[postaction={decorate}](210:1)--(0,0);
        \draw[postaction={decorate}](270:1)--(0,0);
        \draw[postaction={decorate}](330:1)--(0,0);
        \node at (30:1){\textbullet};
        \node [right] at (30:1){$z_{1}$};
        \node[above] at (32:0.6){$i\vec{e}_{1}$};
        \node at (150:1){\textbullet};
        \node [left] at (150:1){$z_{2}$};
        \node [above]at (148:0.6) {$j\vec{e_{2}}$};
        \node at (270:1){\textbullet};
        \node[below] at (270:1){$z_{3}$};
        \node[right] at (270:0.8){$k \vec{e}_{3}$};
        \node[circle,draw,fill=white,scale=0.4] at (0,0){};
    \end{tikzpicture}}\qquad\delta \mathsf{W}_{w}=\tilde{\mathsf{J}}_{w,1}\left(\mathsf{B}_{1}^{i}\mathsf{I}_{z_{1}}-\mathsf{B}_{2}^{j}\mathsf{I}_{z_{2}}\right)+\tilde{\mathsf{J}}_{w,2}\left(\mathsf{B}_{1}^{i}\mathsf{I}_{z_{1}}-\mathsf{B}_{3}^{k}\mathsf{I}_{z_{3}}\right).
    \eea
\end{itemize}

\subsubsection{Witten index and DT partition function}\label{eq:DT-Wittenindex}
Given the superpotential and the $\U(1)^{3}$-charges in \eqref{eq:DTflavorcharges}, and using \eqref{eq:D6antidef}, the chiral fields $\mathsf{I}_{z},\mathsf{J}_{w},\tilde{\mathsf{J}}_{w,\alpha}$ give the following contributions,
\bea
\mathsf{I}_{z}\,\,(z\in s(\vec{Y})):&\qquad  \frac{\sh(\phi_{I}-c_{\bar{4},\fra_{\alpha}}(\scube_{z})-\epsilon_{4})}{\sh(\phi_{I}-c_{\bar{4},\fra_{\alpha}}(\scube_{z}))}=\mathcal{Z}^{\D6_{\bar{4}}\tbar\D0}(c_{\bar{4},\fra_{\alpha}}(\scube_{z}),\phi_{I}),\\
\mathsf{J}_{w}\,\, (w\in p_{1}(\vec{Y})):&\qquad  \frac{\sh(c_{\bar{4},\fra_{\alpha}}(\scube_{w})-\phi_{I})}{\sh(\epsilon_{4}+c_{\bar{4},\fra_{\alpha}}(\scube_{w})-\phi_{I})}=\mathcal{Z}^{\overline{\D6}_{\bar{4}}\tbar \D0}(c_{\bar{4},\fra_{\alpha}}(\scube_{w}),\phi_{I}),\\
\tilde{\mathsf{J}}_{w,1,2}\, (w\in p_{2}(\vec{Y})):&\qquad \left(\frac{\sh(c_{\bar{4},\fra_{\alpha}}(\scube_{w})-\phi_{I})}{\sh(\epsilon_{4}+c_{\bar{4},\fra_{\alpha}}(\scube_{w})-\phi_{I})}\right)^{2}=\mathcal{Z}^{\overline{\D6}_{\bar{4}}\tbar \D0}(c_{\bar{4},\fra_{\alpha}}(\scube_{w}),\phi_{I})^{2},
\eea
where $\fra_{\alpha}, \phi_{I}$ are the Cartan generators of the gauge group of the D6 and D0-branes, respectively. Since the contribution from $\tilde{\mathsf{J}}_{w,1,2}$ always appear in pairs, the contribution from them are always squared.

We then can write down the general formula for the DT partition function with boundary conditions. Focusing on the case with only one D6-brane, we have the following claim.

\begin{definition}\label{def:DT3vertex-rank1-JKintegral}
The equivariant DT3 vertex is 
\bea
\,&\mathcal{Z}^{\DT\tbar\JK}_{\bar{4};\lambda\mu\nu}[\mathfrak{q},q_{1,2,3,4}]=\sum_{k=0}^{\infty}\mathfrak{q}^{k}\mathcal{Z}^{\DT\tbar\JK}_{\bar{4};\lambda\mu\nu}[k],
\eea
where the $k$-instanton sector of the D6$_{\bar{4}}$--D2--D0 partition function with one D6-brane is given as
\bea\label{eq:DT3contourint}
\,&\mathcal{Z}^{\DT\tbar\JK}_{\bar{4};\lambda\mu\nu}[k]=\frac{1}{k!}\left(\frac{\sh(-\epsilon_{14,24,34})}{\sh(-\epsilon_{1,2,3,4})}\right)^{k}\oint_{\eta_{0}} \prod_{I=1}^{k}\frac{d\phi_{I}}{2\pi i}\prod_{I=1}^{k}\mathcal{Z}^{\D6_{\bar{4}}\tbar\D2\tbar\D0}_{\DT;\lambda\mu\nu}(\fra,\phi_{I})\prod_{I<J}^{k}\mathcal{Z}^{\D0\tbar\D0}(\phi_{I},\phi_{J})
\eea
with
\bea\label{eq:DTflavornode-def}
\mathcal{Z}^{\D6_{\bar{4}}\tbar\D2\tbar\D0}_{\DT;\lambda\mu\nu}(\fra,\phi_{I})&=\prod_{\scube\in s(\vec{Y})}\mathcal{Z}^{\D6_{\bar{4}}\tbar\D0}(c_{\bar{4},\fra}(\cube),\phi_{I})\prod_{\scube\in p_{1}(\vec{Y})}\mathcal{Z}^{\overline{\D6}_{\bar{4}}\tbar\D0}(c_{\bar{4},\fra}(\cube),\phi_{I})\prod_{\scube\in p_{2}(\vec{Y})}\mathcal{Z}^{\overline{\D6}_{\bar{4}}\tbar\D0}(c_{\bar{4},\fra}(\cube),\phi_{I})^{2},
\eea
and the reference vector is $\eta_{0}$ and the sets $s(\vec{Y})$, $p_{1}(\vec{Y})$, and $p_{2}(\vec{Y})$ are the sets of the concave corners, half-convex corners, and the full-convex corners.
\end{definition}


\paragraph{Examples}
In the following sections of this paper, we will frequently use the following three simple examples:
\bea
(\lambda,\mu,\nu)=\begin{dcases}(\varnothing,\varnothing,\Bbox)\\
(\Bbox,\Bbox,\varnothing)\\
(\Bbox,\Bbox,\Bbox)
\end{dcases}
\eea
Each case is the most simple example when there are one-leg, two-legs, three-legs, respectively.

\begin{itemize}
\item {\textbf{One-leg}} 
\bea\label{eq:DTexample1}
\adjustbox{valign=c}{\begin{tikzpicture}[decoration={markings,mark=at position 0.7 with {\arrow{latex}}}]
    \draw[postaction={decorate}](30:0.5)--(0,0);
    \draw[postaction={decorate}](150:0.5)--(0,0);
    \node at (30:0.5){\textbullet};
    \node at (150:0.5){\textbullet};
    \draw [postaction={decorate}] (30:0.5)--++(-30:1.5);
    \draw [postaction={decorate}] (150:0.5)--++(210:1.5);
    \draw [postaction={decorate}] (30:0.5)--++(90:1.5);
    \draw [postaction={decorate}] (90:1.7)--(0,0);
    \draw [postaction={decorate}] (150:0.5)--++(90:1.5);
    \node[circle,draw,fill=white,scale=0.4] at (0,0){};
\end{tikzpicture}}\qquad \begin{dcases}
    s(\vec{Y})=\{(2,1,1),(1,2,1)\},\quad p_{1}(\vec{Y})=\{(2,2,1)\},\\
    \delta \mathsf{W}=\mathsf{J}_{(2,2,1)}(\mathsf{B}_{1}\,\mathsf{I}_{(1,2,1)}-\mathsf{B}_{2}\,\mathsf{I}_{(2,1,1)}),\\
    \mathcal{Z}^{\D6_{\bar{4}}\tbar\D2\tbar\D0}_{\DT;\varnothing\varnothing\,\Cbox}(\fra,\phi_{I})=\frac{\sh(\phi_{I}-\fra-\epsilon_{14,24})}{\sh(\phi_{I}-\fra-\epsilon_{1,2})}\frac{\sh(\fra-\phi_{I}-\epsilon_{34})}{\sh(\fra-\phi_{I}-\epsilon_{3})}.
\end{dcases}
\eea

\item {\textbf{ Two-legs}}
\bea\label{eq:DTexample2}
\adjustbox{valign=c}{\begin{tikzpicture}[decoration={markings,mark=at position 0.7 with {\arrow{latex}}}]
    \draw[postaction={decorate}](0,0)--(-30:1.5);
    \draw[postaction={decorate}](0,0)--(210:1.5);
    \draw [postaction={decorate}] (0,0)--(90:0.5);

    \draw [postaction={decorate}] (-30:1.5)+(0,0.5)--(90:0.5);
    \draw [postaction={decorate}] (210:1.5)+(0,0.5)--(90:0.5);

    \draw [postaction={decorate}] (0,1)--++(210:1.5);
    \draw [postaction={decorate}] (0,1)--++(-30:1.5);
     \draw [postaction={decorate}] (0,1)--++(90:1);
     
    \node[] at (0,0){\textbullet};
    \node[] at (0,1){\textbullet};
    \node[circle,draw,fill=white,scale=0.4] at (90:0.5){};
\end{tikzpicture}}\qquad \begin{dcases}
    s(\vec{Y})=\{(1,1,2),(2,2,1)\},\quad p_{1}(\vec{Y})=\{(2,2,2)\},\\
    \delta \mathsf{W}=\mathsf{J}_{(2,2,2)}(\mathsf{B}_{3}\,\mathsf{I}_{(2,2,1)}-\mathsf{B}_{1}\mathsf{B}_{2}\,\mathsf{I}_{(1,1,2)}),\\
    \mathcal{Z}^{\D6_{\bar{4}}\tbar\D2\tbar\D0}_{\DT;\,\Cbox\,\Cbox\,\varnothing}(\fra,\phi_{I})=\frac{\sh(\phi_{I}-\fra-\epsilon_{124})\sh(\phi_{I}-\fra-\epsilon_{34})}{\sh(\phi_{I}-\fra-\epsilon_{12})\sh(\phi_{I}-\fra-\epsilon_{3})}\frac{\sh(\fra-\phi_{I}-\epsilon_{4})}{\sh(\fra-\phi_{I})}.
\end{dcases}
\eea

\item {\textbf{Three-legs}}
\bea\label{eq:DTexample3}
\adjustbox{valign=c}{\begin{tikzpicture}[decoration={markings,mark=at position 0.7 with {\arrow{latex}}}]
    \draw[postaction={decorate}](0,0)--(-30:1.5);
    \draw[postaction={decorate}](0,0)--(210:1.5);
    \draw [postaction={decorate}] (0,0)--(90:0.5);
\draw [postaction={decorate}] (90:2)--(0,0.5);
    \draw [postaction={decorate}] (-30:1.5)+(0,0.5)--(90:0.5);
    \draw [postaction={decorate}] (210:1.5)+(0,0.5)--(90:0.5);

    \draw [postaction={decorate}] (1.732/4,1/4+0.5)--++(-30:1.5);
    \draw [postaction={decorate}] (1.732/4,1/4+0.5)--++(90:1.3);
    \draw [postaction={decorate}] (1.732/4,1/4+0.5)--++(210:0.5);

\draw [postaction={decorate}] (-1.732/4,1/4+0.5)--++(-30:0.5);
    \draw [postaction={decorate}] (-1.732/4,1/4+0.5)--++(90:1.3);
    \draw [postaction={decorate}] (-1.732/4,1/4+0.5)--++(210:1.5);
    

     \node[] at (1.732/4,1/4+0.5){\textbullet};
     \node[] at (-1.732/4,1/4+0.5){\textbullet};
    \node[] at (0,0){\textbullet};
    
    \node[circle,draw,fill=white,scale=0.4] at (90:0.5){};
\end{tikzpicture}}\qquad \begin{dcases}
    s(\vec{Y})=\{(2,1,2),(2,2,1),(1,2,2)\},\quad p_{2}(\vec{Y})=\{(2,2,2)\},\\
    \delta \mathsf{W}=\tilde{\mathsf{J}}_{(2,2,2),1}(\mathsf{B}_{3}\,\mathsf{I}_{(2,2,1)}-\mathsf{B}_{2}\,\mathsf{I}_{(2,1,2)})+\tilde{\mathsf{J}}_{(2,2,2),2}(\mathsf{B}_{3}\mathsf{I}_{(2,2,1)}-\mathsf{B}_{1}\mathsf{I}_{(1,2,2)}),\\
    \mathcal{Z}^{\D6_{\bar{4}}\tbar\D2\tbar\D0}_{\DT;\,\Cbox\,\Cbox\,\Cbox}(\fra,\phi_{I})=\frac{\sh(\phi_{I}-\fra-\epsilon_{124,134,234})}{\sh(\phi_{I}-\fra-\epsilon_{12,23,13})}\left(\frac{\sh(\fra-\phi_{I}-\epsilon_{4})}{\sh(\fra-\phi_{I})}\right)^{2}.
\end{dcases}
\eea
\end{itemize}

\begin{figure}
    \centering
    \includegraphics[width=5cm]{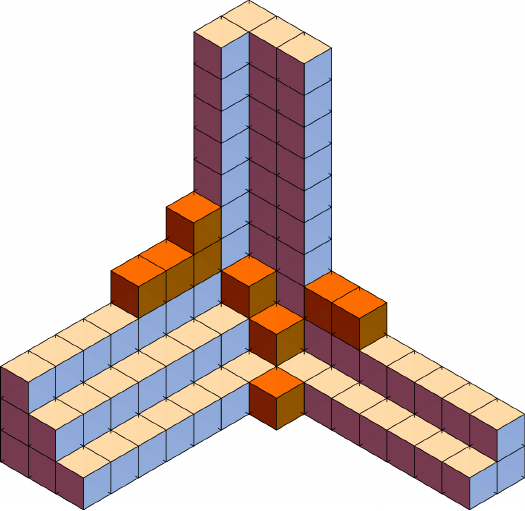}
    \caption{Plane partitions with boundary conditions.}
    \label{fig:DTconfig}
\end{figure}
\paragraph{JK-residue and classifications of poles} 
Evaluating the contour integral \eqref{eq:DT3contourint}, the poles are classified by possible plane partitions with boundary conditions (see Figure \ref{fig:DTconfig}). We denote the set of possible plane partitions with the boundary Young diagrams $(\lambda,\mu,\nu)$ as $\mathcal{DT}_{\lambda\mu\nu}$.
\begin{theorem}\label{thm:DTvertex-JKresidue}
    The DT3 partition function is
    \bea
    &\mathcal{Z}_{\bar{4};\lambda\mu\nu}^{\DT\tbar\JK}[\fq,q_{1,2,3,4}]=\sum_{\pi\in \mathcal{DT}_{\lambda\mu\nu}}\fq^{|\pi|}\mathcal{Z}^{\DT\tbar\JK}_{\bar{4};\lambda\mu\nu}[\pi],\\
    &\mathcal{Z}^{\DT\tbar\JK}_{\bar{4};\lambda\mu\nu}[\pi]=\underset{\phi=\phi_{\pi}}{\Res} \prod_{I=1}^{k}\mathcal{Z}^{\D6_{\bar{4}}\tbar\D0}_{\DT;\lambda\mu\nu}(\fra,\phi_{I})\prod_{I<J}\mathcal{Z}^{\D0\tbar\D0}(\phi_{I},\phi_{J}) 
    \eea
    where $\pi$ is the set of boxes that can be placed on the minimal plane partition without breaking the plane partition condition and $|\pi|$ is the number of such boxes (the orange boxes in Fig.~\ref{fig:DTconfig}).
\end{theorem}

Let us see it explicitly for low levels for \eqref{eq:DTexample1}. Choosing the reference vector $\eta_{0}$ picks
\bea
\phi_{I}=\fra+\epsilon_{1,2},\quad \phi_{I}-\phi_{J}=+\epsilon_{1,2,3,4}.
\eea
For level one, we have
\bea
\phi_{1}=\fra+\epsilon_{1,2}
\eea
which corresponds to the addable boxes at $(2,1,1),(1,2,1)$. 

For the next level, assume we have a box at $\phi_1=\fra+\eps_1$. The poles picked by the JK-residue formalism are
\bea
\phi_2=\fra+\eps_{1,2},\quad \phi_2=\phi_1+\eps_{1,2,3,4}.
\eea
The first one comes from the charge vector $\{e_{1},e_{2}\}$ and the second one comes from the charge vectors $\{e_{1},e_{2}-e_{1}\}$.

The pole coming from $\phi_2=\fra+\eps_{1}=\phi_1$ is canceled from the numerator part $\sh(\phi_1-\phi_2)$ coming from the numerator part of the D0-D0 term. The poles $\phi_2=\phi_1+\eps_4=\fra+\eps_{14}$ and $\phi_2=\phi_1+\eps_2=\fra+\eps_{12}=\fra-\eps_{34}$ are canceled by the numerator of \eqref{eq:DTexample1}. A similar analysis can be done when $\phi_1=\fra+\eps_2$ and the poles giving non-zero JK-residues are
\bea
(\phi_1,\phi_2)=(\fra+\eps_{1},\fra+\eps_{2}),\quad (\fra+\eps_{1},\fra+2\eps_{1}),\quad (\fra+\eps_{1}+\eps_{3}),\quad (\fra+\eps_{2},\fra+2\eps_2),\quad (\fra+\eps_2,\fra+\eps_2+\eps_{3}).
\eea
Note that when performing the iterative residue, the ordering of the poles need to obey melting rule for each level. Only for the first pole, we can perform the residue in either order $(\fra+\eps_1,\fra+\eps_2)$ or $(\fra+\eps_2,\fra+\eps_1)$.


\subsection{Infinite product and regularizations}\label{sec:infinite-product-reg}
In this section, we give a different way to derive the framing node contribution $\mathcal{Z}^{\D6_{\bar{4}}\tbar\D2\tbar\D0}_{\DT;\lambda\mu\nu}(\fra,\phi_I)$. Starting from an empty plane partition, whenever a box is added at $c_{\bar{4},\fra}(\cube)$ and the pole is evaluated, the contribution $\mathcal{Z}^{\D0\tbar\D0}(c_{\bar{4},\fra}(\cube),\phi_I)$ appears and modify the contour integrand. In particular, assume that we have a plane partition $\pi$ and have evaluated the poles corresponding to the boxes of the plane partition. The contour integrand then becomes
\bea
\prod_{I}\mathcal{Z}^{\D6_{\bar{4}}\tbar\D0}(\fra,\phi_I)\prod_{\scube\in\pi}\mathcal{Z}^{\D0\tbar\D0}(c_{123,\fra}(\cube),\phi_{I}).
\eea
The poles picked up from the JK-residue formalism will be the addable boxes of the plane partition and the plane partition grows after the evaluation.

For a general minimal plane partition with boundary conditions $(\lambda,\mu,\nu)$, the framing node contribution becomes
\bea
\mathcal{Z}^{\D6_{\bar{4}}\tbar\D2\tbar\D0}_{\DT;\lambda\mu\nu}(\fra,\phi_{I})&=\mathcal{Z}^{\D6_{\bar{4}}\tbar\D0}(\fra,\phi_{I})\prod_{\scube\in\mathcal{B}_{\lambda\mu\nu}}\mathcal{Z}^{\D0\tbar\D0}(c_{123,\fra}(\cube),\phi_{I}).
\eea
However, since we have infinite number of boxes in $\mathcal{B}_{\lambda\mu\nu}$, we need to regularize it. 

To make the discussion concrete, let us consider the boundary condition $(\varnothing,\varnothing,\Bbox)$. The boundary boxes give
\bea
\prod_{k=1}^{\infty}\mathcal{Z}^{\D0\tbar\D0}(\fra+\epsilon_{3}(k-1),\phi_{I})&=\prod_{k=1}^{\infty}\frac{\sh(\fra+\eps_3(k-1)-\phi_I)\sh(\fra+\eps_3(k-1)-\phi_I-\eps_{14,24,34})}{\sh(\fra+\eps_3(k-1)-\phi_I-\eps_{1,2,3,4})}\\
&\times \prod_{k=1}^{\infty}\frac{\sh(\phi_I-\fra-\eps_3(k-1))\sh(\phi_I-\fra-\eps_3(k-1)-\eps_{14,24,34})}{\sh(\phi_I-\fra-\eps_3(k-1)-\eps_{1,2,3,4})}\\
&=\prod_{k=1}^{\infty}\frac{\sh(\phi_I-\fra-(k-1)\eps_3)\sh(\phi_i-\fra-(k-1)\eps_3)}{\sh(\phi_I-\fra-\eps_3(k-2))\sh(\phi_I-\fra-k\eps_3)}\\
&\times \prod_{k=1}^{\infty}\frac{\sh(\phi_I-\fra-k\eps_3-\eps_{1,2,4})}{\sh(\phi_I-\fra-\eps_3(k-1)-\eps_{1,2,4})}\frac{\sh(\phi_I-\fra-\eps_{3}(k-2)+\eps_{1,2,4})}{\sh(\phi_I-\fra-\eps_{3}(k-1)+\eps_{1,2,4})}\\
&=\frac{\sh(\phi_I-\fra)}{\sh(\phi_I-\fra+\eps_3)}\frac{\sh(\phi_I-\fra+\eps_{31,32,34})}{\sh(\phi_I-\fra-\eps_{1,2,4})}\\
&\eqqcolon \mathcal{Z}^{\overline{\D2}_{3}\tbar\D0}(\fra,\phi_I)
\eea
where we used $\sh(x)=-\sh(-x)$ and $\sum_{a\in\four}\eps_{a}=0$. Since this physically corresponds to adding D2-branes to the system, we denote it as\footnote{The reason why we denote it as the anti D2-branes comes from the quantum algebra. Adding such D2-branes to the boundary conditions correspond to adding the inverse of the D2 vertex operators.} $\mathcal{Z}^{\overline{\D2}_{3}\tbar\D0}(\fra,\phi_I)$. Using the quadrality, we can also introduce 
\bea
\mathcal{Z}^{\overline{\D2}_{a}\tbar\D0}(\fra,\phi_I)=\frac{\sh(\phi_I-\fra)}{\sh(\phi_I-\fra+\eps_a)}\prod_{i\in\bar{a}}\frac{\sh(\phi_I-\fra+\eps_{a}+\eps_{i})}{\sh(\phi_I-\fra-\eps_i)}.
\eea
Using this, the framing node contribution is then given as
\bea
\mathcal{Z}^{\D6_{\bar{4}}\tbar\D2\tbar\D0}_{\DT;\,\varnothing\varnothing\,\Bbox}(\fra,\phi_{I})&=\mathcal{Z}^{\D6_{\bar{4}}\tbar\D0}(\fra,\phi_{I})\mathcal{Z}^{\overline{\D2}_{3}\tbar\D0}(\fra,\phi_I)\\
&=\frac{\sh(\phi_{I}-\fra-\epsilon_{14,24})\sh(\phi_{I}-\fra+\epsilon_{34})}{\sh(\phi_{I}-\fra-\epsilon_{1,2})\sh(\phi_{I}-\fra+\epsilon_{3})}.
\eea

Combinatorially, including the contribution $\mathcal{Z}^{\overline{\D2}_{a}\tbar\D0}(\frb,\phi_I)$ corresponds to adding a sequence of boxes extending semi-infinitely in the $a$-axis starting from the coordinate $\frb$ (we call this a rod). For general $\mathcal{B}_{\lambda\mu\nu}$, we can decompose it into multiple rods extending in the 1,2,3 directions without intersections. The framing node contribution then can be obtained by including the corresponding D2-brane contributions. For example, for $(\Bbox,\Bbox,\varnothing)$, the boundary boxes can be decomposed into a D2$_{1}$-brane at the origin $\fra$ and a D2$_{2}$-brane at $\fra+\eps_2$:
\bea
\mathcal{Z}^{\D6_{\bar{4}}\tbar\overline{\D2}\tbar\D0}_{\DT;\,\Bbox\,\Bbox\,\varnothing}(\fra,\phi_{I})&=\mathcal{Z}^{\D6_{\bar{4}}\tbar\D0}(\fra,\phi_{I})\mathcal{Z}^{\overline{\D2}_{1}\tbar\D0}(\fra,\phi_{I})\mathcal{Z}^{\overline{\D2}_{2}\tbar\D0}(\fra+\epsilon_{2},\phi_{I})\\
&=\frac{\sh(\phi_{I}-\fra-\epsilon_{34})\sh(\phi_{I}-\fra+\epsilon_{3})\sh(\phi_{I}-\fra+\epsilon_{4})}{\sh(\phi_{I}-\fra-\epsilon_{3})\sh(\phi_{I}-\fra+\epsilon_{34})\sh(\phi_{I}-\fra)}.
\eea
For the $(\Bbox,\Bbox,\Bbox)$ case, we can place a D2$_{3}$-brane at the origin $\fra$, a D2$_{1}$-brane at $\fra+\eps_1$, and a D2$_2$-brane at $\fra+\eps_2$:
\bea
\mathcal{Z}^{\D6_{\bar{4}}\tbar\overline{\D2}\tbar\D0}_{\DT;\,\Bbox\,\Bbox\,\Bbox}(\fra,\phi_{I})&=\mathcal{Z}^{\D6_{\bar{4}}\tbar\D0}(\fra,\phi_{I})\mathcal{Z}^{\overline{\D2}_{3}\tbar\D0}(\fra,\phi_{I})\mathcal{Z}^{\overline{\D2}_{1}\tbar\D0}(\fra+\epsilon_{1},\phi_{I})\mathcal{Z}^{\overline{\D2}_{2}\tbar\D0}(\fra+\epsilon_{2},\phi_{I})\\
&=\frac{\sh(\phi_{I}-\fra-\epsilon_{124,134,234})}{\sh(\phi_{I}-\fra-\epsilon_{12,23,13})}\left(\frac{\sh(\phi_{I}-\fra+\epsilon_{4})}{\sh(\phi_{I}-\fra)}\right)^{2}.
\eea

Although the above derivation provides a way to give the factors appearing in the framing node contribution, it does not tell which factor correspond to the fundamental chiral or the anti-fundamental chiral multiplets. Such information is important because they determine how one evaluates the contour integral. Let us give a way to determine this. We first start by rewriting the D2$_{1,2,3}$ contributions using the D6$_{\bar{4}}$, $\overline{\D6}_{\bar{4}}$ contributions:
\bea
\mathcal{Z}^{\overline{\D2}_{\bar{3}}\tbar\D0}(\fra,\phi_I)=\mathcal{Z}^{\overline{\D6}_{\bar{4}}\tbar\D0}(\fra,\phi_{I})\mathcal{Z}^{\overline{\D6}_{\bar{4}}\tbar\D0}(\fra-\eps_{34},\phi_{I})\mathcal{Z}^{\D6_{\bar{4}}\tbar\D0}(\fra+\eps_{1,2},\phi_{I})
\eea
and others are obtained by the triality. We then use
\bea
\mathcal{Z}^{\D6_{\bar{4}}\tbar\D0}(\fra,\phi_{I})\mathcal{Z}^{\overline{\D6}_{\bar{4}}\tbar\D0}(\fra,\phi_{I})=1
\eea
and cancel the extra factors. For example, for the case $(\varnothing,\varnothing,\Bbox)$, we have
\bea
\mathcal{Z}^{\D6_{\bar{4}}\tbar\overline{\D2}\tbar\D0}_{\DT;\,\Bbox\,\Bbox\,\varnothing}(\fra,\phi_{I})&=\mathcal{Z}^{\D6_{\bar{4}}\tbar\D0}(\fra,\phi_{I})\mathcal{Z}^{\overline{\D2}_{3}\tbar\D0}(\fra,\phi_I)\\
&=\mathcal{Z}^{\D6_{\bar{4}}\tbar\D0}(\fra,\phi_{I})\mathcal{Z}^{\overline{\D6}_{\bar{4}}\tbar\D0}(\fra,\phi_{I})\mathcal{Z}^{\overline{\D6}_{\bar{4}}\tbar\D0}(\fra-\eps_{34},\phi_{I})\mathcal{Z}^{\D6_{\bar{4}}\tbar\D0}(\fra+\eps_{1,2},\phi_{I})\\
&=\mathcal{Z}^{\overline{\D6}_{\bar{4}}\tbar\D0}(\fra-\eps_{34},\phi_{I})\mathcal{Z}^{\D6_{\bar{4}}\tbar\D0}(\fra+\eps_{1,2},\phi_{I})
\eea
which indeed matches with \eqref{eq:DTexample1}. Other examples can be done similarly.

\subsection{PT3 counting rules and coordinate system}\label{sec:PTrule-coord}
Let us review the PT box counting rules for each case when we have one, two, and three legs, respectively.

\subsubsection{One-leg}
\begin{figure}
\begin{minipage}{0.15\columnwidth}
    \centering
    \includegraphics[width=2cm]{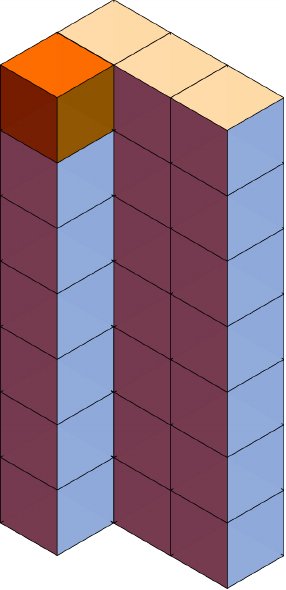}
\end{minipage}
\begin{minipage}{0.15\columnwidth}
    \centering
    \includegraphics[width=2cm]{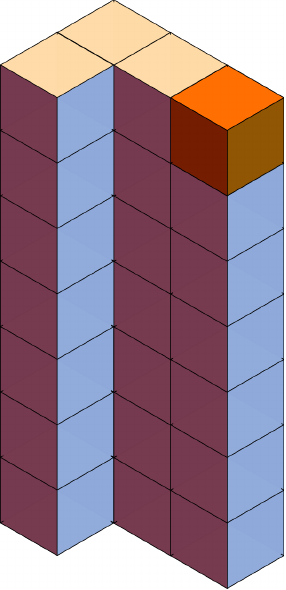}
\end{minipage}
\begin{minipage}{0.15\columnwidth}
    \centering
    \includegraphics[width=2cm]{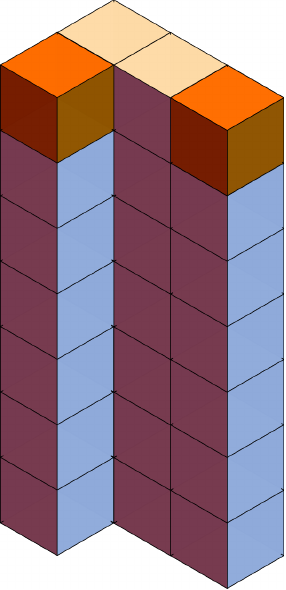}
\end{minipage}
\begin{minipage}{0.15\columnwidth}
    \centering
    \includegraphics[width=2cm]{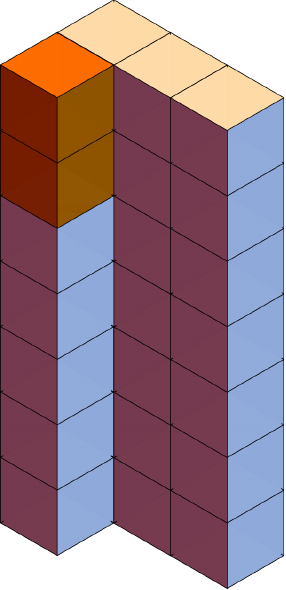}
\end{minipage}
\begin{minipage}{0.15\columnwidth}
    \centering
    \includegraphics[width=2cm]{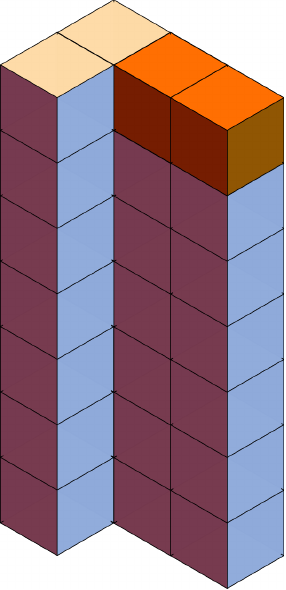}
\end{minipage}
\begin{minipage}{0.15\columnwidth}
    \centering
    \includegraphics[width=2cm]{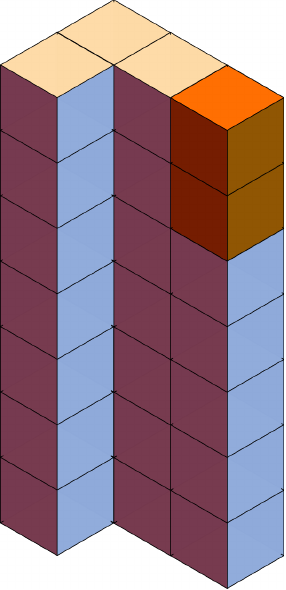}
\end{minipage}
\caption{ PT configurations for $\lambda=\mu=\varnothing$ and $\nu=\{3,1\}$ with one and two boxes. The coordinates of the set of orange boxes are $(\epsilon_{1}-\epsilon_{3})$, $(2\epsilon_{2}-\epsilon_{3})$, $(\epsilon_{1}-\epsilon_{3},2\epsilon_{2}-\epsilon_{3})$, $(\epsilon_{1}-\epsilon_{3},\epsilon_{1}-2\epsilon_{3})$, $(\epsilon_{2}-\epsilon_{3},2\epsilon_{2}-\epsilon_{3})$, and $(2\epsilon_{2}-\epsilon_{3},2\epsilon_{2}-2\epsilon_{3})$, respectively.}
\label{fig:PT3oneleg}
\end{figure}

Starting from the minimal plane partition \eqref{eq:minimal-pp-1-leg} with legs $(\varnothing,\varnothing,\nu)$, we extend the half cylinder of shape $\nu$ in the negative direction of the third coordinate. We then remove the boxes originally present in the standard box counting. For the one-leg case, we end up with the same configuration as shown in Fig.~\ref{fig:PT3oneleg}. However, they now represent a half cylinder starting at the origin and going further in the negative direction. We then can think of the resulting half cylinder as a hollow structure in which the boxes are stacked. The difference with the ordinary box counting is that the gravity is now pointing towards the $(1,1,1)$ direction. 
\begin{condition}\label{cond:PTrule-oneleg}
    Let $(i,j,k)$ belong to the one-leg half cylinder extending in the negative direction. If any of 
    \bea
    (i-1,j,k),\quad (i,j-1,k),\quad (i,j,k-1)
    \eea
    are contained in the box configuration, then $(i,j,k)$ is also included.
\end{condition}
Possible positions where we can add the boxes for the $\nu=\{3,1\}$ case are given in Fig.~\ref{fig:PT3oneleg}.

\subsubsection{Two-legs}
\begin{figure}[tbp]
  \hspace{0.04\columnwidth} 
  \begin{minipage}[b]{0.48\columnwidth}
    \centering
    \includegraphics[width=5cm]{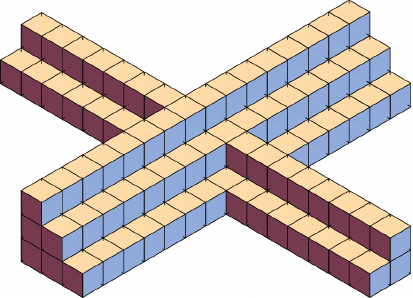}
  \end{minipage}
  \begin{minipage}[b]{0.48\columnwidth}
    \centering
    \includegraphics[width=6cm]{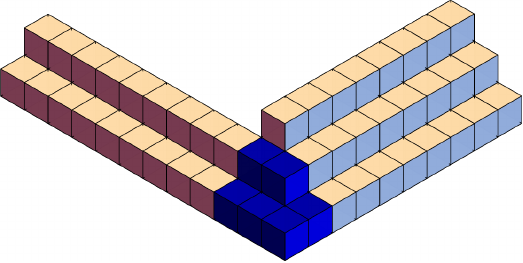}
  \end{minipage}
   \caption{Construction of the PT box-counting with two nontrivial legs: $\lambda=\{3,2,1\}$, $\mu=\{2,1\}$, and $\nu=\varnothing$. Left: Configuration with asymptotic cylinders extended in both directions obtained from the minimal plane partition. Right: Configuration with boxes in the positive directions removed and boxes in the intersection colored. The blue boxes are the boxes belonging to the intersection of the two legs and they are kept.}
   \label{fig:PT3twoleg}
\end{figure}
For the two-legs case, the situation is more complicated. Similar to the one-leg case, we start from the minimal plane partition with two nontrivial legs (see Fig.~\ref{fig:PT3twoleg}) and extend the two legs in the negative directions of $x_{1}, x_{2}$. We then remove all the boxes with coordinates $i,j,k\geq 0$ except the intersection of the two cylinders (the boxes $\mathcal{B}_{\lambda\cap \mu}$). They are colored in blue in the figure. The resulting half cylinder is then understood as a hollow structure and boxes are placed inside it. Again, the gravity is pointing the positive direction $(1,1,1)$.
\begin{condition}\label{cond:PTrule-twolegs}
    Let $(i,j,k)$ belong to the two-legs half cylinder extending in the negative directions. If any of 
    \bea
    (i-1,j,k),\quad (i,j-1,k),\quad (i,j,k-1)
    \eea
    are contained in the box configuration, then $(i,j,k)$ is also included.
\end{condition}

For the situation in Fig.~\ref{fig:PT3twoleg}, there are three allowed configurations for level one:
\bea
\adjustbox{valign=c}{\includegraphics[width=3cm]{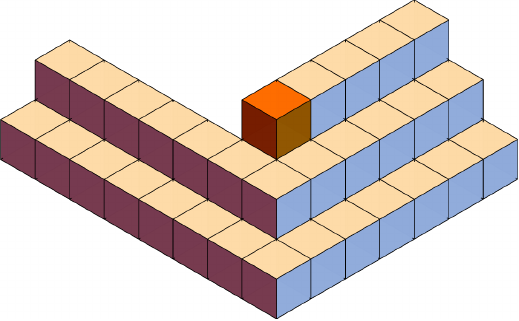}}\qquad \adjustbox{valign=c}{\includegraphics[width=3cm]{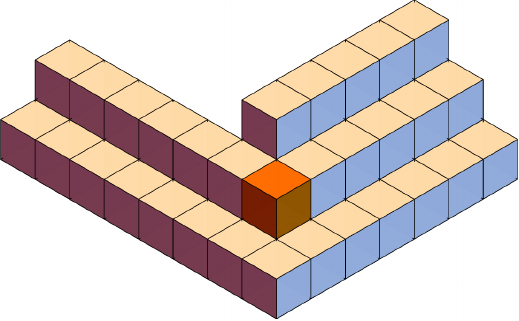}}\qquad \adjustbox{valign=c}{\includegraphics[width=3cm]{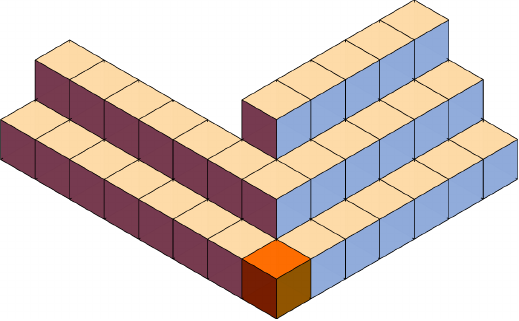}}
\eea
and eight allowed configurations for level two:
\bea
\adjustbox{valign=c}{\includegraphics[width=3cm]{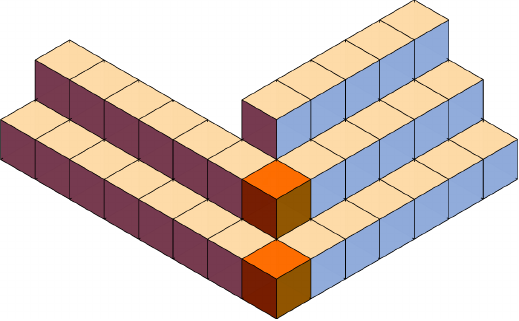}}\quad \adjustbox{valign=c}{\includegraphics[width=3cm]{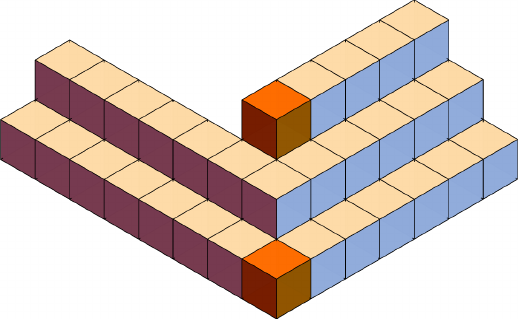}}\quad \adjustbox{valign=c}{\includegraphics[width=3cm]{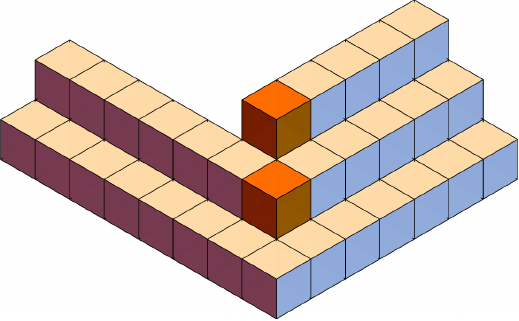}}\quad \adjustbox{valign=c}{\includegraphics[width=3cm]{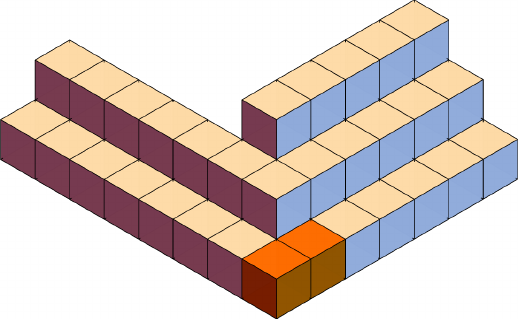}}\\
\adjustbox{valign=c}{\includegraphics[width=3cm]{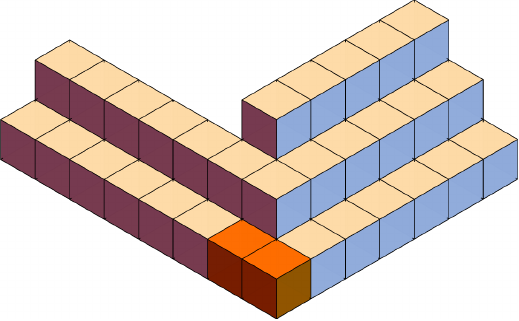}}\quad \adjustbox{valign=c}{\includegraphics[width=3cm]{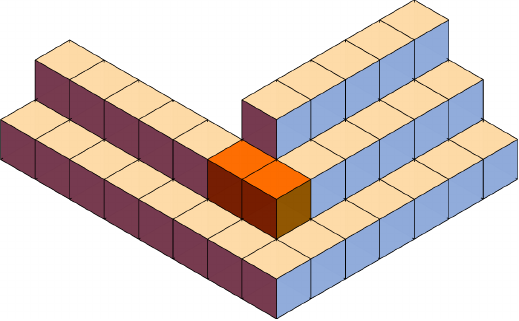}}\quad \adjustbox{valign=c}{\includegraphics[width=3cm]{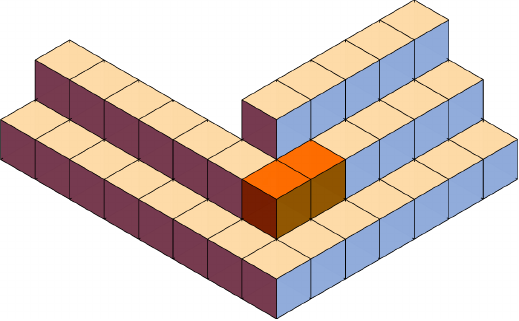}}\quad \adjustbox{valign=c}{\includegraphics[width=3cm]{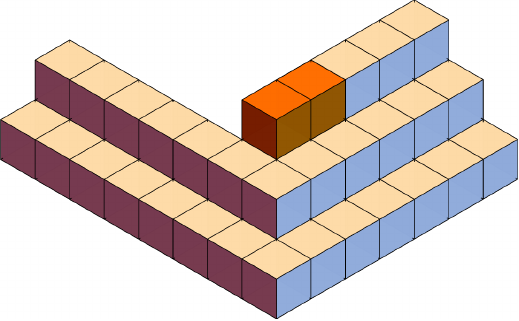}}
\eea

\subsubsection{Three-legs}
\begin{figure}[tbp]
  \hspace{0.04\columnwidth} 
  \begin{minipage}[b]{0.48\columnwidth}
    \centering
    \includegraphics[draft=false,width=4cm]{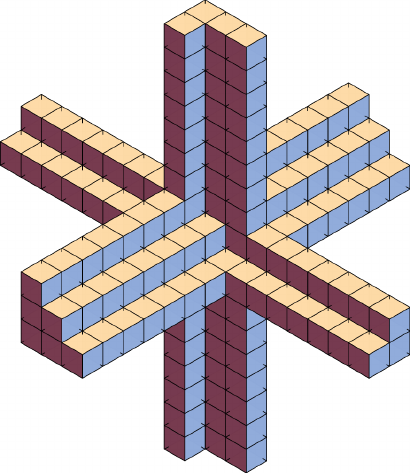}
  \end{minipage}
  \begin{minipage}[b]{0.48\columnwidth}
    \centering
    \includegraphics[draft=false,width=5cm]{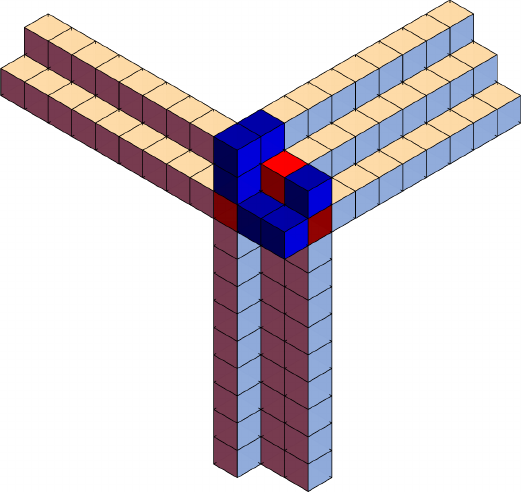}
  \end{minipage}
   \caption{Construction of the PT box-counting with three trivial legs: $\lambda=\{3,2,1\}$, $\mu=\{2,1\}$, and $\nu=\{3,1\}$. Left: Configuration with asymptotic cylinders extended in both directions obtained from the minimal plane partition. Right: Configuration with boxes in the positive directions removed and boxes in the intersection colored. Red boxes are boxes belonging to the three legs. }
   \label{fig:PT3threeleg}
\end{figure}
Let us then move on to the PT counting when there are three nontrivial legs. We review both the descriptions in \cite{Pandharipande:2007kc,Pandharipande:2007sq} and \cite{Gaiotto:2020dsq}. Similarly, we start from the minimal plane partition Fig.~\ref{fig:minimal-pp}. We first extend the three legs in the negative direction as in Fig.~\ref{fig:PT3threeleg}. We then remove all the boxes in the positive quadrant while keeping the boxes in the intersection of the three legs. The boxes that belong to the intersection of exactly two legs, i.e. $\mathcal{B}_{\lambda\cap\mu},\mathcal{B}_{\mu\cap\nu},\mathcal{B}_{\lambda\cap\nu}$, are colored in blue. The boxes belonging to all of the three legs, i.e. $\mathcal{B}_{\lambda\cap\mu\cap \nu}$, are colored in red. See the right of Fig.~\ref{fig:PT3threeleg}.

\paragraph{Pandharipande--Thomas (PT) box counting rules}

Let us first review the allowed box configurations for PT box counting following \cite{Pandharipande:2007kc,Pandharipande:2007sq,Jenne2020TheCP,Jenne:2021irh}. We start from classifying the boxes in the hollow structure into three sets I$^{-}$, II, III.
\begin{itemize}
    \item I$^{-}$ consists of boxes belonging to the negative quadrant of the extended minimal plane partition.
    \item II consists of all boxes lying in exactly 2 cylinders, i.e. the blue positions. We further denote $\text{II}_{\bar{a}}\,(a=1,2,3)$ as the boxes lying in exactly 2 cylinders extending in the $b,c$ axes, such that $a\neq b\neq c$.
    \item III consists of all boxes lying in exactly 3 cylinders, i.e the red positions.
\end{itemize}
 The PT configurations are labeled box configurations in $\text{II}\cup\text{ III} \cup\text{ I}^{-}$. A labeled box configuration is a finite set of boxes in $\text{II}\cup\text{ III} \cup\text{ I}^{-}$ with the boxes in III labeled.\footnote{In this paper, we use the following labels $\{-1,0,1,2,3\}$.} The rules of the labeling are given as follows.
 \begin{condition}[\textbf{PT counting rules}]\label{cond:PTrules}
 The possible PT configurations for the three-legs case obey the following rules.
     \begin{enumerate}
         \item For a box $(i,j,k)\in \text{I}^{-}$, if any of $(i-1,j,k),\,(i,j-1,k),\, (i,j,k-1)$
         are included in the box configuration, then $(i,j,k)$ is also included. In other words, the box $(i,j,k)\in \text{I}^{-}$ supports the boxes $(i-1,j,k),\,(i,j-1,k),\,(i,j,k-1)$ from the positive direction. The box $(i,j,k)$ is an unlabeled box.
         \item If $(i,j,k)\in \text{II}_{\bar{a}}$ for $a=1,2,3$ and if any of 
         $(i-1,j,k),\,(i,j-1,k),\, (i,j,k-1)$ is not a type III box labeled by $a$ that is included in the configuration, then $(i,j,k)$ is also included. In other words, the boxes $(i-1,j,k),\,(i,j-1,k),\,(i,j,k-1)$ that are not a type III box with label $a$ needs to be supported by the box $(i,j,k)\in \text{II}_{\bar{a}}$.

         \item If $(i,j,k)\in\text{III}$ and the boxes $(i-1,j,k),\,(i,j-1,k),\,(i,j,k-1)$ support other boxes from the positive direction, then $(i,j,k)$ is also included. If the supported boxes of $(i,j,k)$ belong to the cylinder extending in only one direction $a$, then $(i,j,k)$ is labeled by $a$. If the supported boxes belong to two cylinders, then $(i,j,k)$ must be unlabeled.

        
         

        \item If $(i,j,k)\in\text{III}$ admits any labeling without breaking the conditions above, it is called freely labeled. Such situation gives a two states degeneracy.
     \end{enumerate}
 \end{condition}

In this paper, to distinguish the type III unlabeled boxes with other unlabeled boxes, we label it with $-1$. For the case when the type III box is freely labeled, we label it with $0$.

When counting the size of the PT configuration, the unlabeled box (the box with label $-1$) contributes as $2$, while other boxes contribute as $1$. On the other hand, when counting the possible PT configurations, the configuration with a type III freely labeled box is counted twice.


\paragraph{Gaiotto--Rap\v{c}\'{a}k (GR) box counting rules}
Gaiotto--Rap\v{c}\'{a}k proposed a generalization of the PT box counting rules in \cite{Gaiotto:2020dsq} and gave transition rules between the box configurations.  
Again, the box configuration grows in the negative direction inside the hollow structure, and for the red positions, we have three kinds of boxes we can place.
\begin{enumerate}
    \item A \textbf{light box} that can support other light boxes.
    \item A \textbf{heavy box} that can support other heavy boxes or boxes at uncolored positions.
    \item An \textbf{ultra-heavy box} that can support boxes in all directions. Only the ultra-heavy box contributes two to the size of the PT configuration.
\end{enumerate}

\begin{condition}[\textbf{GR counting rules}]\label{cond:GRrules}The rules so that we can add boxes at the red positions are given as follows. \begin{enumerate} 
    \item All of the boxes placed at the red position, except of the heavy box, needs to be supported from the positive direction. Such heavy box at the red position does not have to be supported by boxes at the blue positions. If the heavy box is not supported by a box in the blue position belonging to the two legs $i,j\,(i\neq j)$, we then label it by $k $ such that $i\neq  j\neq k$. 
\item 
A heavy box placed at the red position can support heavy boxes and boxes at uncolored positions if the resulting configuration admits a consistent labeling. In particular, all the heavy boxes supported by it can be labeled by either 1, 2 or 3 in a way that on top of a red position labeled by $i$, there are only boxes at uncolored locations inside the cylinder along the axis $i$ or another heavy box labeled by $i$.

\item  A family of adjacent red boxes can support light boxes if they admit at least two consistent labels described above.
\end{enumerate}

\end{condition}

Additionally, we have the transition rules between GR box configurations.

\begin{condition}[\textbf{GR transition rules}]\label{cond:GRtransitionrule}

The transition rules between the GR box configurations are as follows. \begin{itemize}
\item  The light box and the heavy box can change into the ultra-heavy box or vice versa if the resulting configuration is allowed.
\item When a box at an uncolored location on top of the heavy box is removed and the resulting configuration admits multiple consistent labels, only the light box will be generated. 

\item  If an unsupported box is not supporting other boxes and becomes supported, only the light box is generated. In the opposite process, the heavy box can become unsupported but not the light one.
\item  Whenever the resulting configuration contains heavy boxes at red locations admitting two or more consistent labeling, both heavy boxes and light boxes can be generated.
\end{itemize}

\end{condition}

\paragraph{Relation between the two box configurations}
The PT box configurations and the GR box configurations are identified by the following correspondence:
\bea\label{eq:PT-GR-correspondence}
\text{type III freely labeled box (label 0)}\quad &\longleftrightarrow \quad \text{a pair of heavy and light box}, \\
\text{type III labeled box (label 1,2,3) }\quad &\longleftrightarrow \quad \text{heavy box},\\
\text{type III unlabeled box}\quad &\longleftrightarrow\quad  \text{ultra-heavy box}.
\eea
Note that in the PT box configurations, the light box will always appear in pair with a heavy box. The GR box counting rules include additional transition rules between the box configurations which is not obvious in the PT box counting rules. In particular, generation of light boxes is the nontrivial part. However, after proper identifications of the boxes, the total number of the box configurations for each level is the same.


\paragraph{Coordinate system}
The $\eps$-coordinates and $q$-coordinates of the boxes in the PT configurations are assigned by extending the coordinate system of the original plane partition given in \eqref{eq:qcoordinates}. For example, for the most left configuration of Fig.~\ref{fig:PT3oneleg}, the coordinate of the orange box is $(2,1,0)$ and the $\eps$- and $q$-coordinates are $\epsilon_{1}-\epsilon_{3}$ and $q_{1}q_{3}^{-1}$, where we set the coordinate of the origin to be trivial for simplicity. 

We denote the set of box configurations obeying the PT-box counting rules in Cond.~\ref{cond:PTrules} with this coordinate system as $\mathcal{PT}_{\lambda\mu\nu}$. On the other hand, we denote the set of box configurations obeying the GR-box counting rules in Cond.~\ref{cond:GRrules} as $\mathcal{GR}_{\lambda\mu\nu}$. An element of them is denoted as $\pi$ following the notation of the plane partition.

Note that the box configurations $\mathcal{PT}_{\lambda\mu\nu}$ and $\mathcal{GR}_{\lambda\mu\nu}$ are identified under the correspondence \eqref{eq:PT-GR-correspondence}. In particular, a PT box configuration $\pi$ including a freely labeled type III box is understood as a pair of GR box configurations $\Lambda,\Lambda'$, where $\Lambda$ ($\Lambda'$) contains only the light (heavy) box, with other boxes being the same. Namely, we have $\pi=(\Lambda,\Lambda')$.

\paragraph{A different description and coordinate system}
For later use, we introduce a different coordinate system by flipping the $q$-parameters and $\eps$-parameters as $q_{a}\rightarrow q_{a}^{-1}$, $\eps_a\rightarrow -\eps_a$ (see Fig.~\ref{fig:PTreverse}).  To distinguish with the previous coordinate system, we denote the set of PT, GR box configurations obeying this coordinate system as $\widetilde{\mathcal{PT}}_{\lambda\mu\nu}, \widetilde{\mathcal{GR}}_{\lambda\mu\nu}$, respectively. We also denote the set of the DT counting in this coordinate system as $\widetilde{\mathcal{DT}}_{\lambda\mu\nu}$.

\begin{figure}
    \centering
    \includegraphics[width=6cm]{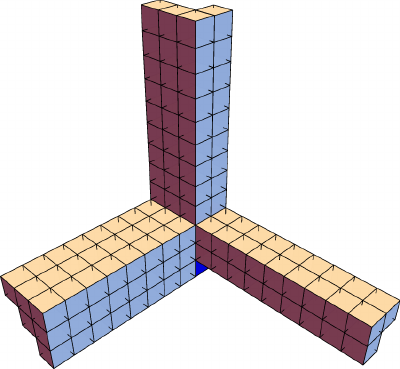}
    \caption{PT box counting hollow structure with the coordinates all reversed.}
    \label{fig:PTreverse}
\end{figure}


\subsection{PT3 counting and \texorpdfstring{$\eta$}{eta} vector}\label{sec:PT3counting-etavector}
The main claim of this section is that changing the reference vector to $\eta=\tilde{\eta}_{0}=(-1,-1,\ldots,-1)$, we obtain the equivariant PT3 vertex \cite{Pandharipande:2007sq,Pandharipande:2007kc}. Namely, the poles classified by the JK residue formalism are the PT3 configurations (in particular $\mathcal{PT}_{\lambda\mu\nu}$) discussed in the previous section. In particular, we will focus on the three examples \eqref{eq:DTexample1}, \eqref{eq:DTexample2}, \eqref{eq:DTexample3} in section~\ref{sec:PToneleg}, \ref{sec:PTtwolegs}, and \ref{sec:PTthreelegs}, respectively.

\begin{definition}\label{def:PTvertex-JKresidue}
The PT3 partition function with boundary conditions is given by
\bea
\,&\mathcal{Z}^{\PT\tbar\JK}_{\bar{4};\lambda\mu\nu}[\mathfrak{q},q_{1,2,3,4}]=\sum_{k=0}^{\infty}\mathfrak{q}^{k}\mathcal{Z}^{\PT\tbar\JK}_{\bar{4};\lambda\mu\nu}[k],
\eea
where
\bea
\,&\mathcal{Z}^{\PT\tbar\JK}_{\bar{4};\lambda\mu\nu}[k]=\frac{1}{k!}\left(\frac{\sh(-\epsilon_{14,24,34})}{\sh(-\epsilon_{1,2,3,4})}\right)^{k}\oint_{\tilde{\eta}_{0}} \prod_{I=1}^{k}\frac{d\phi_{I}}{2\pi i}\prod_{I=1}^{k}\mathcal{Z}^{\D6_{\bar{4}}\tbar\D2\tbar\D0}_{\DT;\lambda\mu\nu}(\fra,\phi_{I})\prod_{I<J}^{k}\mathcal{Z}^{\D0\tbar\D0}(\phi_{I},\phi_{J}).
\eea
The integrand is the same but the reference vector is $\tilde{\eta}_{0}=-\eta_{0}$.
\end{definition}

\begin{proposition}\label{prop:PTJK-pole}
    When at least one of the three-legs is trivial, during the JK-residue procedure, only single order pole appears. When the three-legs are nontrivial, second order poles appear.
\end{proposition}

\begin{theorem}\label{thm:PTvertex-expansion}
    The poles are classified by the PT box configurations $\mathcal{PT}_{\lambda\mu\nu}$ with nontrivial equivariant weights:
    \bea
    \mathcal{Z}^{\PT\tbar\JK}_{\bar{4};\lambda\mu\nu}[\fq,q_{1,2,3,4}]=\sum_{\pi\in \mathcal{PT}_{\lambda\mu\nu} }\fq^{|\pi|}\mathcal{Z}^{\PT\tbar\JK}_{\bar{4};\lambda\mu\nu}[\pi].
    \eea
    
\end{theorem}
 Note here that when counting the size of the configuration $|\pi|$, the unlabeled box is counted as~$2$.

We also have an expression in the GR box configurations $\mathcal{GR}_{\lambda\mu\nu}$ after the identifications \eqref{eq:PT-GR-correspondence}.
\begin{theorem}\label{thm:GRvertex-expansion}
    The poles are classified by the GR box configurations $\mathcal{GR}_{\lambda\mu\nu}$ with nontrivial equivariant weights after the identifications \eqref{eq:PT-GR-correspondence}.
    \bea
    \mathcal{Z}^{\PT\tbar\JK}_{\bar{4};\lambda\mu\nu}[\fq,q_{1,2,3,4}]=\sum_{\pi\in \mathcal{GR}_{\lambda\mu\nu} }\fq^{|\pi|}\mathcal{Z}^{\PT\tbar\JK}_{\bar{4};\lambda\mu\nu}[\pi].
    \eea
    When identifying the GR configurations with the poles, the light boxes will always appear in pair with the heavy boxes. Namely, when $\Lambda$ contains a light box, we have a configuration with a heavy box $\Lambda'$ where other boxes are the same and the GR configuration is understood as this pair $\pi=(\Lambda,\Lambda')$. In other words, we cannot determine the partition function of a configuration containing only the light box, but the appearing partition function correspond to the configuration of the pair of the heavy and light boxes.

\end{theorem}

When performing the JK-residue formalism, as discussed in Example 4 and 5 of section~\ref{sec:JK-residue}, the flag construction provides a way to perform the iterative residue. In the D6-D0 case, the possible ways to perform the iterative residue correspond to how one can stack boxes to form a plane partition. In other words, we can define some \textit{transition rules} between non-zero JK-residue configurations from the flag construction. This procedure is also applicable to the PT counting. We will not encounter nontrivial examples of the transition rules in the main text, but they are discussed in Appendix~\ref{app:sec-PT3vertex-examples}.

\begin{proposition}
    Consider two PT box configurations $\pi_{1},\pi_{2}$ whose corresponding poles are denoted as $\phi_{\pi_1}$ and $\phi_{\pi_2}$. Each of them are the set of the $\eps$-coordinates of the boxes in the two PT configurations and the JK-residue evaluated at $\phi_{\pi_{1,2}}$ are non-zero. Assume that $\phi_{\pi_2}=\phi_{\pi_1}\cup\{\phi_{\ast}\}$ and $|\pi_2|=|\pi_1|+1=k+1$. Namely, the PT configuration $\pi_2$ is a configuration with one box added to $\pi_1$. 

    If we can construct a flag whose $\kappa$-matrix includes the reference vector and the iterative residue can be performed in the order
    \bea
        \underset{\phi_{k+1}=\phi_\ast}{\Res}\,\,
        \underset{\phi_{\pi_1}}{\Res},
    \eea
    then the transition $\pi_1\rightarrow \pi_2$ is allowed. In other words, the box at $\phi_\ast$ is an addable box of the PT configuration $\pi_1$.
    
\end{proposition}

\subsubsection{One-leg}\label{sec:PToneleg}
Let us first use the PT box counting rule in Cond.~\ref{cond:PTrule-oneleg} given in section~\ref{sec:PTrule-coord}. We first start from the minimal plane partition with the boundaries $(\varnothing,\varnothing,\Bbox)$ and then extend it in the negative direction. Removing the boxes from the positive quadrant gives the following hollow structure extending in the negative direction. The possible configurations for low levels are then given as
\bea\label{eq:PToneleg-figure}
\adjustbox{valign=c}{\includegraphics[width=1cm]{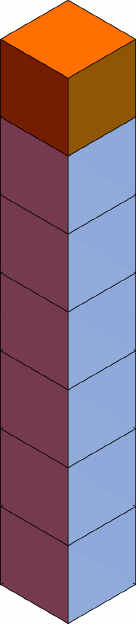}}\qquad \quad \quad \adjustbox{valign=c}{\includegraphics[width=1cm]{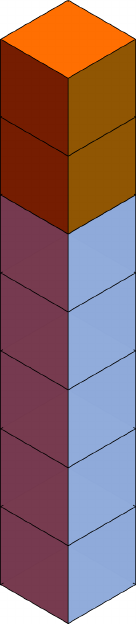}}\qquad \quad \quad \adjustbox{valign=c}{\includegraphics[width=1cm]{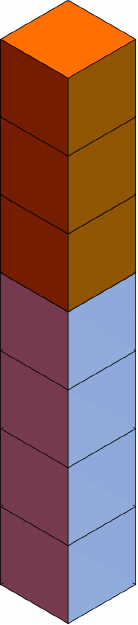}}
\eea
Since we omitted the positive quadrant from the setup, the $\epsilon$-coordinates of the boxes are given as $\fra-\epsilon_{3}, \fra-2\epsilon_{3},\ldots$ and they are decreasing by $\epsilon_{3}$ units.

Let us derive this PT3 counting rule using the DT integral formula \eqref{eq:DTexample1}:
\bea
&\prod_{I=1}^{k}\mathcal{Z}^{\D6_{\bar{4}}\tbar\D2\tbar\D0}_{\DT;\varnothing\varnothing\,\Cbox}(\fra,\phi_{I})\prod_{I< J}\mathcal{Z}^{\D0\tbar\D0}(\phi_{I},\phi_{J})\\
=&\prod_{I=1}^{k}\frac{\sh(\phi_{I}-\fra-\epsilon_{14,24})}{\sh(\phi_{I}-\fra-\epsilon_{1,2})}\frac{\sh(\fra-\phi_{I}-\epsilon_{34})}{\sh(\fra-\phi_{I}-\epsilon_{3})}\prod_{I\neq J}\frac{\sh(\phi_{I}-\phi_{J})\sh(\phi_{I}-\phi_{J}-\epsilon_{14,24,34})}{\sh(\phi_{I}-\phi_{J}-\epsilon_{1,2,3,4})}.
\eea
Choosing $\eta=(-1,\ldots,-1)$, the poles picked up and the corresponding charge vectors are (see Example~5 of section~\ref{sec:JK-residue})
\bea
-e_{I},&\leftrightarrow -\phi_1+\fra-\eps_3=0,\\
e_{I}-e_{J},&\leftrightarrow \phi_1-\phi_J=\eps_{1,2,3,4}.
\eea
Using the Weyl invariance, we may assume that the poles are picked up in the order $\phi_{1},\phi_{2},\ldots,\phi_{k}$. So that the reference vector $\tilde{\eta}_0$ is included in the cone of the charge vectors, the charge vectors need to be chosen in a way such as $\{-e_{1},e_{1}-e_{2},\ldots\}$. Namely, we have $\phi_{I}=\fra-\epsilon_{3}$ and $\phi_{I}=\phi_{J}-\epsilon_{1,2,3,4}$ for $I>J$. 

For level one, the possible pole is
\bea
\phi_1=\fra-\eps_3.
\eea
For the second level, when evaluating the second pole $\phi_{2}$, we have two choices $\phi_{2}=\fra-\epsilon_{3}$ or $\phi_{2}=\phi_{1}-\epsilon_{1,2,3,4}$. For the case $\phi_{2}=\fra-\epsilon_{3}$, the pole is canceled with the numerator $\sh(\phi_{1}-\phi_{2})$ and thus the residue is zero. The poles $\phi_{2}=\phi_{1}-\epsilon_{1,2,4}=\fra-\epsilon_{31,32,34}$ are canceled by $\sh(\phi_{2}-\fra-\epsilon_{14,24})$ and $\sh(-\phi_{2}+\fra-\epsilon_{34})$, respectively after using $\sum_{a\in\four}\epsilon_{a}=0$. Doing this procedure recursively, one can see that the non-zero residues come from the poles
\bea
\phi_{I}=\fra-\epsilon_{3}I,\quad  I=1,\ldots,k.
\eea
The fact that the poles are evaluated in the order $\phi_{1},\ldots,\phi_{k}$ correspond to the PT3 box counting rule explained above.

The JK-residue is obtained as
\bea\label{eq:PToneleg-JK-partfunct}
\mathcal{Z}^{\PT\tbar\JK}_{\bar{4};\varnothing\varnothing\,\Cbox}[k]&=(-1)^{k}\left(\frac{\sh(-\epsilon_{14,24,34})}{\sh(-\epsilon_{1,2,3,4})}\right)^{k}\underset{\phi_{k}=\fra-k\epsilon_{3}}{\Res}\cdots \underset{\phi_{2}=\fra-2\epsilon_{3}}{\Res}\underset{\phi_{1}=\fra-\epsilon_{3}}{\Res}\prod_{I=1}^{k}\mathcal{Z}^{\D6_{\bar{4}}\tbar\D2\tbar\D0}_{\DT;\varnothing\varnothing\,\Cbox}(\fra,\phi_{I})\prod_{I<J}^{k}\mathcal{Z}^{\D0\tbar\D0}(\phi_{I},\phi_{J})\\
&=\prod_{i=1}^{k}\frac{\sh(\epsilon_{4}+i\epsilon_{3})}{\sh(i\epsilon_{3})}=\prod_{i=1}^{k}\frac{[q_{4}q_{3}^{i}]}{[q_{3}^{i}]}
\eea
which gives the PT3 partition function
\bea
\mathcal{Z}^{\PT\tbar\JK}_{\bar{4};\varnothing\varnothing\,
\Cbox}[\mathfrak{q},q_{1,2,3,4}]=\sum_{k=0}^{\infty}\mathfrak{q}^{k}\mathcal{Z}^{\PT\tbar\JK}_{\bar{4};\varnothing\varnothing\,\Cbox}[k]=\sum_{k=0}^{\infty}\mathfrak{q}^{k}\prod_{i=1}^{k}\frac{[q_{4}q_{3}^{i}]}{[q_{3}^{i}]}=\PE\left[\mathfrak{q}\frac{[q_{34}]}{[q_{3}]}\right]
\eea
where in the last equation we used the fact that we have a nice PE formula for this case \cite{Nekrasov:2009JJM}. Note that the overall factor $(-1)^{k}$ comes from the JK-residue computation \eqref{eq:JKresidue-sign}.

\paragraph{General case}
Let us briefly discuss what will happen for the general one-leg case $(\varnothing,\varnothing,\nu)$ (see Appendix~\ref{app:sec-PT3vertex-examples} for explicit examples). For this case, the pole structure of the framing node contribution takes the form as
\bea
\mathcal{Z}^{\D6_{\bar{4}}\tbar\D2\tbar\D0}_{\DT;\varnothing\varnothing\nu}(\fra,\phi_{I})&=\prod_{(i,j)\in A(\nu)} \frac{\sh(\phi_I-\fra-(i-1)\eps_1-(j-1)\eps_2-\eps_4)}{\sh(\phi_I-\fra-(i-1)\eps_1-(j-1)\eps_2)}\\
&\times \prod\limits_{(i,j)\in R(\nu)}\frac{\sh(\fra+i\eps_1+j\eps_2-\phi_I)}{\sh(\fra+(i-1)\eps_1+(j-1)\eps_2-\eps_3-\phi_I)}.
\eea
Taking the reference vector to be $\eta=\eta_0$, the poles picked up at the one-instanton level are
\bea
\phi_1=\fra+(i-1)\eps_1+(j-1)\eps_2,\quad (i,j)\in A(\nu)
\eea
and they correspond to the boxes we can add to the minimal plane partition.

On the other hand, choosing the reference vector to be $\eta=\tilde{\eta}_0$, the poles picked up are 
\bea
\phi_1=\fra-\eps_3+(i-1)\eps_1+(j-1)\eps_2,\quad (i,j)\in R(\nu).
\eea
These poles correspond to the removable boxes of the Young diagram at the $\fra-\eps_3$ plane and they indeed correspond to the poles that one can place at the first level of the PT3 counting.

\subsubsection{Two-legs}\label{sec:PTtwolegs}
Let us explicitly study the case for $(\Bbox,\Bbox,\varnothing)$. Other examples are given in Appendix~\ref{app:sec-PT3vertex-examples}. Again, we start from the minimal plane partition and extend it in the negative direction. Removing boxes from the positive quadrant gives the following hollow structure:
\bea
\adjustbox{valign=c}{\includegraphics[width=5cm]{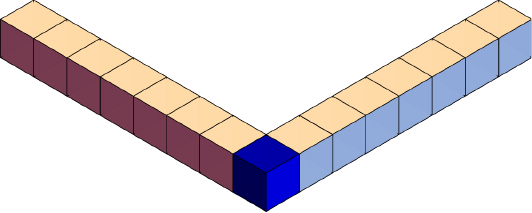}}
\eea
where the blue box has the coordinate $\fra$. Generic PT3 configurations look like
\bea\label{eq:PTtwolegs-figure}
\adjustbox{valign=c}{\includegraphics[width=5cm]{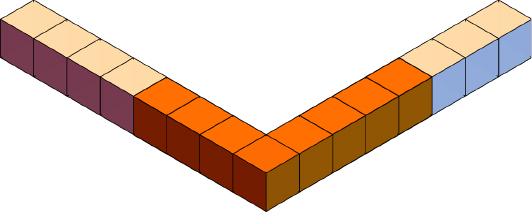}}
\eea
and we can label them by two non-negative integers $[m,n]$ $(m,n\in\mathbb{Z}_{\geq 0})$. The integer $m$ ($n$) counts the number of boxes except the origin box in the first (second) axis. The configuration $[0,0]$ corresponds to the configuration with a box only at the blue position. For example, the configuration above is $[4,3]$.

Let us derive the PT rules using the JK-residue formalism. First, the contour integrand is
\bea
&\prod_{I}\mathcal{Z}^{\D6_{\bar{4}}\tbar\D2\tbar\D0}_{\DT;\,\Cbox\,\Cbox\varnothing}(\fra,\phi_{I})\prod_{I<J}\mathcal{Z}^{\D0\tbar\D0}(\phi_{I},\phi_{J})\\
=&\prod_{I}\frac{\sh(\phi_{I}-\fra-\epsilon_{124})\sh(\phi_{I}-\fra-\epsilon_{34})}{\sh(\phi_{I}-\fra-\epsilon_{12})\sh(\phi_{I}-\fra-\epsilon_{3})}\frac{\sh(\fra-\phi_{I}-\epsilon_{4})}{\sh(\fra-\phi_{I})}\prod_{I\neq J}\frac{\sh(\phi_{I}-\phi_{J})\sh(\phi_{I}-\phi_{J}-\epsilon_{14,24,34})}{\sh(\phi_{I}-\phi_{J}-\epsilon_{1,2,3,4})}.
\eea
Choosing $\eta=(-1,\ldots,-1)$, the poles picked up are
\bea\label{eq:twoleg1box-pole}
\fra-\phi_{I}=0,\quad \phi_{I}-\phi_{J}=+\epsilon_{1,2,3,4}.
\eea
Let us show that the poles are classified as
\bea
\fra, \quad \fra-m\epsilon_{1},\quad \fra-n\epsilon_{2},\quad m,n\in\mathbb{Z}_{\geq 1}.
\eea

For level one, we choose
\bea
\phi_{1}=\fra
\eea
where we used the Weyl invariance again. Let us next consider level two. The possible poles are $\phi_{2}=\fra$ and $\phi_{2}=\phi_{1}-\epsilon_{1,2,3,4}$. The former one is canceled by the numerator $\sh(\phi_{I}-\phi_{J})$. The poles $\phi_{2}=\fra-\epsilon_{3,4}$ are canceled by
\bea
\prod_{I}\sh(\phi_{I}-\fra+\epsilon_{3})\sh(\fra-\phi_{I}-\epsilon_{4})
\eea
and the poles $\phi_{2}=\fra-\epsilon_{1,2}$ remain. For the generic level $k=1+m+n$, assume that we have evaluated the poles at $\{\fra,\fra-i\epsilon_{1},\fra-j\epsilon_{2}\}_{i=1,\ldots,m}^{j=1,\ldots,n}$ and consider the pole at $\phi_{k+1}$. The poles coming from the JK-residue are
\bea
\phi_{k+1}=\fra,\quad \phi_{k+1}=\phi_{I}-\epsilon_{1,2,3,4},\quad (I=1,\ldots,k).
\eea
Again the first pole is canceled by $\sh(\phi_{k+1}-\phi_{1})$. For the second poles, most of them are canceled and it is enough to consider the surface contribution when $\phi_{k+1}=\fra-m\epsilon_{1}-\epsilon_{1,2,3,4}, \fra-n\epsilon_{2}-\epsilon_{1,2,3,4}$. Using the symmetry, we can focus on $\phi_{k+1}=\fra-m\epsilon_{1}-\epsilon_{1,2,3,4}$. If we have a box at $(-m,1,1)$, we also have a box at $(-m+1,1,1)$, and thus for some $J$ the pole is $\phi_{J}=\fra-(m-1)\epsilon_{1}$. Using this, we have
\bea
\phi_{k+1}=\fra-m\epsilon_{1}-\epsilon_{2,3,4}=\phi_{J}-\epsilon_{12,13,14}
\eea
which is canceled by the numerator
\bea
\sh(\phi_{k+1}-\phi_{J}+\epsilon_{12,13,14})\in \prod_{I\neq J}\sh(\phi_{I}-\phi_{J}-\epsilon_{14,24,34}).
\eea
Therefore, the poles are indeed classified as \eqref{eq:twoleg1box-pole}.

The JK-residue for each PT configuration is evaluated as
\bea\label{eq:PTtwoleg-JK-partfunct}
\mathcal{Z}_{\bar{4};\,\Cbox\,\Cbox\varnothing}^{\PT\tbar\JK}[[0,0]]&=-\left(\frac{\sh(-\epsilon_{14,24,34})}{\sh(-\epsilon_{1,2,3,4})}\right)\underset{\phi_{1}=\fra}{\Res}\mathcal{Z}^{\D6_{\bar{4}}\tbar\D2\tbar\D0}_{\DT;\,\Cbox\,\Cbox\varnothing}(\fra,\phi_{1})=-\frac{[q_{12}][q_{13}][q_{23}]}{[q_{1}][q_{2}][q_{3}]}\\
\mathcal{Z}_{\bar{4};\,\Cbox\,\Cbox\varnothing}^{\PT\tbar\JK}[[m,n]]&=(-1)^{m+n+1}\left(\frac{\sh(-\epsilon_{14,24,34})}{\sh(-\epsilon_{1,2,3,4})}\right)^{m+n+1}\underset{\phi=\phi_{[m,n]}}{\Res}\prod_{I}\mathcal{Z}^{\D6_{\bar{4}}\tbar\D2\tbar\D0}_{\DT;\,\Cbox\,\Cbox\varnothing}(\fra,\phi_{I})\prod_{I<J}\mathcal{Z}^{\D0\tbar\D0}(\phi_{I},\phi_{J})\\
&=\frac{[q_{34}]}{[q_{3}]}\times\frac{[q_{3}q_{1}^{m+1}q_{2}^{-n}][q_{4}q_{1}^{m+1}q_{2}^{-n}]}{[q_{1}^{m+1}q_{2}^{-n}][q_{34}q_{1}^{m+1}q_{2}^{-n}]}\times 
\prod_{i=1}^{m}\frac{[q_{4}q_{1}^{i}]}{[q_{1}^{i}]}\prod_{j=1}^{n}\frac{[q_{4}q_{2}^{j}]}{[q_{2}^{j}]}
\eea
where we simply denoted the iterative residue as
\bea
\underset{\phi=\phi_{[m,n]}}{\Res}=\underset{\phi_{m+n+1}=\fra-n\epsilon_{2}}{\Res}\cdots\underset{\phi_{m+1}=\fra-\epsilon_{2}}{\Res}\underset{\phi_{m}=\fra-m\epsilon_{1}}{\Res}\cdots\underset{\phi_{1}=\fra}{\Res}.
\eea
Note that the second equation actually also reproduces the $[0,0]$ configuration. We finally have
\bea
\mathcal{Z}^{\PT\tbar\JK}_{\bar{4};\,
\Cbox\,\Cbox\varnothing}[\mathfrak{q},q_{1,2,3,4}]=\sum_{k=0}^{\infty}\mathfrak{q}^{k}\mathcal{Z}^{\PT\tbar\JK}_{\bar{4};\,
\Cbox\,\Cbox\varnothing}[k]=1+\sum_{m,n=0}^{\infty}\mathfrak{q}^{m+n+1}\mathcal{Z}_{\bar{4};\,\Cbox\,\Cbox\varnothing}^{\PT\tbar\JK}[[m,n]].
\eea

\subsubsection{Three-legs}\label{sec:PTthreelegs}
In this section, we apply the JK-residue formalism to the example \eqref{eq:PTexample3} and discuss the relation with the PT and GR box counting rules. For other examples, see Appendix~\ref{app:sec-PT3vertex-examples}.

\begin{figure}[ht]
    \centering
    \includegraphics[width=6cm]{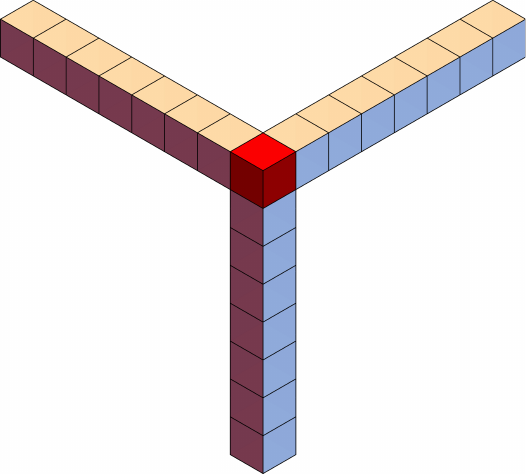}
    \caption{The PT box counting setup for $\lambda=\mu=\nu={1}$.}
    \label{fig:PT3leg-onebox}
\end{figure}

\paragraph{PT box-counting rules}Let us see the possible PT configurations following Cond.~\ref{cond:PTrules}. At level zero, we have the empty configuration $\ket{\text{vac}}$. 

At level one, the origin box is freely labeled and thus we have
\bea
\ket{0^0}
\eea
and it gives a two-states degeneracy.

For the second level, we have four possible configurations:
\bea
\ket{0^{1},-\eps_1},\quad \ket{0^2,-\eps_2},\quad \ket{0^3,-\eps_3},\quad \ket{0^{-1}}.
\eea
For the configurations $\ket{0^{a},-\eps_a}$, since the box at $-\eps_a$ is supported by the origin box, we may label it as $a$ following the third condition of Cond.~\ref{cond:PTrules}. The configuration $\ket{0^{-1}}$ also does not break the PT rules.

For the third level, we have 6 configurations:
\bea
\ket{0^{-1},-\eps_1},\quad \ket{0^{-1},-\eps_2},\quad \ket{0^{-1},-\eps_3},\quad \ket{0^{1},-\eps_1,-2\eps_1},\quad \ket{0^2,-\eps_2,-2\eps_2},\quad \ket{0^3,-\eps_3,-2\eps_3}.
\eea
The possible configurations decompose into two classes of configurations. If we have boxes at three distinct positions $0,\eps_a,-2\eps_a$, since the box at the red position supports boxes in the cylinder $a$, it is labeled $a$. The other class is the configurations with a box with weight 2 in the origin $0^{-1}$ and the remaining box can be placed in any of the three directions.

For the fourth level, we have 9 possible configurations:
\bea
\ket{0^{-1},-\eps_i,-\eps_j}\,(i\neq j),\quad \ket{0^{i},-\eps_i,-2\eps_i,-3\eps_i}\,(i=1,2,3).
\eea
Since the supported boxes extends in two different cylinders, the $0$ for the configuration $\ket{0^{-1},-\eps_i,-\eps_j}$ needs to be unlabeled (labeled by $-1$). For the other configurations, the discussion is similar to the level three and the origin box is labeled depending on the direction the supported boxes extend.

For generic level $k$, the possible configurations are given as
\bea
&\ket{0^{-1},-\epsilon_{1},\ldots,-n_{1}\epsilon_{1},-\epsilon_{2},\ldots,-n_{2}\epsilon_{2},-\epsilon_{3},\ldots,-n_{3}\epsilon_{3}},\quad n_{1}+n_{2}+n_{3}=k-2 \,(n_{i}\geq 0),\\
&\ket{0^{i},-\epsilon_{i},\ldots,-(k-1)\epsilon_{i}}.
\eea

\paragraph{GR box-counting rules}
At level zero, we have the empty configuration $\ket{\text{vac}}$. At level one, we have two possible configurations where we can place a heavy box and a light box at the origin:
\bea
\ket{\text{vac}}\rightarrow \ket{0_{H}}+\ket{0_{L}}.
\eea
This is because the red position admits any labeling.

For the next level, we have the following transitions
\bea
\ket{0_{H}}&\rightarrow \ket{0_{U}}+\ket{0_{H},-\epsilon_{1}}+\ket{0_{H},-\epsilon_{2}}+\ket{0_{H},-\epsilon_{3}}\\
\ket{0_{L}}&\rightarrow \ket{0_{U}}.
\eea
Both the heavy and light boxes can transform into an ultra-heavy box. For the configuration with a heavy box $\ket{0_{H}}$, we need to assign a consistent labeling. For each $\ket{0_{H},-\epsilon_{i}}$, since a box is added to an uncolored position in the cylinder extending in the direction $i$, we can label the heavy box as $i$, which is a consistent configuration. Note that the ultra-heavy box is counted as weight two and thus such configuration appears at level two. 

For level three, we have six states:
\bea
\ket{0_{U}}&\rightarrow \ket{0_{U},-\epsilon_{1}}+\ket{0_{U},-\epsilon_{2}}+\ket{0_{U},-\epsilon_{3}}\\
\ket{0_{H},-\epsilon_{1}}&\rightarrow \ket{0_{H},-\epsilon_{1},-2\epsilon_{1}}+\ket{0_{U},-\epsilon_{1}}\\
\ket{0_{H},-\epsilon_{2}}&\rightarrow \ket{0_{H},-\epsilon_{2},-2\epsilon_{2}}+\ket{0_{U},-\epsilon_{2}}\\
\ket{0_{H},-\epsilon_{3}}&\rightarrow \ket{0_{H},-\epsilon_{3},-2\epsilon_{3}}+\ket{0_{U},-\epsilon_{3}}.
\eea
We can place any boxes on top of the ultra-heavy box. For the configuration $\ket{0_{H},-\epsilon_{i}}$, we can transform the heavy box to an ultra-heavy box or add a box in the cylinder of direction $i$, since the heavy box has the label $i$. Note that we cannot add a box at $-\epsilon_{j}$ for $i\neq j$ since the heavy box at the red position will not have a consistent labeling because there are boxes belonging to different cylinders in two different directions.

Let us also consider the fourth level. In this case, we have nine possible states:
\bea
\ket{0_{U},-\epsilon_{i}}&\rightarrow \ket{0_{U},-\epsilon_{i},-2\epsilon_{i}}+\ket{0_{U},-\epsilon_{i},-\epsilon_{j}}\quad i=1,2,3,\,\,j\neq i\\
\ket{0_{H},-\epsilon_{i},-2\epsilon_{i}}&\rightarrow \ket{0_{H},-\epsilon_{i},-2\epsilon_{i},-3\epsilon_{i}}+\ket{0_{U},-\epsilon_{i},-2\epsilon_{i}}.
\eea
For the configuration $\ket{0_{U},-\epsilon_{i}}$, we can stack boxes in the cylinders as long as they obey the melting rule with gravity in the $(1,1,1)$-direction since the ultra-heavy box supports boxes in all directions. For the configuration $\ket{0_{H},-\epsilon_{i},-2\epsilon_{i}}$, similar to the previous levels, we can change the heavy box to an ultra-heavy box, or add a box in $-3\epsilon_{i}$. Again, we cannot add a box at $-\epsilon_{j}$ for $j\neq i$ because the heavy box will not admit a consistent labeling.

Generally, we have $3+\binom{k}{2}$ possible states for level $k\geq 2$
\bea
&\ket{0_{U},-\epsilon_{1},\ldots,-n_{1}\epsilon_{1},-\epsilon_{2},\ldots,-n_{2}\epsilon_{2},-\epsilon_{3},\ldots,-n_{3}\epsilon_{3}},\quad n_{1}+n_{2}+n_{3}=k-2 \,(n_{i}\geq 0),\\
&\ket{0_{H},-\epsilon_{i},\ldots,-(k-1)\epsilon_{i}}.
\eea

\paragraph{Level one}
Let us see next what will happen for the JK-residue formalism. The contour integrand is
\bea
&\prod_{I}\mathcal{Z}^{\D6_{\bar{4}}\tbar\D2\tbar\D0}_{\DT;\,\Cbox\,\Cbox\,\Cbox}(\fra,\phi_{I})\prod_{I< J}\mathcal{Z}^{\D0\tbar\D0}(\phi_{I},\phi_{J})\\
=&\prod_{I}\frac{\sh(\phi_{I}-\fra-\epsilon_{124,134,234})}{\sh(\phi_{I}-\fra-\epsilon_{12,23,13})}\left(\frac{\sh(\fra-\phi_{I}-\epsilon_{4})}{\sh(\fra-\phi_{I})}\right)^{2}\prod_{I\neq J}\frac{\sh(\phi_{I}-\phi_{J})\sh(\phi_{I}-\phi_{J}-\epsilon_{14,24,34})}{\sh(\phi_{I}-\phi_{J}-\epsilon_{1,2,3,4})}.
\eea
For level one, the integrand is
\bea
\mathcal{Z}^{\D6_{\bar{4}}\tbar\D2\tbar\D0}_{\DT;\,\Cbox\,\Cbox\,\Cbox}(\fra,\phi_{1})=\frac{\sh(\phi_{1}-\fra-\epsilon_{124,134,234})}{\sh(\phi_{1}-\fra-\epsilon_{12,23,13})}\left(\frac{\sh(\fra-\phi_{1}-\epsilon_{4})}{\sh(\fra-\phi_{1})}\right)^{2}.
\eea
Choosing $\eta=(-1,\ldots,-1)$, the pole picked up is
\bea
\phi_{1}=\fra.
\eea
The transition rule is given as
\bea
\boxed{\text{vac}\longrightarrow (0)}
\eea
where we omit the $\fra$ to denote the PT configuration.\footnote{For a general box configuration with coordinates $x,y,z,\ldots$, we denote the configuration as $(x,y,z,\ldots)$ when discussing the JK-formalism. When writing $(x,y,z,\ldots)$, the order of the variables inside is not specified. The transition rules determine how the PT configurations are generated from the vacuum.}

Since this pole is a second-order pole, the residue is evaluated as
\bea
&\underset{\phi_{1}=\fra}{\Res}\mathcal{Z}^{\D6_{\bar{4}}\tbar\D2\tbar\D0}_{\DT;\,\Cbox\,\Cbox\,\Cbox}(\fra,\phi_{1})\\
=&\lim_{\phi_{1}\rightarrow \fra}\frac{\partial }{\partial \phi_{1}}\left((\phi_{1}-\fra)^{2}\frac{\sh(\phi_{1}-\fra-\epsilon_{124,134,234})}{\sh(\phi_{1}-\fra-\epsilon_{12,23,13})}\left(\frac{\sh(\fra-\phi_{1}-\epsilon_{4})}{\sh(\fra-\phi_{1})}\right)^{2}\right)\\
=&\lim_{\phi_{1}\rightarrow \fra}\frac{\partial }{\partial \phi_{1}}\left(\frac{\sh(\fra-\phi_{1}-\epsilon_{4})^{2}\sh(\phi_{1}-\fra-\epsilon_{124,134,234})}{\sh(\phi_{1}-\fra-\epsilon_{12,23,13})}\right)
\eea
where in the second line, we used $\sh(\phi_{1}-\fra)\sim (\phi_{1}-\fra)$ at the limit $\phi_{1}\rightarrow \fra$. Finally we have
\bea\label{eq:3leglevel1}
\mathcal{Z}^{\PT\tbar\JK}_{\bar{4};\,\Cbox\,\Cbox\,\Cbox}[1]&=-\left(\frac{\sh(-\epsilon_{14,24,34})}{\sh(-\epsilon_{1,2,3,4})}\right)\times \underset{\phi_{1}=\fra}{\Res}\mathcal{Z}^{\D6_{\bar{4}}\tbar\D2\tbar\D0}_{\DT;\,\Cbox\,\Cbox\,\Cbox}(\fra,\phi_{1})\\
=&\ch\left(\epsilon _1+\epsilon
   _2+\epsilon _3\right)\\
   &-\frac{1}{2} \sh\left(\epsilon _1+\epsilon _2+\epsilon _3\right) \left(\frac{\ch\left(\epsilon
   _1\right)}{\sh\left(\epsilon _1\right)}+\frac{\ch\left(\epsilon
   _2\right)}{\sh\left(\epsilon _2\right)}+\frac{\ch\left(\epsilon
   _3\right)}{\sh\left(\epsilon _3\right)}+\frac{\ch\left(\epsilon _1+\epsilon
   _2\right)}{\sh\left(\epsilon _1+\epsilon _2\right)}+\frac{\ch\left(\epsilon _1+\epsilon
   _3\right)}{\sh\left(\epsilon _1+\epsilon _3\right)}+\frac{\ch\left(\epsilon _2+\epsilon
   _3\right)}{\sh\left(\epsilon _2+\epsilon _3\right)}\right).
\eea

\paragraph{Level two}
Let us next consider the second level. Again, choosing $\eta=(-1,\ldots,-1)$, the poles picked up come from 
\bea
-\phi_{I}+\fra=0,\quad \phi_{I}-\phi_{J}=+\epsilon_{1,2,3,4}.
\eea
Using the Weyl invariance, let us focus on the following cases 
\bea
(\phi_{1},\phi_{2})&=(\fra,\fra),\quad -\phi_1+\fra=0,\quad -\phi_2+\fra=0,\\
(\phi_1,\phi_2)&=(\fra,\fra-\epsilon_{1,2,3,4}),\quad -\phi_1+\fra=0,\quad -\phi_1+\phi_2=\eps_{1,2,3,4}.
\eea
Note that these poles are all non-degenerate. The corresponding charge vectors are $\{(-1,0),(0,-1)\}$ and $\{(-1,0),(1,-1)\}$, respectively, and the reference vector is contained in the cone constructed from these charge vectors. We thus can perform the iterative residue as
\bea
\underset{\phi_2}{\Res}\,\,\underset{\phi_1}{\Res}.
\eea

Not all pole contribution gives non-zero JK-residues. Let us see which poles remain after evaluating the JK-residue. The residue at $\phi_{1}=\fra$ involves a second order pole:
\bea
&\underset{\phi_{1}=\fra}{\Res}\prod_{I=1}^{2}\mathcal{Z}^{\D6_{\bar{4}}\tbar\D2\tbar\D0}_{\DT;\,\Cbox\,\Cbox\,\Cbox}(\fra,\phi_{I})\mathcal{Z}^{\D0\tbar\D0}(\phi_{1},\phi_{2})\\
=& \lim_{\phi_{1}\rightarrow \fra}\frac{\partial }{\partial \phi_{1}}\left((\phi_{1}-\fra)^{2} \prod_{I=1}^{2}\mathcal{Z}^{\D6_{\bar{4}}\tbar\D2\tbar\D0}_{\DT;\,\Cbox\,\Cbox\,\Cbox}(\fra,\phi_{I})\mathcal{Z}^{\D0\tbar\D0}(\phi_{1},\phi_{2})\right)
\eea

We first start from the pole $(\phi_{1},\phi_{2})=(\fra,\fra)$. In usual instanton computations as the two-legs case, this contribution disappears because of the numerator $\prod_{I\neq J}\sh(\phi_{I}-\phi_{J})$. One would expect that a similar cancellation happens here but this is not the case. The contour integrand can be schematically written as
\bea\label{eq:3leg-level2-1-def}
\prod_{I=1}^{2}\mathcal{Z}^{\D6_{\bar{4}}\tbar\D2\tbar\D0}_{\DT;\,\Cbox\,\Cbox\,\Cbox}(\fra,\phi_{I})\mathcal{Z}^{\D0\tbar\D0}(\phi_{1},\phi_{2})=\frac{\sh(\phi_{1}-\phi_{2})^{2}f_{0}(\phi_{1},\phi_{2})}{\sh(\phi_{1}-\fra)^{2}\sh(\phi_{2}-\fra)^{2}}
\eea
where $f_{0}(\phi_{1},\phi_{2})$ is a function with no zeros nor poles at $\phi_{1}=\phi_{2}=\fra$. The residue is then computed as
\bea\label{eq:3leg-level2-1}
\underset{\phi_{1}=\fra}{\Res}\prod_{I=1}^{2}\mathcal{Z}^{\D6_{\bar{4}}\tbar\D2\tbar\D0}_{\DT;\,\Cbox\,\Cbox\,\Cbox}(\fra,\phi_{I})\mathcal{Z}^{\D0\tbar\D0}(\phi_{1},\phi_{2})&=\left.\partial_{\phi_{1}}\left(\frac{\sh(\phi_{1}-\phi_{2})^{2}f_{0}(\phi_{1},\phi_{2})}{\sh(\phi_{2}-\fra)^{2}}\right)\right|_{\phi_{1}=\fra}\\
&=\partial_{\phi_{1}}f_{0}(\phi_{1},\phi_{2})|_{\phi_{1}=\fra}-\frac{f_{0}(\fra,\phi_{2})\ch(\fra-\phi_{2})}{\sh(\phi_{2}-\fra)}\\
&\xrightarrow{\phi_{2}\simeq \fra} -\frac{2f_{0}(\fra,\fra)}{\sh(\phi_{2}-\fra)}
\eea
where we used $\sh(x)'=\ch(x)/2$ and $\ch(0)=2$. Obviously, we still have a pole at $\phi_{2}=\fra$ and thus the residue is non-zero.

Let us next consider the poles at $(\phi_{1},\phi_{2})=(\fra,\fra-\epsilon_{1,2,3})$:
\bea\label{eq:3leg-level2-2-def}
&\prod_{I=1}^{2}\mathcal{Z}^{\D6_{\bar{4}}\tbar\D2\tbar\D0}_{\DT;\,\Cbox\,\Cbox\,\Cbox}(\fra,\phi_{I})\mathcal{Z}^{\D0\tbar\D0}(\phi_{1},\phi_{2})\\
=&\frac{f_{1}(\phi_{1},\phi_{2})}{\sh(\phi_{1}-\fra)^{2}}\frac{\sh(\phi_{2}-\fra+\epsilon_{1})}{\sh(\phi_{2}-\phi_{1}+\epsilon_{1})}=\frac{f_{2}(\phi_{1},\phi_{2})}{\sh(\phi_{1}-\fra)^{2}}\frac{\sh(\phi_{2}-\fra+\epsilon_{2})}{\sh(\phi_{2}-\phi_{1}+\epsilon_{2})}=\frac{f_{3}(\phi_{1},\phi_{2})}{\sh(\phi_{1}-\fra)^{2}}\frac{\sh(\phi_{2}-\fra+\epsilon_{3})}{\sh(\phi_{2}-\phi_{1}+\epsilon_{3})}
\eea
where $f_{1,2,3}(\phi_{1},\phi_{2})$ are functions with no zeros nor poles at $\phi_{2}=\fra-\epsilon_{1,2,3}$, respectively. Using the triality, let us focus on the $\epsilon_{3}$ case. Looking at the integrand, since there is $\sh(\phi_{2}-\fra+\epsilon_{3})$, one might expect that the residue at $\phi_{2}=\fra-\epsilon_{3}$ vanishes. Again, this is not the case because when evaluating the residue at $\phi_{1}=\fra$, derivatives of the integrand appears:
\bea\label{eq:3leg-level2-2}
\underset{\phi_{1}=\fra}{\Res}\prod_{I=1}^{2}\mathcal{Z}^{\D6_{\bar{4}}\tbar\D2\tbar\D0}_{\DT;\,\Cbox\,\Cbox\,\Cbox}(\fra,\phi_{I})\mathcal{Z}^{\D0\tbar\D0}(\phi_{1},\phi_{2})&=\left.\partial_{\phi_{1}}\left(f_{3}(\phi_{1},\phi_{2})\frac{\sh(\phi_{2}-\fra+\epsilon_{3})}{\sh(\phi_{2}-\phi_{1}+\epsilon_{3})}\right)\right|_{\phi_{1}=\fra}\\
&=\frac{f_{3}(\fra,\phi_{2})}{\sh(\phi_{2}-\fra+\epsilon_{3})}+\partial_{\phi_{1}}f_{3}(\fra,\phi_{2}).
\eea
We still have a single pole appearing in the denominator and thus this pole gives a non-zero residue.

Finally, let us consider the poles at $(\phi_{1},\phi_{2})=(\fra,\fra-\epsilon_{4})$:
\bea
&\prod_{I=1}^{2}\mathcal{Z}^{\D6_{\bar{4}}\tbar\D2\tbar\D0}_{\DT;\,\Cbox\,\Cbox\,\Cbox}(\fra,\phi_{I})\mathcal{Z}^{\D0\tbar\D0}(\phi_{1},\phi_{2})=\frac{f_{4}(\phi_{1},\phi_{2})\sh(\phi_{2}-\fra+\epsilon_{4})^{2}}{\sh(\phi_{1}-\fra)^{2}\sh(-\phi_{1}+\phi_{2}+\epsilon_{4})}
\eea
where $f_{4}(\phi_{1},\phi_{2})$ is a function with no zeros nor poles. The residue at $\phi_{1}=\fra$ is 
\bea
\underset{\phi_{1}=\fra}{\Res}\prod_{I=1}^{2}\mathcal{Z}^{\D6_{\bar{4}}\tbar\D2\tbar\D0}_{\DT;\,\Cbox\,\Cbox\,\Cbox}(\fra,\phi_{I})\mathcal{Z}^{\D0\tbar\D0}(\phi_{1},\phi_{2})=\sh(\phi_{2}-\fra+\epsilon_{4})\partial_{\phi_{1}}f_{4}(\fra,\phi_2)+f_{4}(\fra,\phi_{2}).
\eea
Thus, this gives no nontrivial residue. This situation is similar to other cases where the crystal structure does not grow in the fourth direction.

The transition rules can be deduced as
\bea
\boxed{(0)\longrightarrow (0,0)+(0,-\eps_1)+(0,-\eps_2)+(0,-\eps_3)}
\eea

Denoting the residue of the contour integrand without the Weyl group factor as
\bea
\mathcal{Z}_{(x,y)}\coloneqq\underset{\phi_{2}=\fra+y}{\Res}\,\,\underset{\phi_{1}=\fra+x}{\Res}\left(\frac{\sh(-\epsilon_{14,24,34})}{\sh(-\epsilon_{1,2,3,4})}\right)^{2}\prod_{I=1}^{2}\mathcal{Z}^{\D6_{\bar{4}}\tbar\D2\tbar\D0}_{\DT;\,\Cbox\,\Cbox\,\Cbox}(\fra,\phi_{I})\mathcal{Z}^{\D0\tbar\D0}(\phi_{1},\phi_{2})
\eea
the PT partition function at two-instanton level is
\bea\label{eq:3leglevel2}
\mathcal{Z}_{\bar{4};\,\Cbox\,\Cbox\,\Cbox}^{\PT\tbar\JK}[2]&=\frac{1}{2}\left(\mathcal{Z}_{(0,0)}+2\sum_{i=1}^{3}\mathcal{Z}_{(0,-\epsilon_{i})}\right),\\
\frac{1}{2}\mathcal{Z}_{(0,0)}&=\left(\frac{\sh\left(\epsilon _1+\epsilon _2\right) \sh\left(\epsilon _1+\epsilon _3\right)
   \sh\left(\epsilon _2+\epsilon _3\right)}{\sh\left(\epsilon _1\right)
   \sh\left(\epsilon _2\right) \sh\left(\epsilon _3\right)}\right)^{2},\\
\mathcal{Z}_{(0,-\epsilon_{1})}&=-\frac{\sh\left(-\epsilon _1+\epsilon _2+\epsilon _3\right) \sh\left(2 \epsilon _1+\epsilon
   _2+\epsilon _3\right)}{ \sh\left(\epsilon _1\right) \sh\left(2 \epsilon _1\right)},\\
\mathcal{Z}_{(0,-\epsilon_{2})}&=-\frac{\sh\left(\epsilon _1-\epsilon _2+\epsilon _3\right) \sh\left(\epsilon _1+2 \epsilon
   _2+\epsilon _3\right)}{ \sh\left(\epsilon _2\right) \sh\left(2 \epsilon _2\right)},\\
\mathcal{Z}_{(0,-\epsilon_{3})}&=\frac{\sh\left(-\epsilon _1-\epsilon _2+\epsilon _3\right) \sh\left(\epsilon _1+\epsilon _2+2
   \epsilon _3\right)}{ \sh\left(\epsilon _3\right) \sh\left(2 \epsilon _3\right)}.
\eea
The overall factor is the Weyl group factor $1/2$. For the configurations coming from the poles $(\fra,\fra-\epsilon_{1,2,3})$, we have two situations in taking the residue $(\phi_{1},\phi_{2})=(\fra,\fra-\epsilon_{i}),\,(\fra-\epsilon_{i},\fra)$ and the Weyl group factor is canceled out by such multiplicity.

\paragraph{Level three}
For the next level, let us assume $(\phi_{1},\phi_{2})=(\fra,\fra),(\fra,\fra-\epsilon_{1,2,3})$ and evaluate the pole at $\phi_{3}$. For the case $(\phi_{1},\phi_{2})=(\fra,\fra)$, using \eqref{eq:3leg-level2-1-def} and \eqref{eq:3leg-level2-1}, we have
\bea
&\underset{\phi_{2}=\fra}{\Res}\,\underset{\phi_{1}=\fra}{\Res}\prod_{I=1}^{3}\mathcal{Z}^{\D6_{\bar{4}}\tbar\D2\tbar\D0}_{\DT;\,\Cbox\,\Cbox\,\Cbox}(\fra,\phi_{I})\prod_{I<J}\mathcal{Z}^{\D0\tbar\D0}(\phi_{I},\phi_{J})\\
=&\underset{\phi_{2}=\fra}{\Res}\left(\mathcal{Z}^{\D6_{\bar{4}}\tbar\D2\tbar\D0}_{\DT;\,\Cbox\,\Cbox\,\Cbox}(\fra,\phi_{3})\mathcal{Z}^{\D0\tbar\D0}(\phi_{2},\phi_{3})\partial_{\phi_{1}}\left(f_{0}(\phi_{1},\phi_{2})\mathcal{Z}^{\D0\tbar\D0}(\phi_{1},\phi_{3})\right)\right.\\
-&\left.\left.\frac{f_{0}(\fra,\phi_{2})\ch(\fra-\phi_{2})}{\sh(\phi_{2}-\fra)}\mathcal{Z}^{\D6_{\bar{4}}\tbar\D2\tbar\D0}_{\DT;\,\Cbox\,\Cbox\,\Cbox}(\fra,\phi_{3}) \mathcal{Z}^{\D0\tbar\D0}(\phi_{1},\phi_{3})\mathcal{Z}^{\D0\tbar\D0}(\phi_{2},\phi_{3})
  \right)\right|_{\phi_{1}=\fra}\\
  =&-2f_{0}(\fra,\fra)\mathcal{Z}^{\D6_{\bar{4}}\tbar\D2\tbar\D0}_{\DT;\,\Cbox\,\Cbox\,\Cbox}(\fra,\phi_{3})\mathcal{Z}^{\D0\tbar\D0}(\fra,\phi_{3})^{2}.
\eea
Using
\bea
\mathcal{Z}^{\D6_{\bar{4}}\tbar\D2\tbar\D0}_{\DT;\,\Cbox\,\Cbox\,\Cbox}(\fra,\phi_{3})\mathcal{Z}^{\D0\tbar\D0}(\fra,\phi_{3})^{2}=\left(\frac{\sh(\phi_{3}-\fra)\sh(\phi_{3}-\fra-\epsilon_{14,24,34})}{\sh(\phi_{3}-\fra-\epsilon_{1,2,3,4})}\right)^{2}\frac{\sh(\fra-\phi_{3}-\epsilon_{14,24,34})}{\sh(\fra-\phi_{3}-\epsilon_{1,2,3})}
\eea
the poles with non-zero residues are 
\bea
(\fra,\fra,\fra-\epsilon_{1,2,3}).
\eea

Using the triality, we can focus on $(\phi_{1},\phi_{2})=(\fra,\fra-\epsilon_{3})$. From JK-formalism, the candidates are
\bea
\phi_{3}=\phi_{2}-\epsilon_{1,2,3,4}=\fra-2\epsilon_{3},\fra-\epsilon_{31,32,34}\quad \text{or} \quad \phi_3=\fra.
\eea
From \eqref{eq:3leg-level2-2-def} and \eqref{eq:3leg-level2-2} and attributing
\bea
f_{3}(\phi_{1},\phi_{2})\longrightarrow f_{3}(\phi_{1},\phi_{2})\mathcal{Z}^{\D6_{\bar{4}}\tbar\D2\tbar\D0}_{\DT;\,\Cbox\,\Cbox\,\Cbox}(\fra,\phi_{3})\mathcal{Z}^{\D0\tbar\D0}(\phi_{1},\phi_{3})\mathcal{Z}^{\D0\tbar\D0}(\phi_{2},\phi_{3})
\eea
we have
\bea
&\underset{\phi_{2}=\fra-\epsilon_{3}}{\Res}\,\underset{\phi_{1}=\fra}{\Res}\prod_{I=1}^{3}\mathcal{Z}^{\D6_{\bar{4}}\tbar\D2\tbar\D0}_{\DT;\,\Cbox\,\Cbox\,\Cbox}(\fra,\phi_{I})\prod_{I<J}\mathcal{Z}^{\D0\tbar\D0}(\phi_{I},\phi_{J})\\
=&f_{3}(\fra,\fra-\epsilon_{3})\mathcal{Z}^{\D6_{\bar{4}}\tbar\D2\tbar\D0}_{\DT;\,\Cbox\,\Cbox\,\Cbox}(\fra,\phi_{3})\mathcal{Z}^{\D0\tbar\D0}(\fra,\phi_{3})\mathcal{Z}^{\D0\tbar\D0}(\fra-\epsilon_{3},\phi_{3}).
\eea
The poles $\phi_{3}=\fra-\epsilon_{31,32,34}$ all cancel with the numerator of $\mathcal{Z}^{\D0\tbar\D0}(\phi_{1},\phi_{3})|_{\phi_{1}=\fra}$, while the pole $\phi_{3}=\fra-2\epsilon_{3}$ is non-zero. Moreover, the pole at $\phi_3=\fra$ also gives non-zero contribution.

Summarizing, for the level three, the non-zero contributions come from the poles
\bea
(\fra,\fra,\fra-\epsilon_{i}),\quad (\fra,\fra-\epsilon_{i},\fra-2\epsilon_{i}),\quad i=1,2,3.
\eea
The transition rules from level 2 to level 3 are
    \begin{empheq}[box=\fbox]{align}
\begin{split}
(0,0)&\longrightarrow (0,0,-\eps_1)+(0,0,-\eps_2)+(0,0,-\eps_3)\\
(0,-\eps_1)&\longrightarrow (0,-\eps_1,-2\eps_1)+(0,0,-\eps_1)\\
(0,-\eps_2)&\longrightarrow (0,-\eps_2,-2\eps_2)+(0,0,-\eps_2)\\
(0,-\eps_3)&\longrightarrow (0,-\eps_3,-2\eps_3)+(0,0,-\eps_3)
\end{split}
\end{empheq}
Note that since the iterative residue can also be performed in the order $\underset{\phi_3=\fra}{\Res}\,\underset{\phi_2=\fra-\eps_i}{\Res}\,\underset{\phi_1=\fra}{\Res}$, the transition $(0,-\eps_i)\rightarrow (0,0,-\eps_i)$ is also allowed.

Denoting the residue as
\bea
\mathcal{Z}_{(x,y,z)}\coloneqq (-1)^{3}\underset{\phi_{3}=\fra+z}{\Res}\,\underset{\phi_{2}=\fra+y}{\Res}\,\,\underset{\phi_{1}=\fra+x}{\Res}\left(\frac{\sh(-\epsilon_{14,24,34})}{\sh(-\epsilon_{1,2,3,4})}\right)^{3}\prod_{I=1}^{3}\mathcal{Z}^{\D6_{\bar{4}}\tbar\D2\tbar\D0}_{\DT;\,\Cbox\,\Cbox\,\Cbox}(\fra,\phi_{I})\prod_{I<J}\mathcal{Z}^{\D0\tbar\D0}(\phi_{I},\phi_{J})
\eea
we have
\bea\label{eq:3leglevel3}
\mathcal{Z}_{\bar{4};\,\Cbox\,\Cbox\,\Cbox}^{\PT\tbar\JK}[3]&=\frac{1}{6}\left(3\sum_{i=1}^{3}\mathcal{Z}_{(0,0,-\epsilon_{i})}+6\sum_{i=1}^{3}\mathcal{Z}_{(0,-\epsilon_{i},-2\epsilon_{i})}   \right),\\
\frac{1}{2}\mathcal{Z}_{(0,0,-\epsilon_{1})}&= -\frac{\sh\left(\epsilon _1+\epsilon _2\right) \sh\left(2 \epsilon _1+\epsilon _2\right)
   \sh\left(\epsilon _1+\epsilon _3\right) \sh\left(2 \epsilon _1+\epsilon _3\right)
   \sh\left(\epsilon _2+\epsilon _3\right) \sh\left(-\epsilon _1+\epsilon _2+\epsilon
   _3\right)^2}{ \sh\left(\epsilon _1\right) \sh\left(2 \epsilon _1\right)^2
   \sh\left(\epsilon _2\right) \sh\left(\epsilon _2-\epsilon _1\right) \sh\left(\epsilon
   _3\right) \sh\left(\epsilon _3-\epsilon _1\right)},\\
   \frac{1}{2}\mathcal{Z}_{(0,0,-\epsilon_{2})}&= \frac{\sh\left(\epsilon _1+\epsilon _2\right) \sh\left(\epsilon _1+2 \epsilon _2\right)
   \sh\left(\epsilon _1+\epsilon _3\right) \sh\left(\epsilon _1-\epsilon _2+\epsilon
   _3\right)^2 \sh\left(\epsilon _2+\epsilon _3\right) \sh\left(2 \epsilon _2+\epsilon
   _3\right)}{ \sh\left(\epsilon _1\right) \sh\left(\epsilon _2\right) \sh\left(2
   \epsilon _2\right)^2 \sh\left(\epsilon _2-\epsilon _1\right) \sh\left(\epsilon _3\right)
   \sh\left(\epsilon _3-\epsilon _2\right)},\\
   \frac{1}{2}\mathcal{Z}_{(0,0,-\epsilon_{3})}&=-\frac{\sh\left(\epsilon _1+\epsilon _2\right) \sh\left(\epsilon _1+\epsilon _3\right)
   \sh\left(-\epsilon _1-\epsilon _2+\epsilon _3\right)^2 \sh\left(\epsilon _2+\epsilon
   _3\right) \sh\left(\epsilon _1+2 \epsilon _3\right) \sh\left(\epsilon _2+2 \epsilon
   _3\right)}{ \sh\left(\epsilon _1\right) \sh\left(\epsilon _2\right)
   \sh\left(\epsilon _3\right) \sh\left(2 \epsilon _3\right)^2 \sh\left(\epsilon
   _3-\epsilon _1\right) \sh\left(\epsilon _3-\epsilon _2\right)},\\
   \mathcal{Z}_{(0,-\epsilon_{1},-2\epsilon_{1})}&=\frac{\sh\left(-2 \epsilon _1+\epsilon _2+\epsilon _3\right) \sh\left(-\epsilon _1+\epsilon
   _2+\epsilon _3\right) \sh\left(3 \epsilon _1+\epsilon _2+\epsilon _3\right)}{ \sh\left(2
   \epsilon _1\right)^2 \sh\left(3 \epsilon _1\right)},\\
   \mathcal{Z}_{(0,-\epsilon_{2},-2\epsilon_{2})}&=\frac{\sh\left(\epsilon _1-2 \epsilon _2+\epsilon _3\right) \sh\left(\epsilon _1-\epsilon
   _2+\epsilon _3\right) \sh\left(\epsilon _1+3 \epsilon _2+\epsilon _3\right)}{ \sh\left(2
   \epsilon _2\right)^2 \sh\left(3 \epsilon _2\right)},\\
   \mathcal{Z}_{(0,-\epsilon_{3},-2\epsilon_{3})}&=\frac{\sh\left(-\epsilon _1-\epsilon _2+\epsilon _3\right) \sh\left(-\epsilon _1-\epsilon _2+2
   \epsilon _3\right) \sh\left(\epsilon _1+\epsilon _2+3 \epsilon _3\right)}{ \sh\left(2
   \epsilon _3\right)^2 \sh\left(3 \epsilon _3\right)}.
\eea

\paragraph{Generic level $k\geq 2$}
Two classes of poles only remain:
\bea
(0,-\epsilon_{i},\ldots, -(k-1)\epsilon_{i})
\eea
and
\bea
(0,0,-\epsilon_{1},\ldots,-n_{1}\epsilon_{1},-\epsilon_{2},\ldots,-n_{2}\epsilon_{2},\epsilon_{3},\ldots,-n_{3}\epsilon_{3}),\quad n_{1}+n_{2}+n_{3}=k-2.
\eea
We thus have $\binom{k}{2}+3$ possible configurations. The transition rules can be obtained similarly but we omit the discussion. The JK-residue schematically takes the form as
\bea\label{eq:3leglevelgeneric}
\mathcal{Z}^{\PT\tbar\JK}_{\bar{4};\,\Cbox\,\Cbox\,\Cbox}[k]&=\frac{1}{k!}\left(k!\sum_{i=1}^{3}\mathcal{Z}_{(0,-\epsilon_{i},\ldots,-(k-1)\epsilon_{i})}\right.\\
&\left.+\frac{k!}{2}\sum_{n_{1}+n_{2}+n_{3}=k-2}\mathcal{Z}_{(0,0,-\epsilon_{1},\ldots,-n_{1}\epsilon_{1},-\epsilon_{2},\ldots,-n_{2}\epsilon_{2},\epsilon_{3},\ldots,-n_{3}\epsilon_{3})}\right),
\eea
where each term is understood as the JK-residue of the contour integrand without the Weyl group factor for each configuration. The multiplicity in front of each terms comes from the Weyl invariance. The configuration without the ultra heavy box has the multiplicity $k!$ because of the reordering of the variables $\phi_{1},\ldots,\phi_{k}$. For the configuration with the ultra heavy box, since two out of the $k$ variables needs to pick up the same pole, we have $k!/2$ multiplicity.

\paragraph{Comparison with PT and GR rules} The ultra heavy box at the origin $\ket{0_{U}}$ corresponds to $\ket{0^{-1}}$ in the PT box counting rules, and $(0,0)$ in the JK-formalism. The reason why such ultra heavy box is counted twice in the size is because we are evaluating the residue at the same position twice. 

The two-state degeneracy coming from the heavy box and light box $\ket{0_{H}}+\ket{0_{L}}$, correspond to $0^{0}$ in the PT box counting rules and the pole at $0$ in the JK-residue formalism. One cannot see the two-state degeneracy by just looking at the poles. To see such two-state degeneracy, one needs to take the unrefined limit $\eps_4\rightarrow 0$ ($q_{4}\rightarrow 1$) and see the limit of the residue. For example, for the one-instanton level, we have
\bea
\mathcal{Z}_{(0)}\xrightarrow{q_{4}\rightarrow 1} 2.
\eea
The factor $2$ here corresponds to the two-state degeneracy. See section~\ref{sec:unrefined-vertex} for a detailed discussion. Moreover, in the JK-residue formalism, we cannot decompose the partition function into the contributions of the light and heavy boxes:
\bea
\mathcal{Z}_{(0)}=\mathcal{Z}[0^{0}]=\mathcal{Z}[\ket{0_L}+\ket{0_H}].
\eea
In other words, in the JK-formalism, the light box will always appear in pairs with the heavy box just as in the PT box counting rules.

We note also that since in the JK-residue formalism, the light boxes always appear in pairs with the heavy boxes and the transition rules connect different configurations associated with the poles, the GR transition rules can not be directly observed. In particular, transition  rules such as the generation of light boxes cannot be seen. Instead the JK transition rules are obtained by the GR transition rules after projecting out some of the light boxes. Such kind of phenomenon occurs also for other examples (see Appendix~\ref{app:sec-PT3vertex-examples}). For the moment, we do not know how to derive the GR transition rules directly from the JK-residue formalism and we leave it for future work.

\subsection{Conjugate map and PT3 counting}\label{sec:PT3counting-conjugate}
In this section, we propose a different way to evaluate PT3 partition functions. For each fundamental chiral superfield $\Phi_{\text{f}}$ with the flavor charge $q(\Phi_{\text{f}})=q_{1}^{i-1}q_{2}^{j-1}q_{3}^{k-1}=e^{\epsilon(\Phi_{\text{f}})}$, the contribution to the Witten index is
\bea
\frac{\sh(\phi_{I}-\fra-\epsilon(\Phi_{\text{f}})-\epsilon_{4})}{\sh(\phi_{I}-\fra-\epsilon(\Phi_{\text{f}}))}=\mathcal{Z}^{\D6_{\bar{4}}\tbar\D0}(\fra+\epsilon(\Phi_{\text{f}}),\phi_{I}).
\eea

We introduce a conjugation map which flips the fundamental to anti-fundamental while keeping the flavor charges:
\bea
\frac{\sh(\fra-\phi_{I}-\epsilon(\Phi_{\text{f}})-\epsilon_{4})}{\sh(\fra-\phi_{I}-\epsilon(\Phi_{\text{f}}))}=\mathcal{Z}^{\overline{\D6}_{\bar{4}}\tbar\D0}(\fra-\epsilon(\Phi_{\text{f}})-\epsilon_{4},\phi_{I}),
\eea
where we used \eqref{eq:D6antidef}. Oppositely, the anti-fundamental chiral superfield $\Phi_{\text{af}}$ with the flavor charge $q(\phi_{\text{af}})=e^{\epsilon(\Phi_{\text{af}})}$ transforms as
\bea
\,&\frac{\sh(\fra-\phi_{I}-\epsilon(\Phi_{\text{af}})-\epsilon_{4})}{\sh(\fra-\phi_{I}-\epsilon(\Phi_{\text{af}}))}=\mathcal{Z}^{\overline{\D6}_{\bar{4}}\tbar\D0}(\fra-\epsilon(\Phi_{\text{af}})-\epsilon_{4},\phi_{I})\\
&\longrightarrow \frac{\sh(\phi_{I}-\fra-\epsilon(\Phi_{\text{f}})-\epsilon_{4})}{\sh(\phi_{I}-\fra-\epsilon(\Phi_{\text{f}}))}=\mathcal{Z}^{\D6_{\bar{4}}\tbar\D0}(\fra+\epsilon(\Phi_{\text{f}}),\phi_{I}).
\eea

Under this conjugation map, the contour integrand \eqref{eq:DT3contourint} becomes
\bea
\mathcal{Z}^{\D6_{\bar{4}}\tbar\D2\tbar\D0}_{\DT;\lambda\mu\nu}(\fra,\phi_{I})\longrightarrow \mathcal{Z}^{\D6_{\bar{4}}\tbar\D2\tbar\D0}_{\PT;\lambda\mu\nu}(\fra,\phi_{I})
\eea
where
\bea\label{eq:PTflavornode-def}
\mathcal{Z}^{\D6_{\bar{4}}\tbar\D2\tbar\D0}_{\PT;\lambda\mu\nu}(\fra,\phi_{I})&=\prod_{\scube\in s(\vec{Y})}\mathcal{Z}^{\overline{\D6}_{\bar{4}}\tbar\D0}(\fra-c_{\bar{4},0}(\cube)-\epsilon_{4},\phi_{I})\prod_{\scube\in p_{1}(\vec{Y})}\mathcal{Z}^{\D6_{\bar{4}}\tbar\D0}(\fra-c_{\bar{4},0}(\cube)-\epsilon_{4},\phi_{I})\\
&\times \prod_{\scube\in p_{2}(\vec{Y})}\mathcal{Z}^{\D6_{\bar{4}}\tbar\D0}(\fra-c_{\bar{4},0}(\cube)-\epsilon_{4},\phi_{I})^{2}.
\eea
For example, for the three examples \eqref{eq:DTexample1}, \eqref{eq:DTexample2}, and \eqref{eq:DTexample3}, we have
\begin{align}
\mathcal{Z}^{\D6_{\bar{4}}\tbar\D2\tbar\D0}_{\PT;\varnothing\varnothing\,\Cbox}(\fra,\phi_{I})&=\frac{\sh(\phi_{I}-\fra+\epsilon_{12})}{\sh(\phi_{I}-\fra-\epsilon_{3})}\frac{\sh(-\phi_{I}+\fra-\epsilon_{14,24})}{\sh(-\phi_{I}+\fra-\epsilon_{1,2})},\label{eq:PTexample1}\\
\mathcal{Z}^{\D6_{\bar{4}}\tbar\D2\tbar\D0}_{\PT;\,\Cbox\,\Cbox\varnothing}(\fra,\phi_{I})&=\frac{\sh(\phi_{I}-\fra-\epsilon_{4})\sh(-\phi_{I}+\fra+\epsilon_{3})\sh(-\phi_{I}+\fra-\epsilon_{34})}{\sh(\phi_{I}-\fra)\sh(-\phi_{I}+\fra+\epsilon_{34})\sh(-\phi_{I}+\fra-\epsilon_{3})},\label{eq:PTexample2}\\
\mathcal{Z}^{\D6_{\bar{4}}\tbar\D2\tbar\D0}_{\PT;\,\Cbox\,\Cbox\,\Cbox}(\fra,\phi_{I})&=\frac{\sh(\phi_{I}-\fra-\epsilon_{4})^{2}}{\sh(\phi_{I}-\fra)^{2}}\prod_{i=1}^{3}\frac{\sh(-\phi_{I}+\fra+\epsilon_{i})}{\sh(-\phi_{I}+\fra+\epsilon_{4}+\epsilon_{i})}.\label{eq:PTexample3}
\end{align}

The contour integral formulas for the PT invariants are then given as follows.
\begin{theorem}\label{thm:conjugate-PT3vertex}
    The PT3 partition function can be also computed by
    \bea
    \mathcal{Z}_{\bar{4};\lambda\mu\nu}^{\PT\tbar\JK}[\mathfrak{q},q_{1,2,3,4}]=\sum_{k=0}^{\infty}\mathfrak{q}^{k}\mathcal{Z}_{\bar{4};\lambda\mu\nu}^{\PT\tbar\JK}[k],
    \eea
    where
    \bea
\,&\mathcal{Z}^{\PT\tbar\JK}_{\bar{4};\lambda\mu\nu}[k]=\frac{1}{k!}\left(\frac{\sh(-\epsilon_{14,24,34})}{\sh(-\epsilon_{1,2,3,4})}\right)^{k}\oint_{\eta_{0}} \prod_{I=1}^{k}\frac{d\phi_{I}}{2\pi i}\prod_{I=1}^{k}\mathcal{Z}^{\D6_{\bar{4}}\tbar\D2\tbar\D0}_{\PT;\lambda\mu\nu}(\fra,\phi_{I})\prod_{I<J}^{k}\mathcal{Z}^{\D0\tbar\D0}(\phi_{I},\phi_{J})
\eea
where the reference vector is the conventional $\eta_{0}$.
\end{theorem}

This time, the poles are classified by the PT3 configurations but the coordinates are assigned as Fig.~\ref{fig:PTreverse}. For example, for the case \eqref{eq:PTexample1}, using the Weyl invariance, the poles picked by the JK prescription are
\bea
\phi_{1}=\fra+\epsilon_{3},\quad \phi_{I}-\phi_{J}=+\epsilon_{1,2,3,4}\,\,I>J.
\eea
A similar discussion as section~\ref{sec:PToneleg} gives the poles
\bea
\phi_{I}=\fra+\epsilon_{3}I,\quad I=1,\ldots,k.
\eea 
The positions of the poles are the ones in section~\ref{sec:PToneleg} after changing $\eps_a\rightarrow -\eps_a$. Even so, the PT partition function actually becomes the same.

\begin{theorem}
    The poles are classified by the PT3 configurations with nontrivial equivariant weight:
    \bea
    \mathcal{Z}^{\PT\tbar\JK}_{\bar{4};\lambda\mu\nu}[\fq,q_{1,2,3,4}]=\sum_{\pi\in \widetilde{\mathcal{PT}}_{\lambda\mu\nu} }\fq^{|\pi|}\mathcal{Z}^{\PT\tbar\JK}_{\bar{4};\lambda\mu\nu}[\pi]
    \eea
    Note that the PT partition function actually matches with the previous case so we are using the same notation in the left hand side.
\end{theorem}

\subsection{Relation with topological vertices}\label{sec:top-vertex-JK}
In this section, we will study various limits of the equivariant PT3 vertices given in the previous section. We first start by studying the D6 partition function in \eqref{eq:D6U1PEformula} and show that limits of it produces generalized MacMahon functions. In particular, we study three limits: the unrefined limit, the refined limit, and the Macdonald refined limit. We then move on to the limits of the PT vertex and show that they match with the unrefined topological vertex, refined topological vertex, and the Macdonald refined topological vertex. We show this explicitly for the cases when each leg has one box at most, but it has been confirmed for other cases too. See Appendix~\ref{app:sec-PT3vertex-examples} for formulas for the topological vertices for other examples.

\subsubsection{MacMahon function and generalized MacMahon functions}\label{sec:general-MacMahon}
Let us first review the MacMahon function and generalized MacMahon functions. The \textbf{MacMahon function} is defined as the generating function of the plane partition:
\bea\label{eq:MacMahon-funct}
M(\fq)=\sum_{\pi\in\mathcal{PP}}\fq^{|\pi|}=\prod_{n=1}^{\infty}\frac{1}{(1-\fq^{n})^{n}}=\PE\left[\frac{\fq}{(1-\fq)^{2}}\right],\qquad|\fq|<1.
\eea
Namely, for each plane partition $\pi$, there is only the topological term $\fq^{|\pi|}$ and no additional weight.

We further can introduce a refinement of the MacMahon function, which is called the \textbf{refined MacMahon function}:
\bea\label{eq:refined-MacMahon-funct}
M(q,t)=\sum_{\pi\in\mathcal{PP}}t^{\sum_{i=1}^{\infty}|\pi(-i)|}q^{\sum_{j=1}^{\infty}|\pi(j-1)|}=\prod_{i,j=1}^{\infty}\frac{1}{(1-q^{i}t^{j-1})}=\PE\left[\frac{q}{(1-q)(1-t)}\right].
\eea

The MacMahon function has a further generalization introduced by Vuleti\'c 
     \cite{Vuletic2007AGO}:
 \bea
    M(x;q,t)= \sum_{\pi}F_{\pi}(q,t)x^{|\pi|}=\prod_{\substack{i=1\\n=0}}^{\infty}\left(\frac{1-x^{i}q^{n}t}{1-x^{i}q^{n}}\right)^{i}=\PE\left[\frac{1-t}{1-q}\frac{x}{(1-x)^{2}}\right]
     \eea
     where
     \bea
\,&F_{\pi}(q,t)=\prod_{(i,j)\in\pi}F_{\pi}(i,j)(q,t)
,\quad  
f(n,m)=\begin{dcases}
    \prod_{i=0}^{n-1}\frac{1-q^{i}t^{m+1}}{1-q^{i+1}t^{m}},\quad n\geq 1\\
    1,\quad n=0
\end{dcases}\\
& F_{\pi}(i,j)(q,t)=\prod_{m=0}^{\infty}\frac{f(\pi^{0}_{1}-\pi^{+}_{m+1},m)f(\pi_{1}^{0}-\pi^{-}_{m+1},m)}{f(\pi^{0}_{1}-\pi^{0}_{m+1},m)f(\pi^{0}_{1}-\pi^{0}_{m+2},m)},
\eea
and 
\bea
\pi^{0}=\{\pi_{i,j},\pi_{i+1,j+1},\ldots\},\quad \pi^{+}=\{\pi_{i+1,j},\pi_{i+2,j+1},\ldots\},\quad  \pi^{-}=\{\pi_{i,j+1},\pi_{i+1,j+2},\ldots\}
\eea
The $(i,j)\in\pi$ here is the support of the plane partition. The parameters $x,q,t$ are all independent parameters. The difference from the MacMahon and refined MacMahon function is that there are nontrivial weights for each plane partition $\pi$ in addition to the topological term. Since the right hand side is directly related with the Macdonald kernel function, it is called the Macdonald refinement of MacMahon function. Description using vertex operators were studied in \cite{Cai2015TheVO}.

We further can introduce a refinement in the topological term as
\bea\label{eq:Macd-refined-MacMahon-funct}
M(x,y;q,t)&=\sum_{\pi}F_{\pi}(q,t)x^{\sum_{j=1}^{\infty}|\pi(j-1)|}y^{\sum_{i=1}^{\infty}|\pi(-i)|}\\
&=\prod_{i,j=1}^{\infty}\prod_{n=0}^{\infty}\left(\frac{1-tx^{i}y^{j-1}q^{n}}{1-x^{i}y^{j-1}q^{n}}\right)=\PE\left[\frac{1-t}{1-q}\frac{x}{(1-x)(1-y)}\right],
\eea
which was studied by Foda--Wu \cite{Foda:2017tnv} (see also \cite{Foda:2018jwz}). Since the factors inside the plethystic exponential resembles the Macdonald kernel function, we call this the \textbf{Macdonald refined MacMahon function}. Note that the four parameters $x,y,q,t$ are all independent. Taking the limit $x=y$ gives Vuleti\'c's generalized MacMahon function. Taking the limit $q=t$ gives the refined MacMahon function $M(x,y)$. Further imposing $x=y$ gives the MacMahon function $M(x)$.


\paragraph{Unrefined limit}
Let us study various limits of the D6 partition function:
\bea\label{eq:D6partitionfunct}
\mathcal{Z}^{\D6}_{\bar{4}}[\fq;q_{1,2,3,4}]=\sum_{\pi}\fq^{|\pi|}\mathcal{Z}_{\bar{4}}^{\D4}[\pi]=\PE\left[\frac{[q_{14}][q_{24}][q_{34}]}{[q_{1}][q_{2}][q_{3}]}\frac{\fq}{(1-\fq q_{4}^{-1/2})(1-\fq q_{4}^{1/2})}\right].
\eea
The unrefined limit $q_{4}\rightarrow 1$ of the PE formula is
\bea
\frac{[q_{14}][q_{24}][q_{34}]}{[q_{1}][q_{2}][q_{3}]}\frac{\fq}{(1-\fq q_{4}^{-1/2})(1-\fq q_{4}^{1/2})}\xrightarrow{q_{4}\rightarrow 1} \frac{\fq}{(1-\fq)^{2}}.
\eea
The summand becomes trivial in this limit and
\bea
\mathcal{Z}_{\bar{4}}^{\D6}[\pi]\xrightarrow{q_{4}\rightarrow 1} 1
\eea
which gives
\bea
\mathcal{Z}^{\D6}_{\bar{4}}[\fq;q_{1,2,3,4}]\xrightarrow{q_{4}\rightarrow 1} M(\fq). 
\eea

\paragraph{Refined limit}
To relate the PE formula of the D6 partition function with the PE formula of the refined MacMahon function, we first rescale the equivariant parameters as
\bea\label{eq:parameters-refinedlimit}
(q_{1},q_{2},q_{3},q_{4})\rightarrow (z^{r_{1}}q_{1},z^{r_{2}}q_{2},z^{r_{3}}q_{3},q_{4}),\quad r_{1}+r_{2}+r_{3}=0
\eea
and then take the limit $z\rightarrow 0$. This limit was studied in \cite{Nekrasov:2014nea}. In this paper, we focus on the following four limits
\bea
 r_{1}\gg r_{3}>0 \gg  r_{2},\quad r_{1}\gg 0>r_{3} \gg  r_{2},\quad r_{2}\gg r_{3}>0 \gg  r_{1},\quad r_{2}\gg 0>r_{3} \gg  r_{1}
\eea
The third direction $q_{3}$, which is the middle parameter after rescaling is a special direction and it is called the \textit{preferred direction}.

For the first case, we have
\bea
\frac{[q_{14}][q_{24}][q_{34}]}{[q_{1}][q_{2}][q_{3}]}\frac{\fq}{(1-\fq q_{4}^{-1/2})(1-\fq q_{4}^{1/2})}&\rightarrow \frac{[z^{r_{1}}q_{14}][z^{r_{2}}q_{24}][z^{r_{3}}q_{34}]}{[z^{r_{1}}q_{1}][z^{r_{2}}q_{2}][z^{r_{3}}q_{3}]}\frac{\fq}{(1-\fq q_{4}^{-1/2})(1-\fq q_{4}^{1/2})}\\
&\xrightarrow{z\rightarrow 0}\frac{q_{4}^{-1/2}\fq}{(1-\fq q_{4}^{-1/2})(1-\fq q_{4}^{1/2})}=\frac{t}{(1-q)(1-t)}
\eea
where we used
\bea
\frac{[z^{r}q_{i}q_{4}]}{[z^{r}q_{i}]}=\frac{(z^{r/2}q_{i4}^{1/2}-z^{-r/2}q_{i4}^{-1/2})}{(z^{r/2}q_{i}^{1/2}-z^{-r/2}q_{i}^{-1/2})}=\begin{dcases}
    q_{4}^{-1/2},\quad r>0\\
    q_{4}^{1/2},\quad\,\, r<0
\end{dcases}
\eea
and the parameters correspondences are
\bea\label{eq:refined-parameter}
q=\fq q_{4}^{1/2},\quad t=\fq q_{4}^{-1/2}.
\eea
In this limit, one can also confirm
\bea
\fq^{|\pi|}\mathcal{Z}_{\bar{4}}^{\D6}[\pi]\longrightarrow q^{\sum_{i=1}^{\infty}|\pi(-i)|}t^{\sum_{j=1}^{\infty}|\pi(j-1)|}=\fq^{|\pi|}q_{4}^{\frac{1}{2}\sum_{i=1}^{\infty}(|\pi(-i)|-|\pi(i-1)|)}.
\eea

For other limits of parameters, we have
\bea
\mathcal{Z}_{\bar{4}}^{\D6}[\fq;q_{1,2,3,4}]\longrightarrow \begin{dcases}
   M(t,q),\quad  r_{1}\gg r_{3}>0\gg r_{2}\\
   M(q,t),\quad  r_{1}\gg 0>r_{3}\gg r_{2}\\
   M(t,q),\quad r_{2}\gg r_{3}>0 \gg  r_{1}\\
   M(q,t),\quad r_{2}\gg 0>r_{3} \gg  r_{1}
\end{dcases}
\eea
where we fixed the parameter correspondence \eqref{eq:refined-parameter}. 

In the unrefined limit, we have
\bea
\fq=q=t
\eea
and so when discussing the unrefined limit, we do not distinguish these parameters.



\paragraph{Macdonald refined limit}
To relate with the D6 partition function, we instead take the following limit
\bea\label{eq:parameters-Macreflimit}
(q_{1},q_{2},q_{3},q_{4})\rightarrow (z^{r}q_{1},z^{-r}q_{2},q_{3},q_{4}),
\eea
while keeping the parameter $q_{3}$ fixed. Similarly to the previous case, we call this direction the preferred direction and keep it to be the third direction. We have two cases $r>0$ and $r<0$. For example, for $r>0$, we have
\bea
\frac{[q_{14}][q_{24}][q_{34}]}{[q_{1}][q_{2}][q_{3}]}\frac{\fq}{(1-\fq q_{4}^{-1/2})(1-\fq q_{4}^{1/2})}&\rightarrow \frac{[z^{r}q_{14}][z^{-r}q_{24}][q_{34}]}{[q_{1}][q_{2}][q_{3}]}\frac{\fq}{(1-\fq q_{4}^{-1/2})(1-\fq q_{4}^{1/2})}\\
&\xrightarrow{z\rightarrow 0} \frac{[q_{34}]}{[q_{3}]}\frac{\fq}{(1-\fq q_{4}^{-1/2})(1-\fq q_{4}^{1/2})}\\
&=\frac{(1-q_{34})}{(1-q_{3})}\frac{\fq q_{4}^{-1/2}}{(1-\fq q_{4}^{-1/2})(1-\fq q_{4}^{1/2})}\\
&=\frac{1-t}{1-q}\frac{x}{(1-x)(1-y)}
\eea
where we used the parameter correspondence:
\bea\label{eq:Macref-parameter}
x=\fq q_{4}^{-1/2},\quad y=\fq q_{4}^{1/2},\quad q=q_{3},\quad t=q_{34}
\eea
with the constraint
\bea\label{eq:Macref-constr}
x/y=q/t.
\eea
Actually we also have the same partition function for $r<0$. 

Starting from the Macdonald refined MacMahon function and imposing the condition \eqref{eq:Macref-constr} gives the limit of the D6 partition function. One can also directly confirm the correspondence
\bea
\fq^{|\pi|}\mathcal{Z}^{\D6}_{\bar{4}}[\pi]\longrightarrow \left.F_{\pi}(q,t)x^{\sum_{j=1}^{\infty}|\pi(j-1)|}y^{\sum_{i=1}^{\infty}|\pi(-i)|}\right|_{\eqref{eq:Macref-parameter}}
\eea

Summarizing, we have
\bea
\mathcal{Z}_{\bar{4}}^{\D6}[\fq;q_{1,2,3,4}]\longrightarrow
    M(x,y;q,t)\quad r\neq 0
\eea
where we are assuming \eqref{eq:Macref-constr} for the right hand side.

\subsubsection{Unrefined topological vertex}\label{sec:unrefined-vertex}
\begin{definition}[\cite{Aganagic:2003db,Okounkov:2003sp}]
    The \textbf{unrefined topological vertex} is defined as
\bea
C_{\lambda\mu\nu}(q)=q^{\kappa(\mu)/2}s_{\nu^{\rmT}}(q^{-\rho})\sum_{\eta}s_{\lambda^{\rmT}/\eta}(q^{-\rho-\nu})s_{\mu/\eta}(q^{-\rho-\nu^{\rmT}})
\eea
where $\kappa (\mu)=||\mu||^{2}-||\mu^{\rmT}||^{2}$, $||\mu||^{2}=\sum_{i=1}^{\ell(\mu)}\mu_{i}^{2}$ (see Appendix~\ref{app:top-vertex-symm-funct} for the notations), 
\bea
     \rho=(-\frac{1}{2},-\frac{3}{2},-\frac{5}{2},\ldots),\quad s_{\lambda}(1,q,q^{2},\cdots)=q^{n(\lambda)}\prod_{\Abox\in \lambda}(1-q^{h_{\lambda}(\Abox)})^{-1}
\eea
and $n(\lambda)=\sum_{i}(i-1)\lambda_{i}$ and $h_{\lambda}(\Bbox)=a_{\lambda}(\Bbox)+\ell_{\lambda}(\Bbox)+1$. The $a_{\lambda}(\Bbox),\ell_{\lambda}(\Bbox)$ are the arm and leg length. 

\end{definition}
We denote the normalized topological vertex as
\bea
\wtC_{\lambda\mu\nu}(q)=1+\cdots 
\eea
where after expanding in series of $q$, the first term will be $1$.

The main claim of this section is the following theorem.
\begin{theorem}\label{thm:PT3vertex-unref-corresp}
    In the unrefined limit $q_{4}\rightarrow 1$, the PT partition function becomes the unrefined topological vertex:
    \bea
\mathcal{Z}^{\PT\tbar\JK}_{\bar{4};\lambda\mu\nu}[\fq,q_{1,2,3,4}]\longrightarrow \wtC_{\lambda^{\rmT}\mu^{\rmT}\nu^{\rmT}}(\fq).
    \eea
\end{theorem}

Note that the orientation of \cite{Okounkov:2003sp} is opposite from the definition here and the transpose of the Young diagrams on the topological vertex side is the consequence of it. We have confirmed the above theorem for the examples in Appendix~\ref{app:sec-PT3vertex-examples}. Let us check this explicitly for the cases \eqref{eq:DTexample1}, \eqref{eq:DTexample2}, \eqref{eq:DTexample3}.

\paragraph{One-leg case}
The one leg PT3 partition function becomes
\bea
\mathcal{Z}^{\PT\tbar\JK}_{\bar{4};\,\varnothing\varnothing\,\Cbox}[\fq,q_{1,2,3,4}]=\sum_{k=0}^{\infty}\fq^{k}\prod_{i=1}^{k}\frac{[q_{4}q_{3}^{i}]}{[q_{3}^{i}]}\longrightarrow\wtC_{\varnothing\varnothing\,\Cbox}(\fq)= \sum_{k=0}^{\infty}\fq^{k}=\frac{1}{1-\fq}
\eea
\paragraph{Two-leg case}The two leg PT3 partition function becomes
\bea
\mathcal{Z}^{\PT\tbar\JK}_{\bar{4};\,\Cbox\,\Cbox\,\varnothing}[\fq,q_{1,2,3,4}]=1+\sum_{m,n=0}^{\infty}\mathcal{Z}^{\PT\tbar\JK}_{\bar{4};\,\Cbox\,\Cbox\,\varnothing}[[m,n]]\longrightarrow \wtC_{\Cbox\,\Cbox\varnothing}(\fq)=1+\sum_{m,n=0}^{\infty}\fq^{m+n+1} =\frac{1-\fq+\fq^{2}}{(1-\fq)^{2}} 
\eea
where we used
\bea
\mathcal{Z}^{\PT\tbar\JK}_{\bar{4};\,\Cbox\,\Cbox\,\varnothing}[[m,n]]\rightarrow 1.
\eea

\paragraph{Three-leg case}
For the three-leg case, taking the unrefined limit $\epsilon_{1}+\epsilon_{2}+\epsilon_{3}\rightarrow 0$, we have for level one \eqref{eq:3leglevel1}
\bea
\mathcal{Z}^{\PT\tbar\JK}_{\bar{4};\,\Cbox\,\Cbox\,\Cbox}[1] \rightarrow 2
\eea
where we used $\sh(0)=0$ and $\ch(0)=2$. The origin of the two state degeneracy is this $\ch$ factor.

For level two
\bea
\mathcal{Z}_{(0,0)}\rightarrow 2,\quad \mathcal{Z}_{(0,-\epsilon_{1,2,3})}\rightarrow 1
\eea
which gives
\bea
\mathcal{Z}^{\PT\tbar\JK}_{\bar{4};\,\Cbox\,\Cbox\,\Cbox}[2]\rightarrow 4
\eea
For level three, we have
\bea
\mathcal{Z}_{(0,0,-\epsilon_{1,2,3})}\rightarrow 2,\quad \mathcal{Z}_{(0,-\epsilon_{i},-2\epsilon_{i})}\rightarrow 1
\eea
giving 
\bea
\mathcal{Z}^{\PT\tbar\JK}_{\bar{4};\,\Cbox\,\Cbox\,\Cbox}[3]\rightarrow 6.
\eea
For the generic level \eqref{eq:3leglevelgeneric}, the residues of the configurations with the ultra-heavy box become
\bea
\mathcal{Z}_{(0,0,-\epsilon_{1},\ldots,-n_{1}\epsilon_{1},-\epsilon_{2},\ldots,-n_{2}\epsilon_{2},\epsilon_{3},\ldots,-n_{3}\epsilon_{3})}\longrightarrow 2
\eea
and for those without the ultra heavy box become
\bea
\mathcal{Z}_{(0,-\epsilon_{i},\ldots,-(k-1)\epsilon_{i})}\longrightarrow 1.
\eea
Taking care of the multiplicity and the Weyl group factor, we finally have
\bea
\mathcal{Z}^{\PT\tbar\JK}_{\bar{4};\,\Cbox\,\Cbox\,\Cbox}[k]&\longrightarrow \sum_{i=1}^{3}1+\sum_{n_{1}+n_{2}+n_{3}=k-2}1=3+\binom{k}{2}.
\eea
Therefore, 
\bea
\mathcal{Z}^{\PT\tbar\JK}_{\bar{4};\Cbox\Cbox\Cbox}[\fq,q_{1,2,3,4}]\longrightarrow &\,\,\wtC_{\Cbox\Cbox\Cbox}(\fq)=1+2\mathfrak{q}+\sum_{k\geq 2}\left(\binom{k}{2}+3\right)\mathfrak{q}^{k}=\frac{1-\mathfrak{q}+\mathfrak{q}^{2}-\mathfrak{q}^{3}+\mathfrak{q}^{4}}{(1-\mathfrak{q})^{3}}\\
&\qquad \,\quad=1+2\mathfrak{q}+4\fq^{2}+6\fq^{3}+9\fq^{4}+13\fq^{5}\cdots .
\eea
Note that the two states degeneracy only appear at level one.

\subsubsection{Refined topological vertex}\label{sec:refined-vertex}
\begin{definition}[\cite{Iqbal:2007ii}]
    The \textbf{refined topological vertex}\footnote{For explicit computations of the refined topological vertex and the unrefined topological vertex, we used a \texttt{Mathematica} program \href{https://github.com/kantohm11/SchurFs?tab=readme-ov-file}{SchurFs} accompanied with \cite{Hayashi:2017jze}. } is defined as
    \bea
    C_{\lambda\mu\nu}(t,q)
    &=q^{\frac{||\mu||^{2}+||\nu||^{2}}{2}}t^{-\frac{||\mu^{\rmT}||^{2}}{2}}\widetilde{Z}_{\nu}(t,q)\sum_{\eta}\left(\frac{q}{t}\right)^{\frac{|\eta|+|\lambda|-|\mu|}{2}}s_{\lambda^{\rmT}/\eta}(t^{-\rho}q^{-\nu})s_{\mu/\eta}(t^{-\nu^{\rmT}}q^{-\rho})
\eea
where 
\bea
     \rho=(-\frac{1}{2},-\frac{3}{2},-\frac{5}{2},\cdots),\quad
    \widetilde{Z}_{\nu}(t,q)&=\prod_{(i,j)\in\nu}(1-t^{\ell_{\nu}(i,j)+1}q^{a_{\nu}(i,j)})^{-1}.
\eea
\end{definition}
Note that the unrefined topological vertex is obtained from the unrefined limit $q=t$:
\bea
C_{\lambda\mu\nu}(q,q)=C_{\lambda\mu\nu}(q).
\eea
The refined topological vertex can be expanded in series of $q,t$. We normalize the refined topological vertex and make the first term to start from $1$ and denote it as
\bea
\widetilde{C}_{\lambda\mu\nu}(t,q)=1+\cdots.
\eea

Depending on how we take the limits of the equivariant parameters, we have the following correspondence.
\begin{theorem}\label{thm:PT3vertex-ref-corresp}
    Keeping the third direction to be the preferred direction, the PT3 vertex under the limit \eqref{eq:parameters-refinedlimit} discussed in becomes the refined topological vertex:
    \bea
    \mathcal{Z}^{\PT\tbar\JK}_{\bar{4};\lambda\mu\nu}[\mathfrak{q},q_{1,2,3,4}]\longrightarrow  \begin{dcases}
    \wtC_{\lambda^{\rmT}\mu^{\rmT}\nu^{\rmT}}(t,q),\quad r_{1}\gg r_{3}>0\gg r_{2}\\
    \wtC_{\mu\lambda\nu}(q,t),\quad r_{1}\gg 0>r_{3}\gg r_{2}\\
    \wtC_{\mu\lambda\nu}(t,q),\quad r_{2}\gg r_{3}>0\gg r_{1}\\
    \wtC_{\lambda^{\rmT}\mu^{\rmT}\nu^{\rmT}}(q,t),\quad r_{1}\gg 0>r_{3}\gg r_{2}
    \end{dcases}
    \eea
    with the parameter correspondence
    \bea
    q=\fq q_{4}^{1/2},\quad t=\fq q_{4}^{-1/2}
    \eea
\end{theorem}

Let us see this explicitly for the examples in \eqref{eq:DTexample1}, \eqref{eq:DTexample2}, and \eqref{eq:DTexample3} for the limit $r_{1}\gg r_{3}>0\gg r_{2}$. 
\paragraph{One-leg case}
The refined topological vertex is given as
\bea
\wtC_{\Cbox\,\varnothing\varnothing}(t,q)=\frac{1}{1-t},\quad \wtC_{\varnothing\,\Cbox\,\varnothing}(t,q)=\frac{1}{1-q},\quad \wtC_{\varnothing\varnothing\,\Cbox}(t,q)=\frac{1}{1-t}.
\eea
The PT partition functions are
\bea
\mathcal{Z}^{\PT\tbar\JK}_{\bar{4};\,\Cbox\,\varnothing\varnothing}[\mathfrak{q},q_{1,2,3,4}]=\sum_{k=0}^{\infty}\mathfrak{q}^{k}\prod_{i=1}^{k}\frac{[q_{4}q_{1}^{i}]}{[q_{1}^{i}]}=\PE\left[\mathfrak{q}\frac{[q_{14}]}{[q_{1}]}\right],\\
\mathcal{Z}^{\PT\tbar\JK}_{\bar{4};\varnothing\,\Cbox\,\varnothing}[\mathfrak{q},q_{1,2,3,4}]=\sum_{k=0}^{\infty}\mathfrak{q}^{k}\prod_{i=1}^{k}\frac{[q_{4}q_{2}^{i}]}{[q_{2}^{i}]}=\PE\left[\mathfrak{q}\frac{[q_{24}]}{[q_{2}]}\right],\\
\mathcal{Z}^{\PT\tbar\JK}_{\bar{4};\varnothing\varnothing\,
\Cbox}[\mathfrak{q},q_{1,2,3,4}]=\sum_{k=0}^{\infty}\mathfrak{q}^{k}\prod_{i=1}^{k}\frac{[q_{4}q_{3}^{i}]}{[q_{3}^{i}]}=\PE\left[\mathfrak{q}\frac{[q_{34}]}{[q_{3}]}\right].
\eea
Focusing on $r_{1}\gg r_{3}>0\gg r_{2}$ and taking the limit of the equivariant parameters, we have
\bea
\mathcal{Z}^{\PT\tbar\JK}_{\bar{4};\,\Cbox\,\varnothing\varnothing}[\mathfrak{q},q_{1,2,3,4}]&\longrightarrow \PE[\fq q_{4}^{-1/2}]=\PE[t]=\wtC_{\Cbox\,\varnothing\varnothing}(t,q),\\
\mathcal{Z}^{\PT\tbar\JK}_{\bar{4};\varnothing\,\Cbox\,\varnothing}[\mathfrak{q},q_{1,2,3,4}]&\longrightarrow \PE[\fq q_{4}^{1/2}]=\PE[q]=\wtC_{\varnothing\,\Cbox\,\varnothing}(t,q),\\
\mathcal{Z}^{\PT\tbar\JK}_{\bar{4};\varnothing\varnothing\,
\Cbox}[\mathfrak{q},q_{1,2,3,4}]&\longrightarrow \PE[\fq q_{4}^{-1/2}]=\PE[t]=\wtC_{\varnothing\varnothing\,\Cbox}(t,q).
\eea


\paragraph{Two-leg case}
The refined topological vertices are 
\bea
\wtC_{\Cbox\,\Cbox\varnothing}(t,q)=\frac{1-q+qt}{(1-t)(1-q)},\quad \wtC_{\varnothing\,\Cbox\,\Cbox}(t,q)=\frac{1-q+qt}{(1-t)(1-q)},\quad \wtC_{\Cbox\,\varnothing\,\Cbox}(t,q)=\frac{1-t+qt}{(1-t)^{2}}.
\eea
Under the refined limit with the slope $r_{1}\gg r_{3}>0\gg r_{2}$, we have
\bea
\fq^{m+n+1}\mathcal{Z}_{\bar{4};\,\Cbox\,\Cbox\varnothing}^{\PT\tbar\JK}[[m,n]]\longrightarrow \fq^{m+n+1}q_{4}^{-1/2} q_{4}^{-m/2}q_{4}^{n/2}=t^{m+1}q^{n}
\eea
and thus
\bea
\mathcal{Z}^{\PT\tbar\JK}_{\bar{4};\,
\Cbox\,\Cbox\,\varnothing}[\mathfrak{q},q_{1,2,3,4}]\longrightarrow 1+\sum_{m,n=0}^{\infty}t^{m+1}q^{n}=1+\frac{t}{(1-q)(1-t)}=\wtC_{\Cbox\,\Cbox\,\varnothing}(t,q).
\eea
For the other cases, we have 
\bea
\fq^{m+n+1}\mathcal{Z}_{\bar{4};\,\Cbox\,\varnothing\,\Cbox}^{\PT\tbar\JK}[[m,n]]&\longrightarrow \fq^{m+n+1}\begin{dcases}
    q_{4}^{1/2}q_{4}^{-m/2}q_{4}^{-n/2}\\
    q_{4}^{-m/2}q_{4}^{-1/2}
\end{dcases} =\begin{dcases}
    qt^{m+n},\quad n\neq 0\\
    t^{m+1},\quad n=0
\end{dcases}\\
\fq^{m+n+1}\mathcal{Z}_{\bar{4};\varnothing\,\Cbox\,\Cbox}^{\PT\tbar\JK}[[m,n]]&\longrightarrow\fq^{m+n+1}q_{4}^{-1/2}q_{4}^{m/2}q_{4}^{-n/2}=t^{n+1}q^{m}
\eea 
where $m,n$ are the number of boxes extending in $3,1$ and $2,3$ respectively. Therefore, we have
\bea
\mathcal{Z}^{\PT\tbar\JK}_{\bar{4};\,\Cbox\,\varnothing\,\Cbox}[\mathfrak{q},q_{1,2,3,4}]&\longrightarrow \sum_{n=0}^{\infty}t^{n}+q\sum_{k=1}^{\infty}kt^{k}=\frac{1}{1-t}+q\frac{t}{(1-t)^{2}}=\wtC_{\Cbox\,\varnothing\,\Cbox}(t,q)\\
\mathcal{Z}^{\PT\tbar\JK}_{\bar{4};\varnothing\,\Cbox\,\Cbox}[\mathfrak{q},q_{1,2,3,4}]&\longrightarrow 1+\sum_{m,n=0}^{\infty} q^{m}t^{n+1}=1+\frac{t}{(1-q)(1-t)}=\wtC_{\varnothing\,\Cbox\,\Cbox}(t,q)
\eea

\paragraph{Three-leg case}
The refined topological vertex is expanded as
\bea
\wtC_{\Cbox\,\Cbox\,\Cbox}(t,q)=\frac{1-q^{2}-t+2qt-qt^{2}+q^{2}t^{2}}{(1-q)(1-t)^{2}}=1 + q + t + 3 q t + q^2 t + t^2 + 4 q t^2 + t^3+\cdots.
\eea
For level one, we have
\bea
\mathcal{Z}^{\PT\tbar\JK}_{\bar{4};\,\Cbox\,\Cbox\,\Cbox}[1]\longrightarrow q_{4}^{1/2}+q_{4}^{-1/2}.
\eea
For level two, we have
\bea
\frac{1}{2}\mathcal{Z}_{(0,0)}&\longrightarrow  q_{123}, \quad \mathcal{Z}_{(0,-\epsilon_{1})}\longrightarrow 1,\quad \mathcal{Z}_{(0,-\epsilon_{2})}\longrightarrow 1 ,\quad 
\mathcal{Z}_{(0,-\epsilon_{3})}\longrightarrow 1.
\eea
For level three, we have
\bea
\frac{1}{2}\mathcal{Z}_{(0,0,-\epsilon_{1})}\rightarrow q_{4}^{-1/2} ,\quad \frac{1}{2}\mathcal{Z}_{(0,0,-\epsilon_{2})}\rightarrow q_{4}^{-1/2},\quad \frac{1}{2}\mathcal{Z}_{(0,0,-\epsilon_{3})}\rightarrow q_{4}^{-3/2}\\
 \mathcal{Z}_{(0,-\epsilon_{1},-2\epsilon_{1})}\rightarrow q_{4}^{-1/2},\quad \mathcal{Z}_{(0,-\epsilon_{2},-2\epsilon_{2})}\rightarrow q_{4}^{1/2},\quad \mathcal{Z}_{(0,-\epsilon_{3},-2\epsilon_{3})}\rightarrow q_{4}^{-1/2}.
\eea
Totally, we have
\bea
\mathcal{Z}^{\PT\tbar\JK}_{\bar{4};\,
\Cbox\,\Cbox\,\Cbox}[\mathfrak{q},q_{1,2,3,4}]&\rightarrow 1+\fq (q_{4}^{1/2}+q_{4}^{-1/2})+\fq^{2}(3+q_{4}^{-1})+\fq^{3}(4 q_{4}^{-1/2}+q_{4}^{1/2}+q_{4}^{-3/2})+\cdots\\
&=1+q+t+3qt+t^{2}+4qt^{2}+q^{2}t+t^{3}+\cdots=\wtC_{\Cbox\,\Cbox\,\Cbox}(t,q)
\eea
which matches with the refined topological vertex up to three level. One can also confirm this for higher levels but we omit the discussion.

\subsubsection{Macdonald refined topological vertex}\label{sec:Macdonald-vertex}
\begin{definition}[\cite{Foda:2017tnv}]
The \textbf{Macdonald refined topological vertex}\footnote{For explicit computations of the Macdonald refined topological vertex, we wrote a \texttt{Sage Math} program by using basic \href{https://more-sagemath-tutorials.readthedocs.io/en/latest/tutorial-symmetric-functions.html}{symmetric functions}. For example, one can use the scalar product, which is already implemented, to introduce skew Macdonald symmetric functions. Converting to the power sum and substituting the variables such as $x^{-\rho+1/2}y^{-\nu}$, one can evaluate the explicit Macdonald vertex. See Appendix~\ref{app:sec-PT3vertex-examples} for examples.} (shortly Macdonald vertex) is defined as
\bea
\mathcal{M}_{\lambda\mu\nu}(x,y;q,t)=\prod_{n=0}^{\infty}\left(\prod_{\Abox\in\nu}\frac{1-tq^{n}x^{\ell_{\nu}(\Abox)+1}y^{a_{\nu}(\Abox)}}{1-q^{n}x^{\ell_{\nu}(\Abox)+1}y^{a_{\nu}(\Abox)}}\right)\sum_{\eta}P_{\mu/\eta}(y^{-\rho-1/2}x^{-\nu^{\rmT}};q,t)Q_{\lambda^{\rmT}/\eta}(x^{-\rho+1/2}y^{-\nu};q,t)
\eea
where $P_{\lambda}(x;q,t),Q_{\mu}(x;q,t)$ are skew Macdonald symmetric functions.     
\end{definition}
Note that the parameters $x,y,q,t$ are all independent. The notation used here is slightly different from the original paper \cite{Foda:2017tnv}. The three Young diagrams $(Y_{1},Y_{2},Y_{3})$ there are $(\mu,\lambda^{\rmT},\nu)$ here. The normalized Macdonald vertex is denoted as
\bea
\widetilde{\mathcal{M}}_{\lambda\mu\nu}(x,y;q,t)=1+\cdots
\eea
after expanding $\mathcal{M}_{\lambda\mu\nu}(x,y)$ in series of $x,y$ and dividing it by the first term with the lowest degree of $x,y$.

Taking limits of the normalized Macdonald vertex, one can recover the normalized unrefined and refined topological vertices:
\bea
\wtM_{\lambda\mu\nu}(x,y;q,t)=\begin{dcases}
    \wtC_{\lambda\mu\nu}(x,y),\quad q=t\\
    \wtC_{\lambda\mu\nu}(x),\quad y=x,\,\,t=q
\end{dcases}
\eea

\begin{theorem}\label{thm:PT3vertex-Macdref-corresp}
Keeping the third direction to be the preferred direction and taking the limit \eqref{eq:parameters-Macreflimit}, the PT3 vertex becomes the Macdonald vertex
    \bea
    \mathcal{Z}^{\PT\tbar\JK}_{\bar{4};\lambda\mu\nu}[\mathfrak{q},q_{1,2,3,4}] \longrightarrow \begin{dcases}
    \wtM_{\lambda^{\rmT}\mu^{\rmT}\nu^{\rmT}}(x,y;q,t),\quad r>0 \\
    \wtM_{\mu\lambda\nu}(x,y;q,t),\quad  r<0
    \end{dcases}
    \eea
    where the parameters of the Macdonald vertex are specialized as \eqref{eq:Macref-parameter}
    \bea
x=\fq q_{4}^{-1/2},\quad y=\fq q_{4}^{1/2},\quad q=q_{3},\quad t=q_{34}.
    \eea
\end{theorem}
Note that strictly speaking, the Macdonald vertex in the right hand side is not the full Macdonald vertex but a specialization of it. Namely, the nontrivial constraint $xt=qy$ is imposed. Practically, to confirm this correspondence, one starts from the Macdonald vertex and expand it in the series of $x,y$. Then, one can insert the parameter correspondence \eqref{eq:Macref-parameter}. Comparing with the limits of the JK-residue computation, one obtains the above correspondence.

Let us focus on the $r>0$ case and confirm it explicitly when there is at most one box for each leg.
\paragraph{One-leg case}
The Macdonald vertex is 
\bea
\wtM_{\Cbox\,\varnothing\varnothing}(x,y;q,t)=\frac{1}{1-x},\quad \wtM_{\varnothing\,\Cbox\,\varnothing}(x,y;q,t)=\frac{1}{1-y},\quad \wtM_{\varnothing\varnothing\,\Cbox}(x,y;q,t)=\frac{(tx;q)_{\infty}}{(x;q)_{\infty}}
\eea
Let us focus on the slope $r>0$. For the $(\Bbox,\varnothing,\varnothing)$ and $(\varnothing,\Bbox,\varnothing)$ cases, the computation is the same:
\bea
\mathcal{Z}^{\PT\tbar\JK}_{\bar{4};\,\Cbox\,\varnothing\varnothing}[\mathfrak{q},q_{1,2,3,4}]&\longrightarrow \PE[\fq q_{4}^{-1/2}]=\PE[x]=\wtM_{\Cbox\,\varnothing\varnothing}(x,y;q,t),\\
\mathcal{Z}^{\PT\tbar\JK}_{\bar{4};\varnothing\,\Cbox\,\varnothing}[\mathfrak{q},q_{1,2,3,4}]&\longrightarrow \PE[\fq q_{4}^{1/2}]=\PE[y]=\wtM_{\varnothing\,\Cbox\,\varnothing}(x,y;q,t).
\eea
For the $(\varnothing,\varnothing,\Bbox)$ case, the limit does not affect the partition function and so we have
\bea
\mathcal{Z}^{\PT\tbar\JK}_{\bar{4};\varnothing\varnothing\,
\Cbox}[\mathfrak{q},q_{1,2,3,4}]=\PE\left[\mathfrak{q}\frac{[q_{34}]}{[q_{3}]}\right]\longrightarrow\PE\left[\frac{1-t}{1-q}x\right]=\wtM_{\varnothing\varnothing\,\Cbox}(x,y;q,t)
\eea

\paragraph{Two-leg case}
The Macdonald vertices are given as
\bea
\wtM_{\Cbox\,\Cbox\varnothing}(x,y;q,t)&=\frac{(qxy - qx + tx - qy - xy + q + y - 1)}{(q - 1)(x - 1)(y - 1)},\\
\wtM_{\varnothing\,\Cbox\,\Cbox}(x,y;q,t)&=\frac{(tx;q)_{\infty}}{(x;q)_{\infty}}\frac{1-y+xy}{1-y},\\
\wtM_{\Cbox\,\varnothing\,\Cbox}(x,y;q,t)&=\frac{(tx;q)_{\infty}}{(x;q)_{\infty}}\frac{1-x+xy}{1-x}.
\eea
Let us focus first on the case $(\Bbox,\Bbox,\varnothing)$. The Macdonald vertex under the condition $xt=qy$ is rewritten as
\bea
\wtM_{\Cbox\,\Cbox\varnothing}(x,y;q,t)&= \frac{1-y-q+xy+qx-qxy}{(1-q)(1-x)(1-y)}=1+\frac{x(1-t)}{(1-q)(1-x)(1-y)}
\eea
For the PT partition function, taking the limit \eqref{eq:parameters-Macreflimit} gives
\bea
\fq^{m+n+1}\mathcal{Z}_{\bar{4};\,\Cbox\,\Cbox\varnothing}^{\PT\tbar\JK}[[m,n]]&\longrightarrow \fq^{m+n+1}\frac{[q_{34}]}{[q_{3}]}q_{4}^{-m/2}q_{4}^{n/2}=\frac{1-t}{1-q}x^{m+1}y^{n}
\eea
and thus
\bea
\mathcal{Z}^{\PT\tbar\JK}_{\bar{4};\,
\Cbox\,\Cbox\varnothing}[\mathfrak{q},q_{1,2,3,4}]\longrightarrow 1+\sum_{m,n=0}^{\infty}\frac{1-t}{1-q}x^{m+1}y^{n}=1+\frac{(1-t)x}{(1-q)(1-x)(1-y)}=\wtM_{\Cbox\,\Cbox\varnothing}(x,y;q,t).
\eea

For the case $(\varnothing,\Bbox,\Bbox)$, the limit of the PT partition function is
\bea
\mathcal{Z}_{\bar{4};\varnothing\,\Cbox\,\Cbox}^{\PT\tbar\JK}[[m,n]]\longrightarrow \begin{dcases}
    \prod_{j=1}^{n+1}\frac{[q_{4}q_{3}^{j}]}{[q_{3}^{j}]},\quad m=0\\
    q_{4}^{\frac{m-1}{2}}\prod_{j=1}^{n}\frac{[q_{4}q_{3}^{j}]}{[q_{3}^{j}]},\quad m\neq 0
\end{dcases}
\eea
which gives
\bea
\mathcal{Z}^{\PT\tbar\JK}_{\bar{4};\varnothing\,\Cbox\,\Cbox}[\mathfrak{q},q_{1,2,3,4}]&\longrightarrow \left(1+\sum_{n=0}^{\infty}\fq^{n+1}\prod_{j=1}^{n+1}\frac{[q_{4}q_{3}^{j}]}{[q_{3}^{j}]}\right)+\sum_{m=1}^{\infty}\sum_{n=0}^{\infty}\fq^{m+n+1}q_{4}^{\frac{m-1}{2}}\prod_{j=1}^{n}\frac{[q_{4}q_{3}^{j}]}{[q_{3}^{j}]}\\
&=\sum_{n=0}^{\infty}\prod_{j=1}^{n}\frac{1-tq^{j-1}}{1-q^{j}}x^{n}+\left(\sum_{n=0}^{\infty}x^{n}\prod_{j=1}^{n}\frac{1-tq^{j-1}}{1-q^{j}}\right)\left(x\sum_{m=1}^{\infty}y^{m}\right)\\
&=\frac{(tx;q)_{\infty}}{(x;q)_{\infty}}+\frac{(tx;q)_{\infty}}{(x;q)_{\infty}}\frac{xy}{1-y}=\wtM_{\varnothing\,\Cbox\,\Cbox}(x,y;q,t)
\eea
where we used
\bea
\frac{(tx;q)_{\infty}}{(x;q)_{\infty}}=\sum_{n=0}^{\infty}\frac{(t;q)_{n}}{(q;q)_{n}}x^{n},\quad (x;q)_{n}=\prod_{i=0}^{n-1}(1-xq^{i}).
\eea

Similarly for $(\Bbox,\varnothing,\Bbox)$, we have
\bea
\mathcal{Z}_{\bar{4};\,\Cbox\,\varnothing\,\Cbox}^{\PT\tbar\JK}[[m,n]]\longrightarrow \begin{dcases}
    \prod_{j=1}^{m+1}\frac{[q_{4}q_{3}^{j}]}{[q_{3}^{j}]},\quad n=0\\
    q_{4}^{\frac{-n+1}{2}}\prod_{j=1}^{m}\frac{[q_{4}q_{3}^{j}]}{[q_{3}^{j}]},\quad n\neq 0
\end{dcases}
\eea
and
\bea
\mathcal{Z}_{\bar{4};\,\Cbox\,\varnothing\,\Cbox}^{\PT\tbar\JK}[\fq,q_{1,2,3,4}]\longrightarrow \frac{(tx;q)_{\infty}}{(x;q)_{\infty}}\left(1+\frac{xy}{1-x}\right)=
\wtM_{\Cbox\,\varnothing\,\Cbox}(x,y;q,t).
\eea

\paragraph{Three-leg case}
The Macdonald vertex is
\bea
&\wtM_{\Cbox\,\Cbox\,\Cbox}(x,y;q,t)\\
=&\frac{(tx;q)_{\infty}}{(x;q)_{\infty}}\frac{(tx^2y^2 - tx^2y + qxy^2 - txy^2 - x^2y^2 - qxy + 3txy + x^2y - qy^2 - tx + qy - ty - 2xy + y^2 + t + x - 1)}{(t - 1)(x - 1)(y - 1)}.
\eea
Expanding in series of $x,y$, we have
\bea
\wtM_{\Cbox\,\Cbox\,\Cbox}(x,y;q,t)
&=1+\frac{(t-1)^{2}x+(q-1)^{2}y}{(q-1)(t-1)}+\left(3xy+\frac{(1-t)(1-qt)}{(1-q)(1-q^{2})}x^{2}\right)\\
&+\left(\frac{(1-t)(1-qt)(1-q^{2}t)}{(1-q)(1-q^{2})(1-q^{3})}x^{3}+\frac{4-2t+2q-3qt-q^{2}}{1-q^{2}}x^{2}y+xy^{2}\right)+\cdots
\eea

For the level one, the limit is
\bea
\fq\mathcal{Z}^{\PT\tbar\JK}_{\bar{4};\,\Cbox\,\Cbox\,\Cbox}[1]\rightarrow \fq\frac{-q_1^2 q_3 q_2^2-q_1 \left(q_3^2-4 q_3+1\right) q_2-q_3}{\sqrt{q_1} \sqrt{q_2}
   \left(q_1 q_2-1\right) \left(q_3-1\right) \sqrt{q_3}}
\eea
which matches with the first level of $\wtM_{\Cbox\,\Cbox\,\Cbox}(x,y;q,t)$ after inserting \eqref{eq:Macref-parameter}.

For level two, we have
\bea
\frac{1}{2}\mathcal{Z}_{(0,0)}&\rightarrow \frac{\left(q_1 q_2-1\right)^2 q_3}{q_1 q_2 \left(q_3-1\right)^2},\quad \mathcal{Z}_{(0,-\epsilon_{1})}\rightarrow 1,\quad
\mathcal{Z}_{(0,-\epsilon_{2})}\rightarrow 1,\quad
\mathcal{Z}_{(0,-\epsilon_{3})}\rightarrow -\frac{\left(q_1 q_2-q_3\right) \left(q_1 q_2 q_3^2-1\right)}{q_1 q_2
   \left(q_3-1\right)^2 \left(q_3+1\right)}
   \eea
which gives
\bea
\fq^{2}\mathcal{Z}^{\PT\tbar\JK}_{\bar{4};\,\Cbox\,\Cbox\,\Cbox}[2]\rightarrow \fq^{2}\left(3+\frac{\left(q_1 q_2-1\right) \left(q_1 q_2-q_3\right) q_3}{q_1 q_2
   \left(1-q_3^{2}\right) \left(1-q_3\right)}\right)=3xy+\frac{(1-t)(1-qt)}{(1-q)(1-q^{2})}x^{2}.
\eea
For level three, the limit gives
\bea
&\frac{1}{2}\mathcal{Z}_{(0,0,-\epsilon_{1})}\rightarrow \frac{\left(1-q_1 q_2\right) \sqrt{q_3}}{\sqrt{q_1} \sqrt{q_2} \left(q_3-1\right)} ,\quad \frac{1}{2}\mathcal{Z}_{(0,0,-\epsilon_{2})}\rightarrow \frac{\left(1-q_1 q_2\right) \sqrt{q_3}}{\sqrt{q_1} \sqrt{q_2} \left(q_3-1\right)} ,\\
&\frac{1}{2}\mathcal{Z}_{(0,0,-\epsilon_{3})}\rightarrow -\frac{\left(q_1 q_2-1\right) q_3^{3/2} \left(q_3-q_1 q_2\right)^2}{q_1^{3/2}
   q_2^{3/2} \left(q_3-1\right)^3 \left(q_3+1\right)^2},\quad  \mathcal{Z}_{(0,-\epsilon_{1},-2\epsilon_{1})}\rightarrow q_{4}^{-1/2},\quad \mathcal{Z}_{(0,-\epsilon_{2},-2\epsilon_{2})}\rightarrow q_{4}^{1/2},\\
& \mathcal{Z}_{(0,-\epsilon_{3},-2\epsilon_{3})}\rightarrow \frac{\sqrt{q_3} \left(q_1^3 q_2^3 q_3^3-q_3^3+q_1 q_2 \left(q_3^5+q_3+1\right)
   q_3-q_1^2 q_2^2 \left(q_3^5+q_3^4+1\right)\right)}{q_1^{3/2} q_2^{3/2}
   \left(q_3-1\right)^3 \left(q_3+1\right)^2 \left(q_3^2+q_3+1\right)} 
\eea
and indeed the sum matches with the third order of the Macdonald refined vertex. One can do the similar comparison for generic level, but we omit the discussion.

\subsection{DT/PT correspondence}\label{sec:DTPTcorrespondence}
The DT partition functions and PT partition functions are actually not independent and they are related by the so-called DT/PT correspondence \cite{Pandharipande:2007sq,Pandharipande:2007kc,Toda:2008ASPM,Toda:2010JAMS,Bridgeland:2011JAMS,Stoppa:2011BSMF,Toda2016HallAI,Kononov:2019fni,Jenne2020TheCP,Jenne:2021irh,Kuhn:2023koa}. For the rank-one K-theoretic case, the DT/PT correspondence was recently proven in \cite{Kuhn:2023koa}.

\begin{theorem}\label{thm:DTPT-correspondence}
In our notation, we have the following DT/PT correspondence:
\bea
\mathcal{Z}^{\DT\tbar\JK}_{\bar{4};\lambda\mu\nu}[\mathfrak{q},q_{1,2,3,4}]=\mathcal{Z}^{\DT\tbar\JK}_{\bar{4};\varnothing\varnothing\varnothing}[\mathfrak{q},q_{1,2,3,4}]\mathcal{Z}^{\PT\tbar\JK}_{\bar{4};\lambda\mu\nu}[\mathfrak{q},q_{1,2,3,4}].
\eea
Note that
\bea
\mathcal{Z}^{\DT\tbar\JK}_{\bar{4};\varnothing\varnothing\varnothing}[\mathfrak{q},q_{1,2,3,4}]=\mathcal{Z}^{\D6}_{\bar{4}}[\fq;q_{1,2,3,4}]
\eea
where the right hand side is \eqref{eq:D6U1PEformula}, \eqref{eq:D6partitionfunct}.
\end{theorem}
We have confirmed this up to three instantons $\mathfrak{q}^{3}$ for $|\lambda|+|\mu|+|\nu|\leq 5$. Taking the unrefined (section~\ref{sec:unrefined-vertex}), refined (section~\ref{sec:refined-vertex}), and the Macdonald refined limit (section~\ref{sec:Macdonald-vertex}), we still have the DT/PT correspondence and the right hand side becomes the generalized MacMahon functions (see \eqref{eq:MacMahon-funct}, \eqref{eq:refined-MacMahon-funct}, \eqref{eq:Macd-refined-MacMahon-funct}) times the unrefined, refined, Macdonald refined topological vertices (see Thm.~\ref{thm:PT3vertex-unref-corresp}, \ref{thm:PT3vertex-ref-corresp}, \ref{thm:PT3vertex-Macdref-corresp}).

Let us see this explicitly for the one-instanton level.
\paragraph{One-leg}
For the example $(\lambda,\mu,\nu)=(\varnothing,\varnothing,\Bbox)$, the configurations for the DT side is $\fra+\eps_1,\fra+\eps_2$, which gives
\bea
\mathcal{Z}^{\DT\tbar\JK}_{\bar{4};\,\varnothing\varnothing\,\Bbox}[1]=\frac{\left(q_1 q_2-1\right) \left(q_1^2 q_3-1\right) \left(q_2 q_3-1\right)}{\left(q_1-1\right) \sqrt{q_1}
   \left(q_1-q_2\right) \sqrt{q_2} \left(q_3-1\right) \sqrt{q_3}}-\frac{\left(q_1 q_2-1\right) \left(q_1
   q_3-1\right) \left(q_2^2 q_3-1\right)}{\sqrt{q_1} \left(q_1-q_2\right) \left(q_2-1\right) \sqrt{q_2}
   \left(q_3-1\right) \sqrt{q_3}},
\eea
where the first term and second term correspond to the configurations $\fra+\eps_1$, $\fra+\eps_2$, respectively.

Using 
\bea\label{eq:rank1-DT3_1instanton}
\mathcal{Z}^{\DT\tbar\JK}_{\bar{4};\,\varnothing\,\varnothing\,\varnothing}[1]=-\frac{\left(q_1 q_2-1\right) \left(q_1 q_3-1\right) \left(q_2 q_3-1\right)}{\left(q_1-1\right) \sqrt{q_1}
   \left(q_2-1\right) \sqrt{q_2} \left(q_3-1\right) \sqrt{q_3}}
\eea
and \eqref{eq:PToneleg-JK-partfunct}, we indeed have
\bea
\mathcal{Z}^{\DT\tbar\JK}_{\bar{4};\,\varnothing\,\varnothing\,\Bbox}[1]=\mathcal{Z}^{\DT\tbar\JK}_{\bar{4};\,\varnothing\,\varnothing\,\varnothing}[1]+\mathcal{Z}^{\PT\tbar\JK}_{\bar{4};\,\varnothing\,\varnothing\,\Bbox}[1].
\eea

\paragraph{Two-legs}
For the example $(\lambda,\mu,\nu)=(\Bbox,\Bbox,\varnothing)$, the DT side is computed as
\bea
\mathcal{Z}^{\DT\tbar\JK}_{\bar{4};\,\Bbox\,\Bbox\,\varnothing}[1]&=\frac{\left(q_1^2 q_2^2-1\right) \left(q_1 q_3-1\right) \left(q_2 q_3-1\right)}{\left(q_1-1\right)
   \sqrt{q_1} \left(q_2-1\right) \sqrt{q_2} \left(q_1 q_2-q_3\right) \sqrt{q_3}}\\
   &-\frac{\left(q_1
   q_2-1\right){}^2 \left(q_3+1\right) \left(q_1 q_3-1\right) \left(q_2 q_3-1\right)}{\left(q_1-1\right)
   \sqrt{q_1} \left(q_2-1\right) \sqrt{q_2} \left(q_1 q_2-q_3\right) \left(q_3-1\right) \sqrt{q_3}}
\eea
where the first term corresponds to the configuration with a box at $\fra+\eps_{1}+\eps_2$, while the second one corresponds to the case with a box at $\fra+\eps_3$. Using \eqref{eq:PTtwoleg-JK-partfunct}, we indeed have
\bea
\mathcal{Z}^{\DT\tbar\JK}_{\bar{4};\,\Bbox\,\Bbox\,\varnothing}[1]=\mathcal{Z}^{\DT\tbar\JK}_{\bar{4};\,\varnothing\,\varnothing\,\varnothing}[1]+\mathcal{Z}^{\PT\tbar\JK}_{\bar{4};\,\Bbox\,\Bbox\,\varnothing}[1].
\eea

\paragraph{Three-legs}
For the example $(\lambda,\mu,\nu)=(\Bbox,\Bbox,\Bbox)$, the DT side is computed as
\bea
\mathcal{Z}^{\DT\tbar\JK}_{\bar{4};\,\Bbox\,\Bbox\,\Bbox}[1]&=-\frac{\left(q_1^2 q_2-1\right) \left(q_1 q_2^2-1\right) \left(q_3-1\right) \left(q_1 q_3-1\right) \left(q_2
   q_3-1\right)}{\left(q_1-1\right) \sqrt{q_1} \left(q_2-1\right) \sqrt{q_2} \left(q_1 q_2-1\right)
   \left(q_1-q_3\right) \left(q_2-q_3\right) \sqrt{q_3}}\\
   &+\frac{\left(q_2-1\right) \left(q_1 q_2-1\right)
   \left(q_1^2 q_3-1\right) \left(q_1 q_3^2-1\right) \left(q_2 q_3-1\right)}{\left(q_1-1\right) \sqrt{q_1}
   \left(q_1-q_2\right) \sqrt{q_2} \left(q_2-q_3\right) \left(q_3-1\right) \sqrt{q_3} \left(q_1
   q_3-1\right)}\\
   &-\frac{\left(q_1-1\right) \left(q_1 q_2-1\right) \left(q_1 q_3-1\right) \left(q_2^2
   q_3-1\right) \left(q_2 q_3^2-1\right)}{\sqrt{q_1} \left(q_1-q_2\right) \left(q_2-1\right) \sqrt{q_2}
   \left(q_1-q_3\right) \left(q_3-1\right) \sqrt{q_3} \left(q_2 q_3-1\right)}
\eea
where each term corresponds to the configuration with a box at $\fra+\eps_{12}$, $\fra+\eps_{13}$, and $\fra+\eps_{23}$, respectively. Using \eqref{eq:3leglevel1}, we indeed have
\bea
\mathcal{Z}^{\DT\tbar\JK}_{\bar{4};\,\Bbox\,\Bbox\,\Bbox}[1]=\mathcal{Z}^{\DT\tbar\JK}_{\bar{4};\,\varnothing\,\varnothing\,\varnothing}[1]+\mathcal{Z}^{\PT\tbar\JK}_{\bar{4};\,\Bbox\,\Bbox\,\Bbox}[1].
\eea

\paragraph{Wall-crossing}The fact that the result differs whether we use $\eta_{0}$ or $\tilde{\eta}_{0}$ to evaluate the contour integral is a consequence of the wall-crossing phenomenon. Although we postpone a detailed discussion of this for future work, let us briefly discuss how to see that this occurs using the contour integrals. The origin of the wall-crossing is the existence of nontrivial residues at the \textit{asymptotics} \cite{Hori:2014tda}. In other words, the residues at $\phi_{i}\rightarrow \pm \infty$ give the nontrivial wall-crossing factors. Conversely, if we do not have nontrivial residues at the asymptotics, no wall-crossing occurs.

To be concrete, let us consider the one-instanton contribution for the one-leg case. For the one-instanton case, the D0-D0 term does not appear and the contour integral is just an integration with respect to one variable. It will be convenient to use the multiplicative notation. The contour integral is written as
\bea
\oint \frac{dz}{2\pi i z}g(z),\quad g(z)=\mathcal{G}\times q_{4}^{-1/2}\frac{(z-q_{14})(z-q_{24})(z-q_{34}^{-1})}{(z-q_{1})(z-q_{2})(z-q_{3}^{-1})}
\eea
where we denoted $z=e^{\phi_1-\fra}$ and
\bea
\mathcal{G}=\frac{(1-q_{12})(1-q_{13})(1-q_{23})}{(1-q_{1})(1-q_{2})(1-q_{3})(1-q_{4}^{-1})}.
\eea
We can take the contour to be the unit circle traversed counter-clockwise. Assuming the analytic region $|q_{1,2,3}|<1$, we can take the residues of the poles $z=0,q_{1},q_{2}$ inside the region $|z|<1$:
\bea
\oint_{|z|=1} \frac{dz}{2\pi i z}g(z)&=\underset{z=0}{\Res}\left[g(z)\frac{dz}{z}\right]+\underset{z=q_{1}}{\Res}\left[g(z)\frac{dz}{z}\right]+\underset{z=q_{2}}{\Res}\left[g(z)\frac{dz}{z}\right]\\
&=\underset{z=0}{\Res}\left[g(z)\frac{dz}{z}\right]+\oint_{\eta_0} \frac{dz}{2\pi i z}g(z)
\eea
where the second and third term of the first line correspond to the JK-residue with the reference vector $\eta=\eta_{0}$.

Instead, we can also take the residues of the poles in the region $|z|>1$:
\bea
\oint_{|z|=1} \frac{dz}{2\pi i z}g(z)&=-\underset{z=\infty}{\Res}\left[g(z)\frac{dz}{z}\right]-\underset{z=q_{3}^{-1}}{\Res}\left[g(z)\frac{dz}{z}\right]\\
&=-\underset{z=\infty}{\Res}\left[g(z)\frac{dz}{z}\right]+\oint_{\tilde{\eta}_0} \frac{dz}{2\pi i z}g(z).
\eea
Since the result should be the same, we obtain
\bea
\oint_{\eta_0} \frac{dz}{2\pi i z}g(z)-\oint_{\tilde{\eta}_0} \frac{dz}{2\pi i z}g(z)=-\underset{z=\infty}{\Res}\left[g(z)\frac{dz}{z}\right]-\underset{z=0}{\Res}\left[g(z)\frac{dz}{z}\right].
\eea
The difference between the DT-side and the PT-side comes from the difference between the residues at the asymptotics $z=0,\infty$. The residues are computed as
\bea
\underset{z=0}{\Res}\left[g(z)\frac{dz}{z}\right]&=\lim_{z\rightarrow 0}g(z)=\mathcal{G}\times q_{4}^{+1/2},\\
\underset{z=\infty}{\Res}\left[g(z)\frac{dz}{z}\right]&=-\underset{w=0}{\Res}\left[g(w^{-1})\frac{dw}{w}\right]=-\lim_{w\rightarrow 0}g(w^{-1})=-\mathcal{G}\times q_{4}^{-1/2},
\eea
where we used $z=w^{-1}$ to evaluate the residue at $z=\infty$ and then consider $w=0$. We then have
\bea
\oint_{\eta_0} \frac{dz}{2\pi i z}g(z)-\oint_{\tilde{\eta}_0} \frac{dz}{2\pi i z}g(z)&=\lim_{z\rightarrow \infty}g(z)-\lim_{z\rightarrow 0}g(z)\\
&=-q_{4}^{1/2}\frac{(1-q_{12})(1-q_{13})(1-q_{23})}{(1-q_{1})(1-q_{2})(1-q_{3})}
\eea
which is indeed the one-instanton contribution of the D6-D0 partition function \eqref{eq:rank1-DT3_1instanton}.

Let us do the similar analysis for general boundary conditions $\vec{Y}=(\lambda,\mu,\nu)$. The asymptotics at $\phi_1\rightarrow \pm \infty$ are only necessary to derive the wall-crossing formula. Since this setup preserves four supersymmetries, it is enough to study the asymptotics of the combinations $\mathcal{Z}^{{\D6}_{\bar{4}}\tbar\D0}(\#,\phi_1)$ and $\mathcal{Z}^{\overline{\D6}_{\bar{4}}\tbar\D0}(\#,\phi_{I})$
\bea
&\mathcal{Z}^{\D6_{\bar{4}}\tbar\D0}(\#,\phi_1)\xrightarrow{\phi_1\rightarrow \infty} q_{4}^{-1/2},\quad 
\mathcal{Z}^{\overline{\D6}_{\bar{4}}\tbar\D0}(\#,\phi_1)\xrightarrow{\phi_1\rightarrow \infty} q_{4}^{+1/2} \\
&\mathcal{Z}^{\D6_{\bar{4}}\tbar\D0}(\#,\phi_1)\xrightarrow{\phi_1\rightarrow -\infty} q_{4}^{+1/2},\quad \mathcal{Z}^{\overline{\D6}_{\bar{4}}\tbar\D0}(\#,\phi_1)\xrightarrow{\phi_1\rightarrow -\infty} q_{4}^{-1/2}.
\eea
The asymptotics of the framing node contribution \eqref{eq:DTflavornode-def} are obtained as
\bea
\lim_{\phi_1\rightarrow -\infty}\mathcal{Z}^{\D6_{\bar{4}}\tbar\D2\tbar\D0}_{\DT;\lambda\mu\nu}(\fra,\phi_1)&=q_{4}^{\frac{1}{2}\left(|s(\vec{Y})|-|p_{1}(\vec{Y})|-2|p_{2}(\vec{Y})|\right)}=q_{4}^{1/2},\\
\lim_{\phi_1\rightarrow +\infty}\mathcal{Z}^{\D6_{\bar{4}}\tbar\D2\tbar\D0}_{\DT;\lambda\mu\nu}(\fra,\phi_1)&=q_{4}^{-\frac{1}{2}\left(|s(\vec{Y})|-|p_{1}(\vec{Y})|-2|p_{2}(\vec{Y})|\right)}=q_{4}^{-1/2}
\eea
where in the last line, we used\footnote{We will not give a proof of this relation, but one can easily confirm it for examples.}
\bea
|s(\vec{Y})|-|p_{1}(\vec{Y})|-2|p_{2}(\vec{Y})|=1.
\eea
Thus, the wall-crossing formula for the one-instanton level is
\bea
\mathcal{G}\times \left(\lim_{\phi_1\rightarrow \infty}\mathcal{Z}^{\D6_{\bar{4}}\tbar\D2\tbar\D0}_{\DT;\lambda\mu\nu}(\fra,\phi_1)-\lim_{\phi_1\rightarrow -\infty}\mathcal{Z}^{\D6_{\bar{4}}\tbar\D2\tbar\D0}_{\DT;\lambda\mu\nu}(\fra,\phi_1)\right)=-q_{4}^{1/2}\frac{(1-q_{12})(1-q_{13})(1-q_{23})}{(1-q_{1})(1-q_{2})(1-q_{3})}.
\eea

For higher levels, we need to deal with multi-dimensional contour integrals, we postpone a detailed analysis for future work.

\subsection{Higher rank DT3 and PT3 counting}\label{sec:higherrankDTPT}
Up to the previous section, we only focused on the case with a single D6$_{\bar{4}}$-brane. Including other D6-branes gives generalizations of the DT and PT counting. In this section, we discuss such generalizations and we will see that the DT/PT correspondence also have natural generalizations. See \cite{Toda2016HallAI,Kool:2016opl,Gholampour:2017lkc,Jardim2025HigherRD} for mathematical literatures discussing higher rank DT/PT countings.

We mainly discuss the following three different generalizations.
\begin{enumerate}
    \item Multiple D6-branes spanning the same $\mathbb{C}^{3}$ space (section~\ref{sec:rankNDTPT}). In particular, we consider the setup with $n$ D6$_{\bar{4}}$-branes. For this case, the gauge theory on the D6-branes will be a $U(n)$ gauge theory and we obtain the rank $n$ DT/PT partition functions. 
    \item Tetrahedron instanton generalizations (section~\ref{sec:tetrahedronDTPT}). We can also include D6-branes spanning different subspaces. We consider $n_{\bar{a}}$ D6$_{\bar{a}}$-branes with boundary conditions.
    \item Supergroup generalizations (section~\ref{sec:supergroupDTPT}). As discussed in \cite{Kimura:2023bxy}, we can include D6-branes giving anti-fundamental chiral multiplets, which we denote $\overline{\D6}$-branes, in the supersymmetric quantum mechanics. In particular, we consider $n$ D6$_{\bar{4}}$-branes and $m$ $\overline{\text{D6}}_{\bar{4}}$-branes giving a $U(n|m)$ gauge theory. We also propose tetrahedron instantons generalizations. 
\end{enumerate}

In the following discussion, the following function 
\bea
\text{MF}[\fq,\mu;q_{1,2,3,4}]=\PE\left[\frac{[q_{14}][q_{24}][q_{34}]}{[q_{1}][q_{2}][q_{3}][q_{4}]}\frac{[\mu]}{[\fq\mu^{-1/2}][\fq \mu^{1/2}]}\right]
\eea
will play a significant role. We shortly write this as $\text{MF}[\mu]$ because the $\mu$-dependence is only important. Physically, this is the so-called magnificent four partition function \cite{Nekrasov:2017cih,Nekrasov:2018xsb}. Note that the tetrahedron instanton partition function \eqref{eq:tetrahedronPE} is just a specialization of it.

All of the conjectures in this section have been confirmed by explicit computations for various examples at low levels.

\subsubsection{Parallel D6-branes}\label{sec:rankNDTPT}
Let us first consider the setup with $n$ D6$_{\bar{4}}$-branes (see also Thm.~\ref{thm:D6vacuum}). The D6 partition function is obtained by changing the framing node contribution to 
\bea
\prod_{\alpha=1}^{n}\mathcal{Z}^{\D6_{\bar{4}}\tbar\D0}(\fra_{\alpha},\phi_I).
\eea
The partition function is
\bea
\mathcal{Z}^{\D6}_{\bar{4},n}[\fq,q_{1,2,3,4}]&=\sum_{k=0}^{\infty}\fq^{k}\mathcal{Z}^{\D6}_{\bar{4},n}[k],
\eea
where the $k$-instanton sector is
\bea
\mathcal{Z}^{\D6}_{\bar{4},n}[k]&=\frac{1}{k!}\left(\frac{\sh(-\epsilon_{14,24,34})}{\sh(-\epsilon_{1,2,3,4})}\right)^{k}\oint_{\eta_{0}}\prod_{I=1}^{k}\frac{d\phi_{I}}{2\pi i }\prod_{I=1}^{k}\prod_{\alpha=1}^{n}\mathcal{Z}^{\D6_{\bar{4}}\tbar\D0}(\fra_{\alpha},\phi_I)\prod_{I<J}\mathcal{Z}^{\D0\tbar\D0}(\phi_{I},\phi_{J}).
\eea
After the JK-residue prescription, the poles picked up are multiple plane partitions $\vec{\pi}=(\pi^{(\alpha)})_{\alpha=1}^{n}$ where the origin is $\fra_{\alpha}$ respectively:
\bea
\{\phi_I\}_{I=1}^{k}\longrightarrow \{c_{\bar{4},\fra_{\alpha}}(\cube)\mid \cube \in \pi^{(\alpha)}\}.
\eea
The rank $n$ D6 partition function has a PE-formula (see \eqref{eq:tetrahedronPE})
\bea
\mathcal{Z}^{\D6}_{\bar{4},n}[\fq,q_{1,2,3,4}]=\PE\left[\frac{[q_{14}][q_{24}][q_{34}]}{[q_{1}][q_{2}][q_{3}][q_{4}]}\frac{[q_{4}^{n}]}{[\fq q_{4}^{-n/2}][\fq q_{4}^{n/2}]}\right]
\eea
and it does not depend on the Coulomb branch parameters $\fra_{\alpha}$.

Let us next consider the setup with boundary conditions coming from the D2$_{1,2,3}$-branes. We denote the boundary conditions as $(\vec{\lambda},\vec{\mu},\vec{\nu})$, where $\vec{\lambda}=(\lambda^{(\alpha)})_{\alpha=1}^{n}$, $\vec{\mu}=(\mu^{(\alpha)})_{\alpha=1}^{n}$, $\vec{\nu}=(\nu^{(\alpha)})_{\alpha=1}^{n}$. The rank $n$ DT3 partition function is naturally defined as (see also Def.~\ref{def:DT3vertex-rank1-JKintegral}, \ref{thm:DTvertex-JKresidue})
\bea
\,&\mathcal{Z}^{\DT\tbar\JK}_{\bar{4};\vec{\lambda}\vec{\mu}\vec{\nu}}[\mathfrak{q},q_{1,2,3,4}]=\sum_{k=0}^{\infty}\mathfrak{q}^{k}\mathcal{Z}^{\DT\tbar\JK}_{\bar{4};\vec{\lambda}\vec{\mu}\vec{\nu}}[k],
\eea
where the $k$-instanton sector is given as
\bea\label{eq:rankN-DT3contourint}
\,&\mathcal{Z}^{\DT\tbar\JK}_{\bar{4};\vec{\lambda}\vec{\mu}\vec{\nu}}[k]=\frac{1}{k!}\left(\frac{\sh(-\epsilon_{14,24,34})}{\sh(-\epsilon_{1,2,3,4})}\right)^{k}\oint_{\eta_{0}} \prod_{I=1}^{k}\frac{d\phi_{I}}{2\pi i}\prod_{I=1}^{k}\prod_{\alpha=1}^{n}\mathcal{Z}^{\D6_{\bar{4}}\tbar\D2\tbar\D0}_{\DT;\lambda^{(\alpha)}\mu^{(\alpha)}\nu^{(\alpha)}}(\fra_{\alpha},\phi_{I})\prod_{I<J}^{k}\mathcal{Z}^{\D0\tbar\D0}(\phi_{I},\phi_{J})
\eea
and $\vec{v}=(v_{\alpha})_{\alpha=1}^{n}$, $v_{\alpha}=e^{\fra_{\alpha}}$. The poles are classified by $n$-tuple plane partitions with the boundary conditions $(\lambda^{(\alpha)},\mu^{(\alpha)},\nu^{(\alpha)})$ (elements of $\mathcal{DT}_{\lambda^{(\alpha)}\mu^{(\alpha)}\nu^{(\alpha)}}$) and the poles correspond to the boxes that one can place. Again, the origin of the plane partitions are $\fra_{\alpha}$. Extra bifundamental contributions arise in the partition function and generally it depends on the Coulomb branch paramters $\vec{v}$.

Note that when all of the boundary conditions are trivial, we have
\bea
\mathcal{Z}^{\DT\tbar\JK}_{\bar{4};\vec{\varnothing}\vec{\varnothing}\vec{\varnothing}}[\mathfrak{q},q_{1,2,3,4}]=\mathcal{Z}^{\D6}_{\bar{4},n}[\fq,q_{1,2,3,4}].
\eea

The rank $n$ PT3 partition function can be defined similarly by changing the reference vector from $\eta_{0}$ to $\tilde{\eta}_{0}$ while keeping the integrand (see also Def.~\ref{def:PTvertex-JKresidue}, \ref{thm:PTvertex-expansion}):
\bea
\,&\mathcal{Z}^{\PT\tbar\JK}_{\bar{4};\vec{\lambda}\vec{\mu}\vec{\nu}}[\mathfrak{q},q_{1,2,3,4}]=\sum_{k=0}^{\infty}\mathfrak{q}^{k}\mathcal{Z}^{\PT\tbar\JK}_{\bar{4};\vec{\lambda}\vec{\mu}\vec{\nu}}[k],
\eea
where the $k$-instanton sector is given as
\bea\label{eq:rankN-PT3contourint}
\,&\mathcal{Z}^{\PT\tbar\JK}_{\bar{4};\vec{\lambda}\vec{\mu}\vec{\nu}}[k]=\frac{1}{k!}\left(\frac{\sh(-\epsilon_{14,24,34})}{\sh(-\epsilon_{1,2,3,4})}\right)^{k}\oint_{\tilde{\eta}_{0}} \prod_{I=1}^{k}\frac{d\phi_{I}}{2\pi i}\prod_{I=1}^{k}\prod_{\alpha=1}^{n}\mathcal{Z}^{\D6_{\bar{4}}\tbar\D2\tbar\D0}_{\DT;\lambda^{(\alpha)}\mu^{(\alpha)}\nu^{(\alpha)}}(\fra_{\alpha},\phi_{I})\prod_{I<J}^{k}\mathcal{Z}^{\D0\tbar\D0}(\phi_{I},\phi_{J}).
\eea
The poles give higher rank PT3 box configurations. However, one should note that when all of the three-legs are nontrivial, higher order poles appear and when taking the residue, derivatives act on the integrand. Since the derivatives also act on the contributions coming from the other D6-branes, the partition function becomes much more complicated.\footnote{We omit explicit computations in this paper, but one can easily apply the formalism discussed in this paper and study the pole structure.} Moreover, generally, the partition function depends on the Coulomb branch parameters $\vec{v}$.

The DT, PT partition functions are connected by the DT/PT correspondence (see Thm.~\ref{thm:DTPT-correspondence}). 
\begin{conjecture}\label{thm:rankn-DTPT}
    The rank $n$ DT/PT correspondence is
\bea
\mathcal{Z}^{\DT\tbar\JK}_{\bar{4};\vec{\lambda}\vec{\mu}\vec{\nu}}[\mathfrak{q},q_{1,2,3,4}]&=\mathcal{Z}^{\DT\tbar\JK}_{\bar{4};\vec{\varnothing}\vec{\varnothing}\vec{\varnothing}}[\mathfrak{q},q_{1,2,3,4}]\mathcal{Z}^{\PT\tbar\JK}_{\bar{4};\vec{\lambda}\vec{\mu}\vec{\nu}}[\mathfrak{q},q_{1,2,3,4}]\\
&=\mathcal{Z}^{\D6}_{\bar{4},n}[\fq;q_{1,2,3,4}]\mathcal{Z}^{\PT\tbar\JK}_{\bar{4};\vec{\lambda}\vec{\mu}\vec{\nu}}[\mathfrak{q},q_{1,2,3,4}]\\
&=\text{MF}[q_{4}^{n}]\mathcal{Z}^{\PT\tbar\JK}_{\bar{4};\vec{\lambda}\vec{\mu}\vec{\nu}}[\mathfrak{q},q_{1,2,3,4}].
\eea    
\end{conjecture}

\subsubsection{Tetrahedron instantons}\label{sec:tetrahedronDTPT}
Let us next include other D6-branes with different supports. As discussed in section~\ref{sec:gaugeorigami-index}, to obtain the tetrahedron instanton, we need to modify the flavor node contribution to
\bea
\prod_{a\in\four}\prod_{\alpha=1}^{n_{\bar{a}}}\mathcal{Z}^{\D6_{\bar{a}}\tbar\D0}(\fra_{\bar{a},\alpha},\phi_I).
\eea
The tetrahedron instanton partition function is (see also Thm.~\ref{thm:D6vacuum}):
\bea
\mathcal{Z}^{\D6}_{n_{\bar{1}},n_{\bar{2}},n_{\bar{3}},n_{\bar{4}}}[\fq,q_{1,2,3,4}]&=\sum_{k=0}^{\infty}\fq^{k}\mathcal{Z}^{\D6}_{n_{\bar{1}},n_{\bar{2}},n_{\bar{3}},n_{\bar{4}}}[k],
\eea
where
\bea
\mathcal{Z}^{\D6}_{n_{\bar{1}},n_{\bar{2}},n_{\bar{3}},n_{\bar{4}}}[k]&=\frac{1}{k!}\left(\frac{\sh(-\epsilon_{14,24,34})}{\sh(-\epsilon_{1,2,3,4})}\right)^{k}\oint_{\eta_{0}}\prod_{I=1}^{k}\frac{d\phi_{I}}{2\pi i }\prod_{I=1}^{k}\prod_{a\in\four}\prod_{\alpha=1}^{n_{\bar{a}}}\mathcal{Z}^{\D6_{\bar{a}}\tbar\D0}(\fra_{\bar{a},\alpha},\phi_I)\prod_{I<J}\mathcal{Z}^{\D0\tbar\D0}(\phi_{I},\phi_{J}).
\eea
The poles are classified by multiple plane partitions $\underline{\vec{\pi}}=(\vec{\pi}_{a})_{a\in\four}=(\pi^{(\alpha)}_{\bar{a}})_{a\in\four}^{\alpha=1,\ldots,n_{\bar{a}}}$:
\bea
\{\phi_I\}_{I=1}^{k}\longrightarrow \{c_{\bar{a},\fra_{\bar{a},\alpha}}(\cube)\mid \cube\in\pi_{\bar{a}}^{(\alpha)}\}.
\eea
Note that the plane partitions now span different three-dimensional subspaces. We choose an orientation such that the three coordinates are $(x_{a+1},x_{a+2},x_{a+3})$ for $\bar{a}$, where the indices are understood modulo $4$. The partition function does not depend on the Coulomb branch parameters and has a PE-formula \eqref{eq:tetrahedronPE}.

The tetrahedron DT partition functions with boundary conditions are straightforward to generalize. We denote the boundary conditions for the three-legs of the $n_{\bar{a}}$ D6$_{\bar{a}}$-branes as $(\vec{\lambda}_{\bar{a}},\vec{\mu}_{\bar{a}},\vec{\nu}_{\bar{a}})$, where $\vec{\lambda}_{\bar{a}}=(\lambda_{\bar{a}}^{(\alpha)})_{\alpha=1}^{n}$, $\vec{\mu}_{\bar{a}}=(\mu_{\bar{a}}^{(\alpha)})_{\alpha=1}^{n}$, $\vec{\nu}_{\bar{a}}=(\nu_{\bar{a}}^{(\alpha)})_{\alpha=1}^{n}$ and the orientations are chosen as Fig.~\ref{fig:minimal-pp}. Namely, the arm of $\lambda_{\bar{a},\alpha},\mu_{\bar{a},\alpha},\nu_{\bar{a},\alpha}$ extends in the $a+3,a+1,a+2$-axes, respectively. We collectively denote the boundary Young diagrams as $(\underline{\vec{\lambda}},\underline{\vec{\mu}},\underline{\vec{\nu}})$ where $\underline{\vec{\lambda}}=(\vec{\lambda}_{a})_{a\in\four},\underline{\vec{\mu}}=(\vec{\mu}_{a})_{a\in\four},\underline{\vec{\nu}}=(\vec{\nu}_{a})_{a\in\four}$. 

The tetrahedron DT partition function is defined as
\bea
\,&\mathcal{Z}^{\DT\tbar\JK}_{\underline{\vec{\lambda}}\,\underline{\vec{\mu}}\,\underline{\vec{\nu}}}[\mathfrak{q},q_{1,2,3,4}]=\sum_{k=0}^{\infty}\mathfrak{q}^{k}\mathcal{Z}^{\DT\tbar\JK}_{\underline{\vec{\lambda}}\,\underline{\vec{\mu}}\,\underline{\vec{\nu}}}[k],
\eea
where the $k$-instanton sector is given as
\bea\label{eq:tetrahedron-DT3contourint}
\,&\mathcal{Z}^{\DT\tbar\JK}_{\underline{\vec{\lambda}}\,\underline{\vec{\mu}}\,\underline{\vec{\nu}}}[k]=\frac{1}{k!}\left(\frac{\sh(-\epsilon_{14,24,34})}{\sh(-\epsilon_{1,2,3,4})}\right)^{k}\oint_{\eta_{0}} \prod_{I=1}^{k}\frac{d\phi_{I}}{2\pi i}\prod_{I=1}^{k}\prod_{a\in\four}\prod_{\alpha=1}^{n_{\bar{a}}}\mathcal{Z}^{\D6_{\bar{a}}\tbar\D2\tbar\D0}_{\DT;\lambda_{\bar{a}}^{(\alpha)}\mu_{\bar{a}}^{(\alpha)}\nu_{\bar{a}}^{(\alpha)}}(\fra_{\bar{a},\alpha},\phi_{I})\prod_{I<J}^{k}\mathcal{Z}^{\D0\tbar\D0}(\phi_{I},\phi_{J})
\eea
and $\underline{\vec{v}}=(v_{\bar{a},\alpha})$, $v_{\bar{a},\alpha}=e^{\fra_{\bar{a},\alpha}}$. The framing node contribution $\mathcal{Z}^{\D6_{\bar{a}}\tbar\D2\tbar\D0}_{\DT;\lambda_{\bar{a}}^{(\alpha)}\mu_{\bar{a}}^{(\alpha)}\nu_{\bar{a}}^{(\alpha)}}(\fra_{\bar{a},\alpha},\phi_{I})$ is obtained by the permutation of the parameters $\eps_1\rightarrow \eps_2\rightarrow \eps_3\rightarrow \eps_4\rightarrow \eps_1$. 

The poles are classified by multiple plane partitions with boundary conditions, but this time, we have multiple directions. We note that similar to the previous section, the partition function generally depends on the Coulomb branch parameters. Again, when the boundary conditions are all trivial, we simply obtain the tetrahedron instanton partition function
\bea
\mathcal{Z}^{\DT\tbar\JK}_{\underline{\vec{\varnothing}}\,\underline{\vec{\varnothing}}\,\underline{\vec{\varnothing}}}[\mathfrak{q},q_{1,2,3,4}]=\mathcal{Z}^{\D6}_{n_{\bar{1}},n_{\bar{2}},n_{\bar{3}},n_{\bar{4}}}[\fq,q_{1,2,3,4}].
\eea

For the tetrehedron PT partition function, we simply need to flip the sign of the reference vector. We then have
\bea
\,&\mathcal{Z}^{\PT\tbar\JK}_{\underline{\vec{\lambda}}\,\underline{\vec{\mu}}\,\underline{\vec{\nu}}}[\mathfrak{q},q_{1,2,3,4}]=\sum_{k=0}^{\infty}\mathfrak{q}^{k}\mathcal{Z}^{\PT\tbar\JK}_{\underline{\vec{\lambda}}\,\underline{\vec{\mu}}\,\underline{\vec{\nu}}}[k],
\eea
where the $k$-instanton sector is given as
\bea\label{eq:tetrahedron-PT3contourint}
\,&\mathcal{Z}^{\PT\tbar\JK}_{\underline{\vec{\lambda}}\,\underline{\vec{\mu}}\,\underline{\vec{\nu}}}[k]=\frac{1}{k!}\left(\frac{\sh(-\epsilon_{14,24,34})}{\sh(-\epsilon_{1,2,3,4})}\right)^{k}\oint_{\tilde{\eta}_{0}} \prod_{I=1}^{k}\frac{d\phi_{I}}{2\pi i}\prod_{I=1}^{k}\prod_{a\in\four}\prod_{\alpha=1}^{n_{\bar{a}}}\mathcal{Z}^{\D6_{\bar{a}}\tbar\D2\tbar\D0}_{\DT;\lambda_{\bar{a}}^{(\alpha)}\mu_{\bar{a}}^{(\alpha)}\nu_{\bar{a}}^{(\alpha)}}(\fra_{\bar{a},\alpha},\phi_{I})\prod_{I<J}^{k}\mathcal{Z}^{\D0\tbar\D0}(\phi_{I},\phi_{J}).
\eea
The poles give tetrahedron generalizations of the PT3 configurations.

Even for this tetrahedron instanton generalization, we have the DT/PT correspondence.
\begin{conjecture}\label{thm:tetrahedron-DTPT}
The tetrahedron DT/PT correspondence is 
\bea
\mathcal{Z}^{\DT\tbar\JK}_{\underline{\vec{\lambda}}\,\underline{\vec{\mu}}\,\underline{\vec{\nu}}}[\mathfrak{q},q_{1,2,3,4}]&=
\mathcal{Z}^{\DT\tbar\JK}_{\underline{\vec{\varnothing}}\,\underline{\vec{\varnothing}}\,\underline{\vec{\varnothing}}}[\mathfrak{q},q_{1,2,3,4}]\mathcal{Z}^{\PT\tbar\JK}_{\underline{\vec{\lambda}}\,\underline{\vec{\mu}}\,\underline{\vec{\nu}}}[\mathfrak{q},q_{1,2,3,4}]\\
&=\mathcal{Z}^{\D6}_{n_{\bar{1}},n_{\bar{2}},n_{\bar{3}},n_{\bar{4}}}[\fq,q_{1,2,3,4}]\mathcal{Z}^{\PT\tbar\JK}_{\underline{\vec{\lambda}}\,\underline{\vec{\mu}}\,\underline{\vec{\nu}}}[\mathfrak{q},q_{1,2,3,4}]\\
&=\text{MF}\left[\prod_{a}q_{a}^{n_{\bar{a}}}\right]\mathcal{Z}^{\PT\tbar\JK}_{\underline{\vec{\lambda}}\,\underline{\vec{\mu}}\,\underline{\vec{\nu}}}[\mathfrak{q},q_{1,2,3,4}].
\eea
    
\end{conjecture}

\subsubsection{Supergroup generalizations}\label{sec:supergroupDTPT}
As one can observe from the previous sections, the PT counting and the DT/PT correspondence are phenomena coming from the flip of the sign of the reference vector, i.e., the wall-crossing phenomena. At the level of the supersymmetric quantum mechanics, the fundamental/anti-fundamental characteristic of the chiral multiplets play an important role. The introduction of additional anti D6-branes induces extra anti-fundamental chiral multiplets and the flip of the reference vector will pick such contributions. Thus, generalizations of the partition functions and the DT/PT correspondence should naturally occur. In this section, we show that generalizations of DT/PT correspondence holds.

\paragraph{Pure supergroup gauge theory }
Let us start from the setup with $n$ D6$_{\bar{4}}$-branes and $m$ $\overline{\D6}_{\bar{4}}$-branes which we call the pure $U(n|m)$ gauge theory. The framing node contribution is modified to
\bea
\prod_{\alpha=1}^{n}\mathcal{Z}^{\D6_{\bar{4}}\tbar\D0}(\fra_{\alpha},\phi_I)
\prod_{\beta=1}^{m}\mathcal{Z}^{\overline{\D6}_{\bar{4}}\tbar\D0}(\frb_{\beta}-\eps_4,\phi_I),
\eea
where we shifted the Coulomb branch parameters of the anti D6-branes for convenience.

The partition function is
\bea
\mathcal{Z}^{\D6,+}_{\bar{4},n|m}[\fq,q_{1,2,3,4}]&=\sum_{k=0}^{\infty}\fq^{k}\mathcal{Z}^{\D6,+}_{\bar{4},n|m}[k],
\eea
where the $k$-instanton sector is
\bea
\mathcal{Z}^{\D6,+}_{\bar{4},n|m}[k]&=\frac{1}{k!}\left(\frac{\sh(-\epsilon_{14,24,34})}{\sh(-\epsilon_{1,2,3,4})}\right)^{k}\oint_{\eta_{0}}\prod_{I=1}^{k}\frac{d\phi_{I}}{2\pi i }\prod_{I=1}^{k}\prod_{\alpha=1}^{n}\mathcal{Z}^{\D6_{\bar{4}}\tbar\D0}(\fra_{\alpha},\phi_I)\\
&\times\prod_{\beta=1}^{m}\mathcal{Z}^{\overline{\D6}_{\bar{4}}\tbar\D0}(\frb_{\beta}-\eps_4,\phi_I)\prod_{I<J}\mathcal{Z}^{\D0\tbar\D0}(\phi_{I},\phi_{J}).
\eea
The sign $+$ represents the sign of the reference vector $\eta_{0}$. Choosing the reference vector $\eta_{0}$ picks poles corresponding to plane partitions extending in the 123-plane whose origins are the $\fra_{\alpha}$. In other words, the pole structure is just the rank $n$ computation given in section~\ref{sec:rankNDTPT}, but the existence of the $m$ $\overline{\D6}_{\bar{4}}$-branes play the roles of extra matter contributions. We note that generally, this partition function depends on the Coulomb branch parameters $\fra_{\alpha},\frb_{\beta}$ when $n\neq 0, m\neq 0$, but we did not write the dependence on them explicitly.

Let us consider a different partition function which is defined as
\bea
\mathcal{Z}^{\D6,-}_{\bar{4},n|m}[\fq,q_{1,2,3,4}]&=\sum_{k=0}^{\infty}\fq^{k}\mathcal{Z}^{\D6,-}_{\bar{4},n|m}[k],
\eea
where the $k$-instanton sector is
\bea
\mathcal{Z}^{\D6,-}_{\bar{4},n|m}[k]&=\frac{1}{k!}\left(\frac{\sh(-\epsilon_{14,24,34})}{\sh(-\epsilon_{1,2,3,4})}\right)^{k}\oint_{\tilde{\eta}_{0}}\prod_{I=1}^{k}\frac{d\phi_{I}}{2\pi i }\prod_{I=1}^{k}\prod_{\alpha=1}^{n}\mathcal{Z}^{\D6_{\bar{4}}\tbar\D0}(\fra_{\alpha},\phi_I)\\
&\times \prod_{\beta=1}^{m}\mathcal{Z}^{\overline{\D6}_{\bar{4}}\tbar\D0}(\frb_{\beta}-\eps_4,\phi_I)\prod_{I<J}\mathcal{Z}^{\D0\tbar\D0}(\phi_{I},\phi_{J}).
\eea
This time the reference vector is $\eta=\tilde{\eta}_{0}$ and the sign $-$ denotes this dependence. Although we will not discuss in detail, the poles will be classified by $m$-tuples of plane partitions whose origin is $\frb_{\beta}$ but the coordinates of the boxes are $\frb_{\beta}-(i-1)\eps_1-(j-1)\eps_2-(k-1)\eps_3$ as in the case of the supergroup instanton counting~\cite{Kimura:2019msw,Noshita:2022dxv}. Namely, the plane partition grows in the negative direction of the three-axes. This is not difficult to observe because the framing node contribution is
\bea
\mathcal{Z}^{\overline{\D6}_{\bar{4}}\tbar\D0}(\frb-\eps_4,\phi_I)=\frac{\sh(\frb-\eps_4-\phi_I)}{\sh(\frb-\phi_I)},
\eea
which is the same with the D6-brane contribution after changing the fundamental chiral multiplets to anti-fundamental chiral multiplets and changing the sign of the flavor fugacities $\eps_a\rightarrow -\eps_a$. Choosing $\eta=\tilde{\eta}_0$ picks such poles and the hyperplanes
\bea
\phi_I-\phi_J=-\eps_{1,2,3,4}
\eea
and thus the crystal grows in the negative directions. For this case, the $n$ D6$_{\bar{4}}$-branes will play the roles of extra matter contributions.

Studying the partition functions, we generally have the following relation.
\begin{proposition}
    Let $v_{\alpha}=e^{\fra_{\alpha}}\,(\alpha=1,\ldots,n)$ and $w_{\beta}=e^{\frb_{\beta}}\,(\beta=1,\ldots,m)$ be the flavor fugacities of the theory. We then have 
    \bea
\mathcal{Z}^{\D6,+}_{\bar{4},n|m}[\fq,q_{1,2,3,4};\vec{v},\vec{w}]=\mathcal{Z}^{\D6,-}_{\bar{4},m|n}[\fq,q^{-1}_{1,2,3,4};\vec{w},\vec{v}],
    \eea
    where we explicitly denoted the flavor fugacities dependence in the partition functions.
\end{proposition}
Using this and the fact that the PE-formula \eqref{eq:tetrahedronPE} is invariant under $q_{a}\rightarrow q_{a}^{-1}$, one can also show that when $n=0$, the partition function has a PE-formula and it coincides with the rank $m$ D6$_{\bar{4}}$ partition function.
\begin{corollary}
 We have the following identity: 
 \bea
 \mathcal{Z}^{\D6,-}_{\bar{4},0|m}[\fq,q_{1,2,3,4}]=\mathcal{Z}^{\D6}_{\bar{4},m}[\fq,q_{1,2,3,4}]
 \eea
\end{corollary}

An interesting property is that the partition functions $\mathcal{Z}^{\D6,\pm}_{\bar{4},n|m}[\fq,q_{1,2,3,4}]$ are connected in a similar way as the DT/PT correspondence. 
\begin{conjecture}\label{thm:puresupergroupDTPT}
    The partition function of the pure supergroup gauge theory coming from $n$ D6$_{\bar{4}}$-branes and $m$ $\overline{\D6}_{\bar{4}}$-branes evaluated with the reference vector $\eta_{0},\tilde{\eta}_{0}$ obeys the wall-crossing formula
    \bea
     \mathcal{Z}^{\D6,+}_{\bar{4},n|m}[\fq,q_{1,2,3,4}]=\text{MF}[q_{4}^{n-m}]\mathcal{Z}^{\D6,-}_{\bar{4},n|m}[\fq,q_{1,2,3,4}].
    \eea
    
    In particular, when $n=m$, there is no wall-crossing
     \bea
     \mathcal{Z}^{\D6,+}_{\bar{4},n|n}[\fq,q_{1,2,3,4}]=\mathcal{Z}^{\D6,-}_{\bar{4},n|n}[\fq,q_{1,2,3,4}].
    \eea
    
\end{conjecture}
An interesting fact is that although the partition function generally depends on the Coulomb branch parameters, the difference between them factorizes into a Coulomb branch parameter independent function.

The fact that no wall-crossing occurs when $n=m$ can be also confirmed by studying the asymptotics of the contour integrand. For the limit $\phi_I\rightarrow \infty$, we have
\bea
\prod_{\alpha=1}^{n}\mathcal{Z}^{\D6_{\bar{4}}\tbar\D0}(\fra_{\alpha},\phi_I)\prod_{\beta=1}^{m}\mathcal{Z}^{\overline{\D6}_{\bar{4}}\tbar\D0}(\frb_{\beta}-\eps_4,\phi_I)\rightarrow q_{4}^{-\frac{1}{2}(n-m)}
\eea
and for the limit $\phi_I\rightarrow -\infty$, we have
\bea
\prod_{\alpha=1}^{n}\mathcal{Z}^{\D6_{\bar{4}}\tbar\D0}(\fra_{\alpha},\phi_I)\prod_{\beta=1}^{m}\mathcal{Z}^{\overline{\D6}_{\bar{4}}\tbar\D0}(\frb_{\beta}-\eps_4,\phi_I)\rightarrow q_{4}^{+\frac{1}{2}(n-m)}.
\eea
Therefore, when $n=m$, the difference between the asymptotics vanish and we have no wall-crossing.

\paragraph{Supergroup generalizations of DT and PT counting}
Let us now introduce the D2-branes giving nontrivial boundary conditions. Since we also have the $\overline{\D6}_{\bar{4}}$-branes, we also can introduce boundary conditions on them. We introduce $n$ $\D6_{\bar{4}}$-branes and $m$ $\overline{\D6}_{\bar{4}}$-branes to the system and denote the boundary conditions as $(\vec{\lambda}_{+},\vec{\mu}_{+},\vec{\nu}_{+})$ for the D6-branes and $(\vec{\lambda}_{-},\vec{\mu}_{-},\vec{\nu}_{-})$ for the $\overline{\D6}$-branes, where
\bea
\vec{\lambda}_{+}=(\lambda_{+}^{(\alpha)})_{\alpha=1}^{n},\quad \vec{\mu}_{+}=(\mu_{+}^{(\alpha)})_{\alpha=1}^{n},\quad \vec{\nu}_{+}=(\nu_{+}^{(\alpha)})_{\alpha=1}^{n},\\
\vec{\lambda}_{-}=(\lambda_{-}^{(\alpha)})_{\alpha=1}^{m},\quad \vec{\mu}_{-}=(\mu_{-}^{(\alpha)})_{\alpha=1}^{m},\quad \vec{\nu}_{-}=(\nu_{-}^{(\alpha)})_{\alpha=1}^{m},\\
\eea
The orientation of the Young diagrams for the $\overline{\D6}$-brane part is the one given in Fig.~\ref{fig:PTreverse}. 

The framing node contribution becomes 
\bea
\prod_{\alpha=1}^{n}\mathcal{Z}^{\D6_{\bar{4}}\tbar\D2\tbar\D0}_{\DT;\lambda_{+}^{(\alpha)}\mu_{+}^{(\alpha)}\nu_{+}^{(\alpha)}}(\fra_{\alpha},\phi_{I})\prod_{\beta=1}^{m}\mathcal{Z}^{\D6_{\bar{4}}\tbar\D2\tbar\D0}_{\PT;\lambda_{-}^{(\beta)}\mu_{-}^{(\beta)}\nu_{-}^{(\beta)}}(\frb_{\beta},\phi_{I}).
\eea
Note that the contribution $\mathcal{Z}^{\D6_{\bar{4}}\tbar\D2\tbar\D0}_{\PT;\lambda_{-}^{(\beta)}\mu_{-}^{(\beta)}\nu_{-}^{(\beta)}}(\frb_{\beta},\phi_{I})$ generalizes the $\overline{\D6}_{\bar{4}}$-brane contribution to include the leg boundary conditions. One will see that if the boundary conditions are trivial, it indeed reproduces the $\overline{\D6}_{\bar{4}}$ contribution.

We then define the generalized partition function as
\bea
\,&\mathcal{Z}^{+}_{\bar{4};\,\vec{\lambda}_{+}\vec{\mu}_{+}\vec{\nu}_{+}|\vec{\lambda}_{-}\vec{\mu}_{-}\vec{\nu}_{-}}[\mathfrak{q},q_{1,2,3,4}]=\sum_{k=0}^{\infty}\mathfrak{q}^{k}\mathcal{Z}^{+}_{\bar{4};\,\vec{\lambda}_{+}\vec{\mu}_{+}\vec{\nu}_{+}|\vec{\lambda}_{-}\vec{\mu}_{-}\vec{\nu}_{-}}[k],
\eea
where the $k$-instanton sector is
\bea\label{eq:DTPTvertex-positive}
\,\mathcal{Z}^{+}_{\bar{4};\,\vec{\lambda}_{+}\vec{\mu}_{+}\vec{\nu}_{+}|\vec{\lambda}_{-}\vec{\mu}_{-}\vec{\nu}_{-}}[k]&=\frac{1}{k!}\left(\frac{\sh(-\epsilon_{14,24,34})}{\sh(-\epsilon_{1,2,3,4})}\right)^{k}\oint_{\eta_{0}} \prod_{I=1}^{k}\frac{d\phi_{I}}{2\pi i}\prod_{I=1}^{k}\prod_{\alpha=1}^{n}\mathcal{Z}^{\D6_{\bar{4}}\tbar\D2\tbar\D0}_{\DT;\lambda_{+}^{(\alpha)}\mu_{+}^{(\alpha)}\nu_{+}^{(\alpha)}}(\fra_{\alpha},\phi_{I})\\
& \times \prod_{I=1}^{k}\prod_{\beta=1}^{m}\mathcal{Z}^{\D6_{\bar{4}}\tbar\D2\tbar\D0}_{\PT;\lambda_{-}^{(\beta)}\mu_{-}^{(\beta)}\nu_{-}^{(\beta)}}(\frb_{\beta},\phi_{I}) \prod_{I<J}^{k}\mathcal{Z}^{\D0\tbar\D0}(\phi_{I},\phi_{J})
\eea
where we used the \eqref{eq:DTflavornode-def} and \eqref{eq:PTflavornode-def} (see section~\ref{sec:PT3counting-conjugate}). Again, the sign $+$ denotes the fact that we are choosing the reference vector to be $\eta=\eta_0$.

The pole structure is much more complicated compared to the previous cases and we will get a mix of DT and PT counting. For the contributions from the D6$_{\bar{4}}$-branes, we obtain the usual DT3 counting. On the other hand, for the contributions from the $\overline{\D6}_{\bar{4}}$-branes, we have the PT3 counting.\footnote{Strictly speaking, the conjugate PT3 counting.} Since it is a mix of DT and PT counting, we will call this partition function the \textbf{(DT|PT) vertex} or the \textbf{(DT|PT) partition function}. 

Choosing the reference vector to be $\eta=\tilde{\eta}_0$, we can define a different partition function:
\bea
\,&\mathcal{Z}^{-}_{\bar{4};\,\vec{\lambda}_{+},\vec{\mu}_{+}\vec{\nu}_{+}|\vec{\lambda}_{-}\vec{\mu}_{-}\vec{\nu}_{-}}[\mathfrak{q},q_{1,2,3,4}]=\sum_{k=0}^{\infty}\mathfrak{q}^{k}\mathcal{Z}^{-}_{\bar{4};\,\vec{\lambda}_{+}\vec{\mu}_{+}\vec{\nu}_{+}|\vec{\lambda}_{-}\vec{\mu}_{-}\vec{\nu}_{-}}[k],
\eea
where the $k$-instanton sector is given as
\bea\label{eq:DTPTvertex-negative}
\,\mathcal{Z}^{-}_{\bar{4};\,\vec{\lambda}_{+}\vec{\mu}_{+}\vec{\nu}_{+}|\vec{\lambda}_{-}\vec{\mu}_{-}\vec{\nu}_{-}}[k]&=\frac{1}{k!}\left(\frac{\sh(-\epsilon_{14,24,34})}{\sh(-\epsilon_{1,2,3,4})}\right)^{k}\oint_{\tilde{\eta}_{0}} \prod_{I=1}^{k}\frac{d\phi_{I}}{2\pi i}\prod_{I=1}^{k}\prod_{\alpha=1}^{n}\mathcal{Z}^{\D6_{\bar{4}}\tbar\D2\tbar\D0}_{\DT;\lambda_{+}^{(\alpha)}\mu_{+}^{(\alpha)}\nu_{+}^{(\alpha)}}(\fra_{\alpha},\phi_{I})\\
& \times \prod_{I=1}^{k}\prod_{\beta=1}^{m}\mathcal{Z}^{\D6_{\bar{4}}\tbar\D2\tbar\D0}_{\PT;\lambda_{-}^{(\beta)}\mu_{-}^{(\beta)}\nu_{-}^{(\beta)}}(\frb_{\beta},\phi_{I}) \prod_{I<J}^{k}\mathcal{Z}^{\D0\tbar\D0}(\phi_{I},\phi_{J}).
\eea
This time, for the contributions from the $\D6_{\bar{4}}$-branes, we have the PT3 counting, while for the contributions from the $\overline{\D6}_{\bar{4}}$-branes, we have the DT3 counting.\footnote{Strictly speaking, this is the conjugate version of the DT3 counting. The conjugate DT3 counting is the DT3 counting with all of the $\eps_a$ parameters flipped to $-\eps_a$.} Since the DT, PT properties are flipped, we call it the \textbf{(PT|DT) vertex } or the \textbf{(PT|DT) partition function}.

The (DT|PT) vertex and the (PT|DT) vertex are connected with each other similar to other DT/PT correspondence.
\begin{conjecture}
    The DT/PT correspondence of the (DT|PT) vertex and the (PT|DT) vertex is 
    \bea
\mathcal{Z}^{+}_{\bar{4};\,\vec{\lambda}_{+}\vec{\mu}_{+}\vec{\nu}_{+}|\vec{\lambda}_{-}\vec{\mu}_{-}\vec{\nu}_{-}}[\mathfrak{q},q_{1,2,3,4}]=\text{MF}[q_{4}^{n-m}]\,\mathcal{Z}^{-}_{\bar{4};\,\vec{\lambda}_{+}\vec{\mu}_{+}\vec{\nu}_{+}|\vec{\lambda}_{-}\vec{\mu}_{-}\vec{\nu}_{-}}[\mathfrak{q},q_{1,2,3,4}]
    \eea
    where $n$ is the number of the D6$_{\bar{4}}$-branes and $m$ is the number of the $\overline{\D6}_{\bar{4}}$-branes.

    In particular, if $n=m$, the (DT|PT) vertex and the (PT|DT) vertex are the same.
\end{conjecture}
Note that Conj.~\ref{thm:puresupergroupDTPT} is just a special case by setting all of the boundary conditions to be trivial. The higher rank DT/PT correspondence in Conj.~\ref{thm:rankn-DTPT} comes by setting $m=0$.

\paragraph{Supergroup tetrahedron instantons }
Let us generalize the situation when we also have different orientation of D6-branes and $\overline{\D6}$-branes. Let us first start when we do not have any D2-branes and boundary conditions.

We consider the setup with $n_{\bar{a}}$ D6$_{\bar{a}}$-branes and $m_{\bar{a}}$ $\overline{\D6}_{\bar{a}}$-branes. The framing node contribution is modified to
\bea
\prod_{a\in\four}\prod_{\alpha=1}^{n_{\bar{a}}}\mathcal{Z}^{\D6_{\bar{a}}\tbar\D0}(\fra_{\bar{a},\alpha},\phi_I)
\prod_{\beta=1}^{m_{\bar{a}}}\mathcal{Z}^{\overline{\D6}_{\bar{a}}\tbar\D0}(\frb_{\bar{a},\beta}-\eps_a,\phi_I).
\eea
The partition function is
\bea
\mathcal{Z}^{\D6,+}_{\substack{(n_{\bar{1}},n_{\bar{2}},n_{\bar{3}},n_{\bar{4}})\\(m_{\bar{1}},m_{\bar{2}},m_{\bar{3}},m_{\bar{4}})}}[\fq,q_{1,2,3,4}]&=\sum_{k=0}^{\infty}\fq^{k}\mathcal{Z}^{\D6,+}_{\substack{(n_{\bar{1}},n_{\bar{2}},n_{\bar{3}},n_{\bar{4}})\\(m_{\bar{1}},m_{\bar{2}},m_{\bar{3}},m_{\bar{4}})}}[k],
\eea
where
\bea
\mathcal{Z}^{\D6,+}_{\substack{(n_{\bar{1}},n_{\bar{2}},n_{\bar{3}},n_{\bar{4}})\\(m_{\bar{1}},m_{\bar{2}},m_{\bar{3}},m_{\bar{4}})}}[k]&=\frac{1}{k!}\left(\frac{\sh(-\epsilon_{14,24,34})}{\sh(-\epsilon_{1,2,3,4})}\right)^{k}\oint_{\eta_{0}}\prod_{I=1}^{k}\frac{d\phi_{I}}{2\pi i }\prod_{I=1}^{k}\prod_{a\in\four}\prod_{\alpha=1}^{n_{\bar{a}}}\mathcal{Z}^{\D6_{\bar{a}}\tbar\D0}(\fra_{\bar{a},\alpha},\phi_I)\\
&\times\prod_{a\in\four}\prod_{\beta=1}^{m_{\bar{a}}}\mathcal{Z}^{\overline{\D6}_{\bar{a}}\tbar\D0}(\frb_{\bar{a},\beta},\phi_I) \prod_{I<J}\mathcal{Z}^{\D0\tbar\D0}(\phi_{I},\phi_{J}).
\eea
Note again that the sign $+$ denotes the choice of the reference vector $\eta=\eta_0$. The pole structure is similar to the tetehedron instanton case and we have multiple plane partitions in different orientations. The anti D6-brane contributions just play the roles of extra matter contributions.

Again, we can consider a different partition function by choosing the reference vector $\eta=\tilde{\eta}_0$:
\bea
\mathcal{Z}^{\D6,-}_{\substack{(n_{\bar{1}},n_{\bar{2}},n_{\bar{3}},n_{\bar{4}})\\(m_{\bar{1}},m_{\bar{2}},m_{\bar{3}},m_{\bar{4}})}}[\fq,q_{1,2,3,4}]&=\sum_{k=0}^{\infty}\fq^{k}\mathcal{Z}^{\D6,-}_{\substack{(n_{\bar{1}},n_{\bar{2}},n_{\bar{3}},n_{\bar{4}})\\(m_{\bar{1}},m_{\bar{2}},m_{\bar{3}},m_{\bar{4}})}}[k],
\eea
where
\bea
\mathcal{Z}^{\D6,-}_{\substack{(n_{\bar{1}},n_{\bar{2}},n_{\bar{3}},n_{\bar{4}})\\(m_{\bar{1}},m_{\bar{2}},m_{\bar{3}},m_{\bar{4}})}}[k]&=\frac{1}{k!}\left(\frac{\sh(-\epsilon_{14,24,34})}{\sh(-\epsilon_{1,2,3,4})}\right)^{k}\oint_{\tilde{\eta}_{0}}\prod_{I=1}^{k}\frac{d\phi_{I}}{2\pi i }\prod_{I=1}^{k}\prod_{a\in\four}\prod_{\alpha=1}^{n_{\bar{a}}}\mathcal{Z}^{\D6_{\bar{a}}\tbar\D0}(\fra_{\bar{a},\alpha},\phi_I)\\
&\times\prod_{a\in\four}\prod_{\beta=1}^{m_{\bar{a}}}\mathcal{Z}^{\overline{\D6}_{\bar{a}}\tbar\D0}(\frb_{\bar{a},\beta},\phi_I) \prod_{I<J}\mathcal{Z}^{\D0\tbar\D0}(\phi_{I},\phi_{J}).
\eea
In this case, the poles picked up correspond to plane partitions in different orientations and extending in the negative axes.

The ``DT/PT correspondence" of this theory is given as follows.
\begin{conjecture}\label{thm:tetrahedronsupergroupDTPT}
 The partition functions of the supergroup tetrahedron instantons obey the identity
    \bea
\mathcal{Z}^{\D6,+}_{\substack{(n_{\bar{1}},n_{\bar{2}},n_{\bar{3}},n_{\bar{4}})\\(m_{\bar{1}},m_{\bar{2}},m_{\bar{3}},m_{\bar{4}})}}[\fq,q_{1,2,3,4}]=\text{MF}\left[\prod_{a\in\four}q_{a}^{n_{\bar{a}}-m_{\bar{a}}}\right]\mathcal{Z}^{\D6,-}_{\substack{(n_{\bar{1}},n_{\bar{2}},n_{\bar{3}},n_{\bar{4}})\\(m_{\bar{1}},m_{\bar{2}},m_{\bar{3}},m_{\bar{4}})}}[\fq,q_{1,2,3,4}].
    \eea
    If $\prod_{a\in\four}q_{a}^{n_{\bar{a}}-m_{\bar{a}}}=1$, which is $n_{\bar{a}}-m_{\bar{a}}=n_{\bar{b}}-m_{\bar{b}}$ for $a\neq b$, then no wall crossing occurs.
\end{conjecture}

\paragraph{Tetrahedron (DT|PT) vertex}
Generalizations to include the D2-branes as boundary conditions are straightforward. We denote the boundary conditions for the three-legs of the $n_{\bar{a}}$ D6$_{\bar{a}}$-branes as $(\vec{\lambda}_{+,\bar{a}},\vec{\mu}_{+,\bar{a}},\vec{\nu}_{+,\bar{a}})$ and for the $m_{\bar{a}}$ $\overline{\D6}_{\bar{a}}$-branes as $(\vec{\lambda}_{-,\bar{a}},\vec{\mu}_{-,\bar{a}},\vec{\nu}_{-,\bar{a}})$. We collectively denote them as $(\underline{\vec{\lambda}}_{+},\underline{\vec{\mu}}_{+},\underline{\vec{\nu}}_{+})$ and $(\underline{\vec{\lambda}}_{-},\underline{\vec{\mu}}_{-},\underline{\vec{\nu}}_-)$.

The framing node contribution is
\bea
\prod_{a\in\four}\prod_{\alpha=1}^{n_{\bar{a}}}\mathcal{Z}^{\D6_{\bar{a}}\tbar\D2\tbar\D0}_{\DT;\lambda_{+,\bar{a}}^{(\alpha)}\mu_{+,\bar{a}}^{(\alpha)}\nu_{+,\bar{a}}^{(\alpha)}}(\fra_{\bar{a},\alpha},\phi_{I})\prod_{\beta=1}^{m_{\bar{a}}}\mathcal{Z}^{\D6_{\bar{a}}\tbar\D2\tbar\D0}_{\PT;\lambda_{-,\bar{a}}^{(\beta)}\mu_{-,\bar{a}}^{(\beta)}\nu_{-,\bar{a}}^{(\beta)}}(\frb_{\bar{a},\beta},\phi_{I}).
\eea

The tetrahedron (DT|PT) vertex is defined as
\bea
\,&\mathcal{Z}^{+}_{\,\underline{\vec{\lambda}}_{+}\underline{\vec{\mu}}_{+}\underline{\vec{\nu}}_{+}|\underline{\vec{\lambda}}_{-}\underline{\vec{\mu}}_{-}\underline{\vec{\nu}}_{-}}[\mathfrak{q},q_{1,2,3,4}]=\sum_{k=0}^{\infty}\mathfrak{q}^{k}\mathcal{Z}^{+}_{\,\underline{\vec{\lambda}}_{+}\underline{\vec{\mu}}_{+}\underline{\vec{\nu}}_{+}|\underline{\vec{\lambda}}_{-}\underline{\vec{\mu}}_{-}\underline{\vec{\nu}}_{-}}[k],
\eea
where the $k$-instanton sector is
\bea\label{eq:tetrehedronDTPTvertex-positive}
\,\mathcal{Z}^{+}_{\,\underline{\vec{\lambda}}_{+}\underline{\vec{\mu}}_{+}\underline{\vec{\nu}}_{+}|\underline{\vec{\lambda}}_{-}\underline{\vec{\mu}}_{-}\underline{\vec{\nu}}_{-}}[k]&=\frac{1}{k!}\left(\frac{\sh(-\epsilon_{14,24,34})}{\sh(-\epsilon_{1,2,3,4})}\right)^{k}\oint_{\eta_{0}} \prod_{I=1}^{k}\frac{d\phi_{I}}{2\pi i}\prod_{I=1}^{k}\prod_{a\in\four}\prod_{\alpha=1}^{n_{\bar{a}}}\mathcal{Z}^{\D6_{\bar{a}}\tbar\D2\tbar\D0}_{\DT;\lambda_{+,\bar{a}}^{(\alpha)}\mu_{+,\bar{a}}^{(\alpha)}\nu_{+,\bar{a}}^{(\alpha)}}(\fra_{\bar{a},\alpha},\phi_{I})\\
& \times \prod_{I=1}^{k}\prod_{a\in\four}\prod_{\beta=1}^{m_{\bar{a}}}\mathcal{Z}^{\D6_{\bar{a}}\tbar\D2\tbar\D0}_{\PT;\lambda_{-,\bar{a}}^{(\beta)}\mu_{-,\bar{a}}^{(\beta)}\nu_{-,\bar{a}}^{(\beta)}}(\frb_{\bar{a},\beta},\phi_{I}) \prod_{I<J}^{k}\mathcal{Z}^{\D0\tbar\D0}(\phi_{I},\phi_{J}).
\eea
Choosing the reference vector $\eta=\eta_{0}$ gives the tetrahedron DT counting of the D6-branes and the tetrahedron PT counting of the $\overline{\D6}$-branes.

The tetrahedron (PT|DT) vertex is defined as
\bea
\,&\mathcal{Z}^{-}_{\,\underline{\vec{\lambda}}_{+}\underline{\vec{\mu}}_{+}\underline{\vec{\nu}}_{+}|\underline{\vec{\lambda}}_{-}\underline{\vec{\mu}}_{-}\underline{\vec{\nu}}_{-}}[\mathfrak{q},q_{1,2,3,4}]=\sum_{k=0}^{\infty}\mathfrak{q}^{k}\mathcal{Z}^{-}_{\,\underline{\vec{\lambda}}_{+}\underline{\vec{\mu}}_{+}\underline{\vec{\nu}}_{+}|\underline{\vec{\lambda}}_{-}\underline{\vec{\mu}}_{-}\underline{\vec{\nu}}_{-}}[k],
\eea
where the $k$-instanton sector is
\bea\label{eq:tetrehedronDTPTvertex-negative}
\,\mathcal{Z}^{-}_{\,\underline{\vec{\lambda}}_{+}\underline{\vec{\mu}}_{+}\underline{\vec{\nu}}_{+}|\underline{\vec{\lambda}}_{-}\underline{\vec{\mu}}_{-}\underline{\vec{\nu}}_{-}}[k]&=\frac{1}{k!}\left(\frac{\sh(-\epsilon_{14,24,34})}{\sh(-\epsilon_{1,2,3,4})}\right)^{k}\oint_{\tilde{\eta}_{0}} \prod_{I=1}^{k}\frac{d\phi_{I}}{2\pi i}\prod_{I=1}^{k}\prod_{a\in\four}\prod_{\alpha=1}^{n_{\bar{a}}}\mathcal{Z}^{\D6_{\bar{a}}\tbar\D2\tbar\D0}_{\DT;\lambda_{+,\bar{a}}^{(\alpha)}\mu_{+,\bar{a}}^{(\alpha)}\nu_{+,\bar{a}}^{(\alpha)}}(\fra_{\bar{a},\alpha},\phi_{I})\\
&\times \prod_{I=1}^{k}\prod_{a\in\four}\prod_{\beta=1}^{m_{\bar{a}}}\mathcal{Z}^{\D6_{\bar{a}}\tbar\D2\tbar\D0}_{\PT;\lambda_{-,\bar{a}}^{(\beta)}\mu_{-,\bar{a}}^{(\beta)}\nu_{-,\bar{a}}^{(\beta)}}(\frb_{\bar{a},\beta},\phi_{I}) \prod_{I<J}^{k}\mathcal{Z}^{\D0\tbar\D0}(\phi_{I},\phi_{J}).
\eea
Choosing the reference vector $\eta=\tilde{\eta}_{0}$ gives the tetrahedron PT counting of the D6-branes and the tetrahedron DT counting of the $\overline{\D6}$-branes.

The DT/PT correspondence is given as follows.
\begin{conjecture}
    The DT/PT correspondence of the tetrahedron (DT|PT) vertex and the (PT|DT) vertex is 
    \bea
\mathcal{Z}^{+}_{\,\underline{\vec{\lambda}}_{+}\underline{\vec{\mu}}_{+}\underline{\vec{\nu}}_{+}|\underline{\vec{\lambda}}_{-}\underline{\vec{\mu}}_{-}\underline{\vec{\nu}}_{-}}[\mathfrak{q},q_{1,2,3,4}]=\text{MF}\left[\prod_{a\in\four}q_{a}^{n_{\bar{a}}-m_{\bar{a}}}\right]\,\mathcal{Z}^{-}_{\,\underline{\vec{\lambda}}_{+}\underline{\vec{\mu}}_{+}\underline{\vec{\nu}}_{+}|\underline{\vec{\lambda}}_{-}\underline{\vec{\mu}}_{-}\underline{\vec{\nu}}_{-}}[\mathfrak{q},q_{1,2,3,4}]
    \eea
If $\prod_{a\in\four}q_{a}^{n_{\bar{a}}-m_{\bar{a}}}=1$, which is $n_{\bar{a}}-m_{\bar{a}}=n_{\bar{b}}-m_{\bar{b}}$ for $a\neq b$, then both sides become the same. 
\end{conjecture}

Note that all of the DT/PT correspondence we have discussed is a reduction of this theorem. For any cases, the two partition functions evaluated using different reference vectors are equal to each other up to factors depending on $\text{MF}[\mu]$. In other words, this factor measures the difference of the two partition functions. Moreover, the parameter $\mu$ is equal to the central charges of the total system of D6, $\overline{\D6}$-branes, where we associated the central charge of the D6$_{\bar{a}}$-branes and $\overline{\D6}_{\bar{a}}$-branes as $q_{a}$ and $q_{a}^{-1}$, respectively.

\section{BPS/CFT correspondence and PT3 \texorpdfstring{$qq$}{qq}-characters}\label{sec:BPS/CFT-PT3qq}
In this section, we will introduce the $qq$-characters\footnote{We do not distinguish the $qq$-characters and quiver W-algebra \cite{Kimura:2015rgi,Kimura:2016dys,Kimura:2017hez,Shiraishi:1995rp,Awata:1996dx,Awata:1995zk,Frenkel:1998ojj, Frenkel:1997CMP} in this paper. Strictly speaking, the quiver W-algebra is the operator formalism of the $qq$-characters.} corresponding to the PT3 counting, which we call the PT3 $qq$-characters. The PT3 $qq$-characters contain the information of the PT3 partition functions and thus it establishes the BPS/CFT correspondence. Basically, we follow the discussion in \cite{Kimura:2023bxy, Kimura:2024xpr,Kimura:2024osv}. See \cite{Nekrasov:2015wsu,Nekrasov:2016ydq,Kimura:2015rgi} for original references (see also \cite{Kimura:2019hnw,Kimura:2022zsm}).

In section~\ref{sec:vertexop-def}, we introduce the vertex operators associated with the D-branes. The free field realizations of the DT and PT contour integral formulas are given in section~\ref{sec:freefield-contourintegral}. The existence of the free field realizations lead us to the concept of the PT3 $qq$-characters. We then give an algebraic derivation of them using the screening charge and the collision limit in section~\ref{sec:PT3qq-screeningcharge}. Finally, the relation with shifted quantum toroidal $\mathfrak{gl}_{1}$ is discussed in section~\ref{sec:QTgl1}.

\subsection{Vertex operators}\label{sec:vertexop-def}
Let us review the vertex operators introduced in \cite{Kimura:2023bxy}. We define
\bea
\bfP_{a}=1-q_{a},\quad \bfP_{abc}=\bfP_{a}\bfP_{b}\bfP_{c},\quad  \bfP_{\four}=\prod_{a\in\four}\bfP_{a},\quad \bfP_{a}^{[n]}=\bfP_{a}|_{q_{a}\rightarrow q_{a}^{n}}
\eea
and the vertex operators are given as follows.
\begin{definition}[\cite{Kimura:2023bxy}]\label{def:vertex-op}
    Let $\{\mathsf{a}_{n}\}_{n\in\mathbb{Z}_{\geq 0}}$ be free bosons obeying the commutation relation
    \bea
   \relax [\mathsf{a}_{n},\mathsf{a}_{m}]=-\frac{1}{n}\mathbf{P}_{\four}^{[n]}\delta_{n+m,0}.
    \eea
    The D0, D2, D6 vertex operators are $\mathsf{A}(x),\mathsf{S}_{a}(x),\mathsf{W}_{\bar{a}}(x)$: 
    \bea
    \mathsf{A}(x)&=\mathsf{a}_{0}(x):\exp\left(\sum_{n\neq 0}\mathsf{a}_{n}x^{-n}\right):,\quad  \mathsf{S}_{a}(x)=\mathsf{s}_{a,0}(x):\exp\left(\sum_{n\neq 0}\mathsf{s}_{a,n}x^{-n}\right):,\\
    \mathsf{W}_{\bar{a}}(x)&=\mathsf{w}_{\bar{a},0}(x):\exp\left(\sum_{n\neq 0}\mathsf{w}_{\bar{a},n}x^{-n}\right):, \quad \mathsf{a}_{n}=\begin{dcases}
       \bfP^{[-n]}_{a}\mathsf{s}_{a,n},\\
       \bfP^{[-n]}_{\bar{a}}\mathsf{w}_{\bar{a},n}.
   \end{dcases}
    \eea
for $a\in\four$, where $\mathsf{a}_{0}(x),\mathsf{s}_{a,0}(x), \mathsf{w}_{\bar{a},0}(x)$ are some zero-modes. We denote the normal ordering symbol by $:\cdot:$.
\end{definition}

We will use the following zero-modes for explicit computations:
\bea\label{eq:zero-modes}
\mathsf{a}_{0}(x)=e^{\mathsf{t}_{0}},\quad \mathsf{s}_{a,0}(x)=x^{-(\log q_{a})^{-1}\mathsf{t}_{0}}e^{-(\log q_{a})^{-1}\tilde{\partial}_{\mathsf{t}}},\quad \mathsf{w}_{\bar{a},0}(x)=x^{-\log q_{a}\tilde{\mathsf{t}}_{0}}e^{-\log q_{a}\log (-q_{a})\tilde{\mathsf{t}}_{0}}e^{-\log q_{a}\partial_{\mathsf{t}}}
\eea
where we introduced two independent sets of zero modes
\bea
\relax [\partial_{\mathsf{t}}, \mathsf{t}_{0}]=[\widetilde{\partial}_{\mathsf{t}}, \tilde{\mathsf{t}}_{0}]=1
\eea
with other commutation relations to be trivial.

\begin{definition}[\cite{Kimura:2023bxy}]\label{def:structurefunct-QA}
  The structure functions are 
    \bea\label{app-eq:struct_funct}
\mathscr{V}_{a}(x)&=\frac{1-q_{a}x}{1-x},\quad 
g_{\bar{a}}(x)=\frac{\prod_{i\neq a}(1-q_{i}x)(1-q_{\bar{a}}x)}{(1-x)\prod_{i\neq a}(1-q_{a}^{-1}q_{i}^{-1}x)},\\
\mathscr{S}_{ab}(x)&=\frac{(1-q_{a}x)(1-q_{b})}{(1-x)(1-q_{a}q_{b}x)},\quad \mathcal{A}_{\mathbb{C}^{4}}(x)=\frac{\prod_{i=1}^{4}(1-q_{i}x)\prod^{4}_{i=1}(1-q_{i}^{-1}x)}{(1-x)^{2}\prod_{i\neq j}(1-q_{i}q_{j}x)}.
  \eea
\end{definition}

\begin{proposition}[\cite{Kimura:2023bxy}]\label{prop:OPE-formula}
Under the explicit zero-modes \eqref{eq:zero-modes}, some of the operator product formulas are
\bea\label{eq:contractions}
\mathsf{A}(x)\mathsf{S}_{a}(x')=g_{\bar{a}}\left(x'/x\right)^{-1}: \mathsf{A}(x)\mathsf{S}_{a}(x'):,&\quad \mathsf{S}_{a}(x')\mathsf{A}(x)=g_{\bar{a}}(q_{a}x/x'):\mathsf{A}(x)\mathsf{S}_{a}(x'):,\\
    \mathsf{A}(x)\mathsf{W}_{\bar{a}}(x')=\mathscr{V}_{a}\left(x'/x\right)^{-1}:\mathsf{A}(x)\mathsf{W}_{\bar{a}}(x'):,&\quad 
    \mathsf{W}_{\bar{a}}(x')\mathsf{A}(x)=q_{a}^{-1}\mathscr{V}_{a}(q_{a}^{-1}x/x'):\mathsf{W}_{\bar{a}}(x')\mathsf{A}(x):,\\
    \mathsf{S}_{a}(x)\mathsf{S}_{b}(x')=\mathscr{S}_{\overline{ab}}(q_{a}x'/x):\mathsf{S}_{a}(x)\mathsf{S}_{b}(x'):,&\quad 
        \mathsf{S}_{b}(x')\mathsf{S}_{a}(x)=\mathscr{S}_{\overline{ab}}(q_{b}x/x'):\mathsf{S}_{a}(x)\mathsf{S}_{b}(x'):,\\
\mathsf{W}_{\bar{a}}(x)\mathsf{S}_{a}(x')=\frac{x'}{1-q_{a}^{-1}x'/x}:\mathsf{W}_{\bar{a}}(x)\mathsf{S}_{a}(x'):,&\quad 
\mathsf{S}_{a}(x')\mathsf{W}_{\bar{a}}(x)=\frac{-q_{a}x}{1-q_{a}x/x'}:\mathsf{W}_{\bar{a}}(x)\mathsf{S}_{a}(x'):.
\eea
\end{proposition}
After using the definitions of the vertex operators in Def.~\ref{def:vertex-op} and the zero-modes in \eqref{eq:zero-modes}, we also have the relations
\bea\label{eq:vertexoprelation}
\mathsf{A}(x)&={:\frac{\mathsf{S}_{a}(x)}{\mathsf{S}_{a}(q_{a}x)}:}=\mathsf{a}_{0}(x){:\frac{\mathsf{W}_{abc}(x)\mathsf{W}_{abc}(q_{ab}x)\mathsf{W}_{abc}(q_{ac}x)\mathsf{W}_{abc}(q_{bc}x)}{\mathsf{W}_{abc}(q_{a}x)\mathsf{W}_{abc}(q_{b}x)\mathsf{W}_{abc}(q_{c}x)\mathsf{W}_{abc}(q_{abc}x)}:},\\
\mathsf{S}_{a}(x)&=\mathsf{s}_{a,0}(x){:\frac{\mathsf{W}_{abc}(x)\mathsf{W}_{abc}(q_{bc}x)}{\mathsf{W}_{abc}(q_{b}x)\mathsf{W}_{abc}(q_{c}x)}:}.
\eea

\subsection{Free field realizations and PT3 \texorpdfstring{$qq$}{qq}-characters}\label{sec:freefield-contourintegral}
The contour integral formula associated with the DT and PT counting in Thm.~\ref{thm:DTvertex-JKresidue} and Def.~\ref{def:PTvertex-JKresidue} have free field realizations. To relate the integrand with the vertex operators, let us rewrite it using the structure functions in Def.~\ref{def:structurefunct-QA}.
\paragraph{Structure functions and contour integral formulas}
Let us first rewrite the contour integral formulas in a more convenient way using the structure functions:
\bea
&\mathcal{Z}^{\D6_{a}\tbar\D0}(\fra,\phi_{I})=q_{4}^{-1/2}\mathscr{V}_{4}\left(\frac{v}{x_{I}}\right),\quad\mathcal{Z}^{\overline{\D6}_{a}\tbar\D0}(\fra,\phi_{I})=q_{4}^{1/2}\mathscr{V}_{4}\left(\frac{v}{x_{I}}\right)^{-1},\quad \mathcal{Z}^{\overline{\D2}_{a}\tbar\D0}(\fra,\phi_{I})=g_{\bar{a}}\left(\frac{v}{x_{I}}\right)^{-1},\\
&\mathcal{Z}^{\D0\tbar\D0}(\phi_{I},\phi_{J})=\mathcal{A}_{\mathbb{C}^{4}}\left(\frac{x_{I}}{x_{J}}\right)^{-1},\quad  \mathcal{G}\coloneqq\frac{\sh(-\epsilon_{14,24,34})}{\sh(-\epsilon_{1,2,3,4})}=\frac{(1-q_{12})(1-q_{13})(1-q_{23})}{(1-q_{1})(1-q_{2})(1-q_{3})(1-q_{123})},
\eea
where $x_{I}=e^{\phi_{I}},\,\,v=e^{\fra}$. Redefining the topological terms as $\fq\rightarrow q_{4}^{1/2}\fq$ and including the $q_{4}^{1/2}$ factor to the contour integrand, the $U(1)$ D6 partition function and the DT partition functions becomes
\bea
\mathcal{Z}_{\bar{4}}^{\D6\tbar\QA}[k]&=\frac{\mathcal{G}^{k}}{k!}\oint_{\eta_{0}}\prod_{I=1}^{k}\frac{dx_{I}}{2\pi i x_{I}} \prod_{I=1}^{k}\mathscr{V}_{4}\left(\frac{v}{x_{I}}\right)\prod_{I<J}\mathcal{A}_{\mathbb{C}^{4}}\left(\frac{x_{I}}{x_{J}}\right)^{-1},\\
\mathcal{Z}^{\DT\tbar\QA}_{\bar{4};\lambda\mu\nu}[k]&=\frac{\mathcal{G}^{k}}{k!}\oint _{\eta_{0}}\prod_{I=1}^{k}\frac{dx_{I}}{2\pi i x_{I}}\prod_{I=1}^{k}\mathcal{Z}^{\D6_{\bar{4}}\tbar\D2\tbar\D0}_{\DT\tbar\QA;\lambda\mu\nu}(v,x_{I})\prod_{I<J}\mathcal{A}_{\mathbb{C}^{4}}\left(\frac{x_{I}}{x_{J}}\right)^{-1},
\eea
where
\bea
\mathcal{Z}^{\D6_{\bar{4}}\tbar\D2\tbar\D0}_{\DT\tbar\QA;\lambda\mu\nu}(v,x_{I})=\prod_{\scube \in s(\vec{Y})}\mathscr{V}_{4}\left(\frac{\chi_{\bar{4},v}(\cube)}{x_{I}}\right)\prod_{\scube \in p_{1}(\vec{Y})}\mathscr{V}_{4}\left(\frac{\chi_{\bar{4},v}(\cube)}{x_{I}}\right)^{-1}\prod_{\scube \in p_{2}(\vec{Y})}\mathscr{V}_{4}\left(\frac{\chi_{\bar{4},v}(\cube)}{x_{I}}\right)^{-2}.
\eea

Similarly, after redefining of the topological term $\fq\rightarrow q_{4}^{1/2} \fq $, the PT partition function becomes 
\bea
\mathcal{Z}^{\PT\tbar\QA}_{\bar{4};\lambda\mu\nu}[k]&=\frac{\mathcal{G}^{k}}{k!}\oint _{\tilde{\eta}_{0}}\prod_{I=1}^{k}\frac{dx_{I}}{2\pi i x_{I}}\prod_{I=1}^{k}\mathcal{Z}^{\D6_{\bar{4}}\tbar\D2\tbar\D0}_{\DT\tbar\QA;\lambda\mu\nu}(v,x_{I})\prod_{I<J}\mathcal{A}_{\mathbb{C}^{4}}\left(\frac{x_{I}}{x_{J}}\right)^{-1}.
\eea
For the PT partition function defined by the conjugate map in section~\ref{sec:PT3counting-conjugate}, we redefine the topological term as $\fq\rightarrow q_{4}^{-1/2}\fq$ and obtain
\bea
\widetilde{\mathcal{Z}}^{\PT\tbar\QA}_{\bar{4};\lambda\mu\nu}[k]&=\frac{\mathcal{G}^{k}}{k!}\oint _{\eta_{0}}\prod_{I=1}^{k}\frac{dx_{I}}{2\pi i x_{I}}\prod_{I=1}^{k}\mathcal{Z}^{\D6_{\bar{4}}\tbar\D2\tbar\D0}_{\PT\tbar\QA;\lambda\mu\nu}(v,x_{I})\prod_{I<J}\mathcal{A}_{\mathbb{C}^{4}}\left(\frac{x_{I}}{x_{J}}\right)^{-1}
\eea
where
\bea
\mathcal{Z}^{\D6_{\bar{4}}\tbar\D2\tbar\D0}_{\PT\tbar\QA;\lambda\mu\nu}(v,x_{I})=\prod_{\scube \in s(\vec{Y})}\mathscr{V}_{4}\left(\frac{q_{4}^{-1}\chi_{\bar{4},v^{-1}}(\cube)^{-1}}{x_{I}}\right)^{-1}\prod_{\scube \in p_{1}(\vec{Y})}\mathscr{V}_{4}\left(\frac{q_{4}^{-1}\chi_{\bar{4},v^{-1}}(\cube)^{-1}}{x_{I}}\right)\prod_{\scube \in p_{2}(\vec{Y})}\mathscr{V}_{4}\left(\frac{q_{4}^{-1}\chi_{\bar{4},v^{-1}}(\cube)^{-1}}{x_{I}}\right)^{2}.
\eea
Due to the modification of the topological terms, after evaluation of the contour integral, we have the relation\footnote{In the JK-residue computation, we used the symmetric notation $\sh(\phi)=e^{\phi/2}-e^{-\phi/2}$, while for the quantum algebraic computation, we are using the index $\mathbb{I}[x]=1-x^{-1}$ to compute. The extra zero-modes come from $\mathbb{I}[x^{-1}]=-x\mathbb{I}[x]$. This difference is known as the Nekrasov--Okounkov twist by the virtual canonical bundle~\cite{Nekrasov:2014nea}.}
\bea
\mathcal{Z}^{\PT\tbar\QA}_{\bar{4};\lambda\mu\nu}[k]=q_{4}^{k}\,\widetilde{\mathcal{Z}}^{\PT\tbar\QA}_{\bar{4};\lambda\mu\nu}[k].
\eea

After evaluating the poles (see Thm.~\ref{thm:tetraJKpoles}, \ref{thm:DTvertex-JKresidue}, \ref{thm:PTvertex-expansion}), we denote the partition functions in terms of plane partitions, DT configurations, PT configurations, GR configurations as
\bea
\sum_{\pi}\fq^{|\pi|}
\begin{dcases}
\mathcal{Z}^{\D6\tbar\QA}_{\bar{4}}[\pi],\quad \pi\in\mathcal{PP},\\
\mathcal{Z}^{\DT\tbar\QA}_{\bar{4};\lambda\mu\nu}[\pi],\quad \pi\in \mathcal{DT}_{\lambda\mu\nu},\\
\mathcal{Z}^{\PT\tbar\QA}_{\bar{4};\lambda\mu\nu}[\pi],\quad \pi\in \mathcal{PT}_{\lambda\mu\nu},\,\mathcal{GR}_{\lambda\mu\nu},\\
\widetilde{\mathcal{Z}}^{\PT\tbar\QA}_{\bar{4};\lambda\mu\nu}[\pi],\quad \pi \in \widetilde{\mathcal{PT}}_{\lambda\mu\nu},\,\widetilde{\mathcal{GR}}_{\lambda\mu\nu}.
\end{dcases}
\eea
Note that when expanding in the GR configurations, we are pairing the light and heavy boxes as discussed in Thm.~\ref{thm:GRvertex-expansion}.

\paragraph{Free field realizations}
Using the vertex operators defined in the previous section, the contour integral formulas can be written in a concise way. The partition functions become
\bea\label{eq:freefield-D6DTPT}
\mathcal{Z}^{\D6\tbar\QA}_{\bar{4}}[k]&=\frac{\mathcal{G}^{k}}{k!}\oint_{\eta_{0}}\prod_{I=1}^{k}\frac{dx_{I}}{2\pi i x_{I}}\left\langle \prod_{I=1}^{k}\mathsf{A}(x_{I})^{-1}\mathsf{W}_{\bar{4}}(v) \right\rangle,\\
\mathcal{Z}^{\DT\tbar\QA}_{\bar{4};\lambda\mu\nu}[k]&=\frac{\mathcal{G}^{k}}{k!}\oint_{\eta_{0}}\prod_{I=1}^{k}\frac{dx_{I}}{2\pi i x_{I}}\left\langle \prod_{I=1}^{k}\mathsf{A}(x_{I})^{-1}\mathsf{H}_{\bar{4};\lambda\mu\nu}(v)\right\rangle,\\
\mathcal{Z}^{\PT\tbar\QA}_{\bar{4};\lambda\mu\nu}[k]&=\frac{\mathcal{G}^{k}}{k!}\oint_{\tilde{\eta}_{0}}\prod_{I=1}^{k}\frac{dx_{I}}{2\pi i x_{I}}\left\langle \prod_{I=1}^{k}\mathsf{A}(x_{I})^{-1}\mathsf{H}_{\bar{4};\lambda\mu\nu}(v)\right\rangle,\\
\widetilde{\mathcal{Z}}^{\PT\tbar\QA}_{\bar{4};\lambda\mu\nu}[k]&=\frac{\mathcal{G}^{k}}{k!}\oint_{\eta_{0}}\prod_{I=1}^{k}\frac{dx_{I}}{2\pi i x_{I}}\left\langle \prod_{I=1}^{k}\mathsf{A}(x_{I})^{-1}\widetilde{\mathsf{H}}_{\bar{4};\lambda\mu\nu}(v)\right\rangle.
\eea
where we defined
\bea\label{eq:highestweight-def}
\mathsf{H}_{\bar{4};\lambda\mu\nu}(v)&\coloneqq {:\frac{\prod\limits_{\scube \in s(\vec{Y})}\mathsf{W}_{\bar{4}}(\chi_{\bar{4},v}(\cube))}{\prod\limits_{\scube \in p_{1}(\vec{Y})}\mathsf{W}_{\bar{4}}(\chi_{\bar{4},v}(\cube))\prod\limits_{\scube \in p_{2}(\vec{Y})}\mathsf{W}_{\bar{4}}(\chi_{\bar{4},v}(\cube))^{2}} :},\\
\widetilde{\mathsf{H}}_{\bar{4};\lambda\mu\nu}(v)&\coloneqq {:\frac{\prod\limits_{\scube \in p_{1}(\vec{Y})}\mathsf{W}_{\bar{4}}(q_{4}^{-1}\chi_{\bar{4},v^{-1}}(\cube)^{-1})\prod\limits_{\scube \in p_{2}(\vec{Y})}\mathsf{W}_{\bar{4}}(q_{4}^{-1}\chi_{\bar{4},v^{-1}}(\cube)^{-1})^{2}}{\prod\limits_{\scube \in s(\vec{Y})}\mathsf{W}_{\bar{4}}(q_{4}^{-1}\chi_{\bar{4},v^{-1}}(\cube)^{-1})} :}.
\eea
The vertex operators $\mathsf{W}_{\bar{4}}(v),\mathsf{H}_{\bar{4};\lambda\mu\nu}(v),\widetilde{\mathsf{H}}_{\bar{4};\lambda\mu\nu}(v)$ correspond to the framing node of the quiver quantum mechanics and it is called the \textbf{highest weight}.

Using the relation \eqref{eq:vertexoprelation} and the fact that $\mathsf{s}_{a,0}(x)$ and $\mathsf{a}_{0}(x)$ commute with each other, instead of using the $\mathsf{W}_{\bar{4}}(x)$ operators, we can also write the contour integral formulas using the $\mathsf{S}_{a}(x)$ operators. For example, for the one-leg case, we can instead use the following as the highest weight:
\bea
{:\frac{\mathsf{W}_{\bar{4}}(v)}{\prod\limits_{(i,j)\in \nu}\mathsf{S}_{3}(vq_{1}^{i-1}q_{2}^{j-1})}:}.
\eea
This is similar to the infinite product expression discussed in section~\ref{sec:infinite-product-reg}. The plane partition with no boundary conditions correspond to $\mathsf{W}_{\bar{4}}(v)$ and adding boxes in the boundaries correspond to taking the infinite product 
\bea
{:\mathsf{W}_{\bar{4}}(v)\prod_{i=1}^{\infty}\mathsf{A}(q_{3}^{i-1}v)^{-1}:}=
{:\mathsf{W}_{\bar{4}}(v)\prod_{i=1}^{\infty}\frac{\mathsf{S}_{3}(q_{3}^{i}v)}{\mathsf{S}_{3}(vq_{3}^{i-1})}:}={:\frac{\mathsf{W}_{\bar{4}}(v)}{\mathsf{S}_{3}(v)}:}.
\eea
The factor $:\mathsf{S}_{a}(v)^{-1}:$ corresponds to the factor $\mathcal{Z}^{\overline{\D2}_{a}\tbar\D0}(\fra,\phi_I)$.

\paragraph{Contour integral expression of $qq$-characters}
The free field realizations given above imply that the partition functions have operator lift ups. Such operators are called the \textbf{$qq$-characters} and they are essential objects in discussing the BPS/CFT correspondence.

Let us be more concrete and focus on the $U(1)$ D6 partition function. We can define an operator
\bea
\mathsf{T}_{\bar{4}}(v)&=\sum_{k=0}^{\infty}\fq^{k}\frac{\mathcal{G}^{k}}{k!}\oint_{\eta_{0}}\prod_{I=1}^{k}\frac{dx_{I}}{2\pi i x_{I}} \prod_{I=1}^{k}\mathsf{A}(x_{I})^{-1}\mathsf{W}_{\bar{4}}(v) \\
&=\sum_{k=0}^{\infty}\fq^{k}\frac{\mathcal{G}^{k}}{k!}\oint_{\eta_{0}}\prod_{I=1}^{k}\frac{dx_{I}}{2\pi i x_{I}} \prod_{I=1}^{k}\mathscr{V}_{4}\left(\frac{v}{x_{I}}\right)\prod_{I<J}\mathcal{A}_{\mathbb{C}^{4}}\left(\frac{x_{I}}{x_{J}}\right)^{-1}:\prod_{I=1}^{k}\mathsf{A}(x_{I})^{-1}\mathsf{W}_{\bar{4}}(v):\\
&=\sum_{\pi\in\mathcal{PP}}\fq^{|\pi|}\mathcal{Z}^{\D6\tbar\QA}_{\bar{4}}[\pi] :\mathsf{W}_{\bar{4}}(v)\prod_{\scube\in\pi}\mathsf{A}^{-1}(\chi_{\bar{4},x}(\cube)):
\eea
where in the last line, we explicitly evaluated the residues at the poles labeled by plane partitions. Note that the vacuum expectation value of this operator $\langle \mathsf{T}_{\bar{4}}(v) \rangle$ simply gives the partition function. This is the D6 $qq$-character originally introduced in \cite{Kimura:2023bxy} to study the BPS/CFT correspondence of the tetrahedron instantons.

Similarly, for the DT partition function, one can introduce the DT3 $qq$-characters \cite{Kimura:2024osv}:
\bea
\mathsf{DT}_{\bar{4};\lambda\mu\nu}(v)&=\sum_{\pi\in\mathcal{DT}_{\lambda\mu\nu}}\fq^{|\pi|}\mathcal{Z}^{\DT\tbar\QA}_{\bar{4};\lambda\mu\nu}[\pi]\Lambda_{\bar{4},\pi}^{\lambda\mu\nu}(v),\quad \Lambda_{\bar{4},\pi}^{\lambda\mu\nu}(v)=\mathsf{H}_{\bar{4};\lambda\mu\nu}(v) \prod_{\scube\in\pi}\mathsf{A}^{-1}(\chi_{\bar{4},v}(\cube)):.
\eea

One would then ask what would happen for the PT case. We may define the following two types of PT $qq$-characters in the contour integral form:
\bea
\mathsf{PT}_{\bar{4};\lambda\mu\nu}(v)&=\sum_{k=0}^{\infty}\frac{\fq^{k}\mathcal{G}^{k}}{k!}\oint_{\tilde{\eta}_{0}}\prod_{I=1}^{k}\frac{dx_{I}}{2\pi i x_{I}}\prod_{I=1}^{k}\mathcal{Z}^{\D6_{\bar{4}}\tbar\D2\tbar\D0}_{\DT\tbar\QA;\lambda\mu\nu}(v,x_{I})\prod_{I<J}\mathcal{A}_{\mathbb{C}^{4}}\left(\frac{x_{I}}{x_{J}}\right)^{-1}:\mathsf{H}_{\bar{4};\lambda\mu\nu}(v)\prod_{I=1}^{k}\mathsf{A}^{-1}(x_{I}):,\\
\widetilde{\mathsf{PT}}_{\bar{4};\lambda\mu\nu}(v)&=\sum_{k=0}^{\infty}\frac{\fq^{k}\mathcal{G}^{k}}{k!}\oint_{\eta_{0}}\prod_{I=1}^{k}\frac{dx_{I}}{2\pi i x_{I}}\prod_{I=1}^{k}\mathcal{Z}^{\D6_{\bar{4}}\tbar\D2\tbar\D0}_{\PT\tbar\QA;\lambda\mu\nu}(v,x_{I})\prod_{I<J}\mathcal{A}_{\mathbb{C}^{4}}\left(\frac{x_{I}}{x_{J}}\right)^{-1}:\widetilde{\mathsf{H}}_{\bar{4};\lambda\mu\nu}(v)\prod_{I=1}^{k}\mathsf{A}^{-1}(x_{I}):.
\eea
The difference with the other two types of $qq$-characters is that second order poles appear in the integrand when the three legs are all nontrivial (see section~\ref{sec:PTthreelegs} and Appendix~\ref{app:sec-PT3vertex-examples}). Due to this property, derivative of the operator part appears, and the expression is not simply expanded by the $\pi\in\mathcal{PT}_{\lambda\mu\nu},\widetilde{\mathcal{PT}}_{\lambda\mu\nu}$ configurations, but rather extra terms appear.

\paragraph{Two-legs PT3 $qq$-characters}
Let us first consider the case when one of the three legs is trivial. For this case, since the poles are always single order (see Prop.~\ref{prop:PTJK-pole}), after evaluating the residue, the PT3 $qq$-characters have the following expansions
\bea
\mathsf{PT}_{\bar{4};\lambda\mu\nu}(v)&=\sum_{\pi\in\mathcal{PT}_{\lambda\mu\nu}}\fq^{|\pi|}\mathcal{Z}^{\PT\tbar\QA}_{\bar{4};\lambda\mu\nu}[\pi]\mathsf{V}^{\lambda\mu\nu}_{\bar{4},\pi}(v),\quad \mathsf{V}^{\lambda\mu\nu}_{\bar{4},\pi}(v)=:\mathsf{H}_{\bar{4};\lambda\mu\nu}(v)\prod_{\scube\in\pi}\mathsf{A}^{-1}(\chi_{\bar{4},v}(\cube)):,\\
\widetilde{\mathsf{PT}}_{\bar{4};\lambda\mu\nu}(v)&=\sum_{\pi\in\widetilde{\mathcal{PT}}_{\lambda\mu\nu}}\fq^{|\pi|}\widetilde{\mathcal{Z}}^{\PT\tbar\QA}_{\bar{4};\lambda\mu\nu}[\pi]\widetilde{\mathsf{V}}^{\lambda\mu\nu}_{\bar{4},\pi}(v),\quad \mathsf{V}^{\lambda\mu\nu}_{\bar{4},\pi}(v)=:\widetilde{\mathsf{H}}_{\bar{4};\lambda\mu\nu}(v)\prod_{\scube\in\pi}\mathsf{A}^{-1}(\chi_{\bar{4},v}(\cube)):.
\eea

\paragraph{Three-legs PT3 $qq$-character}
When all of the three-legs are nontrivial, the expansion of the PT $qq$-character becomes much more complicated because of the existence of derivative operators acting on the vertex operator part.

Let us see this explicitly for the case $\lambda=\mu=\nu=\Bbox$ (see Fig.~\ref{fig:PT3leg-onebox}) up to level three. At level one, the integrand is
\bea
\mathcal{G}\times x_{1}^{-1}\mathcal{Z}^{\D6_{\bar{4}}\tbar\D2\tbar\D0}_{\DT\tbar\QA;\,\Bbox\,\Bbox\,\Bbox}(v,x_{1}):\mathsf{H}_{\bar{4};\,\Bbox\,\Bbox\,\Bbox}(v)\mathsf{A}^{-1}(x_{1}):.
\eea
Note that we have a second-order pole at $x_{1}=v$. The JK-residue is then evaluated as
\bea\label{eq:PT3qq-level1-JK}
&-\mathcal{G}\times \underset{x_{1}=v}{\Res}x_{1}^{-1}\mathcal{Z}^{\D6_{\bar{4}}\tbar\D2\tbar\D0}_{\DT\tbar\QA;\,\Bbox\,\Bbox\,\Bbox}(v,x_{1}):\mathsf{H}_{\bar{4};\,\Bbox\,\Bbox\,\Bbox}(v)\mathsf{A}^{-1}(x_{1}):\\
=&-\mathcal{G}\times \lim_{x_{1}\rightarrow v}\frac{\partial}{\partial x_{1}}\left((x_{1}-v)^{2}x_{1}^{-1}\mathcal{Z}^{\D6_{\bar{4}}\tbar\D2\tbar\D0}_{\DT\tbar\QA;\,\Bbox\,\Bbox\,\Bbox}(v,x_{1}):\mathsf{H}_{\bar{4};\,\Bbox\,\Bbox\,\Bbox}(v)\mathsf{A}^{-1}(x_{1}):\right)\\
=&-\mathcal{G}\lim_{x_{1}\rightarrow v}\frac{\partial }{\partial x_{1}}\left((x_{1}-v)^{2}x_{1}^{-1}\mathcal{Z}^{\D6_{\bar{4}}\tbar\D2\tbar\D0}_{\DT\tbar\QA;\,\Bbox\,\Bbox\,\Bbox}(v,x_{1})\right):\mathsf{H}_{\bar{4};\,\Bbox\,\Bbox\,\Bbox}(v)\mathsf{A}^{-1}(v):\\
&-\mathcal{G}\times \left.\left(\left(1-\frac{v}{x_{1}}\right)^{2}\mathcal{Z}^{\D6_{\bar{4}}\tbar\D2\tbar\D0}_{\DT\tbar\QA;\,\Bbox\,\Bbox\,\Bbox}(v,x_{1})\right)\right|_{x_{1}\rightarrow v}:\mathsf{H}_{\bar{4};\,\Bbox\,\Bbox\,\Bbox}(v)v\partial_{v}\mathsf{A}^{-1}(v):\\
=&\mathcal{Z}_{\bar{4};\,\Bbox\,\Bbox\,\Bbox}^{\PT\tbar\QA}[1]:\mathsf{H}_{\bar{4};\,\Bbox\,\Bbox\,\Bbox}(v)\mathsf{A}^{-1}(v):-(1-q_{4}):\mathsf{H}_{\bar{4};\,\Bbox\,\Bbox\,\Bbox}(v)v\partial_{v}\mathsf{A}^{-1}(v):.
\eea
The existence of the derivative term is the nontrivial part for the PT $qq$-characters.

For the level two, the poles are $(x_{1},x_{2})=(v,v),(v,q_{1,2,3}^{-1}v)$. Let us first consider the pole $(x_{1},x_{2})=(v,v)$. The contour integrand is schematically the same as \eqref{eq:3leg-level2-1} but now the function $f_{0}(\phi_1,\phi_2)$ includes the vertex operator part. Taking the residue at $x_{1}=v$ gives \eqref{eq:3leg-level2-1}, and derivatives of the vertex operators appear because of $\partial_{\phi_1}f_{0}(\phi_1,\phi_2)$. However, after evaluating the residue at $x_{2}=v$, the derivative terms will not appear and the contribution from $x_{1}=x_{2}=v$ is
\bea
\frac{1}{2}\mathcal{Z}_{(0,0)}^{\QA}:\mathsf{H}_{\bar{4};\,\Bbox\,\Bbox\,\Bbox}(v)\mathsf{A}^{-2}(v):
\eea
where we denoted the JK-residue at $(x_{1},x_{2})=(v,v)$ in the quantum algebraic notation as $\mathcal{Z}_{(0,0)}^{\QA}$. Note that the existence of the ultra-heavy box is identified with two vertex operators $\mathsf{A}^{-2}(v)$ at the same position.

The analysis for $(x_{1},x_{2})=(v,q_{1,2,3}^{-1}v)$ is the same and when evaluating the second pole $x_{2}=q_{1,2,3}^{-1}v$ no derivatives appear (see \eqref{eq:3leg-level2-2-def} and \eqref{eq:3leg-level2-2}):
\bea\mathcal{Z}_{(0,-\eps_i)}^{\QA}:\mathsf{H}_{\bar{4};\,\Bbox\,\Bbox\,\Bbox}(v)\mathsf{A}^{-1}(v)\mathsf{A}^{-1}(q_{i}^{-1}v):,\quad i=1,2,3.
\eea
Therefore, for level two, the PT $qq$-character is
\bea\label{eq:PT3qq-level2-JK}
\frac{1}{2}\mathcal{Z}_{(0,0)}^{\QA}:\mathsf{H}_{\bar{4};\,\Bbox\,\Bbox\,\Bbox}(v)\mathsf{A}^{-2}(v):+\sum_{i=1}^{3}\mathcal{Z}_{(0,-\eps_i)}^{\QA}:\mathsf{H}_{\bar{4};\,\Bbox\,\Bbox\,\Bbox}(v)\mathsf{A}^{-1}(v)\mathsf{A}^{-1}(q_{i}^{-1}v):
\eea

For level three, we have $(x_{1},x_{2},x_{3})=(v,v,q_{i}^{-1}v),(v,q_{i}^{-1}v,q_{i}^{-2}v)$ for $i=1,2,3$. After integrating over the variables $(x_{1},x_{2})$, the pole at $x_{3}=q_{i}^{-1}v$ is a simple pole and thus no-derivative terms of vertex operators appear. The PT $qq$-character at level three thus takes the form as
\bea
&\frac{1}{2}\sum_{i=1}^{3}\mathcal{Z}^{\QA}_{(0,0,-\eps_i)}:\mathsf{H}_{\bar{4};\,\Bbox\,\Bbox\,\Bbox}(v)\mathsf{A}^{-2}(v)\mathsf{A}^{-1}(q_{i}^{-1}v):+\sum_{i=1}^{3}\mathcal{Z}^{\QA}_{(0,-\eps_i,-2\eps_i)}:\mathsf{H}_{\bar{4};\,\Bbox\,\Bbox\,\Bbox}(v)\mathsf{A}^{-1}(v)\mathsf{A}^{-1}(q_{i}^{-1}v)\mathsf{A}^{-1}(q_{i}^{-2}v):
\eea

\begin{remark}
For the two-legs case, the vertex operator part includes the information of the poles and has a one-to-one correspondence with the PT/GR configurations. However, for the three-legs case, the general rules when the derivative terms appear are not so clear and we do not know how to identify them with the PT/GR configurations. A naive expectation is that the derivative terms correspond to the unlabelled type III boxes or the existence of the light boxes. In section~\ref{sec:PT3-threelegs-gl1module}, we will see that the Drinfeld currents $K^{\pm}(z)$ do not act diagonally and non-diagonal terms appear. Such non-diagonal terms correspond to the derivatives appearing in the $qq$-character. A detailed description of such non-diagonal properties is left for future work.
\end{remark}

\subsection{PT3 \texorpdfstring{$qq$}{qq}-characters and screening charges}\label{sec:PT3qq-screeningcharge}
One way to characterize the $qq$-characters is to use the commutativity with the screening charge. In this section, we will show the relations between the PT $qq$-characters and the screening charges. Before moving on to the PT $qq$-characters, let us briefly review the derivation of the D6 $qq$-character using the screening charge following \cite{Kimura:2023bxy}. We first introduce some formulas related to the multiplicative delta function.
\begin{definition}
    The multiplicative delta function is defined as
    \bea
    \delta(x)=\sum_{n\in\mathbb{Z}}x^{n}.
    \eea
\end{definition}
\begin{proposition}
    The delta function obeys
    \bea
     \delta(x)=\left[\frac{1}{1-x}\right]_{|x|<1}-\left[\frac{-x^{-1}}{1-x^{-1}}\right]_{|x|>1}.
    \eea
\end{proposition}
Namely, the difference between the rational functions of the right hand side evaluated at different analytic regions give the delta function. 

The delta function obeys
\bea\label{eq:delta-reflect}
\delta(x)=\delta(x^{-1})
\eea
and
\bea\label{eq:delta-rational}
f(x)\delta\left(\frac{x}{a}\right)=f(a)\delta\left(\frac{x}{a}\right)
\eea
for a rational function $f(x)$. For later use, we denote the series expansion of a rational function $f(x)$ in $x^{\mp 1}$ as $[f(x)]^{x}_{\pm}$. Sometimes we omit the variable $x$ and shortly write $[\cdots]_{\pm}$.

Let us then define the screening charges.
\begin{definition}
    The screening charges are defined as an infinite sum of the D2-brane vertex operators $\mathsf{S}_{a}(x)$:
    \bea
\mathscr{Q}_{a}(x)=\sum_{k\in\mathbb{Z}}\mathsf{S}_{a}(q_{a}^{k}x),\quad a\in\four.
    \eea
\end{definition}
The D6$_{\bar{4}}$ $qq$-character is a sum of operators generated from $\mathsf{W}_{\bar{4}}(v)$ commuting with the screening charge $\mathscr{Q}_{4}(x)$:
\bea
\mathsf{T}_{\bar{4}}(v)=\mathsf{W}_{\bar{4}}(v)+\cdots,\quad [\mathsf{T}_{\bar{4}}(v),\mathscr{Q}_{4}(x)]=0.
\eea
The operators appearing in the $qq$-character are recursively obtained by the iWeyl reflection
\bea
\mathsf{W}_{\bar{4}}(v)\rightarrow :\mathsf{W}_{\bar{4}}(v)\mathsf{A}^{-1}(v):.
\eea
Physically, this process corresponds to adding instanton corrections. For example, the commutation relation of $\mathsf{W}_{\bar{4}}(v)$ and the screening current $\mathsf{S}_{4}(x)$ is
\bea
\relax[\mathsf{W}_{\bar{4}}(v),\mathsf{S}_{4}(x)]=q_{4}v\delta\left(q_{4}v/x\right):\mathsf{W}_{4}(v)\mathsf{S}_{4}(q_{4}v):.
\eea
After iWeyl reflection, we have
\bea
\relax[:\mathsf{W}_{\bar{4}}(v)\mathsf{A}^{-1}(v):,\mathsf{S}_{4}(x)]=v\delta\left(x/v\right)\frac{\prod_{i=1}^{3}(1-q_{i})}{\prod_{1\leq i<j\leq 3}(1-q_{i}q_{j})}:\mathsf{W}_{\bar{4}}(v)\mathsf{S}_{4}(q_{4}v):+\cdots
\eea
where we used \eqref{eq:vertexoprelation} and only extracted the contribution coming from the pole $x=v$. The combination
\bea
\mathsf{W}_{\bar{4}}(v)-q_{4}\frac{\prod_{1\leq i<j\leq 3}(1-q_{i}q_{j})}{\prod_{i=1}^{3}(1-q_{i})}:\mathsf{W}_{\bar{4}}(v)\mathsf{A}^{-1}(v):
\eea
is then not singular at $x=v,q_{4}v$ when we consider the commutating relation with $\mathsf{S}_{4}(x)$. Note that the coefficient indeed gives the one-instanton contribution. 

Choosing the highest weight to be $\mathsf{H}_{\bar{4};\lambda\mu\nu}(v)$, one can show that the DT $qq$-character $\mathsf{DT}_{\bar{4};\lambda\mu\nu}(v)$ commutes with the screening charge
\bea
\relax [\mathsf{DT}_{\bar{4};\lambda\mu\nu}(v),\mathscr{Q}_{4}(x)]=0.
\eea

The fact that the contour integrand of the PT partition function $\mathcal{Z}^{\PT\tbar\QA}_{\bar{4};\lambda\mu\nu}[k]$ has the same free-field realization but the reference vector is different implies that we need a new screening charge.
\begin{definition}
    A new set of screening charges $\widetilde{\mathscr{Q}}_{a}(x)$ are defined as
    \bea
    \widetilde{\mathscr{Q}}_{a}(x)=\sum_{k\in\mathbb{Z}}\mathsf{S}_{a}(q_{a}^{k}x)^{-1},\quad a\in \four.
    \eea
\end{definition}
In other words, the screening charge $\mathscr{Q}_{a}(x)$ corresponds to taking the reference vector to be $\eta=\eta_{0}$, while the screening charge $\widetilde{\mathscr{Q}}_{a}(x)$ corresponds to taking the reference vector to be $\eta=\tilde{\eta}_{0}$.

Note that the new set of screening charges commute with each other similar to the conventional screening charges.
\begin{proposition}
    When $a\neq b$, the screening charges commute with each other:
    \bea
    \relax [\widetilde{\mathscr{Q}}_{a}(x),\widetilde{\mathscr{Q}}_{b}(x')]=0.
    \eea
\end{proposition}

Using the screening charge, the main claim of this section is the following theorem.
\begin{theorem}
When one of the three-legs is trivial, the PT $qq$-character $\mathsf{PT}_{\bar{4};\lambda\mu\nu}(v)$ commutes with the screening charge $\widetilde{\mathscr{Q}}_{4}(x)$:
\bea
\mathsf{PT}_{\bar{4};\lambda\mu\nu}(v)=\mathsf{H}_{\bar{4};\lambda\mu\nu}(v)+\cdots ,\quad [\mathsf{PT}_{\bar{4};\lambda\mu\nu}(v),\widetilde{\mathscr{Q}}_{4}(x)]=0.
\eea
\end{theorem}
In other words, the PT $qq$-character has the same highest weight with the DT $qq$-character but the screening charge that it commutes with is the unique difference. When the three legs are nontrivial, we do not know how to discuss the commutativity with the screening charge because of the existence of the second-order poles. Instead, we will give a different derivation using the collision procedure in section~\ref{sec:PT3-collisionlimit}.

On the other hand, since the conjugate PT partition function $\widetilde{\mathcal{Z}}^{\PT\tbar\QA}_{\bar{4};\lambda\mu\nu}[k]$ is evaluated using the reference vector $\eta=\eta_{0}$, the conjugate PT $qq$-character $\widetilde{\mathsf{PT}}_{\bar{4};\lambda\mu\nu}(x)$ actually commutes with the conventional screening charge.
\begin{theorem}
    When one of the three legs is trivial, the conjugate PT $qq$-character $\widetilde{\mathsf{PT}}_{\bar{4};\lambda\mu\nu}(v)$ commutes with the screening charge $\mathscr{Q}_{4}(x)$:
    \bea
\widetilde{\mathsf{PT}}_{\bar{4};\lambda\mu\nu}(v)=\widetilde{\mathsf{H}}_{\bar{4};\lambda\mu\nu}(v)+\cdots ,\quad [\widetilde{\mathsf{PT}}_{\bar{4};\lambda\mu\nu}(v),\mathscr{Q}_{4}(x)]=0.
    \eea
\end{theorem}

\subsubsection{Two-legs PT3 $qq$-characters}\label{sec:twolegs-PTqq}
We will mainly discuss using the PT $qq$-characters and the discussion is parallel for the conjugate PT $qq$-characters. 
Let us focus on the case when one of the three legs is always trivial. Note again that for such case, the PT configurations and the GR configurations are the same. Since we are interested in the relation with the screening charge, we define the structure function\footnote{Although we wrote the structure function generally using $\lambda,\mu,\nu$, note that we are assuming one of them to be trivial. } $\mathscr{PW}^{\bar{4},\lambda\mu\nu}_{\pi,v}(x)$ as the contraction between the screening charge and the vertex operators of the monomial terms
\bea
\widetilde{\mathsf{S}}_{4}(x)^{-1}\mathsf{V}^{\lambda\mu\nu}_{\bar{4},\pi}(v)&=(-q_{4}v)^{-1}\left[\mathscr{PW}^{\bar{4},\lambda\mu\nu}_{\pi,v}(x)^{-1}\right]^{x}_{-}:\wtS_{4}(x)^{-1}\mathsf{V}^{\lambda\mu\nu}_{\bar{4},\pi}(v):,\\
\mathsf{V}^{\lambda\mu\nu}_{\bar{4},\pi}(v)\widetilde{\mathsf{S}}_{4}(x)^{-1}&=(-q_{4}v)^{-1}\left[\mathscr{PW}^{\bar{4},\lambda\mu\nu}_{\pi,v}(x)^{-1}\right]^{x}_{+}:\wtS_{4}(x)^{-1}\mathsf{V}^{\lambda\mu\nu}_{\bar{4},\pi}(v):.
\eea

An interesting property is that the structure function obeys the following property.
\begin{proposition}\label{prop:PTtwolegs-structure}
    When one of the three legs is trivial, the structure function $\mathscr{PW}^{\bar{4},\lambda\mu\nu}_{\pi,v}(x)$ has zeros only at the positions related with the addable and removable boxes of the PT configuration $\pi$:
    \bea
\mathscr{PW}^{\bar{4},\lambda\mu\nu}_{\pi,v}(x)\propto \prod_{\scube\in A(\pi)}\left(1-\chi_{\bar{4},v}(\cube)/x\right) \prod_{\scube\in R(\pi)}\left(1-q_{4}\chi_{\bar{4},v}(\scube)/x\right).
    \eea
\end{proposition}
Recall that the zeros of the structure functions appearing at the usual D6 $qq$-characters and DT $qq$-characters are at $x=\chi_{\bar{4},v}(\scube)$ for addable boxes and $x=q_{4}^{-1}\chi_{\bar{4},v}(\scube)$ for removable boxes. The origin of this difference is due to the fact that we are using the reference vector $\tilde{\eta}_{0}$ instead of $\eta_{0}$. 

To make the discussion concrete, let us see explicit examples for the one-leg case and two-legs case.

\paragraph{One-leg case}
Let us consider the case when $(\lambda,\mu,\nu)=(\varnothing,\varnothing,\Bbox)$ (see section~\ref{sec:PToneleg} and \eqref{eq:PToneleg-figure}). The vacuum configuration gives
\bea
\mathscr{PW}^{\bar{4},\varnothing\varnothing\,\Bbox}_{\pi,v}(x)^{-1}=\frac{(1-q_{14}v/x)(1-q_{24}v/x)}{(1-q_{3}^{-1}v/x)}
\eea
and indeed we have a pole corresponding to the addable box at $x=q_{3}^{-1}v$. Generally, when we have $k$-boxes at $x=q_{3}^{-1}v,\ldots, q_{3}^{-k}v$, the structure function is 
\bea\label{eq:oneleg-structurefunct-ex}
\mathscr{PW}^{\bar{4},\varnothing\varnothing\,\Bbox}_{\pi,v}(x)^{-1}=(1-q_{4}v/x)\frac{(1-q_{2}^{-1}q_{3}^{-k-1}v/x)(1-q_{1}^{-1}q_{3}^{-k-1}v/x)}{(1-q_{4}q_{3}^{-k}v/x)(1-q_{3}^{-k-1}v/x)}
\eea
and we have an addable pole at $x=q_{3}^{-k-1}v$ and a removable pole at $x=q_{4}q_{3}^{-k}v$.

For an example of a nontrivial case with $(\lambda,\mu,\nu)=(\varnothing,\varnothing,\Yboxdim{4pt}\yng(2,1))$, the vacuum configuration gives
\bea
\mathscr{PW}^{\bar{4},\varnothing\varnothing\,\Yboxdim{4pt}\yng(2,1)}_{\pi,v}(x)^{-1}=\frac{(1-q_{4}q_{1}^{2}v/x)(1-q_{4}q_{2}^{2}v/x)(1-q_{124}v/x)}{(1-q_{3}^{-1}q_{2}v/x)(1-q_{3}^{-1}q_{1}v/x)}
\eea
and indeed the poles correspond to the positions where can add boxes inside the hollow structure.

\paragraph{Two-legs case}
Let us consider the case $(\lambda,\mu,\nu)=(\Yboxdim{5pt}\yng(1),\yng(1),\varnothing)$ (see section~\ref{sec:PTtwolegs} and \eqref{eq:PTtwolegs-figure}). The structure function for the vacuum configuration is
\bea
\mathscr{PW}^{\bar{4},\Yboxdim{4pt}\yng(1)\,\yng(1)\,\varnothing}_{\pi,v}(x)^{-1}=(1-q_{4}v/x)\mathscr{S}_{23}(q_{4}v/x)\mathscr{S}_{13}(q_{24}v/x)=\frac{(1-q_{12}^{-1}v/x)(1-q_{3}^{-1}v/x)}{(1-v/x)}
\eea
and we have a pole corresponding to the addable box at the origin. For level one with a box at the origin, the structure function is
\bea
\mathscr{PW}^{\bar{4},\Yboxdim{4pt}\yng(1)\,\yng(1)\,\varnothing}_{\pi,v}(x)^{-1}=(1-q_{4}v/x)\mathscr{S}_{23}(q_{14}v/x)\mathscr{S}_{13}(q_{24}v/x)g_{\bar{4}}(q_{4}v/x)=\frac{(1-q_{12}^{-1}v/x)^{2}(1-q_{13,23}^{-1}v/x)}{(1-q_{1,2}^{-1}v/x)(1-q_{4}v/x)}
\eea
and indeed it obeys Prop.~\ref{prop:PTtwolegs-structure}. Generally, for the configuration $\pi=[m,n]$, we have
\bea\label{eq:twoleg-structurefunct-ex}
\mathscr{PW}^{\bar{4},\Yboxdim{4pt}\yng(1)\,\yng(1)\,\varnothing}_{\pi,v}(x)^{-1}&=(1-q_{4}v/x)\mathscr{S}_{23}\left(\frac{q_{4}q_{1}^{-m}v}{x}\right)\mathscr{S}_{13}\left(\frac{q_{4}q_{2}^{-n}v}{x}\right)\\
&=(1-q_{4}v/x)\frac{(1-q_{2,3}^{-1}q_{1}^{-m-1}v/x)(1-q_{1,3}^{-1}q_{2}^{-n-1}v/x)}{(1-q_{1}^{-m-1}v/x)(1-q_{4}q_{1}^{-m}v/x)(1-q_{2}^{-n-1}v/x)(1-q_{4}q_{2}^{-n}v/x)}
\eea
which gives addable poles at $x=q_{1}^{-m-1}v,q_{2}^{-n-1}v$ and removable poles at $x=q_{4}q_{1}^{-m}v,q_{4}q_{2}^{-n}v$.

\paragraph{Commutativity with the screening charge}
Using Prop.~\ref{prop:PTtwolegs-structure}, we can show the commutativity of the PT $qq$-character with the screening charge. We first have
\bea
\,&[\mathsf{V}^{\lambda\mu\nu}_{\bar{4},\pi}(v),\mathsf{S}_{4}(x)^{-1}]\\
=&(q_{4}v)^{-1}\left(\sum_{\scube\in A(\pi)}\underset{x=\chi_{\bar{4},v}(\scube)}{\Res}{x}^{-1} \mathscr{PW}^{\bar{4},\lambda\mu\nu}_{\pi,v}(x)^{-1} \delta\left(\frac{x}{\chi_{\bar{4},v}(\cube)}\right)  :\mathsf{V}^{\lambda\mu\nu}_{\bar{4},\pi}(v)\mathsf{S}_{4}(\chi_{\bar{4},v}(\cube))^{-1}: \right.\\
    &+\left. \sum_{\scube\in R(\pi)}\underset{x=q_{4}\chi_{\bar{4},x}(\scube)}{\Res}{x'}^{-1} \mathscr{PW}^{\bar{4},\lambda\mu\varnothing}_{\pi,v}(x)^{-1} \delta\left(\frac{x}{q_{4}\chi_{\bar{4},v}(\cube)}\right)  :\mathsf{V}^{\lambda\mu\nu}_{\bar{4},\pi}(v)\mathsf{S}_{4}(q_{4}\chi_{\bar{4},v}(\cube))^{-1}: \right).
\eea
The commutativity with the PT $qq$-character is then given as
\bea
\relax &[\mathsf{PT}_{\bar{4};\lambda\mu\nu}(v),\mathsf{S}_{4}(x)^{-1}]\\
=&(q_{4}v)^{-1}\sum_{\pi\in\mathcal{PT}_{\lambda\mu\nu}}\mathcal{Z}^{\PT\tbar\QA}_{\bar{4};\lambda\mu\nu}[\pi]\left(\sum_{\scube\in A(\pi)}\underset{x=\chi_{\bar{4},v}(\scube)}{\Res}{x}^{-1} \mathscr{PW}^{\bar{4},\lambda\mu\nu}_{\pi,v}(x)^{-1} \delta\left(\frac{x}{\chi_{\bar{4},v}(\cube)}\right)  :\mathsf{V}^{\lambda\mu\nu}_{\bar{4},\pi}(v)\mathsf{S}_{4}(\chi_{\bar{4},v}(\cube))^{-1}: \right.\\
    &+\left. \sum_{\scube\in R(\pi)}\underset{x=q_{4}\chi_{\bar{4},x}(\scube)}{\Res}{x'}^{-1} \mathscr{PW}^{\bar{4},\lambda\mu\varnothing}_{\pi,v}(x)^{-1} \delta\left(\frac{x}{q_{4}\chi_{\bar{4},v}(\cube)}\right)  :\mathsf{V}^{\lambda\mu\nu}_{\bar{4},\pi}(v)\mathsf{S}_{4}(q_{4}\chi_{\bar{4},v}(\cube))^{-1}: \right).
\eea
Shifting the second term as $\pi'=\pi-\scube$, the second term is rewritten as
\bea
\sum_{\pi'\in\mathcal{PT}_{\lambda\mu\nu}}\mathcal{Z}^{\PT\tbar\QA}_{\bar{4};\lambda\mu\nu}[\pi'+\cube] \sum_{\scube\in A(\pi')}\underset{x=q_{4}\chi_{\bar{4},v}(\scube)}{\Res}{x}^{-1} \mathscr{PW}^{\bar{4},\lambda\mu\nu}_{\pi'+\scube,v}(x)^{-1} \delta\left(\frac{x}{q_{4}\chi_{\bar{4},v}(\cube)}\right)  :\mathsf{V}^{\lambda\mu\nu}_{\bar{4},\pi'+\scube}(v)\mathsf{S}_{4}(q_{4}\chi_{\bar{4},v}(\cube))^{-1}:.
\eea
Using \eqref{eq:vertexoprelation}, the vertex operator part becomes 
\bea
{:\mathsf{V}^{\lambda\mu\nu}_{\bar{4},\pi'+\scube}(v)\mathsf{S}_{4}(q_{4}\chi_{\bar{4},v}(\cube))^{-1}:}={:\mathsf{V}^{\lambda\mu\nu}_{\bar{4},\pi'}(v)\mathsf{S}_{4}(\chi_{\bar{4},v}(\cube))^{-1}:}.
\eea
We then have
\bea
\relax &[\mathsf{PT}_{\bar{4};\lambda\mu\nu}(v),\mathsf{S}_{4}(x)^{-1}]\\
=&(q_{4}v)^{-1}\sum_{\pi\in\mathcal{PT}_{\lambda\mu\nu}}\mathcal{Z}^{\PT\tbar\QA}_{\bar{4};\lambda\mu\nu}[\pi]\sum_{\scube\in A(\pi)}\underset{x=\chi_{\bar{4},v}(\scube)}{\Res}{x}^{-1} \mathscr{PW}^{\bar{4},\lambda\mu\nu}_{\pi,v}(x)^{-1}  :\mathsf{V}^{\lambda\mu\nu}_{\bar{4},\pi}(v)\mathsf{S}_{4}(\chi_{\bar{4},v}(\cube))^{-1}:\\
&\times \left(\delta\left(\frac{x}{\chi_{\bar{4},v}(\cube)}\right)-\delta\left(\frac{x}{q_{4}\chi_{\bar{4},v}(\cube)}\right)\right)
\eea
which is written in a $q_{4}$-difference form and commutes with the screening charge $\widetilde{\mathscr{Q}}_{4}(x)$ if the PT partition function obeys a recursion relation given in the following proposition. Using the following proposition, we obtain the commutativity.
\begin{proposition}
    When one of the three legs is trivial, the PT partition function obeys the following recursion relation
    \bea
\frac{\mathcal{Z}^{\PT\tbar\QA}_{\bar{4};\lambda\mu\nu}[\pi+\cube]}{\mathcal{Z}^{\PT\tbar\QA}_{\bar{4};\lambda\mu\nu}[\pi]}=-\frac{\underset{x=\chi_{\bar{4},v}(\scube)}{\Res}{x}^{-1} \mathscr{PW}^{\bar{4},\lambda\mu\nu}_{\pi,v}(x)^{-1}}{\underset{x=q_{4}\chi_{\bar{4},v}(\scube)}{\Res}{x}^{-1} \mathscr{PW}^{\bar{4},\lambda\mu\nu}_{\pi+\scube,v}(x)^{-1}}.
    \eea
\end{proposition}
Conversely, when one of the legs is trivial, the PT partition function can be bootstrapped using the above recursion formula.

For example, for the one-leg case with $(\lambda,\mu,\nu)=(\varnothing,\varnothing,\Bbox)$, the recursion formula is solved as
\bea
\frac{\mathcal{Z}^{\PT\tbar\QA}_{\bar{4};\varnothing\varnothing\,\Bbox}[k+1]}{\mathcal{Z}^{\PT\tbar\QA}_{\bar{4};\varnothing\varnothing\,\Bbox}[k]}=\frac{1-q_{4}q_{3}^{k+1}}{1-q_{3}^{k+1}},\quad \mathcal{Z}^{\PT\tbar\QA}_{\bar{4};\varnothing\varnothing\,\Bbox}[k]=\prod_{i=1}^{k}\frac{1-q_{4}q_{3}^{i}}{1-q_{3}^{i}}
\eea
where we imposed the initial condition $\mathcal{Z}^{\PT\tbar\QA}_{\bar{4};\varnothing\varnothing\,\Bbox}[0]=1$. This indeed matches with \eqref{eq:PToneleg-JK-partfunct} up to redefinition of the topological term.

For the two-legs case with  $(\lambda,\mu,\nu)=(\Yboxdim{5pt}\yng(1),\yng(1),\varnothing)$, we have
\bea
\frac{\mathcal{Z}^{\PT\tbar\QA}_{\bar{4};\,\Yboxdim{5pt}\yng(1)\,\yng(1)\,\varnothing}[[m,n+1]]}{\mathcal{Z}^{\PT\tbar\QA}_{\bar{4};\,\Yboxdim{5pt}\yng(1)\,\yng(1)\,\varnothing}[[m,n]]}&=\frac{1-q_{4}q_{2}^{n+1}}{1-q_{2}^{n+1}}\frac{\mathscr{S}_{23}(q_{4}q_{1}^{-m}/q_{2}^{-n-1})}{\mathscr{S}_{23}(q_{1}^{-m}/q_{2}^{-n-1})},\\
\frac{\mathcal{Z}^{\PT\tbar\QA}_{\bar{4};\,\Yboxdim{5pt}\yng(1)\,\yng(1)\,\varnothing}[[m+1,n]]}{\mathcal{Z}^{\PT\tbar\QA}_{\bar{4};\,\Yboxdim{5pt}\yng(1)\,\yng(1)\,\varnothing}[[m,n]]}&=\frac{1-q_{4}q_{1}^{m+1}}{1-q_{1}^{m+1}}\frac{\mathscr{S}_{13}(q_{4}q_{2}^{-n}/q_{1}^{-m-1})}{\mathscr{S}_{13}(q_{2}^{-n}/q_{1}^{-m-1})},\\
\frac{\mathcal{Z}^{\PT\tbar\QA}_{\bar{4};\,\Yboxdim{5pt}\yng(1)\,\yng(1)\,\varnothing}[[0,0]]}{\mathcal{Z}^{\PT\tbar\QA}_{\bar{4};\,\Yboxdim{5pt}\yng(1)\,\yng(1)\,\varnothing}[\varnothing]}&=-q_{4}\frac{(1-q_{12})(1-q_{23})(1-q_{13})}{(1-q_{1})(1-q_{2})(1-q_{3})}.
\eea
Imposing the initial condition $\mathcal{Z}^{\PT\tbar\QA}_{\bar{4};\,\Yboxdim{5pt}\yng(1)\,\yng(1)\,\varnothing}[\varnothing]=1$, we can solve the recursion relation:
\bea
\mathcal{Z}^{\PT\tbar\QA}_{\bar{4};\,\Bbox\,\Bbox\,\varnothing}[[m,n]]&=\frac{1-q_{34}}{1-q_{3}}\times \prod_{i=1}^{m}\frac{1-q_{4}q_{1}^{i}}{1-q_{1}^{i}}\prod_{j=1}^{n}\frac{1-q_{4}q_{2}^{j}}{1-q_{2}^{j}} \times  \frac{(1-q_{3}q_{1}^{m+1}/q_{2}^{n})(1-q_{4}q_{1}^{m+1}/q_{2}^{n})}{(1-q_{1}^{m+1}/q_{2}^{n})(1-q_{34}q_{1}^{m+1}/q_{2}^{n})}.
\eea
This indeed matches with \eqref{eq:PTtwoleg-JK-partfunct} up to redefinition of the topological term.

\paragraph{Conjugate PT $qq$-character}
We briefly summarize the formulas and properties related to the conjugate PT $qq$-characters. The structure functions are defined as
\bea
\mathsf{S}_{4}(x)\widetilde{\mathsf{V}}^{\lambda\mu\nu}_{\bar{4},\pi}(v)&=(-q_{4}v)\left[\widetilde{\mathscr{PW}}^{\bar{4},\lambda\mu\nu}_{\pi,v}(x)^{-1}\right]^{x}_{+}:\mathsf{S}_{4}(x)\widetilde{\mathsf{V}}^{\lambda\mu\nu}_{\bar{4},\pi}(v):,\\
\widetilde{\mathsf{V}}^{\lambda\mu\nu}_{\bar{4},\pi}(v)\mathsf{S}_{4}(x)&=(-q_{4}v)\left[\widetilde{\mathscr{PW}}^{\bar{4},\lambda\mu\nu}_{\pi,v}(x)^{-1}\right]^{x}_{-}:\mathsf{S}_{4}(x)\widetilde{\mathsf{V}}^{\lambda\mu\nu}_{\bar{4},\pi}(v):.
\eea
\begin{proposition}\label{prop:PTtwolegs-structure-conjugate}
    When one of the three-legs are trivial, the structure function $\widetilde{\mathscr{PW}}^{\bar{4},\lambda\mu\nu}_{\pi,v}(x)$ has zeros only at the positions related with the addable and removable boxes of the conjugate PT configuration $\pi$:
    \bea
\widetilde{\mathscr{PW}}^{\bar{4},\lambda\mu\nu}_{\pi,v}(x)\propto \prod_{\scube\in A(\pi)}\left(1-\chi_{\bar{4},v}(\cube)/x\right) \prod_{\scube\in R(\pi)}\left(1-q_{4}^{-1}\chi_{\bar{4},v}(\scube)/x\right).
    \eea
\end{proposition}
After studying the commutativity with the screening charge $\mathscr{Q}_{4}(x)$, we will see that the conjugate PT partition functions obey the following recursion formula.
\begin{proposition}
    When one of the three-legs is trivial, the conjugate PT partition function obeys the recursion formula
    \bea
\frac{\widetilde{\mathcal{Z}}^{\PT\tbar\QA}_{\bar{4};\lambda\mu\nu}[\pi+\cube]}{\widetilde{\mathcal{Z}}^{\PT\tbar\QA}_{\bar{4};\lambda\mu\nu}[\pi]}=-\frac{\underset{x=\chi_{\bar{4},v}(\scube)}{\Res}{x}^{-1} \widetilde{\mathscr{PW}}^{\bar{4},\lambda\mu\nu}_{\pi,v}(x)^{-1}}{\underset{x=q_{4}^{-1}\chi_{\bar{4},v}(\scube)}{\Res}{x}^{-1} \widetilde{\mathscr{PW}}^{\bar{4},\lambda\mu\nu}_{\pi+\scube,v}(x)^{-1}}.
    \eea
\end{proposition}

\subsubsection{Three-legs PT3 $qq$-characters and collision limit}\label{sec:PT3-collisionlimit}
Compared to the two-legs PT3 $qq$-character, the contraction between the three-legs $qq$-character and the screening charge contains a second order pole. Let us focus on the case $(\lambda,\mu,\nu)=(\Bbox,\Bbox,\Bbox)$:
\bea
\mathsf{H}_{\bar{4};\,\Bbox\,\Bbox\,\Bbox}(v)\mathsf{S}_{4}(x)^{-1}&=x^{-2}\frac{(1-q_{1}x/v)(1-q_{2}x/v)(1-q_{3}x/v)}{(1-x/v)^{2}}:\mathsf{H}_{\bar{4};\,\Bbox\,\Bbox\,\Bbox}(v)\mathsf{S}_{4}(x)^{-1}:,\\
\mathsf{S}_{4}(x)^{-1}\mathsf{H}_{\bar{4};\,\Bbox\,\Bbox\,\Bbox}(v)&=-(q_{4}v)^{-1}\frac{(1-q_{1}^{-1}v/x)(1-q_{2}^{-1}v/x)(1-q_{3}^{-1}v/x)}{(1-v/x)^{2}}:\mathsf{H}_{\bar{4};\,\Bbox\,\Bbox\,\Bbox}(v)\mathsf{S}_{4}(x)^{-1}:.
\eea
The commutation relation then includes derivatives of delta functions due to the second order pole and one needs to deal with them to study the commutativity.

Unfortunately, in the context of quantum algebra, a detailed discussion in this direction is missing. Instead, we will use the collision limit to derive $qq$-characters with higher order poles. The strategy to take the collision is as follows.\footnote{ See Appendix~\ref{app-sec:A1qqcharacter-collisionlimit} for a review of the application to the $A_{1}$ $qq$-character.} Let $\mathsf{T}_{1,2}(x_{1,2})$ be $qq$-characters commuting with the screening charge whose highest weights are $\mathscr{O}_{1,2}(x_{1,2})$, respectively. We further assume that the operator product with the screening charge contains higher order poles at $x_{2}=\alpha x_{1}$. 

The $qq$-character with the highest weight ${:\mathcal{O}_{1}(x_{1})\mathcal{O}_{2}(x_{2}):}$ with generic $x_{1,2}$ can be obtained by studying the commutation relation with the screening charge. The $qq$-character will take the form as
\bea\label{eq:higherrankqq}
:\mathcal{O}_{1}(x_{1})\mathcal{O}_{2}(x_{2}):+\mathfrak{q}\left(\# :\mathcal{O}_{1}(x_{1})\mathsf{A}^{-1}(x_{1})\mathcal{O}_{2}(x_{2}):+\#:\mathcal{O}_{1}(x_{1})\mathsf{A}^{-1}(x_{2})\mathcal{O}_{2}(x_{2}):\right)+\cdots,
\eea
with some coefficients depending on the operators and the spectral parameters. 

A different way to derive this is to study the fusing of the $qq$-characters $\mathsf{T}_{1}(x_{1})\mathsf{T}_{2}(x_{2})$:
\bea\label{eq:higherrankqq2}
\mathsf{f}_{12}\left(\frac{x_{2}}{x_{1}}\right)\mathsf{T}_{1}(x_{1})\mathsf{T}_{2}(x_{2})&={:\mathcal{O}_{1}(x_{1})\mathcal{O}_{2}(x_{2}):}+\cdots ,\quad
\mathcal{O}_{1}(x_{1})\mathcal{O}_{2}(x_{2})=\mathsf{f}_{12}\left(\frac{x_{2}}{x_{1}}\right)^{-1}:\mathcal{O}_{1}(x_{1})\mathcal{O}_{2}(x_{2}):.
\eea
Since $\mathsf{T}_{1}(x_{1}), \mathsf{T}_{2}(x_{2})$ commute with the screening charge, the fused $qq$-character also commutes with it. 

To get the $qq$-character with $:\mathcal{O}_{1}(x_{1})\mathcal{O}_{2}(\alpha x_{1}):$, we simply take the limit $x_{2}\rightarrow\alpha x_{1}$ of the right hand side in \eqref{eq:higherrankqq}, \eqref{eq:higherrankqq2}. Each of the coefficients $\#$ is singular generally, but after combining all the terms, we may get a well-defined operator sum. If so, the new operator sum gives a new $qq$-character.

Let us perform this explicitly. We focus only on the PT $qq$-character and we omit the discussion for the conjugate PT $qq$-character. The highest weight is
\bea
\mathsf{H}_{\bar{4};\,\Bbox\,\Bbox\,\Bbox}(v)&={:\mathsf{H}_{\bar{4};\,\Bbox\,\Bbox\,\varnothing}(v)\frac{\mathsf{W}_{\bar{4}}(q_{3}v)\mathsf{W}_{\bar{4}}(q_{3}q_{2}q_{1}v)}{\mathsf{W}_{\bar{4}}(q_{3}q_{1}v)\mathsf{W}_{\bar{4}}(q_{3}q_{2}v)}:}\simeq :\mathsf{H}_{\bar{4};\,\Bbox\,\Bbox\,\varnothing}(v)\mathsf{S}_{3}(q_{3}v)^{-1}:
\eea
where in the last line we instead used the screening charge expression. Note that the overall zero-modes only differ but the commutation relation does not change. A natural way to decompose the highest weight is $:\mathsf{H}_{\bar{4};\,\Bbox\,\Bbox\,\varnothing}(v_1)\mathsf{S}_{3}(q_{3}v_{2})^{-1}:$. One would then expect that we need to fuse the two-legs PT3 $qq$-character $\mathsf{PT}_{\bar{4};\,\Bbox\,\Bbox\varnothing}(v_{1})$ and the screening charge $\widetilde{\mathscr{Q}_{3}}(q_{3}v_{2})$ and then take the collision limit $v_{2}\rightarrow v_{1}$.

The main claim of this section is the following conjecture.
\begin{conjecture}\label{conj:PT3leg-fusion-total}
     The three-legs PT3 $qq$-character $\mathsf{PT}_{\bar{4};\,\Bbox\,\Bbox\,\Bbox}(v_{1})$ is obtained by the collision limit of the two-legs PT3 $qq$-character $\mathsf{PT}_{\bar{4};\,\Bbox\,\Bbox\,\varnothing}(v_{1})$ and the screening charge $\widetilde{\mathscr{Q}}_{3}(q_{3}v_{2})$ with $v_{2}\rightarrow v_{1}$:
    \bea
    \mathsf{f}\left(\frac{v_{2}}{v_{1}}\right)\mathsf{PT}_{\bar{4};\,\Bbox\,\Bbox\,\varnothing}(v_{1})\widetilde{\mathscr{Q}}_{3}(q_{3}v_{2})\xrightarrow{v_{2}\rightarrow v_{1}} \mathsf{PT}_{\bar{4};\,\Bbox\,\Bbox\,\Bbox}(v_{1}).
    \eea
\end{conjecture}

Let us show this is indeed the case for low levels. We first decompose the screening charge into two parts
\bea
\widetilde{\mathscr{Q}}_{3}(q_{3}v_{2})
&=\sum_{k=0}^{\infty}\mathfrak{q}^{k}{:\mathsf{S}_{3}(q_{3}v_{2})^{-1}\prod_{i=1}^{k}\mathsf{A}^{-1}(q_{3}^{-i+1}v_{2}):}+\sum_{k=1}^{\infty}\mathfrak{q}^{-k}:\mathsf{S}_{3}(q_{3}v_{2})^{-1}\prod_{i=1}^{k}\mathsf{A}(v_{2}q_{3}^{i}):\\
&=\widetilde{\mathscr{Q}}^{+}_{3}(q_{3}v_{2})+\widetilde{\mathscr{Q}}^{-}_{3}(q_{3}v_{2})
\eea
where we included the topological term $\mathfrak{q}$ for convenience.

\paragraph{Negative part}Let us first consider the fusion with the negative part of the screening charge $\widetilde{\mathscr{Q}}_{3}^{-}(q_{3}v_{2})$. An interesting property is that after taking the limit $v_{2}\rightarrow v_{1}$, actually the fusion with the negative part vanishes.
\begin{proposition}\label{prop:PT3leg-negativescreening}
    For each term of the two-legs PT $qq$-character $\mathsf{PT}_{\bar{4};\Bbox\,\Bbox\,\varnothing}(v_{1})$, the fusion with each term of $\widetilde{\mathscr{Q}}_{3}^{-}(v_{2})$ vanishes at the limit $v_{2}\rightarrow v_{1}$:
    \bea
\mathsf{f}\left(\frac{v_{2}}{v_{1}}\right)\mathsf{V}^{\,\Bbox\,\Bbox\,\varnothing}_{\bar{4},[m,n]}(v_{1})\mathsf{S}_{3}(q_{3}v_{2}q_{3}^{k})\xrightarrow {v_{2}\rightarrow v_{1}}0,
    \eea
    for $m,n\in\mathbb{Z}\geq 0$ and $k>0$. Hence, we have
    \bea
    \mathsf{f}\left(\frac{v_{2}}{v_{1}}\right)\mathsf{PT}_{\bar{4};\,\Bbox\,\Bbox\,\varnothing}(v_{1})\widetilde{\mathscr{Q}}^{-}_{3}(q_{3}v_{2})\xrightarrow{v_{2}\rightarrow v_{1}}0.
    \eea
\end{proposition}
This is because after direct computation, the operator product is proportional to $(v_{2}/v_{1};q_{3})_{k+1}$ which contains a zero at $v_{2}=v_{1}$.

The fusion with the positive part is computed from
\bea
\mathsf{f}\left(\frac{v_{2}}{v_{1}}\right)\mathsf{PT}_{\bar{4};\,\Bbox\,\Bbox\,\varnothing}(v_{1})\widetilde{\mathscr{Q}}^{+}_{3}(q_{3}v_{2})={:\mathsf{H}_{\bar{4};\,\Bbox\,\Bbox\,\varnothing}(v_1)\mathsf{S}_{3}(q_{3}v_{2})^{-1}:}+\cdots 
\eea
where $\mathsf{f}(x)$ is the inverse of the OPE between $\mathsf{H}_{\bar{4};\,\Bbox\,\Bbox\,\varnothing}(v_1)$ and $\mathsf{S}_{3}(q_{3}v_{2})^{-1}$. After expanding the left hand side, we then take the limit $v_{2}\rightarrow v_{1}$. Let us explicitly perform this computation for low levels.

\paragraph{Level one}
For level one, we have
\bea
&\mathfrak{q}\left(q_{4}\mathscr{V}_{4}\left(q_{4}^{-1}\frac{v_{2}}{v_{1}}\right)^{-1}g_{\bar{1}}\left(q_{1}\frac{v_{2}}{v_{1}}\right)g_{\bar{2}}\left(\frac{v_{2}}{v_{1}}\right) :\mathsf{H}_{\bar{4};\,\Bbox\,\Bbox\,\varnothing}(v_{1})\mathsf{S}_{3}(q_{3}v_{2})^{-1}\mathsf{A}^{-1}(v_{2}): \right.\\
&+\quad \left. \mathcal{Z}^{\PT\tbar\QA}_{\bar{4};\,\Bbox\,\Bbox\,\varnothing}[[0,0]]g_{\bar{3}}\left(q_{3}\frac{v_{2}}{v_{1}}\right)^{-1}:\mathsf{H}_{\bar{4};\,\Bbox\,\Bbox\,\varnothing}(v_{1})\mathsf{A}^{-1}(v_{1})\mathsf{S}_{3}(q_{3}v_{2})^{-1}:\right).
\eea
Denoting $v_{2}=zv_{1}$, the coefficients are singular at $z=1$:
\bea
q_{4}\mathscr{V}_{4}\left(q_{4}^{-1}\frac{v_{2}}{v_{1}}\right)^{-1}g_{\bar{1}}\left(q_{1}\frac{v_{2}}{v_{1}}\right)g_{\bar{2}}\left(\frac{v_{2}}{v_{1}}\right)&=q_{4}\frac{(1-q_{12}z)(1-q_{4}z)(1-q_{3}z)}{(1-z)(1-q_{34}z)(1-q_{3}^{-1}z)},\\
g_{\bar{3}}\left(q_{3}\frac{v_{2}}{v_{1}}\right)^{-1}&=\frac{(1-q_{1,2,4}^{-1}z)(1-q_{3}z)}{(1-z)(1-q_{34,13,23}z)}.
\eea

The operator parts are expanded around $z=1$ as
\bea
:\mathsf{H}_{\bar{4};\,\Bbox\,\Bbox\,\varnothing}(v_{1})\mathsf{S}_{3}(q_{3}v_{2})^{-1}\mathsf{A}^{-1}(v_{2}):&=:\mathsf{H}_{\bar{4};\,\Bbox\,\Bbox\,\varnothing}(v_{1})\mathsf{S}_{3}(q_{3}v_{1})^{-1}\mathsf{A}^{-1}(v_{1}):\\
&\quad -(1-z):\mathsf{H}_{\bar{4};\,\Bbox\,\Bbox\,\varnothing}(v_{1})\partial_{\log v_{2}}\left(\mathsf{S}_{3}(q_{3}v_{2})^{-1}\mathsf{A}^{-1}(v_{2})\right)|_{v_{2}=v_{1}}:,\\
:\mathsf{H}_{\bar{4};\,\Bbox\,\Bbox\,\varnothing}(v_{1})\mathsf{A}^{-1}(v_{1})\mathsf{S}_{3}(q_{3}v_{2})^{-1}:&=:\mathsf{H}_{\bar{4};\,\Bbox\,\Bbox\,\varnothing}(v_{1})\mathsf{A}^{-1}(v_{1})\mathsf{S}_{3}(q_{3}v_{1})^{-1}:\\
&\quad -(1-z):\mathsf{H}_{\bar{4};\,\Bbox\,\Bbox\,\varnothing}(v_{1})\mathsf{A}^{-1}(v_{1})\partial_{\log v_{2}}\mathsf{S}_{3}(q_{3}v_{2})^{-1}|_{v_{2}=v_{1}}:
\eea
where we used
\bea
f(zx)
&=f(x)+(z-1)\partial_{\log x}f(x)+\cdots.
\eea

The terms proportional to $1/(1-z)$ are then
\bea
&\left(q_{4}\mathscr{V}_{4}\left(q_{4}^{-1}\frac{v_{2}}{v_{1}}\right)^{-1}g_{\bar{1}}\left(q_{1}\frac{v_{2}}{v_{1}}\right)g_{\bar{2}}\left(\frac{v_{2}}{v_{1}}\right)+\mathcal{Z}^{\PT\tbar\QA}_{\bar{4};\,\Bbox\,\Bbox\,\varnothing}[[0,0]]g_{\bar{3}}\left(q_{3}\frac{v_{2}}{v_{1}}\right)^{-1}\right):\mathsf{H}_{\bar{4};\,\Bbox\,\Bbox\,\varnothing}(v_{1})\mathsf{A}^{-1}(v_{1})\mathsf{S}_{3}(q_{3}v_{1})^{-1}:\\
\xrightarrow {z\rightarrow 1} &\mathcal{Z}^{\PT\tbar\QA}_{\bar{4};\,\Bbox\,\Bbox\,\Bbox}[1] :\mathsf{H}_{\bar{4};\,\Bbox\,\Bbox\,\varnothing}(v_{1})\mathsf{A}^{-1}(v_{1})\mathsf{S}_{3}(q_{3}v_{1})^{-1}:
\eea
and the singularity at $z=1$ is resolved.

The remaining terms are
\bea
&-(1-z)q_{4}\mathscr{V}_{4}\left(q_{4}^{-1}\frac{v_{2}}{v_{1}}\right)^{-1}g_{\bar{1}}\left(q_{1}\frac{v_{2}}{v_{1}}\right)g_{\bar{2}}\left(\frac{v_{2}}{v_{1}}\right):\mathsf{H}_{\bar{4};\,\Bbox\,\Bbox\,\varnothing}(v_{1})\partial_{\log v_{2}}\left(\mathsf{S}_{3}(q_{3}v_{2})^{-1}\mathsf{A}^{-1}(v_{2})\right)|_{v_{2}=v_{1}}:\\
\xrightarrow{z\rightarrow 1}& -(1-q_{4}):\mathsf{H}_{\bar{4};\,\Bbox\,\Bbox\,\varnothing}(v_{1})\partial_{\log v_{2}}\left(\mathsf{S}_{3}(q_{3}v_{2})^{-1}\mathsf{A}^{-1}(v_{2})\right)|_{v_{2}=v_{1}}:,
\eea
and
\bea
&-(1-z)\mathcal{Z}^{\PT\tbar\QA}_{\bar{4};\,\Bbox\,\Bbox\,\varnothing}[[0,0]]g_{\bar{3}}\left(q_{3}\frac{v_{2}}{v_{1}}\right)^{-1}:\mathsf{H}_{\bar{4};\,\Bbox\,\Bbox\,\varnothing}(v_{1})\mathsf{A}^{-1}(v_{1})\partial_{\log v_{2}}\mathsf{S}_{3}(q_{3}v_{2})^{-1}|_{v_{2}=v_{1}}:\\
\xrightarrow{z\rightarrow 1}&\,\,(1-q_{4}):\mathsf{H}_{\bar{4};\,\Bbox\,\Bbox\,\varnothing}(v_{1})\mathsf{A}^{-1}(v_{1})\partial_{\log v_{2}}\mathsf{S}_{3}(q_{3}v_{2})^{-1}|_{v_{2}=v_{1}}:.
\eea
Combining these terms gives
\bea
-(1-q_{4}):\mathsf{H}_{\bar{4};\,\Bbox\,\Bbox\,\varnothing}(v_{1})\mathsf{S}_{3}(q_{3}v_{1})^{-1}\partial_{\log v_{1}}\mathsf{A}^{-1}(v_{1}):.
\eea


Therefore, at level one we have
\bea
&\mathfrak{q}\left(\mathcal{Z}^{\PT\tbar\QA}_{\bar{4};\,\Bbox\,\Bbox\,\Bbox}[1] :\mathsf{H}_{\bar{4};\,\Bbox\,\Bbox\,\varnothing}(v_{1})\mathsf{A}^{-1}(v_{1})\mathsf{S}_{3}(q_{3}v_{1})^{-1}:-(1-q_{4}):\mathsf{H}_{\bar{4};\,\Bbox\,\Bbox\,\varnothing}(v_{1})\mathsf{S}_{3}(q_{3}v_{1})^{-1}\partial_{\log v_{1}}\mathsf{A}^{-1}(v_{1}):\right)\\
\simeq &\mathfrak{q}\left(\mathcal{Z}^{\PT\tbar\QA}_{\bar{4};\,\Bbox\,\Bbox\,\Bbox}[1] :\mathsf{H}_{\bar{4};\,\Bbox\,\Bbox\,\Bbox}(v_{1})\mathsf{A}^{-1}(v_{1}):-(1-q_{4}):\mathsf{H}_{\bar{4};\,\Bbox\,\Bbox\,\Bbox}(v_{1})v_{1}\partial_{v_{1}}\mathsf{A}^{-1}(v_{1}):\right)
\eea
which matches with \eqref{eq:PT3qq-level1-JK}. In the last line, we simply incorporate the non-essential zero-modes for comparison.

\paragraph{Level two}For level two, the terms related are
\bea
\mathfrak{q}^{2}\times \mathsf{f}\left(\frac{v_{2}}{v_{1}}\right)&\left(\mathcal{Z}^{\PT\tbar\QA}_{\bar{4};\,\Bbox\,\Bbox\,\varnothing}[[0,0]] :\mathsf{H}_{\bar{4};\,\Bbox\,\Bbox\,\varnothing}(v_{1})\mathsf{A}^{-1}(v_{1})::\mathsf{S}_{3}(q_{3}v_{2})^{-1}\mathsf{A}^{-1}(v_{2}):\right.\\
&+\mathcal{Z}^{\PT\tbar\QA}_{\bar{4};\,\Bbox\,\Bbox\,\varnothing}[[1,0]] :\mathsf{H}_{\bar{4};\,\Bbox\,\Bbox\,\varnothing}(v_{1})\mathsf{A}^{-1}(v_{1})\mathsf{A}^{-1}(q_{1}^{-1}v_{1})::\mathsf{S}_{3}(q_{3}v_{2})^{-1}:\\
&+\mathcal{Z}^{\PT\tbar\QA}_{\bar{4};\,\Bbox\,\Bbox\,\varnothing}[[0,1]] :\mathsf{H}_{\bar{4};\,\Bbox\,\Bbox\,\varnothing}(v_{1})\mathsf{A}^{-1}(v_{1})\mathsf{A}^{-1}(q_{2}^{-1}v_{1})::\mathsf{S}_{3}(q_{3}v_{2})^{-1}:\\
&\left.+\mathsf{H}_{\bar{4};\,\Bbox\,\Bbox\,\varnothing}(v_{1}):\mathsf{S}_{3}(q_{3}v_{2})^{-1}\mathsf{A}^{-1}(v_{2})\mathsf{A}^{-1}(q_{3}^{-1}v_{2}):\right).
\eea
Studying the contraction of the vertex operators, one will see that each term are non-singular at $v_{2}\rightarrow v_{1}$. After taking this limit, we obtain
\bea
&\mathfrak{q}^{2}\left(\frac{1}{2}\mathcal{Z}^{\QA}:\mathsf{H}_{\bar{4};\,\Bbox\,\Bbox\,\varnothing}(v_{1})\mathsf{S}_{3}(q_{3}v_{1})^{-1}\mathsf{A}^{-2}(v_{1}):+\sum_{i=1}^{3}\mathcal{Z}^{\QA}_{(0,-\eps_i)}:\mathsf{H}_{\bar{4};\,\Bbox\,\Bbox\,\varnothing}(v_{1})\mathsf{S}_{3}(q_{3}v_{1})^{-1}\mathsf{A}^{-1}(v)\mathsf{A}^{-1}(q_{i}^{-1}v_{1}):\right)\\
\simeq&\mathfrak{q}^{2}\left( \frac{1}{2}\mathcal{Z}_{(0,0)}^{\QA}:\mathsf{H}_{\bar{4};\,\Bbox\,\Bbox\,\Bbox}(v)\mathsf{A}^{-2}(v):+\sum_{i=1}^{3}\mathcal{Z}_{(0,-\eps_i)}^{\QA}:\mathsf{H}_{\bar{4};\,\Bbox\,\Bbox\,\Bbox}(v)\mathsf{A}^{-1}(v)\mathsf{A}^{-1}(q_{i}^{-1}v):\right)
\eea
where $\mathcal{Z}^{\QA}_{(a,b)}$ are the three-legs PT partition functions. This final formula indeed matches with \eqref{eq:PT3qq-level2-JK}. 

One can perform the computation for higher levels and we expect it generally holds.
\begin{conjecture}\label{conj:PT3leg-positivescreening}
    The three-legs PT3 $qq$-character $\mathsf{PT}_{\bar{4};\,\Bbox\,\Bbox\,\Bbox}(v_{1})$ is obtained by the collision limit of the two-legs PT3 $qq$-character $\mathsf{PT}_{\bar{4};\,\Bbox\,\Bbox\,\varnothing}(v_{1})$ and the positive part of the screening charge $\widetilde{\mathscr{Q}}_{3}^{+}(q_{3}v_{2})$ with $v_{2}\rightarrow v_{1}$:
    \bea
    \mathsf{f}\left(\frac{v_{2}}{v_{1}}\right)\mathsf{PT}_{\bar{4};\,\Bbox\,\Bbox\,\varnothing}(v_{1})\widetilde{\mathscr{Q}}^{+}_{3}(q_{3}v_{2})\xrightarrow{v_{2}\rightarrow v_{1}} \mathsf{PT}_{\bar{4};\,\Bbox\,\Bbox\,\Bbox}(v_{1}).
    \eea
\end{conjecture}
Combining Conj.~\ref{conj:PT3leg-positivescreening} and Prop.~\ref{prop:PT3leg-negativescreening}, we obtain Conj.~\ref{conj:PT3leg-fusion-total}.

\paragraph{General three-legs PT3 $qq$-characters}
Starting from the minimal plane partition with $(\lambda,\mu,\nu)$, we can decompose it into a minimal plane partition with $(\lambda,\mu,\varnothing)$ and a stack of one-dimensional rods extending in the 3-axis. We denote the coordinates of them as $y_{1}^{\ast},\ldots,y_{|\nu|}^{\ast}$. For the example Fig.~\ref{fig:minimal-pp}, the coordinates are
\bea
vq_{3}^{3},\quad vq_{1}q_{3}^{3},\quad vq_{2}q_{3}^{2},\quad vq_{2}^{2}q_{3}^{2}.
\eea

For general three-legs PT $qq$-characters $\mathsf{PT}_{\bar{4};\lambda\mu\nu}(v)$, we first start from two-legs PT $qq$-characters $\mathsf{PT}_{\bar{4};\lambda\mu\varnothing}(v)$ and fuse $|\nu|$ number of screening charges:
\bea
\mathsf{f}(v,y_{1},\ldots,y_{|\nu|})\mathsf{PT}_{\bar{4};\lambda\mu\varnothing}(v)\widetilde{\mathscr{Q}}_{3}(y_{1})\widetilde{\mathscr{Q}}_{3}(y_{2})\cdots \widetilde{\mathscr{Q}}_{3}(y_{|\nu|}),\quad y_{i}\rightarrow y_{i}^{\ast} 
\eea
where $\mathsf{f}(v,y_{1},\ldots,y_{|\nu|})$ is a suitable function. We note that we have to take the limits in way that for each step the Young diagram condition for $\nu$ is satisfied. For the example in Fig.~\ref{fig:minimal-pp}, we can take the limits in the order
\bea
y_1\rightarrow vq_{3}^{3},\quad y_2\rightarrow  vq_{1}q_{3}^{3},\quad y_3\rightarrow vq_{2}q_{3}^{2},\quad y_4\rightarrow vq_{2}^{2}q_{3}^{2}.
\eea

\subsection{Relation with (shifted) quantum toroidal \texorpdfstring{$\mathfrak{gl}_{1}$}{gl(1)}}\label{sec:QTgl1}
In this section, we discuss the relation between the PT3 $qq$-characters and the quantum toroidal $\mathfrak{gl}_{1}$ \cite{ding1997generalization,miki2007q,FFJMM1,Feigin2011plane,Feigin2011} (see \cite{DIMreview} for a review). The discussion here is the trigonometric version of the algebraic interpretation of the PT counting discussed in \cite{Gaiotto:2020dsq} (see also \cite{Gaberdiel:2017hcn,Gaberdiel:2018nbs}).

\begin{definition}
Let $\mathsf{q}_{1},\mathsf{q}_{2},\mathsf{q}_{3}$ be the deformation parameters with the condition $\sfq_{1}\sfq_{2}\sfq_{3}=1$. The quantum toroidal $\mathfrak{gl}_{1}$, which is denoted $\mathcal{E}$, is generated by three Drinfeld currents 
\begin{equation}
    E(z)=\sum_{m\in\mathbb{Z}}E_{m}z^{-m},\quad F(z)=\sum_{m\in\mathbb{Z}}F_{m}z^{-m},\quad K^{\pm}(z)=\sum_{r\geq 0}K^{\pm}_{\pm r} z^{\mp r}
\end{equation}
where $K_{0}^{\pm}=K^{\pm}$ and central elements 
\begin{equation}
    C,\quad K^{-}=(K^{+})^{-1}.
\end{equation}
The defining relations are
\begin{align}
\begin{split}
    E(z)E(w)=\sfg(z/w)E(w)E(z),&\quad F(z)F(w)=\sfg(z/w)^{-1}F(w)F(z),\\
    K^{\pm}(z)K^{\pm}(w)=K^{\pm}(w)K^{\pm}(z),&\quad K^{-}(z)K^{+}(w)=\frac{\sfg(C^{-1}z/w)}{\sfg(Cz/w)}K^{+}(w)K^{-}(z),\\
    K^{\pm}(C^{(1\mp 1)/2}z)E(w)&=\sfg(z/w)E(w)K^{\pm}(C^{(1\mp1)/2}z),\\
    K^{\pm}(C^{(1\pm 1)/2}z)F(w)&=\mathsf{g}(z/w)^{-1}F(w)K^{\pm}(C^{(1\pm 1)/2}z),\\
    [E(z),F(w)]=\tilde{g}&\left(\delta\left(\frac{Cw}{z}\right)K^{+}(z)-\delta\left(\frac{Cz}{w}\right)K^{-}(w)\right)
\end{split}
\end{align}
where 
\begin{equation}\label{eq:gl1structurefunction}
    \sfg(z)=\frac{\prod_{i=1}^{3}(1-\sfq_{i}z)}{\prod_{i=1}^{3}(1-\sfq_{i}^{-1}z)},\quad \kappa_{r}=\prod_{i=1}^{3}(\sfq_{i}^{r/2}-\sfq_{i}^{-r/2}),
\end{equation}
and $\tilde{g}=1/\kappa_{1}$.
\end{definition}

Additionally, one needs the Serre relations but we did not write the explicit form (see \cite{DIMreview}). 

The coproduct structure is
\begin{align}\label{eq:coproduct}
\begin{split}
\Delta E(z)&=E(z)\otimes 1+K^{-}(C_{1}z)\otimes E(C_{1}z),\\
\Delta F(z)&=F(C_{2}z)\otimes K^{+}(C_{2}z)+1\otimes F(z),\\
\Delta K^{+}(z)&=K^{+}(z)\otimes K^{+}(C_{1}^{-1}z),\\
\Delta K^{-}(z)&=K^{-}(C_{2}^{-1}z)\otimes K^{-}(z),\\
\Delta(X)&=X\otimes X,\quad X=C,K^{-},
\end{split}
\end{align}
where $C_{1}=C\otimes 1$ and $C_{2}=1\otimes C$. Using this coproduct, we can consider tensor product representations.

The representations of the quantum toroidal $\mathfrak{gl}_{1}$ are obtained by determining the values of the central elements $C,K^{-}$. In this paper, we only consider the vertical representations with the central charge $C=1$ \cite{FFJMM1,Feigin2011,Feigin2011plane}. Let us first review two basic representations: the vector representation and the MacMahon representation.

\paragraph{Vector representation}
There are three types of vector representations with central charges $(C,K^{-})=(1,1)$ and the actions of the Drinfeld currents are 
\begin{align}\label{eq:vectorrep}
  \begin{split}
        K^{\pm}(z)[u]^{(c)}_{j}=&\left[\Psi_{[u]^{(c)}_{j}}(z)\right]^{z}_{\pm}[u]^{(c)}_{j}\eqqcolon[S_{c}\left(u\sfq_{c}^{j}/z\right)]^{z}_{\pm}[u]^{(c)}_{j},\\
        E(z)[u]^{(c)}_{j}=&\mathcal{E}\delta\left(u\sfq_{c}^{j}/z\right)[u]^{(c)}_{j+1},\\
        F(z)[u]^{(c)}_{j}=&\mathcal{F}\delta\left(u\sfq_{c}^{j-1}/z\right)[u]^{(c)}_{j-1},\quad c=1,2,3,\quad j\in\mathbb{Z}
   \end{split}
\end{align}
where
\begin{equation}
    \mathcal{E}\mathcal{F}=\Tilde{g}\frac{(1-\sfq_{c+1}^{-1})(1-\sfq_{c-1}^{-1})}{(1-\sfq_{c})},\quad S_{c}(z)=\frac{(1-\sfq_{c-1}z)(1-\sfq_{c+1}z)}{(1-z)(1-\sfq_{c-1}\sfq_{c+1}z)}.
\end{equation}
We denote these representations $\mathcal{V}_{c}(u),\,(c=1,2,3)$. The bases $\{[u]_{j}^{(c)}\}_{j\in\mathbb{Z}}$ are represented by 1d partitions.

\paragraph{MacMahon representation}
MacMahon representations are representations with central charges $(C,K^{-})=(1,K^{1/2})$ where $K\in\mathbb{C}^{\times}$ is a generic parameter. The action of the Drinfeld currents is given as
\bea\label{eq:MacMahonrep}
    K^{\pm}(z)|u,\pi\rangle&=\left[\Psi_{\pi,u}(z)\right]^{z}_{\pm}|u,\pi\rangle,\\
    E(z)|u,\pi\rangle&=\sum_{\scube\in A(\pi)}E(\pi\rightarrow \pi+\cube)\delta\left(\frac{z}{\chi_{u}(\cube)}\right)
    |u,\pi+\cube\rangle,\\
    F(z)|u,\pi\rangle&=\sum_{\scube\in R(\pi)}F(\pi\rightarrow \pi-\cube)\delta\left(\frac{z}{\chi_{u}(\cube)}\right)
    |u,\pi-\cube\rangle,
\eea 
where  
\begin{equation}\label{eq:CartanMacMahon}
    \chi_{u}(\cube)=u\sfq_{1}^{i-1}\sfq_{2}^{j-1}\sfq_{3}^{k-1},\quad  \Psi_{\pi,u}(z)=K^{-1/2}\frac{1-Ku/z}{1-u/z}\prod_{\scube\in\pi}\sfg\left(\frac{\chi_{u}(\cube)}{z}\right)^{-1}.
\end{equation}
 We denote this representation $M(u,K)$. 

 A general property of the charge function $\Psi_{\pi,u}(z)$ is that the pole structure is determined by the addable and removable boxes of the plane partition:
 \bea
\Psi_{\pi,u}(z)\propto \prod_{\scube\in A(\pi)}\frac{1}{(1-\chi_{u}(\cube)/z)}\prod_{\scube\in R(\pi)}\frac{1}{(1-\chi_{u}(\cube)/z)}.
 \eea
 We shortly call the poles corresponding to the addable (removable) boxes \textit{addable (removable) poles}, respectively. The coefficients $E(\pi\rightarrow \pi+\cube),F(\pi\rightarrow \pi-\cube)$ can be \textit{bootstrapped} by imposing the algebraic relations to hold. For example, if we us the $EF$ relation, we obtain the constraint
 \bea
E(\pi\rightarrow \pi+\cube) F(\pi+\cube\rightarrow \pi)=-\tilde{g}\times \Res_{z=\chi_{u}(\scube)}z^{-1}\Psi_{\pi,u}(z).
 \eea

\paragraph{Shifted quantum toroidal algebra}
The series expansion of the Drinfeld currents $K^{\pm}(z)$ above starts from $z^{0}$ and so the Cartan eigenvalues of the vertical representations always have the same number of numerators and denominators. We can relax this condition and consider rational functions with different number of numerators and denominators. Roughly speaking, the currents $K^{\pm}(z)$ are modified to $z^{r_{\pm}}K^{\pm}(z)$, while keeping the defining relations to be the same. Such algebras are called the shifted quantum toroidal algebra (see \cite{Hernandez:2020xih,Noshita:2021dgj,Li:2020rij,Galakhov:2021xum,Galakhov:2021vbo} for recent references). When using the current realizations, derivation of representations can be done similarly to the unshifted case. We will not discuss subtle properties about the mode expansions but simply use this current description.


\subsubsection{DT counting}
When the plane partition has nontrivial boundary conditions, we need to modify the vacuum charge function so that it includes the boundary contributions \cite{Feigin2011plane}.

The vacuum charge function is given as
\begin{equation}
    \Psi^{\DT,\lambda\mu\nu}_{\varnothing,u}(z)=K^{-1/2}\frac{1-Ku/z}{1-u/z}\prod_{\scube\in \mathcal{B}_{\lambda,\mu,\nu}}\sfg\left(\frac{\chi_{u}(\cube)}{z}\right)^{-1}.
\end{equation}
Since the right hand side is an infinite product, one needs to regularize it properly.

As an example, let us consider the situation when the boundary conditions are $(\lambda,\mu,\nu)=(\varnothing,\varnothing,\Bbox)$:
\bea
 \Psi^{\DT,\varnothing\varnothing\,\Bbox}_{\varnothing,u}(z)&=K^{-1/2}\frac{1-Ku/z}{1-u/z}\prod_{i=1}^{\infty}\sfg\left(\frac{u\sfq_{3}^{i-1}}{z}\right)^{-1}=K^{-1/2}\frac{1-Ku/z}{1-u/z}S_{3}\left(\frac{u}{z}\right)^{-1}\\
 &=K^{-1/2}\frac{(1-Ku/z)(1-\sfq_{12}u/z)}{(1-\sfq_{1}u/z)(1-\sfq_{2}u/z)}
\eea
where we used $\sfg(z)=S_{3}(z)/S_{3}(\sfq_{3}z).$ Namely, the semi-infinite sequence of boxes placed at the position $u$ extending in the direction $3$ is understood as inserting the $S_{3}(u/z)^{-1}$ function. Charge functions for situations with boundary boxes extending in other directions are obtained by using the function $S_{c}(z)$.

For the two-legs case $(\lambda,\mu,\nu)=(\Bbox,\Bbox,\varnothing)$, we can insert a rod extending in the $1$-axis at the origin $u$ and a rod extending in $2$-axis at $u\sfq_{2}$ and so the vacuum charge function is 
\bea
\Psi^{\DT,\,\Bbox\,\Bbox\,\varnothing}_{\varnothing,u}(z)&=K^{-1/2}\frac{1-Ku/z}{1-u/z}S_{1}\left(\frac{u}{z}\right)^{-1}S_{2}\left(\frac{u\sfq_{2}}{z}\right)^{-1}=K^{-1/2}\frac{(1-Ku/z)(1-u/z)}{(1-\sfq_{3}u/z)(1-\sfq_{12}u/z)}.
\eea

For the three-legs case $(\lambda,\mu,\nu)=(\Bbox,\Bbox,\Bbox)$, the vacuum charge function is
\bea
\Psi^{\DT,\,\Bbox\,\Bbox\,\Bbox}_{\varnothing,u}(z)&=K^{-1/2}\frac{1-Ku/z}{1-u/z}S_{1}\left(\frac{u}{z}\right)^{-1}S_{2}\left(\frac{u\sfq_{2}}{z}\right)^{-1}S_{3}\left(\frac{u\sfq_{3}}{z}\right)^{-1}\\
&=K^{-1/2}\frac{(1-Ku/z)(1-u/z)^{2}}{(1-\sfq_{12}u/z)(1-\sfq_{13}u/z)(1-\sfq_{23}u/z)}.
\eea

For the general boundary condition $\vec{Y}=(\lambda,\mu,\nu)$, the vacuum charge function is simplified as 
\bea
\Psi^{\DT,\lambda\mu\nu}_{\varnothing,u}(z)=K^{-1/2}\frac{(1-Ku/z)\prod\limits_{\scube\in p_{1}(\vec{Y})}\left(1-\chi_{u}(\cube)/z\right)\prod\limits_{\scube\in p_{2}(\vec{Y})}\left(1-\chi_{u}(\cube)/z\right)^{2}}{\prod\limits_{\scube\in s(\vec{Y})}\left(1-\chi_{u}(\cube)/z\right)}
\eea
where the numbers of the factors in the numerator and the denominators are the same.

The action of the Drinfeld currents are obtained similar to \eqref{eq:MacMahonrep}. Let $\pi$ be the set of boxes we can add to the minimal plane partition and the charge function is generally
\bea
\Psi_{\pi,u}(z)=
\Psi^{\DT,\lambda\mu\nu}_{\varnothing,u}(z)\prod_{\scube\in\pi}\sfg\left(\frac{\chi_{u}(\cube)}{z}\right)^{-1}.
\eea
The poles are the coordinates of the addable and removable boxes of the configuration. The action of the Drinfeld currents are obtained by inserting this charge function to \eqref{eq:MacMahonrep}.

\subsubsection{One-leg and two-legs PT counting}
To construct the module of the PT counting, we introduce a conjugate map on the vacuum charge function:
\bea
\Psi^{\DT,\lambda\mu\nu}_{\varnothing,u}(z)&\longmapsto \Psi^{\PT,\lambda\mu\nu}_{\varnothing,u}(z)=\frac{1}{\Psi^{\DT,\lambda\mu\nu}_{\varnothing,u^{-1}}(z^{-1})},\\
K&\longmapsto K^{-1}.
\eea
In particular, when we have a one-dimensional rod extending in the $c$-axis whose coordinate is $\chi_{u}(\cube)$, the conjugate map acts as 
\bea
S_{c}\left(\chi_{u}(\cube)/z\right)^{-1}\longmapsto S_{c}\left(\chi_{u^{-1}}(\cube)/z^{-1}\right )=S_{c}\left(\sfq_{c}\frac{\chi_{u^{-1}}(\cube)^{-1}}{z}\right).
\eea

For example, for the one-leg and two-legs cases, we have
\bea
\Psi^{\PT,\varnothing\varnothing\,\Bbox}_{\varnothing,u}(z)&=K^{1/2}\frac{(1-\sfq_{1}^{-1}u/z)(1-\sfq_{2}^{-1}u/z)}{(1-uK/z)(1-u\sfq_{3}/z)},\\
\Psi^{\PT,\Bbox\,\Bbox\,\varnothing}_{\varnothing,u}(z)&=K^{1/2}\frac{(1-\sfq_{12}u/z)(1-\sfq_{3}u/z)}{(1-Ku/z)(1-u/z)}
\eea
and for the three-legs case, we have
\bea
\Psi^{\PT,\Bbox\,\Bbox\,\Bbox}_{\varnothing,u}(z)&=K^{1/2}\frac{(1-\sfq_{12}^{-1}u/z)(1-\sfq_{13}^{-1}u/z)(1-\sfq_{23}^{-1}u/z)}{(1-Ku/z)(1-u/z)^{2}}.
\eea
Generally, we have
\bea
\Psi^{\PT,\lambda\mu\nu}_{\varnothing,u}(z)=K^{1/2}\frac{\prod\limits_{\scube\in s(\vec{Y})}\left(1-\chi_{u^{-1}}(\cube)^{-1}/z\right)}{(1-Ku/z)\prod\limits_{\scube\in p_{1}(\vec{Y})}\left(1-\chi_{u^{-1}}(\cube)^{-1}/z\right)\prod\limits_{\scube\in p_{2}(\vec{Y})}\left(1-\chi_{u^{-1}}(\cube)^{-1}/z\right)^{2}}.
\eea

An observation is that there is always a pole at $1-Ku/z$ and using such pole the Drinfeld current $E(z)$ generates a plane partition whose origin is at $Ku$. Another observation is that for the one-leg and two-legs cases, the pole structure is single order similar to the DT counting case, while for the three-legs case, a second order pole exists. We thus need to treat the three-legs case separately.

Since we are interested only in the module structure obtained from the poles except $1-Ku/z$, we remove the contributions with the parameter $K$ and construct the module. The module will not be a module of the quantum toroidal $\mathfrak{gl}_{1}$ but the shifted quantum toroidal $\mathfrak{gl}_{1}$. We will see that the resulting module is related with the PT counting.\footnote{Strictly speaking, the module is related with the conjugate PT counting (see Fig.~\ref{fig:PTreverse}).}

\paragraph{One-leg case}Let us focus on the case $(\varnothing,\varnothing,\Bbox)$. The vacuum charge function is
\bea
\Psi^{\PT,\varnothing\varnothing\,\Bbox}_{\varnothing,u}(z)&=\frac{(1-\sfq_{1}^{-1}u/z)(1-\sfq_{2}^{-1}u/z)}{(1-u\sfq_{3}/z)}
\eea
where we omitted the contribution coming from the pole $1-Ku/z$. The pole at $z=u\sfq_{3}$ means we can add a box $z=u\sfq_{3}$ at level one. 

The charge function at level one is computed as
\bea
\Psi^{\PT,\varnothing\varnothing\,\Bbox}_{\varnothing,u}(z)\sfg\left(\frac{u\sfq_{3}}{z}\right)^{-1}=\frac{(1-\sfq_{1}^{-1}\sfq_{3}u/z)(1-\sfq_{2}^{-1}\sfq_{3}u/z)(1-u/z)}{(1-u\sfq_{3}/z)(1-u\sfq_{3}^{2}/z)}.
\eea
The pole $z=u\sfq_{3}$ corresponds to the removable box while the pole $z=u\sfq_{3}^{2}$ corresponds to the addable box to the configuration. Recursively, the level $k$ configuration has boxes at $u\sfq_{3},u\sfq_{3}^{2},\ldots, u\sfq_{3}^{k}$ and the charge function is
\bea
\Psi^{\PT,\varnothing\varnothing\,\Bbox}_{\varnothing,u}(z)\prod_{i=1}^{k}\sfg\left(\frac{u\sfq_{3}^{i}}{z}\right)^{-1}=\frac{(1-u/z)(1-\sfq_{1}\sfq_{3}^{k+1}u/z)(1-\sfq_{2}\sfq_{3}^{k+1}u/z)}{(1-u\sfq_{3}^{k}/z)(1-u\sfq_{3}^{k+1}/z)}
\eea
where we have an addable box at $z=u\sfq_{3}^{k+1}$ and a removable box at $z=u\sfq_{3}^{k}$. The bases of the module indeed matches with the PT configurations in section~\ref{sec:PTrule-coord}.

Let us construct the action of the Drinfeld currents. Based on the analysis above, the bases are labeled by non-negative integers $\ket{k}\,(k\in\mathbb{Z}_{\geq 0})$:
\bea
E(z)\ket{k}&=\mathcal{E}_{k}\delta\left(\frac{z}{u\sfq_{3}^{k+1}}\right)\ket{k+1},\\
F(z)\ket{k}&=\mathcal{F}_{k}\delta\left(\frac{z}{u\sfq_{3}^{k}}\right)\ket{k-1},\quad \mathcal{F}_{0}=0,\\
K^{\pm}(z)\ket{k}&=\left[\Psi_{k,u}(z)\right]_{\pm}^{z}\ket{k}
\eea
where
\bea
\Psi_{k,u}(z)=\Psi^{\PT,\varnothing\varnothing\,\Bbox}_{\varnothing,u}(z)\prod_{i=1}^{k}\sfg\left(\frac{u\sfq_{3}^{i}}{z}\right)^{-1}.
\eea
Using the $EF$ relation, we obtain the constraint
\bea
\mathcal{E}_{k}\mathcal{F}_{k+1}=\frac{(1-\sfq_{12})(1-\sfq_{3}^{-k-1})}{(1-\sfq_{1})(1-\sfq_{2})},\quad \mathcal{F}_{0}=0.
\eea
Under this constraint, the defining relations are satisfied. For later use, we use
\bea
\mathcal{E}_{k}=\frac{(1-\sfq_{12})}{(1-\sfq_{1})(1-\sfq_{2})},\quad \mathcal{F}_{k}=1-\sfq_{3}^{-k}.
\eea

\paragraph{Two-legs case}Again, let us consider the example $(\Bbox\,\Bbox\,\varnothing)$. The vacuum charge function is
\bea
\Psi^{\PT,\,\Bbox\,\Bbox\,\varnothing}_{\varnothing,u}(z)&=\frac{(1-\sfq_{12}u/z)(1-\sfq_{3}u/z)}{(1-u/z)}
\eea
and the pole corresponds to the box we can add at $z=u$. For level one, we have
\bea
\Psi^{\PT,\,\Bbox\,\Bbox\,\varnothing}_{\varnothing,u}(z)\sfg\left(\frac{u}{z}\right)^{-1}=\frac{(1-\sfq_{12}u/z)^{2}(1-\sfq_{13,23}u/z)}{(1-\sfq_{1,2}u/z)(1-u/z)}.
\eea
The pole at $z=u$ correspond to the removable box and the poles $z=\sfq_{1,2}u$ correspond to the addable boxes. Recursively, one will see that the generic configuration $\pi=[m,n]$ has $k=m+n+1$ boxes at $z=u,u\sfq_{1},u\sfq_{1}^{2},\ldots,u\sfq_{1}^{m},\sfq_{2}u,\ldots,u\sfq_{2}^{n}$. The charge function is computed as
\bea
&\Psi^{\PT,\,\Bbox\,\Bbox\,\varnothing}_{\varnothing,u}(z)\sfg\left(\frac{u}{z}\right)^{-1}\prod_{i=1}^{m}\sfg\left(\frac{u\sfq_{1}^{i}}{z}\right)^{-1}\prod_{j=1}^{n}\sfg\left(\frac{u\sfq_{2}^{j}}{z}\right)^{-1}\\
=&(1-u/z)\frac{(1-\sfq_{2,3}\sfq_{1}^{m+1}u/z)(1-\sfq_{1,3}\sfq_{2}^{n+1}u/z)}{(1-u\sfq_{1}^{m+1}/z)(1-u\sfq_{1}^{m}/z)(1-u\sfq_{2}^{n+1}/z)(1-u\sfq_{2}^{n}/z)}.
\eea

The bases of the module are $\ket{0},\ket{[m,n]}\,\,(m,n\in\mathbb{Z}_{\geq 0})$. The action of the Drinfeld currents $E(z),F(z)$ are
\bea
E(z)\ket{0}&=E(\varnothing\rightarrow [0,0])\delta\left(\frac{z}{u}\right)\ket{[0,0]},\\
E(z)\ket{[m,n]}&=E([m,n]\rightarrow [m+1,n])\delta\left(\frac{z}{u\sfq_{1}^{m+1}}\right)\ket{[m+1,n]}\\
&+E([m,n]\rightarrow [m,n+1])\delta\left(\frac{z}{u\sfq_{2}^{n+1}}\right)\ket{[m,n+1]},\\
F(z)\ket{[0,0]}&=F([0,0]\rightarrow \varnothing)\delta\left(\frac{z}{u}\right)\ket{0},\\
F(z)\ket{[m,n]}&=F([m,n]\rightarrow [m-1,n])\delta\left(\frac{z}{u\sfq_{1}^{m}}\right)\ket{[m-1,n]}\\
&+E([m,n]\rightarrow [m,n-1])\delta\left(\frac{z}{u\sfq_{2}^{n}}\right)\ket{[m,n-1]}.
\eea
For the Cartan part, we have
\bea
 K^{\pm}(z)|0\rangle&=\Psi^{\PT,\Bbox\,\Bbox\,\varnothing}_{\varnothing,u}(z)|\pi\rangle,\quad K^{\pm}(z)|[m,n]\rangle=\Psi_{[m,n],u}(z)|[m,n]\rangle,\\
 \Psi_{[m,n],u}(z)&=\Psi^{\PT,\Bbox\,\Bbox\,\varnothing}_{\varnothing,u}(z)\sfg\left(\frac{u}{z}\right)^{-1}\prod_{i=1}^{m}\sfg\left(\frac{u\sfq_{1}^{i}}{z}\right)^{-1}\prod_{j=1}^{n}\sfg\left(\frac{u\sfq_{2}^{j}}{z}\right)^{-1}.
\eea
The coefficients can be derived by studying the $EF$-relation similarly. 

Generally, when one of the tree-legs is trivial, the Cartan eigenvalue of the Drinfeld currents $K^{\pm}(z)$ take the form
\bea
\Psi^{\PT,\lambda\mu\nu}_{\varnothing,u}(z)\prod_{\scube\in\pi}\sfg\left(\frac{\chi_{u}(\cube)}{z}\right)^{-1}
\eea
for $\pi\in \widetilde{\mathcal{PT}}_{\lambda\mu\nu}$.

\begin{remark}
    A different way to derive the module is to use the coproduct structure and study the tensor product representations. We leave a detailed study of such properties for future work.
\end{remark}

\begin{remark}
We can also construct a module corresponding to the PT counting directly. We take the vacuum charge function as
\bea
\widetilde{\Psi}^{\PT,\lambda\mu\nu}_{\varnothing,u}(z)=\frac{1}{\Psi^{\DT,\lambda\mu\nu}_{\varnothing,u}(z)}.
\eea
For this case, $E(z)$ and $F(z)$ act as annihilation and creation operators, respectively. This kind of property also appears when one studies representations whose central charges have negative levels \cite{Noshita:2022dxv}.

For example, for the one-leg case, we have
\bea
\widetilde{\Psi}^{\PT,\lambda\mu\nu}_{\varnothing,u}(z)=\frac{(1-\sfq_1 u/z)(1-\sfq_2 u/z)}{(1-\sfq_{3}^{-1}u/z)}
\eea
and
\bea
E(z)\ket{k}&=\mathcal{E}_{k}\delta\left(\frac{z}{u\sfq_{3}^{-k}}\right)\ket{k-1},\quad \mathcal{E}_0=0\\
F(z)\ket{k}&=\mathcal{F}_{k}\delta\left(\frac{z}{u\sfq_{3}^{-k-1}}\right)\ket{k+1}.
\eea
From the defining relations, we can see that adding boxes modifies the charge functions by multiplying $\sfg(z)$ instead of $\sfg(z)^{-1}$. In particular, since we have a pole at $z=u\sfq_3^{-1}$, we can add a box there. The charge function then becomes
\bea
\widetilde{\Psi}^{\PT,\varnothing\varnothing\,\Bbox}_{\varnothing,u}(z)\sfg\left(\frac{u\sfq_{3}^{-1}}{z}\right)=\frac{(1-u\sfq_{1}\sfq_{3}^{-1}/z)(1-u\sfq_{2}\sfq_{3}^{-1}/z)(1-u/z)}{(1-\sfq_{3}^{-1}u/z)(1-\sfq_{3}^{-2}u/z)}.
\eea
Generally, we have
\bea
\widetilde{\Psi}^{\PT,\lambda\mu\nu}_{\varnothing,u}(z)\prod_{\scube\in\pi}\sfg\left(\frac{\chi_{u}(\cube)}{z}\right)
\eea
for $\pi \in \mathcal{PT}_{\lambda\mu\nu}$.

\end{remark}

\subsubsection{Three-legs PT counting}\label{sec:PT3-threelegs-gl1module}
Let us move on to the three-legs case. To make the discussion explicit, we focus only on the case $(\lambda,\mu,\nu)=(\Bbox,\Bbox,\Bbox)$. The vacuum charge function is
\bea
\Psi^{\PT,\,\Bbox\,\Bbox\,\Bbox}_{\varnothing,u}(z)&=\frac{(1-\sfq_{12}^{-1}u/z)(1-\sfq_{13}^{-1}u/z)(1-\sfq_{23}^{-1}u/z)}{(1-u/z)^{2}},
\eea
where we omitted the factors related with $K$. To deal with the second order poles, let us introduce formulas related with derivatives of the delta function.

\paragraph{Delta function and its derivatives }
\begin{definition}
    The $n$-th derivative of the delta function is defined as
    \bea
    \delta^{(n)}(x)=\left.\left(\frac{d}{dy}\right)^{n}\delta(y)\right|_{y=x}.
    \eea
\end{definition}
\begin{proposition}
In terms of series expansion, the $n$-th derivative of the delta function is
\bea
\delta^{(n)}(x)&=\sum_{k\geq 0}\frac{k!}{(k-n)!}x^{k-n}+\sum_{k>0}(-1)^{n}\frac{(k+n-1)!}{(k-1)!}x^{-k-n}\\
&=\sum_{k=0}^{\infty}\frac{(k+n)!}{k!}x^{k}+\sum_{k=0}^{\infty}(-1)^{n}\frac{(k+n)!}{k!}x^{-k-n-1}.
\eea
\end{proposition}
In particular, for the first derivative, we have
\bea
\delta^{(1)}(x)=\sum_{k\in\mathbb{Z}}kx^{k-1}=\sum_{k\geq0}(k+1)x^{k}+\sum_{k\geq 1}(-k)x^{-k-1}.
\eea

\begin{proposition}\label{prop:deltaderiv-reflection}
    The reflection property is 
    \bea
    \delta^{(n)}(x^{-1})=(-1)^{n}x^{n+1}\delta^{(n)}(x).
    \eea
    For the $n=1$ case, we have
    \bea
    x^{-1}\delta^{(1)}(x^{-1})=-x\delta^{(1)}(x).
    \eea
\end{proposition}

\begin{proposition}
    We have the following series expansion for rational functions with higher order poles:
    \bea
    \frac{1}{(1-x)^{n}}&=\sum_{k=0}^{\infty}\frac{(k+n-1)!}{(n-1)!k!}x^{k}=\sum_{k=0}^{\infty}\binom{k+n-1}{n-1}x^{k},\\
    \frac{(-1)^{n}x^{-n}}{(1-x^{-1})^{n}}&=(-1)^{n}x^{-n}\sum_{k=0}^{\infty}\binom{k+n-1}{n-1}x^{-k}=\sum_{k=0}^{\infty}(-1)^{n}\binom{k+n-1}{n-1}x^{-k-n},
    \eea
    where $\binom{n}{m}$ is the binomial coefficient.
\end{proposition}

Using the previous propositions, we obtain the following relation.
\begin{theorem}\label{thm:rat-delta-deriv}
    The rational functions and the delta function are related as
    \bea
     \left[\frac{1}{(1-x)^{n}}\right]_{|x|<1}-\left[\frac{1}{(1-x)^{n}}\right]_{|x|>1}&=\frac{1}{(n-1)!}\delta^{(n-1)}\left(x\right).
    \eea
    In particular, we have
    \bea
\left[\frac{1}{(1-x)^{2}}\right]_{|x|<1}-\left[\frac{1}{(1-x)^{2}}\right]_{|x|>1}&=\delta^{(1)}\left(x\right).
    \eea
\end{theorem}

Since the reflection property Prop.~\ref{prop:deltaderiv-reflection} involves factors of $x^{n+1}$ appearing in one side, it is convenient to introduce a dressed delta function.
\begin{definition}
    The dressed $n$-th derivative of the delta function is defined as
    \bea
    \tilde{\delta}^{(n)}(x)=x^{\frac{n+1}{2}}\delta^{(n)}(x).
    \eea
    In particular, for the $n=1$ case, we have
    \bea
    \tilde{\delta}^{(1)}(x)=x\delta^{(1)}(x).
    \eea
\end{definition}
Using this dressed version, the reflection property is concisely written as
\bea
\tilde{\delta}^{(n)}(x^{-1})=(-1)^{n}\tilde{\delta}^{(n)}(x).
\eea
Focusing on the $n=1$ case, the relation between $\tilde{\delta}^{(1)}(x)$ and $\delta^{(1)}(x)$ is
\bea\label{eq:logderivdelta-relation}
\tilde{\delta}^{(1)}(x)=x\delta^{(1)}(x)=x\partial_{x}\delta(x)=\partial_{x}(x\delta(x))-\delta(x)=\delta^{(1)}(x)-\delta(x).
\eea
This can be also confirmed as
\bea
\delta^{(1)}(x)-\delta(x)=\sum_{k\in\mathbb{Z}}kx^{k-1}-\sum_{k\in\mathbb{Z}}x^{k}=\sum_{k\in\mathbb{Z}}(k-1)x^{k-1}=\sum_{k\in\mathbb{Z}}kx^{k}=x\delta^{(1)}(x).
\eea

\begin{proposition}
We also have the following relation:
\bea
\delta^{(1)}\left(\frac{z}{u}\right)\delta\left(\frac{w}{z}\right)=\delta^{(1)}\left(\frac{z}{u}\right)\delta\left(\frac{w}{u}\right)+\left(\frac{w}{u}\right)\delta\left(\frac{z}{u}\right)\delta^{(1)}\left(\frac{w}{u}\right).
\eea
Denoting $\tilde{\delta}^{(1)}(x)=x\delta^{(1)}(x)$ and multiplying $z/u$ on both hand sides, we have
\bea
\tilde{\delta}^{(1)}\left(\frac{z}{u}\right)\delta\left(\frac{w}{z}\right)=\tilde{\delta}^{(1)}\left(\frac{z}{u}\right)\delta\left(\frac{w}{u}\right)+\delta\left(\frac{z}{u}\right)\tilde{\delta}^{(1)}\left(\frac{w}{u}\right)
\eea
\end{proposition}
\begin{proof}
    Performing the series expansion of the left hand side gives
    \bea
    \delta^{(1)}\left(\frac{z}{u}\right)\delta\left(\frac{w}{z}\right)&=\sum_{k\in\mathbb{Z}}k\left(\frac{z}{u}\right)^{k-1}\sum_{l\in\mathbb{Z}}\left(\frac{w}{z}\right)^{l}\\
    &=\sum_{k,l\in\mathbb{Z}}(k-l)\left(\frac{z}{u}\right)^{k-1-l}\left(\frac{w}{u}\right)^{l}+\sum_{k,l}\left(\frac{z}{u}\right)^{k-1-l}\left(\frac{w}{u}\right)^{l}\\
    &=\delta^{(1)}\left(\frac{z}{u}\right)\delta\left(\frac{w}{u}\right)+\left(\frac{w}{u}\right)\delta\left(\frac{z}{u}\right)\delta^{(1)}\left(\frac{w}{u}\right).
    \eea
\end{proof}

Similar to the property \eqref{eq:delta-rational}, the product of a rational function $f(x)$ and $\delta^{(1)}(x)$ has the following simplication.
\begin{proposition}\label{prop:derivedeltafunct-rational}
    Let $f(x)$ be a rational function, and we have
    \bea
    f(x)\delta^{(1)}(x)&=f(1)\delta^{(1)}(x)-f'(1)\delta(x),\\
    f(x)\tilde{\delta}^{(1)}(x)&=f(1)\tilde{\delta}^{(1)}(x)-f'(1)\delta(x).
    \eea
\end{proposition}
\begin{proof}
    Let us prove the first equation. The second one can be obtained by multiplying both hand sides of the first one by $x$. Using the definition of $\delta^{(1)}(x)$, we have
    \bea
    f(x)\delta^{(1)}(x)=f(x)\frac{d}{dx}\delta(x)=\frac{d}{dx}\left(f(x)\delta(x)\right)-f'(x)\delta(x)=f(1)\delta^{(1)}(x)-f'(1)\delta(x).
    \eea

We can also show it explicitly using the series expansion. Assume the Laurent expansion $f(x)=\sum_{n\in\mathbb{Z}}f_{n}x^{n}$. We then have
\bea
f(x)\delta^{(1)}(x)&=\sum_{n\in\mathbb{Z}}f_{n}x^{n}\sum_{m\in\mathbb{Z}}mx^{m-1}=\sum_{m,n\in\mathbb{Z}}f_{n}mx^{m-1+n}\\
&=\sum_{n,m\in\mathbb{Z}}f_{n}(m+1-n)x^{m}\\
&=\left(\sum_{n\in\mathbb{Z}}f_{n}\sum_{m\in\mathbb{Z}}mx^{m-1}\right)-\left(\sum_{n\in\mathbb{Z}}n f_{n}\sum_{m\in\mathbb{Z}}x^{m}\right)\\
&=f(1)\delta^{(1)}(x)-f'(1)\delta(x).
\eea
\end{proof}

\paragraph{Bootstrapping the module} Gaiotto--Rap\v{c}\'{a}k constructed the algebraic action of the Drinfeld currents of the affine Yangian $\mathfrak{gl}_{1}$ by assuming the GR box counting rules and studying the defining relations \cite{Gaiotto:2020dsq} (see also \cite{Gaberdiel:2017hcn,Gaberdiel:2018nbs} for discussion of anti-fundamental modules). We will not attempt to reproduce the whole discussion in the trigonometric language in this paper but give a rough sketch of it for the case $(\lambda,\mu,\nu)=(\Bbox,\Bbox,\Bbox)$. See Appendix~\ref{app-sec:A1collision-module-bootstrap} for a similar discussion on the quantum affine algebra $\mathcal{U}_{\sfq}(\widehat{\mathfrak{sl}}_{2})$.

From level 0 to level one, we have the transition rule
\bea
\ket{\varnothing}\rightarrow \ket{0_L}+\ket{0_H}.
\eea
The action of $K^{\pm}(z)$ on the vacuum is
\bea
K^{\pm}(z)\ket{\varnothing}=\left[\Psi^{\PT,\, \Bbox\,\Bbox\,\Bbox}_{\varnothing,u}(z)\right]_{\pm}\ket{\varnothing}.
\eea
Since the eigenvalue contains a second order pole at $z=u$, corresponding to the addable box, it is natural to expect that the action of $E(z)$ on the vacuum also contains a second order delta function. We thus first start from the following ansatz:
\bea
E(z)\ket{\varnothing}=E(\varnothing\rightarrow 0_{L})\tilde{\delta}^{(1)}\left(\frac{z}{u}\right)\ket{0_{L}}+E(\varnothing\rightarrow 0_{H})\delta\left(\frac{z}{u}\right)\ket{0_{H}}.
\eea
One can instead use $\tilde{\delta}^{(1)}\left(u/z\right)$ in the right hand side but it differs by replacing the coefficient with an extra sign due to the reflection property in Prop.~\ref{prop:deltaderiv-reflection}. This is also one reason we used $\tilde{\delta}^{(1)}(x)$ instead of $\delta^{(1)}(x)$.

Let us then use the $KE$ relation on the vacuum. We first have
\bea
K^{\pm}(z)E(w)\ket{\varnothing}&=E(\varnothing\rightarrow 0_{L})\tilde{\delta}^{(1)}\left(\frac{w}{u}\right)K^{\pm}(z)\ket{0_{L}}+E(\varnothing\rightarrow 0_{H})\delta\left(\frac{w}{u}\right)K^{\pm}(z)\ket{0_{H}}.
\eea
We also have
\bea
\sfg\left(\frac{w}{z}\right)^{-1}E(w)K^{\pm}(z)\ket{\varnothing}&=\sfg\left(\frac{w}{z}\right)^{-1}\Psi^{\PT,\, \Bbox\,\Bbox\,\Bbox}_{\varnothing,u}(z)\left(E(\varnothing\rightarrow 0_{L})\tilde{\delta}^{(1)}\left(\frac{w}{u}\right)\ket{0_{L}}+E(\varnothing\rightarrow 0_{H})\delta\left(\frac{w}{u}\right)\ket{0_{H}}\right)\\
&=E(\varnothing\rightarrow 0_{L})\sfg\left(\frac{u}{z}\right)^{-1}\Psi^{\PT,\, \Bbox\,\Bbox\,\Bbox}_{\varnothing,u}(z)\tilde{\delta}^{(1)}\left(\frac{w}{u}\right)\ket{0_{L}}\\
&-E(\varnothing\rightarrow 0_{L})\left.\partial_{\log w}\sfg\left(\frac{w}{z}\right)^{-1}\right|_{w=u}\Psi^{\PT,\, \Bbox\,\Bbox\,\Bbox}_{\varnothing,u}(z)\delta\left(\frac{w}{u}\right)\ket{0_{L}}\\
&+E(\varnothing\rightarrow 0_{H})\Psi^{\PT,\, \Bbox\,\Bbox\,\Bbox}_{\varnothing,u}(z)\sfg\left(\frac{u}{z}\right)^{-1}\delta\left(\frac{w}{u}\right)\ket{0_{H}}
\eea
where we used
\bea
\,&\sfg\left(\frac{w}{z}\right)^{-1}=\sfg\left(\frac{u}{z}\right)^{-1}+\left.\partial_{w}\sfg\left(\frac{w}{z}\right)^{-1}\right|_{w=u}(w-u)+\ldots,\\
&(w-u)\delta^{(1)}\left(\frac{w}{u}\right)=-u\delta\left(\frac{w}{u}\right).
\eea
Comparing both hand sides gives
\bea\label{eq:shiftgl1-levelone-Cartan}
K^{\pm}(z)\ket{0_{L}}&=\left[\Psi^{\PT,\, \Bbox\,\Bbox\,\Bbox}_{\varnothing,u}(z)\sfg\left(\frac{u}{z}\right)^{-1}\right]_{\pm}\ket{0_{L}},\\
K^{\pm}(z)\ket{0_{H}}&=\left[\Psi^{\PT,\, \Bbox\,\Bbox\,\Bbox}_{\varnothing,u}(z)\sfg\left(\frac{u}{z}\right)^{-1}\right]_{\pm}\ket{0_{H}}-\frac{E(\varnothing\rightarrow 0_{L})}{E(\varnothing\rightarrow 0_{H})}\left[\Psi^{\PT,\, \Bbox\,\Bbox\,\Bbox}_{\varnothing,u}(z)\tilde{\partial}\,\sfg\left(\frac{u}{z}\right)^{-1}\right]_{\pm}\ket{0_{L}},
\eea
where we denoted the derivatives with the log variables as
\bea
\tilde{\partial} f(x)=x\partial_{x}f(x)=\partial_{\log x}f(x).
\eea
Thus, the currents $K^{\pm}(z)$ no more commute with each other but rather form a Jordan block. 

The charge functions are computed as
\bea\label{eq:shiftgl1-level1-eigenvalue}
\Psi^{\PT,\, \Bbox\,\Bbox\,\Bbox}_{\varnothing,u}(z)\sfg\left(\frac{u}{z}\right)^{-1}&=\frac{(1-\sfq_{1}^{-1}u/z)(1-\sfq_{2}^{-1}u/z)(1-\sfq_{3}^{-1}u/z)}{(1-u/z)^{2}},\\
\Psi^{\PT,\, \Bbox\,\Bbox\,\Bbox}_{\varnothing,u}(z)\tilde{\partial}\sfg\left(\frac{u}{z}\right)^{-1}&=\frac{(1-\sfq_{1,2,3}^{-1}u/z)}{(1-u/z)^{2}}\sum_{c=1}^{3}\frac{\sfq_{c} u/z}{1-\sfq_{c} u/z}-\sum_{c=1}^{3}\frac{(1-\sfq_{c-1}^{-1}u/z)(1-\sfq_{c+1}^{-1}u/z)\sfq_{c}^{-1}u/z}{(1-u/z)^{2}}.
\eea
One of the poles at $z=u$ in \eqref{eq:shiftgl1-level1-eigenvalue} is understood as the pole for removing the box at the origin, while the other pole is understood as the pole for adding a box at the origin. Since it is a second-order pole, it is natural to impose the following ansatz
\bea
F(z)\ket{0_{L}}&=F(0_{L}\rightarrow \varnothing)\delta\left(\frac{z}{u}\right)\ket{\varnothing}\\
F(z)\ket{0_{H}}&=\left(F^{(2)}(0_{H}\rightarrow \varnothing)\tilde{\delta}^{(1)}\left(\frac{z}{u}\right)+F^{(1)}(0_{H}\rightarrow \varnothing)\delta\left(\frac{z}{u}\right)\right)\ket{\varnothing}.
\eea
The $KF$ and $EF$ relations will impose constraints on the coefficients.

Let us next consider the transition between level one and level two. Using \eqref{eq:shiftgl1-level1-eigenvalue}, for the heavy box, we have poles at $z=u,\sfq_{1,2,3}u$, while for the light box, we have a pole only at $z=u$. The poles imply that we can add a box at $z=u$ for both configurations and we can add boxes at $z=\sfq_{1,2,3}u$ only for the heavy box. Placing boxes at $z=u$ twice implies the ultra-heavy box and it is compatible with the JK-residue computation explained in \eqref{eq:3leg-level2-1}. 

For the transitions $\ket{0_{H,L}}\rightarrow \ket{0_U}$, the action of $E(z)$ are
\bea
E(z)\ket{0_L}&\sim E(0_L\rightarrow 0_U)\delta\left(\frac{z}{u}\right)\ket{0_U},\\
E(z)\ket{0_H}&\sim \left(E^{(2)}(0_H\rightarrow 0_U)\tilde{\delta}^{(1)}\left(\frac{z}{u}\right)+E^{(1)}(0_H\rightarrow 0_U)\delta\left(\frac{z}{u}\right)\right)\ket{0_U}
\eea
where $\sim $ means that we extracted the terms only related to the states written explicitly. The $KE$ relation on $\ket{0_L}$ gives 
\bea\label{eq:shiftgl1-ultraheavy-eigenvalue}
K^{\pm}(z)\ket{0_U}&=\Psi^{\PT,\, \Bbox\,\Bbox\,\Bbox}_{\varnothing,u}(z)\sfg\left(\frac{u}{z}\right)^{-2}\ket{0_U}\\
&=\frac{(1-\sfq_{1}^{-1}u/z)^{2}(1-\sfq_{2}^{-1}u/z)^{2}(1-\sfq_{3}^{-1}u/z)^{2}}{(1-u/z)^{2}(1-\sfq_{1}u/z)(1-\sfq_{2}u/z)(1-\sfq_{3}u/z)}\ket{0_U}.
\eea
The $KE$ relation on $\ket{0_H}$ gives a nontrivial constraint so that this equation is satisfied. After studying the $EE$ relations on $\ket{0_{L,H}}$, one will see that they are all compatible.

For the action on $F(z)$, we have
\bea
F(z)\ket{0_U}&\sim F(0_U\rightarrow 0_H)\delta\left(\frac{z}{u}\right)\ket{0_H}\\
&+\left(F^{(2)}(0_U\rightarrow 0_L)\tilde{\delta}^{(1)}\left(\frac{z}{u}\right)+F^{(1)}(0_U\rightarrow 0_L)\delta\left(\frac{z}{u}\right)\right)\ket{0_L}.
\eea
The assumption that we have the delta function comes from the fact that the pole $z=u$ at \eqref{eq:shiftgl1-ultraheavy-eigenvalue} corresponds to the removable boxes and it is a second-order pole. The $FF,FK,EF$ relations impose further constraints and one can determine the coefficients.

For the transition\footnote{Although we used the same notation as in the previous sections, note that the $\epsilon_{1,2,3}$ parameters obey $\eps_1+\eps_2+\eps_3=0$ in this section.} $\ket{0_H}\rightarrow \ket{\eps_1}+\ket{\eps_2}+\ket{\eps_3}$, the action of $E(z)$ is similar to the standard rules \eqref{eq:MacMahonrep}:
\bea
E(z)\ket{0_H}\sim E(0_H\rightarrow 0_H+\eps_3)\delta\left(\frac{z}{u\sfq_{3}}\right)\ket{0_H,\eps_3}
\eea
where we focused only on the box $\eps_3$ using the triality. The $KE$ relation gives
\bea\label{eq:shiftgl1-leveltwo-onelegextend-eigenvalue}
K^{\pm}(z)\ket{0_H,\eps_3}&=\Psi^{\PT,\, \Bbox\,\Bbox\,\Bbox}_{\varnothing,u}(z)\sfg\left(\frac{u}{z}\right)^{-1}\sfg\left(\frac{u\sfq_{3}}{z}\right)^{-1}\ket{0_H,\eps_3}\\
&=\frac{(1-\sfq_{12}u/z)(1-\sfq_{1}\sfq_{3}^{2}u/z)(1-\sfq_{2}\sfq_{3}^{2}u/z)}{(1-u/z)(1-\sfq_{3}^{2}u/z)}.
\eea
Note that we cannot add a box on top of the light box because the charge function $\Psi^{\PT,\, \Bbox\,\Bbox\,\Bbox}_{\varnothing,u}(z)\sfg\left(\frac{u}{z}\right)^{-1}$ does not contain a pole at $z=u\sfq_{3}$. 

On the other hand, for the action $F(z)$, we have
\bea
F(z)\ket{0_H,\eps_3}\sim F(0_H+\eps_3\rightarrow 0_L)\delta\left(\frac{z}{u\sfq_{3}}\right)\ket{0_L}.
\eea
Note that the right hand side needs to be $\ket{0_L}$ but not $\ket{0_H}$. This can be deduced from the $EF$ relation on $\ket{0_H}$. The contribution coming from the pole $u\sfq_{3}$ is
\bea
\relax [E(z),F(w)]\ket{0_H}&\sim -F(w)E(z)\ket{0_H}\\
&\sim F(0_H+\eps_3\rightarrow 0_L)E(0_H\rightarrow 0_H+\eps_3)\delta\left(\frac{z}{u\sfq_3}\right)\delta\left(\frac{w}{u\sfq_{3}}\right)\ket{0_L}
\eea
and
\bea
\delta\left(\frac{w}{z}\right)\left(K^{+}(z)-K^{-}(z)\right)\ket{0_H}&\sim-\frac{E(\varnothing\rightarrow 0_{L})}{E(\varnothing\rightarrow 0_{H})}\underset{z=u\sfq_{3}}{\Res}z^{-1}\Psi^{\PT,\, \Bbox\,\Bbox\,\Bbox}_{\varnothing,u}(z)\tilde{\partial}\,\sfg\left(\frac{u}{z}\right)^{-1}\delta\left(\frac{w}{z}\right)\delta\left(\frac{z}{u\sfq_{3}}\right)\ket{0_{L}}.
\eea
Since the pole at $z=u\sfq_{3}$ is only contained in the derivative term $\tilde{\partial} \sfg\left(u/z\right)^{-1}$, using \eqref{eq:shiftgl1-levelone-Cartan} and \eqref{eq:shiftgl1-level1-eigenvalue}, so that the $EF$ relation is compatible, the action of $F(z)$ needs to produce the light box but not the heavy box.

For higher levels, the action of the Drinfeld currents is easy to derive. For example, for the ultra-heavy box, the poles $z=u\sfq_{1,2,3}$ of \eqref{eq:shiftgl1-ultraheavy-eigenvalue} are the positions where one can add boxes. Since it is a simple pole, the action of $E(z)$ follows the standard rule \eqref{eq:MacMahonrep}. If we further add a box at $z=u\sfq_{3}$ on top of the ultra-heavy box, $K^{\pm}(z)$ acts diagonally and the charge function is
\bea
\Psi^{\PT,\, \Bbox\,\Bbox\,\Bbox}_{\varnothing,u}(z)\sfg\left(\frac{u}{z}\right)^{-2}\sfg\left(\frac{u\sfq_{3}}{z}\right)^{-1}=\frac{(1-\sfq_{12}u/z)^{2}(1-\sfq_{13}u/z)(1-\sfq_{23}u/z)(1-\sfq_{1,2}\sfq_{3}^{2}u/z))}{(1-u/z)(1-\sfq_{1}u/z)(1-\sfq_{2}u/z)(1-\sfq_{3}u/z)(1-\sfq_{3}^{2}u/z)}.
\eea
The pole at $z=u\sfq_{3}$ is the removable pole and the pole $z=u\sfq_{3}^{2}$ is an addable pole. All of the poles are single order and so the standard rule is applicable. Note that we still have addable poles at $z=\sfq_{1,2}u$ and we can add boxes there. This is compatible with the JK-residue computation where we can place boxes extending in all the three directions on the ultra-heavy box.

On the other hand, for the configuration $\ket{0_H,\eps_3}$, the eigenvalue \eqref{eq:shiftgl1-leveltwo-onelegextend-eigenvalue} implies that we can only add a box at $z=u\sfq_{3}$. Note that there are no poles at $z=u\sfq_{1,2}$ in this charge function. For the configuration with $k$-boxes extending in the 3-direction, the charge function is
\bea
\Psi^{\PT,\, \Bbox\,\Bbox\,\Bbox}_{\varnothing,u}(z)\prod_{i=1}^{k}\sfg\left(\frac{u\sfq_{3}^{i-1}}{z}\right)^{-1}
\eea
and it contains a pole at $z=u\sfq_{3}^{k}$, implying that we can add a box there.

Summarizing, for higher levels, the bases are
\bea
&\ket{0_{U},\epsilon_{1},\ldots,n_{1}\epsilon_{1},\epsilon_{2},\ldots,n_{2}\epsilon_{2},\epsilon_{3},\ldots,n_{3}\epsilon_{3}},\quad n_{1}+n_{2}+n_{3}=k-2 \,(n_{i}\geq 0),\\
&\ket{0_{H},\epsilon_{i},\ldots,(k-1)\epsilon_{i}},\quad i=1,2,3
\eea
and the action of $K^{\pm}(z)$ gives the charge functions
\bea\label{eq:PT3-three-leg-generic-eigenvalues}
&\Psi^{\PT,\, \Bbox\,\Bbox\,\Bbox}_{\varnothing,u}(z)\sfg\left(\frac{u}{z}\right)^{-2}\prod_{i=1}^{n_{1}}\sfg\left(\frac{u\sfq_{1}^{i-1}}{z}\right)^{-1}\prod_{j=1}^{n_{2}}\sfg\left(\frac{u\sfq_{2}^{j-1}}{z}\right)^{-1}\prod_{k=1}^{n_{3}}\sfg\left(\frac{u\sfq_{3}^{k-1}}{z}\right)^{-1},\\
&\Psi^{\PT,\, \Bbox\,\Bbox\,\Bbox}_{\varnothing,u}(z)\prod_{i=1}^{k-1}\sfg\left(\frac{u\sfq_{c}^{i-1}}{z}\right)^{-1}
\eea
respectively.


\subsubsection{PT3 $q$-characters}
The $qq$-characters have a one-to-one correspondence with the vertical representations of the (shifted) quantum toroidal $\mathfrak{gl}_{1}$. In particular, if we take the Nekrasov--Shatashvili (NS) limit of the $qq$-characters, we will get the $q$-characters \cite{Knight:1995JA,Frenkel:1998ojj} of the (shifted) quantum toroidal $\mathfrak{gl}_1$. From the quantum algebra viewpoint, the $q$-character is defined as the trace of $K^{\pm}(z)$ over the underlying module 
\bea\label{eq:q-character-def}
\Tr\left(\fq^{d}K^{\pm}(z)\right)
\eea
where $d$ is some degree counting operator.

To make the discussion concrete, let us first focus on the MacMahon representation. The degree counting operator simply counts the number of boxes and using \eqref{eq:CartanMacMahon} we obtain
\bea
T^{K}_{\bar{4}}(u)=\Tr_{M(u,K)}\left(\fq^{d}K^{\pm}(z)\right)=\sum_{\pi}\fq^{|\pi|}\Psi_{\pi,u}(z)=\sum_{\pi}\fq^{|\pi|}K^{-1/2}\frac{1-Ku/z}{1-u/z}\prod_{\scube\in\pi}\sfg\left(\frac{\chi_{u}(\cube)}{z}\right)^{-1}.
\eea
For later use, we remove the factors depending on $K$ and define the following 
\bea
T_{\bar{4}}(u)=\sum_{\pi}\sfq^{|\pi|}\frac{1}{1-u/z}\prod_{\scube\in \pi}\sfg\left(\frac{\chi_{u}(\cube)}{z}\right)^{-1}.
\eea
This will be the $q$-character of the shifted quantum toroidal $\mathfrak{gl}_1$.

To relate the $qq$-characters with the $q$-characters, we take the contraction with the screening current $\mathsf{S}_{4}(q_{4}z)$:
\bea
\mathcal{T}_{\bar{4}}(u)&=\bra{0}\mathsf{S}_{4}(q_{4}z)\mathsf{T}_{\bar{4}}(u)\ket{0}=(-q_{4}u)\sum_{\pi\in\mathcal{PP}}\fq^{|\pi|}\mathcal{Z}_{\bar{4}}^{\D6\tbar\QA}[\pi]\left[\frac{1}{(1-u/z)}\prod_{\scube \in\pi}g_{\bar{4}}\left(\frac{\chi_{\bar{4},u}(\cube)}{z}\right)^{-1}\right].
\eea
This is also called the $qq$-character in the context of quantum toroidal algebra. After taking the NS limit $q_{4}\rightarrow 1$, the coefficient becomes trivial $\mathcal{Z}^{\D6\tbar\QA}_{\bar{4}}[\pi]\rightarrow1$ and we obtain
\bea
\mathcal{T}_{\bar{4}}(u)\xrightarrow{q_{4}\rightarrow 1} -u T_{\bar{4}}(u)
\eea
where we used
\bea
\chi_{\bar{4},u}(\cube)\xrightarrow{q_{4}\rightarrow 1} \chi_{u}(\cube),\quad g_{\bar{4}}(x)\xrightarrow{q_{4}\rightarrow 1} \sfg (x).
\eea
Thus, the NS limit of the $qq$-character coincides with the $q$-character.

This discussion can be generalizable to the PT $qq$-characters but we need to be careful depending on the number of legs. 

\paragraph{Two-legs PT $q$-characters}
When at least one of the three-legs is trivial, the module is constructed in a similar way as the MacMahon representation. The Cartan eigenvalue for a PT configuration $\pi\in\widetilde{\mathcal{PT}}_{\lambda\mu\nu}$ is 
\bea
\Psi^{\PT,\,\lambda\mu\nu}_{\varnothing,u}(z)\prod_{\scube\in \pi}\sfg\left(\frac{\chi_{u}(\cube)}{z}\right)^{-1}
\eea
and so the $q$-character is determined by
\bea
PT_{\lambda\mu\nu}(u)=\sum_{\pi\in\widetilde{\mathcal{PT}}_{\lambda\mu\nu}}\sfq^{|\pi|}\Psi^{\PT,\,\lambda\mu\nu}_{\varnothing,u}(z)\prod_{\scube\in\pi}\sfg\left(\frac{\chi_{u}(\cube)}{z}\right)^{-1}.
\eea

One of the main claim of this section is that this $q$-character is obtained from the NS limit of the \textit{conjugate} PT $qq$-character. Following the description above, let us take the contraction with the screening current $\mathsf{S}_{4}(q_{4}z)$:
\bea
\mathcal{PT}_{\lambda\mu\nu}(u)&=\bra{0}\mathsf{S}_{4}(q_{4}z)\widetilde{\mathsf{PT}}_{\bar{4};\lambda\mu\nu}(u)\ket{0}=(-q_{4}u)\sum_{\pi\in\mathcal{PT}_{\lambda\mu\nu}}\fq^{|\pi|}\widetilde{\mathcal{Z}}^{\PT\tbar\QA}_{\bar{4};\lambda\mu\nu}[\pi]\widetilde{\mathscr{PW}}^{\bar{4},\lambda\mu\nu}_{\pi,u}(z)^{-1}.
\eea
Taking the NS-limit, we have
\bea
\widetilde{\mathcal{Z}}^{\PT\tbar\QA}_{\bar{4};\lambda\mu\nu}[\pi]\xrightarrow{q_{4}\rightarrow 1} 1,\quad \widetilde{\mathscr{PW}}^{\bar{4},\lambda\mu\nu}_{\pi,u}(z)^{-1}\xrightarrow{q_{4}\rightarrow 1}\Psi^{\PT,\,\lambda\mu\nu}_{\varnothing,u}(z)\prod_{\scube\in \pi}\sfg\left(\frac{\chi_{u}(\cube)}{z}\right)^{-1}
\eea
which gives
\bea
\mathcal{PT}_{\lambda\mu\nu}(u)\xrightarrow{q_{4}\rightarrow 1}  -u PT_{\lambda\mu\nu}(u).
\eea

\paragraph{Three-legs PT $q$-characters}
For the three-legs case, the discussion will be complicated because of the existence of second order poles. We only focus on the case $(\lambda,\mu,\nu)=(\Bbox,\Bbox,\Bbox)$, but the following discussion can be generalized.

Let us first start from the $q$-character of the module constructed in section~\ref{sec:PT3-threelegs-gl1module}. We still can use the definition given in \eqref{eq:q-character-def}. For level one, we have two states $\ket{0_{L,H}}$ whose action of $K^{\pm}(z)$ is \eqref{eq:shiftgl1-levelone-Cartan}. The action of $K^{\pm}(z)$ may be non-diagonal but after taking the trace, such non-diagonal term will not appear and we obtain
\bea
2\times \Psi^{\PT,\, \Bbox\,\Bbox\,\Bbox}_{\varnothing,u}(z)\sfg\left(\frac{u}{z}\right)^{-1}.
\eea
The multiplicity factor here comes from the two-states degeneracy of the light and heavy boxes and the fact that the diagonal part of the Jordan block is the same.

For level two, the contributions come from the ultra-heavy box $\ket{0_U}$ and $\ket{0_H,\eps_{1,2,3}}$ and thus, we have
\bea
\Psi^{\PT,\, \Bbox\,\Bbox\,\Bbox}_{\varnothing,u}(z)\sfg\left(\frac{u}{z}\right)^{-2}+\sum_{i=1}^{3}\Psi^{\PT,\, \Bbox\,\Bbox\,\Bbox}_{\varnothing,u}(z)\sfg\left(\frac{u}{z}\right)^{-1}\sfg\left(\frac{u\sfq_{i}}{z}\right)^{-1}.
\eea

For higher levels, since the Drinfeld current $K^{\pm}(z)$ acts diagonally, we simply need to take the sum over the eigenvalues of the bases using \eqref{eq:PT3-three-leg-generic-eigenvalues}. 

The $q$-character is then given as
\bea
PT_{\Bbox\,\Bbox\,\Bbox}(u)&=\Psi^{\PT,\, \Bbox\,\Bbox\,\Bbox}_{\varnothing,u}(z)+2\fq\Psi^{\PT,\, \Bbox\,\Bbox\,\Bbox}_{\varnothing,u}(z)\sfg\left(\frac{u}{z}\right)^{-1}\\
&+\fq^{2}\left(\Psi^{\PT,\, \Bbox\,\Bbox\,\Bbox}_{\varnothing,u}(z)\sfg\left(\frac{u}{z}\right)^{-2}+\sum_{i=1}^{3}\Psi^{\PT,\, \Bbox\,\Bbox\,\Bbox}_{\varnothing,u}(z)\sfg\left(\frac{u}{z}\right)^{-1}\sfg\left(\frac{u\sfq_{i}}{z}\right)^{-1}\right)+\cdots
\eea
where we explicitly wrote only up to level two. Strictly speaking, one needs to determine how the degree counting operator acts on the bases of the PT module. For example, whether the ultra-heavy box is counted twice is a fact that needs to be confirmed. We leave all of this for future work.

Let us next move on to the three-legs PT $qq$-character. Again, let us take the contraction of the conjugate PT $qq$-character with the screening current:
\bea
\mathcal{PT}_{\Bbox\,\Bbox\,\Bbox}(u)&=\bra{0}\mathsf{S}_{4}(q_{4}z)\widetilde{\mathsf{PT}}_{\bar{4};\Bbox\,\Bbox\,\Bbox}(u)\ket{0}\\
&=\bra{0}\mathsf{S}_{4}(q_{4}z)\widetilde{\mathsf{H}}_{\bar{4};\Bbox\,\Bbox\,\Bbox}(u)\ket{0}+\fq \left(\widetilde{\mathcal{Z}}^{\PT\tbar\QA}_{\bar{4};\,\Bbox\,\Bbox\,\Bbox}[1]\bra{0}\mathsf{S}_{4}(q_{4}z):\widetilde{\mathsf{H}}_{\bar{4};\Bbox\,\Bbox\,\Bbox}(u)\mathsf{A}^{-1}(u):\ket{0}\right.\\
&\left.+(1-q_{4}^{-1})\bra{0}\mathsf{S}_{4}(q_{4}z):\widetilde{\mathsf{H}}_{\bar{4};\Bbox\,\Bbox\,\Bbox}(u)\partial_{\log u}\mathsf{A}^{-1}(u):\ket{0}\right)\\
&+\fq^{2}\left( \frac{1}{2}\widetilde{\mathcal{Z}}^{\QA}_{(0,0)}\bra{0}\mathsf{S}_{4}(q_{4}z):\widetilde{\mathsf{H}}_{\bar{4};\Bbox\,\Bbox\,\Bbox}(u)\mathsf{A}^{-2}(u):\ket{0}\right.\\
&\left.+\sum_{i=1}^{3}\widetilde{\mathcal{Z}}^{\QA}_{(0,\eps_i)}\bra{0}\mathsf{S}_{4}(q_{4}z):\widetilde{\mathsf{H}}_{\bar{4};\Bbox\,\Bbox\,\Bbox}(u)\mathsf{A}^{-1}(u)\mathsf{A}^{-1}(uq_{i}):\ket{0}\right)+\cdots 
\eea
At the NS-limit (the unrefined limit in section~\ref{sec:unrefined-vertex}) the coefficients become
\bea
\widetilde{\mathcal{Z}}^{\PT\tbar\QA}_{\bar{4};\,\Bbox\,\Bbox\,\Bbox}[1]\rightarrow 2,\quad \widetilde{\mathcal{Z}}^{\QA}_{(0,0)}\rightarrow 2,\quad \widetilde{\mathcal{Z}}^{\QA}_{(0,\eps_i)}\rightarrow 1.
\eea
An interesting property is that the coefficient containing derivatives of the vertex operators vanishes $1-q_{4}^{-1}\rightarrow 0$. Moreover, using
\bea
\bra{0}\mathsf{S}_{4}(q_{4}z)\widetilde{\mathsf{H}}_{\bar{4};\Bbox\,\Bbox\,\Bbox}(u)\ket{0}&\xrightarrow{q_{4}\rightarrow 1} -u\Psi^{\PT,\, \Bbox\,\Bbox\,\Bbox}_{\varnothing,u}(z),\\
\bra{0}\mathsf{S}_{4}(q_{4}z):\widetilde{\mathsf{H}}_{\bar{4};\Bbox\,\Bbox\,\Bbox}(u)\mathsf{A}^{-1}(u):\ket{0}&\xrightarrow{q_{4}\rightarrow 1}-u\Psi^{\PT,\, \Bbox\,\Bbox\,\Bbox}_{\varnothing,u}(z)\sfg\left(\frac{u}{z}\right)^{-1},\\
\bra{0}\mathsf{S}_{4}(q_{4}z):\widetilde{\mathsf{H}}_{\bar{4};\Bbox\,\Bbox\,\Bbox}(u)\mathsf{A}^{-2}(u):\ket{0}&\xrightarrow{q_{4}\rightarrow 1}-u\Psi^{\PT,\, \Bbox\,\Bbox\,\Bbox}_{\varnothing,u}(z)\sfg\left(\frac{u}{z}\right)^{-2},\\
\bra{0}\mathsf{S}_{4}(q_{4}z):\widetilde{\mathsf{H}}_{\bar{4};\Bbox\,\Bbox\,\Bbox}(u)\mathsf{A}^{-1}(u)\mathsf{A}^{-1}(uq_{i}):\ket{0}&\xrightarrow{q_{4}\rightarrow 1}-u\Psi^{\PT,\, \Bbox\,\Bbox\,\Bbox}_{\varnothing,u}(z)\sfg\left(\frac{u}{z}\right)^{-1}\sfg\left(\frac{u\sfq_{i}}{z}\right)^{-1}
\eea
we obtain
\bea
\mathcal{PT}_{\Bbox\,\Bbox\,\Bbox}(u)\rightarrow -u PT_{\Bbox\,\Bbox\,\Bbox}(u).
\eea
The nontrivial part is that after taking the NS-limit, the PT3 partition function (including the Weyl group factor) gives the multiplicity of the eigenspace and that the non-diagonal terms all vanish as expected. The computation for higher levels is similar and one can show exactly that the NS-limit of the $qq$-character indeed becomes the $q$-character of the PT3 module.

\section{Conclusion and discussion}
In the first half of this paper, we showed that even if we start from the same contour integrand, choosing the reference vector $\eta=\eta_0$ gives the DT-side, while choosing $\eta=\tilde{\eta}_0$ gives the PT-side. We also discussed the relation of the PT3 vertices with the known topological vertices and showed that under suitable limits, we can obtain all of them. Moreover, we also discussed natural generalizations of DT, PT counting and propose new DT/PT correspondences.

In the second half of this paper, by following \cite{Kimura:2023bxy,Kimura:2024osv}, we introduced a new $qq$-character, which we call the Pandhariphande--Thomas (PT) $qq$-characters. The monomial terms of them are the PT configurations plus some non-diagonal terms coming from the existence of higher order poles. Since the information of the PT3 partition function is included in the coefficients of the vertex operators, we managed to establish the BPS/CFT correspondence even for the PT-side. We also studied the algebraic structures of the PT3 $qq$-characters.

Let us list some possible directions we hope to address in near future.

\paragraph{Wall-crossing and $qq$-characters}
The fact that changing the reference vector gives both the DT-side and the PT-side is a consequence of the wall-crossing phenomenon. Studying this phenomenon from the quantum mechanics view-point following \cite{Hori:2014tda} would be interesting. In section~\ref{sec:DTPTcorrespondence}, we gave a brief discussion of it but giving a full proof from this viewpoint might be interesting.

Understanding the wall-crossing and the DT/PT correspondence at the level of $q$-characters or $qq$-characters is also interesting. One of the motivation of this work was to understand such a phenomenon. Although we managed to construct both the DT and PT $qq$-characters independently, we do not know how to directly relate them to each other in the sense of DT/PT correspondence for the moment. One strategy to accomplish this is to study the DT/PT correspondence at the level of contour integrals. For example, we can perform a similar discussion of the wall-crossing in section~\ref{sec:DTPTcorrespondence} for the $qq$-characters. For the one-leg one box case, consider the vertex operator valued contour integral
\bea
\mathcal{G}\times \oint_{|z|=1} \frac{dz}{2\pi i z}\mathcal{Z}^{\D6_{\bar{4}}\tbar\D2\tbar\D0}_{
\DT\tbar\QA;\varnothing\varnothing\,\Bbox}(v,zv):\frac{\mathsf{W}_{\bar{4}}(v)}{\mathsf{S}_{3}(v)}\mathsf{A}^{-1}(zv):,
\eea
where the contour is taken to be the unit circle traversed counterclockwise. Since the vertex operator part is holomorphic, the poles only come from the coefficient part. Assuming $|q_{1,2,3}|<1$, we then can perform the residue inside the region $|z|<1$, which becomes
\bea
&\mathcal{G}\times \left(\underset{z=0}{\Res}z^{-1} \mathcal{Z}^{\D6_{\bar{4}}\tbar\D2\tbar\D0}_{
\DT\tbar\QA;\varnothing\varnothing\,\Bbox}(v,zv):\frac{\mathsf{W}_{\bar{4}}(v)}{\mathsf{S}_{3}(v)}\mathsf{A}^{-1}(zv): \right.\\
&\left.+ \underset{z=q_{1}}{\Res}z^{-1} \mathcal{Z}^{\D6_{\bar{4}}\tbar\D2\tbar\D0}_{
\DT\tbar\QA;\varnothing\varnothing\,\Bbox}(v,zv):\frac{\mathsf{W}_{\bar{4}}(v)}{\mathsf{S}_{3}(v)}\mathsf{A}^{-1}(zv):+\underset{z=q_{2}}{\Res}z^{-1} \mathcal{Z}^{\D6_{\bar{4}}\tbar\D2\tbar\D0}_{
\DT\tbar\QA;\varnothing\varnothing\,\Bbox}(v,zv):\frac{\mathsf{W}_{\bar{4}}(v)}{\mathsf{S}_{3}(v)}\mathsf{A}^{-1}(zv):\right)\\
=&\mathcal{G}q_{4}\lim_{z\rightarrow 0}:\frac{\mathsf{W}_{\bar{4}}(v)}{\mathsf{S}_{3}(v)}\mathsf{A}^{-1}(zv):+\mathcal{Z}^{\DT\tbar\QA}_{\bar{4};\varnothing\varnothing\,\Bbox}[(q_{1})]:\frac{\mathsf{W}_{\bar{4}}(v)}{\mathsf{S}_{3}(v)}\mathsf{A}^{-1}(q_{1}v):+\mathcal{Z}^{\DT\tbar\QA}_{\bar{4};\varnothing\varnothing\,\Bbox}[(q_{2})]:\frac{\mathsf{W}_{\bar{4}}(v)}{\mathsf{S}_{3}(v)}\mathsf{A}^{-1}(q_{1}v):.
\eea
The second and third term is the DT3 $qq$-character of level one.

Instead, we can also take the residue at the region $|z|>1$:
\bea
&-\mathcal{G}\times \left(\underset{z=\infty}{\Res}z^{-1} \mathcal{Z}^{\D6_{\bar{4}}\tbar\D2\tbar\D0}_{
\DT\tbar\QA;\varnothing\varnothing\,\Bbox}(v,zv):\frac{\mathsf{W}_{\bar{4}}(v)}{\mathsf{S}_{3}(v)}\mathsf{A}^{-1}(zv): + \underset{z=q_{3}^{-1}}{\Res}z^{-1} \mathcal{Z}^{\D6_{\bar{4}}\tbar\D2\tbar\D0}_{
\DT\tbar\QA;\varnothing\varnothing\,\Bbox}(v,zv):\frac{\mathsf{W}_{\bar{4}}(v)}{\mathsf{S}_{3}(v)}\mathsf{A}^{-1}(zv):\right)\\
=&\,\,\mathcal{G}\lim_{z\rightarrow \infty}:\frac{\mathsf{W}_{\bar{4}}(v)}{\mathsf{S}_{3}(v)}\mathsf{A}^{-1}(zv):+\mathcal{Z}^{\PT\tbar\QA}_{\bar{4};\varnothing\varnothing\,\Bbox}[(q_{3}^{-1})]:\frac{\mathsf{W}_{\bar{4}}(v)}{\mathsf{S}_{3}(v)}\mathsf{A}^{-1}(q_{3}^{-1}v):
\eea
where the second term is the one-instanton level of the PT3 $qq$-character.

Since the contour integral should be the same, we have
\bea
\mathsf{DT}_{\bar{4};\varnothing\varnothing\,\Bbox}(v)[1]-\mathsf{PT}_{\bar{4};\varnothing\varnothing\,\Bbox}(v)[1]=\mathcal{G}\lim_{z\rightarrow \infty}:\frac{\mathsf{W}_{\bar{4}}(v)}{\mathsf{S}_{3}(v)}\mathsf{A}^{-1}(zv):-\mathcal{G}q_{4}\lim_{z\rightarrow 0}:\frac{\mathsf{W}_{\bar{4}}(v)}{\mathsf{S}_{3}(v)}\mathsf{A}^{-1}(zv):
\eea
where we shortly denoted the one-instanton part of the DT, PT $qq$-characters as $\mathsf{DT}_{\bar{4};\varnothing\varnothing\,\Bbox}(v)[1],\mathsf{PT}_{\bar{4};\varnothing\varnothing\,\Bbox}(v)[1]$. How to relate the right hand side with the level one part of the D6 $qq$-character
\bea
(1-q_{4})\mathcal{G}:\mathsf{W}_{\bar{4}}(v)\mathsf{A}^{-1}(v):=-q_{4}\frac{(1-q_{12})(1-q_{23})(1-q_{31})}{(1-q_{1})(1-q_{2})(1-q_{3})}:\mathsf{W}_{\bar{4}}(v)\mathsf{A}^{-1}(v):
\eea
will be the nontrivial part.


\paragraph{PT3 counting and 3d gauge theories}
In a recent paper \cite{Crew:2020psc}, the one-leg PT3 vertex was identified with the vortex partition function of a 3d $\mathcal{N}=4$ ADHM theory with an appropriate vacuum at the infinity. The DT/PT correspondence implies a wall-crossing phenomenon in the vortex partition function. Understanding it from the 3d view point might lead to new dualities.

Since the 3d $\mathcal{N}=4$ ADHM theory appears from the world-volume theory of multiple D2-branes parallel to a D6-brane, it is natural to expect that the full equivariant PT3 vertex appears from intersecting D2-branes inside a D6-brane. More generally, we can include multiple D6-branes and consider a vortex origami system. The vortex partition functions and relations with the PT3 vertex will be discussed in a future work. The higher rank generalizations of the DT/PT correspondence implies a wall-crossing phenomenon of the vortex origami. Studying it is also interesting.

\paragraph{PT4 counting and JK-residue}
A straightforward generalization is a study of the relation between PT4 counting and JK-residue. For this case, solid partitions with nontrivial boundary conditions appear. In particular, we have leg boundary conditions and surface boundary conditions. The leg boundary condition is similar to the PT3 case, but the surface boundary condition gives new type of PT vertices \cite{Bae:2022pif,Bae:2024bpx}. Derivation of the PT4 vertices from the JK-residue method will be discussed in \cite{Kimura-Noshita-PT4}.


\acknowledgments
The authors thank Jiaqun Jiang, Satoshi Nawata, Jiahao Zheng for discussions on the JK-residue formalism and lending the authors a numerical program that helped the authors learn how to perform the JK-residue computation. We also thank Henry Liu for a wonderful \href{https://github.com/liu-henry-hl/boxcounting}{box counting} computer program, where some consistency check was done.
The work of TK was supported by CNRS through MITI interdisciplinary programs, EIPHI Graduate School (No.~ANR-17-EURE-0002) and Bourgogne-Franche-Comté region. GN is supported by JSPS KAKENHI Grant-in-Aid for JSPS fellows Grant No.~JP25KJ0124.

\appendix

\section{Refined and unrefined topological vertices}\label{app:top-vertex-symm-funct}
In this section, we review basic formulas for the refined and unrefined topological vertices. We use the Iqbal-Kozcaz-Vafa (IKV) \cite{Iqbal:2007ii} form, which use the Schur functions, instead of the Awata-Kanno form \cite{Awata:2005fa,Awata:2008ed} using Macdonald polynomials. We note that both forms are known to be essentially equivalent \cite{Awata:2008ed,Awata:2011ce}.

\subsection{Symmetric Functions}
The Macdonald symmetric functions are related as
\bea
Q_{\lambda}(x;q,t)=b_{\lambda}(q,t)P_{\lambda}(x;q,t),\quad  b_{\lambda}(q,t)=\langle P_{\lambda},P_{\lambda}\rangle^{-1},\quad b_{\lambda}(q,t)=\prod_{\Abox\in\lambda}\frac{1-q^{a_{\lambda}(\Abox)}t^{\ell_{\lambda}(\Abox)+1}}{1-q^{a_{\lambda}(\Abox)+1}t^{\ell_{\lambda}(\Abox)}}
\eea

For the unrefined limit $q=t$, we have
\bea
P_{\lambda}(x;t,t)=Q_{\lambda}(x;t,t)=s_{\lambda}(x).
\eea

For the skew Macdonald symmetric functions, they are defined as
\bea
Q_{\lambda/\mu}(x;q,t)&=\sum_{\nu}f^{\lambda}_{\mu\nu}(q,t)Q_{\nu}(x;q,t)\\
P_{\lambda/\mu}(x;q,t)&=\sum_{\nu}f^{\lambda^{\rmT}}_{\mu^{\rmT}\nu^{\rmT}}(q,t) P_{\nu}(x;q,t)
\eea
where we used
\bea
P_{\mu}(q,t)P_{\nu}(q,t)=\sum_{\lambda}f^{\lambda}_{\mu\nu}(q,t)P_{\lambda}(q,t),\quad f^{\lambda}_{\mu\nu}(q,t)=\langle Q_{\mu} ,P_{\mu}P_{\nu}\rangle\\
Q_{\mu}(q,t)Q_{\nu}(q,t)=\sum_{\lambda}f^{\lambda^{\rmT}}_{\mu^{\rmT}\nu^{\rmT}}(t,q)Q_{\lambda}(q,t),\quad f^{\lambda^{\rmT}}_{\mu^{\rmT}\nu^{\rmT}}(q,t)=\langle P_{\mu} ,Q_{\mu}Q_{\nu}\rangle
\eea
The relation between these skew Macdonald symmetric functions is
\bea
Q_{\lambda/\mu}=b_{\lambda}b_{\mu}^{-1}P_{\lambda/\mu}.
\eea
In the limit $q\rightarrow t$, the skew Macdonald functions become the skew Schur functions:
\bea
Q_{\lambda/\mu}(x;t,t)=P_{\lambda/\mu}(x;t,t)=s_{\lambda/\mu}(x).
\eea

\subsection{Formulas}
\paragraph{Young diagrams and formulas}
The Young diagram (partition) is a sequence of non-increasing non-negative integers 
\begin{align}
    \lambda=\{\lambda_{i}\in\mathbb{Z}_{\geq 0}\,|\,\lambda_{1}\geq \lambda_{2}\geq \cdots\}.
\end{align}
We denote the transpose of $\lambda$ as $\lambda^{\rmT}$. The size and norm are defined as 
\begin{align}
    |\lambda|=\sum_{i=1}^{\ell(\lambda)}\lambda_{i},\quad ||\lambda||^{2}=\sum_{i=1}^{\ell(\lambda)}\lambda_{i}^{2},
\end{align}
where $\ell(\lambda)$ is the length of the Young diagram. The arm length, leg length, arm co-length, leg co-length, and content ($c_{\lambda}(i,j)$)) for $(i,j)\in\lambda$ are defined as 
\begin{align}
    a_{\lambda}(i,j)=\lambda_{i}-j,&\quad \ell_{\lambda}(i,j)=\lambda_{j}^{\rmT}-i,\\
    a'_{\lambda}(i,j)=j-1,&\quad \ell'_{\lambda}(i,j)=i-1,\\
    c_{\lambda}(i,j)=j-i.&
\end{align}
The hook length is defined as 
\begin{align}
    h_{\lambda}(i,j)=\lambda_{i}-j+\lambda_{j}^{\rmT}-i+1=a_{\lambda}(i,j)+\ell_{\lambda}(i,j)+1.
\end{align}
We also have the following quantities
\begin{align}
    n(\lambda)=\sum_{i=1}^{\ell(\lambda)}(i-1)\lambda_{i},&\quad \kappa(\lambda)=2\sum_{x\in\lambda}c(x)=2\sum_{(i,j)\in\lambda}(j-i),\quad a'_{\lambda}(x)-\ell'_{\lambda}(x)=c_{\lambda}(x)
\end{align}
satisfying the following properties.
\begin{align}
    \kappa(\lambda)&=\sum_{i}\lambda_{i}(\lambda_{i}+1-2i)=2(n(\lambda^{\rmT})-n(\lambda))=||\lambda||^{2}-||\lambda^{\rmT}||^{2},\\
    n(\lambda)&=\frac{1}{2}\sum_{j=1}^{\lambda_{1}}\lambda_{j}^{\rmT}(\lambda_{j}^{\rmT}-1)=\sum_{x\in\lambda}\ell'_{\lambda}(x)=\sum_{x\in\lambda}\ell_{\lambda}(x),\\
    n(\lambda^{\rmT})&=\frac{1}{2}\sum_{i=1}^{\ell(\lambda)}\lambda_{i}(\lambda_{i}-1)=\sum_{x\in\lambda}a'_{\lambda}(x)=\sum_{x\in\lambda}a_{\lambda}(x),\\
    \sum_{x\in\lambda}h_{\lambda}(x)&=n(\lambda)+n(\lambda^{\rmT})+|\lambda|=\frac{1}{2}(||\lambda||^{2}+||\lambda^{\rmT}||^{2}).
\end{align}

\paragraph{Generating function of plane partitions}
The generating function for the number of generalized skew plane partitions of shape $(\lambda,\mu,\nu)$ is given by
\bea
Z_{\lambda\mu\nu}(q)=\sum_{\pi\in \mathcal{PP}_{\lambda\mu\nu}}q^{|\pi|}.
\eea
The cyclic symmetry is understood as
\bea
Z_{\lambda\mu\nu}(q)=Z_{\mu\nu\lambda}(q)=Z_{\nu\lambda\mu}(q).
\eea
The normalized unrefined topological vertex also obeys this symmetry.

\paragraph{Topological vertex}
The topological vertex is defined as
\bea
C_{\lambda\mu\nu}(q)=q^{\frac{\kappa(\mu)}{2}}\sum_{\eta}s_{\lambda^{\rmT}/\eta}(q^{-\nu-\rho})s_{\mu/\eta}(q^{-\nu^{\rmT}-\rho})
\eea

\paragraph{Refined topological vertex}
The refined topological vertex is defined as
\bea
    C_{\lambda\mu\nu}(t,q)&=\left(\frac{q}{t}\right)^{\frac{||\mu||^{2}+||\nu||^{2}}{2}}t^{\frac{\kappa(\mu)}{2}}t^{\frac{||\nu||^{2}}{2}}\widetilde{Z}_{\nu}(t,q)\sum_{\eta}\left(\frac{q}{t}\right)^{\frac{|\eta|+|\lambda|-|\mu|}{2}}s_{\lambda^{\rmT}/\eta}(t^{-\rho}q^{-\nu})s_{\mu/\eta}(t^{-\nu^{\rmT}}q^{-\rho})\\
    &=q^{\frac{||\mu||^{2}+||\nu||^{2}}{2}}t^{-\frac{||\mu^{\rmT}||^{2}}{2}}\widetilde{Z}_{\nu}(t,q)\sum_{\eta}\left(\frac{q}{t}\right)^{\frac{|\eta|+|\lambda|-|\mu|}{2}}s_{\lambda^{\rmT}/\eta}(t^{-\rho}q^{-\nu})s_{\mu/\eta}(t^{-\nu^{\rmT}}q^{-\rho})
\eea
where 
\bea
     \rho=(-\frac{1}{2},-\frac{3}{2},-\frac{5}{2},\cdots),\quad
    \widetilde{Z}_{\nu}(t,q)&=\prod_{(i,j)\in\nu}(1-t^{\ell_{\nu}(i,j)+1}q^{a_{\nu}(i,j)})^{-1},
\eea
and $s_{\lambda/\eta}(x),\,x=(x_{1},x_{2},\ldots)$ is the skew Schur function.

The unrefined topological vertex is obtained by specializing the parameters $C_{\lambda\mu\nu}(q)=C_{\lambda\mu\nu}(q,q).$

For later use, we also introduce the normalized refined topological vertex, which we denote as $\widetilde{C}_{\lambda\mu\nu}(t,q)$ which after series expansion in $q,t$, the first term will be always 1:
\bea
\widetilde{C}_{\lambda\mu\nu}(t,q)=1+\cdots.
\eea

\paragraph{Macdonald refined topological vertex}
The Macdonald vertex is defined as
\bea
\mathcal{M}_{\lambda\mu\nu}(x,y;q,t)=\prod_{n=0}^{\infty}\left(\prod_{\Abox\in\nu}\frac{1-tq^{n}x^{\ell_{\nu}(\Abox)+1}y^{a_{\nu}(\Abox)}}{1-q^{n}x^{\ell_{\nu}(\Abox)+1}y^{a_{\nu}(\Abox)}}\right)\sum_{\eta}P_{\mu/\eta}(y^{-\rho-1/2}x^{-\nu^{\rmT}};q,t)Q_{\lambda^{\rmT}/\eta}(x^{-\rho+1/2}y^{-\nu};q,t)
\eea
where $P_{\lambda}(x;q,t),Q_{\mu}(x;q,t)$ are skew Macdonald symmetric functions. The normalized Macdonald vertex is denoted as
\bea
\widetilde{\mathcal{M}}_{\lambda\mu\nu}(x,y;q,t)=1+\cdots
\eea
after expanding in series of $x,y$. Note that the $x,y$ are the $t,q$ in the refined topological vertex and $q,t$ are extra parameters.

In the limit $t\rightarrow q$, the Macdonald refined vertex becomes
\bea
\mathcal{M}_{\lambda\mu\nu}(x,y;t,t)&=\prod_{\Abox\in\nu}\frac{1}{(1-x^{\ell_{\nu}(\Abox)+1)}y^{a_{\nu}(\Abox)})}\sum_{\eta}s_{\mu/\eta}(y^{-\rho-1/2}x^{-\nu^{\rmT}})s_{\lambda^{\rmT}/\eta}(x^{-\rho+1/2}y^{-\nu})\\
&=\left(y^{-\frac{|\lambda|}{2}}x^{\frac{|\mu|}{2}}\right)\prod_{\Abox\in\nu}\frac{1}{(1-x^{\ell_{\nu}(\Abox)+1)}y^{a_{\nu}(\Abox)})}\sum_{\eta}\left(\frac{y}{x}\right)^{\frac{|\eta|}{2}}s_{\mu/\eta}(y^{-\rho}x^{-\nu^{\rmT}})s_{\lambda^{\rmT}/\eta}(x^{-\rho}y^{-\nu})\\
&\propto C_{\lambda\mu\nu}(x,y)
\eea

\section{Examples: PT3 partition functions}\label{app:sec-PT3vertex-examples}
In this section, we summarize the results for PT3 counting for the one-leg, two-legs, three-legs cases respectively. For the one-leg and two-legs case, the PT counting rules, GR counting rules, and the JK-residue all coincide so we only list down the results of the JK-residue formalism. 

Additionally, the unrefined topological vertex, refined topological vertex, and the Macdonald refined topological vertex for each configuration are listed. One can also confirm that the PT3 vertex matches with these topological vertices after taking the proper limits respectively (see Thm.~\ref{thm:PT3vertex-unref-corresp}, \ref{thm:PT3vertex-ref-corresp}, \ref{thm:PT3vertex-Macdref-corresp}).
\Yboxdim{4pt}

\subsection{One-leg PT3 counting}
The results for the one-leg PT3 counting for the boundary condition $(\varnothing,\varnothing,\nu)$ with $|\nu|\leq 4$ are summarized.

\subsubsection{$\nu=\{2\}$ (Fig.~\ref{fig:PT3-1leg-ex1})}
\begin{figure}[ht]
    \centering
    \includegraphics[width=1.2cm]{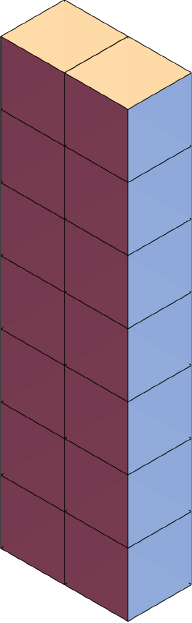}
    \caption{The PT box counting setup for $(\lambda,\mu,\nu)=(\varnothing,\varnothing,\{2\})$}
    \label{fig:PT3-1leg-ex1}
\end{figure}

\paragraph{PT configurations}
\bea
\begin{tabular}{|c|c|}\hline
   \text{instanton number $k$}  &  \text{poles} \\\hline\hline
    1 & $(-\epsilon_{3}+\epsilon_{2})$   \\\hline
    2 &  $(-\epsilon_{3}+\epsilon_{2},-\epsilon_{3}),(-\epsilon_{3}+\epsilon_{2},-2\epsilon_{3}+\epsilon_{2}) $ \\\hline
    3 & $\begin{array}{c}(-\epsilon_{3}+\epsilon_{2},-\epsilon_{3},\epsilon_{2}-2\epsilon_{3})\\
     (\epsilon_{2}-\epsilon_{3},\epsilon_{2}-2\epsilon_{3},\epsilon_{2}-3\epsilon_{3})\end{array}$\\\hline
\end{tabular}
\eea

\paragraph{JK residues}
The framing node contribution is
\bea
\mathcal{Z}^{\D6_{\bar{4}}\tbar\D2\tbar\D0}_{\DT;\,\varnothing\,\varnothing\,\yng(2)}(\fra,\phi_{I})&=\frac{\sh(\fra-\phi_I+\eps_1+2\eps_2)\sh(\phi_I-\fra+\eps_1-\eps_2+\eps_3)\sh(\phi_I-\fra+\eps_{23})}{\sh(\fra-\phi_I+\eps_2-\eps_3)\sh(\phi_I-\fra-\eps_1)\sh(\phi_I-\fra-2\eps_{2})}.
\eea

We list the explicit residues including the Weyl group factor only up to level two.
For level one, we have
\bea
(\epsilon_{2}-\epsilon_{3}),\quad -\frac{\sh\left(2 \epsilon _2\right) \sh\left(\epsilon _1+\epsilon _2\right)
   \sh\left(\epsilon _1+\epsilon _3\right)}{\sh\left(\epsilon _2\right)
   \sh\left(\epsilon _3\right) \sh\left(\epsilon _1-\epsilon _2+\epsilon _3\right)}
\eea
For level two, we have
\bea
(\epsilon_{2}-\epsilon_{3},-\epsilon_{3}),&\qquad \frac{\sh\left(\epsilon _1+\epsilon _2\right) \sh\left(\epsilon _1+2 \epsilon _2\right)}{2
   \sh\left(\epsilon _3\right) \sh\left(\epsilon _3-\epsilon _2\right)}\\
   \left(\epsilon _2-\epsilon _3,\epsilon _2-2 \epsilon _3\right),&\quad-\frac{\sh\left(2
   \epsilon _2\right) \sh\left(\epsilon _1+\epsilon _2\right) \sh\left(\epsilon _1+\epsilon
   _3\right) \sh\left(\epsilon _3-2 \epsilon _2\right) \sh\left(-\epsilon _1-\epsilon
   _2+\epsilon _3\right) \sh\left(\epsilon _1+2 \epsilon _3\right)}{2 \sh\left(\epsilon
   _2\right) \sh\left(\epsilon _3\right) \sh\left(2 \epsilon _3\right)
   \sh\left(\epsilon _3-\epsilon _2\right) \sh\left(\epsilon _1-\epsilon _2+\epsilon
   _3\right) \sh\left(\epsilon _1-\epsilon _2+2 \epsilon _3\right)}
\eea

\paragraph{Unrefined vertex}
\bea\Yboxdim{4pt}
\wtC_{\varnothing\varnothing\,{\yng(1,1)}}(q)=\frac{1}{(1-q)(1-q^{2})}
\eea

\paragraph{Refined vertex}
\bea\Yboxdim{4pt}
\wtC_{\varnothing\varnothing\,{\yng(1,1)}}(t,q)=\frac{1}{(1-t)(1-t^{2})},\quad\wtC_{\varnothing\,{\yng(1,1)}\,\varnothing}(t,q)=\frac{1}{(1-q)(1-q^{2})},\quad \wtC_{{\yng(1,1)}\varnothing\varnothing\,}(t,q)=\frac{1}{(1-t)(1-t^{2})} 
\eea

\paragraph{Macdonald vertex}
\bea\Yboxdim{4pt}
&\wtM_{\varnothing\varnothing\,\yng(1,1)}(x,y;q,t)=\frac{(tx;q)_{\infty}(tx^{2};q)_{\infty}}{(x;q)_{\infty}(x^{2};q)_{\infty}},\quad\wtM_{\varnothing\,\yng(1,1)\,\varnothing}(x,y;q,t)=\frac{1}{(1-y)(1-y^{2})}\\
&\wtM_{{\yng(1,1)}\varnothing\varnothing\,}(x,y;q,t)=\frac{1-tx+qx-qt}{(1-x)(1-x^{2})(1-qt)} 
\eea

\subsubsection{$\nu=\{1,1\}$ (Fig.~\ref{fig:PT3-1leg-ex2})}
\begin{figure}[ht]
    \centering
    \includegraphics[width=1.2cm]{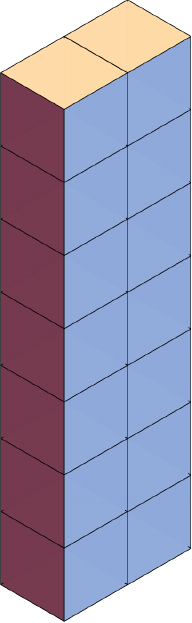}
    \caption{The PT box counting setup for $(\lambda,\mu,\nu)=(\varnothing,\varnothing,\{1,1\})$}
    \label{fig:PT3-1leg-ex2}
\end{figure}
\paragraph{PT configurations}
\bea
\begin{tabular}{|c|c|}\hline
   \text{instanton number $k$}  &  \text{poles} \\\hline\hline
    1 & $(\epsilon_{3}-\epsilon_{1})$   \\\hline
    2 &  $(\epsilon_{3}-\epsilon_{1},\epsilon_{3}),(\epsilon_{3}-\epsilon_{1},2\epsilon_{3}-\epsilon_{1}) $ \\\hline
    3 & $(\epsilon_{3}-\epsilon_{1},\epsilon_{3},-\epsilon_{1}+2\epsilon_{3})$\\
    & $(-\epsilon_{1}+\epsilon_{3},-\epsilon_{1}+2\epsilon_{3},-\epsilon_{1}+3\epsilon_{3})$\\\hline
\end{tabular}
\eea

\paragraph{JK residues}The framing node contribution is
\bea
\mathcal{Z}^{\D6_{\bar{4}}\tbar\D2\tbar\D0}_{\DT;\,\varnothing\,\varnothing\,\yng(1,1)}(\fra,\phi_{I})&=\frac{\sh(\fra-\phi_I+\eps_2+2\eps_1)\sh(\phi_I-\fra+\eps_2-\eps_1+\eps_3)\sh(\phi_I-\fra+\eps_{13})}{\sh(\fra-\phi_I+\eps_1-\eps_3)\sh(\phi_I-\fra-\eps_2)\sh(\phi_I-\fra-2\eps_{1})}.
\eea

The JK-residues are given as follows. For level one,
\bea
(\eps_3-\eps_1),\quad -\frac{\sh\left(2 \epsilon _1\right) \sh\left(\epsilon _1+\epsilon _2\right)
   \sh\left(\epsilon _2+\epsilon _3\right)}{\sh\left(\epsilon _1\right)
   \sh\left(\epsilon _3\right) \sh\left(-\epsilon _1+\epsilon _2+\epsilon _3\right)}.
\eea
For level two, we have
\bea
(-\eps_3,-\eps_3+\eps_1),&\quad \frac{\sh\left(\epsilon _1+\epsilon _2\right) \sh\left(2 \epsilon _1+\epsilon
   _2\right)}{2\sh\left(\epsilon _3\right) \sh\left(\epsilon _3-\epsilon
   _1\right)},\\
 (\eps_1-\eps_3,\eps_1-2\eps_3),  &\quad -\frac{\sh\left(2 \epsilon _1\right) \sh\left(\epsilon _1+\epsilon _2\right)
   \sh\left(\epsilon _3-2 \epsilon _1\right) \sh\left(-\epsilon _1-\epsilon _2+\epsilon
   _3\right) \sh\left(\epsilon _2+\epsilon _3\right) \sh\left(\epsilon _2+2 \epsilon
   _3\right)}{2\sh\left(\epsilon _1\right) \sh\left(\epsilon _3\right) \sh\left(2
   \epsilon _3\right) \sh\left(\epsilon _3-\epsilon _1\right) \sh\left(-\epsilon
   _1+\epsilon _2+\epsilon _3\right) \sh\left(-\epsilon _1+\epsilon _2+2 \epsilon _3\right)}.
\eea

\paragraph{Unrefined vertex}
\bea\Yboxdim{4pt}
\wtC_{\varnothing\varnothing\,{\yng(2)}}(q)=\frac{1}{(1-q)(1-q^{2})}
\eea

\paragraph{Refined vertex}
\bea\Yboxdim{4pt}
\wtC_{\varnothing\varnothing\,{\yng(2)}}(t,q)=\frac{1}{(1-t)(1-qt)},\quad\wtC_{\varnothing\,{\yng(2)}\,\varnothing}(t,q)=\frac{1}{(1-q)(1-q^{2})},\quad \wtC_{{\yng(2)}\varnothing\varnothing\,}(t,q)=\frac{1}{(1-t)(1-t^{2})} 
\eea

\paragraph{Macdonald vertex}
\bea\Yboxdim{4pt}
&\wtM_{\varnothing\varnothing\,\yng(2)}(x,y;q,t)=\frac{(tx;q)_{\infty}(txy;q)_{\infty}}{(x;q)_{\infty}(xy;q)_{\infty}},\quad\wtM_{\varnothing\,\yng(2)\,\varnothing}(x,y;q,t)=\frac{1-ty+qy-qt}{(1-y)(1-y^{2})(1-qt)} \\
&\wtM_{{\yng(2)}\varnothing\varnothing\,}(x,y;q,t)=\frac{1}{(1-x)(1-x^2)}
\eea

\subsubsection{$\nu=\{3\}$ (Fig.~\ref{fig:PT3-1leg-ex3})}
\begin{figure}[ht]
    \centering
    \includegraphics[width=1.2cm]{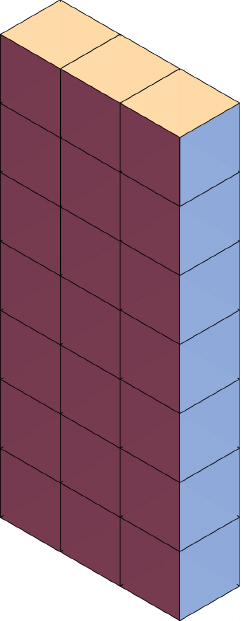}
    \caption{The PT box counting setup for $(\lambda,\mu,\nu)=(\varnothing,\varnothing,\{3\})$}
    \label{fig:PT3-1leg-ex3}
\end{figure}
\paragraph{PT configurations}
\bea
\begin{tabular}{|c|c|}\hline
   \text{instanton number $k$}  &  \text{poles} \\\hline\hline
    1 & $(-\epsilon_{3}+2\epsilon_{2})$   \\\hline
    2 &  $
\begin{array}{c}
 (-\epsilon _3+2\epsilon _2 , -\epsilon _3+\epsilon _2) \\
  (-\epsilon _3+2\epsilon _2 , -2 \epsilon _3+2\epsilon _2) \\
\end{array}$ \\\hline
    3 & $\begin{array}{c}
 (-\epsilon _3 , -\epsilon -3+2\epsilon _2 , -\epsilon _3+\epsilon _2 )\\
 (-\epsilon _3+2\epsilon _2 , -\epsilon -3+\epsilon _2 , -2 \epsilon _3+2\epsilon _2 )\\
 (-\epsilon _3+2\epsilon _2 , -2\epsilon _3+2\epsilon _2 , -3 \epsilon _3+2\epsilon _2 )\\
\end{array}$\\\hline
\end{tabular}
\eea

\paragraph{JK residues}
The framing node contribution is
\bea
\mathcal{Z}^{\D6_{\bar{4}}\tbar\D2\tbar\D0}_{\DT;\,\varnothing\,\varnothing\,\yng(3)}(\fra,\phi_{I})&=\frac{\sh(\fra-\phi_I+\eps_1+3\eps_2)\sh(\phi_I-\fra+\eps_1-2\eps_2+\eps_3)\sh(\phi_I-\fra+\eps_{23})}{\sh(\fra-\phi_I+2\eps_2-\eps_3)\sh(\phi_I-\fra-\eps_1)\sh(\phi_I-\fra-3\eps_2)}
\eea
For level one, we have
\bea
(-\epsilon_{3}+2\epsilon_{2}),\quad -\frac{\sh\left(3 \epsilon _2\right) \sh\left(\epsilon _1+\epsilon _2\right)
   \sh\left(\epsilon _1+\epsilon _3\right)}{\sh\left(\epsilon _2\right)
   \sh\left(\epsilon _3\right) \sh\left(\epsilon _1-2 \epsilon _2+\epsilon _3\right)}
\eea
For level two, we have
\bea
(2\eps_2-\eps_3,\eps_2-\eps_3),&\quad \frac{\sh\left(3 \epsilon _2\right) \sh\left(\epsilon _1+\epsilon _2\right)
   \sh\left(\epsilon _1+2 \epsilon _2\right) \sh\left(\epsilon _1+\epsilon
   _3\right)}{2\sh\left(\epsilon _2\right) \sh\left(\epsilon _3\right)
   \sh\left(\epsilon _1-2 \epsilon _2+\epsilon _3\right) \sh\left(\epsilon _3-\epsilon
   _2\right)}\\
  (2\eps_2-\eps_3,2\eps_2-2\eps_3), &\quad -\frac{\sh\left(3 \epsilon _2\right) \sh\left(\epsilon _1+\epsilon _2\right)
   \sh\left(\epsilon _1+\epsilon _3\right) \sh\left(\epsilon _3-3 \epsilon _2\right)
   \sh\left(-\epsilon _1-\epsilon _2+\epsilon _3\right) \sh\left(\epsilon _1+2 \epsilon
   _3\right)}{2\sh\left(\epsilon _2\right) \sh\left(\epsilon _3\right) \sh\left(2
   \epsilon _3\right) \sh\left(\epsilon _1-2 \epsilon _2+\epsilon _3\right)
   \sh\left(\epsilon _3-\epsilon _2\right) \sh\left(\epsilon _1-2 \epsilon _2+2 \epsilon
   _3\right)}.
\eea

\paragraph{Unrefined vertex}
\bea\Yboxdim{4pt}
\wtC_{\varnothing\varnothing\,{\yng(1,1,1)}}(q)=\frac{1}{(1-q)(1-q^{2})(1-q^{3})}
\eea

\paragraph{Refined vertex}
\bea\Yboxdim{4pt}
&\wtC_{\varnothing\varnothing\,{\yng(1,1,1)}}(t,q)=\frac{1}{(1-t)(1-t^{2})(1-t^{3})},\quad\wtC_{\varnothing\,{\yng(1,1,1)}\,\varnothing}(t,q)=\frac{1}{(1-q)(1-q^{2})(1-q^{3})},\\
&\wtC_{{\yng(1,1,1)}\varnothing\varnothing\,}(t,q)=\frac{1}{(1-t)(1-t^{2})(1-t^{3})} 
\eea

\paragraph{Macdonald refined vertex}
\bea\Yboxdim{4pt}
&\wtM_{\varnothing\varnothing\,\yng(1,1,1)}(x,y;q,t)=\frac{(tx;q)_{\infty}(tx^{2};q)_{\infty}(tx^{3};q)_{\infty}}{(x;q)_{\infty}(x^{2};q)_{\infty}(x^{3};q)_{\infty}},\quad\wtM_{\varnothing\,\yng(1,1,1)\,\varnothing}(x,y;q,t)=\frac{1}{(1-y)(1-y^{2})(1-y^{3})}\\
&\wtM_{{\yng(1,1,1)}\varnothing\varnothing\,}(x,y;q,t)=\frac{\#}{(1-x)(1-x^{2})(1-x^{3})(1-qt)(1-qt^{2})} 
\eea
where the numerator is 
\bea
\#&=1 - q t - q^2 t + q^2 x + q x - q^3 t x - q^2 t x - q t x - t x + q^2 t^2 x + q t^2 x \\
&+ q^2 x^2 + q x^2 - 2 q t x^2 - q^2 t x^2 + q t^2 x^2 + q^2 t^2 x^2 - q^3 t x^2 + q^3 x^3 \\
&- q t x^3 - q^2 t x^3 + t^2 x^3
\eea

\subsubsection{$\nu=\{2,1\}$ (Fig.~\ref{fig:PT3-1leg-ex4})}
\begin{figure}[ht]
    \centering
    \includegraphics[width=1.2cm]{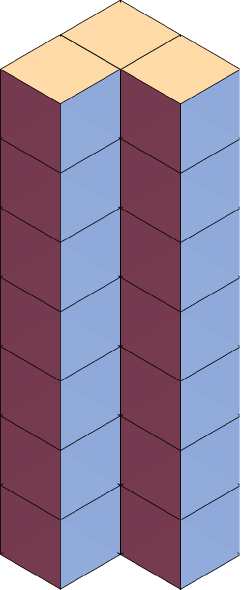}
    \caption{The PT box counting setup for $(\lambda,\mu,\nu)=(\varnothing,\varnothing,\{2,1\})$}
    \label{fig:PT3-1leg-ex4}
\end{figure}
\paragraph{Pole structure}
\bea
\begin{tabular}{|c|c|}\hline
   \text{instanton number $k$}  &  \text{poles} \\\hline\hline
    1 & $\left(
 -\epsilon _3+\epsilon _1),(
 -\epsilon _3+\epsilon _2 
\right)$   \\\hline
    2 &  $
\begin{array}{c}
 (-\epsilon _3+\epsilon _1 , -\epsilon _3+\epsilon _2) \\
 (-\epsilon _3+\epsilon _1 , -2 \epsilon _3+\epsilon _1) \\
 (-\epsilon _3+\epsilon _2 ,- 2 \epsilon _3+\epsilon _2) \\
\end{array}$ \\\hline
    3 & $\begin{array}{c}
 (-\epsilon _3 , -\epsilon _3+\epsilon _1 , -\epsilon _3+\epsilon _2 )\\
 (-\epsilon _3+\epsilon _1 , -\epsilon _3+\epsilon _2 , -2 \epsilon _3+\epsilon _1 )\\
 (-\epsilon _3+\epsilon _1 , -\epsilon _3+\epsilon _2 , -2 \epsilon _3+\epsilon _2 )\\
 (-\epsilon _3+\epsilon _1 , -2 \epsilon _3+\epsilon _1 , -3 \epsilon _3+\epsilon _1 )\\
 (-\epsilon _3+\epsilon _2 , -2 \epsilon _3+\epsilon _2 ,-3 \epsilon _3+\epsilon _2) \\
\end{array}$\\\hline
\end{tabular}
\eea

\paragraph{JK residues}The framing node contribution is 
\bea
\mathcal{Z}^{\D6_{\bar{4}}\tbar\D2\tbar\D0}_{\DT;\,\varnothing\,\varnothing\,\yng(2,1)}(\fra,\phi_{I})&=\frac{\sh(\fra-\phi_I+2\eps_1+\eps_2)\sh(\fra-\phi_I+\eps_1+2\eps_2)\sh(\phi_I-\fra+\eps_3)}{\sh(\fra+\eps_1-\eps_3-\phi_I)\sh(\fra-\phi_I+\eps_2-\eps_3)}\\
&\times \frac{\sh(\phi_I-\fra+\eps_1-\eps_2+\eps_3)\sh(\phi_I-\fra-\eps_1+\eps_{2}+\eps_3)}{\sh(\phi_I-\fra-2\eps_1)\sh(\phi_I-\fra-2\eps_2)\sh(\phi_I-\fra-\eps_{12})}.
\eea

For level one, we have
\bea
(\eps_2-\eps_3),&\quad -\frac{\sh\left(\epsilon _1+\epsilon _2\right) \sh\left(2 \epsilon _2-\epsilon _1\right)
   \sh\left(2 \epsilon _1+\epsilon _3\right)}{\sh\left(\epsilon _2-\epsilon _1\right)
   \sh\left(\epsilon _3\right) \sh\left(2 \epsilon _1-\epsilon _2+\epsilon
   _3\right)}\\
  (\eps_1-\eps_3),&\quad -\frac{\sh\left(\epsilon _2-2 \epsilon _1\right) \sh\left(\epsilon
   _1+\epsilon _2\right) \sh\left(2 \epsilon _2+\epsilon _3\right)}{\sh\left(\epsilon
   _2-\epsilon _1\right) \sh\left(\epsilon _3\right) \sh\left(-\epsilon _1+2 \epsilon
   _2+\epsilon _3\right)}
\eea
For level two, we have
\bea
(\eps_1-\eps_3,\eps_2-\eps_3),&\quad \frac{\sh\left(2 \epsilon _1\right) \sh\left(2 \epsilon _2\right) \sh\left(\epsilon
   _1+\epsilon _3\right) \sh\left(\epsilon _2+\epsilon _3\right) \sh\left(\epsilon
   _1+\epsilon _2\right)^2}{2\sh\left(\epsilon _1\right) \sh\left(\epsilon _2\right)
   \sh\left(\epsilon _3\right)^2 \sh\left(\epsilon _1-\epsilon _2+\epsilon _3\right)
   \sh\left(-\epsilon _1+\epsilon _2+\epsilon _3\right)}\\
  (\eps_2-\eps_3,\eps_2-2\eps_3) &\quad -\frac{\sh\left(2 \epsilon
   _2-\epsilon _1\right) \sh\left(2 \epsilon _1+\epsilon _3\right) \sh\left(\epsilon _1-2
   \epsilon _2+\epsilon _3\right) \sh\left(-\epsilon _1-\epsilon _2+\epsilon _3\right)
   \sh\left(2 \epsilon _1+2 \epsilon _3\right) \sh\left(\epsilon _1+\epsilon
   _2\right)}{2\sh\left(\epsilon _2-\epsilon _1\right) \sh\left(\epsilon _3\right)
   \sh\left(2 \epsilon _3\right) \sh\left(\epsilon _1-\epsilon _2+\epsilon _3\right)
   \sh\left(2 \epsilon _1-\epsilon _2+\epsilon _3\right) \sh\left(2 \epsilon _1-\epsilon
   _2+2 \epsilon _3\right)}\\
  (\eps_1-\eps_3,\eps_1-2\eps_3), &\quad -\frac{\sh\left(\epsilon _2-2 \epsilon _1\right)
   \sh\left(-\epsilon _1-\epsilon _2+\epsilon _3\right) \sh\left(-2 \epsilon _1+\epsilon
   _2+\epsilon _3\right) \sh\left(2 \epsilon _2+\epsilon _3\right) \sh\left(2 \epsilon
   _2+2 \epsilon _3\right) \sh\left(\epsilon _1+\epsilon _2\right)}{2\sh\left(\epsilon
   _2-\epsilon _1\right) \sh\left(\epsilon _3\right) \sh\left(2 \epsilon _3\right)
   \sh\left(-\epsilon _1+\epsilon _2+\epsilon _3\right) \sh\left(-\epsilon _1+2 \epsilon
   _2+\epsilon _3\right) \sh\left(-\epsilon _1+2 \epsilon _2+2 \epsilon _3\right)}
\eea

\paragraph{Unrefined vertex}
\bea\Yboxdim{4pt}
\wtC_{\varnothing\varnothing\,{\yng(2,1)}}(q)=\frac{1}{(1-q)(1-q^{2})(1-q^{3})}
\eea

\paragraph{Refined vertex}
\bea\Yboxdim{4pt}
&\wtC_{\varnothing\varnothing\,{\yng(2,1)}}(t,q)=\frac{1}{(1-t)^{2}(1-qt^{2})},\quad\wtC_{\varnothing\,{\yng(2,1)}\,\varnothing}(t,q)=\frac{1}{(1-q)^{2}(1-q^{3})},\\
&\wtC_{{\yng(2,1)}\varnothing\varnothing\,}(t,q)=\frac{1}{(1-t)^{2}(1-t^{3})} 
\eea

\paragraph{Macdonald refined vertex}
\bea\Yboxdim{4pt}
\wtM_{\varnothing\varnothing\,\yng(2,1)}(x,y;q,t)&=\frac{(tx;q)^{2}_{\infty}(tx^{2}y;q)_{\infty}}{(x;q)^{2}_{\infty}(x^{2}y;q)_{\infty}},\\
\wtM_{{\yng(2,1)}\,\varnothing\varnothing\,}(x,y;q,t)&=\frac{1+x-tx^{2}+qx^{2}-qt^{2}-t^{2}x^{2}+qtx^{2}-qt^{2}x}{(1-x)(1-x^{2})(1-x^{3})(1-qt^{2})} \\
\wtM_{\varnothing\,\yng(2,1)\,\varnothing}(x,y;q,t)&=\frac{1+y-ty^{2}+qy^{2}-qt^{2}-t^{2}y^{2}+qty^{2}-qt^{2}y}{(1-y)(1-y^{2})(1-y^{3})(1-qt^{2})} .
\eea

\subsubsection{$\nu=\{1,1,1\}$ (Fig.~\ref{fig:PT3-1leg-ex5})}
\begin{figure}[ht]
    \centering
    \includegraphics[width=1.2cm]{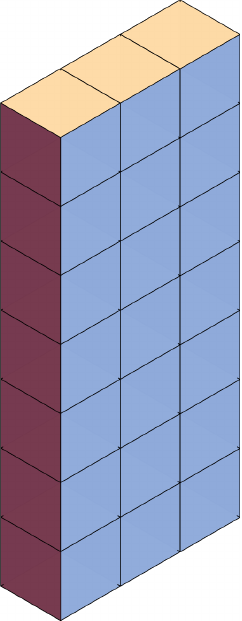}
    \caption{The PT box counting setup for $(\lambda,\mu,\nu)=(\varnothing,\varnothing,\{1,1,1\})$}
    \label{fig:PT3-1leg-ex5}
\end{figure}
\paragraph{Pole structure}
\bea
\begin{tabular}{|c|c|}\hline
   \text{instanton number $k$}  &  \text{poles} \\\hline\hline
    1 & $(-\epsilon_{3}+2\epsilon_{1})$   \\\hline
    2 &  $
\begin{array}{c}
 (-\epsilon _3+2\epsilon _1 , -\epsilon _3+\epsilon _1) \\
  (-\epsilon _3+2\epsilon _1 ,-2 \epsilon _3+2\epsilon _1) \\
\end{array}$ \\\hline
    3 & $\begin{array}{c}
 (-\epsilon _3 , -\epsilon_3+2\epsilon _1 , -\epsilon _3+\epsilon _1 )\\
 (-\epsilon _3+2\epsilon _1 , -\epsilon _3+\epsilon _1 , -2 \epsilon _3+2\epsilon _1 )\\
 (-\epsilon _3+2\epsilon _1 , -2\epsilon _3+2\epsilon _1 , -3 \epsilon _3+2\epsilon _1)\\
\end{array}$\\\hline
\end{tabular}
\eea

\paragraph{JK residues}
The framing node contribution is
\bea
\mathcal{Z}^{\D6_{\bar{4}}\tbar\D2\tbar\D0}_{\DT;\,\varnothing\,\varnothing\,\yng(1,1,1)}(\fra,\phi_{I})&=\frac{\sh(\fra-\phi_I+\eps_2+3\eps_1)\sh(\phi_I-\fra+\eps_2-2\eps_1+\eps_3)\sh(\phi_I-\fra+\eps_{13})}{\sh(\fra-\phi_I+2\eps_1-\eps_3)\sh(\phi_I-\fra-\eps_2)\sh(\phi_I-\fra-3\eps_1)}.
\eea
For level one, we have
\bea
(-\eps_3+2\eps_1)&,\quad -\frac{\sh\left(3 \epsilon _1\right) \sh\left(\epsilon _1+\epsilon _2\right)
   \sh\left(\epsilon _2+\epsilon _3\right)}{\sh\left(\epsilon _1\right)
   \sh\left(\epsilon _3\right) \sh\left(-2 \epsilon _1+\epsilon _2+\epsilon _3\right)}
\eea
For level two, we have
\bea
(2\eps_1-\eps_3,\eps_1-\eps_3),&\quad \frac{\sh\left(3 \epsilon _1\right) \sh\left(\epsilon _1+\epsilon _2\right)
   \sh\left(2 \epsilon _1+\epsilon _2\right) \sh\left(\epsilon _2+\epsilon
   _3\right)}{2\sh\left(\epsilon _1\right) \sh\left(\epsilon _3\right)
   \sh\left(\epsilon _3-\epsilon _1\right) \sh\left(-2 \epsilon _1+\epsilon _2+\epsilon
   _3\right)}\\
 (\eps_1-\eps_3,2\eps_1-\eps_3),  &\quad -\frac{\sh\left(3 \epsilon _1\right) \sh\left(\epsilon _1+\epsilon _2\right)
   \sh\left(\epsilon _3-3 \epsilon _1\right) \sh\left(-\epsilon _1-\epsilon _2+\epsilon
   _3\right) \sh\left(\epsilon _2+\epsilon _3\right) \sh\left(\epsilon _2+2 \epsilon
   _3\right)}{2\sh\left(\epsilon _1\right) \sh\left(\epsilon _3\right) \sh\left(2
   \epsilon _3\right) \sh\left(\epsilon _3-\epsilon _1\right) \sh\left(-2 \epsilon
   _1+\epsilon _2+\epsilon _3\right) \sh\left(-2 \epsilon _1+\epsilon _2+2 \epsilon _3\right)}
\eea

\paragraph{Unrefined vertex}
\bea\Yboxdim{4pt}
\wtC_{\varnothing\varnothing\,{\yng(3)}}(q)=\frac{1}{(1-q)(1-q^{2})(1-q^{3})}
\eea

\paragraph{Refined vertex}
\bea\Yboxdim{4pt}
&\wtC_{\varnothing\varnothing\,{\yng(3)}}(t,q)=\frac{1}{(1-t)(1-tq)(1-tq^2)},\quad\wtC_{\varnothing\,{\yng(3)}\,\varnothing}(t,q)=\frac{1}{(1-q)(1-q^{2})(1-q^{3})},\\
&\wtC_{{\yng(3)}\varnothing\varnothing\,}(t,q)=\frac{1}{(1-t)(1-tq)(1-tq^2)} 
\eea

\paragraph{Macdonald refined vertex}
\bea\Yboxdim{4pt}
&\wtM_{\varnothing\varnothing\,\yng(3)}(x,y;q,t)=\frac{(tx;q)_{\infty}(txy;q)_{\infty}(tx y^2;q)_{\infty}}{(x;q)_{\infty}(xy;q)_{\infty}(xy^2;q)_{\infty}},\quad \wtM_{{\yng(3)}\varnothing\varnothing\,}(x,y;q,t)=\frac{1}{(1-x)(1-x^2)(1-x^3)},\\
&\wtM_{\varnothing\,\yng(3)\,\varnothing}(x,y;q,t)=\frac{\#}{(1-y)(1-y^2)(1-y^3)(1-qt)(1-q^2t)}
\eea
where the numerator is 
\bea
\#&=-q^3ty^2 + q^2t^2y^2 + q^3y^3 - q^2ty^3 + q^3t^2\\
&- q^3ty + q^2t^2y - q^2ty^2 + qt^2y^2 - qty^3 + t^2y^3 \\
&- q^2ty + qt^2y + q^2y^2 - qty^2 - q^2t + q^2y\\
&- qty + qy^2 - ty^2 - qt + qy - ty + 1
\eea

\subsection{Two-legs PT3 counting}
We summarize the formulas for two-legs PT3 counting with the boundary condition $(\lambda,\mu,\varnothing)$ for $3\leq |\lambda|+|\mu| \leq 4$.

\subsubsection{$\lambda=\{2\}, \mu=\{1\}$ (Fig.~\ref{fig:PT3-2leg-ex1})}
\begin{figure}[ht]
    \centering
    \includegraphics[width=5cm]{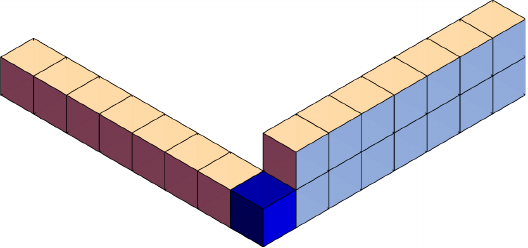}
    \caption{The PT box counting setup for $(\lambda,\mu,\nu)=(\{2\},\{1\},\varnothing)$}
    \label{fig:PT3-2leg-ex1}
\end{figure}
\paragraph{Pole structure}
\bea
\begin{tabular}{|c|c|}\hline
   \text{instanton number $k$}  &  \text{poles} \\\hline\hline
    1 & $(0), (-\epsilon_{1}+\epsilon_{3})$   \\\hline
    2 &  $
\begin{array}{c}
 (0 , -\epsilon _2) \\
  (0,  -\epsilon _1+\epsilon_3) \\
  (-\epsilon_1+\epsilon_3, -2\epsilon_1+\epsilon_3)
\end{array}$ \\\hline
    3 & $\begin{array}{c}
 (0 , -\epsilon _1 , -\epsilon _1+\epsilon_3 )\\
 (0 , -\epsilon _2,  -2\epsilon _2 )\\
 (0,-\epsilon _2, -\epsilon_1+\epsilon_3)\\
 (0,-\epsilon _1+\epsilon_3, -2\epsilon _1+\epsilon_3)\\
 (-\epsilon_1+\epsilon_3, -2\epsilon _1+\epsilon_3, -3\epsilon _1+\epsilon_3)\\
\end{array}$\\\hline
\end{tabular}
\eea

\paragraph{JK residues}
The framing node contribution is
\bea
\mathcal{Z}^{\D6_{\bar{4}}\tbar\D2\tbar\D0}_{\DT;\,\yng(2)\,\yng(1)\,\varnothing}(\fra,\phi_{I})&=\frac{\sh\left(\fra+\epsilon _1+\epsilon _2+\epsilon _3-\phi_I\right) \sh\left(\fra+\epsilon _2+2 \epsilon _3-\phi_I\right)
   \sh\left(-\fra+\epsilon _1+\phi_I\right) }{\sh\left(\fra-\phi_I\right) \sh\left(\fra-\epsilon _1+\epsilon _3-\phi_I\right)
   }\\
   &\times \frac{\sh\left(-\fra+\epsilon _1+\epsilon _2-\epsilon _3+\phi_I\right)
   \sh\left(-\fra+\epsilon _3+\phi_I\right)}{\sh\left(-\fra-\epsilon _1-\epsilon _2+\phi_I\right) \sh\left(-\fra-2 \epsilon _3+\phi_I\right) \sh\left(-\fra-\epsilon
   _2-\epsilon _3+\phi_I\right)}.
\eea
For level one, we have
\bea
(0),&\quad -\frac{\text{sh}\left(\epsilon _1+\epsilon _3\right) \text{sh}\left(-\epsilon _1-\epsilon _2+\epsilon _3\right) \text{sh}\left(\epsilon
   _2+2 \epsilon _3\right)}{\text{sh}\left(\epsilon _2\right) \text{sh}\left(2 \epsilon _3\right) \text{sh}\left(\epsilon _3-\epsilon
   _1\right)}\\
   (\eps_3-\eps_1),&\quad \frac{\text{sh}\left(2 \epsilon _1+\epsilon _2\right) \text{sh}\left(\epsilon _2+\epsilon _3\right) \text{sh}\left(2 \epsilon _3-\epsilon
   _1\right)}{\text{sh}\left(\epsilon _1\right) \text{sh}\left(\epsilon _3-\epsilon _1\right) \text{sh}\left(-2 \epsilon _1-\epsilon
   _2+\epsilon _3\right)}.
\eea
For level two, we have
\bea
(0,-\eps_2),&\quad \frac{\text{sh}\left(\epsilon _1+\epsilon _3\right) \text{sh}\left(-\epsilon _1-\epsilon _2+\epsilon _3\right) \text{sh}\left(\epsilon
   _1-\epsilon _2+\epsilon _3\right) \text{sh}\left(2 \epsilon _2+2 \epsilon _3\right)}{2 \text{sh}\left(\epsilon _2\right)
   \text{sh}\left(2 \epsilon _2\right) \text{sh}\left(2 \epsilon _3\right) \text{sh}\left(-\epsilon _1+\epsilon _2+\epsilon _3\right)}\\
(0,-\eps_1+\eps_3),&\quad -\frac{\text{sh}\left(2 \epsilon _1\right) \text{sh}\left(\epsilon _1+\epsilon _2\right) \text{sh}\left(\epsilon _1+\epsilon _3\right)
   \text{sh}\left(\epsilon _2+\epsilon _3\right){}^2 \text{sh}\left(-\epsilon _1+\epsilon _2+2 \epsilon _3\right)}{2
   \text{sh}\left(\epsilon _1\right){}^2 \text{sh}\left(\epsilon _2\right) \text{sh}\left(\epsilon _3\right) \text{sh}\left(\epsilon _3-2
   \epsilon _1\right) \text{sh}\left(-\epsilon _1+\epsilon _2+\epsilon _3\right)}\\
(-\eps_1+\eps_3,-2\eps_1+\eps_3),&\quad \frac{\text{sh}\left(2 \epsilon _1+\epsilon _2\right) \text{sh}\left(3 \epsilon _1+\epsilon _2\right) \text{sh}\left(\epsilon _2+\epsilon
   _3\right) \text{sh}\left(-\epsilon _1+\epsilon _2+\epsilon _3\right) \text{sh}\left(2 \epsilon _3-2 \epsilon _1\right)
   \text{sh}\left(2 \epsilon _3-\epsilon _1\right)}{2 \text{sh}\left(\epsilon _1\right) \text{sh}\left(2 \epsilon _1\right)
   \text{sh}\left(\epsilon _3-2 \epsilon _1\right) \text{sh}\left(\epsilon _3-\epsilon _1\right) \text{sh}\left(-3 \epsilon _1-\epsilon
   _2+\epsilon _3\right) \text{sh}\left(-2 \epsilon _1-\epsilon _2+\epsilon _3\right)}
\eea

\paragraph{Unrefined vertex}
\bea\Yboxdim{4pt}
\wtC_{\,{\yng(1,1)}\,\,\yng(1)\,\varnothing}(q)=\frac{1-q^{2}+q^{3}}{(1-q)^{2}(1-q^{2})}
\eea

\paragraph{Refined vertex}
\bea\Yboxdim{4pt}
&\wtC_{\,{\yng(1,1)}\,\,\yng(1)\,\varnothing}(t,q)=\frac{1-q+t-t^{2}+qt^{2}}{(1-t)(1-t^2)(1-q)},\quad \wtC_{\varnothing\,{\yng(1,1)}\,\,\yng(1)\,}(t,q)=\frac{1-q^2+q^2t}{(1-t)(1-q)(1-q^2)},\\
&\wtC_{\yng(1)\,\,{\yng(1,1)}\,\varnothing\,}(t,q)=\frac{1-q^2+q^2 t-t+qt}{(1-q)(1-q^2)(1-t)}
\eea

\paragraph{Macdonald refined vertex}
\bea\Yboxdim{4pt}
\wtM_{\yng(1)\,\varnothing\,\yng(1,1)}(x,y;q,t)&=\frac{(tx;q)_{\infty}(tx^{2};q)_{\infty}}{(x;q)_{\infty}(x^{2};q)_{\infty}}\frac{1-x^{2}+x^{2}y}{1-x},\quad \wtM_{\varnothing\,\yng(1,1)\,\yng(1)}(x,y;q,t)=\frac{(tx;q)_{\infty}}{(x;q)_{\infty}}\frac{1-y^{2}+xy^{2}}{(1-y)(1-y^{2})}, \\
\wtM_{{\yng(1,1)}\,\yng(1)\,\varnothing\,}(x,y;q,t)&=\frac{1-y+x-x^{2}-q^{2}+x^{2}y+q^{2}y-tx^{2}+qx^{2}-qtx+q^{2}x^{2}-q^{2}x^{2}y}{(1-y)(1-x)(1-x^{2})(1-q^{2})}.
\eea
For the last case, the series expansion is
\bea
\wtM_{{\yng(1,1)}\,\yng(1)\,\varnothing\,}(x,y;q,t)&=1+\frac{2-qt-q^{2}}{1-q^{2}}x+\left(\frac{1-qt}{1-q^{2}}xy+\frac{2-q-t}{1-q}x^{2}\right)\\
&+\left(\frac{1-qt}{1-q^{2}}xy^{2}+\frac{1-t}{1-q}x^{2}y+\frac{3-t+q-2qt-q^{2}}{1-q^{2}}x^{3}\right)+\cdots
\eea

\subsubsection{$\lambda=\{1,1\}, \mu=\{1\}$ (Fig.~\ref{fig:PT3-2leg-ex2})}
\begin{figure}[ht]
    \centering
    \includegraphics[width=5cm]{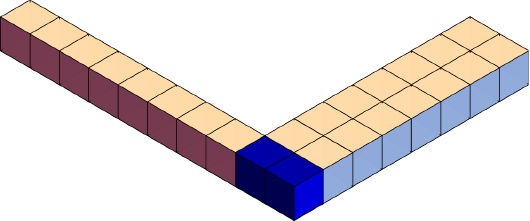}
    \caption{The PT box counting setup for $(\lambda,\mu,\nu)=(\{1,1\},\{1\},\varnothing)$}
    \label{fig:PT3-2leg-ex2}
\end{figure}
\paragraph{Pole structure}
\bea
\begin{tabular}{|c|c|}\hline
   \text{instanton number $k$}  &  \text{poles} \\\hline\hline
    1 & $ (\epsilon_{2})$   \\\hline
    2 &  $
\begin{array}{c}
 (0 , \epsilon _2) \\
  (-\epsilon_1+\epsilon_2,  \epsilon_2) 
\end{array}$ \\\hline
    3 & $\begin{array}{c}
 (0 , -\epsilon _1+\epsilon_2 , \epsilon_2 )\\
 (0 , +\epsilon _2,  -\epsilon _2 )\\
 (-\epsilon_1+\epsilon_2, -2\epsilon _1+\epsilon_2, \epsilon_2)
\end{array}$\\\hline
\end{tabular}
\eea

\paragraph{JK residues}
The framing node contribution is 
\bea
\mathcal{Z}^{\D6_{\bar{4}}\tbar\D2\tbar\D0}_{\DT;\,\yng(1,1)\,\yng(1)\,\varnothing}(\fra,\phi_{I})&=\frac{\text{sh}\left(\fra+\epsilon _1+2 \epsilon _2+\epsilon _3-\phi_I\right) \text{sh}\left(-\fra+\epsilon _1+\epsilon _2+\phi_I\right)
   \text{sh}\left(-\fra-\epsilon _2+\epsilon _3+\phi_I\right)}{\text{sh}\left(\fra+\epsilon _2-\phi_I\right) \text{sh}\left(-\fra-\epsilon _1-2
   \epsilon _2+\phi_I\right) \text{sh}\left(-\fra-\epsilon _3+\phi_I\right)}.
\eea
For the level one, we have
\bea
(\eps_2),&\quad -\frac{\text{sh}\left(\epsilon _1+2 \epsilon _2\right) \text{sh}\left(\epsilon _1+\epsilon _3\right) \text{sh}\left(\epsilon _2+\epsilon
   _3\right)}{\text{sh}\left(\epsilon _1\right) \text{sh}\left(\epsilon _2\right) \text{sh}\left(\epsilon _3-\epsilon _2\right)}.
\eea
For the level two, we have
\bea
(0,\eps_2),&\quad -\frac{\text{sh}\left(\epsilon _1+\epsilon _2\right) \text{sh}\left(\epsilon _1+2 \epsilon _2\right) \text{sh}\left(\epsilon _1+\epsilon
   _3\right) \text{sh}\left(\epsilon _1-\epsilon _2+\epsilon _3\right) \text{sh}\left(\epsilon _2+\epsilon _3\right) \text{sh}\left(2
   \epsilon _2+\epsilon _3\right)}{2 \text{sh}\left(\epsilon _1\right) \text{sh}\left(\epsilon _2\right) \text{sh}\left(2 \epsilon
   _2\right) \text{sh}\left(\epsilon _2-\epsilon _1\right) \text{sh}\left(\epsilon _3\right) \text{sh}\left(\epsilon _3-\epsilon
   _2\right)}\\
(\eps_2,\eps_2-\eps_1),&\quad \frac{\text{sh}\left(2 \epsilon _2\right) \text{sh}\left(\epsilon _1+2 \epsilon _2\right) \text{sh}\left(\epsilon _1+\epsilon _3\right)
   \text{sh}\left(2 \epsilon _1+\epsilon _3\right) \text{sh}\left(\epsilon _2+\epsilon _3\right) \text{sh}\left(-\epsilon _1+\epsilon
   _2+\epsilon _3\right)}{2 \text{sh}\left(\epsilon _1\right) \text{sh}\left(2 \epsilon _1\right) \text{sh}\left(\epsilon _2\right)
   \text{sh}\left(\epsilon _2-\epsilon _1\right) \text{sh}\left(\epsilon _3-\epsilon _2\right) \text{sh}\left(\epsilon _1-\epsilon
   _2+\epsilon _3\right)}
\eea

\paragraph{Unrefined vertex}
\bea\Yboxdim{4pt}
\wtC_{\,{\yng(2)}\,\,\yng(1)\,\varnothing}(q)=\frac{1-q+q^{3}}{(1-q)^{2}(1-q^{2})}
\eea

\paragraph{Refined vertex}
\bea\Yboxdim{4pt}
&\wtC_{\,{\yng(2)}\,\,\yng(1)\,\varnothing}(t,q)=\frac{1-q+qt^{2}}{(1-t)(1-t^2)(1-q)},\quad \wtC_{\varnothing \,{\yng(2)}\,\,\yng(1)\, }(t,q)=\frac{q^3 (-t)+q^3+q^2 t^2-q^2+q t-q+1}{(1-q)(1-q^2)(1-t)},\\
&\wtC_{\yng(1)\,\,{\yng(2)}\,\varnothing\,}(t,q)=\frac{1-t+q^2t}{(1-t)^{2}(1-qt)}
\eea

\paragraph{Macdonald refined vertex}
\bea\Yboxdim{4pt}
\wtM_{\adjustbox{valign=c}{\yng(1)}\,\varnothing
\,\adjustbox{valign=c}{\yng(2)}}(x,y;q,t)&=\frac{(tx;q)_{\infty}(txy;q)_{\infty}}{(x;q)_{\infty}(xy;q)_{\infty}}\frac{1-x+xy^{2}}{1-x},\quad \wtM_{\varnothing\,\adjustbox{valign=c}{\yng(2)}\,\adjustbox{valign=c}{\yng(1)}}(x,y;q,t)=\frac{(tx;q)_{\infty}}{(x;q)_{\infty}}\frac{\#}{(1-qt)(1-y)(1-y^{2})}, \\
\wtM_{{\yng(2)}\,\yng(1)\,\varnothing\,}(x,y;q,t)&=\frac{1-y-q+qy+x^2y-tx^2+qx^2-qx^2y}{(1-x)(1-x^2)(1-y)(1-q)}.
\eea
where
\bea
\#&=1-y-y^{2}+xy-qt+y^{3}-txy+qxy+qty-xy^3+x^2y^2+qty^2\\
&-qtxy+txy^3-qxy^3-qty^3-tx^2y^3+qx^2y^3+qtxy^3-qtx^2y^2
\eea
For the last case, the series expansion is
\bea
\wtM_{{\yng(2)}\,\yng(1)\,\varnothing\,}(x,y;q,t)=1+x+\frac{2-q-t}{1-q}x^{2}+\left(\frac{2-q-t}{1-q}x^{3}+\frac{1-t}{1-q}x^{2}y\right)+\cdots
\eea

\subsubsection{$\lambda=\{3\}, \mu=\{1\}$ (Fig.~\ref{fig:PT3-2leg-ex3})}
\begin{figure}[ht]
    \centering
    \includegraphics[width=5cm]{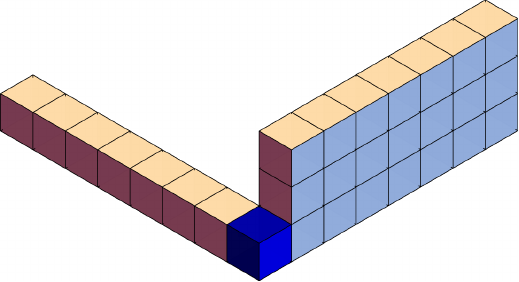}
    \caption{The PT box counting setup for $(\lambda,\mu,\nu)=(\{3\},\{1\},\varnothing)$}
    \label{fig:PT3-2leg-ex3}
\end{figure}
\paragraph{Pole structure}
\bea
\begin{tabular}{|c|c|}\hline
   \text{instanton number $k$}  &  \text{poles} \\\hline\hline
    1 & $ (0), (\epsilon_1-2\epsilon_3)$   \\\hline
    2 &  $
\begin{array}{c}
 (0 , -\epsilon _2) \\
  (0,-\epsilon_1+2\epsilon_3) \\
  (-\epsilon_1+2\epsilon_3,-2\epsilon_1+2\epsilon_3) \\
  (-\epsilon_1+2\epsilon_3,-\epsilon_1+\epsilon_3) 
\end{array}$ \\\hline
    3 & $\begin{array}{c}
( 0 , -\epsilon _2 , -2 \epsilon _2) \\
 (0 , -\epsilon _2 , -\epsilon _1+2 \epsilon _3) \\
 (0 , -\epsilon _1+2 \epsilon _3 , -2 \epsilon _1+2 \epsilon _3) \\
 (0 , -\epsilon _1+2 \epsilon _3 , -\epsilon _1+\epsilon _3 )\\
 (-\epsilon _1+2 \epsilon _3 , -2 \epsilon _1+2 \epsilon _3 , -3 \epsilon _1+2 \epsilon _3 )\\
 (-\epsilon _1+2 \epsilon _3 , -2 \epsilon _1+2 \epsilon _3 , -\epsilon _1+\epsilon _3 )
\end{array}$\\\hline
\end{tabular}
\eea

\paragraph{JK residues}
The framing node contribution is
\bea
\mathcal{Z}^{\D6_{\bar{4}}\tbar\D2\tbar\D0}_{\DT;\,\yng(3)\,\yng(1)\,\varnothing}(\fra,\phi_{I})&=\frac{\text{sh}\left(\fra+\epsilon _1+\epsilon _2+\epsilon _3-\phi_I\right) \text{sh}\left(\fra+\epsilon _2+3 \epsilon _3-\phi_I\right)
   }{\text{sh}\left(\fra-\phi_I\right) \text{sh}\left(\fra-\epsilon _1+2 \epsilon _3-\phi
   _I\right) }\\
   &\times \frac{\text{sh}\left(-\fra+\epsilon _1+\phi_I\right) \text{sh}\left(-\fra+\epsilon _1+\epsilon _2-2 \epsilon _3+\phi_I\right)
   \text{sh}\left(-\fra+\epsilon _3+\phi_I\right)}{\text{sh}\left(-\fra-\epsilon _1-\epsilon _2+\phi_I\right) \text{sh}\left(-\fra-3 \epsilon _3+\phi_I\right)
   \text{sh}\left(-\fra-\epsilon _2-\epsilon _3+\phi_I\right)}
\eea
For level one, we have
\bea
(0),&\quad -\frac{\text{sh}\left(\epsilon _1+\epsilon _3\right) \text{sh}\left(-\epsilon _1-\epsilon _2+2 \epsilon _3\right)
   \text{sh}\left(\epsilon _2+3 \epsilon _3\right)}{\text{sh}\left(\epsilon _2\right) \text{sh}\left(3 \epsilon _3\right)
   \text{sh}\left(2 \epsilon _3-\epsilon _1\right)}\\
   (\eps_1-2\eps_3),&\quad \frac{\text{sh}\left(\epsilon _1+\epsilon _2\right) \text{sh}\left(2 \epsilon _3\right) \text{sh}\left(-2 \epsilon _1-\epsilon
   _2+\epsilon _3\right) \text{sh}\left(\epsilon _2+\epsilon _3\right) \text{sh}\left(3 \epsilon _3-\epsilon
   _1\right)}{\text{sh}\left(\epsilon _1\right) \text{sh}\left(\epsilon _3\right) \text{sh}\left(-\epsilon _1-\epsilon _2+\epsilon
   _3\right) \text{sh}\left(2 \epsilon _3-\epsilon _1\right) \text{sh}\left(-2 \epsilon _1-\epsilon _2+2 \epsilon _3\right)}
\eea
For level two, we have
\bea
(0 , -\epsilon _2),&\quad \frac{\text{sh}\left(\epsilon _1+\epsilon _3\right) \text{sh}\left(\epsilon _1-\epsilon _2+\epsilon _3\right) \text{sh}\left(-\epsilon
   _1-\epsilon _2+2 \epsilon _3\right) \text{sh}\left(2 \epsilon _2+3 \epsilon _3\right)}{2 \text{sh}\left(\epsilon _2\right)
   \text{sh}\left(2 \epsilon _2\right) \text{sh}\left(3 \epsilon _3\right) \text{sh}\left(-\epsilon _1+\epsilon _2+2 \epsilon
   _3\right)} \\
  (0,-\epsilon_1+2\epsilon_3),&\quad -\frac{\text{sh}\left(\epsilon _1+\epsilon _2\right) \text{sh}\left(\epsilon _3-2 \epsilon _1\right) \text{sh}\left(\epsilon
   _1+\epsilon _3\right) \text{sh}\left(\epsilon _2+\epsilon _3\right) \text{sh}\left(\epsilon _2+2 \epsilon _3\right)
   \text{sh}\left(-\epsilon _1+\epsilon _2+3 \epsilon _3\right)}{2 \text{sh}\left(\epsilon _1\right) \text{sh}\left(\epsilon _2\right)
   \text{sh}\left(\epsilon _3\right) \text{sh}\left(\epsilon _3-\epsilon _1\right) \text{sh}\left(2 \epsilon _3-2 \epsilon _1\right)
   \text{sh}\left(-\epsilon _1+\epsilon _2+2 \epsilon _3\right)} \\
  (-\epsilon_1+2\epsilon_3,-2\epsilon_1+2\epsilon_3),&\quad \frac{\text{sh}\left(\epsilon _2+\epsilon
   _3\right) \text{sh}\left(-\epsilon _1+\epsilon _2+\epsilon _3\right) \text{sh}\left(3 \epsilon _3-2 \epsilon _1\right)
   \text{sh}\left(3 \epsilon _3-\epsilon _1\right)}{2  \text{sh}\left(2 \epsilon _3-2 \epsilon _1\right) \text{sh}\left(-3 \epsilon _1-\epsilon _2+2
   \epsilon _3\right) \text{sh}\left(-2 \epsilon _1-\epsilon _2+2 \epsilon _3\right)} \\
&\quad \times \frac{\text{sh}\left(\epsilon _1+\epsilon _2\right) \text{sh}\left(2 \epsilon _1+\epsilon _2\right) \text{sh}\left(2
   \epsilon _3\right) \text{sh}\left(-3 \epsilon _1-\epsilon _2+\epsilon _3\right) }{\text{sh}\left(\epsilon _1\right) \text{sh}\left(2 \epsilon _1\right)
   \text{sh}\left(\epsilon _3\right) \text{sh}\left(\epsilon _3-\epsilon _1\right) \text{sh}\left(-\epsilon _1-\epsilon
   _2+\epsilon _3\right)}\\
  (-\epsilon_1+2\epsilon_3,-\epsilon_1+\epsilon_3),&\quad  \frac{\text{sh}\left(2 \epsilon _1+\epsilon _2\right) \text{sh}\left(\epsilon _2+\epsilon _3\right) \text{sh}\left(\epsilon
   _2+2 \epsilon _3\right) \text{sh}\left(3 \epsilon _3-\epsilon _1\right)}{2 \text{sh}\left(\epsilon _1\right)
   \text{sh}\left(\epsilon _3-\epsilon _1\right){}^2 \text{sh}\left(-2 \epsilon _1-\epsilon _2+2 \epsilon _3\right)}
\eea

\paragraph{Unrefined vertex}
\bea\Yboxdim{4pt}
\wtC_{\,{\yng(1,1,1)}\,\,\yng(1)\,\varnothing}(q)=\frac{1-q^3+q^4}{(1-q)^{2}(1-q^{2})(1-q^3)}
\eea

\paragraph{Refined vertex}
\bea
&\wtC_{\,{\yng(1,1,1)}\,\,\yng(1)\,\varnothing}(t,q)=\frac{1+t-t^3-q+qt^3}{(1-q)(1-t)(1-t^2)(1-t^3)},\quad \wtC_{\,{\varnothing\,\yng(1,1,1)}\,\,\yng(1)\,}(t,q)= \frac{q^3 t-q^3+1}{(1-q)(1-q^2)(1-q^3)(1-t)},\\
&\wtC_{\yng(1)\,\varnothing\,{\yng(1,1,1)}\,\,}(t,q)=\frac{1-t^3+qt^3}{(1-t)^2(1-t^2)(1-t^3)}
\eea

\paragraph{Macdonald refined vertex}

\bea\Yboxdim{4pt}
\wtM_{\adjustbox{valign=c}{\yng(1)}\,\varnothing
\,\adjustbox{valign=c}{\yng(1,1,1)}}(x,y;q,t)&=\frac{(tx;q)_{\infty}(tx^2;q)_{\infty}(tx^3;q)_{\infty}}{(x;q)_{\infty}(x^2;q)_{\infty}(x^3;q)_{\infty}}\frac{1-x^3+x^3y}{1-x},\\
\wtM_{\varnothing\,\adjustbox{valign=c}{\yng(1,1,1)}\,\adjustbox{valign=c}{\yng(1)}}(x,y;q,t)&=\frac{(tx;q)_{\infty}}{(x;q)_{\infty}}\frac{1-y^3+xy^3}{(1-y)(1-y^2)(1-y^3)}, \\
\wtM_{{\yng(1,1,1)}\,\yng(1)\,\varnothing\,}(x,y;q,t)&=\frac{\#}{(1-x)(1-x^2)(1-x^3)(1-y)(1-q^3)(1-qt)}.
\eea
where
\bea
\#&=q^4tx^3y - q^4x^4y + q^3tx^4y - q^4tx^3 + q^4x^4 \\
&- q^3tx^4 - q^3tx^3 + q^2t^2x^3 + q^3x^4 - q^2tx^4 \\
&- q^3x^3y + q^3t^2x - q^3tx^2 + q^2t^2x^2 + q^3x^3 - q^2tx^3 + qt^2x^3\\
&- qtx^4 + t^2x^4 - q^4ty + q^4xy - q^3txy - qtx^3y + qx^4y - tx^4y \\
&+ q^4t - q^4x + q^3tx - q^2tx^2 + qt^2x^2 + q^2x^3 - qx^4 + tx^4 \\
&- q^2tx + q^2x^2 - qtx^2 + qx^3 - tx^3 + q^3y + x^3y - q^3 \\
&- qtx + qx^2 - tx^2 - x^3 + qty - qxy + txy - qt + qx - tx + x - y + 1
\eea

\subsubsection{$\lambda=\{2,1\}, \mu=\{1\}$ (Fig.~\ref{fig:PT3-2leg-ex4})}
\begin{figure}[ht]
    \centering
    \includegraphics[width=5cm]{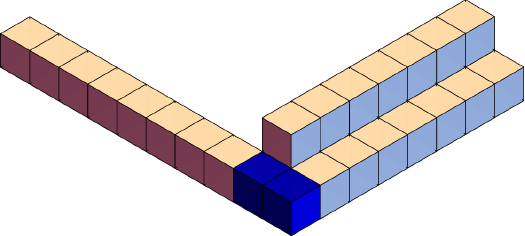}
    \caption{The PT box counting setup for $(\lambda,\mu,\nu)=(\{2,1\},\{1\},\varnothing)$}
    \label{fig:PT3-2leg-ex4}
\end{figure}
\paragraph{Pole structure}
\bea
\begin{tabular}{|c|c|}\hline
   \text{instanton number $k$}  &  \text{poles} \\\hline\hline
    1 & $ (\epsilon_2), (-\epsilon_1+\epsilon_3)$   \\\hline
    2 &  $
\begin{array}{c}
 (0 , \epsilon _2) \\
  (-\epsilon_1+\epsilon_2,\epsilon_2) \\
  (\epsilon_2, -\epsilon_1+\epsilon_3) \\
  (-\epsilon_1+\epsilon_3,-2\epsilon_1+\epsilon_3) 
\end{array}$ \\\hline
    3 & $\begin{array}{c}
( 0 , -\epsilon_1+\epsilon _2 , \epsilon _2) \\
 (0 , -\epsilon _2 , \epsilon _2) \\
 (0 , \epsilon _2 ,  -\epsilon _1+\epsilon _3) \\
 (-\epsilon _1+ \epsilon _2 , -2 \epsilon _1+ \epsilon _2 ,  \epsilon _2 )\\
 (-\epsilon _1+ \epsilon _2 ,  \epsilon _2 , -\epsilon _1+\epsilon _3 )\\
 ( \epsilon _2 , -\epsilon_1+\epsilon _3 , -2\epsilon _1+\epsilon _3 )\\
 (-\epsilon _1+ \epsilon _3 , -2\epsilon_1+ \epsilon _3 , -3\epsilon _1+\epsilon _3 )
\end{array}$\\\hline
\end{tabular}
\eea

\paragraph{JK residues}
The framing node contribution is
\bea
\mathcal{Z}^{\D6_{\bar{4}}\tbar\D2\tbar\D0}_{\DT;\,\yng(2,1)\,\yng(1)\,\varnothing}(\fra,\phi_{I})&=\frac{ \text{sh}\left(-\fra+\epsilon _1+\phi_I\right) \text{sh}\left(-\fra+\epsilon _1+\epsilon _2-\epsilon _3+\phi
   _I\right) \text{sh}\left(-\fra-\epsilon _2+\epsilon _3+\phi_I\right)}{\text{sh}\left(\fra+\epsilon _2-\phi_I\right)
   \text{sh}\left(\fra-\epsilon _1+\epsilon _3-\phi_I\right)}\\
   &\times \frac{\text{sh}\left(\fra+\epsilon _1+2 \epsilon _2+\epsilon _3-\phi_I\right) \text{sh}\left(\fra+\epsilon _2+2 \epsilon _3-\phi
   _I\right)}{ \text{sh}\left(-\fra-\epsilon _1-2 \epsilon _2+\phi_I\right)
   \text{sh}\left(-\fra-2 \epsilon _3+\phi_I\right) \text{sh}\left(-\fra-\epsilon _2-\epsilon _3+\phi_I\right)}.
\eea
For level one, we have
\bea
(\eps_2),&\quad -\frac{\text{sh}\left(\epsilon _1+\epsilon _2\right) \text{sh}\left(2 \epsilon _3\right) \text{sh}\left(\epsilon _1+\epsilon
   _3\right) \text{sh}\left(-\epsilon _1-2 \epsilon _2+\epsilon _3\right) \text{sh}\left(\epsilon _2+\epsilon
   _3\right)}{\text{sh}\left(\epsilon _1\right) \text{sh}\left(\epsilon _2\right) \text{sh}\left(\epsilon _3\right)
   \text{sh}\left(-\epsilon _1-\epsilon _2+\epsilon _3\right) \text{sh}\left(2 \epsilon _3-\epsilon _2\right)}\\
(\eps_3-\eps_1),&\quad \frac{\text{sh}\left(2 \epsilon _1+2 \epsilon _2\right) \text{sh}\left(\epsilon _2+\epsilon _3\right)
   \text{sh}\left(-\epsilon _1-\epsilon _2+2 \epsilon _3\right)}{\text{sh}\left(\epsilon _1\right) \text{sh}\left(-2
   \epsilon _1-2 \epsilon _2+\epsilon _3\right) \text{sh}\left(-\epsilon _1-\epsilon _2+\epsilon _3\right)}.
\eea
For level two, we have
\bea
 (0 , \epsilon _2),&\quad -\frac{\text{sh}\left(\epsilon _1+\epsilon _2\right) \text{sh}\left(\epsilon _1+\epsilon _3\right) \text{sh}\left(-\epsilon
   _1-2 \epsilon _2+\epsilon _3\right) \text{sh}\left(\epsilon _1-\epsilon _2+\epsilon _3\right) \text{sh}\left(2 \epsilon
   _2+\epsilon _3\right) \text{sh}\left(\epsilon _2+2 \epsilon _3\right)}{2 \text{sh}\left(\epsilon _2\right)
   \text{sh}\left(2 \epsilon _2\right) \text{sh}\left(\epsilon _2-\epsilon _1\right) \text{sh}\left(\epsilon _3\right)
   \text{sh}\left(\epsilon _3-\epsilon _1\right) \text{sh}\left(2 \epsilon _3-\epsilon _2\right)} \\
  (-\epsilon_1+\epsilon_2,\epsilon_2),&\quad \frac{\text{sh}\left(-\epsilon _1-2 \epsilon _2+\epsilon
   _3\right) \text{sh}\left(\epsilon _2+\epsilon _3\right) \text{sh}\left(-\epsilon _1+\epsilon _2+\epsilon _3\right)
   \text{sh}\left(\epsilon _1+2 \epsilon _3\right)}{2 \text{sh}\left(\epsilon _3-\epsilon
   _2\right) \text{sh}\left(-\epsilon _1-\epsilon _2+\epsilon _3\right) \text{sh}\left(2 \epsilon _3-\epsilon _2\right)
   \text{sh}\left(\epsilon _1-\epsilon _2+2 \epsilon _3\right)} \\
   &\quad \times \frac{\text{sh}\left(\epsilon _1+\epsilon _2\right) \text{sh}\left(2 \epsilon _3\right) \text{sh}\left(2 \epsilon
   _1+\epsilon _3\right) \text{sh}\left(\epsilon _3-2 \epsilon _2\right) }{\text{sh}\left(\epsilon _1\right) \text{sh}\left(2 \epsilon _1\right)
   \text{sh}\left(\epsilon _2-\epsilon _1\right) \text{sh}\left(\epsilon _3\right) }\\
  (\epsilon_2, -\epsilon_1+\epsilon_3),&\quad -\frac{\text{sh}\left(2 \epsilon _1+\epsilon _2\right) \text{sh}\left(\epsilon _1+2 \epsilon _2\right)
   \text{sh}\left(\epsilon _1+\epsilon _3\right) \text{sh}\left(\epsilon _2+\epsilon _3\right){}^2 \text{sh}\left(2 \epsilon
   _3-\epsilon _1\right)}{2 \text{sh}\left(\epsilon _1\right){}^2 \text{sh}\left(\epsilon _2\right) \text{sh}\left(\epsilon
   _3-\epsilon _1\right) \text{sh}\left(\epsilon _3-\epsilon _2\right) \text{sh}\left(-2 \epsilon _1-\epsilon _2+\epsilon
   _3\right)} \\
  (-\epsilon_1+\epsilon_3,-2\epsilon_1+\epsilon_3) ,&\quad \frac{
   \text{sh}\left(-2 \epsilon _1-\epsilon _2+2 \epsilon _3\right) \text{sh}\left(-\epsilon _1-\epsilon _2+2 \epsilon
   _3\right)}{2 \text{sh}\left(\epsilon _1\right) \text{sh}\left(2 \epsilon _1\right) }\\
   &\times \frac{\text{sh}\left(2 \epsilon _1+2 \epsilon _2\right) \text{sh}\left(3 \epsilon _1+2 \epsilon _2\right)
   \text{sh}\left(\epsilon _2+\epsilon _3\right) \text{sh}\left(-\epsilon _1+\epsilon _2+\epsilon _3\right)}{\text{sh}\left(-3 \epsilon _1-2
   \epsilon _2+\epsilon _3\right) \text{sh}\left(-2 \epsilon _1-2 \epsilon _2+\epsilon _3\right) \text{sh}\left(-2 \epsilon
   _1-\epsilon _2+\epsilon _3\right) \text{sh}\left(-\epsilon _1-\epsilon _2+\epsilon _3\right)}
\eea

\paragraph{Unrefined vertex}
\bea\Yboxdim{4pt}
\wtC_{\,{\yng(2,1)}\,\,\yng(1)\,\varnothing}(q)=\frac{1-q+q^{2}-q^3+q^4}{(1-q)^{3}(1-q^{3})}
\eea
\paragraph{Refined vertex}
\bea
&\wtC_{\,{\yng(2,1)}\,\,\yng(1)\,\varnothing}(t,q)=\frac{1+t^2-t^3-q+qt^3}{(1-q)(1-t)^{2}(1-t^{3})},\quad \wtC_{\varnothing \,{\yng(2,1)}\,\,\yng(1)}(t,q)=\frac{1-q-q^3+q^4+qt-q^4t+q^3t^2}{(1-t)(1-q)^{2}(1-q^3)}\\
&\wtC_{\,{\yng(1)}\,\varnothing \,\yng(2,1)}(t,q)=\frac{1-t+qt-qt^2+q^2t^2}{(1-t)^3 (1-qt^2)}
\eea

\paragraph{Macdonald refined vertex}

\bea\Yboxdim{4pt}
\wtM_{\adjustbox{valign=c}{\yng(1)}\,\varnothing
\,\adjustbox{valign=c}{\yng(2,1)}}(x,y;q,t)&=\frac{(tx;q)^{2}_{\infty}(tx^2y;q)_{\infty}}{(x;q)^{2}_{\infty}(x^2 y;q)_{\infty}}\frac{1-x+xy-x^2y+x^2y^2}{1-x},\\
\wtM_{\varnothing\,\adjustbox{valign=c}{\yng(2,1)}\,\adjustbox{valign=c}{\yng(1)}}(x,y;q,t)&=\frac{(tx;q)_{\infty}}{(x;q)_{\infty}}\frac{\#_{1}}{(1-y)(1-y^2)(1-y^3)(1-qt^2)}, \\
\wtM_{{\yng(2,1)}\,\yng(1)\,\varnothing\,}(x,y;q,t)&=\frac{\#_2}{(1-x)(1-x^2)(1-x^3)(1-y)(1-q)(1-qt^2)}
\eea
where 
\bea
\#_{1}&=-qt^2x^2y^4 + qt^2xy^5 + qtx^2y^5 - t^2x^2y^5 - qt^2x^2y^3 + qt^2xy^4 - qt^2y^5 - qtxy^5 + t^2xy^5\\
&+ qx^2y^5 - tx^2y^5 - qxy^5 + txy^5 - qt^2xy^2 + qt^2y^3 + x^2y^4 - xy^5 - qt^2xy + qt^2y^2 + qtxy^2 \\
&- t^2xy^2 + x^2y^3 - xy^4 + y^5 + qxy^2 - txy^2 - qt^2 + xy^2 - y^3 + xy - y^2 + 1,\\
\#_{2}&=q^2t^2x^4y - q^2t^2x^4 + q^2t^2x^3y - q^2tx^4y - q^2t^2x^3 + qt^3x^3 + q^2tx^4 - qt^2x^4 + t^3x^4 - qt^2x^3y \\
&- q^2x^4y + 2qtx^4y - t^2x^4y + qt^3x^2 + q^2x^4 - 2qtx^4 + t^2x^4 - q^2t^2xy - tx^4y + q^2t^2x - qt^2x^2 + qx^4 - q^2t^2y \\
&+ q^2txy - qx^3y + x^4y + q^2t^2 - q^2tx + qx^3 - tx^3 - x^4 + qt^2y + q^2xy - 2qtxy\\ 
&+ t^2xy + x^3y - qt^2 - q^2x + 2qtx - t^2x - tx^2 + txy - tx + x^2 + qy - xy - q + x - y + 1
\eea

\subsubsection{$\lambda=\{1,1,1\}, \mu=\{1\}$ (Fig.~\ref{fig:PT3-2leg-ex5})}
\begin{figure}[ht]
    \centering
    \includegraphics[width=5cm]{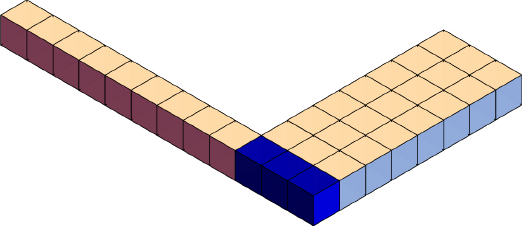}
    \caption{The PT box counting setup for $(\lambda,\mu,\nu)=(\{1,1,1\},\{1\},\varnothing)$}
    \label{fig:PT3-2leg-ex5}
\end{figure}
\paragraph{Pole structure}

\bea
\begin{tabular}{|c|c|}\hline \text{instanton number $k$} & \text{poles} \\\hline\hline 1 & $ (2\epsilon_2)$ \\\hline 2 & $ \begin{array}{c} (-\epsilon_1 + 2\epsilon_2, 2\epsilon_2) \\ (2\epsilon_2, \epsilon_2) \end{array}$ \\\hline 3 & $\begin{array}{c} ( 0 , 2\epsilon _2 , \epsilon _2) \\ (-\epsilon_1 + 2\epsilon_2, -2\epsilon_1 + \epsilon _2 , 2\epsilon _2) \\ (-\epsilon_1 + 2\epsilon_2 , 2\epsilon _2 , \epsilon _2) \end{array}$\\\hline \end{tabular}
\eea

\paragraph{JK residues}
The framing node contribution is
\bea
\mathcal{Z}^{\D6_{\bar{4}}\tbar\D2\tbar\D0}_{\DT;\,\yng(1,1,1)\,\yng(1)\,\varnothing}(\fra,\phi_{I})&=\frac{\text{sh}\left(\fra+\epsilon _1+3 \epsilon _2+\epsilon _3-\phi_I\right) \text{sh}\left(-\fra+\epsilon _1+\epsilon _2+\phi
   _I\right) \text{sh}\left(-\fra-2 \epsilon _2+\epsilon _3+\phi_I\right)}{\text{sh}\left(\fra+2 \epsilon _2-\phi_I\right)
   \text{sh}\left(-\fra-\epsilon _1-3 \epsilon _2+\phi_I\right) \text{sh}\left(-\fra-\epsilon _3+\phi_I\right)}.
\eea
For level one, we have
\bea
(2\eps_2),&\quad -\frac{\text{sh}\left(\epsilon _1+3 \epsilon _2\right) \text{sh}\left(\epsilon _1+\epsilon _3\right) \text{sh}\left(\epsilon
   _2+\epsilon _3\right)}{\text{sh}\left(\epsilon _1\right) \text{sh}\left(\epsilon _2\right) \text{sh}\left(\epsilon _3-2
   \epsilon _2\right)}.
\eea
For level two, we have
\bea
(-\epsilon_1 + 2\epsilon_2, 2\epsilon_2),&\quad \frac{\text{sh}\left(3 \epsilon _2\right) \text{sh}\left(\epsilon _1+3 \epsilon _2\right) \text{sh}\left(\epsilon
   _1+\epsilon _3\right) \text{sh}\left(2 \epsilon _1+\epsilon _3\right) \text{sh}\left(\epsilon _2+\epsilon _3\right)
   \text{sh}\left(-\epsilon _1+\epsilon _2+\epsilon _3\right)}{2 \text{sh}\left(\epsilon _1\right) \text{sh}\left(2 \epsilon
   _1\right) \text{sh}\left(\epsilon _2\right) \text{sh}\left(\epsilon _2-\epsilon _1\right) \text{sh}\left(\epsilon _3-2
   \epsilon _2\right) \text{sh}\left(\epsilon _1-2 \epsilon _2+\epsilon _3\right)} \\
(2\epsilon_2, \epsilon_2),&\quad -\frac{\text{sh}\left(\epsilon _1+2 \epsilon _2\right) \text{sh}\left(\epsilon _1+3 \epsilon _2\right)
   \text{sh}\left(\epsilon _1+\epsilon _3\right) \text{sh}\left(\epsilon _1-\epsilon _2+\epsilon _3\right)
   \text{sh}\left(\epsilon _2+\epsilon _3\right) \text{sh}\left(2 \epsilon _2+\epsilon _3\right)}{2 \text{sh}\left(\epsilon
   _1\right) \text{sh}\left(\epsilon _2\right) \text{sh}\left(2 \epsilon _2\right) \text{sh}\left(\epsilon _2-\epsilon
   _1\right) \text{sh}\left(\epsilon _3-2 \epsilon _2\right) \text{sh}\left(\epsilon _3-\epsilon _2\right)}.
\eea

\paragraph{Unrefined vertex}
\bea\Yboxdim{4pt}
\wtC_{\,{\yng(3)}\,\,\yng(1)\,\varnothing}(q)=\frac{1-q+q^4}{(1-q)^{2}(1-q^2)(1-q^{3})}
\eea
\paragraph{Refined vertex}
\bea\Yboxdim{4pt}
&\wtC_{\,{\yng(3)}\,\,\yng(1)\,\varnothing}(t,q)=\frac{1-q+qt^3}{(1-q)(1-t)(1-t^2)(1-t^3)},\quad \wtC_{\,\yng(1)\,\varnothing\,{\yng(3)}\,}(t,q)=\frac{1-t+q^3t}{(1-t)^2(1-qt)(1-q^2t)},\\
&\wtC_{\varnothing\,{\yng(3)}\,\,\yng(1)}(t,q)=\frac{q^6 t-q^6-q^5 t^2+q^5-q^4 t+q^4+q^3 t^3-q^3 t+q^2 t^2-q^2+q t-q+1}{(1-q)(1-q^2)(1-q^3)(1-t)}.
\eea

\paragraph{Macdonald refined vertex}

\bea\Yboxdim{4pt}
\wtM_{\adjustbox{valign=c}{\yng(1)}\,\varnothing
\,\adjustbox{valign=c}{\yng(3)}}(x,y;q,t)&=\frac{(tx;q)_{\infty}(txy;q)_{\infty}(txy^2;q)_{\infty}}{(x;q)_{\infty}(xy;q)_{\infty}(x^2 y;q)_{\infty}}\frac{1-x+xy^3}{1-x},\\
\wtM_{{\yng(3)}\,\yng(1)\,\varnothing\,}(x,y;q,t)&=\frac{1-y-q+qy+x^3y-tx^3+qx^3-qx^3y}{(1-x)(1-x^2)(1-x^3)(1-y)(1-q)},\\
\wtM_{\varnothing\,\adjustbox{valign=c}{\yng(3)}\,\adjustbox{valign=c}{\yng(1)}}(x,y;q,t)&=\frac{(tx;q)_{\infty}}{(x;q)_{\infty}}\frac{\#}{(1-y)(1-y^2)(1-y^3)(1-qt)(1-q^2t)}, 
\eea
where
\bea
\#&=-q^3t^2x^2y^5 - q^3tx^3y^5 + q^2t^2x^3y^5 + q^3t^2xy^6 + q^3tx^2y^6 - q^2t^2x^2y^6 + q^3x^3y^6 - q^2tx^3y^6\\
& + q^3t^2x^3y^3 - q^3tx^3y^4 + q^2t^2x^3y^4 + q^3tx^2y^5 - q^2t^2x^2y^5 - q^2tx^3y^5 + qt^2x^3y^5 - q^3t^2y^6\\
&- q^3txy^6 + q^2t^2xy^6 - q^3x^2y^6 + 2q^2tx^2y^6 - qt^2x^2y^6 - qtx^3y^6 + t^2x^3y^6 - q^3t^2xy^4 - q^2tx^3y^4 + qt^2x^3y^4\\
&+ q^3t^2y^5 + 2q^2tx^2y^5 - qt^2x^2y^5 + q^2x^3y^5 - qtx^3y^5- 2q^2txy^6 + qt^2xy^6 - q^2x^2y^6 + 2qtx^2y^6\\
&- t^2x^2y^6 + q^3t^2x^2y^2 - q^3t^2xy^3- q^3tx^2y^3 + q^2t^2x^2y^3 - q^2tx^3y^3 + q^3t^2y^4 + q^3txy^4 - q^2t^2xy^4 \\
&+ q^2x^3y^4 - qtx^3y^4 - q^2x^2y^5 + 2qtx^2y^5 + qx^3y^5 - tx^3y^5 + q^2ty^6 + q^2xy^6 - 2qtxy^6 - qx^2y^6\\
&+ tx^2y^6 - q^3tx^2y^2 + q^2t^2x^2y^2 + q^3txy^3 - q^2t^2xy^3 + q^3x^2y^3 - 2q^2tx^2y^3 + qt^2x^2y^3 - qtx^3y^3\\
&+ 2q^2txy^4 - qt^2xy^4 + qx^3y^4 - tx^3y^4 - q^2ty^5 - qx^2y^5 + tx^2y^5 + qty^6 + qxy^6 - txy^6 + q^3t^2xy\\
&- q^3t^2y^2 - 2q^2tx^2y^2 + qt^2x^2y^2 + 2q^2txy^3 - qt^2xy^3 + q^2x^2y^3 - 2qtx^2y^3 + t^2x^2y^3 - q^2ty^4 - q^2xy^4 \\
&+ 2qtxy^4 - qty^5 - x^2y^5 + xy^6 - q^3t^2y - q^3txy + q^2t^2xy + q^2x^2y^2 - 2qtx^2y^2 - q^2xy^3 \\
&+ 2qtxy^3 + qx^2y^3 - tx^2y^3 + x^3y^3 - qty^4 - qxy^4 + txy^4 - y^6 + q^3t^2 - 2q^2txy + qt^2xy\\
&+ q^2ty^2 + qx^2y^2 - tx^2y^2 - qxy^3 + txy^3 - xy^4 + y^5 + q^2ty \\
&+ q^2xy - 2qtxy + qty^2 + x^2y^2 - xy^3 + y^4 - q^2t + qty + qxy - txy - qt + xy - y^2 - y + 1
\eea

\subsubsection{$\lambda=\{2\}, \mu=\{2\}$ (Fig.~\ref{fig:PT3-2leg-ex6})}
\begin{figure}[ht]
    \centering
    \includegraphics[width=5cm]{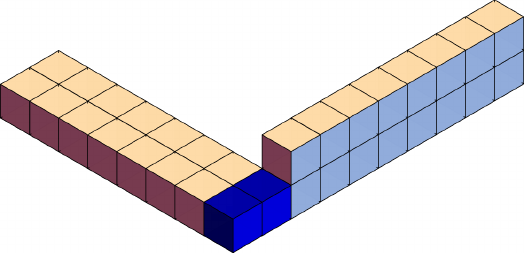}
    \caption{The PT box counting setup for $(\lambda,\mu,\nu)=(\{2\},\{2\},\varnothing)$}
    \label{fig:PT3-2leg-ex6}
\end{figure}
\paragraph{Pole structure}

\bea
\begin{tabular}{|c|c|}\hline \text{instanton number $k$} & \text{poles} \\\hline\hline 1 & $ (\epsilon_1), (-\epsilon_1+\epsilon_3)$ \\\hline 2 & $ \begin{array}{c} (0, \epsilon_1) \\ (\epsilon_1, \epsilon_1 - \epsilon_2) \\ (\epsilon_1, -\epsilon_1 + \epsilon_3) \\ (-\epsilon_1 + \epsilon_3, -2\epsilon_1 + \epsilon_3) \end{array}$ \\\hline 3 & $\begin{array}{c} ( 0 , \epsilon_1 , \epsilon_1 - \epsilon_2) \\ ( 0 , \epsilon_1 , -\epsilon_1 + \epsilon_3) \\ (\epsilon_1, \epsilon_1 - \epsilon_2 , \epsilon_1 - 2\epsilon_2) \\ (\epsilon_1 , \epsilon_1 - \epsilon_2 , -\epsilon_1 + \epsilon_3) \\ (\epsilon_1, -\epsilon_1 + \epsilon_3 , -2\epsilon_1 + \epsilon_3) \\ (-\epsilon_1 + \epsilon_3 , -2\epsilon_1 + \epsilon_3 , -3\epsilon_1 + \epsilon_3) \end{array}$\\\hline \end{tabular}
\eea

\paragraph{JK residues}
The framing node contribution is
\bea
\mathcal{Z}^{\D6_{\bar{4}}\tbar\D2\tbar\D0}_{\DT;\,\yng(2)\,\yng(2)\,\varnothing}(\fra,\phi_{I})&=\frac{\text{sh}\left(\fra+2 \epsilon _1+\epsilon _2+\epsilon _3-\phi_I\right) \text{sh}\left(\fra+\epsilon _2+2 \epsilon _3-\phi
   _I\right) }{\text{sh}\left(\fra+\epsilon _1-\phi_I\right)
   \text{sh}\left(\fra-\epsilon _1+\epsilon _3-\phi_I\right)}\\
   &\times \frac{\text{sh}\left(-\fra+\epsilon _1+\phi_I\right) \text{sh}\left(-\fra+\epsilon _1+\epsilon _2-\epsilon _3+\phi
   _I\right) \text{sh}\left(-\fra-\epsilon _1+\epsilon _3+\phi_I\right)}{ \text{sh}\left(-\fra-2 \epsilon _1-\epsilon _2+\phi_I\right)
   \text{sh}\left(-\fra-2 \epsilon _3+\phi_I\right) \text{sh}\left(-\fra-\epsilon _2-\epsilon _3+\phi_I\right)}
\eea
For level one, we have
\bea
(\eps_1),&\quad -\frac{\text{sh}\left(2 \epsilon _1\right) \text{sh}\left(\epsilon _1+\epsilon _3\right) \text{sh}\left(-2 \epsilon
   _1-\epsilon _2+\epsilon _3\right) \text{sh}\left(\epsilon _2+\epsilon _3\right) \text{sh}\left(-\epsilon _1+\epsilon _2+2
   \epsilon _3\right)}{\text{sh}\left(\epsilon _1\right) \text{sh}\left(\epsilon _2\right) \text{sh}\left(\epsilon _3-2
   \epsilon _1\right) \text{sh}\left(-\epsilon _1+\epsilon _2+\epsilon _3\right) \text{sh}\left(2 \epsilon _3-\epsilon
   _1\right)}\\
(\eps_3-\eps_1),&\quad \frac{\text{sh}\left(3 \epsilon _1+\epsilon _2\right) \text{sh}\left(\epsilon _2+\epsilon _3\right) \text{sh}\left(2
   \epsilon _3-2 \epsilon _1\right)}{\text{sh}\left(\epsilon _1\right) \text{sh}\left(\epsilon _3-2 \epsilon _1\right)
   \text{sh}\left(-3 \epsilon _1-\epsilon _2+\epsilon _3\right)}
\eea
For level two, we have{\small
\bea
 (0, \epsilon_1),&\quad \frac{\text{sh}\left(\epsilon _1+\epsilon _3\right) \text{sh}\left(2 \epsilon _1+\epsilon _3\right) \text{sh}\left(-2
   \epsilon _1-\epsilon _2+\epsilon _3\right) \text{sh}\left(-\epsilon _1-\epsilon _2+\epsilon _3\right)
   \text{sh}\left(\epsilon _2+2 \epsilon _3\right) \text{sh}\left(-\epsilon _1+\epsilon _2+2 \epsilon _3\right)}{2
   \text{sh}\left(\epsilon _2\right) \text{sh}\left(\epsilon _2-\epsilon _1\right) \text{sh}\left(2 \epsilon _3\right)
   \text{sh}\left(\epsilon _3-2 \epsilon _1\right) \text{sh}\left(\epsilon _3-\epsilon _1\right) \text{sh}\left(2 \epsilon
   _3-\epsilon _1\right)} \\ 
 (\epsilon_1, \epsilon_1 - \epsilon_2),&\quad \frac{\text{sh}\left(-2 \epsilon _1-\epsilon _2+\epsilon _3\right) \text{sh}\left(\epsilon _1-\epsilon
   _2+\epsilon _3\right) \text{sh}\left(\epsilon _2+\epsilon _3\right) \text{sh}\left(2 \epsilon _2+\epsilon _3\right)
   \text{sh}\left(-\epsilon _1+2 \epsilon _2+2 \epsilon _3\right)}{2 \text{sh}\left(\epsilon _1\right)
   \text{sh}\left(\epsilon _2\right) \text{sh}\left(2 \epsilon _2\right) } \\
   &\quad \times \frac{\text{sh}\left(2 \epsilon _1\right) \text{sh}\left(\epsilon _2-2 \epsilon _1\right) \text{sh}\left(\epsilon
   _1+\epsilon _3\right) }{\text{sh}\left(\epsilon _2-\epsilon _1\right)
   \text{sh}\left(-2 \epsilon _1+\epsilon _2+\epsilon _3\right) \text{sh}\left(-\epsilon _1+\epsilon _2+\epsilon _3\right)
   \text{sh}\left(-\epsilon _1+2 \epsilon _2+\epsilon _3\right) \text{sh}\left(2 \epsilon _3-\epsilon _1\right)}  \\
 (\epsilon_1, -\epsilon_1 + \epsilon_3),&\quad -\frac{\text{sh}\left(3 \epsilon _1\right) \text{sh}\left(2 \epsilon _1+\epsilon _2\right) \text{sh}\left(\epsilon
   _1+\epsilon _3\right) \text{sh}\left(\epsilon _2+\epsilon _3\right){}^2 \text{sh}\left(-2 \epsilon _1+\epsilon _2+2
   \epsilon _3\right)}{2 \text{sh}\left(\epsilon _1\right){}^2 \text{sh}\left(\epsilon _2\right) \text{sh}\left(\epsilon
   _3-3 \epsilon _1\right) \text{sh}\left(\epsilon _3-\epsilon _1\right) \text{sh}\left(-2 \epsilon _1+\epsilon _2+\epsilon
   _3\right)} \\ 
 (-\epsilon_1 + \epsilon_3, -2\epsilon_1 + \epsilon_3),&\quad \frac{\text{sh}\left(3 \epsilon _1+\epsilon _2\right) \text{sh}\left(4 \epsilon _1+\epsilon _2\right)
   \text{sh}\left(\epsilon _2+\epsilon _3\right) \text{sh}\left(-\epsilon _1+\epsilon _2+\epsilon _3\right) \text{sh}\left(2
   \epsilon _3-3 \epsilon _1\right) \text{sh}\left(2 \epsilon _3-2 \epsilon _1\right)}{2 \text{sh}\left(\epsilon _1\right)
   \text{sh}\left(2 \epsilon _1\right) \text{sh}\left(\epsilon _3-3 \epsilon _1\right) \text{sh}\left(\epsilon _3-2 \epsilon
   _1\right) \text{sh}\left(-4 \epsilon _1-\epsilon _2+\epsilon _3\right) \text{sh}\left(-3 \epsilon _1-\epsilon _2+\epsilon
   _3\right)}
\eea}

\paragraph{Unrefined vertex}
\bea\Yboxdim{4pt}
\wtC_{\,{\yng(1,1)}\,\,\yng(1,1)\,\varnothing}(q)=\frac{1-q^2+q^4}{(1-q)^{2}(1-q^2)^{2}}
\eea
\paragraph{Refined vertex}
\bea\Yboxdim{4pt}
&\wtC_{\,{\yng(1,1)}\,\,\yng(1,1)\,\varnothing}(t,q)=\frac{1-q^2+qt-t^2+q^2t^2}{(1-q)(1-q^2)(1-t)(1-t^2)},\quad \wtC_{\varnothing\,{\yng(1,1)}\,\,\yng(1,1)\,}(t,q)=\frac{1-q^2+q^2t^2}{(1-q)(1-q^2)(1-t)(1-t^2)},\\
&\wtC_{\,{\yng(1,1)}\,\varnothing\,\yng(1,1)\,}(t,q)=\frac{q^2 t^4-q t^5-q t^4+q t^3+q t^2+t^5-t^3-t^2+1}{(1-t)^{2}(1-t^{2})^{2}}
\eea

\paragraph{Macdonald refined vertex}
\bea\Yboxdim{4pt}
\wtM_{\adjustbox{valign=c}{\yng(1,1)}\,\varnothing
\,\adjustbox{valign=c}{\yng(1,1)}}(x,y;q,t)&=\frac{(tx;q)_{\infty}(tx^2;q)_{\infty}}{(x;q)_{\infty}(x^2 ;q)_{\infty}}\frac{\#}{(1-x)(1-x^2)(1-qt)},\\
\wtM_{{\yng(1,1)}\,\yng(1,1)\,\varnothing\,}(x,y;q,t)&=\frac{1-y^2+xy-x^2-q^2+x^2y^2+q^2y^2-tx^2y+qx^2y-qtxy+q^2x^2-q^2x^2y^2}{(1-x)(1-x^2)(1-y)(1-y^2)(1-q^2)},\\
\wtM_{\varnothing\,\adjustbox{valign=c}{\yng(1,1)}\,\adjustbox{valign=c}{\yng(1,1)}}(x,y;q,t)&=\frac{(tx;q)_{\infty}(tx^2;q)_{\infty}}{(x;q)_{\infty}(x^2 ;q)_{\infty}}\frac{1-y^2+x^2y^2}{(1-y)(1-y^2)}, 
\eea
where
\bea
\#&=qtx^5y - qtx^4y^2 + qx^5y^2 - tx^5y^2 - qtx^5 + qtx^4y - qx^5y + tx^5y- qtx^3y - qx^4y \\
&+ tx^4y - x^5y + x^4y^2 + qtx^3 + qx^4 - tx^4+ x^5 - qtx^2y + qx^3y - tx^3y - x^4y + qtx^2\\
&- qx^3 + tx^3 + qx^2y - tx^2y + x^3y - qx^2 + tx^2 - x^3 + x^2y - qt + qx - tx - x^2 + 1.
\eea

\subsubsection{$\lambda=\{2\}, \mu=\{1,1\}$ (Fig.~\ref{fig:PT3-2leg-ex7})}
\begin{figure}[ht]
    \centering
    \includegraphics[width=5cm]{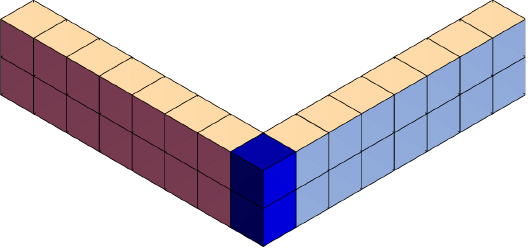}
    \caption{The PT box counting setup for $(\lambda,\mu,\nu)=(\{2\},\{1,1\},\varnothing)$}
    \label{fig:PT3-2leg-ex7}
\end{figure}
\paragraph{Pole structure}

\bea
\begin{tabular}{|c|c|}\hline \text{instanton number $k$} & \text{poles} \\\hline\hline 1 & $ (\epsilon_3)$ \\\hline 2 & $ \begin{array}{c} (0, \epsilon_3) \\ (-\epsilon_1 + \epsilon_3, \epsilon_3) \\ (-\epsilon_2 + \epsilon_3, \epsilon_3) \end{array}$ \\\hline 3 & $\begin{array}{c} ( 0 , -\epsilon_1 + \epsilon_3 , \epsilon_3) \\ ( 0 , -\epsilon_2 + \epsilon_3 , \epsilon_3) \\ (-\epsilon_1 + \epsilon_3 , \epsilon_3 - 2\epsilon_1 , \epsilon_3) \\ (-\epsilon_1 + \epsilon_3 , \epsilon_3 - \epsilon_2 , \epsilon_3) \\ (-\epsilon_2 + \epsilon_3 , -2\epsilon_2 + \epsilon_3 , \epsilon_3) \end{array}$\\\hline \end{tabular}
\eea

\paragraph{JK residues}
The framing node contribution is
\bea
\mathcal{Z}^{\D6_{\bar{4}}\tbar\D2\tbar\D0}_{\DT;\,\yng(2)\,\yng(1,1)\,\varnothing}(\fra,\phi_{I})&=\frac{\text{sh}\left(\fra+\epsilon _1+\epsilon _2+2 \epsilon _3-\phi_I\right) \text{sh}\left(-\fra+\epsilon _1+\epsilon
   _2-\epsilon _3+\phi_I\right) \text{sh}\left(-\fra+\epsilon _3+\phi_I\right)}{\text{sh}\left(\fra+\epsilon _3-\phi_I\right)
   \text{sh}\left(-\fra-\epsilon _1-\epsilon _2+\phi_I\right) \text{sh}\left(-\fra-2 \epsilon _3+\phi_I\right)}
\eea
For level one, we have
\bea
(\eps_3),&\quad \frac{\text{sh}\left(\epsilon _1+\epsilon _2\right){}^2 \text{sh}\left(2 \epsilon _3\right) \text{sh}\left(\epsilon
   _1+\epsilon _3\right) \text{sh}\left(\epsilon _2+\epsilon _3\right)}{\text{sh}\left(\epsilon _1\right)
   \text{sh}\left(\epsilon _2\right) \text{sh}\left(\epsilon _3\right){}^2 \text{sh}\left(-\epsilon _1-\epsilon _2+\epsilon
   _3\right)}.
\eea
For level two, we have
\bea
(0, \epsilon_3),&\quad -\frac{\text{sh}\left(\epsilon _1+\epsilon _2\right) \text{sh}\left(\epsilon _1+\epsilon _3\right) \text{sh}\left(-\epsilon
   _1-\epsilon _2+\epsilon _3\right) \text{sh}\left(\epsilon _2+\epsilon _3\right) \text{sh}\left(\epsilon _1+2 \epsilon
   _3\right) \text{sh}\left(\epsilon _2+2 \epsilon _3\right)}{2 \text{sh}\left(\epsilon _1\right) \text{sh}\left(\epsilon
   _2\right) \text{sh}\left(\epsilon _3\right) \text{sh}\left(2 \epsilon _3\right) \text{sh}\left(\epsilon _3-\epsilon
   _1\right) \text{sh}\left(\epsilon _3-\epsilon _2\right)} \\
(-\epsilon_1 + \epsilon_3, \epsilon_3),&\quad  \frac{\text{sh}\left(\epsilon _1+\epsilon _2\right){}^2 \text{sh}\left(2 \epsilon _1+\epsilon _2\right) \text{sh}\left(2
   \epsilon _3\right) \text{sh}\left(2 \epsilon _1+\epsilon _3\right) \text{sh}\left(\epsilon _2+\epsilon _3\right)
   \text{sh}\left(-\epsilon _1+\epsilon _2+\epsilon _3\right) \text{sh}\left(2 \epsilon _3-\epsilon _1\right)}{2
   \text{sh}\left(\epsilon _1\right) \text{sh}\left(2 \epsilon _1\right) \text{sh}\left(\epsilon _2-\epsilon _1\right)
   \text{sh}\left(\epsilon _3\right){}^2 \text{sh}\left(\epsilon _3-\epsilon _1\right) \text{sh}\left(-2 \epsilon
   _1-\epsilon _2+\epsilon _3\right) \text{sh}\left(-\epsilon _1-\epsilon _2+\epsilon _3\right)}\\ 
(-\epsilon_2 + \epsilon_3, \epsilon_3) ,&\quad -\frac{\text{sh}\left(\epsilon _1+\epsilon _2\right){}^2 \text{sh}\left(\epsilon _1+2 \epsilon _2\right) \text{sh}\left(2
   \epsilon _3\right) \text{sh}\left(\epsilon _1+\epsilon _3\right) \text{sh}\left(\epsilon _1-\epsilon _2+\epsilon
   _3\right) \text{sh}\left(2 \epsilon _2+\epsilon _3\right) \text{sh}\left(2 \epsilon _3-\epsilon _2\right)}{2
   \text{sh}\left(\epsilon _2\right) \text{sh}\left(2 \epsilon _2\right) \text{sh}\left(\epsilon _2-\epsilon _1\right)
   \text{sh}\left(\epsilon _3\right){}^2 \text{sh}\left(-\epsilon _1-2 \epsilon _2+\epsilon _3\right)
   \text{sh}\left(\epsilon _3-\epsilon _2\right) \text{sh}\left(-\epsilon _1-\epsilon _2+\epsilon _3\right)}
\eea

\paragraph{Unrefined vertex}
\bea\Yboxdim{4pt}
\wtC_{\,{\yng(1,1)}\,\,\yng(2)\,\varnothing}(q)=\frac{1-q+2q^3-q^4-q^5+q^6}{(1-q)^{2}(1-q^2)^2}
\eea

\paragraph{Refined vertex}
\bea\Yboxdim{4pt}
&\wtC_{\,{\yng(1,1)}\,\,\yng(2)\,\varnothing}(t,q)=\frac{1-q-q^2+q^3+qt-q^3t+qt^2+q^2t^2-q^3t^2-qt^3+q^3t^3}{(1-q)(1-q^2)(1-t)(1-t^2)},\\
&\wtC_{\varnothing\,{\yng(1,1)}\,\,\yng(2)\,}(t,q)=\frac{1-q-q^2+q^3+qt+q^2t-q^3t-q^4t+q^4t^2}{(1-q)(1-q^2)(1-t)(1-qt)},\\
&\wtC_{\,{\yng(2)}\,\varnothing\,\yng(1,1)\,}(t,q)=\frac{1-t+qt-t^2+qt^2+t^3-qt^3-qt^4+q^2t^4}{(1-t)^{2}(1-t^{2})^{2}}
\eea

\paragraph{Macdonald refined vertex}

\bea\Yboxdim{4pt}
\wtM_{\adjustbox{valign=c}{\yng(2)}\,\varnothing
\,\adjustbox{valign=c}{\yng(1,1)}}(x,y;q,t)&=\frac{(tx;q)_{\infty}(tx^2;q)_{\infty}}{(x;q)_{\infty}(x^2 ;q)_{\infty}}\frac{1-x+xy-x^2+x^2y+x^3-x^3y-x^4y+x^4y^2}{(1-x)(1-x^2)},\\
\wtM_{{\yng(1,1)}\,\yng(2)\,\varnothing\,}(x,y;q,t)&=\frac{\#}{(1-x)(1-x^2)(1-y)(1-y^2)(1-q)(1-q^2)(1-qt)},\\
\wtM_{\varnothing\,\adjustbox{valign=c}{\yng(1,1)}\,\adjustbox{valign=c}{\yng(2)}}(x,y;q,t)&=\frac{(tx;q)_{\infty}(tx y;q)_{\infty}}{(x;q)_{\infty}(x y;q)_{\infty}}\frac{1-y-y^2+xy+y^3+xy^2-xy^3-xy^4+x^2y^4}{(1-y)(1-y^2)}, 
\eea
where
\bea
&\#=-q^4tx^3y^3 + q^4tx^3y^2 - q^3t^2x^3y^2 + q^4tx^2y^3 + q^3tx^3y^3 + q^4tx^3y - q^4tx^2y^2 + q^3tx^3y^2\\
&- q^2t^2x^3y^2 + q^4txy^3 - q^3tx^2y^3 + q^3x^3y^3 + q^2tx^3y^3 - q^4tx^3 + q^3t^2x^3 - q^4tx^2y - q^3tx^3y\\
&- q^4txy^2 + q^3t^2xy^2 + q^3tx^2y^2 - 2q^3x^3y^2 + q^2tx^3y^2 + qt^2x^3y^2 - q^4ty^3 - q^3txy^3 - q^3x^2y^3\\
&- q^2tx^2y^3 - q^2x^3y^3 - qtx^3y^3 + q^4tx^2 - q^2t^3x^2 - q^3tx^3 + 2q^2t^2x^3 - qt^3x^3 - q^4txy + q^3tx^2y \\
&+ q^2t^2x^2y - qt^3x^2y - q^3x^3y - 2q^2tx^3y + 2qt^2x^3y - t^3x^3y + q^4ty^2 - q^3txy^2 + q^2t^2xy^2 + q^3x^2y^2 \\
&+ q^2tx^2y^2 - qtx^3y^2 + t^2x^3y^2 + q^3ty^3 - q^3xy^3 - q^2txy^3 + q^2x^2y^3 + qtx^2y^3 - qx^3y^3 + q^4tx - q^3t^2x\\
&- q^3tx^2 + q^2t^2x^2 + 2q^3x^3 - 2q^2tx^3 + q^4ty + q^3txy + q^3x^2y + qt^2x^2y + 2q^2x^3y - qtx^3y + t^2x^3y \\
&- q^3ty^2+ 2q^3xy^2 - q^2txy^2 - qt^2xy^2 - q^2x^2y^2 - qtx^2y^2 + 2qx^3y^2 - 2tx^3y^2 + q^3y^3 + q^2ty^3 + q^2xy^3 \\
&+ qtxy^3+ qx^2y^3 + x^3y^3 - q^4t + q^3tx - q^2t^2x - q^3x^2 - q^2tx^2 + 2qt^2x^2 - q^3ty + q^3xy + q^2txy \\
&- q^2x^2y - 2qtx^2y + t^2x^2y + qx^3y - q^3y^2 - q^2ty^2 + qtxy^2 - t^2xy^2 - qx^2y^2 - q^2y^3 - qty^3 + qxy^3 \\
&- x^2y^3 + q^3t - 2q^3x + q^2tx+ qt^2x + q^2x^2 - qtx^2 - qx^3 + tx^3 - q^3y - q^2ty - q^2xy - qtxy - tx^2y - x^3y \\
&+ q^2y^2 + qty^2 - 2qxy^2 + 2txy^2+ x^2y^2- qy^3 - xy^3 + q^3 + q^2t - qtx + t^2x + qx^2 - tx^2 + q^2y \\
&+ qty - qxy + x^2y+ qy^2 + y^3 - q^2 - qt + 2qx - 2tx + qy + xy - y^2 - q - y + 1
\eea

\subsubsection{$\lambda=\{1,1\}, \mu=\{1,1\}$ (Fig.~\ref{fig:PT3-2leg-ex8})}
\begin{figure}[ht]
    \centering
    \includegraphics[width=5cm]{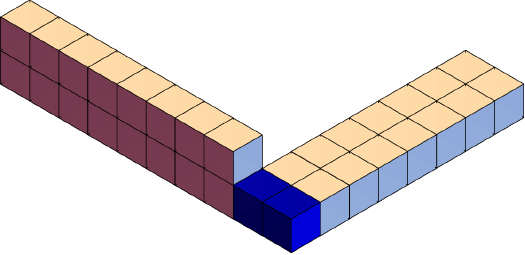}
    \caption{The PT box counting setup for $(\lambda,\mu,\nu)=(\{1,1\},\{1,1\},\varnothing)$}
    \label{fig:PT3-2leg-ex8}
\end{figure}
\paragraph{Pole structure}
\bea
\begin{tabular}{|c|c|}\hline \text{instanton number $k$} & \text{poles} \\\hline\hline 1 & $ (\epsilon_2), (-\epsilon_2+\epsilon_3)$ \\\hline 2 & $ \begin{array}{c} (0, \epsilon_2) \\ (-\epsilon_1+\epsilon_2, \epsilon_2) \\ (\epsilon_2, -\epsilon_2+\epsilon_3) \\ (-\epsilon_2+\epsilon_3, -2\epsilon_2+\epsilon_3) \end{array}$ \\\hline 3 & $\begin{array}{c} ( 0 , \epsilon_2 , \epsilon_3 - \epsilon_2) \\ ( 0 , \epsilon_2 , -\epsilon_1 + \epsilon_2) \\ (\epsilon_2, \epsilon_2 - \epsilon_1 , \epsilon_2 - 2\epsilon_1) \\ (\epsilon_2 , \epsilon_2 - \epsilon_1 , -\epsilon_2 + \epsilon_3) \\ (\epsilon_2, -\epsilon_2 + \epsilon_3 , -2\epsilon_2 + \epsilon_3) \\ (-\epsilon_2 + \epsilon_3 , -2\epsilon_2 + \epsilon_3 , -3\epsilon_2 + \epsilon_3) \end{array}$\\\hline \end{tabular}
\eea

\paragraph{JK residues}
The framing node contribution is
\bea
\mathcal{Z}^{\D6_{\bar{4}}\tbar\D2\tbar\D0}_{\DT;\,\yng(1,1)\,\yng(1,1)\,\varnothing}(\fra,\phi_{I})&=\frac{\text{sh}\left(\fra+\epsilon _1+2 \epsilon _2+\epsilon _3-\phi_I\right) \text{sh}\left(\fra+\epsilon _1+2 \epsilon _3-\phi
   _I\right) }{\text{sh}\left(\fra+\epsilon _2-\phi_I\right)
   \text{sh}\left(\fra-\epsilon _2+\epsilon _3-\phi_I\right)}\\
   &\times \frac{\text{sh}\left(-\fra+\epsilon _2+\phi_I\right) \text{sh}\left(-\fra+\epsilon _1+\epsilon _2-\epsilon _3+\phi
   _I\right) \text{sh}\left(-\fra-\epsilon _2+\epsilon _3+\phi_I\right)}{ \text{sh}\left(-\fra-\epsilon _1-2 \epsilon _2+\phi_I\right)
   \text{sh}\left(-\fra-2 \epsilon _3+\phi_I\right) \text{sh}\left(-\fra-\epsilon _1-\epsilon _3+\phi_I\right)}
\eea
For level one, we have
\bea
(\epsilon_2),&\quad -\frac{\text{sh}\left(2 \epsilon _2\right) \text{sh}\left(\epsilon _1+\epsilon _3\right) \text{sh}\left(-\epsilon _1-2
   \epsilon _2+\epsilon _3\right) \text{sh}\left(\epsilon _2+\epsilon _3\right) \text{sh}\left(\epsilon _1-\epsilon _2+2
   \epsilon _3\right)}{\text{sh}\left(\epsilon _1\right) \text{sh}\left(\epsilon _2\right) \text{sh}\left(\epsilon _3-2
   \epsilon _2\right) \text{sh}\left(\epsilon _1-\epsilon _2+\epsilon _3\right) \text{sh}\left(2 \epsilon _3-\epsilon
   _2\right)} \\
(-\epsilon_2+\epsilon_3),&\quad \frac{\text{sh}\left(\epsilon _1+3 \epsilon _2\right) \text{sh}\left(\epsilon _1+\epsilon _3\right) \text{sh}\left(2
   \epsilon _3-2 \epsilon _2\right)}{\text{sh}\left(\epsilon _2\right) \text{sh}\left(-\epsilon _1-3 \epsilon _2+\epsilon
   _3\right) \text{sh}\left(\epsilon _3-2 \epsilon _2\right)}
\eea
For level two, we have
\bea
(0, \epsilon_2),&\quad -\frac{\text{sh}\left(-\epsilon _1-2 \epsilon _2+\epsilon _3\right) \text{sh}\left(-\epsilon _1-\epsilon _2+\epsilon
   _3\right) }{2
   \text{sh}\left(\epsilon _1\right) \text{sh}\left(\epsilon _2-\epsilon _1\right) \text{sh}\left(2 \epsilon _3\right)
   } \\ 
   &\times \frac{\text{sh}\left(\epsilon _2+\epsilon _3\right) \text{sh}\left(2 \epsilon _2+\epsilon _3\right)
   \text{sh}\left(\epsilon _1+2 \epsilon _3\right) \text{sh}\left(\epsilon _1-\epsilon _2+2 \epsilon _3\right)}{\text{sh}\left(\epsilon _3-2 \epsilon _2\right) \text{sh}\left(\epsilon _3-\epsilon _2\right) \text{sh}\left(2 \epsilon
   _3-\epsilon _2\right)}\\
(-\epsilon_1+\epsilon_2, \epsilon_2),&\quad \frac{\text{sh}\left(2 \epsilon _2\right) \text{sh}\left(2 \epsilon _2-\epsilon _1\right) \text{sh}\left(\epsilon
   _1+\epsilon _3\right) \text{sh}\left(2 \epsilon _1+\epsilon _3\right) \text{sh}\left(-\epsilon _1-2 \epsilon _2+\epsilon
   _3\right) }{2 \text{sh}\left(\epsilon _1\right) \text{sh}\left(2
   \epsilon _1\right) \text{sh}\left(\epsilon _2\right) \text{sh}\left(\epsilon _2-\epsilon _1\right)
   \text{sh}\left(\epsilon _1-2 \epsilon _2+\epsilon _3\right)} \\ 
   &\times\frac{\text{sh}\left(\epsilon _2+\epsilon _3\right) \text{sh}\left(-\epsilon _1+\epsilon _2+\epsilon _3\right)
   \text{sh}\left(2 \epsilon _1-\epsilon _2+2 \epsilon _3\right)}{ \text{sh}\left(\epsilon _1-\epsilon _2+\epsilon _3\right)
   \text{sh}\left(2 \epsilon _1-\epsilon _2+\epsilon _3\right) \text{sh}\left(2 \epsilon _3-\epsilon _2\right)} \\
(\epsilon_2, -\epsilon_2+\epsilon_3),&\quad -\frac{\text{sh}\left(3 \epsilon _2\right) \text{sh}\left(\epsilon _1+2 \epsilon _2\right) \text{sh}\left(\epsilon
   _1+\epsilon _3\right){}^2 \text{sh}\left(\epsilon _2+\epsilon _3\right) \text{sh}\left(\epsilon _1-2 \epsilon _2+2
   \epsilon _3\right)}{2 \text{sh}\left(\epsilon _1\right) \text{sh}\left(\epsilon _2\right){}^2 \text{sh}\left(\epsilon
   _3-3 \epsilon _2\right) \text{sh}\left(\epsilon _1-2 \epsilon _2+\epsilon _3\right) \text{sh}\left(\epsilon _3-\epsilon
   _2\right)} \\ 
(-\epsilon_2+\epsilon_3, -2\epsilon_2+\epsilon_3),&\quad \frac{\text{sh}\left(\epsilon _1+3 \epsilon _2\right) \text{sh}\left(\epsilon _1+4 \epsilon _2\right)
   \text{sh}\left(\epsilon _1+\epsilon _3\right) \text{sh}\left(\epsilon _1-\epsilon _2+\epsilon _3\right) \text{sh}\left(2
   \epsilon _3-3 \epsilon _2\right) \text{sh}\left(2 \epsilon _3-2 \epsilon _2\right)}{2 \text{sh}\left(\epsilon _2\right)
   \text{sh}\left(2 \epsilon _2\right) \text{sh}\left(-\epsilon _1-4 \epsilon _2+\epsilon _3\right) \text{sh}\left(\epsilon
   _3-3 \epsilon _2\right) \text{sh}\left(-\epsilon _1-3 \epsilon _2+\epsilon _3\right) \text{sh}\left(\epsilon _3-2
   \epsilon _2\right)}
\eea

\paragraph{Unrefined vertex}
\bea\Yboxdim{4pt}
\wtC_{\,{\yng(2)}\,\,\yng(2)\,\varnothing}(q)=\frac{1-q^2+q^4}{(1-q)^{2}(1-q^2)^2}
\eea
\paragraph{Refined vertex}
\bea\Yboxdim{4pt}
&\wtC_{\,{\yng(2)}\,\,\yng(2)\,\varnothing}(t,q)=\frac{1-q^2+q^2t^2}{(1-q)(1-q^2)(1-t)(1-t^2)},\\
&\wtC_{\varnothing\,{\yng(2)}\,\,\yng(2)\,}(t,q)=\frac{1-q^2-q^3+q^5+q^2t+q^3t-q^4t-q^5t+q^4t^2}{(1-q)(1-q^2)(1-t)(1-qt)},\\
&\wtC_{\,{\yng(2)}\,\varnothing\,\yng(2)\,}(t,q)=\frac{1-t^2+q^2t^2}{(1-t)^{2}(1-t^{2})(1-qt)}
\eea

\paragraph{Macdonald refined vertex}
\bea\Yboxdim{4pt}
\wtM_{\adjustbox{valign=c}{\yng(2)}\,\varnothing
\,\adjustbox{valign=c}{\yng(2)}}(x,y;q,t)&=\frac{(tx;q)_{\infty}(tx y;q)_{\infty}}{(x;q)_{\infty}(x y ;q)_{\infty}}\frac{1-x^2+x^2y^2}{(1-x)(1-x^2)},\\
\wtM_{{\yng(2)}\,\yng(2)\,\varnothing\,}(x,y;q,t)&=\frac{1-y^2-q^2+x^2y^2+q^2y^2-tx^2y+qx^2y-qtx^2+q^2x^2-q^2x^2y^2}{(1-x)(1-x^2)(1-y)(1-y^2)(1-q^2)},\\
\wtM_{\varnothing\,\adjustbox{valign=c}{\yng(2)}\,\adjustbox{valign=c}{\yng(2)}}(x,y;q,t)&=\frac{(tx;q)_{\infty}(tx y;q)_{\infty}}{(x;q)_{\infty}(x y;q)_{\infty}}\frac{\#}{(1-y)(1-y^2)}, 
\eea
where
\bea
\#&=-qtx^2y^4 + qtxy^5 + qx^2y^5 - tx^2y^5 + qtxy^4 - qty^5 - qxy^5 + txy^5 - qtxy^3- qxy^4\\
&+ txy^4 + x^2y^4 - xy^5 - qtxy^2 + qty^3 + qxy^3 - txy^3 + qy^4 - ty^4 - xy^4 + y^5+ qty^2\\
&+ qxy^2 - txy^2 - qy^3 + ty^3 + xy^3 - qy^2 + ty^2 + xy^2 - y^3 - qt + qy - ty - y^2 + 1
\eea

\subsection{Three-legs PT3 counting}
In this section, we summarize the results for PT3 counting with three-legs for the cases when $|\lambda|+|\mu|+|\nu|\leq 5$. The PT box counting rules \cite{Pandharipande:2007kc,Pandharipande:2007sq}, GR box counting rules \cite{Gaiotto:2020dsq}, JK-residues, unrefined topological vertex, refined topological vertex, and the Macdonald refined topological vertex are summarized. 

For the GR box counting rules, we quote the results in Appendix of \cite{Gaiotto:2020dsq}. See it for a detailed analysis. Compared to their notations, the signs of the $\eps_{a}$-parameters are all flipped as $-\eps_{a}$.  

Note that for the PT box counting rules, we write the labels of the boxes at the red positions as $\ket{0^{a}}$ where $a=-1,0,1,2,3$. For $a=0$, this is understood as the two-state degeneracy, while for $a=-1$, it is understood as the ultra heavy box. For the ultra heavy box, it is counted twice when counting the boxes, while for the two-state degeneracy, it is counted twice when counting the possible configurations. 

To list the PT configurations, Henry Liu's \href{https://github.com/liu-henry-hl/boxcounting}{sage math code} was helpful. We thank him for writing a wonderful computer program.

\subsubsection[Example 1]{ $(\lambda,\mu,\nu)=(\yng(1),\yng(1,1),\yng(1))$ (Fig.~\ref{fig:3leg-ex1})}
\begin{figure}[h]
    \centering
    \includegraphics[width=5cm]{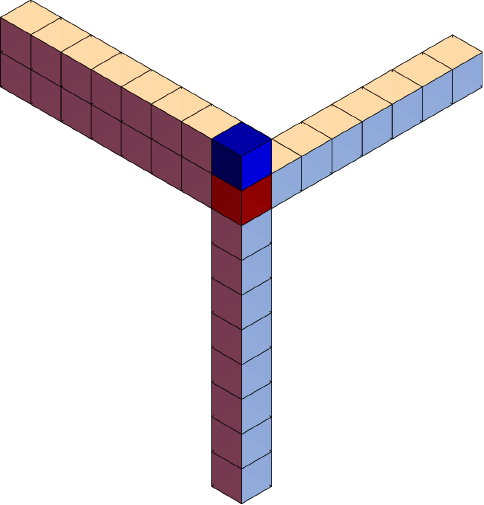}
    \caption{The PT box counting setup for $\lambda=\{1\},\mu=\{1,1\}, \nu=\{1\}$.}
    \label{fig:3leg-ex1}
\end{figure}
Let us study the case $\lambda=\yng(1),\mu=\yng(1,1),\nu=\yng(1)$ (see Fig.~\ref{fig:3leg-ex1}). The red position has the coordinate $0$ and the blue position has $\epsilon_{3}$. The blue position belongs to the 2, 3 directions and thus it is labeled by 1.

\paragraph{PT box-counting rules}
\begin{itemize}
    \item \textbf{Level 1}: 
    \bea
    \ket{\epsilon_3},\quad \ket{0^1}
    \eea
    \item \textbf{Level 2}:
    \bea
    \ket{0^1,-\epsilon_1},\quad \ket{\eps_3,\eps_3-\eps_2},\quad \ket{0^0,\eps_3}
    \eea
    \item \textbf{Level 3}:
    \bea
    \ket{\eps_3,\eps_3-\eps_2,\eps_3-2\eps_2},\quad \ket{\epsilon_3,\epsilon_3-\epsilon_2, 0^{0}},\quad
    \ket{\epsilon_3,0^1,-\epsilon_1},\quad \ket{\epsilon_3,0^3,-\epsilon_3},\quad 
    \ket{0^1,-\eps_1,-2\epsilon_1},\quad \ket{0^{-1},\eps_3}
    \eea
\end{itemize}

\paragraph{GR box-counting rules}
\begin{itemize}
\item \textbf{Level 0 $\rightarrow $ Level 1}:
\bea
\ket{\text{vac}}&\rightarrow \ket{0_{H}}+\ket{\epsilon_{3}}.
\eea

\item \textbf{Level 1 $\rightarrow $ Level 2}: 
\bea
\ket{\epsilon_{3}}&\rightarrow \ket{\epsilon_{3},\epsilon_{3}-\epsilon_{2}}+\ket{\epsilon_{3},0_{H}}+\ket{\epsilon_{3},0_{L}}\\
\ket{0_{H}}&\rightarrow \ket{0_{H},-\epsilon_{1}}+\ket{\epsilon_{3},0_{L}}
\eea



\item \textbf{Level 2 $\rightarrow $ Level 3}:
\bea
|\epsilon_3,\epsilon_3-\epsilon_2\rangle &\rightarrow|\epsilon_3,\epsilon_3-\epsilon_2,0_L\rangle+ |\epsilon_3,\epsilon_3-\epsilon_2,0_H\rangle +|\epsilon_3,\epsilon_3-\epsilon_2,\epsilon_3-2\epsilon_2\rangle  \\ 
|\epsilon_3,0_H\rangle &\rightarrow  |\epsilon_3,0_U\rangle + |\epsilon_3,0_L,\epsilon_3-\epsilon_2\rangle + |\epsilon_3,0_H,\epsilon_3-\epsilon_2\rangle + |\epsilon_3,0_H,-\epsilon_1\rangle +|\epsilon_3,0_H,-\epsilon_3\rangle \\ 
|\epsilon_3,0_L\rangle &\rightarrow |\epsilon_3,0_U\rangle + |\epsilon_3,\epsilon_3-\epsilon_2,0_L\rangle  \\ 
 |0_H, -\epsilon_1\rangle&\rightarrow |0_H, -\epsilon_1,\epsilon_3\rangle + |0_H, -\epsilon_1,-2\epsilon_1\rangle
\eea





\end{itemize}

\paragraph{JK prescription and transition rules}
The framing node contribution is
\bea
\mathcal{Z}^{\D6_{\bar{4}}\tbar\D2\tbar\D0}_{\DT;\,\yng(1)\,\yng(1,1)\,\yng(1)}(\fra,\phi_{I})&=\frac{\sh\left(-\fra+\epsilon _2+\phi_I\right) \sh\left(-\fra+\epsilon
   _1-\epsilon _3+\phi_I\right) \sh\left(-\fra+\epsilon _3+\phi_I\right)}{ \sh\left(-\fra-\epsilon _1-\epsilon _2+\phi_I\right)
   \sh\left(-\fra-\epsilon _2-2 \epsilon _3+\phi_I\right) \sh\left(-\fra-\epsilon _1-\epsilon _3+\phi_I\right)}\\
   &\times \frac{\sh\left(\fra+\epsilon _1+\epsilon _2+\epsilon _3-\phi_I\right) \sh\left(\fra+\epsilon _1+\epsilon
   _2+2 \epsilon _3-\phi_I\right) }{\sh\left(\fra-\phi _I\right)
   \sh\left(\fra+\epsilon _3-\phi_I\right)}.
\eea
Let us first list up the poles giving up non-zero residues picked up from the JK-formalism for each level. Choosing the reference vector to be $ \tilde{\eta}_{0}$, the charge vectors are chosen from 
\bea
-e_{I},\quad e_{I}-e_{J}\,\,(I>J),
\eea
where we used the Weyl invariance and picked up the poles in the order $\phi_1,\phi_2,\ldots$.

\begin{itemize}
\item \textbf{Level 1}: For level one, setting the reference vector to be $\tilde{\eta}_{0}$, the charge vector is $\{-e_{1}\}$ and the poles picked up are
\bea
\phi_{1}=\fra,\quad \phi_{1}=\fra+\epsilon_3.
\eea
The transition rule is
\bea
\boxed{\text{vac}\longrightarrow (0)+(\eps_3).}
\eea
\item \textbf{Level 2}: 
After the JK-formalism, one will see that the following poles give non-zero residues
\bea
(\fra,\fra+\eps_3),\quad (\fra+\eps_3,\fra+\eps_3-\eps_2),\quad (\fra,\fra-\eps_1).
\eea

Let us study the flag structure and the transition rule of the PT configuration $(0,\eps_3)$. The hyperplanes giving this pole are
\bea
-\phi_{1}+\fra=0,\quad \fra+\eps_3-\phi_2=0,\quad -\phi_1+\phi_2-\eps_3=0
\eea
whose charge vectors are
\bea
Q_{\ast}=\left\{\vectxy{-1}{0},\,\vectxy{0}{-1},\, \vectxy{-1}{1}\right\}
\eea
and thus it is a degenerate pole. Let us denote the three charge vectors as $Q_{1},Q_{2},Q_{3}$, respectively. We have three choices of linear independent charge vectors:
\bea
\left\{\vectxy{-1}{0},\,\vectxy{0}{-1}\right\},\quad \left\{\vectxy{-1}{0},\, \vectxy{-1}{1}\right\},\quad\left\{\vectxy{0}{-1},\, \vectxy{-1}{1}\right\}.
\eea
Since the level one configuration is $(0)$ or $(\eps_3)$, we can focus on the following bases
\bea
\mathcal{B}(F^{(1)},Q_{\ast})=\{Q_{1},Q_{2}\},\quad  \mathcal{B}(F^{(2)},Q_{\ast})=\{Q_{1},Q_{3}\},\\
\mathcal{B}(F^{(3)},Q_{\ast})=\{Q_{2},Q_{1}\},\quad \mathcal{B}(F^{(4)},Q_{\ast})=\{Q_{2},Q_{3}\}.
\eea
Note that the flags $F^{(1)},F^{(2)}$ are the same and $F^{(3)},F^{(4)}$ are the same. Thus, we can choose one of each and focus on $F^{(1)},F^{(3)}$. The corresponding $\kappa$ matrices are
\bea
\kappa_{F^{(1)}}=\{Q_{1},Q_{1}+Q_{2}+Q_{3}\}=\{(-1,0),(-2,0)\},\quad \kappa_{F^{(3)}}=\{(0,-1),(-2,0)\}.
\eea
Only the cone of the second case $\kappa_{F^{(3)}}$, the reference vector $\tilde{\eta}_{0}=(-1,-1)$ is included. Thus, the iterative residue is taken as
\bea
\underset{\phi_{1}=\fra}{\Res}\,\,
\underset{\phi_{2}=\fra+\eps_{3}}{\Res}.
\eea
To obtain the PT configuration $(\eps_{3},0)$, the following transition
\bea
(\eps_{3})\longrightarrow (0,\eps_{3})
\eea
is allowed and we can add a box at the origin after adding the box at $\eps_3$. On the other hand, the following transition is not allowed
\bea
(0)\,\,\adjustbox{valign=c}{\tikz[baseline]{
  \draw[->] (0,0) -- (0.6,0);
  \node at (0.3,0) {$\times$};
}}\,\,(0,\eps_{3}). 
\eea

Due to this flag construction, we cannot place the $0$ first to obtain the PT configuration $(0,\eps_3)$. Note also that the pole $\phi_{1}=\fra$ will be a second order pole and thus derivatives of the integrand appears.

A similar analysis can be done for the other poles. Such poles are non-degenerate poles and the residue is easily taken. The transition rule is also easy to determined and it is summarized as
\begin{empheq}[box=\fbox]{align}
\begin{split}
(0)&\longrightarrow (0,-\eps_1)\\
(\eps_3)&\longrightarrow (\eps_3,\eps_3-\eps_2)+(\eps_3,0)
\end{split}
\end{empheq}

\item \textbf{Level 3}: The third level is much more complicated because of the existence of the second order poles. The residue of a second order pole gives the derivative of the integrand which introduces much more terms. 

For the configuration $(\eps_3,\eps_3-\eps_2)$, the poles relevant are
\bea
\frac{1}{\sh(\fra-\phi_3)^{2}\sh(\fra-2\eps_2+\eps_3-\phi_3)}
\eea
and we have a second order pole at the origin and a single pole at $\fra+\eps_3-2\eps_2$. The transition rule is given as
\bea
\boxed{(\eps_3,\eps_3-\eps_2)\longrightarrow (\eps_3,\eps_3-\eps_2,0)+(\eps_3,\eps_3-\eps_2,\eps_3-2\eps_2)}
\eea
For the configuration $(0,-\eps_1)$, the relevant poles are
\bea
\frac{1}{\sh(\fra+\eps_3-\phi_3)\sh(\fra-2\eps_1-\phi_3)}
\eea
and no second order poles appear. The transition rule is simply
\bea
\boxed{(0,-\eps_1)\longrightarrow (0,-\eps_1,\eps_3)+(0,-\eps_1,-2\eps_1)}
\eea
For $(\eps_3,0)$, after explicit computation, one will see that the poles coming from $\phi_{3}=\fra-\eps_{1,3}$ and $\phi_3=\fra+\eps_3-\eps_2$ are all single order poles. Moreover, due to the derivative, the pole at $\phi_3=\fra$ will not vanish as discussed in section~\ref{sec:PTthreelegs}. We thus have
\bea
\boxed{(\eps_3,0)\longrightarrow (\eps_3,0,0)+(\eps_3,0,\eps_3-\eps_2)+(\eps_3,0,-\eps_1)+(\eps_3,0,-\eps_3)}
\eea

\end{itemize}

Taking the unrefined limit gives the contribution coming from $(\eps_3,0,\eps_3-\eps_2)$ to be weight two.

\paragraph{JK-residues}Let us explicitly list down the residues evaluated at the corresponding poles up to level two. Note again that the residues listed here includes the Weyl group factor $k!$ in the denominator. 

For level one, we have
\bea
(0), & \quad -\frac{\sh\left(\epsilon _3-\epsilon _1\right) \sh\left(\epsilon _2+\epsilon
   _3\right) \sh\left(\epsilon _1+\epsilon _2+2 \epsilon _3\right)}{\sh\left(\epsilon
   _1\right) \sh\left(\epsilon _3\right) \sh\left(\epsilon _2+2 \epsilon _3\right)}, \\
(\epsilon _3), &\quad  \frac{\sh\left(\epsilon _1+\epsilon _2\right)^2
   \sh\left(2 \epsilon _3\right) \sh\left(\epsilon _1+\epsilon _3\right)
   \sh\left(\epsilon _2+\epsilon _3\right)}{\sh\left(\epsilon _1\right)
   \sh\left(\epsilon _2\right) \sh\left(\epsilon _3\right)^2 \sh\left(-\epsilon
   _1-\epsilon _2+\epsilon _3\right)} .
   \eea
For level two, we have
\bea
(0,-\epsilon_{1}),&\quad \frac{\sh\left(\epsilon _3-\epsilon _1\right) \sh\left(\epsilon _2+\epsilon _3\right)
   \sh\left(-\epsilon _1+\epsilon _2+\epsilon _3\right) \sh\left(2 \epsilon _1+\epsilon
   _2+2 \epsilon _3\right)}{2 \sh\left(\epsilon _1\right) \sh\left(2 \epsilon _1\right)
   \sh\left(\epsilon _1+\epsilon _3\right) \sh\left(\epsilon _2+2 \epsilon _3\right)},\\
(\epsilon_{3},\epsilon_{3}-\epsilon_{2}),&\quad\frac{\sh\left(\epsilon _1+\epsilon _2\right) \sh\left(\epsilon _1+2 \epsilon
   _2\right)^2 \sh\left(2 \epsilon _3\right) \sh\left(\epsilon _1+\epsilon _3\right)
   \sh\left(\epsilon _1-\epsilon _2+\epsilon _3\right) \sh\left(\epsilon _2+\epsilon
   _3\right) \sh\left(2 \epsilon _3-\epsilon _2\right)}{2 \sh\left(\epsilon _1\right)
   \sh\left(\epsilon _2\right) \sh\left(2 \epsilon _2\right) \sh\left(\epsilon
   _3\right) \sh\left(-\epsilon _1-2 \epsilon _2+\epsilon _3\right) \sh\left(\epsilon
   _3-\epsilon _2\right)^2 \sh\left(-\epsilon _1-\epsilon _2+\epsilon _3\right)},\\
(0,\epsilon_{3}),&\qquad \frac{\ch\left(\epsilon _1+2 \epsilon _3\right) \sh\left(\epsilon _1+\epsilon _2\right)
   \sh\left(\epsilon _2+\epsilon _3\right) \sh\left(\epsilon _1+\epsilon _2+\epsilon
   _3\right)}{4 \sh\left(\epsilon _1\right) \sh\left(\epsilon _3\right)
   \sh\left(\epsilon _3-\epsilon _2\right)}\\
   &\quad-\frac{\ch\left(\epsilon _2\right)
   \sh\left(\epsilon _1+\epsilon _2\right) \sh\left(\epsilon _2+\epsilon _3\right)
   \sh\left(\epsilon _1+2 \epsilon _3\right) \sh\left(\epsilon _1+\epsilon _2+\epsilon
   _3\right)}{2 \sh\left(\epsilon _1\right) \sh\left(\epsilon _2\right)
   \sh\left(\epsilon _3\right) \sh\left(\epsilon _3-\epsilon
   _2\right)}\\
   &\quad -\frac{\ch\left(\epsilon _1\right) \sh\left(\epsilon _1+\epsilon _2\right)
   \sh\left(\epsilon _2+\epsilon _3\right) \sh\left(\epsilon _1+2 \epsilon _3\right)
   \sh\left(\epsilon _1+\epsilon _2+\epsilon _3\right)}{4 \sh\left(\epsilon _1\right)^2
   \sh\left(\epsilon _3\right) \sh\left(\epsilon _3-\epsilon
   _2\right)}\\
   &\quad-\frac{\ch\left(2 \epsilon _3\right) \sh\left(\epsilon _1+\epsilon _2\right)
   \sh\left(\epsilon _2+\epsilon _3\right) \sh\left(\epsilon _1+2 \epsilon _3\right)
   \sh\left(\epsilon _1+\epsilon _2+\epsilon _3\right)}{4 \sh\left(\epsilon _1\right)
   \sh\left(\epsilon _3\right) \sh\left(2 \epsilon _3\right) \sh\left(\epsilon
   _3-\epsilon _2\right)}\\
   &\quad-\frac{\ch\left(\epsilon _1+\epsilon _3\right) \sh\left(\epsilon
   _1+\epsilon _2\right) \sh\left(\epsilon _2+\epsilon _3\right) \sh\left(\epsilon _1+2
   \epsilon _3\right) \sh\left(\epsilon _1+\epsilon _2+\epsilon _3\right)}{2
   \sh\left(\epsilon _1\right) \sh\left(\epsilon _3\right) \sh\left(\epsilon
   _1+\epsilon _3\right) \sh\left(\epsilon _3-\epsilon
   _2\right)}\\
   &\quad+\frac{\ch\left(-\epsilon _1-\epsilon _2+\epsilon _3\right)
   \sh\left(\epsilon _1+\epsilon _2\right) \sh\left(\epsilon _2+\epsilon _3\right)
   \sh\left(\epsilon _1+2 \epsilon _3\right) \sh\left(\epsilon _1+\epsilon _2+\epsilon
   _3\right)}{4 \sh\left(\epsilon _1\right) \sh\left(\epsilon _3\right)
   \sh\left(\epsilon _3-\epsilon _2\right) \sh\left(-\epsilon _1-\epsilon _2+\epsilon
   _3\right)}\\
   &\quad-\frac{\ch\left(\epsilon _3-\epsilon _2\right) \sh\left(\epsilon _1+\epsilon
   _2\right) \sh\left(\epsilon _2+\epsilon _3\right) \sh\left(\epsilon _1+2 \epsilon
   _3\right) \sh\left(\epsilon _1+\epsilon _2+\epsilon _3\right)}{4 \sh\left(\epsilon
   _1\right) \sh\left(\epsilon _3\right) \sh\left(\epsilon _3-\epsilon
   _2\right)^2}\\
   &\quad-\frac{\ch\left(\epsilon _2+\epsilon _3\right) \sh\left(\epsilon
   _1+\epsilon _2\right) \sh\left(\epsilon _1+2 \epsilon _3\right) \sh\left(\epsilon
   _1+\epsilon _2+\epsilon _3\right)}{4 \sh\left(\epsilon _1\right) \sh\left(\epsilon
   _3\right) \sh\left(\epsilon _3-\epsilon _2\right)}\\
   &\quad+\frac{\ch\left(\epsilon _1+\epsilon
   _2+\epsilon _3\right) \sh\left(\epsilon _1+\epsilon _2\right) \sh\left(\epsilon
   _2+\epsilon _3\right) \sh\left(\epsilon _1+2 \epsilon _3\right)}{2 \sh\left(\epsilon
   _1\right) \sh\left(\epsilon _3\right) \sh\left(\epsilon _3-\epsilon _2\right)}.
\eea

Let us denote the residues of the integrand without the Weyl group factor at the poles $\phi_{\ast}$ as $\mathcal{Z}_{\phi_{\ast}}$. Taking the unrefined limit $q_{4}\rightarrow 1$ gives the following weights. Each level matches with the unrefined topological vertex.
\begin{itemize}
    \item \textbf{Level 1}: 
    \bea
    \mathcal{Z}_{(0)}\rightarrow 1,\quad \mathcal{Z}_{(\eps_3)}\rightarrow 1
    \eea
    \item \textbf{Level 2}:
    \bea
    \mathcal{Z}_{(0,-\eps_1)}\rightarrow 1,\quad \mathcal{Z}_{(\eps_3,\eps_3-\eps_2)}\rightarrow 1,\quad \mathcal{Z}_{(0,\eps_3)}\rightarrow 2.
    \eea
    \item \textbf{Level 3}:
    \bea
    \mathcal{Z}_{(0,-\eps_1,\eps_3)}\rightarrow 1,\quad \mathcal{Z}_{(0,-\eps_1,-2\eps_1)}\rightarrow 1,\quad \mathcal{Z}_{(\eps_3,0,0)}\rightarrow 2,\\
    \mathcal{Z}_{(\eps_3,0,\eps_3-\eps_2)}\rightarrow 2,\quad \mathcal{Z}_{(\eps_3,0,-\eps_1)}\rightarrow 1,\quad \mathcal{Z}_{(\eps_3,0,-\eps_3)}\rightarrow 1
    \eea
    The contribution from $(\eps_3,0,0)$ is only one after incorporating the Weyl group factor properly.
\end{itemize}

\paragraph{Unrefined vertex}
\bea
\wtC_{{\Yboxdim{4pt}\yng(1),\yng(2),\yng(1)}}(q)=\frac{1-q+q^{3}-q^{5}+q^6}{(1-q)^{3}(1-q^{2})}=1+2q^{2}+4q^2+7q^3+\cdots
\eea

\paragraph{Refined vertex}
\bea
\wtC_{{\Yboxdim{4pt}\yng(1),\yng(2),\yng(1)}}(t,q)&=-\frac{q^4 t-q^4-q^3 t^3-q^3 t^2+2 q^3 t+q^2 t^3-q^2 t^2-2 q^2 t+2 q^2+q t^2-2 q
   t+t-1}{(q-1)^2 (q+1) (t-1)^2}\\
   &=1+q+t+3qt+2q^2 t+t^2+4qt^2+t^3+\cdots
\eea

\paragraph{Macdonald vertex}
\bea\Yboxdim{4pt}
\wtM_{\yng(1)\,\yng(2)\,\yng(1)}(x,y;q,t)&=\frac{(tx;q)_{\infty}}{(x;q)_{\infty}}\frac{\#}{(1-x)(1-y)(1-y^2)(1-qt)}
\eea
where
\bea
\#&=-qtx^3y^3 - q^2x^2y^4 + qtx^2y^4 + qx^3y^4 - tx^3y^4 + qtx^3y^2 - qtx^2y^3 - qx^3y^3 + tx^3y^3 + 2q^2xy^4 \\
&- qtxy^4 - qx^2y^4 + tx^2y^4 + q^2x^2y^2 - 2qtx^2y^2 - q^2xy^3 + 3qtxy^3 + 2qx^2y^3 - 2tx^2y^3 + x^3y^3\\
&- q^2y^4 + qtx^2y - 2q^2xy^2 + qx^2y^2 - tx^2y^2 - x^3y^2 + q^2y^3 - qty^3 - qxy^3 + txy^3 \\
&+ x^2y^3 - xy^4 + q^2xy - 3qtxy - qx^2y + tx^2y + q^2y^2 + qty^2 + x^2y^2 - 2xy^3 \\
&+ y^4 + qtx - q^2y + qty + qxy - txy - x^2y + 2xy^2 - qt + 2xy - 2y^2 - x + 1
\eea
In the series expansion, we have
\bea
\,&\wtM_{\yng(1)\,\yng(2)\,\yng(1)}(x,y;q,t)=1+\left(\frac{1-t}{1-q}x+\frac{1-q^{2}}{1-qt}y\right)+\left(\frac{(1-t)(1-qt)}{(1-q)(1-q^2)}x^2+\frac{3-2t+2q-3qt}{1-qt}xy\right)\\
&+\left(\frac{(1-t)(1-qt)(1-q^{2}t)}{(1-q)(1-q^{2})(1-q^{3})}x^{3}+\frac{4-4t+t^2-5qt+3qt^2+q^2t}{(1-q)(1-qt)}x^2y+\frac{(1+q)(2-q-t)}{1-qt}xy^{2}\right)+\cdots
\eea


\subsubsection[Example 2]{ $(\lambda,\mu,\nu)=(\yng(1),\yng(1,1,1),\yng(1))$ (Fig.~\ref{fig:3leg-ex2})}
\begin{figure}[h]
    \centering
    \includegraphics[width=5cm]{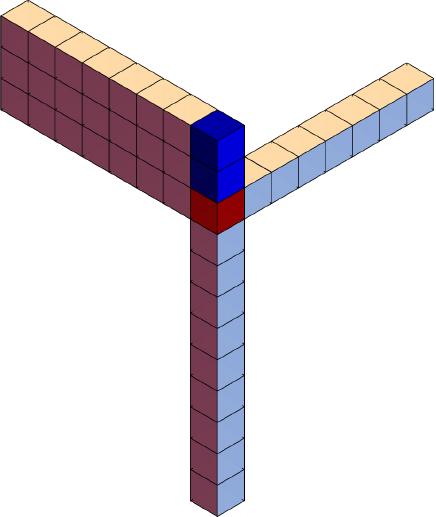}
    \caption{The PT box counting setup for $\lambda=\{1\},\mu=\{1,1,1\}, \nu=\{1\}$.}
    \label{fig:3leg-ex2}
\end{figure}
We have one red position at $0$ and two blue positions at $\epsilon_{3},2\epsilon_{3}$. The two blue positions are intersections of the 2, 3 cylinders and so they are labeled by 1.

\paragraph{PT box-counting rules}
\begin{itemize}
    \item \textbf{Level 1}
    \bea
    \ket{0^1},\quad \ket{2\eps_3}
    \eea
    \item \textbf{Level 2}
    \bea
    \ket{0^1,2\eps_3},\quad \ket{2\eps_3,2\eps_3-\eps_1},\quad \ket{\eps_3,2\eps_3},\quad \ket{0^1,-\eps_1}
    \eea
    \item \textbf{Level 3}
    \bea
    \ket{\eps_3,2\eps_3,2\eps_3-\eps_2},\quad \ket{2\eps_3-2\eps_2, 2\eps_3,2\eps_3-\eps_2},\quad \ket{0^1,2\eps_3,-\eps_1},\\
    \ket{0^0,\eps_3,2\eps_3},\quad\ket{0^1,-\eps_1,-2\eps_1},\quad \ket{0^1,2\eps_3-\eps_2,2\eps_3}
    \eea
    
\end{itemize}

\paragraph{GR box-counting rules}
\begin{itemize}
\item \textbf{Level 0 $\rightarrow $ Level 1}:
\bea
\ket{\text{vac}}&\rightarrow \ket{0_{H}}+\ket{2\epsilon_{3}}
\eea

\item \textbf{Level 1 $\rightarrow $ Level 2}: 
\bea
\ket{2\epsilon_{3}}&\rightarrow \ket{2\epsilon_{3}, 2\epsilon_{3}-\epsilon_{2}}+\ket{2\epsilon_{3},\epsilon_{3}}+\ket{2\epsilon_{3},0_{H}}\\
\ket{0_H}&\rightarrow \ket{0_{H},2\epsilon_{3}}+\ket{0_{H},-\epsilon_{1}}
\eea


\item \textbf{Level 2 $\rightarrow $ Level 3}: 
\bea
|2\epsilon_3,2\epsilon_3-\epsilon_2\rangle &\rightarrow |2\epsilon_3,2\epsilon_3-\epsilon_2,\epsilon_3\rangle + |2\epsilon_3,2\epsilon_3-\epsilon_2,0_H\rangle + |2\epsilon_3,2\epsilon_3-\epsilon_2,2\epsilon_3-2\epsilon_2\rangle,\\ 
    |2\epsilon_3,\epsilon_3\rangle &\rightarrow  |2\epsilon_3,\epsilon_3,2\epsilon_3-\epsilon_2\rangle +|2\epsilon_3,\epsilon_3,0_H\rangle +|2\epsilon_3,\epsilon_3,0_L\rangle,  \\ 
|2\epsilon_3,0_H\rangle&\rightarrow|2\epsilon_3,0_H,2\epsilon_3-\epsilon_2\rangle +|2\epsilon_3,\epsilon_3,0_L\rangle +|2\epsilon_3,0_H,-\epsilon_1\rangle,  \\ 
 |0_H, -\epsilon_1\rangle &\rightarrow |0_H, -\epsilon_1,2\epsilon_3\rangle+|0_H, -\epsilon_1,-2\epsilon_1\rangle.
\eea




\end{itemize}

\paragraph{JK prescription and transition rules}
The framing node contribution is
\bea
\mathcal{Z}^{\D6_{\bar{4}}\tbar\D2\tbar\D0}_{\DT;\,\yng(1)\,\yng(1,1,1)\,\yng(1)}(\fra,\phi_{I})&=\frac{ \sh\left(-\fra+\epsilon _2+\phi _I\right) \sh\left(-\fra+\epsilon _1-2
   \epsilon _3+\phi _I\right) \sh\left(-\fra+\epsilon _3+\phi _I\right)}{ \sh\left(-\fra-\epsilon _1-\epsilon _2+\phi _I\right)
   \sh\left(-\fra-\epsilon _2-3 \epsilon _3+\phi _I\right) \sh\left(-\fra-\epsilon _1-\epsilon _3+\phi
   _I\right)}\\
   &\times\frac{\sh\left(\fra+\epsilon _1+\epsilon _2+\epsilon _3-\phi _I\right) \sh\left(\fra+\epsilon _1+\epsilon
   _2+3 \epsilon _3-\phi _I\right)}{\sh\left(\fra-\phi _I\right)
   \sh\left(\fra+2 \epsilon _3-\phi _I\right)}.
\eea

\begin{itemize}
\item \textbf{Level 1}: For level one, setting the reference vector $\tilde{\eta}_{0}$, the poles picked up are
\bea
\phi_{1}=\fra,\quad \fra+2\eps_{3}
\eea
and thus the transition rule is 
\bea
\boxed{\text{vac}\longrightarrow (0)+(2\eps_3)}
\eea
\item \textbf{Level two}: For level two, the configurations giving non-zero JK-residues are
\bea
(2\eps_3,2\eps_3-\eps_2),\quad (2\eps_3,\eps_3),\quad (2\eps_3,0),\quad (0,-\eps_1).
\eea

For the first two configurations, picking $\phi_{1}=\fra+ 2\eps_3$, the second pole with non-zero residue comes from 
\bea
-\phi_2+\phi_1-\eps_{3,2}=0
\eea
whose corresponding charge vector is $(1,-1)$. Such poles are non-degenerate and the transition rule is just $(2\eps_3)\rightarrow (2\eps_3,2\eps_3-\eps_2)+(2\eps_3,\eps_3) $. A similar discussion is applicable to $(0)\rightarrow (0,-\eps_1)$.

For the pole configuration $(0,2\eps_3)$ coming from the hyperplanes
\bea
\fra-\phi_{1}=0,\quad \fra+2\eps_3-\phi_2=0,
\eea
the charge vectors are $(-1,0),(0,-1)$. The possible bases are $\{(-1,0),(0,-1)\}$ and $\{(0,-1),(-1,0)\}$, and thus starting from either $(0)$ or $(2\eps_3)$, we can obtain the PT configuration $(0,2\eps_3)$.

Summarizing, the transition rules are
\begin{empheq}[box=\fbox]{align}
\begin{split}
(2\eps_3)&\longrightarrow (2\eps_3,2\eps_3-\eps_2)+(2\eps_3,\eps_3)+(2\eps_3,0)\\
(0)&\longrightarrow (0,2\eps_3)+(0,-\eps_1)
\end{split}
\end{empheq}

\item \textbf{Level three}: Let us next consider the level three. The poles giving non-zero JK-residues are
\bea
\,& (0 , -2 \epsilon _1 , -\epsilon _1 ),\quad
 (0 , -\epsilon _1 , 2 \epsilon _3 ),\quad
 (0 , \epsilon _3 , 2 \epsilon _3 ),\\
\,& (0 , 2 \epsilon _3 , 2 \epsilon _3-\epsilon _2),\quad 
 (\epsilon _3 , 2 \epsilon _3 , 2 \epsilon _3-\epsilon _2 ),\quad
 (2 \epsilon _3 , 2 \epsilon _3-2 \epsilon _2 , 2 \epsilon _3-\epsilon _2 ).
\eea
A similar analysis shows that the transition rules are determined as
\begin{empheq}[box=\fbox]{align}
\begin{split}
(2\eps_3,2\eps_3-\eps_2)&\longrightarrow (2\eps_3,2\eps_3-\eps_2,\eps_3)+(2\eps_3,2\eps_3-\eps_2,0)+(2\eps_3,2\eps_3-\eps_2,2\eps_3-2\eps_2) \\
(2\eps_3,\eps_3)&\longrightarrow (2\eps_3,\eps_3,2\eps_3-\eps_2)+(2\eps_3,\eps_3,0) \\
(2\eps_3,0)&\longrightarrow (2\eps_3,0,2\eps_3-\eps_2)+(2\eps_3,0,-\eps_1) \\
(0,-\eps_1)&\longrightarrow (0,-\eps_1,2\eps_3)+(0,-\eps_1,-2\eps_1)
\end{split}
\end{empheq}

For the configuration $(2\eps_3,\eps_3,0)$, the related hyperplanes are
\bea
\fra+2\eps_3-\phi_1=0,\quad -\phi_2+\phi_1-\eps_3=0,\quad \fra-\phi_3=0,\quad -\phi_3+\phi_2-\eps_3=0
\eea
whose corresponding charge vectors are 
\bea
Q_{\ast}=\{(-1,0,0),(1,-1,0),(0,0,-1),(0,1,-1)\}
\eea
and thus it is a degenerate pole. Let us consider the following bases
\bea
\mathcal{B}(F^{(1)},Q_{\ast})=\{(-1,0,0),(1,-1,0),(0,0,-1)\},\quad \mathcal{B}(F^{(2)},Q_{\ast})=\{(-1,0,0),(0,0,-1),(1,-1,0)\}
\eea
whose corresponding $\kappa$ vectors are
\bea
\kappa_{F^{(1)}}=\{(-1,0,0),(0,-1,0),(0,0,-2)\},\quad 
\kappa_{F^{(2)}}=\{(-1,0,0),(0,0,-1),(0,0,-2)\} .
\eea
One can choose other choice of bases, but they will be the same with the above cases. One will see that the reference vector $\tilde{\eta}_{0}$ is only inside the cone obtained from $\kappa_{F^{(1)}}$. The iterated residue is then taken in the way
\bea
\underset{\phi_3=\fra}{\Res}\underset{\phi_2=\fra+\eps_3}{\Res}\underset{\phi_1=\fra+2\eps_3}{\Res}
\eea
but \textit{not}
\bea
\underset{\phi_2=\fra+\eps_3}{\Res}\underset{\phi_3=\fra}{\Res}\underset{\phi_1=\fra+2\eps_3}{\Res}.
\eea
Thus, we do not have the transition between $(2\eps_3,0)\rightarrow (2\eps_3,0,\eps_3)$. Note also that when evaluating this pole, the pole will be a second order pole.

\end{itemize}

\paragraph{JK-residues}
For level one, we have
\bea
 (0),& \quad -\frac{\sh\left(\epsilon _2+\epsilon _3\right) \sh\left(2 \epsilon _3-\epsilon
   _1\right) \sh\left(\epsilon _1+\epsilon _2+3 \epsilon _3\right)}{\sh\left(\epsilon
   _1\right) \sh\left(2 \epsilon _3\right) \sh\left(\epsilon _2+3 \epsilon _3\right)}, \\
 (2 \epsilon _3), &\quad \frac{\sh\left(\epsilon _1+\epsilon _2\right)
   \sh\left(3 \epsilon _3\right) \sh\left(\epsilon _1+\epsilon _3\right)
   \sh\left(-\epsilon _1-\epsilon _2+\epsilon _3\right) \sh\left(\epsilon _2+2 \epsilon
   _3\right)}{\sh\left(\epsilon _2\right) \sh\left(\epsilon _3\right) \sh\left(2
   \epsilon _3\right) \sh\left(\epsilon _3-\epsilon _1\right) \sh\left(-\epsilon
   _1-\epsilon _2+2 \epsilon _3\right)} .
\eea
For level two, we have
\bea
 (0, -\epsilon _1), &\qquad \frac{\sh\left(\epsilon _2+\epsilon _3\right) \sh\left(-\epsilon _1+\epsilon _2+\epsilon
   _3\right) \sh\left(2 \epsilon _3-\epsilon _1\right) \sh\left(2 \epsilon _1+\epsilon
   _2+3 \epsilon _3\right)}{2 \sh\left(\epsilon _1\right) \sh\left(2 \epsilon _1\right)
   \sh\left(\epsilon _1+2 \epsilon _3\right) \sh\left(\epsilon _2+3 \epsilon _3\right)},\\
 (0 ,2 \epsilon _3),&\qquad-\frac{\sh\left(\epsilon _1+\epsilon _2\right) \sh\left(\epsilon _1+\epsilon _3\right)
   \sh\left(\epsilon _3-\epsilon _2\right) \sh\left(\epsilon _2+\epsilon _3\right)
   \sh\left(\epsilon _1+\epsilon _2+2 \epsilon _3\right) \sh\left(\epsilon _1+3 \epsilon
   _3\right)}{2 \sh\left(\epsilon _1\right) \sh\left(\epsilon _2\right)
   \sh\left(\epsilon _3\right)^2 \sh\left(\epsilon _1+2 \epsilon _3\right)
   \sh\left(2 \epsilon _3-\epsilon _2\right)}, \\
 (\epsilon _3 , 2 \epsilon _3 ),&\qquad \frac{\sh\left(\epsilon _1+\epsilon _2\right){}^2 \sh\left(3
   \epsilon _3\right) \sh\left(\epsilon _1+\epsilon _3\right)
   \sh\left(-\epsilon _1-\epsilon _2+\epsilon _3\right)
   \sh\left(\epsilon _2+\epsilon _3\right) \sh\left(\epsilon _1+2
   \epsilon _3\right) \sh\left(\epsilon _2+2 \epsilon _3\right)}{2
   \sh\left(\epsilon _1\right) \sh\left(\epsilon _2\right)
   \sh\left(\epsilon _3\right){}^2 \sh\left(2 \epsilon _3\right)
   \sh\left(\epsilon _3-\epsilon _1\right) \sh\left(\epsilon
   _3-\epsilon _2\right) \sh\left(-\epsilon _1-\epsilon _2+2 \epsilon
   _3\right)},\\
 (2 \epsilon _3 , 2 \epsilon _3-\epsilon _2),&\qquad\frac{\sh\left(\epsilon _1+\epsilon _2\right) \sh\left(\epsilon _1+2 \epsilon _2\right)
   \sh\left(3 \epsilon _3\right) \sh\left(\epsilon _1+\epsilon _3\right)
   \sh\left(-\epsilon _1-2 \epsilon _2+\epsilon _3\right) }{2 \sh\left(\epsilon _2\right) \sh\left(2 \epsilon _2\right)
   \sh\left(\epsilon _3\right) \sh\left(\epsilon _3-\epsilon _1\right)
   }  \\
   &\qquad \times \frac{\sh\left(\epsilon _1-\epsilon
   _2+\epsilon _3\right) \sh\left(\epsilon _2+2 \epsilon _3\right) \sh\left(3 \epsilon
   _3-\epsilon _2\right)}{\sh\left(\epsilon _3-\epsilon _2\right) \sh\left(-\epsilon _1-2 \epsilon _2+2 \epsilon
   _3\right) \sh\left(2 \epsilon _3-\epsilon _2\right) \sh\left(-\epsilon _1-\epsilon
   _2+2 \epsilon _3\right)}.
\eea

The weight factor of each PT configurations are obtained by taking the unrefined limit of the residue. One will see that each configuration except the $(2\eps_3,0,\eps_3)$ contributes as $1$, while the configuration $(2\eps_3,0,\eps_3)$ contributes as $2$.
\begin{itemize}
    \item \textbf{Level 1}: 
    \bea
    \mathcal{Z}_{(0)}\rightarrow 1,\quad \mathcal{Z}_{(2\eps_3)}\rightarrow 1
    \eea
    \item \textbf{Level 2}:
    \bea
    \mathcal{Z}_{(2\eps_3,2\eps_3-\eps_3)}\rightarrow 1,\quad \mathcal{Z}_{(2\eps_3,\eps_3)}\rightarrow 1,\quad \mathcal{Z}_{(2\eps_3,0)}\rightarrow 1,\quad \mathcal{Z}_{(0,-\eps_1)}\rightarrow 1.
    \eea
    \item \textbf{Level 3}:
    \bea
   \,& \mathcal{Z}_{(0 , -2 \epsilon _1 , -\epsilon _1 )}\rightarrow 1,\quad
 \mathcal{Z}_{(0 , -\epsilon _1 , 2 \epsilon _3 )}\rightarrow 1,\quad
 \mathcal{Z}_{(0 , \epsilon _3 , 2 \epsilon _3 )}\rightarrow 2,\\
\,& \mathcal{Z}_{(0 , 2 \epsilon _3 , 2 \epsilon _3-\epsilon _2)}\rightarrow 1,\quad 
 \mathcal{Z}_{(\epsilon _3 , 2 \epsilon _3 , 2 \epsilon _3-\epsilon _2 )}\rightarrow 1,\quad
 \mathcal{Z}_{(2 \epsilon _3 , 2 \epsilon _3-2 \epsilon _2 , 2 \epsilon _3-\epsilon _2 )}\rightarrow 1.
    \eea
\end{itemize}

\paragraph{Unrefined vertex}
\bea
\wtC_{\yng(1),\,\yng(3),\,\yng(1)}(q)=\frac{1-q+q^4-q^7+q^8}{(1-q)^3(1-q^2)(1-q^3)}=1+2q+4q^2+7q^3+\cdots
\eea

\paragraph{Refined vertex}
\bea
\wtC_{\yng(1),\,\yng(3),\,\yng(1)}(t,q)=1+q+t+3qt+t^2+2q^2t+4qt^2+t^3+\cdots
\eea

\paragraph{Macdonald refined vertex}

\bea
\wtM_{\yng(1),\,\yng(3),\,\yng(1)}(x,y;q,t)&=1+\left(\frac{1-t}{1-q}x+\frac{1-q^{3}}{1-q^{2}t}y\right)+\left(\frac{(1-t)(1-qt)}{(1-q)(1-q^2)}x^2+\frac{3-2t+2q-2qt+2q^2-3q^2t}{1-q^2t}xy\right)+\cdots
\eea

\subsubsection[Example 3]{$(\lambda,\mu,\nu)=(\yng(1),\yng(2,1),\yng(1))$ (Fig.~\ref{fig:3leg-ex3})}
\begin{figure}[h]
    \centering
    \includegraphics[width=5cm]{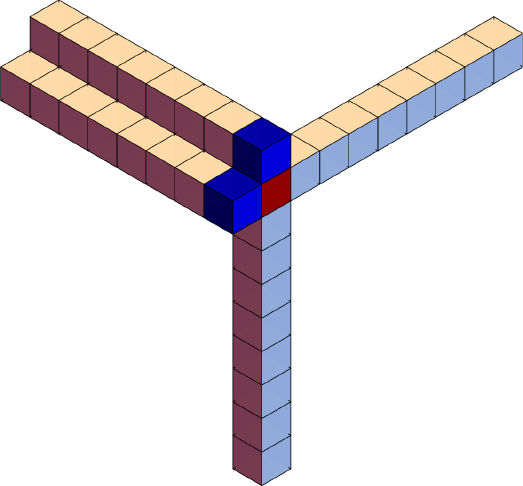}
    \caption{The PT box counting setup for $\lambda=\{1\},\mu=\{2,1\}, \nu=\{1\}$.}
    \label{fig:3leg-ex3}
\end{figure}
For the case $(\lambda,\mu,\nu)=(\yng(1),\yng(2,1),\yng(1))$ in Fig.~\ref{fig:3leg-ex3}, the red position is at $0$ and the blue positions are at $\epsilon_{1,3}$. The blue position $\epsilon_{1}$ has the label 3 and the position $\epsilon_{3}$ has the label $1$.

\paragraph{PT box-counting rules}
\begin{itemize}
    \item \textbf{Level 1}
    \bea
    \ket{\eps_1},\quad \ket{\eps_3}
    \eea
    \item \textbf{Level 2}
    \bea
    \ket{\eps_3,\eps_3-\eps_2},\quad \ket{\eps_1,\eps_3},\quad \ket{0^3,\eps_3},\quad \ket{\eps_1,\eps_1-\eps_2},\quad \ket{0^1,\eps_1}
    \eea
    \item \textbf{Level 3}
    \bea
    \ket{\eps_1,0^1,-\eps_3},\quad \ket{-\eps_3,0^3,\eps_3},\quad \ket{\eps_1,\eps_1-\eps_2,\eps_1-2\eps_2},\quad \ket{\eps_1,\eps_3,\eps_3-\eps_2},\\
    \ket{0^3,\eps_3,\eps_3-\eps_2},\quad \ket{\eps_3-2\eps_2,\eps_3,\eps_3-\eps_2},\quad \ket{\eps_1,\eps_1-\eps_2,0^1},\quad \ket{\eps_1,\eps_1-\eps_2,\eps_3},\quad \ket{\eps_1,0^0,\eps_3}
    \eea

\end{itemize}

\paragraph{GR box-counting rules}

\begin{itemize}
\item \textbf{Level 0 $\rightarrow $ Level 1}:
\bea
\ket{\text{vac}}\rightarrow \ket{\epsilon_{1}}+\ket{\epsilon_{3}}.
\eea

\item \textbf{Level 1 $\rightarrow $ Level 2}:
\bea
|\epsilon_3\rangle &\rightarrow |\epsilon_3,\epsilon_3-\epsilon_2\rangle+ |\epsilon_3,\epsilon_1\rangle+|\epsilon_3,0_H\rangle\\
|\epsilon_1\rangle &\rightarrow  |\epsilon_1,\epsilon_1-\epsilon_2\rangle+ |\epsilon_3,\epsilon_1\rangle+|\epsilon_1,0_H\rangle
\eea

\item\textbf{Level 2 $\rightarrow $ Level 3}
\bea
 |\epsilon_3,\epsilon_3-\epsilon_2\rangle&\rightarrow  |\epsilon_3,\epsilon_3-\epsilon_2,\epsilon_1\rangle + |\epsilon_3,\epsilon_3-\epsilon_2,0_H\rangle +|\epsilon_3,\epsilon_3-\epsilon_2,\epsilon_3-2\epsilon_2\rangle \\ 
 |\epsilon_1,\epsilon_1-\epsilon_2\rangle &\rightarrow |\epsilon_3,\epsilon_1,\epsilon_1-\epsilon_2\rangle + |\epsilon_1,\epsilon_1-\epsilon_2,0_H\rangle +|\epsilon_1,\epsilon_1-\epsilon_2,\epsilon_1-2\epsilon_2\rangle \\ 
|\epsilon_3,\epsilon_1\rangle &\rightarrow |\epsilon_3,\epsilon_1,\epsilon_3-\epsilon_2\rangle+ |\epsilon_3,\epsilon_1,\epsilon_1-\epsilon_2\rangle+|\epsilon_3,\epsilon_1,0_L\rangle+|\epsilon_3,\epsilon_1,0_H\rangle\\ 
|\epsilon_3,0_H\rangle &\rightarrow|\epsilon_3,\epsilon_1,0_L\rangle  + |\epsilon_3,0_H,\epsilon_3-\epsilon_2\rangle +|\epsilon_3,0_H,-\epsilon_3\rangle \\ 
|\epsilon_1,0_H\rangle &\rightarrow|\epsilon_3,\epsilon_1,0_L\rangle + |\epsilon_1,0_H,\epsilon_1-\epsilon_2\rangle +|\epsilon_1,0_H,-\epsilon_1\rangle\eea



\end{itemize}

\paragraph{JK prescription and transition rules}
Let us move on to the JK-prescription. The framing node contribution is
\bea
\mathcal{Z}^{\D6_{\bar{4}}\tbar\D2\tbar\D0}_{\DT;\,\yng(1)\,\yng(2,1)\,\yng(1)}(\fra,\phi_{I})&=\frac{ \sh\left(-\fra+\epsilon _2+\phi _I\right) \sh\left(-\fra+\epsilon
   _1-\epsilon _3+\phi _I\right) \sh\left(-\fra-\epsilon _1+\epsilon _3+\phi
   _I\right)}{
   \sh\left(-\fra-2 \epsilon _1-\epsilon _2+\phi _I\right) \sh\left(-\fra-\epsilon _2-2 \epsilon _3+\phi
   _I\right) \sh\left(-\fra-\epsilon _1-\epsilon _3+\phi _I\right)}\\
   &\times \frac{\sh\left(\fra+2 \epsilon _1+\epsilon _2+\epsilon _3-\phi _I\right) \sh\left(\fra+\epsilon _1+\epsilon
   _2+2 \epsilon _3-\phi _I\right)}{\sh\left(\fra+\epsilon _1-\phi _I\right) \sh\left(\fra+\epsilon _3-\phi _I\right)}.
\eea
Let us set the reference vector to be $\tilde{\eta}_{0}$.

\begin{itemize}
    \item \textbf{Level 1}: The reference vector picks up the poles 
    \bea
    \fra+\eps_1-\phi_1=0,\quad \fra+\eps_3-\phi_{1}=0
    \eea
    giving the configurations
    \bea
    (\eps_1),\quad (\eps_3).
    \eea
    The transition rule is
    \begin{empheq}[box=\fbox]{align}
\begin{split}
\text{vac}&\longrightarrow (\eps_1)+(\eps_3)
\end{split}
\end{empheq}
        
    \item \textbf{Level 2}: Moving on to the level two, the poles picked up come from the charge vectors $-e_{I},e_{I}-e_{J}\,(I,J=1,2)$. The configurations with non-zero JK-residues are 
    \bea
(\eps_3,\eps_3-\eps_2),\quad (\eps_3,\eps_1),\quad (\eps_3,0),\quad (\eps_1,\eps_1-\eps_2),\quad (\eps_1,0).
    \eea
    The transition rules are
\begin{empheq}[box=\fbox]{align}
\begin{split}
(\eps_3)&\longrightarrow (\eps_3,\eps_3-\eps_2)+(\eps_3,\eps_1)+(\eps_3,0)\\
(\eps_1)&\longrightarrow (\eps_1,\eps_1-\eps_2)+(\eps_1,\eps_3)+(\eps_1,0)
\end{split}
\end{empheq}
Starting from $(\eps_3)$, using the hyperplanes $-\phi_{2}+\phi_1=\eps_{1,2,3,4}$, one can extend the crystal in the negative direction. The poles coming from $-\phi_2+\phi_1=\eps_{1,4}$ cancel with the numerator and such contribution does not appear. Note that during this process, no degenerate poles appear. We also can consider taking the pole $-\phi_{2}+\fra+\eps_1=0$. It is not a degenerate pole and we can have the transition $(\eps_3)\rightarrow (\eps_3,\eps_1)$.

Using the symmetry, a similar discussion holds for the configuration $(\eps_1)$. Note here that the configuration $(\eps_3,\eps_1)$ can be obtained from either $(\eps_3)$ or $(\eps_1)$ because it is not a degenerate pole.

    \item \textbf{Level 3}: For the third level, after classification of poles, the non-zero JK residues come from the following configurations:
    \bea
\,&(\epsilon_3,\epsilon_3-\epsilon_2,\epsilon_1),\quad (\epsilon_3,\epsilon_3-\epsilon_2,0),\quad (\epsilon_3,\epsilon_3-\epsilon_2,\epsilon_3-2\epsilon_2)\\
&(\epsilon_3,\epsilon_1,\epsilon_1-\epsilon_2),\quad (\epsilon_1,\epsilon_1-\epsilon_2,0),\quad (\epsilon_1,\epsilon_1-\epsilon_2,\epsilon_1-2\epsilon_2)\\
 &(\epsilon_3,\epsilon_1,0),\,\quad 
 (\epsilon_3,0,-\epsilon_3) ,\quad (\epsilon_1,0,-\epsilon_1)
    \eea
    The transition rules are summarized as
    \begin{empheq}[box=\fbox]{align}
\begin{split}
(\epsilon_3,\epsilon_3-\epsilon_2)&\longrightarrow  (\epsilon_3,\epsilon_3-\epsilon_2,\epsilon_1) + (\epsilon_3,\epsilon_3-\epsilon_2,0) +(\epsilon_3,\epsilon_3-\epsilon_2,\epsilon_3-2\epsilon_2) \\ 
 (\epsilon_1,\epsilon_1-\epsilon_2) &\longrightarrow (\epsilon_3,\epsilon_1,\epsilon_1-\epsilon_2) + (\epsilon_1,\epsilon_1-\epsilon_2,0) +(\epsilon_1,\epsilon_1-\epsilon_2,\epsilon_1-2\epsilon_2) \\ 
(\epsilon_3,\epsilon_1) &\longrightarrow (\epsilon_3,\epsilon_1,\epsilon_3-\epsilon_2)+ (\epsilon_3,\epsilon_1,\epsilon_1-\epsilon_2)+(\epsilon_3,\epsilon_1,0)\\ 
(\epsilon_3,0) &\longrightarrow  (\epsilon_3,0,\epsilon_3-\epsilon_2) +(\epsilon_3,0,-\epsilon_3) \\ 
(\epsilon_1,0) &\longrightarrow  (\epsilon_1,0,\epsilon_1-\epsilon_2) +(\epsilon_1,0,-\epsilon_1)
\end{split}
\end{empheq}

The configuration $(\eps_3,\eps_1,0)$ is the only nontrivial part. The configuration $(\eps_3,\eps_1)$ is obtained from the hyperplanes
\bea
\fra+\eps_1-\phi_1=0,\quad \fra+\eps_3-\phi_2=0.
\eea
The hyperplanes 
\bea
-\phi_3+\phi_1=\eps_1,\quad -\phi_3+\phi_2=\eps_3
\eea
intersect at $\phi_3=0$ and thus it is a second order degenerate pole. 
\end{itemize}

\paragraph{JK-residues}
For level one, we have
\bea
 (\epsilon _1), &\qquad  -\frac{\sh\left(\epsilon _1+\epsilon _2\right)
   \sh\left(\epsilon _3-2 \epsilon _1\right) \sh\left(\epsilon _1+\epsilon _3\right)
   \sh\left(\epsilon _2+\epsilon _3\right) \sh\left(\epsilon _2+2 \epsilon
   _3\right)}{\sh\left(\epsilon _1\right) \sh\left(\epsilon _2\right)
   \sh\left(\epsilon _3\right) \sh\left(\epsilon _3-\epsilon _1\right)
   \sh\left(-\epsilon _1+\epsilon _2+2 \epsilon _3\right)}, \\
 (\epsilon _3), &\qquad  \frac{\sh\left(\epsilon _1+\epsilon _2\right) \sh\left(2
   \epsilon _1+\epsilon _2\right) \sh\left(\epsilon _1+\epsilon _3\right)
   \sh\left(\epsilon _2+\epsilon _3\right) \sh\left(2 \epsilon _3-\epsilon
   _1\right)}{\sh\left(\epsilon _1\right) \sh\left(\epsilon _2\right)
   \sh\left(\epsilon _3\right) \sh\left(\epsilon _3-\epsilon _1\right) \sh\left(-2
   \epsilon _1-\epsilon _2+\epsilon _3\right)}. 
   \eea
For level two, we have
\bea
(0,\epsilon_{1}),&\qquad \frac{\sh\left(\epsilon _1+\epsilon _2\right) \sh\left(\epsilon _3-2 \epsilon _1\right)
   \sh\left(2 \epsilon _1+\epsilon _3\right) \sh\left(\epsilon _2+\epsilon _3\right)
   \sh\left(-\epsilon _1+\epsilon _2+\epsilon _3\right) \sh\left(\epsilon _1+\epsilon
   _2+2 \epsilon _3\right)}{2 \sh\left(\epsilon _1\right) \sh\left(2 \epsilon _1\right)
   \sh\left(\epsilon _2-\epsilon _1\right) \sh\left(\epsilon _3\right)^2
   \sh\left(-\epsilon _1+\epsilon _2+2 \epsilon _3\right)},\\
(0,\epsilon_{3}),&\qquad\frac{\sh\left(\epsilon _1+\epsilon _2\right) \sh\left(-\epsilon _1-\epsilon _2+\epsilon
   _3\right) \sh\left(\epsilon _2+\epsilon _3\right) \sh\left(2 \epsilon _1+\epsilon
   _2+\epsilon _3\right) \sh\left(2 \epsilon _3-\epsilon _1\right) \sh\left(\epsilon _1+2
   \epsilon _3\right)}{2 \sh\left(\epsilon _1\right)^2 \sh\left(\epsilon _3\right)
   \sh\left(2 \epsilon _3\right) \sh\left(\epsilon _3-\epsilon _2\right)
   \sh\left(-2 \epsilon _1-\epsilon _2+\epsilon _3\right)},\\
(\epsilon_{1},\epsilon_{1}-\epsilon_{2}),&\qquad -\frac{\sh\left(\epsilon _1+\epsilon _2\right) \sh\left(\epsilon _3-2 \epsilon _1\right)
   \sh\left(\epsilon _1+\epsilon _3\right) \sh\left(\epsilon _1-\epsilon _2+\epsilon
   _3\right) }{2 \sh\left(\epsilon _2\right) \sh\left(2 \epsilon _2\right)
   \sh\left(\epsilon _2-\epsilon _1\right) \sh\left(\epsilon _3\right)}\\
   &\qquad \times \frac{\sh\left(-2 \epsilon _1+\epsilon _2+\epsilon _3\right) \sh\left(2 \epsilon
   _2+\epsilon _3\right) \sh\left(\epsilon _2+2 \epsilon _3\right) \sh\left(2 \epsilon
   _2+2 \epsilon _3\right)}{
   \sh\left(\epsilon _3-\epsilon _1\right) \sh\left(-\epsilon _1+\epsilon _2+\epsilon
   _3\right) \sh\left(-\epsilon _1+\epsilon _2+2 \epsilon _3\right) \sh\left(-\epsilon
   _1+2 \epsilon _2+2 \epsilon _3\right)}, \\
(\epsilon_{1},\epsilon_{3}),&\qquad-\frac{\sh\left(2 \epsilon _1\right) \sh\left(\epsilon _1+\epsilon _2\right)^3
   \sh\left(2 \epsilon _3\right) \sh\left(\epsilon _1+\epsilon _3\right)^2
   \sh\left(\epsilon _2+\epsilon _3\right)^3}{2 \sh\left(\epsilon _1\right)^3
   \sh\left(\epsilon _2\right)^2 \sh\left(\epsilon _3\right)^3
   \sh\left(-\epsilon _1-\epsilon _2+\epsilon _3\right) \sh\left(-\epsilon _1+\epsilon
   _2+\epsilon _3\right)},\\
(\epsilon_{3},\epsilon_{3}-\epsilon_{2}),&\qquad\frac{\sh\left(2 \epsilon _1+\epsilon _2\right) \sh\left(\epsilon _1+2 \epsilon _2\right)
   \sh\left(2 \epsilon _1+2 \epsilon _2\right) \sh\left(\epsilon _1+\epsilon _3\right)}{2 \sh\left(\epsilon _1\right) \sh\left(\epsilon _2\right)
   \sh\left(2 \epsilon _2\right) \sh\left(\epsilon _3-\epsilon _1\right)}\\
   &\qquad \times \frac{
   \sh\left(\epsilon _1-\epsilon _2+\epsilon _3\right) \sh\left(\epsilon _2+\epsilon
   _3\right) \sh\left(2 \epsilon _3-\epsilon _1\right) \sh\left(-\epsilon _1-\epsilon
   _2+2 \epsilon _3\right)}{
   \sh\left(-2 \epsilon _1-2 \epsilon _2+\epsilon _3\right) \sh\left(\epsilon _3-\epsilon
   _2\right) \sh\left(-2 \epsilon _1-\epsilon _2+\epsilon _3\right) \sh\left(-\epsilon
   _1-\epsilon _2+\epsilon _3\right)}
\eea

Taking the unrefined limit, the nontrivial weight factor appears at the configuration $(\eps_1,\eps_3,0)$ where the second order degenerate pole appears. Only for this configuration, the weight is $2$.
\begin{itemize}
    \item \textbf{Level 1}:
    \bea
\mathcal{Z}_{(\eps_1)}\rightarrow 1,\quad \mathcal{Z}_{(\eps_3)}\rightarrow 1.
    \eea
    \item \textbf{Level 2}:
    \bea
 \mathcal{Z}_{(\eps_3,\eps_3-\eps_2)}\rightarrow 1,\quad \mathcal{Z}_{(\eps_3,\eps_1)}\rightarrow 1,\quad \mathcal{Z}_{(\eps_3,0)}\rightarrow 1,\quad \mathcal{Z}_{(\eps_1,\eps_1-\eps_2)}\rightarrow 1 ,\quad   \mathcal{Z}_{(\eps_1,0)}\rightarrow 1
    \eea
    \item \textbf{Level 3}:
    \bea
    \,&\mathcal{Z}_{(\epsilon_3,\epsilon_3-\epsilon_2,\epsilon_1)}\rightarrow 1,\quad \mathcal{Z}_{(\epsilon_3,\epsilon_3-\epsilon_2,0)}\rightarrow 1,\quad \mathcal{Z}_{(\epsilon_3,\epsilon_3-\epsilon_2,\epsilon_3-2\epsilon_2)}\rightarrow 1,\\
&\mathcal{Z}_{(\epsilon_3,\epsilon_1,\epsilon_1-\epsilon_2)}\rightarrow 1,\quad \mathcal{Z}_{(\epsilon_1,\epsilon_1-\epsilon_2,0)}\rightarrow 1,\quad \mathcal{Z}_{(\epsilon_1,\epsilon_1-\epsilon_2,\epsilon_1-2\epsilon_2)}\rightarrow 1,\\
 &\mathcal{Z}_{(\epsilon_3,\epsilon_1,0)}\rightarrow 2,\,\quad 
 \mathcal{Z}_{(\epsilon_3,0,-\epsilon_3)}\rightarrow 1 ,\quad \mathcal{Z}_{(\epsilon_1,0,-\epsilon_1)}\rightarrow 1.
    \eea
    
\end{itemize}

\paragraph{Unrefined vertex}
\bea
\wtC_{\yng(1),\,\yng(2,1),\,\yng(1)}(q)=\frac{1-2q+3q^2-3q^3+3q^4-3q^5+3q^6-2q^7+q^8}{(1-q)^{4}(1-q^3)}=1+2q+5q^2+10q^3+\cdots
\eea

\paragraph{Refined vertex}
\bea
\wtC_{\yng(1),\,\yng(2,1),\,\yng(1)}(t,q)=1+2t+3qt+2t^2+3q^2t+5qt^2+2t^3+\cdots
\eea

\paragraph{Macdonald refined vertex}

\bea
\wtM_{\yng(1),\,\yng(2,1),\,\yng(1)}(x,y;q,t)&=1+\frac{2(1-t)}{1-q}x\\
&+\left(\frac{(1-t)(2-t+q-2qt)}{(1-q)(1-q^2)}x^2+\frac{(1-t)(3+t-q+qt-q^2-qt^2-2q^2t)}{(1-q)(1-qt^2)}xy\right)+\cdots
\eea

\subsubsection[Example 4]{$(\lambda,\mu,\nu)=(\yng(1,1),\yng(1,1),\yng(1))$ (Fig.~\ref{fig:3leg-ex4})}
\begin{figure}[h]
    \centering
    \includegraphics[width=5cm]{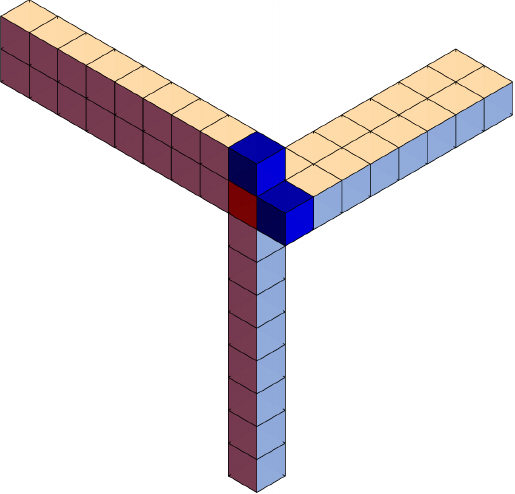}
    \caption{The PT box counting setup for $\lambda=\{1,1\},\mu=\{1,1\}, \nu=\{1\}$.}
    \label{fig:3leg-ex4}
\end{figure}
Let us study the case $\lambda=\yng(1,1),\mu=\yng(1,1),\nu=\yng(1)$ (see Fig.~\ref{fig:3leg-ex4}). The red position has the coordinate $0$ and the blue positions has $\eps_2, \epsilon_{3}$. The blue position $\eps_2$ is labeled 3, and the blue position $\eps_3$ is labeled 1.

\paragraph{PT box-counting rules}
\begin{itemize}
    \item \textbf{Level 1}
    \bea
    \ket{\eps_3},\quad \ket{\eps_2}
    \eea
    \item \textbf{Level 2}
    \bea
\ket{\eps_3,\eps_3-\eps_2},\quad \ket{\eps_3,\eps_2},\quad \ket{\eps_2,\eps_2-\eps_1},\quad \ket{0^3,\eps_3},\quad \ket{0^1,\eps_2}
    \eea
    \item \textbf{Level 3}
    \bea
    \ket{0^1,\eps_2,\eps_2-\eps_1},\quad \ket{\eps_3,\eps_3-\eps_2,\eps_3-2\eps_2},\quad \ket{\eps_3,\eps_2-\eps_1,\eps_2},\quad \ket{\eps_3,\eps_3-\eps_2,\eps_2},\\
    \ket{\eps_2-\eps_1,\eps_2,\eps_2-2\eps_1},\quad \ket{-\eps_3,0^3,\eps_3},\quad \ket{0^3,\eps_3,\eps_3-\eps_2},\quad \ket{0^0,\eps_3,\eps_2}
    \eea

\end{itemize}

\paragraph{GR box-counting rules}
\begin{itemize}
    \item \textbf{Level 0 $\rightarrow $ Level 1}:
    \bea
|\text{vac}\rangle \rightarrow |\epsilon_3\rangle + |\epsilon_2\rangle
\eea
    \item \textbf{Level 1 $\rightarrow $ Level 2}:
    \bea
|\epsilon_3\rangle &\rightarrow |\epsilon_3,\epsilon_3-\epsilon_2\rangle+|\epsilon_3,\epsilon_2\rangle +|\epsilon_3,0_H\rangle  \\
|\epsilon_2\rangle &\rightarrow |\epsilon_2,\epsilon_2-\epsilon_1\rangle +|\epsilon_3,\epsilon_2\rangle+|\epsilon_2,0_H\rangle 
\eea
    \item \textbf{Level 2 $\rightarrow $ Level 3}:
    \bea
|\epsilon_3,\epsilon_3-\epsilon_2\rangle&\rightarrow |\epsilon_3,\epsilon_3-\epsilon_2,\epsilon_3-2\epsilon_2\rangle  +|\epsilon_3,\epsilon_2,\epsilon_3-\epsilon_2\rangle +|\epsilon_3,\epsilon_3-\epsilon_2,0_H\rangle\\ 
 |\epsilon_2,\epsilon_2-\epsilon_1\rangle &\rightarrow |\epsilon_2,\epsilon_2-\epsilon_1,\epsilon_2-2\epsilon_1\rangle + |\epsilon_3,\epsilon_2,\epsilon_2-\epsilon_1\rangle + |\epsilon_2,\epsilon_2-\epsilon_1,0_H\rangle\\ 
|\epsilon_3,\epsilon_2\rangle  &\rightarrow |\epsilon_3,\epsilon_2,\epsilon_3-\epsilon_2\rangle +|\epsilon_3,\epsilon_2,\epsilon_2-\epsilon_1\rangle +|\epsilon_3,\epsilon_2,0_L\rangle  +|\epsilon_3,\epsilon_2,0_H\rangle \\ 
|\epsilon_3,0_H\rangle &\rightarrow|\epsilon_3,0_H,-\epsilon_3\rangle + |\epsilon_3,\epsilon_3-\epsilon_2,0_H\rangle + |\epsilon_3,\epsilon_2,0_L\rangle \\ 
|\epsilon_2,0_H\rangle  &\rightarrow|\epsilon_2,\epsilon_2-\epsilon_1,0_H\rangle + |\epsilon_3,\epsilon_2,0_L\rangle
\eea
\end{itemize}

\paragraph{JK prescription and transition rules}
The framing node contribution is
\bea
\mathcal{Z}^{\D6_{\bar{4}}\tbar\D2\tbar\D0}_{\DT;\,\yng(1,1)\,\yng(1,1)\,\yng(1)}(\fra,\phi_{I})&=\frac{\sh\left(-\fra+\epsilon _2+\phi _I\right) \sh\left(-\fra+\epsilon
   _1-\epsilon _3+\phi _I\right) \sh\left(-\fra-\epsilon _2+\epsilon _3+\phi
   _I\right)}{
   \sh\left(-\fra-\epsilon _1-2 \epsilon _2+\phi _I\right) \sh\left(-\fra-\epsilon _2-2 \epsilon _3+\phi
   _I\right) \sh\left(-\fra-\epsilon _1-\epsilon _3+\phi _I\right)}\\
   &\times \frac{\sh\left(\fra+\epsilon _1+2 \epsilon _2+\epsilon _3-\phi _I\right) \sh\left(\fra+\epsilon _1+\epsilon
   _2+2 \epsilon _3-\phi _I\right) }{\sh\left(\fra+\epsilon _2-\phi _I\right) \sh\left(\fra+\epsilon _3-\phi _I\right)}.
\eea

\begin{itemize}
    \item \textbf{Level 1}: Choosing the reference vector $\tilde{\eta}_0$, the poles picked up are
    \bea
{\fra+\eps_2-\phi_1=0,\quad \fra+\eps_3-\phi_1=0}
    \eea
    and thus the transition rule is
    \bea
    \boxed{\text{vac}\longrightarrow (\eps_2)+(\eps_3)}
    \eea
    \item \textbf{Level 2}: For level two, the poles picked up come from the charge vectors $-e_{I},e_{I}-e_{J}$. The configurations with non-zero JK-residues are
    \bea
    (\eps_3,\eps_3-\eps_2),\quad (\eps_3,\eps_2),\quad  (\eps_3,0),\quad (\eps_2,\eps_2-\eps_1),\quad (\eps_2,0).
    \eea
    Similar to the previous example, no degenerate poles appear during this process and thus the transition rules are given as
    \begin{empheq}[box=\fbox]{align}
\begin{split}
(\eps_3)&\longrightarrow (\eps_3,\eps_3-\eps_2)+(\eps_3,\eps_2)+(\eps_3,0)\\
(\eps_2)&\longrightarrow (\eps_2,\eps_2-\eps_1)+(\eps_3,\eps_2)+(\eps_2,0)
    \end{split}
    \end{empheq}
    \item \textbf{Level 3}: For the third level, the non-zero JK residues come from the configurations
    \bea
(\eps_3,\eps_3-\eps_2,\eps_3-2\eps_2),\quad (\eps_3,\eps_3-\eps_2,\eps_2),\quad (\eps_3,\eps_3-\eps_{2},0),\quad (\eps_3,\eps_2,0),\\
 (\eps_2,\eps_2-\eps_1,\eps_2-2\eps_1),\quad (\eps_2,\eps_2-\eps_1,\eps_3),\quad (\eps_2,\eps_2-\eps_{1},0),\quad (\eps_3,0,-\eps_3).
    \eea
    The transition rules are summarized as
     \begin{empheq}[box=\fbox]{align}
\begin{split}
(\eps_3,\eps_3-\eps_2)&\longrightarrow (\eps_3,\eps_3-\eps_2,\eps_3-2\eps_2)+(\eps_3,\eps_3-\eps_2,\eps_2)+(\eps_3,\eps_3-\eps_{2},0)\\
(\eps_2,\eps_2-\eps_1)&\longrightarrow (\eps_2,\eps_2-\eps_1,\eps_2-2\eps_1)+(\eps_2,\eps_2-\eps_1,\eps_3)+(\eps_2,\eps_2-\eps_{1},0)\\
(\eps_3,\eps_2)&\longrightarrow (\eps_3,\eps_2,\eps_3-\eps_2)+(\eps_3,\eps_2,\eps_2-\eps_1)+ (\eps_3,\eps_2,0)\\
(\eps_3,0)&\longrightarrow (\eps_3,0,-\eps_3)+(\eps_3,0,\eps_3-\eps_2)\\
(\eps_2,0)&\longrightarrow (\eps_2,0,\eps_2-\eps_1)
\end{split}
\end{empheq}

The nontrivial configuration is $(\eps_3,\eps_2,0)$. The hyperplanes related with this pole are
\bea
\fra+\eps_3-\phi_1=0,\quad \fra+\eps_2-\phi_2=0,\quad \phi_2-\phi_3=\eps_2,\quad \phi_1-\phi_3=\eps_3
\eea
and it is a degenerate pole with second order.

\end{itemize}

\paragraph{JK-residues}
For level one, we have
\bea
 (\epsilon _2), &\quad  -\frac{\sh\left(2 \epsilon _2\right) \sh\left(\epsilon
   _1+\epsilon _3\right) \sh\left(-\epsilon _1-\epsilon _2+\epsilon _3\right)
   \sh\left(\epsilon _2+\epsilon _3\right) \sh\left(\epsilon _1+2 \epsilon
   _3\right)}{\sh\left(\epsilon _1\right) \sh\left(\epsilon _2\right) \sh\left(2
   \epsilon _3\right) \sh\left(\epsilon _3-\epsilon _2\right) \sh\left(\epsilon
   _1-\epsilon _2+\epsilon _3\right)}, \\
 (\epsilon _3), &\quad  \frac{\sh\left(\epsilon _1+\epsilon _2\right)
   \sh\left(\epsilon _1+2 \epsilon _2\right) \sh\left(\epsilon _1+\epsilon _3\right)
   \sh\left(\epsilon _2+\epsilon _3\right) \sh\left(2 \epsilon _3-\epsilon
   _2\right)}{\sh\left(\epsilon _1\right) \sh\left(\epsilon _2\right)
   \sh\left(\epsilon _3\right) \sh\left(-\epsilon _1-2 \epsilon _2+\epsilon _3\right)
   \sh\left(\epsilon _3-\epsilon _2\right)} .
\eea
For level two, we have{\small
\bea
 (0, \epsilon _2),&\qquad -\frac{\sh\left(\epsilon _3-\epsilon _1\right) \sh\left(-\epsilon _1-\epsilon _2+\epsilon
   _3\right) \sh\left(\epsilon _2+\epsilon _3\right) \sh\left(2 \epsilon _2+\epsilon
   _3\right) \sh\left(\epsilon _1+2 \epsilon _3\right) \sh\left(\epsilon _1+\epsilon _2+2
   \epsilon _3\right)}{2 \sh\left(\epsilon _1\right) \sh\left(\epsilon _2-\epsilon
   _1\right) \sh\left(\epsilon _3\right) \sh\left(2 \epsilon _3\right)
   \sh\left(\epsilon _3-\epsilon _2\right) \sh\left(\epsilon _2+2 \epsilon _3\right)},\\
 (0 , \epsilon _3 ),&\qquad\frac{\sh\left(\epsilon _1+\epsilon _2\right) \sh\left(-\epsilon _1-\epsilon _2+\epsilon
   _3\right) \sh\left(\epsilon _2+\epsilon _3\right) \sh\left(\epsilon _1+2 \epsilon
   _2+\epsilon _3\right) \sh\left(\epsilon _1+2 \epsilon _3\right) \sh\left(2 \epsilon
   _3-\epsilon _2\right)}{2 \sh\left(\epsilon _1\right) \sh\left(\epsilon _2\right)
   \sh\left(\epsilon _3\right) \sh\left(2 \epsilon _3\right) \sh\left(-\epsilon
   _1-2 \epsilon _2+\epsilon _3\right) \sh\left(\epsilon _3-\epsilon _2\right)},\\
 (\epsilon _2 ,\epsilon _2-\epsilon _1) ,&\qquad{\frac{\sh\left(2 \epsilon _2\right) \sh\left(2 \epsilon _2-\epsilon _1\right)
   \sh\left(\epsilon _1+\epsilon _3\right) \sh\left(2 \epsilon _1+\epsilon _3\right)
   \sh\left(-\epsilon _1-\epsilon _2+\epsilon _3\right) \sh\left(\epsilon _2+\epsilon
   _3\right) \sh\left(-\epsilon _1+\epsilon _2+\epsilon _3\right) \sh\left(2 \epsilon
   _1+2 \epsilon _3\right)}{2 \sh\left(\epsilon _1\right) \sh\left(2 \epsilon _1\right)
   \sh\left(\epsilon _2\right) \sh\left(\epsilon _2-\epsilon _1\right) \sh\left(2
   \epsilon _3\right) \sh\left(\epsilon _1-\epsilon _2+\epsilon _3\right)^2 \sh\left(2
   \epsilon _1-\epsilon _2+\epsilon _3\right)}},\\
 (\epsilon _2 , \epsilon _3 ),&\qquad -\frac{\sh\left(2 \epsilon _2\right)^2 \sh\left(\epsilon _1+\epsilon _2\right)^2
   \sh\left(\epsilon _1+\epsilon _3\right)^3 \sh\left(\epsilon _2+\epsilon
   _3\right)^2 \sh\left(\epsilon _1-\epsilon _2+2 \epsilon _3\right)}{2
   \sh\left(\epsilon _1\right)^2 \sh\left(\epsilon _2\right)^3
   \sh\left(\epsilon _3\right)^2 \sh\left(\epsilon _3-2 \epsilon _2\right)
   \sh\left(\epsilon _1-\epsilon _2+\epsilon _3\right)^2},\\
 (\epsilon _3 , \epsilon _3-\epsilon _2 ),&\qquad\frac{\sh\left(\epsilon _1+2 \epsilon _2\right)^2 \sh\left(\epsilon _1+3 \epsilon
   _2\right) \sh\left(\epsilon _1+\epsilon _3\right) \sh\left(\epsilon _1-\epsilon
   _2+\epsilon _3\right) \sh\left(\epsilon _2+\epsilon _3\right) \sh\left(2 \epsilon _3-2
   \epsilon _2\right) \sh\left(2 \epsilon _3-\epsilon _2\right)}{2 \sh\left(\epsilon
   _1\right) \sh\left(\epsilon _2\right) \sh\left(2 \epsilon _2\right)
   \sh\left(-\epsilon _1-3 \epsilon _2+\epsilon _3\right) \sh\left(\epsilon _3-2 \epsilon
   _2\right) \sh\left(-\epsilon _1-2 \epsilon _2+\epsilon _3\right) \sh\left(\epsilon
   _3-\epsilon _2\right)^2}.
\eea}

The weight factors are given as follows.
\begin{itemize}
    \item \textbf{Level 1}:
    \bea
    \mathcal{Z}_{(\eps_2)}\rightarrow 1,\quad 
    \mathcal{Z}_{(\eps_3)}\rightarrow 1.
    \eea
    \item \textbf{Level 2}:
    \bea
 \mathcal{Z}_{(\eps_3,\eps_3-\eps_2)}\rightarrow 1,\quad \mathcal{Z}_{(\eps_3,\eps_2)}\rightarrow 1,\quad  \mathcal{Z}_{(\eps_3,0)}\rightarrow 1,\quad \mathcal{Z}_{(\eps_2,\eps_2-\eps_1)}\rightarrow 1,\quad \mathcal{Z}_{(\eps_2,0)}\rightarrow 1
    \eea
    \item \textbf{Level 3}:
    \bea
    \mathcal{Z}_{(\eps_3,\eps_3-\eps_2,\eps_3-2\eps_2)}\rightarrow 1,\quad \mathcal{Z}_{(\eps_3,\eps_3-\eps_2,\eps_2)}\rightarrow 1,\quad \mathcal{Z}_{(\eps_3,\eps_3-\eps_{2},0)}\rightarrow 2,\quad \mathcal{Z}_{(\eps_3,\eps_2,0)}\rightarrow 1,\\
 \mathcal{Z}_{(\eps_2,\eps_2-\eps_1,\eps_2-2\eps_1)}\rightarrow 1,\quad \mathcal{Z}_{(\eps_2,\eps_2-\eps_1,\eps_3)}\rightarrow 1,\quad \mathcal{Z}_{(\eps_2,\eps_2-\eps_{1},0)}\rightarrow 1,\quad \mathcal{Z}_{(\eps_3,0,-\eps_3)}\rightarrow 1.
    \eea
\end{itemize}

\paragraph{Unrefined vertex}
\bea
\wtC_{\yng(2),\,\yng(2),\,\yng(1)}(q)=\frac{1-q+q^3-q^4+q^5-q^7+q^8}{(1-q)^{3}(1-q^2)^2}=1+2q+5q^2+9q^3+\cdots
\eea

\paragraph{Refined vertex}
\bea
\wtC_{\yng(2),\,\yng(2),\,\yng(1)}(t,q)=1+2t+2qt+3t^2+q^2t+4qt^2+4t^3+\cdots
\eea

\paragraph{Macdonald refined vertex}

\bea
\wtM_{\yng(2),\,\yng(2),\,\yng(1)}(x,y;q,t)&=1+\frac{2-t+q-2qt}{1-q^{2}}x+\left(\frac{(1-qt)(3-2t-q)}{(1-q)(1-q^2)}x^{2}+2xy\right)+\cdots
\eea

\subsubsection[Example 5]{ $(\lambda,\mu,\nu)=(\yng(1,1),\yng(2),\yng(1))$ (Fig.~\ref{fig:3leg-ex5})}
\begin{figure}[h]
    \centering
    \includegraphics[width=5cm]{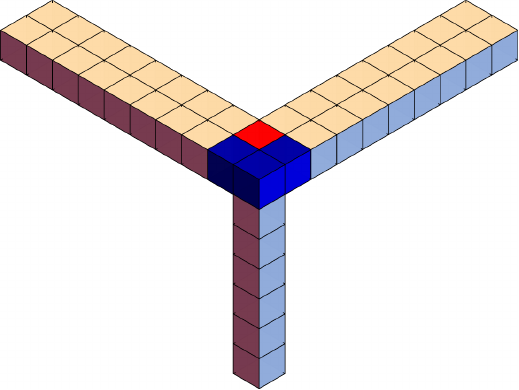}
    \caption{The PT box counting setup for $\lambda=\{1,1\},\mu=\{2\}, \nu=\{1\}$.}
    \label{fig:3leg-ex5}
\end{figure}
Let us study the case $\lambda=\yng(1,1),\mu=\yng(2),\nu=\yng(1)$ (see Fig.~\ref{fig:3leg-ex5}). The red position has the coordinate $0$ and the blue positions has $\epsilon_{1},\eps_2,\eps_1+\eps_2$. The blue positions are all labeled $3$.

\paragraph{PT box-counting rules}
\begin{itemize}
\item \textbf{Level 1}
\bea
\ket{\eps_1+\eps_2},\quad \ket{0^3}
\eea
    \item \textbf{Level 2}
    \bea
    \ket{\eps_1,\eps_1+\eps_2},\quad \ket{\eps_2,\eps_1+\eps_2},\quad\ket{0^3,-\eps_3},\quad \ket{0^3,\eps_1+\eps_2}
    \eea
    \item \textbf{Level 3}
    \bea
\ket{\eps_1,\eps_2,\eps_1+\eps_2},\quad \ket{-\eps_3,0^3,\eps_1+\eps_2},\quad \ket{\eps_1,\eps_1-\eps_2,\eps_1+\eps_2},\quad \ket{-\eps_3,0^3,-2\eps_3},\\
\ket{\eps_1,0^3,\eps_1+\eps_2},\quad  \ket{\eps_1+\eps_2,-\eps_1+\eps_2,\eps_2},\quad \ket{0^3,\eps_1+\eps_2,\eps_2}
    \eea

\end{itemize}

\paragraph{GR box-counting rules}
\begin{itemize}
    \item \textbf{Level 0 $\rightarrow $ Level 1}:
    \begin{eqnarray}
|\text{vac}\rangle \rightarrow |\epsilon_1+\epsilon_2\rangle + |0_H\rangle
\end{eqnarray}
    \item \textbf{Level 1 $\rightarrow $ Level 2}:
    \bea
|\epsilon_1+\epsilon_2\rangle&\rightarrow |\epsilon_1+\epsilon_2,0_H\rangle+ |\epsilon_1+\epsilon_2,\epsilon_1\rangle + |\epsilon_1+\epsilon_2,\epsilon_2\rangle \\
|0_H\rangle &\rightarrow  |\epsilon_1+\epsilon_2,0_H\rangle +|0_H,-\epsilon_3\rangle
\eea
    
    \item \textbf{Level 2 $\rightarrow $ Level 3}:
    \bea
 |\epsilon_1+\epsilon_2,0_H\rangle&\rightarrow|\epsilon_1+\epsilon_2,\epsilon_1,0_H\rangle+  |\epsilon_1+\epsilon_2,\epsilon_2,0_H\rangle +  |\epsilon_1+\epsilon_2,0_H,-\epsilon_3\rangle\\ 
  |\epsilon_1+\epsilon_2,\epsilon_1\rangle &\rightarrow   |\epsilon_1+\epsilon_2,\epsilon_1,\epsilon_2\rangle + |\epsilon_1+\epsilon_2,\epsilon_1,\epsilon_1-\epsilon_2\rangle  +|\epsilon_1+\epsilon_2,\epsilon_1,0_H\rangle\\ 
 |\epsilon_1+\epsilon_2,\epsilon_2\rangle &\rightarrow   |\epsilon_1+\epsilon_2,\epsilon_1,\epsilon_2\rangle +  |\epsilon_1+\epsilon_2,\epsilon_2,\epsilon_2-\epsilon_1\rangle + |\epsilon_1+\epsilon_2,\epsilon_2,0_H\rangle\\ 
|0_H,-\epsilon_3\rangle &\rightarrow |0_H,-\epsilon_3,-2\epsilon_3\rangle + |0_H,-\epsilon_3,\epsilon_1+\epsilon_2\rangle 
\eea
\end{itemize}

\paragraph{JK prescription and transition rules}
The framing node contribution is
\bea
\mathcal{Z}^{\D6_{\bar{4}}\tbar\D2\tbar\D0}_{\DT;\,\yng(1,1)\,\yng(2)\,\yng(1)}(\fra,\phi_{I})&=\frac{ \sh\left(-\fra+\epsilon _1+\phi _I\right) \sh\left(-\fra+\epsilon
   _2+\phi _I\right) \sh\left(-\fra-\epsilon _1-\epsilon _2+\epsilon _3+\phi _I\right)}{ \sh\left(-\fra-2 \epsilon _1-2 \epsilon
   _2+\phi _I\right) \sh\left(-\fra-\epsilon _1-\epsilon _3+\phi _I\right) \sh\left(-\fra-\epsilon
   _2-\epsilon _3+\phi _I\right)}\\
   &\times \frac{\sh\left(\fra+\epsilon _1+\epsilon _2+\epsilon _3-\phi _I\right) \sh\left(\fra+2 \epsilon _1+2
   \epsilon _2+\epsilon _3-\phi _I\right)}{\sh\left(\fra-\phi
   _I\right) \sh\left(\fra+\epsilon _1+\epsilon _2-\phi _I\right)}.
\eea

\begin{itemize}
    \item \textbf{Level 1}: The reference vector $\tilde{\eta}_0$ picks up the poles
    \bea
    \fra-\phi_1=0,\quad \fra+\eps_1+\eps_2-\phi_1=0.
    \eea
    The transition rule is
    \begin{empheq}[box=\fbox]{align}
\begin{split}
\text{vac}&\longrightarrow (0)+(\eps_1+\eps_2)
\end{split}
\end{empheq}
    \item \textbf{Level 2}: The configurations giving non-zero JK residues are
    \bea
    (\eps_1+\eps_2,\eps_1),\quad (\eps_1+\eps_2,\eps_2),\quad (0,-\eps_3).
    \eea
    The transition rule is determined as
    \begin{empheq}[box=\fbox]{align}
\begin{split}
(\eps_1+\eps_2)&\longrightarrow (\eps_1+\eps_2,\eps_1)+(\eps_1+\eps_2,\eps_2)\\
(0)&\longrightarrow (0,-\eps_3)
\end{split}
\end{empheq}
For this case, no configurations are neither degenerate poles nor higher order poles.\footnote{Even so, somehow for the configuration $(0,-\eps_3)$, the weight is $2$ in the unrefined limit. }

We note that naively the JK-prescription also picks up the pole configuration $(0,\eps_1+\eps_2)$, but due to the numerator part $\sh(\phi_2-\phi_1-\eps_1-\eps_2)$ coming from the D0-D0 contribution, the residue is zero. Due to this, somehow the PT, GR box counting rules are not compatible. We do not know how to deal with this mismatch for the moment and leave a detailed analysis for future work. Using our formalism, the relation with the topological vertices and the DT/PT correspondence all seem to holds.

    \item \textbf{Level 3}: For the third level, the nontrivial configurations are
    \bea
   & (\eps_1+\eps_2,\eps_1,\eps_2),\quad (\eps_1+\eps_2,\eps_1,\eps_1-\eps_2),\quad (\eps_1+\eps_2,\eps_1,0),\quad (\eps_1+\eps_2,\eps_2,\eps_2-\eps_1),\\
   & (\eps_1+\eps_2,\eps_2,0),\quad (0,-\eps_3,\eps_1+\eps_2),\quad (0,-\eps_3,-2\eps_3).
    \eea
    The transition rules are determined as
  \begin{empheq}[box=\fbox]{align}
\begin{split}
(\eps_1+\eps_2,\eps_1)&\longrightarrow (\eps_1+\eps_2,\eps_1,\eps_2)+(\eps_1+\eps_2,\eps_1,\eps_1-\eps_2)+(\eps_1+\eps_2,\eps_1,0)\\
(\eps_1+\eps_2,\eps_2)&\longrightarrow (\eps_1+\eps_2,\eps_2,\eps_1)+(\eps_1+\eps_2,\eps_2,\eps_2-\eps_1)+(\eps_1+\eps_2,\eps_2,0)\\
(0,-\eps_3)&\longrightarrow (0,-\eps_3,\eps_1+\eps_2)+(0,-\eps_3,-2\eps_3)
\end{split}
\end{empheq}
All of the configurations are non-degenerate pole configurations. For the configuration $(0,-\eps_3,\eps_1+\eps_2)$, after evaluating the poles at $\phi_1=\fra,\phi_2=\fra-\eps_3$, the pole at $\phi_3=\fra+\eps_{1}+\eps_2$ is not canceled and so such kind of configuration is allowed.

\end{itemize}

\paragraph{JK-residues}
For level one, we have
\bea
(0),&\qquad \frac{\sh\left(-\epsilon _1-\epsilon _2+\epsilon _3\right)
   \sh\left(2 \epsilon _1+2 \epsilon _2+\epsilon
   _3\right)}{\sh\left(2 \epsilon _1+2 \epsilon _2\right)
   \sh\left(\epsilon _3\right)},\\
(\epsilon_{1}+\epsilon_{2}),&\qquad -\frac{\sh\left(2 \epsilon _1+\epsilon _2\right) \sh\left(\epsilon
   _1+2 \epsilon _2\right) \sh\left(\epsilon _3\right)
   \sh\left(\epsilon _1+\epsilon _3\right) \sh\left(\epsilon
   _2+\epsilon _3\right)}{\sh\left(\epsilon _1\right)
   \sh\left(\epsilon _2\right) \sh\left(\epsilon _1+\epsilon
   _2\right) \sh\left(\epsilon _3-\epsilon _1\right)
   \sh\left(\epsilon _3-\epsilon _2\right)}.
\eea
For level two, we have
\bea
(0,-\epsilon_{3}),&\qquad \frac{\sh\left(\epsilon _1+\epsilon _2\right) \sh\left(-\epsilon
   _1-\epsilon _2+\epsilon _3\right)^2 \sh\left(2 \epsilon _1+2 \epsilon
   _2+2 \epsilon _3\right)}{2 \sh\left(2 \epsilon _1+2 \epsilon _2\right)
   \sh\left(\epsilon _3\right) \sh\left(2 \epsilon _3\right)
   \sh\left(\epsilon _1+\epsilon _2+\epsilon _3\right)},\\
(\epsilon_{1},\epsilon_{1}+\epsilon_{2}),&\qquad -\frac{\sh\left(2 \epsilon _1\right) \sh\left(2 \epsilon _1+\epsilon
   _2\right) \sh\left(\epsilon _1+2 \epsilon _2\right)
   \sh\left(\epsilon _1+\epsilon _3\right) \sh\left(\epsilon
   _1-\epsilon _2+\epsilon _3\right) \sh\left(\epsilon _2+\epsilon
   _3\right)^2 \sh\left(2 \epsilon _2+\epsilon _3\right)}{2
   \sh\left(\epsilon _1\right)^2 \sh\left(\epsilon _2\right)
   \sh\left(2 \epsilon _2\right) \sh\left(\epsilon _2-\epsilon
   _1\right) \sh\left(\epsilon _3-\epsilon _1\right)
   \sh\left(\epsilon _3-\epsilon _2\right) \sh\left(-\epsilon
   _1+\epsilon _2+\epsilon _3\right)}, \\
(\epsilon_{2},\epsilon_{1}+\epsilon_{2}),&\qquad \frac{\sh\left(2 \epsilon _2\right) \sh\left(2 \epsilon _1+\epsilon
   _2\right) \sh\left(\epsilon _1+2 \epsilon _2\right)
   \sh\left(\epsilon _1+\epsilon _3\right)^2 \sh\left(2 \epsilon
   _1+\epsilon _3\right) \sh\left(\epsilon _2+\epsilon _3\right)
   \sh\left(-\epsilon _1+\epsilon _2+\epsilon _3\right)}{2
   \sh\left(\epsilon _1\right) \sh\left(2 \epsilon _1\right)
   \sh\left(\epsilon _2\right)^2 \sh\left(\epsilon _2-\epsilon
   _1\right) \sh\left(\epsilon _3-\epsilon _1\right)
   \sh\left(\epsilon _3-\epsilon _2\right) \sh\left(\epsilon
   _1-\epsilon _2+\epsilon _3\right)}.
\eea

The weight factors are summarized as
\begin{itemize}
    \item \textbf{Level 1}:
    \bea
\mathcal{Z}_{(0)}\rightarrow 1,\quad \mathcal{Z}_{(\eps_1+\eps_2)}\rightarrow 1
    \eea
    \item \textbf{Level 2}:
    \bea
\mathcal{Z}_{(\eps_1+\eps_2,\eps_1)}\rightarrow 1,\quad \mathcal{Z}_{(\eps_1+\eps_2,\eps_2)}\rightarrow 1,\quad \mathcal{Z}_{(0,-\eps_3)}\rightarrow 2.
    \eea
    \item \textbf{Level 3}:
    \bea
 & \mathcal{Z}_{(\eps_1+\eps_2,\eps_1,\eps_2)}\rightarrow 1,\quad \mathcal{Z}_{(\eps_1+\eps_2,\eps_1,\eps_1-\eps_2)}\rightarrow 1,\quad \mathcal{Z}_{(\eps_1+\eps_2,\eps_1,0)}\rightarrow 1,\quad \mathcal{Z}_{(\eps_1+\eps_2,\eps_2,\eps_2-\eps_1)}\rightarrow 1,\\
   & \mathcal{Z}_{(\eps_1+\eps_2,\eps_2,0)}\rightarrow 1,\quad \mathcal{Z}_{(0,-\eps_3,\eps_1+\eps_2)}\rightarrow 1,\quad \mathcal{Z}_{(0,-\eps_3,-2\eps_3)}\rightarrow 1.
    \eea
\end{itemize}

\paragraph{Unrefined vertex}
\bea
\wtC_{\yng(2)\,\yng(1,1)\,\yng(1)}(q)&=\frac{1-q-q^2+2q^3-q^5+q^6-q^8+q^9}{(1-q^2)^2(1-q)^3}\\
&=1+2q+4q^2+7q^3+12q^4+\cdots.
\eea
\paragraph{Refined vertex}
\bea
\wtC_{\yng(2)\,\yng(1,1)\,\yng(1)}(t,q)&=1+(q+t)+(3qt+t^2)+(q^2t+5qt^2+t^3)+\cdots.
\eea
\paragraph{Macdonald refined vertex}
\bea
\wtM_{\yng(2)\,\yng(1,1)\,\yng(1)}(x,y;q,t)=1+\left(\frac{1-t}{1-q}x+\frac{1-qt}{1-t^{2}}y\right)+\left(\frac{(1-t)(1-qt)}{(1-q)(1-q^2)}x^2+\frac{(3qt+2q-2t-3)}{(q-1)(t+1)}xy\right)+\cdots
\eea

\subsubsection[Example 6]{$(\lambda,\mu,\nu)=(\yng(2),\yng(1,1),\yng(1))$ (Fig.~\ref{fig:3leg-ex6})}
\begin{figure}[h]
    \centering
    \includegraphics[width=5cm]{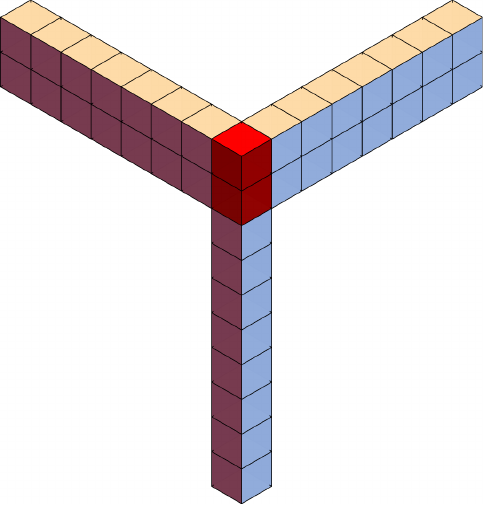}
    \caption{The PT box counting setup for $\lambda=\{2\},\mu=\{1,1\}, \nu=\{1\}$.}
    \label{fig:3leg-ex6}
\end{figure}
For the case $\lambda=\yng(2),\mu=\yng(1,1),\nu=\yng(1)$ (see Fig.~\ref{fig:3leg-ex6}), there are two red positions at $0,\epsilon_{3}$.

\paragraph{PT box-counting rules}
\begin{itemize}
    \item \textbf{Level 1}
    \bea
    \ket{\eps_3^0}
    \eea
    \item \textbf{Level 2}
    \bea
    \ket{\eps_3^0,0^0},\quad \ket{\eps_3^{-1}},\quad \ket{\eps_3^2,\eps_3-\eps_2},\quad \ket{\eps_3^1,\eps_3-\eps_1}
    \eea
    \item \textbf{Level 3}
    \bea
    \ket{0^1,\eps_3^1,\eps_3-\eps_1},\quad \ket{\eps_3^{-1},\eps_3-\eps_1},\quad \ket{0^0,\eps_3^{-1}},\quad \ket{\eps_3^{1},\eps_3-\eps_1,\eps_3-2\eps_1},\\
    \ket{-\eps_3,0^3,\eps_3^3},\quad \ket{0^2,\eps_3^{2},\eps_3-\eps_2},\quad \ket{\eps_3-2\eps_2,\eps_3^{2},\eps_3-\eps_2},\quad \ket{\eps_3^{-1},\eps_3-\eps_2}
    \eea
    
\end{itemize}

\paragraph{GR box-counting rules}
\begin{itemize}
\item \textbf{Level 0 $\rightarrow $ Level 1}:
\bea
\ket{\text{vac}}\rightarrow \ket{\epsilon_{3,H}}+\ket{\epsilon_{3,L}}
\eea 
\item \textbf{Level 1 $\rightarrow $ Level 2}:
\bea
 |\epsilon_{3,H}\rangle &\rightarrow  |\epsilon_{3,U}\rangle + |\epsilon_{3,H},\epsilon_3-\epsilon_1\rangle + |\epsilon_{3,H},\epsilon_3-\epsilon_2\rangle +|\epsilon_{3,L},0_L\rangle +|\epsilon_{3,H},0_H\rangle  \\
|\epsilon_{3,L}\rangle  &\rightarrow |\epsilon_{3,U}\rangle+|\epsilon_{3,L},0_L\rangle
\eea


\item\textbf{Level 2 $\rightarrow $ Level 3}:
\bea
|\epsilon_{3,L},0_L \rangle &\rightarrow |\epsilon_{3,U},0_L \rangle \\ 
|\epsilon_{3,H}, 0_H\rangle  &\rightarrow|\epsilon_{3,H},0_H,-\epsilon_3\rangle +|\epsilon_{3,U},0_L\rangle +|\epsilon_{3,U},0_H\rangle+|\epsilon_{3,H},0_H,\epsilon_3-\epsilon_2\rangle +|\epsilon_{3,H},0_H,\epsilon_3-\epsilon_1\rangle   \\ 
 |\epsilon_{3,H},\epsilon_3-\epsilon_2\rangle &\rightarrow|\epsilon_{3,U},\epsilon_3-\epsilon_2\rangle  + |\epsilon_{3,H},\epsilon_3-\epsilon_2,0_H\rangle +   |\epsilon_{3,H},\epsilon_3-\epsilon_2,\epsilon_3-2\epsilon_2\rangle\\ 
|\epsilon_{3,H},\epsilon_3-\epsilon_1\rangle &\rightarrow|\epsilon_{3,U},\epsilon_3-\epsilon_1\rangle +|\epsilon_{3,H},\epsilon_3-\epsilon_1,0_H\rangle +|\epsilon_{3,H},\epsilon_3-\epsilon_1,\epsilon_3-2\epsilon_1\rangle\\ 
 |\epsilon_{3,U}\rangle  &\rightarrow |\epsilon_{3,U},\epsilon_3-\epsilon_1\rangle+ |\epsilon_{3,U},\epsilon_3-\epsilon_2\rangle+|\epsilon_{3,U},0_L\rangle +|\epsilon_{3,U},0_H\rangle
\eea




\end{itemize}

\paragraph{JK prescription and transition rules}
The framing node contribution is
\bea
\mathcal{Z}^{\D6_{\bar{4}}\tbar\D2\tbar\D0}_{\DT;\,\yng(2)\,\yng(1,1)\,\yng(1)}(\fra,\phi_{I})&=\frac{ \sh\left(-\fra+\epsilon
   _1-\epsilon _3+\phi _I\right) \sh\left(-\fra+\epsilon _2-\epsilon _3+\phi _I\right)
   \sh\left(-\fra+\epsilon _3+\phi _I\right)}{
   \sh\left(-\fra-\epsilon _1-\epsilon _2+\phi _I\right) \sh\left(-\fra-\epsilon _1-2 \epsilon _3+\phi
   _I\right) \sh\left(-\fra-\epsilon _2-2 \epsilon _3+\phi _I\right)}\\
   &\times\frac{\sh\left(\fra+\epsilon _1+\epsilon _2+2 \epsilon _3-\phi _I\right)^2}{\sh\left(\fra+\epsilon _3-\phi _I\right)^2}.
\eea

\begin{itemize}
    \item \textbf{Level 1}: The reference vector $\tilde{\eta}_0$ picks up the pole at $\phi_1=\fra+\eps_3$ and the transition rule is
    \bea
    \boxed{\text{vac}\rightarrow (\eps_3)}.
    \eea
    Note that it is a second order pole. After taking the unrefined limit, this contribution becomes 2.
    \item \textbf{Level 2}: For the second level, the poles come from the charge vectors $-e_{I},e_{I}-e_{J}$. The non-zero JK residue configurations are
    \bea    
    (\eps_3,\eps_3),\quad (\eps_3,0),\quad (\eps_3,\eps_3-\eps_1),\quad (\eps_3,\eps_3-\eps_2).
    \eea    
    The transition rule is determined as
    \begin{empheq}[box=\fbox]{align}
\begin{split}
(\eps_3)\longrightarrow (\eps_3,\eps_3)+(\eps_3,0)+(\eps_3,\eps_3-\eps_1)+(\eps_3,\eps_3-\eps_2)
\end{split}
\end{empheq}
The configuration $(\eps_3,\eps_3)$ is allowed because of a similar discussion as in section~\ref{sec:PTthreelegs}. Other poles are non-degenerate poles coming from the charge vector $e_{1}-e_{2}$ and the transition rules are obvious.

    \item \textbf{Level 3}: The non-zero JK-residue configurations are
    \bea
(\eps_3,\eps_3,0),\quad (\eps_3,\eps_3,\eps_3-\eps_1),\quad (\eps_3,\eps_3,\eps_3-\eps_2),\quad (\eps_3,\eps_3-\eps_1,\eps_3-2\eps_1),\\
(\eps_3,0,\eps_3-\eps_1),\quad (\eps_3,0,\eps_3-\eps_2),\quad (\eps_3,0,-\eps_3),\quad (\eps_3,\eps_3-\eps_2,\eps_3-2\eps_2) .
\eea
The transition rules are determined as
    \begin{empheq}[box=\fbox]{align}
\begin{split}
(\eps_3,\eps_3)&\longrightarrow (\eps_3,\eps_3,0)+(\eps_3,\eps_3,\eps_3-\eps_1)+(\eps_3,\eps_3,\eps_3-\eps_2)\\
(\eps_3,0)&\longrightarrow (\eps_3,0,\eps_3-\eps_1)+(\eps_3,0,\eps_3-\eps_2)+(\eps_3,0,-\eps_3)\\
(\eps_3,\eps_3-\eps_1)&\longrightarrow (\eps_3,\eps_3-\eps_1,0)+(\eps_3,\eps_3-\eps_1,\eps_3-2\eps_1) \\
(\eps_3,\eps_3-\eps_2)&\longrightarrow (\eps_3,\eps_3-\eps_2,0)+(\eps_3,\eps_3-\eps_2,\eps_3-2\eps_2) 
\end{split}
\end{empheq}
The configuration $(\eps_3,\eps_3,0)$ actually has a second order pole at $\phi_3=\fra$.

\end{itemize}

\paragraph{JK-residues}
For level one, we have{\small
\bea
(\epsilon_{3}),&\qquad \frac{\ch\left(\epsilon _1\right) \sh\left(\epsilon _1+\epsilon
   _2\right) \sh\left(\epsilon _1+\epsilon _2+\epsilon _3\right)
   \sh\left(2 \epsilon _3\right)}{2 \sh\left(\epsilon _1\right)
   \sh\left(\epsilon _3\right) \sh\left(-\epsilon _1-\epsilon
   _2+\epsilon _3\right)}+\frac{\ch\left(\epsilon _2\right)
   \sh\left(\epsilon _1+\epsilon _2\right) \sh\left(\epsilon
   _1+\epsilon _2+\epsilon _3\right) \sh\left(2 \epsilon _3\right)}{2
   \sh\left(\epsilon _2\right) \sh\left(\epsilon _3\right)
   \sh\left(-\epsilon _1-\epsilon _2+\epsilon
   _3\right)}\\
   &\qquad +\frac{\ch\left(\epsilon _1+\epsilon _3\right)
   \sh\left(\epsilon _1+\epsilon _2\right) \sh\left(\epsilon
   _1+\epsilon _2+\epsilon _3\right) \sh\left(2 \epsilon _3\right)}{2
   \sh\left(\epsilon _3\right) \sh\left(\epsilon _1+\epsilon
   _3\right) \sh\left(-\epsilon _1-\epsilon _2+\epsilon
   _3\right)}+\frac{\ch\left(\epsilon _2+\epsilon _3\right)
   \sh\left(\epsilon _1+\epsilon _2\right) \sh\left(\epsilon
   _1+\epsilon _2+\epsilon _3\right) \sh\left(2 \epsilon _3\right)}{2
   \sh\left(\epsilon _3\right) \sh\left(-\epsilon _1-\epsilon
   _2+\epsilon _3\right) \sh\left(\epsilon _2+\epsilon
   _3\right)}\\
   &\qquad-\frac{\ch\left(-\epsilon _1-\epsilon _2+\epsilon _3\right)
   \sh\left(\epsilon _1+\epsilon _2\right) \sh\left(\epsilon
   _1+\epsilon _2+\epsilon _3\right) \sh\left(2 \epsilon _3\right)}{2
   \sh\left(\epsilon _3\right) \sh\left(-\epsilon _1-\epsilon
   _2+\epsilon _3\right)^2}-\frac{\ch\left(\epsilon _1+\epsilon
   _2+\epsilon _3\right) \sh\left(\epsilon _1+\epsilon _2\right)
   \sh\left(2 \epsilon _3\right)}{\sh\left(\epsilon _3\right)
   \sh\left(-\epsilon _1-\epsilon _2+\epsilon
   _3\right)}\\
   &\qquad +\frac{\ch\left(2 \epsilon _3\right) \sh\left(\epsilon
   _1+\epsilon _2\right) \sh\left(\epsilon _1+\epsilon _2+\epsilon
   _3\right)}{2 \sh\left(\epsilon _3\right) \sh\left(-\epsilon
   _1-\epsilon _2+\epsilon _3\right)}
\eea}
For level two, we have{\small
\bea
(\epsilon_{3},\epsilon_{3}),&\qquad\frac{\sh\left(\epsilon _1+\epsilon _2\right)^4 \sh\left(2
   \epsilon _3\right)^2 \sh\left(\epsilon _1+\epsilon _3\right)^2
   \sh\left(\epsilon _2+\epsilon _3\right)^2}{\sh\left(\epsilon
   _1\right)^2 \sh\left(\epsilon _2\right)^2 \sh\left(\epsilon
   _3\right)^4 \sh\left(-\epsilon _1-\epsilon _2+\epsilon _3\right)^2},\\
(\epsilon_{3},-\epsilon_{1}+\epsilon_{3}),&\qquad-\frac{\sh\left(\epsilon _1+\epsilon _2\right) \sh\left(2 \epsilon
   _1+\epsilon _2\right) \sh\left(2 \epsilon _3\right)
   \sh\left(-\epsilon _1+\epsilon _2+\epsilon _3\right) \sh\left(2
   \epsilon _1+\epsilon _2+\epsilon _3\right) \sh\left(2 \epsilon
   _3-\epsilon _1\right)}{2 \sh\left(\epsilon _1\right) \sh\left(2
   \epsilon _1\right) \sh\left(\epsilon _3\right) \sh\left(\epsilon
   _3-\epsilon _1\right) \sh\left(-2 \epsilon _1-\epsilon _2+\epsilon
   _3\right) \sh\left(-\epsilon _1-\epsilon _2+\epsilon _3\right)},\\
(\epsilon_{3},-\epsilon_{2}+\epsilon_{3}),&\qquad-\frac{\sh\left(\epsilon _1+\epsilon _2\right) \sh\left(\epsilon
   _1+2 \epsilon _2\right) \sh\left(2 \epsilon _3\right)
   \sh\left(\epsilon _1-\epsilon _2+\epsilon _3\right)
   \sh\left(\epsilon _1+2 \epsilon _2+\epsilon _3\right) \sh\left(2
   \epsilon _3-\epsilon _2\right)}{2 \sh\left(\epsilon _2\right)
   \sh\left(2 \epsilon _2\right) \sh\left(\epsilon _3\right)
   \sh\left(-\epsilon _1-2 \epsilon _2+\epsilon _3\right)
   \sh\left(\epsilon _3-\epsilon _2\right) \sh\left(-\epsilon
   _1-\epsilon _2+\epsilon _3\right)},\\
   (0,\epsilon_{3}),&\qquad-\frac{\ch\left(\epsilon _1\right) \sh\left(\epsilon _1+\epsilon
   _2+2 \epsilon _3\right) \sh\left(\epsilon _1+\epsilon _2+\epsilon
   _3\right)}{4 \sh\left(\epsilon _1\right) \sh\left(\epsilon
   _3\right)}-\frac{\ch\left(\epsilon _2\right) \sh\left(\epsilon
   _1+\epsilon _2+2 \epsilon _3\right) \sh\left(\epsilon _1+\epsilon
   _2+\epsilon _3\right)}{4 \sh\left(\epsilon _2\right)
   \sh\left(\epsilon _3\right)}\\
   &\qquad-\frac{\ch\left(\epsilon _1+\epsilon
   _2\right) \sh\left(\epsilon _1+\epsilon _2+2 \epsilon _3\right)
   \sh\left(\epsilon _1+\epsilon _2+\epsilon _3\right)}{4
   \sh\left(\epsilon _1+\epsilon _2\right) \sh\left(\epsilon
   _3\right)}-\frac{\ch\left(2 \epsilon _3\right) \sh\left(\epsilon
   _1+\epsilon _2+2 \epsilon _3\right) \sh\left(\epsilon _1+\epsilon
   _2+\epsilon _3\right)}{4 \sh\left(\epsilon _3\right) \sh\left(2
   \epsilon _3\right)}\\
   &\qquad +\frac{\ch\left(\epsilon _3-\epsilon _1\right)
   \sh\left(\epsilon _1+\epsilon _2+2 \epsilon _3\right)
   \sh\left(\epsilon _1+\epsilon _2+\epsilon _3\right)}{4
   \sh\left(\epsilon _3\right) \sh\left(\epsilon _3-\epsilon
   _1\right)}-\frac{\ch\left(\epsilon _1+\epsilon _3\right)
   \sh\left(\epsilon _1+\epsilon _2+2 \epsilon _3\right)
   \sh\left(\epsilon _1+\epsilon _2+\epsilon _3\right)}{4
   \sh\left(\epsilon _3\right) \sh\left(\epsilon _1+\epsilon
   _3\right)}\\
   &\qquad +\frac{\ch\left(\epsilon _3-\epsilon _2\right)
   \sh\left(\epsilon _1+\epsilon _2+2 \epsilon _3\right)
   \sh\left(\epsilon _1+\epsilon _2+\epsilon _3\right)}{4
   \sh\left(\epsilon _3\right) \sh\left(\epsilon _3-\epsilon
   _2\right)}+\frac{\ch\left(-\epsilon _1-\epsilon _2+\epsilon _3\right)
   \sh\left(\epsilon _1+\epsilon _2+2 \epsilon _3\right)
   \sh\left(\epsilon _1+\epsilon _2+\epsilon _3\right)}{4
   \sh\left(\epsilon _3\right) \sh\left(-\epsilon _1-\epsilon
   _2+\epsilon _3\right)}\\
   &\qquad -\frac{\ch\left(\epsilon _2+\epsilon _3\right)
   \sh\left(\epsilon _1+\epsilon _2+2 \epsilon _3\right)
   \sh\left(\epsilon _1+\epsilon _2+\epsilon _3\right)}{4
   \sh\left(\epsilon _3\right) \sh\left(\epsilon _2+\epsilon
   _3\right)}-\frac{\ch\left(\epsilon _1+2 \epsilon _3\right)
   \sh\left(\epsilon _1+\epsilon _2+2 \epsilon _3\right)
   \sh\left(\epsilon _1+\epsilon _2+\epsilon _3\right)}{4
   \sh\left(\epsilon _3\right) \sh\left(\epsilon _1+2 \epsilon
   _3\right)}\\
   &\qquad -\frac{\ch\left(\epsilon _2+2 \epsilon _3\right)
   \sh\left(\epsilon _1+\epsilon _2+2 \epsilon _3\right)
   \sh\left(\epsilon _1+\epsilon _2+\epsilon _3\right)}{4
   \sh\left(\epsilon _3\right) \sh\left(\epsilon _2+2 \epsilon
   _3\right)}-\frac{3 \ch\left(\epsilon _3\right) \sh\left(\epsilon
   _1+\epsilon _2+2 \epsilon _3\right) \sh\left(\epsilon _1+\epsilon
   _2+\epsilon _3\right)}{4 \sh\left(\epsilon
   _3\right)^2}\\
   &\qquad +\frac{\ch\left(\epsilon _1+\epsilon _2+2 \epsilon
   _3\right) \sh\left(\epsilon _1+\epsilon _2+\epsilon _3\right)}{2
   \sh\left(\epsilon _3\right)}+\frac{\ch\left(\epsilon _1+\epsilon
   _2+\epsilon _3\right) \sh\left(\epsilon _1+\epsilon _2+2 \epsilon
   _3\right)}{2 \sh\left(\epsilon _3\right)}
\eea}

The weight factors are summarized as follows.
\begin{itemize}
    \item \textbf{Level 1}:
    \bea
    \mathcal{Z}_{(\eps_3)}\rightarrow 2
    \eea
    \item \textbf{Level 2}:
    \bea
\mathcal{Z}_{(\eps_3,\eps_3)}\rightarrow 2,\quad \mathcal{Z}_{(\eps_3,0)}\rightarrow 2,\quad \mathcal{Z}_{(\eps_3,\eps_3-\eps_1)}\rightarrow 1,\quad \mathcal{Z}_{(\eps_3,\eps_3-\eps_2)}\rightarrow 1
    \eea
    The configuration $(\eps_3,\eps_3)$ gives the weight 2 but including the Weyl group factor $1/2$, the contribution to the unrefined vertex is $1$.
    
    \item \textbf{Level 3}:
    \bea
    \mathcal{Z}_{(\eps_3,\eps_3,0)}\rightarrow 4 ,\quad \mathcal{Z}_{(\eps_3,\eps_3,\eps_3-\eps_1)}\rightarrow 2,\quad \mathcal{Z}_{(\eps_3,\eps_3,\eps_3-\eps_2)}\rightarrow 2,\quad \mathcal{Z}_{(\eps_3,\eps_3-\eps_1,\eps_3-2\eps_1)}\rightarrow 1,\\
\mathcal{Z}_{(\eps_3,0,\eps_3-\eps_1)}\rightarrow 1,\quad \mathcal{Z}_{(\eps_3,0,\eps_3-\eps_2)}\rightarrow 1,\quad \mathcal{Z}_{(\eps_3,0,-\eps_3)}\rightarrow 1,\quad \mathcal{Z}_{(\eps_3,\eps_3-\eps_2,\eps_3-2\eps_2)}\rightarrow 1 
    \eea
    Including the Weyl group factor $1/6$, and the multiplicity $3$ for the configurations $(\eps_3,\eps_3,0),\,(\eps_3,\eps_3,\eps_3-\eps_{1,2})$, they give the weight $2,1,1$, respectively in the computation of the unrefined vertex. Thus, at level $3$, they contribute as $9q^{3}$.
\end{itemize}

\paragraph{Unrefined vertex}
\bea
\wtC_{\yng(1,1)\,\yng(2)\,\yng(1)}(q)&=\frac{1-q+q^3-q^4+2q^6-q^7-q^8+q^9}{(1-q)^{3}(1-q^{2})^{2}}\\
&=1+2q  +5q^2+9q^3+16q^4+\cdots
\eea
\paragraph{Refined vertex}
\bea
\wtC_{\yng(1,1)\,\yng(2)\,\yng(1)}(t,q)&=1+(q +t)+(q^2+3qt+t^2)+(4q^2t+4qt^2+t^3)+\cdots
\cdots
\eea
\paragraph{Macdonald refined vertex}

\bea
\wtM_{\yng(1,1)\,\yng(2)\,\yng(1)}(x,y;q,t)&=1+\left(\frac{1-t}{1-q}x+\frac{(1+q)^{2}(1-q)(1-t)}{(1-qt)^{2}}y\right)\\
&+\left(\frac{(1-t)(1-qt)}{(1-q)(1-q^2)}x^2+\frac{(3qt-q+t-3)(q+1)(t-1)}{(1-qt)^2}xy+\frac{(1-q)(1-q^2)}{(1-t)(1-qt)}y^2\right)+\cdots
\eea

\section{\texorpdfstring{$A_{1}$}{A1} \texorpdfstring{$qq$}{qq}-character and collision limit}\label{app-sec:A1qqcharacter-collisionlimit}
In this section, we review the $A_{1}$ $qq$-character associated with the so-called $q$-deformed Virasoro algebra and its collision limit. For the $q$-Virsoro algebra, see the original papers \cite{Shiraishi:1995rp,Awata:1996dx,Awata:1995zk}. For the $qq$-character perspective, see \cite{Nekrasov:2015wsu,Nekrasov:2016ydq}. For the quiver W-algebra perspective, see \cite{Kimura:2015rgi,Kimura:2019hnw,Kimura:2020jxl}.

\subsection{\texorpdfstring{$A_{1}$}{A1} \texorpdfstring{$qq$}{qq}-character and \texorpdfstring{$q$}{q}-Virasoro algebra}
\paragraph{Partition function}The partition function is defined as
\bea
\mathcal{Z}_{k}&=\frac{1}{k!}\left(\frac{\sh(-\epsilon_{12})}{\sh(-\epsilon_{1})\sh(-\epsilon_{2})}\right)^{k}\oint_{\eta_{0}} \prod_{I=1}^{k}\frac{d\phi_{I}}{2\pi i}\prod_{I=1}^{k}\frac{\sh(\phi_{I}-\fra-\epsilon_{1})\sh(-\phi_{I}+\fra+\epsilon_{2})}{\sh(\phi_{I}-\fra)\sh(\fra-\phi_{I}+\epsilon_{12})}\\
&\qquad \times \prod_{I\neq J}\frac{\sh(\phi_{I}-\phi_{J})\sh(\phi_{I}-\phi_{J}-\epsilon_{12})}{\sh(\phi_{I}-\phi_{J}-\epsilon_{1})\sh(\phi_{I}-\phi_{J}-\epsilon_{2})}\\
&=\frac{1}{k!}\left(\frac{(1-q_{12})}{(1-q_{1})(1-q_{2})}\right)^{k}\oint_{\eta_{0}}\prod_{I=1}^{k}\frac{dx_{I}}{2\pi i x_{I}}\prod_{I=1}^{k}\mathscr{S}_{12}\left(\frac{x}{x_{I}}\right)\prod_{I\neq J}\mathscr{S}_{12}\left(\frac{x_{I}}{x_{J}}\right)^{-1}.
\eea
Based on the JK-formalism, the poles picked up are
\bea
\phi_{I}=\fra,\quad \phi_{I}-\phi_{J}=\epsilon_{1,2}.
\eea
Taking $\phi_{1}=\fra$, one can see that the poles $\phi_{2}=\fra+\epsilon_{1,2}$ are canceled by the numerator $\sh(\phi_{I}-\fra-\epsilon_{1})\sh(-\phi_{I}+\fra+\epsilon_{2})$. The pole $\phi_{2}=\fra=\phi_{1}$ is also canceled by the numerator $\sh(\phi_{I}-\phi_{J})$. Thus, for $k\geq 2$, the partition function vanishes. For $k=1$, evaluating the residue, we have
\bea
\mathcal{Z}_{1}=1.
\eea
Therefore, the partition function is
\bea
\mathcal{Z}[\fq]=1+\fq .
\eea

\paragraph{Vertex operators}Let us introduce the following free boson\footnote{We use the same notation for the operators as in the main text but note that the OPE formulas are different.}
\bea
\relax[\mathsf{a}_{n},\mathsf{a}_{m}]=-\frac{1}{n}\delta_{n+m,0}(1-q_{1}^{n})(1-q_{2}^{n})(1+q_{12}^{-n}).
\eea
The root vertex operator is defined as
\bea
\mathsf{A}(x)=\mathsf{a}_{0}(x):\exp\left(\sum_{n\neq 0}\mathsf{a}_{n}x^{-n}\right):.
\eea
We also introduce the following vertex operators 
\bea
\mathsf{Y}(x)=\mathsf{y}_{0}(x):\exp\left(\sum_{n\neq 0}\mathsf{y}_{n}x^{-n}\right):,\quad \mathsf{S}_{a}(x)=\mathsf{s}_{a,0}(x):\exp\left(\sum_{n\neq 0}\mathsf{s}_{a,n}x^{-n}\right):,\quad a=1,2
\eea
where
\bea
\mathsf{y}_{n}=\frac{\mathsf{a}_{n}}{(1+q_{12}^{n})},\quad \mathsf{s}_{a,n}=\frac{\mathsf{a}_{n}}{1-q_{a}^{-n}}.
\eea
The operators $\mathsf{a}_{0}(x),\mathsf{y}_{a}(x),\mathsf{s}_{a,0}(x)$ are zero-modes.

The OPEs are
\bea
\mathsf{Y}(x)\mathsf{S}_{1}(x')&=\frac{1-q_{12}^{-1}x'/x}{1-q_{1}^{-1}x'/x}:\mathsf{Y}(x)\mathsf{S}_{1}(x'):,\quad \mathsf{S}_{1}(x')\mathsf{Y}(x)=q_{2}^{-1}\frac{1-q_{12}x/x'}{1-q_{1}x/x'}:\mathsf{Y}(x)\mathsf{S}_{1}(x'):,\\
\mathsf{Y}(x)\mathsf{S}_{2}(x')&=\frac{1-q_{12}^{-1}x'/x}{1-q_{2}^{-1}x'/x}:\mathsf{Y}(x)\mathsf{S}_{2}(x'):,\quad \mathsf{S}_{2}(x')\mathsf{Y}(x)=q_{1}^{-1}\frac{1-q_{12}x/x'}{1-q_{2}x/x'}:\mathsf{Y}(x)\mathsf{S}_{1}(x'):,\\
\mathsf{A}(x)\mathsf{S}_{1}(x')&=\frac{(1-x'/x)(1-q_{12}^{-1}x'/x)}{(1-q_{2}x'/x)(1-q_{1}^{-1}x'/x)}:\mathsf{A}(x)\mathsf{S}_{1}(x'):,\quad \mathsf{S}_{1}(x')\mathsf{A}(x)=q_{2}^{-2}\frac{(1-x/x')(1-q_{12}x/x')}{(1-q_{2}^{-1}x/x')(1-q_{1}x/x')}:\mathsf{A}(x)\mathsf{S}_{1}(x'):,\\
\mathsf{A}(x)\mathsf{S}_{2}(x')&=\frac{(1-x'/x)(1-q_{12}^{-1}x'/x)}{(1-q_{1}x'/x)(1-q_{2}^{-1}x'/x)}:\mathsf{A}(x)\mathsf{S}_{2}(x'):,\quad \mathsf{S}_{2}(x')\mathsf{A}(x)=q_{1}^{-2}\frac{(1-x/x')(1-q_{12}x/x')}{(1-q_{1}^{-1}x/x')(1-q_{2}x/x')}:\mathsf{A}(x)\mathsf{S}_{2}(x'):,\\
\mathsf{A}(x)\mathsf{Y}(x')&=\mathscr{S}_{12}\left(\frac{x'}{x}\right)^{-1}:\mathsf{A}(x)\mathsf{Y}(x'):,\quad \mathsf{Y}(x')\mathsf{A}(x)=\mathscr{S}_{12}\left(q_{12}^{-1}\frac{x}{x'}\right)^{-1}:\mathsf{A}(x)\mathsf{Y}(x'):
\eea
where the zero-modes are defined as
\bea
\,&\mathsf{a}_{0}(x)=e^{\mathsf{t}_{0}},\quad \mathsf{y}_{0}(x)=e^{\mathsf{t}_{0}/2},\quad \mathsf{s}_{a,0}(x)=x^{\mathsf{s}_{a,0}}e^{\tilde{\mathsf{s}}_{a,0}},\\
&\mathsf{s}_{1,0}=-(\log q_{1})^{-1}\mathsf{t}_{0}+\frac{\log q_{2}}{\log q_{1}},\quad \mathsf{s}_{2,0}=-(\log q_{2})^{-1}\mathsf{t}_{0}+\frac{\log q_{1}}{\log q_{2}},\quad \tilde{\mathsf{s}}_{1,0}=-2\log q_{2}\partial_{\mathsf{t}},\quad \tilde{\mathsf{s}}_{2,0}=-2\log q_{1}\partial_{\mathsf{t}}.
\eea
Under these zero-modes, we also have the following relations:
\bea\label{eq:A1relation}
\mathsf{A}(x)={q_{2}:\frac{\mathsf{S}_{1}(x)}{\mathsf{S}_{1}(q_{1}x)}:}={q_{1}:\frac{\mathsf{S}_{2}(x)}{\mathsf{S}_{2}(q_{2}x)}:}={:\mathsf{Y}(x)\mathsf{Y}(q_{12}^{-1}x):}.
\eea

\paragraph{Free field realization of contour integral} The contour integral formula has the following free field realization
\bea
\mathcal{Z}_{k}=\frac{1}{k!}\left(\frac{(1-q_{12})}{(1-q_{1})(1-q_{2})}\right)^{k}\oint_{\eta_{0}}\prod_{I=1}^{k}\frac{dx_{I}}{2\pi i x_{I}}\left\langle \prod_{I=1}^{k}\mathsf{A}(x_{I})^{-1}\mathsf{Y}(x)\right\rangle.
\eea
It is natural to define the following $qq$-character:
\bea
\mathsf{T}(x)&=\sum_{k=0}^{\infty}\fq^{k}\frac{1}{k!}\left(\frac{(1-q_{12})}{(1-q_{1})(1-q_{2})}\right)^{k}\oint_{\eta_{0}}\prod_{I=1}^{k}\frac{dx_{I}}{2\pi i x_{I}}\prod_{I=1}^{k}\mathsf{A}(x_{I})^{-1}\mathsf{Y}(x)\\
&=\mathsf{Y}(x)+\fq :\mathsf{Y}(x)\mathsf{A}(x)^{-1}:
\eea
where we evaluated the poles in the last line. This operator is the generator of the $q$-Virasoro algebra. Note that we have
\bea
\langle  \mathsf{T}(x)\rangle =\mathcal{Z}[\fq]
\eea
which is the BPS/CFT correspondence.

\paragraph{Commutativity with the screening charge}
The screening charges are defined as
\bea
\mathscr{Q}_{a}(x)=\sum_{k\in \mathbb{Z}}\mathsf{S}_{a}(xq_{a}^{k}),\quad a=1,2.
\eea
Let us focus on the commutativity with $\mathscr{Q}_{2}(x')$.

The OPEs give the following commutation relations:
\bea
\relax[\mathsf{Y}(x),\mathsf{S}_{2}(x')]&=(1-q_{1}^{-1})\delta\left(q_{2}^{-1}\frac{x'}{x}\right):\mathsf{Y}(x)\mathsf{S}_{2}(q_{2}x):,\\
[:\mathsf{Y}(x)\mathsf{A}^{-1}(x):,\mathsf{S}_{2}(x')]&=(1-q_{1})\delta\left(\frac{x'}{x}\right):\mathsf{Y}(x)\mathsf{A}^{-1}(x)\mathsf{S}_{2}(x):.
\eea
Using \eqref{eq:A1relation}, we have
\bea
\relax[\mathsf{T}(x),\mathsf{S}_{2}(x')]=(1-q_{1}^{-1})\left(\delta\left(q_{2}^{-1}x'/x\right)-\delta\left(x'/x\right)\right):\mathsf{Y}(x)\mathsf{S}_{2}(q_{2}x):
\eea
where we set $\fq=1$. The topological term can be included by modifying the zero-modes, but we omit the discussion. Therefore, we finally have
\bea
\relax[\mathsf{T}(x),\mathscr{Q}_{2}(x')]=0.
\eea

\paragraph{Higher rank $qq$-character}
Higher rank $qq$-characters can be obtained by choosing the highest weight to have multiple numbers of the $\mathsf{Y}$-operators. Let us focus on the rank 2 $qq$-character. Starting from $:\mathsf{Y}(x)\mathsf{Y}(x'):$ and studying the commutativity with the screening charge, we can see that the following is the rank 2 $qq$-character:
\bea\label{eq:rank2qqcharacter}
:\mathsf{Y}(x)\mathsf{Y}(x'):+\mathscr{S}_{12}\left(q_{12}^{-1}\frac{x'}{x}\right):\frac{\mathsf{Y}(x)\mathsf{Y}(x')}{\mathsf{A}(x')}:+\mathscr{S}_{12}\left(\frac{x'}{x}\right):\frac{\mathsf{Y}(x)\mathsf{Y}(x')}{\mathsf{A}(x)}:+:\frac{\mathsf{Y}(x)\mathsf{Y}(x')}{\mathsf{A}(x)\mathsf{A}(x')}:.
\eea

Another way to obtain this is to fuse the $A_{1}$ $qq$-characters:
\bea
\mathsf{f}\left(\frac{x_{2}}{x_{1}}\right)\mathsf{T}(x_{1})\mathsf{T}(x_{2}),\quad \mathsf{f}(x)=\exp\left(\sum_{n=1}^{\infty}\frac{1}{n}\frac{(1-q_{1}^{n})(1-q_{2}^{n})}{(1+q_{12}^{n})}x^{n}\right).
\eea
Obviously, since $\mathsf{T}(x)$ commutes with the screening charges, the fused $qq$-character also commutes with them.

In the contour integral formalism, one will simply replace $\mathsf{Y}(x)$ to $:\mathsf{Y}(x_{1})\mathsf{Y}(x_{2}):$ and evaluate the contour integral.

\subsection{Collision limit and higher order poles}

\paragraph{Contour integrals with higher order poles } 

Let us start from a toy model:
\bea
\mathcal{Z}_{k}&=\frac{1}{k!}\left(\frac{\sh(-\epsilon_{12})}{\sh(-\epsilon_{1})\sh(-\epsilon_{2})}\right)^{k}\oint_{\eta_{0}} \prod_{I=1}^{k}\frac{d\phi_{I}}{2\pi i}\prod_{I=1}^{k}\left(\frac{\sh(\phi_{I}-\fra-\epsilon_{1})\sh(-\phi_{I}+\fra+\epsilon_{2})}{\sh(\phi_{I}-\fra)\sh(\fra-\phi_{I}+\epsilon_{12})}\right)^{2}\\
&\qquad \times \prod_{I\neq J}\frac{\sh(\phi_{I}-\phi_{J})\sh(\phi_{I}-\phi_{J}-\epsilon_{12})}{\sh(\phi_{I}-\phi_{J}-\epsilon_{1})\sh(\phi_{I}-\phi_{J}-\epsilon_{2})}\\
&=\frac{1}{k!}\left(\frac{(1-q_{12})}{(1-q_{1})(1-q_{2})}\right)^{k}\oint_{\eta_{0}}\prod_{I=1}^{k}\frac{dx_{I}}{2\pi i x_{I}}\prod_{I=1}^{k}\mathscr{S}_{12}\left(\frac{x}{x_{I}}\right)^{2}\prod_{I\neq J}\mathscr{S}_{12}\left(\frac{x_{I}}{x_{J}}\right)^{-1}.
\eea
We have a second order pole at $\phi_{I}=\fra$. For level 1, the pole is
\bea
\phi_{1}=\fra
\eea
and evaluating the residue gives
\bea
\mathcal{Z}_{1}=\frac{1+q_{1}+q_{2}-6q_{1}q_{2}+q_{1}^{2}q_{2}+q_{1}q_{2}^{2}+q_{12}^{2}}{(1-q_{12})^{2}}\eqqcolon \mathfrak{c}(q_{1},q_{2})
\eea
where this residue is obtained by taking the derivative of the integrand. 

An interesting phenomenon occurs at the second level. Using the JK-residue prescription, we may choose the poles $\phi_{2}=\fra,\,\,\phi_{1}+\epsilon_{1,2}$. Since the level one is a second order pole, one will see that the pole at $\phi_{2}=\fra$ does \textit{not} vanish:
\bea
(\phi_{1},\phi_{2})=(\fra,\fra),\quad \mathcal{Z}_{2}=1.
\eea
Using the terminology of the main text, this corresponds to placing an ultra-heavy box at the origin.

For level $k\geq 3$, one will see that the contour integral vanishes and thus we have
\bea
\mathcal{Z}[\fq]=1+\left(\frac{1+q_{1}+q_{2}-6q_{1}q_{2}+q_{1}^{2}q_{2}+q_{1}q_{2}^{2}+q_{12}^{2}}{(1-q_{12})^{2}}\right)\fq+\fq^{2}.
\eea

Another interesting phenomenon is that at the NS-limit $q_{1}\rightarrow 1$ or $q_{2}\rightarrow 1$, the one-instanton contribution becomes
\bea
\mathcal{Z}_{1}\longrightarrow 2.
\eea
Namely, at level one, we may place the light and heavy boxes at the origin, which gives the two state degeneracy. This is familiar to what happened in the three-legs computation in section~\ref{sec:PTthreelegs} where at the limit\footnote{This limit is the NS-limit of the D6 partition functions.} $q_{4}\rightarrow 1$, the partition function becomes 2, instead of 1.

The free field realization of the contour integral formula is
\bea
\mathcal{Z}_{k}=\frac{1}{k!}\left(\frac{(1-q_{12})}{(1-q_{1})(1-q_{2})}\right)^{k}\oint_{\eta_{0}}\prod_{I=1}^{k}\frac{dx_{I}}{2\pi i x_{I}}\left\langle \prod_{I=1}^{k}\mathsf{A}(x_{I})^{-1}:\mathsf{Y}(x)^{2}:\right\rangle.
\eea
This free field realization makes us want to define the following $qq$-character:
\bea
\mathsf{T}^{(2)}(x)&\coloneqq \sum_{k=0}^{\infty}\fq^{k}\frac{1}{k!}\left(\frac{(1-q_{12})}{(1-q_{1})(1-q_{2})}\right)^{k}\oint_{\eta_{0}}\prod_{I=1}^{k}\frac{dx_{I}}{2\pi i x_{I}} \prod_{I=1}^{k}\mathsf{A}(x_{I})^{-1}:\mathsf{Y}(x)^{2}:\\
&=\sum_{k=0}^{\infty}\frac{\fq^{k}}{k!}\left(\frac{(1-q_{12})}{(1-q_{1})(1-q_{2})}\right)^{k}\oint_{\eta_{0}}\prod_{I=1}^{k}\frac{dx_{I}}{2\pi i x_{I}}\prod_{I=1}^{k}\mathscr{S}_{12}\left(\frac{x}{x_{I}}\right)^{2}\prod_{I\neq J}\mathscr{S}_{12}\left(\frac{x_{I}}{x_{J}}\right)^{-1}:\prod_{I=1}^{k}\mathsf{A}(x_{I})^{-1}\mathsf{Y}(x)^{2}:.
\eea
The contour integrand includes some rational function and additionally a vertex operator term.

Given this contour integral, one would like to evaluate the poles and expand the formula. However, this time we have a second order pole. The existence of this second order pole makes the situation complicated. When there is only a single order pole, the residue of the contour integrand is evaluated by taking the residue of the rational function part and simply inserting the information of the poles to the vertex operator part. This is because after taking the OPE, the vertex operator part is not singular anymore. When there is a second order pole, when taking the residue, we need to take derivatives of the contour integrand. Such derivative also acts on the vertex operator part.

Let us see this explicitly. For level one $k=1$, the residue is performed as
\bea
&\frac{(1-q_{12})}{(1-q_{1})(1-q_{2})}\times \underset{x_{1}=x}{\Res}\left(x_{1}^{-1}\mathscr{S}_{12}\left(\frac{x}{x_{1}}\right)^{2}:\mathsf{Y}(x)^{2}\mathsf{A}^{-1}(x_{1}):\right)\\
=&\frac{(1-q_{12})}{(1-q_{1})(1-q_{2})}\times\lim_{x_{1}\rightarrow x} \frac{\partial}{\partial x_{1}}\left( x_{1}^{-1}(x_{1}-x)^{2}\mathscr{S}_{12}\left(\frac{x}{x_{1}}\right)^{2} :\mathsf{Y}(x)^{2}\mathsf{A}^{-1}(x_{1}): \right)\\
=&\frac{(1-q_{12})}{(1-q_{1})(1-q_{2})}\times \left\{\left.\frac{\partial}{\partial x_{1}}\left(x_{1}^{-1}(x_{1}-x)^{2}\mathscr{S}_{12}\left(\frac{x}{x_{1}}\right)^{2}\right)\right|_{x_{1}=x}:\mathsf{Y}(x)^{2}\mathsf{A}^{-1}(x):\right.\\
&\qquad \left.+\left.\left(x_{1}^{-1}(x_{1}-x)^{2}\mathscr{S}_{12}\left(\frac{x}{x_{1}}\right)^{2}\right)\right|_{x_{1}=x}:\mathsf{Y}(x)^{2}\partial_{x}\mathsf{A}^{-1}(x):\right\}\\
=&\mathfrak{c}(q_{1},q_{2}) :\frac{\mathsf{Y}^{2}(x)}{\mathsf{A}(x)}:+\frac{(1-q_{1})(1-q_{2})}{(1-q_{12})}:\mathsf{Y}(x)^{2}\partial_{\log x}\mathsf{A}^{-1}(x):
\eea
where note that $\partial_{\log x}=x\partial_{x}$.

For level two, schematically, the integrand is
\bea
\frac{(x_{1}-x_{2})^{2}f_{0}(x_{1},x_{2})}{(1-v/x_{1})^{2}(1-v/x_{2})^{2}}
\eea
where $f_{0}(x_{1},x_{2})$ is a product of some rational function and the vertex operator part, which has no poles and zeros at $x_{1}=x_{2}=x$. After taking the residue at $x_{1}=x$:
\bea
\underset{x_{1}=x}{\Res}\frac{(x_{1}-x_{2})^{2}f_{0}(x_{1},x_{2})}{(x_{1}-x)^{2}(x_{2}-x)^{2}}&=\left.\partial_{x_{1}}\left(\frac{(x_{1}-x_{2})^{2}f_{0}(x_{1},x_{2})}{(x_{2}-x)^{2}}\right)\right|_{x_{1}=x}\\
&=\left.\frac{2(x_{1}-x_{2})f_{0}(x_{1},x_{2})}{(x_{2}-x)^{2}}\right|_{x_{1}=x}+\partial_{x_{1}}f_{0}(x_{1},x_{2})\\
&=-\frac{2f_{0}(x,x_{2})}{(x_{2}-x)}+\partial_{x_{1}}f_{0}(x_{1},x_{2})
\eea
Further taking the residue at $x_{2}=x$, no derivative of the vertex operator part appears. After a detailed computation, one obtains
\bea
:\frac{\mathsf{Y}(x)^{2}}{\mathsf{A}^{2}(x)}:
\eea
where the coefficient is simply 1.

The $qq$-character $\mathsf{T}^{(2)}(x)$ is then given as
\bea
\mathsf{T}^{(2)}(x)=:\mathsf{Y}^{2}(x):+\fq\left(\mathfrak{c}(q_{1},q_{2}):\frac{\mathsf{Y}^{2}(x)}{\mathsf{A}(x)}:+\frac{(1-q_{1})(1-q_{2})}{(1-q_{12})}:\mathsf{Y}(x)^{2}\partial_{\log x}\mathsf{A}^{-1}(x):\right)+\fq^{2}:\frac{\mathsf{Y}(x)^{2}}{\mathsf{A}^{2}(x)}:.
\eea
Note that the vacuum expectation value does not coincide with the partition function anymore:
\bea
\langle \mathsf{T}^{(2)}(x) \rangle  \neq \mathcal{Z}[\fq].
\eea
However, the coefficients of the $qq$-character without derivative terms still contain the information of the partition function.

\paragraph{Collision limit of $qq$-characters}
 Let us derive the $qq$-character whose highest weight is the $:\mathsf{Y}^{2}(x):$ by taking the collision limit $x'\rightarrow x$ in \eqref{eq:rank2qqcharacter}. The first and last term is easy since they simply become 
 \bea
:\mathsf{Y}(x)^{2}:,\quad :\frac{\mathsf{Y}(x)^{2}}{\mathsf{A}(x)^{2}}:,
 \eea
respectively. For the second and third terms, since $\mathscr{S}_{12}(z)$ has a pole at $z=1$, naively taking the limit $z\rightarrow 1$ is singular and thus much care is necessary.

Let us see what will happen explicitly for the second and third terms. Let $z=x'/x$ and consider the limit $z=e^{\varepsilon}\rightarrow 1$. The second term is
\bea
\mathscr{S}_{12}(z^{-1}):\frac{\mathsf{Y}(x)\mathsf{Y}(zx)}{\mathsf{A}(zx)}:&=\mathscr{S}_{12}(z^{-1})\left(:\frac{\mathsf{Y}(x)^{2}}{\mathsf{A}(x)}:+(z-1):\mathsf{Y}(x)\partial_{\log x}(\mathsf{Y}(x)\mathsf{A}^{-1}(x)):+\mathcal{O}((z-1)^{2})\right)\\
&=\mathscr{S}_{12}(z^{-1}):\frac{\mathsf{Y}(x)^{2}}{\mathsf{A}(x)}:+z\frac{(1-q_{1}z^{-1})(1-q_{2}z^{-1})}{(1-q_{12}z^{-1})}:\mathsf{Y}(x)\partial_{\log x}(\mathsf{Y}(x)\mathsf{A}^{-1}(x)):+\mathcal{O}(z-1),\\
\mathscr{S}_{12}(z):\frac{\mathsf{Y}(x)\mathsf{Y}(zx)}{\mathsf{A}(x)}:&=\mathscr{S}_{12}(z)\left(:\frac{\mathsf{Y}(x)^{2}}{\mathsf{A}(x)}:+(z-1):\frac{\mathsf{Y}(x)}{\mathsf{A}(x)}\partial_{\log x}\mathsf{Y}(x):+\mathcal{O}((z-1)^{2})\right)\\
&=\mathscr{S}_{12}(z):\frac{\mathsf{Y}(x)^{2}}{\mathsf{A}(x)}:-\frac{(1-q_{1}z)(1-q_{2}z)}{(1-q_{12}z)}:\frac{\mathsf{Y}(x)}{\mathsf{A}(x)}\partial_{\log x}\mathsf{Y}(x):+\mathcal{O}(z-1)
\eea
where we used
\bea
f(zx)&=f(x+\varepsilon x)=f(x)+\varepsilon x\partial_{x}f(x)+\mathcal{O}(\varepsilon^2)\\
&=f(x)+(z-1)\partial_{\log x}f(x)+\mathcal{O}((z-1)^2).
\eea
Taking the limit $z\rightarrow 1$ gives
\bea\label{eq:A1qq-collision}
&\mathscr{S}_{12}(z^{-1}):\frac{\mathsf{Y}(x)\mathsf{Y}(zx)}{\mathsf{A}(zx)}:+\mathscr{S}_{12}(z):\frac{\mathsf{Y}(x)\mathsf{Y}(zx)}{\mathsf{A}(x)}:\\
\xrightarrow{z\rightarrow 1}&\left(\mathscr{S}_{12}(z)+\mathscr{S}_{12}(z^{-1})\right):\frac{\mathsf{Y}(x)^{2}}{\mathsf{A}(x)}:+\frac{(1-q_{1})(1-q_{2})}{(1-q_{12})}:\mathsf{Y}(x)\partial_{\log x}\mathsf{A}^{-1}(x):\\
&=\mathfrak{c}(q_{1},q_{2}):\frac{\mathsf{Y}(x)^{2}}{\mathsf{A}(x)}:+\frac{(1-q_{1})(1-q_{2})}{(1-q_{12})}:\mathsf{Y}(x)\partial_{\log x}\mathsf{A}^{-1}(x),
\eea
where in the last line, we used
\bea\label{eq:A1qq-collision-part}
\mathscr{S}_{12}(z)&\simeq \frac{(1-q_{1})(1-q_{2})}{(1-q_{12})}\frac{1}{1-z}+\frac{q_{1}+q_{2}-3q_{12}+q_{12}^{2}}{(1-q_{12})^{2}}+\cdots,\\
\mathscr{S}_{12}(z^{-1})&\simeq -\frac{(1-q_{1})(1-q_{2})}{(1-q_{12})}\frac{1}{1-z}+\frac{1-3q_{12}+q_{1}^{2}q_{2}+q_{1}q_{2}^{2}}{(1-q_{12})^{2}}+\cdots ,\\
\mathfrak{c}(q_{1},q_{2})&=\,\underset{z=1}{\Res}\frac{\mathscr{S}_{12}(z)}{z-1}+\underset{z=1}{\Res}\frac{\mathscr{S}_{12}(z^{-1})}{z-1}.
\eea

Therefore, we have the $qq$-character
\bea\label{eq:A1qq-collision-part2}
\mathsf{T}^{(2)}(x)&={:\mathsf{Y}^{2}(x):}+\frac{q_{1}+q_{2}-3q_{12}+q_{12}^{2}}{(1-q_{12})^{2}}:\frac{\mathsf{Y}^{2}(x)}{\mathsf{A}(x)}:\\
&+\left(\frac{1-3q_{12}+q_{1}^{2}q_{2}+q_{1}q_{2}^{2}}{(1-q_{12})^{2}}:\frac{\mathsf{Y}^{2}(x)}{\mathsf{A}(x)}:+\frac{(1-q_{1})(1-q_{2})}{(1-q_{12})}:\mathsf{Y}(x)^{2}\partial_{\log x}\mathsf{A}^{-1}(x):\right)+:\frac{\mathsf{Y}(x)^{2}}{\mathsf{A}^{2}(x)}:\\
&={:\mathsf{Y}^{2}(x):}+\left(\mathfrak{c}(q_{1},q_{2}):\frac{\mathsf{Y}^{2}(x)}{\mathsf{A}(x)}:+\frac{(1-q_{1})(1-q_{2})}{(1-q_{12})}:\mathsf{Y}(x)^{2}\partial_{\log x}\mathsf{A}^{-1}(x):\right)+:\frac{\mathsf{Y}(x)^{2}}{\mathsf{A}^{2}(x)}:.
\eea
The topological term is always set $\fq=1$ for convenience.

It is natural to identify the partition functions of each configurations to be
\bea\label{eq:A1-partfunct-identify}
\mathcal{Z}[\varnothing]=1,\quad \mathcal{Z}[\Bbox_{L}]=\frac{q_{1}+q_{2}-3q_{12}+q_{12}^{2}}{(1-q_{12})^{2}},\\
\mathcal{Z}[\Bbox_{H}]=\frac{1-3q_{12}+q_{1}^{2}q_{2}+q_{1}q_{2}^{2}}{(1-q_{12})^{2}},\quad \mathcal{Z}[\Bbox_{U}]=1.
\eea
Note that all of the factors become $1$ at the limit $q_{2}\rightarrow1 $.


\paragraph{Commutation relation with the screening charge}
Compared to the $q$-Virasoro algebra, deriving the $qq$-character $\mathsf{T}^{(2)}(x)$ directly using the commutativity with the screening charge is difficult. In the following, we will only compute the commutativity of the vertex operators with the screening currents where second order poles appear. To obtain the commutativity with the full screening charge, we need to organize the vertex operators where derivative terms appear, but this part is difficult. We leave this for future work.

Again, we focus on the screening charge $\mathscr{Q}_{2}(x)$. We have
\bea
\,&\mathsf{Y}(x)^{2} \mathsf{S}_{2}(x')=\left(\frac{1-q_{12}^{-1}x'/x}{1-q_{2}^{-1}x'/x}\right)^{2}:\mathsf{Y}(x)^{2} \mathsf{S}_{2}(x'):,\quad \mathsf{S}_{2}(x')\mathsf{Y}(x)^{2}=q_{1}^{-2}\left(\frac{1-q_{12}x/x'}{1-q_{2}x/x'}\right)^{2}:\mathsf{Y}(x)^{2}\mathsf{S}_{2}(x'):,\\
&:\mathsf{Y}(x)^{2}\mathsf{A}(x)^{-1}: \mathsf{S}_{2}(x')=\frac{1-q_{1}x'/x}{1-x'/x}\frac{1-q_{12}^{-1}x'/x}{1-q_{2}^{-1}x'/x}:\mathsf{Y}(x)^{2}\mathsf{A}(x)^{-1} \mathsf{S}_{2}(x'):,\\
&\mathsf{S}_{2}(x'):\mathsf{Y}(x)^{2}\mathsf{A}(x)^{-1}: =\frac{1-q_{1}^{-1}x/x'}{1-x/x'}\frac{1-q_{12}x/x'}{1-q_{2}x/x'}:\mathsf{Y}(x)^{2}\mathsf{A}(x)^{-1} \mathsf{S}_{2}(x'):,\\
&:\mathsf{Y}(x)^{2}\partial_{\log x}\mathsf{A}(x)^{-1}: \mathsf{S}_{2}(x')=x\left(\frac{1-q_{12}^{-1}x'/x}{1-q_{2}^{-1}x'/x}\right)^{2}\partial_{x}\left(\frac{(1-q_{1}x'/x)(1-q_{2}^{-1}x'/x)}{(1-x'/x)(1-q_{12}^{-1}x'/x)}\right):\mathsf{Y}(x)^{2}\partial_{\log x}\mathsf{A}(x)^{-1} \mathsf{S}_{2}(x'):,\\
&\mathsf{S}_{2}(x'):\mathsf{Y}(x)^{2}\partial_{\log x}\mathsf{A}(x)^{-1}:=x \left(\frac{1-q_{12}x/x'}{1-q_{2}x/x'}\right)^{2}\partial_{x}\left(\frac{(1-q_{1}^{-1}x/x')(1-q_{2}x/x')}{(1-x/x')(1-q_{12}x/x')}\right) :\mathsf{Y}(x)^{2}\partial_{\log x}\mathsf{A}(x)^{-1} \mathsf{S}_{2}(x'):,  \\
&:\mathsf{Y}(x)^{2}\mathsf{A}(x)^{-2}: \mathsf{S}_{2}(x')=\left(\frac{1-q_{1}x'/x}{1-x'/x}\right)^{2}:\mathsf{Y}(x)^{2}\mathsf{A}(x)^{-2} \mathsf{S}_{2}(x'):,\\
&\mathsf{S}_{2}(x'):\mathsf{Y}(x)^{2}\mathsf{A}(x)^{-2}: =q_{1}^{2}\left(\frac{1-q_{1}^{-1}x/x'}{1-x/x'}\right)^{2}:\mathsf{Y}(x)^{2}\mathsf{A}(x)^{-2} \mathsf{S}_{2}(x'):.
\eea

Let us first consider the commutativity of $:\mathsf{Y}^{2}(x):$ and $\mathsf{S}_{2}(x')$. To deal with the double pole, we start from $:\mathsf{Y}(x)\mathsf{Y}(zx):$ and then take the limit $z\rightarrow 1$. We first have 
\bea
\,&:\mathsf{Y}(x)\mathsf{Y}(xz): \mathsf{S}_{2}(x')=\frac{1-q_{12}^{-1}x'/x}{1-q_{2}^{-1}x'/x}\frac{1-q_{12}^{-1}z^{-1}x'/x}{1-q_{2}^{-1}z^{-1}x'/x}:\mathsf{Y}(x)\mathsf{Y}(xz) \mathsf{S}_{2}(x'):,
\eea
which gives
\bea
\relax[:\mathsf{Y}(x)\mathsf{Y}(xz):, \mathsf{S}_{2}(x')]&=\delta\left(q_{2}^{-1}x'/x\right)\frac{(1-q_{1}^{-1})(1-q_{1}^{-1}z^{-1})}{(1-z^{-1})}:\mathsf{Y}(x)\mathsf{Y}(xz) \mathsf{S}_{2}(q_{2}x):\\
&+\delta\left(q_{2}^{-1}z^{-1}x'/x\right)\frac{(1-q_{1}^{-1})(1-q_{1}^{-1}z)}{(1-z)}:\mathsf{Y}(x)\mathsf{Y}(xz) \mathsf{S}_{2}(q_{2}zx):.
\eea
Expanding around $z=1$ and taking the limit $z\rightarrow 1$, we have
\bea
&\delta\left(\frac{x'}{q_{2}x}\right)\frac{(1-q_{1}^{-1})(1-q_{1}^{-1}z^{-1})}{(1-z^{-1})}:\mathsf{Y}(x)\mathsf{Y}(xz) \mathsf{S}_{2}(q_{2}x):\\
=&(1-q_{1}^{-1})\delta\left(\frac{x'}{q_{2}x}\right)\frac{(1-q_{1}^{-1}z^{-1})}{(1-z^{-1})}\left(:\mathsf{Y}(x)^{2}\mathsf{S}_{2}(q_{2}x):+(z-1):\mathsf{Y}(x)\partial_{\log x}\mathsf{Y}(x)\mathsf{S}_{2}(q_{2}x):\right)\\
= &\delta\left(\frac{x'}{q_{2}x}\right)\left(\frac{(1-q_{1}^{-1})(1-q_{1}^{-1}z^{-1})}{(1-z^{-1})}:\mathsf{Y}(x)^{2}\mathsf{S}_{2}(q_{2}x):+(1-q_{1}^{-1})^{2}:\mathsf{Y}(x)\partial_{\log x}\mathsf{Y}(x)\mathsf{S}_{2}(q_{2}x):\right)
\eea
and 
\bea
&\delta\left(\frac{x'}{zq_{2}x}\right)\frac{(1-q_{1}^{-1})(1-q_{1}^{-1}z)}{(1-z)}:\mathsf{Y}(x)\mathsf{Y}(xz) \mathsf{S}_{2}(q_{2}zx):\\
=&\left(\delta\left(\frac{x'}{q_{2}x}\right)+(z^{-1}-1)\left.\partial_{\log y}\delta(y)\right|_{y=x'/q_{2}x}\right)\frac{(1-q_{1}^{-1})(1-q_{1}^{-1}z)}{(1-z)}\\
&\times \left(:\mathsf{Y}(x)^{2}\mathsf{S}_{2}(q_{2}x):+(z-1):\mathsf{Y}(x)\partial_{\log x}(\mathsf{Y}(x)\mathsf{S}_{2}(q_{2}x)):\right)\\
=&\delta\left(\frac{x'}{q_{2}x}\right)\frac{(1-q_{1}^{-1})(1-q_{1}^{-1}z)}{(1-z)}:\mathsf{Y}(x)^{2}\mathsf{S}_{2}(q_{2}x)-\delta\left(\frac{x'}{q_{2}x}\right)(1-q_{1}^{-1})^{2}x:\mathsf{Y}(x)\partial_{x}(\mathsf{Y}(x)\mathsf{S}_{2}(q_{2}x)):\\
&+(1-q_{1}^{-1})^{2}\partial_{\log y}\delta(y)|_{y=x'/xq_{2}}:\mathsf{Y}(x)^{2}\mathsf{S}_{2}(q_{2}x):.
\eea
Therefore, using
\bea
\frac{(1-q_{1}^{-1}z^{-1})}{1-z^{-1}}+\frac{(1-q_{1}^{-1}z)}{1-z}=\frac{1-q_{1}^{-1}z-z+q_{1}^{-1}}{1-z}\longrightarrow 1+q_{1}^{-1}
\eea
we get
\bea
\relax [\mathsf{Y}(x)^{2},\mathsf{S}_{2}(x')]&=+(1-q_{1}^{-1})^{2}\left(\delta^{(1)}\left(\frac{x'}{q_{2}x}\right)-\delta\left(\frac{x'}{q_{2}x}\right)\right):\mathsf{Y}(x)^{2}\mathsf{S}_{2}(q_{2}x):+\delta\left(\frac{x'}{q_{2}x}\right)(1-q_{1}^{-2}):\mathsf{Y}(x)^{2}\mathsf{S}_{2}(q_{2}x):\\
&-\delta\left(\frac{x'}{q_{2}x}\right)(1-q_{1}^{-1})^{2}:\mathsf{Y}(x)^{2}\left.\partial_{\log y}\mathsf{S}_{2}(y)\right|_{y=q_{2}x}:\\
&=+(1-q_{1}^{-1})^{2}{\delta^{(1)}}\left(\frac{x'}{q_{2}x}\right):\mathsf{Y}(x)^{2}\mathsf{S}_{2}(q_{2}x):+\delta\left(\frac{x'}{q_{2}x}\right)2q_{1}^{-1}(1-q_{1}^{-1}):\mathsf{Y}(x)^{2}\mathsf{S}_{2}(q_{2}x):\\
&-\delta\left(\frac{x'}{q_{2}x}\right)(1-q_{1}^{-1})^{2}:\mathsf{Y}(x)^{2}\left.\partial_{\log y}\mathsf{S}_{2}(y)\right|_{y=q_{2}x}:
\eea
where we used \eqref{eq:logderivdelta-relation}.

A more simple way to derive this relation is to expand the contraction formula around $x'=q_{2}x$ directly:
\bea
\mathsf{Y}(x)^{2} \mathsf{S}_{2}(x')
&\simeq \frac{(1-q_{1}^{-1})^{2}}{(1-q_{2}^{-1}x'/x)^{2}}:\mathsf{Y}(x)^{2}\mathsf{S}_{2}(q_{2}x):-\frac{2q_{1}^{-1}(1-q_{1}^{-1})}{1-q_{2}^{-1}x'/x}:\mathsf{Y}(x)^{2}\mathsf{S}_{2}(q_{2}x):\\
&\qquad -\frac{(1-q_{1}^{-1})^{2}}{(1-q_{2}^{-1}x'/x)}:\mathsf{Y}(x)^{2}\left.\partial_{\log y}\mathsf{S}_{2}(y)\right|_{y=q_{2}x}:.
\eea
Combining with the computation of the different analytic region, the result is
\bea
\relax[\mathsf{Y}^{2}(x),\mathsf{S}_{2}(x')]&=+(1-q_{1}^{-1})^{2}{\delta^{(1)}}\left(\frac{q_{2}^{-1}x'}{x}\right):\mathsf{Y}(x)^{2}\mathsf{S}_{2}(q_{2}x):+2q_{1}^{-1}(1-q_{1}^{-1})\delta\left(\frac{q_{2}^{-1}x'}{x}\right):\mathsf{Y}(x)^{2}\mathsf{S}_{2}(q_{2}x):\\
&-(1-q_{1}^{-1})^{2}\delta\left(\frac{x'}{q_{2}x}\right):\mathsf{Y}(x)^{2}\left.\partial_{\log y}\mathsf{S}_{2}(y)\right|_{y=q_{2}x}:
\eea
where we used Thm.~\ref{thm:rat-delta-deriv}. Indeed the result matches with the previous computation.

Since there are no higher order poles for the commutation coming from $:\mathsf{Y}(x)^{2}\mathsf{A}(x)^{-1}:$, we simply obtain
\bea
\relax[:\mathsf{Y}(x)^{2}\mathsf{A}(x)^{-1}:, \mathsf{S}_{2}(x')]&=\delta\left(\frac{x'}{x}\right)\frac{(1-q_{1})(1-q_{12}^{-1})}{(1-q_{2}^{-1})}:\mathsf{Y}(x)^{2}\mathsf{A}(x)^{-1} \mathsf{S}_{2}(x):\\
&-\delta\left(\frac{x'}{q_{2}x}\right)\frac{(1-q_{12}^{-1})(1-q_{1})}{(1-q_{2}^{-1})}:\mathsf{Y}(x)^{2}\mathsf{A}(x)^{-1} \mathsf{S}_{2}(q_{2}x):.
\eea

The OPE between $:\mathsf{Y}(x)^{2}\partial_{\log x}\mathsf{A}(x)^{-1}:$ and $\mathsf{S}_{2}(x')$ is simplified as
\bea
:\mathsf{Y}(x)^{2}\partial_{\log x}\mathsf{A}(x)^{-1}:\mathsf{S}_{2}(x')&=\frac{(1-q_{1}^{-1})xx'(q_{2}x^{2}+q_{1}q_{2}^{2}x^{2}-2q_{2}xx'-2q_{12}xx'+{x'}^{2}+q_{12}{x'}^{2})}{(x-x')^{2}(q_{2}x-x')^{2}}\\
&\qquad\times :\mathsf{Y}(x)^{2}\partial_{\log x}\mathsf{A}(x)^{-1}\mathsf{S}_{2}(x'):\\
&\simeq \frac{(1-q_{1}^{-1})(1-q_{12})}{(1-q_{2})}\frac{:\mathsf{Y}(x)^{2}\partial_{\log x}\mathsf{A}(x)^{-1}\mathsf{S}_{2}(x):}{(1-x'/x)^{2}}\\
&-\frac{(1-q_{1}^{-1})(1-3q_{2}+q_{12}+q_{1}q_{2}^{2})}{(1-q_{2})^{2}}\frac{:\mathsf{Y}(x)^{2}\partial_{\log x}\mathsf{A}(x)^{-1}\mathsf{S}_{2}(x):}{1-x'/x}\\
&-\frac{(1-q_{1}^{-1})(1-q_{12})}{(1-q_{2})}\frac{:\mathsf{Y}(x)^{2}\partial_{\log x}\mathsf{A}(x)^{-1}\partial_{\log x}\mathsf{S}_{2}(x):}{(1-x'/x)}
\eea
and
\bea
:\mathsf{Y}(x)^{2}\partial_{\log x}\mathsf{A}(x)^{-1}:\mathsf{S}_{2}(x')
&\simeq \frac{(1-q_{1}^{-1})(1-q_{12})}{(1-q_{2})}\frac{:\mathsf{Y}(x)^{2}\partial_{\log x}\mathsf{A}(x)^{-1}\mathsf{S}_{2}(q_{2}x):}{(1-x'/q_{2}x)^{2}}\\
&-\frac{(1-q_{1}^{-1})(1+q_{2}-3q_{12}+q_{1}q_{2}^{2})}{(1-q_{2})^{2}}\frac{:\mathsf{Y}(x)^{2}\partial_{\log x}\mathsf{A}(x)^{-1}\mathsf{S}_{2}(q_{2}x):}{(1-x'/q_{2}x)}\\
&-\frac{(1-q_{1}^{-1})(1-q_{12})}{(1-q_{2})}\frac{:\mathsf{Y}(x)^{2}\partial_{\log x}\mathsf{A}(x)^{-1}\left.\partial_{\log y}\mathsf{S}_{2}(y)\right|_{y=q_{2}x}:}{(1-x'/q_{2}x)}
\eea
where we expanded around $x'=x$ and $x'=q_{2}x$, respectively. We thus have
\bea
\relax&[:\mathsf{Y}(x)^{2}\partial_{\log x}\mathsf{A}(x)^{-1}:,\mathsf{S}_{2}(x')]\\
=&\frac{(1-q_{1}^{-1})(1-q_{12})}{(1-q_{2})}\left(:\mathsf{Y}(x)^{2}\partial_{\log x}\mathsf{A}(x)^{-1}\mathsf{S}_{2}(x):\delta^{(1)}\left(\frac{x'}{x}\right)+:\mathsf{Y}(x)^{2}\partial_{\log x}\mathsf{A}(x)^{-1}\mathsf{S}_{2}(q_{2}x):\delta^{(1)}\left(\frac{x'}{q_{2}x}\right)\right)\\
&-\frac{(1-q_{1}^{-1})}{(1-q_{2})^{2}}\left((1-3q_{2}+q_{12}+q_{1}q_{2}^{2}):\mathsf{Y}(x)^{2}\partial_{\log x}\mathsf{A}(x)^{-1}\mathsf{S}_{2}(x):\delta\left(\frac{x'}{x}\right)\right.\\
&\left.\qquad +(1+q_{2}-3q_{12}+q_{1}q_{2}^{2}):\mathsf{Y}(x)^{2}\partial_{\log x}\mathsf{A}(x)^{-1}\mathsf{S}_{2}(q_{2}x):\delta\left(\frac{x'}{q_{2}x}\right) \right)\\
&-\frac{(1-q_{1}^{-1})(1-q_{12})}{(1-q_{2})}\left(:\mathsf{Y}(x)^{2}\partial_{\log x}\mathsf{A}(x)^{-1}\partial_{\log x}\mathsf{S}_{2}(x):\delta\left(\frac{x'}{x}\right)\right.\\
&\left.\qquad +:\mathsf{Y}(x)^{2}\partial_{\log x}\mathsf{A}(x)^{-1}\left.\partial_{\log y}\mathsf{S}_{2}(y)\right|_{y=q_{2}x}:\delta\left(\frac{x'}{q_{2}x}\right)\right).
\eea

Similarly, we have
\bea
:\mathsf{Y}(x)^{2}\mathsf{A}(x)^{-2}:\mathsf{S}_{2}(x')&\simeq \frac{(1-q_{1})^{2}}{(1-x'/x)^{2}}:\mathsf{Y}(x)^{2}\mathsf{A}(x)^{-2}\mathsf{S}_{2}(x):+\frac{2q_{1}(1-q_{1})}{1-x'/x}:\mathsf{Y}(x)^{2}\mathsf{A}(x)^{-2}\mathsf{S}_{2}(x):\\
&\quad -\frac{(1-q_{1})^{2}}{1-x'/x}:\mathsf{Y}(x)^{2}\mathsf{A}(x)^{-2}\partial_{\log x}\mathsf{S}_{2}(x):
\eea
and
\bea
\relax[:\mathsf{Y}(x)^{2}\mathsf{A}(x)^{-2}:,\mathsf{S}_{2}(x')]&=(1-q_{1})^{2}\delta^{(1)}\left(\frac{x'}{x}\right):\mathsf{Y}(x)^{2}\mathsf{A}(x)^{-2}\mathsf{S}_{2}(x):+2q_{1}(1-q_{1})\delta\left(\frac{x'}{x}\right):\mathsf{Y}(x)^{2}\mathsf{A}(x)^{-2}\mathsf{S}_{2}(x):\\
&-(1-q_{1})^{2}\delta\left(\frac{x'}{x}\right):\mathsf{Y}(x)^{2}\mathsf{A}(x)^{-2}\partial_{\log x}\mathsf{S}_{2}(x):.
\eea

\subsection{Quantum affine algebra}\label{sec:quantum-affine-qq-character}
\subsubsection{Spin $1/2$ representation and $A_{1}$ $qq$-character}
In this section, we review the relation of the $A_{1}$ $qq$-character and the quantum affine $\mathcal{U}_{\sfq}(\widehat{\mathfrak{sl}}_{2})$ algebra. See \cite{Kimura:2020jxl} for a review. We follow the expression for the quantum affine algebra in~\cite{Drinfeld:1986in,Ding-Frenkel-quantumaffine}.
\paragraph{Drinfeld presentation and vertical representation} Let $E(z),F(z),K^{\pm}(z)$ be the Drinfeld currents of the quantum affine $\mathcal{U}_{\sfq}(\widehat{\mathfrak{sl}}_{2})$ algebra:
\bea
E(z)=\sum_{m\in\mathbb{Z}}E_{m}z^{-m},\quad F(z)=\sum_{m\in\mathbb{Z}}F_{m}z^{-m},\quad K^{\pm}(z)=\sum_{r\geq 0}K^{\pm }_{\pm r}z^{\mp r}.
\eea
We use the same notation for the Drinfeld currents as those of the quantum toroidal algebra in section~\ref{sec:QTgl1} because the structure function is only the part that differs. The defining relations are
\bea
E(z)E(w)=G(z/w)E(w)E(z),&\quad F(z)F(w)=G(z/w)^{-1}F(w)F(z),\\
K^{\pm}(z)K^{\pm}(w)=K^{\pm}(w)K^{\pm}(z),&\quad K^{-}(z)K^{+}(w)=K^{+}(w)K^{-}(z),\\
K^{\pm}(z)E(w)=G(z/w)E(w)K^{\pm}(z),&\quad K^{\pm}(z)F(w)=G(z/w)^{-1}F(w)K^{\pm}(z),\\
[E(z),F(w)]&=\frac{1}{\sfq^{1/2}-\sfq^{-1/2}}\left(\delta\left(\frac{w}{z}\right)K^{+}(z)-\delta\left(\frac{z}{w}\right)K^{-}(w)\right)
\eea
where the structure function is defined as
\bea
G(z)=\sfq^{-1}\frac{1-\sfq z}{1-\sfq^{-1}z}
\eea
and additionally we have a central element $C$, but we kept it $C=1$. Namely, we are focusing on the vertical representations of the quantum affine algebra. 

The coproduct structure is 
\bea\label{eq:qaffine-coproduct}
\Delta (E(z))&=E(z)\otimes 1+K^{-}(z)\otimes E(z),\\
\Delta (F(z))&=F(z)\otimes K^{+}(z)+1\otimes F(z),\\
\Delta (K^{\pm}(z))&=K^{\pm}(z)\otimes K^{\pm}(z),
\eea
where we kept the central charge to be $C=1$ again.

Note also that the structure function obeys the condition
\bea
G(z)G(z^{-1})=1.
\eea
The relation with the structure function of the quantum toroidal $\mathfrak{gl}_{1}$ is 
\bea
\mathsf{g}(z)=\prod_{i=1}^{3}\frac{(1-\mathsf{q}_{i}z)}{(1-\mathsf{q}_{i}^{-1}z)}\xrightarrow {\sfq_{2}\rightarrow 0,\,\, \sfq_{3}\rightarrow \infty}\sfq_{1}^{-1}\frac{1-\sfq_{1}z}{1-\sfq_{1}^{-1}z}=G(z)
\eea
where we kept $\sfq_{2}\sfq_{3}=\sfq_{1}^{-1}$ fixed while taking the limit and identified $\sfq_{1}=\sfq$.

The spin $1/2$ representation is obtained by imposing the following highest weight:
\bea
K^{\pm}(z)\ket{0}=\left[\Psi_{\varnothing,u}(z)\right]_{\pm}\ket{0},\quad \Psi_{\varnothing,u}(z)=\sfq^{-1/2}  \frac{1-\sfq u/z}{1-u/z}.
\eea
Since we have a pole at $z=u$, we can add a box at the origin and assume
\bea
E(z)\ket{0}=\mathcal{E}\delta\left(\frac{u}{z}\right)\ket{\Bbox},
\eea
where $\mathcal{E}$ is some coefficient. Using the $EK$-relation, we have
\bea
K^{\pm}(z)\ket{\Bbox}&=\left[\Psi_{\varnothing,u}(z)G\left(\frac{u}{z}\right)^{-1}\right]_{\pm}\ket{\Bbox}\\
&=\left[\sfq^{1/2}\frac{1-\sfq^{-1}u/z}{1-u/z}\right]_{\pm}\ket{\Bbox}.
\eea
We only have a removable pole and no more addable poles and so the growth of the crystal stops here, giving a two-dimensional representation space:
\bea
E(z)\ket{\Bbox}=0.
\eea

The action of $F(z)$ is determined as
\bea
F(z)\ket{\Bbox}&=\mathcal{F}\delta\left(\frac{u}{z}\right)\ket{0},\quad F(z)\ket{0}=0.
\eea
Using the $EF$ relation, we obtain the condition
\bea
\mathcal{E}\mathcal{F}=1.
\eea
We can simply impose $\mathcal{E}=\mathcal{F}=1$.

\paragraph{Relation with the $A_{1}$ $qq$-character}
The monomial terms of the fundamental $qq$-character of $A_{1}$ theory has a one to one correspondence with the bases of the spin $1/2$ representation of the quantum affine algebra. The $A_{1} $ $qq$-character has two terms
\bea
\mathsf{Y}(x),\quad :\mathsf{Y}(x)\mathsf{A}(x)^{-1}:.
\eea
The first term corresponds to the vacuum configuration while the second term corresponds to the configuration with one box at the coordinate $x$. To relate with the representation above, we need to replace $u\rightarrow x$. 

Let us see the relation with the screening charge
\bea
\bra{0}\mathsf{S}_{2}(q_{2}x')\mathsf{Y}(x)\ket{0}=q_{1}^{-1}\frac{1-q_{1}x/x'}{1-x/x'} &\xrightarrow{q_{2}\rightarrow 1} \sfq^{-1}\frac{1-\sfq x/x'}{1-x/x'}\\
\bra{0}\mathsf{S}_{2}(q_{2}x'):\mathsf{Y}(x)\mathsf{A}^{-1}(x):\ket{0}=q_{1}^{+1}\frac{(1-q_{12}^{-1}x/x')}{(1-q_{2}^{-1}x/x')}&\xrightarrow{q_{2}\rightarrow 1} \sfq\frac{1-\sfq^{-1}x/x'}{1-x/x'}.
\eea
After taking the NS limit $q_{2}\rightarrow 1$, the factors coincide with the Cartan eigenvalues $\Psi_{\varnothing,x}(x')$ and $\Psi_{\varnothing,x}(x')G(x'/x)^{-1}$ of $K^{\pm}(z)$ up to trivial zero-modes:
\bea
\bra{0}\mathsf{S}_{2}(q_{2}x')\mathsf{T}(x)\ket{0}\xrightarrow{q_{2}\rightarrow 1} \sfq^{-1/2}\left(\Psi_{\varnothing,x}(x')+\sfq \Psi_{\varnothing,x}(x')G\left(\frac{x}{x'}\right)^{-1}\right)
\eea

This is nothing but the fact that after taking the NS limit, the $A_{1}$ $qq$-character becomes the $q$-character of the spin $1/2$ representation of the quantum affine $\mathfrak{sl}_{2}$ algebra.

\subsubsection{Bootstrapping the module} \label{app-sec:A1collision-module-bootstrap}
We set the highest weight to be
\bea
\widetilde{\Psi}_{\varnothing,u}(z)=\Psi_{\varnothing,u}(z)^{2}=\sfq^{-1}\left(\frac{1-\sfq u/z}{1-u/z}\right)^{2}.
\eea
Let us construct a module with this highest weight 
following the bootstrap method of \cite{Gaiotto:2020dsq}. We want to construct a module with the following transition rules
\bea
\begin{tikzpicture}
\node at (-1.5,0){0};
\node at (0,0) {$\ket{0}$};    
\node at (1.5,0.5){$\ket{\Bbox_L}$};
\node at (1.5,-0.5){$\ket{\Bbox_H}$};
\node at (3,0){$\ket{\Bbox_{U}}$};
\node at (4.5,0){0};
\draw[->] (-1.3,0)--(-.3,0);
\draw[->] (0.2,0.1)--(1,0.5);
\draw[->] (0.2,-0.1)--(1,-0.5);
\draw[->] (1.8,0.5)--(2.5,0.1);
\draw[->] (1.8,-0.5)--(2.5,-0.1);
\draw[->] (3.4,0)--(4.3,0);
\end{tikzpicture}
\eea
where the arrows indicate the action of the operator $E(z)$. The action of the $F(z)$ is the opposite but we omit it.

\paragraph{Vacuum $\leftrightarrow $ Level one}
We first start from the following ansatz:
\bea
E(z)\ket{0}=E(\varnothing\rightarrow \Bbox_{L})\tilde{\delta}^{(1)}\left(\frac{z}{u}\right)\ket{\Bbox_{L}}+E(\varnothing\rightarrow \Bbox_{H})\delta\left(\frac{z}{u}\right)\ket{\Bbox_{H}}.
\eea
Let us then use the $KE$ relation on the vacuum. We first have
\bea
K^{\pm}(z)E(w)\ket{0}&=E(\varnothing\rightarrow \Bbox_{L})\tilde{\delta}^{(1)}\left(\frac{w}{u}\right)K^{\pm}(z)\ket{\Bbox_{L}}+E(\varnothing\rightarrow \Bbox_{H})\delta\left(\frac{w}{u}\right)K^{\pm}(z)\ket{\Bbox_{H}}.
\eea
We also have
\bea
G\left(\frac{w}{z}\right)^{-1}E(w)K^{\pm}(z)\ket{0}&=G\left(\frac{w}{z}\right)^{-1}\widetilde{\Psi}_{\varnothing,u}(z)\left(E(\varnothing\rightarrow \Bbox_{L})\tilde{\delta}^{(1)}\left(\frac{w}{u}\right)\ket{\Bbox_{L}}+E(\varnothing\rightarrow \Bbox_{H})\delta\left(\frac{w}{u}\right)\ket{\Bbox_{H}}\right)\\
&=E(\varnothing\rightarrow \Bbox_{L})G\left(\frac{u}{z}\right)^{-1}\widetilde{\Psi}_{\varnothing,u}(z)\tilde{\delta}^{(1)}\left(\frac{w}{u}\right)\ket{\Bbox_{L}}\\
&-E(\varnothing\rightarrow \Bbox_{L})\left.\partial_{\log w}G\left(\frac{w}{z}\right)^{-1}\right|_{w=u}\widetilde{\Psi}_{\varnothing,u}(z)\delta\left(\frac{w}{u}\right)\ket{\Bbox_{L}}\\
&+E(\varnothing\rightarrow \Bbox_{H})\widetilde{\Psi}_{\varnothing,u}(z)G\left(\frac{u}{z}\right)^{-1}\delta\left(\frac{w}{u}\right)\ket{\Bbox_{H}}
\eea
where we used
\bea
\,&G\left(\frac{w}{z}\right)^{-1}=G\left(\frac{u}{z}\right)^{-1}+\left.\partial_{w}G\left(\frac{w}{z}\right)^{-1}\right|_{w=u}(w-u)+\ldots,\\
&(w-u)\delta^{(1)}\left(\frac{w}{u}\right)=-u\delta\left(\frac{w}{u}\right).
\eea
Comparing both hand sides gives
\bea
K^{\pm}(z)\ket{\Bbox_{L}}&=\left[\widetilde{\Psi}_{\varnothing,u}(z)G\left(\frac{u}{z}\right)^{-1}\right]_{\pm}\ket{\Bbox_{L}},\\
K^{\pm}(z)\ket{\Bbox_{H}}&=\left[\widetilde{\Psi}_{\varnothing,u}(z)G\left(\frac{u}{z}\right)^{-1}\right]_{\pm}\ket{\Bbox_{H}}-\frac{E(\varnothing\rightarrow \Bbox_{L})}{E(\varnothing\rightarrow \Bbox_{H})}\left[\widetilde{\Psi}_{\varnothing,u}(z)u\partial_{u}G\left(\frac{u}{z}\right)^{-1}\right]_{\pm}\ket{\Bbox_{L}}.
\eea
For later use, we will denote the derivatives with the log variables as
\bea
\tilde{\partial} f(x)=x\partial_{x}f(x)=\partial_{\log x}f(x).
\eea

Let us then move on to the action of $F(z)$. We impose the following ansatz
\bea
F(z)\ket{\Bbox_{L}}&=F(\Bbox_{L}\rightarrow \varnothing)\delta\left(\frac{z}{u}\right)\ket{0}\\
F(z)\ket{\Bbox_{H}}&=\left(F^{(2)}(\Bbox_{H}\rightarrow \varnothing)\tilde{\delta}^{(1)}\left(\frac{z}{u}\right)+F^{(1)}(\Bbox_{H}\rightarrow \varnothing)\delta\left(\frac{z}{u}\right)\right)\ket{0}.
\eea

The $KF$ relation on $\ket{\Bbox_L},\ket{\Bbox_H}$ give
\bea
K^{\pm}(z)F(w)\ket{\Bbox_L}&=F(\Bbox_L\rightarrow \varnothing)\widetilde{\Psi}_{\varnothing,u}(z)\delta\left(\frac{w}{u}\right)\ket{0},\\
G\left(\frac{w}{z}\right)F(w)K^{\pm}(z)\ket{\Bbox_L}&=F(\Bbox_L\rightarrow \varnothing)G\left(\frac{w}{z}\right)\widetilde{\Psi}_{\varnothing,u}(z)G\left(\frac{u}{z}\right)^{-1}\delta\left(\frac{w}{u}\right)\ket{0}\\
&=F(\Bbox_L\rightarrow \varnothing)\widetilde{\Psi}_{\varnothing,u}(z)\delta\left(\frac{w}{u}\right)\ket{0}
\eea
and 
\bea
K^{\pm}(z)F(w)\ket{\Bbox_H}&=\widetilde{\Psi}_{\varnothing,u}(z)\left(F^{(2)}(\Bbox_H\rightarrow \varnothing)\tilde{\delta}^{(1)}\left(\frac{w}{u}\right)+F^{(1)}(\Bbox_H\rightarrow \varnothing)\delta\left(\frac{w}{u}\right)\right)\ket{0},\\
G\left(\frac{w}{z}\right)F(w)K^{\pm}(z)\ket{\Bbox_H}&=G\left(\frac{w}{z}\right)F(w)\left(\widetilde{\Psi}_{\varnothing,u}(z)G\left(\frac{u}{z}\right)^{-1}\ket{\Bbox_{H}}\right.\\
&\left.-\frac{E(\varnothing\rightarrow \Bbox_{L})}{E(\varnothing\rightarrow \Bbox_{H})}\widetilde{\Psi}_{\varnothing,u}(z)u\partial_{u}G\left(\frac{u}{z}\right)^{-1}\ket{\Bbox_{L}}\right)\\
&=\left( F^{(2)}(\Bbox_H\rightarrow \varnothing)\widetilde{\Psi}_{\varnothing,u}(z)G\left(\frac{u}{z}\right)^{-1}G\left(\frac{w}{z}\right)\tilde{\delta}^{(1)}\left(\frac{w}{u}\right)\right.\\
&+F^{(1)}(\Bbox_H\rightarrow \varnothing)\widetilde{\Psi}_{\varnothing,u}(z)\delta\left(\frac{w}{u}\right)\\
&\left.-F(\Bbox_L\rightarrow \varnothing)\frac{E(\varnothing\rightarrow \Bbox_{L})}{E(\varnothing\rightarrow \Bbox_{H})}\widetilde{\Psi}_{\varnothing,u}(z)G\left(\frac{u}{z}\right)\tilde{\partial}G\left(\frac{u}{z}\right)^{-1}\delta\left(\frac{w}{u}\right) \right)\ket{0}.
\eea
Using
\bea\label{eq:Gfunctdeltaderiv}
G\left(\frac{w}{z}\right)^{-1}\tilde{\delta}^{(1)}\left(\frac{w}{u}\right)=G\left(\frac{u}{z}\right)^{-1}\tilde{\delta}^{(1)}\left(\frac{w}{u}\right)-\tilde{\partial}G\left(\frac{u}{z}\right)^{-1}\delta\left(\frac{w}{u}\right)
\eea
we obtain
\bea
    E(\varnothing\rightarrow \Bbox_H)F^{(2)}(\Bbox_H\rightarrow \varnothing) =E(\varnothing\rightarrow \Bbox_L)F(\Bbox_L\rightarrow \varnothing).
\eea

For the $EF$ relations on $\ket{0}$, we have
\bea
\,[E(z),F(w)]\ket{0}&=-E(\varnothing\rightarrow \Bbox_L)F(\Bbox_L\rightarrow \varnothing)\tilde{\delta}^{(1)}\left(\frac{z}{u}\right)\delta\left(\frac{w}{u}\right)\ket{0}\\
&-E(\varnothing\rightarrow \Bbox_H)F^{(2)}(\Bbox_H\rightarrow \varnothing)\delta\left(\frac{z}{u}\right)\tilde{\delta}^{(1)}\left(\frac{w}{u}\right)\ket{0}\\
&-E(\varnothing\rightarrow \Bbox_H)F^{(1)}(\Bbox_H\rightarrow \varnothing)\delta\left(\frac{z}{u}\right)\delta\left(\frac{w}{u}\right)\ket{0}
\eea
and
\bea
\frac{1}{[\sfq]}\delta\left(\frac{z}{w}\right)\left(K^{+}(z)-K^{-}(z)\right)\ket{0}&=\frac{1}{[\sfq]}\delta\left(\frac{z}{w}\right)\left([\tilde{\Psi}_{\varnothing,u}(z)]_{+}-[\tilde{\Psi}_{\varnothing,u}(z)]_{-}\right)\ket{0}\\
&=\frac{1}{[\sfq]}\delta\left(\frac{z}{w}\right)\left( A_{0}\tilde{\delta}^{(1)}\left(\frac{z}{u}\right)+A_{1}\delta\left(\frac{z}{u}\right)\right)\ket{0}
\eea
where
\bea
A_{0}=\underset{z=u}{\Res}z^{-1}\left(1-\frac{z}{u}\right)\tilde{\Psi}_{\varnothing,u}(z),\quad A_{1}=\underset{z=u}{\Res}z^{-1}\tilde{\Psi}_{\varnothing,u}(z).
\eea
Thus,
\bea
E(\varnothing\rightarrow \Bbox_L)F(\Bbox_L\rightarrow \varnothing)&=-\frac{A_{0}}{[\sfq]},\\
E(\varnothing\rightarrow \Bbox_H)F^{(2)}(\Bbox_H\rightarrow \varnothing)&=-\frac{A_{0}}{[\sfq]}
E(\varnothing\rightarrow \Bbox_H)F^{(1)}(\Bbox_H\rightarrow \varnothing)&=-\frac{A_{1}}{[\sfq]}.
\eea
Note that the constraint obtained from the $KF$ relation is automatically satisfied under these conditions.

\paragraph{Level one $\leftrightarrow $ Level two}We impose the ansatz
\bea
E(z)\ket{\Bbox_{L}}&=E(\Bbox_L\rightarrow \Bbox_U)\delta\left(\frac{z}{u}\right)\ket{\Bbox_U}\\
E(z)\ket{\Bbox_H}&=\left(E^{(2)}(\Bbox_H\rightarrow \Bbox_U)\tilde{\delta}^{(1)}\left(\frac{z}{u}\right)+E^{(1)}(\Bbox_H\rightarrow \Bbox_U)\delta\left(\frac{z}{u}\right)\right)\ket{\Bbox_U}.
\eea

Let us start from the $EE$-relation:
\bea
\sfq^{-1}\left(1-\sfq w/z\right)E(z)E(w)=(1-\sfq^{-1} w/z)E(w)E(z).
\eea
The left hand side on the vacuum gives 
\bea
&\sfq^{-1}\left(1-\sfq w/z\right)\left(E(\varnothing\rightarrow \Bbox_L)E(\Bbox_L\rightarrow \Bbox_U)\delta\left(\frac{z}{u}\right)\tilde{\delta}^{(1)}\left(\frac{w}{u}\right)\right.+E(\varnothing\rightarrow \Bbox_H)E^{(2)}(\Bbox_H\rightarrow \Bbox_U)\delta\left(\frac{w}{u}\right)\tilde{\delta}^{(1)}\left(\frac{z}{u}\right)\\
&\left.+E(\varnothing\rightarrow \Bbox_H)E^{(1)}(\Bbox_H\rightarrow \Bbox_U)\delta\left(\frac{z}{u}\right)\delta\left(\frac{w}{u}\right)\right)\ket{\Bbox_U}\\
=&\left(-(1-\sfq^{-1})E(\varnothing\rightarrow \Bbox_L)E(\Bbox_L\rightarrow \Bbox_U)\delta\left(\frac{z}{u}\right)\tilde{\delta}^{(1)}\left(\frac{w}{u}\right)\right.+E(\varnothing\rightarrow \Bbox_L)E(\Bbox_L\rightarrow \Bbox_U)\delta\left(\frac{z}{u}\right)\delta\left(\frac{w}{u}\right)\\
&-(1-\sfq^{-1})E(\varnothing\rightarrow \Bbox_H)E^{(2)}(\Bbox_H\rightarrow \Bbox_U)\delta\left(\frac{w}{u}\right)\tilde{\delta}^{(1)}\left(\frac{z}{u}\right)-E(\varnothing\rightarrow \Bbox_H)E^{(2)}(\Bbox_H\rightarrow \Bbox_U)\delta\left(\frac{w}{u}\right)\delta\left(\frac{z}{u}\right)\\
&\left.-(1-\sfq^{-1})E(\varnothing\rightarrow \Bbox_H)E^{(1)}(\Bbox_H\rightarrow \Bbox_U)\delta\left(\frac{w}{u}\right)\delta\left(\frac{z}{u}\right)\right)\ket{\Bbox_U}
\eea
and the right hand side gives
\bea
&\left(1-\sfq^{-1}w/z\right)\left(E(\varnothing\rightarrow \Bbox_L)E(\Bbox_L\rightarrow \Bbox_U)\delta\left(\frac{w}{u}\right)\tilde{\delta}^{(1)}\left(\frac{z}{u}\right)\right.+E(\varnothing\rightarrow \Bbox_H)E^{(2)}(\Bbox_H\rightarrow \Bbox_U)\delta\left(\frac{z}{u}\right)\tilde{\delta}^{(1)}\left(\frac{w}{u}\right)\\
&\left.+E(\varnothing\rightarrow \Bbox_H)E^{(1)}(\Bbox_H\rightarrow \Bbox_U)\delta\left(\frac{z}{u}\right)\delta\left(\frac{w}{u}\right)\right)\ket{\Bbox_U}\\
=&\left(  (1-\sfq^{-1}) E(\varnothing\rightarrow \Bbox_L)E(\Bbox_L\rightarrow \Bbox_U)\delta\left(\frac{w}{u}\right)\tilde{\delta}^{(1)}\left(\frac{z}{u}\right)-\sfq^{-1} E(\varnothing\rightarrow \Bbox_L)E(\Bbox_L\rightarrow \Bbox_U)\delta\left(\frac{u}{z}\right)\left(\frac{w}{u}\right) \right.\\
&+(1-\sfq^{-1}) E(\varnothing\rightarrow \Bbox_H)E^{(2)}(\Bbox_H\rightarrow \Bbox_U)\delta\left(\frac{z}{u}\right)\tilde{\delta}^{(1)}\left(\frac{w}{u}\right)+\sfq^{-1}E(\varnothing\rightarrow \Bbox_H)E^{(2)}(\Bbox_H\rightarrow \Bbox_U)\delta\left(\frac{z}{u}\right)\delta\left(\frac{w}{u}\right)\\
&\left.+(1-\sfq^{-1})E(\varnothing\rightarrow \Bbox_H)E^{(1)}(\Bbox_H\rightarrow \Bbox_U)\delta\left(\frac{z}{u}\right)\delta\left(\frac{w}{u}\right)\right)\ket{\Bbox_U}
\eea
where during the computation we used Prop.~\ref{prop:derivedeltafunct-rational}. Comparing both hand sides, we have
\bea
E(\varnothing\rightarrow \Bbox_H)E^{(2)}(\Bbox_H\rightarrow \Bbox_U)&=-E(\varnothing\rightarrow \Bbox_L)E(\Bbox_L\rightarrow \Bbox_U)\\
E^{(1)}(\Bbox_H\rightarrow \Bbox_U)&=-\frac{1+\sfq^{-1}}{1-\sfq^{-1}}E^{(2)}(\Bbox_H\rightarrow \Bbox_U).
\eea

Let us study the $KE$ relation on $\ket{\Bbox_L},\ket{\Bbox_H}$. The action on $\ket{\Bbox_L}$ is
\bea
K^{\pm}(z)E(w)\ket{\Bbox_L}&=E(\Bbox_L\rightarrow \Bbox_U)\delta\left(\frac{w}{u}\right)K^{\pm}(z)\ket{\Bbox_U},\\
G\left(\frac{w}{z}\right)^{-1}E(w)K^{\pm}(z)\ket{\Bbox_L}&=E(\Bbox_L\rightarrow \Bbox_U)\delta\left(\frac{w}{u}\right)\tilde{\Psi}_{\varnothing,u}(z)G\left(\frac{u}{z}\right)^{-2}\ket{\Bbox_U}
\eea
which gives
\bea\label{eq:A1quiver-ultra-Cartan}
K^{\pm}(z)\ket{\Bbox_U}=\left[\tilde{\Psi}_{\varnothing,u}(z)G\left(\frac{u}{z}\right)^{-2}\right]_{\pm}\ket{\Bbox_U}.
\eea
The action on $\ket{\Bbox_H}$ gives
\bea
K^{\pm}(z)E(w)\ket{\Bbox_H}&=\left(E^{(2)}(\Bbox_H\rightarrow \Bbox_U)\tilde{\delta}^{(1)}\left(\frac{w}{u}\right)+E^{(1)}(\Bbox_H\rightarrow \Bbox_U)\delta\left(\frac{w}{u}\right)\right)K^{\pm}(z)\ket{\Bbox_U},\\
G\left(\frac{w}{z}\right)^{-1}E(w)K^{\pm}(z)\ket{\Bbox_H}&=\left( E^{(2)}(\Bbox_H\rightarrow \Bbox_U) \tilde{\Psi}_{\varnothing,u}(z)G\left(\frac{u}{z}\right)^{-1}G\left(\frac{
w}{z}\right)^{-1}\tilde{\delta}^{(1)}\left(\frac{w}{u}\right)\right.\\
&+E^{(1)}(\Bbox_H\rightarrow \Bbox_U) \tilde{\Psi}_{\varnothing,u}(z)G\left(\frac{u}{z}\right)^{-2}  \delta\left(\frac{w}{u}\right)\\
&\left.-\frac{E(\varnothing\rightarrow \Bbox_L)E(\Bbox_L\rightarrow \Bbox_U)}{E(\varnothing\rightarrow \Bbox_H)} \tilde{\Psi}_{\varnothing,u}(z)G\left(\frac{u}{z}\right)^{-1}\tilde{\partial}G\left(\frac{u}{z}\right)^{-1} \delta\left(\frac{w}{u}\right)   \right)\ket{\Bbox_U}.
\eea
Using \eqref{eq:Gfunctdeltaderiv} and \eqref{eq:A1quiver-ultra-Cartan} gives the constraint
\bea
E^{(2)}(\Bbox_H\rightarrow \Bbox_U)\frac{E(\varnothing\rightarrow \Bbox_H)}{E(\varnothing\rightarrow \Bbox_L)}=-E(\Bbox_L\rightarrow \Bbox_U).
\eea

For the $F(z)$ action, we can impose the following ansatz
\bea
F(z)\ket{\Bbox_U}&=F(\Bbox_U\rightarrow \Bbox_H)\delta\left(\frac{z}{u}\right)\ket{\Bbox_H}\\
&+\left(F^{(2)}(\Bbox_U\rightarrow \Bbox_L)\tilde{\delta}^{(1)}\left(\frac{z}{u}\right)+F^{(1)}(\Bbox_U\rightarrow \Bbox_L)\delta\left(\frac{z}{u}\right)\right)\ket{\Bbox_L}.
\eea
Further constraints will be imposed after studying the $FF$ relation and $KF$ relations but we omit the discussion.

The action of the Drinfeld currents are summarized as
\begin{empheq}[box=\fbox]{align}
\begin{split}
E(z)\ket{0}&=E(\varnothing\rightarrow \Bbox_{L})\tilde{\delta}^{(1)}\left(\frac{z}{u}\right)\ket{\Bbox_{L}}+E(\varnothing\rightarrow \Bbox_{H})\delta\left(\frac{z}{u}\right)\ket{\Bbox_{H}},\\
E(z)\ket{\Bbox_{L}}&=E(\Bbox_L\rightarrow \Bbox_U)\delta\left(\frac{z}{u}\right)\ket{\Bbox_U},\\
E(z)\ket{\Bbox_H}&=\left(E^{(2)}(\Bbox_H\rightarrow \Bbox_U)\tilde{\delta}^{(1)}\left(\frac{z}{u}\right)+E^{(1)}(\Bbox_H\rightarrow \Bbox_U)\delta\left(\frac{z}{u}\right)\right)\ket{\Bbox_U},\\
E(z)\ket{\Bbox_U}&=0,\\
F(z)\ket{0}&=0,\\
F(z)\ket{\Bbox_{L}}&=F(\Bbox_{L}\rightarrow \varnothing)\delta\left(\frac{z}{u}\right)\ket{0},\\
F(z)\ket{\Bbox_{H}}&=\left(F^{(2)}(\Bbox_{H}\rightarrow \varnothing)\tilde{\delta}^{(1)}\left(\frac{z}{u}\right)+F^{(1)}(\Bbox_{H}\rightarrow \varnothing)\delta\left(\frac{z}{u}\right)\right)\ket{0},\\
F(z)\ket{\Bbox_U}&=F(\Bbox_U\rightarrow \Bbox_H)\delta\left(\frac{z}{u}\right)\ket{\Bbox_H}\\
&+\left(F^{(2)}(\Bbox_U\rightarrow \Bbox_L)\tilde{\delta}^{(1)}\left(\frac{z}{u}\right)+F^{(1)}(\Bbox_U\rightarrow \Bbox_L)\delta\left(\frac{z}{u}\right)\right)\ket{\Bbox_L},\\
K^{\pm}(z)\ket{0}&=\tilde{\Psi}_{\varnothing,u}(z)\ket{0},\\
K^{\pm}(z)\ket{\Bbox_{L}}&=\left[\widetilde{\Psi}_{\varnothing,u}(z)G\left(\frac{u}{z}\right)^{-1}\right]_{\pm}\ket{\Bbox_{L}},\\
K^{\pm}(z)\ket{\Bbox_{H}}&=\left[\widetilde{\Psi}_{\varnothing,u}(z)G\left(\frac{u}{z}\right)^{-1}\right]_{\pm}\ket{\Bbox_{H}}-\frac{E(\varnothing\rightarrow \Bbox_{L})}{E(\varnothing\rightarrow \Bbox_{H})}\left[\widetilde{\Psi}_{\varnothing,u}(z)u\partial_{u}G\left(\frac{u}{z}\right)^{-1}\right]_{\pm}\ket{\Bbox_{L}},\\
K^{\pm}(z)\ket{\Bbox_U}&=\tilde{\Psi}_{\varnothing,u}(z)G\left(\frac{u}{z}\right)^{-2}\ket{\Bbox_U}.
\end{split}
\end{empheq}

\subsubsection{Relation with the $qq$-character}
The contraction of $\mathsf{T}^{(2)}(x)$ and the screening current $\mathsf{S}_{2}(q_{2}x')$ give 
\bea
\bra{0}\mathsf{S}_{2}(q_{2}x')\mathsf{Y}(x)^{2}\ket{0}&\xrightarrow{q_{2}\rightarrow 1}\sfq^{-2}\left(\frac{1-\sfq x/x'}{1-x/x'}\right)^{2}=\sfq^{-1}\widetilde{\Psi}_{\varnothing,x}(x'),\\
\bra{0}\mathsf{S}_{2}(q_{2}x'):\mathsf{Y}(x)^{2}\mathsf{A}(x)^{-1}:\ket{0}&\xrightarrow{q_{2}\rightarrow 1}\frac{(1-\sfq^{-1}x/x')(1-\sfq x/x')}{(1-x/x')^{2}}=\widetilde{\Psi}_{\varnothing,x}(x')G\left(\frac{x}{x'}\right)^{-1},\\
\bra{0}\mathsf{S}_{2}(q_{2}x'):\mathsf{Y}(x)^{2}\partial_{\log x}\mathsf{A}(x)^{-1}\ket{0}&\xrightarrow{q_{2}\rightarrow1}\frac{(-1+\sfq^{2})xx'}{\sfq(x-x')^{2}}=\widetilde{\Psi}_{\varnothing,x}(x')x\partial_{x}G\left(\frac{x}{x'}\right)^{-1},\\
\bra{0}\mathsf{S}_{2}(q_{2}x'):\mathsf{Y}(x)^{2}\mathsf{A}(x)^{-2}:\ket{0}&\xrightarrow{q_{2}\rightarrow 1}\sfq^{2}\left(\frac{1-\sfq^{-1}x/x'}{1-x/x'}\right)^{2}=\sfq\widetilde{\Psi}_{\varnothing,x}(x')G\left(\frac{x}{x'}\right)^{-2}
\eea
Using
\bea
\mathfrak{c}(q_{1},q_{2})\xrightarrow{q_{2}\rightarrow 1} 2,\quad \frac{(1-q_{1})(1-q_{2})}{1-q_{12}}\xrightarrow{q_{2}\rightarrow 1}0
\eea
we obtain
\bea
\bra{0}\mathsf{S}_{2}(x')\mathsf{T}^{(2)}(x)\ket{0}\xrightarrow{q_{2}\rightarrow 1} \sfq^{-1}\left(\widetilde{\Psi}_{\varnothing,x}(x')+2\sfq\widetilde{\Psi}_{\varnothing,x}(x')G\left(\frac{x}{x'}\right)^{-1}+\sfq^{2}\widetilde{\Psi}_{\varnothing,x}(x')G\left(\frac{x}{x'}\right)^{-2}\right).
\eea
Namely, we obtain the $q$-character of this representation. The $q$-character structure is intuitive. The first term corresponds to the vacuum state. For the double state degeneracy coming from $\ket{\Bbox_{L,H}}$, the Drinfeld current $K^{\pm}(z)$ does not act diagonally but instead it forms a Jordan block. The eigenvalues of this two-state degeneracy comes from the diagonal components of it and it is the same $\widetilde{\Psi}_{\varnothing,x}(x')G\left(\frac{x}{x'}\right)^{-1}$. Because of this, there is a factor $2$ giving the multiplicity of this eigenspace. The last term is the contribution coming from the ultra-heavy box corresponding to place two boxes at the origin.

In other words, the $q$-character is understood as the trace over the underlying module:
\bea
\Tr\left(\sfq^{d} K^{\pm}(x')\right)
\eea
where $d$ is some degree counting operator, which we omit the discussion. One may generally set $\sfq^{d}$ to $p^{d}$, where $p$ is some generic parameter and it comes from the redefinition of the topological term in the $qq$-character.


We note that the $qq$-character contains much more information than the $q$-character because of the existence of the non-diagonal term. From the expressions \eqref{eq:A1qq-collision}, \eqref{eq:A1qq-collision-part}, \eqref{eq:A1qq-collision-part2}, we can interpret each monomial terms as
\bea
:\mathsf{Y}^{2}(x): &\longleftrightarrow K^{\pm}(z)\ket{0},\\
\frac{q_{1}+q_{2}-3q_{12}+q_{12}^{2}}{(1-q_{12})^{2}}:\frac{\mathsf{Y}^{2}(x)}{\mathsf{A}(x)}: &\longleftrightarrow K^{\pm}(z)\ket{\Bbox_{L}},\\
\frac{1-3q_{12}+q_{1}^{2}q_{2}+q_{1}q_{2}^{2}}{(1-q_{12})^{2}}:\frac{\mathsf{Y}^{2}(x)}{\mathsf{A}(x)}:+\frac{(1-q_{1})(1-q_{2})}{(1-q_{12})}:\mathsf{Y}(x)^{2}\partial_{\log x}\mathsf{A}^{-1}(x): &\longleftrightarrow  K^{\pm}(z)\ket{\Bbox_{H}},\\
:\frac{\mathsf{Y}^{2}(x)}{\mathsf{A}(x)^{2}}:&\longleftrightarrow K^{\pm}(z)\ket{\Bbox_{U}}.
\eea
We decomposed the term $\mathfrak{c}(q_{1,2})$ into the sum of two parts so that the correspondence between the light and heavy boxes can be seen explicitly. This identification is compatible with the one in~\eqref{eq:A1-partfunct-identify}.

\bibliographystyle{utphys}
\bibliography{Worigami}

\providecommand{\href}[2]{#2}\begingroup\raggedright\begin{thebibliography}{100}

\bibitem{Nekrasov:2015wsu}
N.~Nekrasov, ``{BPS/CFT correspondence: non-perturbative Dyson-Schwinger equations and qq-characters},'' \href{https://dx.doi.org/10.1007/JHEP03(2016)181}{{\em JHEP} {\bfseries 03} (2016) 181}, \href{https://arxiv.org/abs/1512.05388}{{\ttfamily arXiv:1512.05388 [hep-th]}}.

\bibitem{Nekrasov:2016qym}
N.~Nekrasov, ``{BPS/CFT correspondence II: Instantons at crossroads, moduli and compactness theorem},'' \href{https://dx.doi.org/10.4310/ATMP.2017.v21.n2.a4}{{\em Adv. Theor. Math. Phys.} {\bfseries 21} (2017) 503--583}, \href{https://arxiv.org/abs/1608.07272}{{\ttfamily arXiv:1608.07272 [hep-th]}}.

\bibitem{Nekrasov:2016ydq}
N.~Nekrasov, ``{BPS/CFT Correspondence III: Gauge Origami partition function and qq-characters},'' \href{https://dx.doi.org/10.1007/s00220-017-3057-9}{{\em Commun. Math. Phys.} {\bfseries 358} no.~3, (2018) 863--894}, \href{https://arxiv.org/abs/1701.00189}{{\ttfamily arXiv:1701.00189 [hep-th]}}.

\bibitem{Nekrasov:2017rqy}
N.~Nekrasov, ``{BPS/CFT correspondence IV: sigma models and defects in gauge theory},'' \href{https://dx.doi.org/10.1007/s11005-018-1115-7}{{\em Lett. Math. Phys.} {\bfseries 109} no.~3, (2019) 579--622}, \href{https://arxiv.org/abs/1711.11011}{{\ttfamily arXiv:1711.11011 [hep-th]}}.

\bibitem{Nekrasov:2017gzb}
N.~Nekrasov, ``{BPS/CFT correspondence V: BPZ and KZ equations from qq-characters},'' \href{https://arxiv.org/abs/1711.11582}{{\ttfamily arXiv:1711.11582 [hep-th]}}.

\bibitem{Jeffrey1993LocalizationFN}
L.~C. Jeffrey and F.~Kirwan, ``Localization for nonabelian group actions,'' \href{https://dx.doi.org/10.1016/0040-9383(94)00028-J}{{\em Topology} {\bfseries 34} (1993) 291--327}.

\bibitem{Szenes2003ToricRA}
A.~Szenes and M.~Vergne, ``{Toric reduction and a conjecture of Batyrev and Materov},'' \href{https://dx.doi.org/10.1007/s00222-004-0375-2}{{\em Invent. Math.} {\bfseries 158} (2003) 453--495}, \href{https://arxiv.org/abs/math/0306311}{{\ttfamily arXiv:math/0306311 [math.AT]}}.

\bibitem{Brion1999ArrangementOH}
M.~Brion and M.~Vergne, ``{Arrangement of hyperplanes. I. Rational functions and Jeffrey-Kirwan residue},'' \href{https://dx.doi.org/10.1016/S0012-9593(01)80005-7}{{\em Ann. Sci. Éc. Norm. Supér.} {\bfseries 32} (1999) 715--741}.

\bibitem{Benini:2013xpa}
F.~Benini, R.~Eager, K.~Hori, and Y.~Tachikawa, ``{Elliptic Genera of 2d ${\mathcal{N}}$ = 2 Gauge Theories},'' \href{https://dx.doi.org/10.1007/s00220-014-2210-y}{{\em Commun. Math. Phys.} {\bfseries 333} no.~3, (2015) 1241--1286}, \href{https://arxiv.org/abs/1308.4896}{{\ttfamily arXiv:1308.4896 [hep-th]}}.

\bibitem{Benini:2013nda}
F.~Benini, R.~Eager, K.~Hori, and Y.~Tachikawa, ``{Elliptic genera of two-dimensional N=2 gauge theories with rank-one gauge groups},'' \href{https://dx.doi.org/10.1007/s11005-013-0673-y}{{\em Lett. Math. Phys.} {\bfseries 104} (2014) 465--493}, \href{https://arxiv.org/abs/1305.0533}{{\ttfamily arXiv:1305.0533 [hep-th]}}.

\bibitem{Hori:2014tda}
K.~Hori, H.~Kim, and P.~Yi, ``{Witten Index and Wall Crossing},'' \href{https://dx.doi.org/10.1007/JHEP01(2015)124}{{\em JHEP} {\bfseries 01} (2015) 124}, \href{https://arxiv.org/abs/1407.2567}{{\ttfamily arXiv:1407.2567 [hep-th]}}.

\bibitem{Hwang:2014uwa}
C.~Hwang, J.~Kim, S.~Kim, and J.~Park, ``{General instanton counting and 5d SCFT},'' \href{https://dx.doi.org/10.1007/JHEP07(2015)063}{{\em JHEP} {\bfseries 07} (2015) 063}, \href{https://arxiv.org/abs/1406.6793}{{\ttfamily arXiv:1406.6793 [hep-th]}}. [Addendum: JHEP 04, 094 (2016)].

\bibitem{Cordova:2014oxa}
C.~Cordova and S.-H. Shao, ``{An Index Formula for Supersymmetric Quantum Mechanics},'' \href{https://dx.doi.org/10.5427/jsing.2016.15b}{{\em J. Singul.} {\bfseries 15} (2016) 14--35}, \href{https://arxiv.org/abs/1406.7853}{{\ttfamily arXiv:1406.7853 [hep-th]}}.

\bibitem{Nekrasov:2016gud}
N.~Nekrasov and N.~S. Prabhakar, ``{Spiked Instantons from Intersecting D-branes},'' \href{https://dx.doi.org/10.1016/j.nuclphysb.2016.11.014}{{\em Nucl. Phys. B} {\bfseries 914} (2017) 257--300}, \href{https://arxiv.org/abs/1611.03478}{{\ttfamily arXiv:1611.03478 [hep-th]}}.

\bibitem{Rapcak:2018nsl}
M.~Rap\v{c}\'ak, Y.~Soibelman, Y.~Yang, and G.~Zhao, ``{Cohomological Hall algebras, vertex algebras and instantons},'' \href{https://dx.doi.org/10.1007/s00220-019-03575-5}{{\em Commun. Math. Phys.} {\bfseries 376} no.~3, (2019) 1803--1873}, \href{https://arxiv.org/abs/1810.10402}{{\ttfamily arXiv:1810.10402 [math.QA]}}.

\bibitem{Rapcak:2020ueh}
M.~Rapcak, Y.~Soibelman, Y.~Yang, and G.~Zhao, ``{Cohomological Hall algebras and perverse coherent sheaves on toric Calabi{\textendash}Yau $3$-folds},'' \href{https://dx.doi.org/10.4310/CNTP.2023.v17.n4.a2}{{\em Commun. Num. Theor. Phys.} {\bfseries 17} no.~4, (2023) 847--939}, \href{https://arxiv.org/abs/2007.13365}{{\ttfamily arXiv:2007.13365 [math.QA]}}.

\bibitem{Pomoni:2021hkn}
E.~Pomoni, W.~Yan, and X.~Zhang, ``{Tetrahedron Instantons},'' \href{https://dx.doi.org/10.1007/s00220-022-04376-z}{{\em Commun. Math. Phys.} {\bfseries 393} no.~2, (2022) 781--838}, \href{https://arxiv.org/abs/2106.11611}{{\ttfamily arXiv:2106.11611 [hep-th]}}.

\bibitem{Pomoni:2023nlf}
E.~Pomoni, W.~Yan, and X.~Zhang, ``{Probing M-theory with tetrahedron instantons},'' \href{https://dx.doi.org/10.1007/JHEP11(2023)177}{{\em JHEP} {\bfseries 11} (2023) 177}, \href{https://arxiv.org/abs/2306.06005}{{\ttfamily arXiv:2306.06005 [hep-th]}}.

\bibitem{Fasola:2023ypx}
N.~Fasola and S.~Monavari, ``{Tetrahedron instantons in Donaldson-Thomas theory},'' \href{https://dx.doi.org/10.1016/j.aim.2024.110099}{{\em Adv. Math.} {\bfseries 462} (Feb., 2025) 110099}, \href{https://arxiv.org/abs/2306.07145}{{\ttfamily arXiv:2306.07145 [math.AG]}}.

\bibitem{Nekrasov:2017cih}
N.~Nekrasov, ``{Magnificent four},'' \href{https://dx.doi.org/10.4310/ATMP.2020.v24.n5.a4}{{\em Adv. Theor. Math. Phys.} {\bfseries 24} no.~5, (2020) 1171--1202}, \href{https://arxiv.org/abs/1712.08128}{{\ttfamily arXiv:1712.08128 [hep-th]}}.

\bibitem{Nekrasov:2018xsb}
N.~Nekrasov and N.~Piazzalunga, ``{Magnificent Four with Colors},'' \href{https://dx.doi.org/10.1007/s00220-019-03426-3}{{\em Commun. Math. Phys.} {\bfseries 372} no.~2, (2019) 573--597}, \href{https://arxiv.org/abs/1808.05206}{{\ttfamily arXiv:1808.05206 [hep-th]}}.

\bibitem{Billo:2021xzh}
M.~Bill\`o, M.~Frau, F.~Fucito, L.~Gallot, A.~Lerda, and J.~F. Morales, ``{On the D(\textendash{}1)/D7-brane systems},'' \href{https://dx.doi.org/10.1007/JHEP04(2021)096}{{\em JHEP} {\bfseries 04} (2021) 096}, \href{https://arxiv.org/abs/2101.01732}{{\ttfamily arXiv:2101.01732 [hep-th]}}.

\bibitem{Kool:2025qou}
M.~Kool and J.~V. Rennemo, ``{Proof of a magnificent conjecture},'' \href{https://arxiv.org/abs/2507.02852}{{\ttfamily arXiv:2507.02852 [math.AG]}}.

\bibitem{Bao:2025hfu}
J.~Bao and M.~Yamazaki, ``{Crystals and double quiver algebras from Jeffrey-Kirwan residues},'' \href{https://dx.doi.org/10.21468/SciPostPhys.18.4.143}{{\em SciPost Phys.} {\bfseries 18} no.~4, (2025) 143}, \href{https://arxiv.org/abs/2501.03365}{{\ttfamily arXiv:2501.03365 [hep-th]}}.

\bibitem{Bao:2024ygr}
J.~Bao, R.-K. Seong, and M.~Yamazaki, ``{The origin of Calabi-Yau crystals in BPS states counting},'' \href{https://dx.doi.org/10.1007/JHEP03(2024)140}{{\em JHEP} {\bfseries 03} (2024) 140}, \href{https://arxiv.org/abs/2401.02792}{{\ttfamily arXiv:2401.02792 [hep-th]}}.

\bibitem{Kimura:2023bxy}
T.~Kimura and G.~Noshita, ``{Gauge origami and quiver W-algebras},'' \href{https://dx.doi.org/10.1007/JHEP05(2024)208}{{\em JHEP} {\bfseries 05} (2024) 208}, \href{https://arxiv.org/abs/2310.08545}{{\ttfamily arXiv:2310.08545 [hep-th]}}.

\bibitem{Szabo:2023ixw}
R.~J. Szabo and M.~Tirelli, ``{Instanton counting and Donaldson{\textendash}Thomas theory on toric Calabi{\textendash}Yau four-orbifolds},'' \href{https://dx.doi.org/10.4310/ATMP.2023.v27.n6.a2}{{\em Adv. Theor. Math. Phys.} {\bfseries 27} no.~6, (2023) 1665--1757}, \href{https://arxiv.org/abs/2301.13069}{{\ttfamily arXiv:2301.13069 [hep-th]}}.

\bibitem{Bonelli:2020gku}
G.~Bonelli, N.~Fasola, A.~Tanzini, and Y.~Zenkevich, ``{ADHM in 8d, coloured solid partitions and Donaldson-Thomas invariants on orbifolds},'' \href{https://dx.doi.org/10.1016/j.geomphys.2023.104910}{{\em J. Geom. Phys.} {\bfseries 191} (2023) 104910}, \href{https://arxiv.org/abs/2011.02366}{{\ttfamily arXiv:2011.02366 [hep-th]}}.

\bibitem{Benini:2018hjy}
F.~Benini, G.~Bonelli, M.~Poggi, and A.~Tanzini, ``{Elliptic non-Abelian Donaldson-Thomas invariants of $\mathbb{C}^3$},'' \href{https://dx.doi.org/10.1007/JHEP07(2019)068}{{\em JHEP} {\bfseries 07} (2019) 068}, \href{https://arxiv.org/abs/1807.08482}{{\ttfamily arXiv:1807.08482 [hep-th]}}.

\bibitem{Cao:2014bca}
Y.~Cao and N.~C. Leung, ``{Donaldson-Thomas theory for Calabi-Yau 4-folds},'' \href{https://arxiv.org/abs/1407.7659}{{\ttfamily arXiv:1407.7659 [math.AG]}}.

\bibitem{Borisov:2017GT}
D.~Borisov and D.~Joyce, ``Virtual fundamental classes for moduli spaces of sheaves on calabi--yau four-folds,'' \href{https://dx.doi.org/10.2140/gt.2017.21.3231}{{\em Geom. Topol.} {\bfseries 21} no.~6, (2017) 3231--3311}, \href{https://arxiv.org/abs/1504.00690}{{\ttfamily arXiv:1504.00690 [math.AG]}}.

\bibitem{Cao:2017swr}
Y.~Cao and M.~Kool, ``{Zero-dimensional Donaldson\textendash{}Thomas invariants of Calabi\textendash{}Yau 4-folds},'' \href{https://dx.doi.org/10.1016/j.aim.2018.09.011}{{\em Adv. Math.} {\bfseries 338} (2018) 601--648}, \href{https://arxiv.org/abs/1712.07347}{{\ttfamily arXiv:1712.07347 [math.AG]}}.

\bibitem{Cao:2019tnw}
Y.~Cao and M.~Kool, ``{Curve counting and DT/PT correspondence for Calabi-Yau 4-folds},'' \href{https://dx.doi.org/10.1016/j.aim.2020.107371}{{\em Adv. Math.} {\bfseries 375} (2020) 107371}, \href{https://arxiv.org/abs/1903.12171}{{\ttfamily arXiv:1903.12171 [math.AG]}}.

\bibitem{Cao:2019tvv}
Y.~Cao, M.~Kool, and S.~Monavari, ``{K-Theoretic DT/PT Correspondence for Toric Calabi\textendash{}Yau 4-Folds},'' \href{https://dx.doi.org/10.1007/s00220-022-04472-0}{{\em Commun. Math. Phys.} {\bfseries 396} no.~1, (2022) 225--264}, \href{https://arxiv.org/abs/1906.07856}{{\ttfamily arXiv:1906.07856 [math.AG]}}.

\bibitem{Oh:2020rnj}
J.~Oh and R.~P. Thomas, ``{Counting sheaves on Calabi\textendash{}Yau 4-folds, I},'' \href{https://dx.doi.org/10.1215/00127094-2022-0059}{{\em Duke Math. J.} {\bfseries 172} no.~7, (2023) 1333--1409}, \href{https://arxiv.org/abs/2009.05542}{{\ttfamily arXiv:2009.05542 [math.AG]}}.

\bibitem{Monavari:2022rtf}
S.~Monavari, ``{Canonical vertex formalism in DT theory of toric Calabi-Yau 4-folds},'' \href{https://dx.doi.org/10.1016/j.geomphys.2022.104466}{{\em J. Geom. Phys.} {\bfseries 174} (2022) 104466}, \href{https://arxiv.org/abs/2203.11381}{{\ttfamily arXiv:2203.11381 [math.AG]}}.

\bibitem{Maulik:2003rzb}
D.~Maulik, N.~Nekrasov, A.~Okounkov, and R.~Pandharipande, ``{Gromov\textendash{}Witten theory and Donaldson\textendash{}Thomas theory, I},'' \href{https://dx.doi.org/10.1112/S0010437X06002302}{{\em Compos. Math.} {\bfseries 142} no.~05, (2006) 1263--1285}, \href{https://arxiv.org/abs/math/0312059}{{\ttfamily arXiv:math/0312059}}.

\bibitem{Maulik:2004txy}
D.~Maulik, N.~Nekrasov, A.~Okounkov, and R.~Pandharipande, ``{Gromov\textendash{}Witten theory and Donaldson\textendash{}Thomas theory, II},'' \href{https://dx.doi.org/10.1112/S0010437X06002314}{{\em Compos. Math.} {\bfseries 142} no.~05, (2006) 1286--1304}, \href{https://arxiv.org/abs/math/0406092}{{\ttfamily arXiv:math/0406092}}.

\bibitem{Okounkov:2003sp}
A.~Okounkov, N.~Reshetikhin, and C.~Vafa, ``{Quantum Calabi-Yau and classical crystals},'' \href{https://dx.doi.org/10.1007/0-8176-4467-9_16}{{\em Prog. Math.} {\bfseries 244} (2006) 597},
\href{https://arxiv.org/abs/hep-th/0309208}{{\ttfamily arXiv:hep-th/0309208 [hep-th]}}.

\bibitem{Ooguri:2009ijd}
H.~Ooguri and M.~Yamazaki, ``{Crystal Melting and Toric Calabi-Yau Manifolds},'' \href{https://dx.doi.org/10.1007/s00220-009-0836-y}{{\em Commun. Math. Phys.} {\bfseries 292} (2009) 179--199}, \href{https://arxiv.org/abs/0811.2801}{{\ttfamily arXiv:0811.2801 [hep-th]}}.

\bibitem{Nekrasov:2023nai}
N.~Nekrasov and N.~Piazzalunga, ``{Global Magni4icence, or: 4G Networks},'' \href{https://dx.doi.org/10.3842/SIGMA.2024.106}{{\em SIGMA} {\bfseries 20} (2024) 106}, \href{https://arxiv.org/abs/2306.12995}{{\ttfamily arXiv:2306.12995 [hep-th]}}.

\bibitem{Piazzalunga:2023qik}
N.~Piazzalunga, ``{The one-legged K-theoretic vertex of fourfolds from 3d gauge theory},'' \href{https://arxiv.org/abs/2306.12405}{{\ttfamily arXiv:2306.12405 [hep-th]}}.

\bibitem{DelZotto:2021gzy}
M.~Del~Zotto, N.~Nekrasov, N.~Piazzalunga, and M.~Zabzine, ``{Playing With the Index of M-Theory},'' \href{https://dx.doi.org/10.1007/s00220-022-04479-7}{{\em Commun. Math. Phys.} {\bfseries 396} no.~2, (2022) 817--865}, \href{https://arxiv.org/abs/2103.10271}{{\ttfamily arXiv:2103.10271 [hep-th]}}.

\bibitem{Iqbal:2007ii}
A.~Iqbal, C.~Kozcaz, and C.~Vafa, ``{The Refined topological vertex},'' \href{https://dx.doi.org/10.1088/1126-6708/2009/10/069}{{\em JHEP} {\bfseries 10} (2009) 069},
\href{https://arxiv.org/abs/hep-th/0701156}{{\ttfamily arXiv:hep-th/0701156 [hep-th]}}.

\bibitem{Aganagic:2003db}
M.~Aganagic, A.~Klemm, M.~Marino, and C.~Vafa, ``{The Topological vertex},'' \href{https://dx.doi.org/10.1007/s00220-004-1162-z}{{\em Commun. Math. Phys.} {\bfseries 254} (2005) 425--478},
\href{https://arxiv.org/abs/hep-th/0305132}{{\ttfamily arXiv:hep-th/0305132 [hep-th]}}.

\bibitem{Awata:2008ed}
H.~Awata and H.~Kanno, ``{Refined BPS state counting from Nekrasov's formula and Macdonald functions},'' \href{https://dx.doi.org/10.1142/S0217751X09043006}{{\em Int. J. Mod. Phys.} {\bfseries A24} (2009) 2253--2306},
\href{https://arxiv.org/abs/0805.0191}{{\ttfamily arXiv:0805.0191 [hep-th]}}.

\bibitem{Nekrasov:2014nea}
N.~Nekrasov and A.~Okounkov, ``{Membranes and sheaves.},'' \href{https://dx.doi.org/10.14231/AG-2016-015}{{\em Algebr. Geom.} {\bfseries 3} no.~3, (2016) 320--369}, \href{https://arxiv.org/abs/1404.2323}{{\ttfamily arXiv:1404.2323 [math.AG]}}.

\bibitem{Pandharipande:2007sq}
R.~Pandharipande and R.~P. Thomas, ``{The 3\textendash{}fold vertex via stable pairs},'' \href{https://dx.doi.org/10.2140/gt.2009.13.1835}{{\em Geom. Topol.} {\bfseries 13} no.~4, (2009) 1835--1876}, \href{https://arxiv.org/abs/0709.3823}{{\ttfamily arXiv:0709.3823 [math.AG]}}.

\bibitem{Pandharipande:2007kc}
R.~Pandharipande and R.~P. Thomas, ``{Curve counting via stable pairs in the derived category},'' \href{https://dx.doi.org/10.1007/s00222-009-0203-9}{{\em Invent. Math.} {\bfseries 178} (2009) 407--447}, \href{https://arxiv.org/abs/0707.2348}{{\ttfamily arXiv:0707.2348 [math.AG]}}.

\bibitem{Toda:2008ASPM}
Y.~Toda, ``Generating functions of stable pair invariants via wall-crossings in derived categories,'' \href{https://dx.doi.org/10.2969/aspm/05910389}{{\em Adv. Stud. Pure Math.} 389--434}, \href{https://arxiv.org/abs/0806.0062}{{\ttfamily arXiv:0806.0062 [math.AG]}}.

\bibitem{Toda:2010JAMS}
Y.~Toda, ``{Curve counting theories via stable objects I. DT/PT correspondence},'' \href{https://dx.doi.org/10.1090/s0894-0347-10-00670-3}{{\em J. Amer. Math. Soc.} {\bfseries 23} no.~4, (2010) 1119–1157}, \href{https://arxiv.org/abs/0902.4371}{{\ttfamily arXiv:0902.4371 [math.AG]}}.

\bibitem{Bridgeland:2011JAMS}
T.~Bridgeland, ``Hall algebras and curve-counting invariants,'' \href{https://dx.doi.org/10.1090/s0894-0347-2011-00701-7}{{\em J. Amer. Math. Soc.} {\bfseries 24} no.~4, (2011) 969–998}, \href{https://arxiv.org/abs/1002.4374}{{\ttfamily arXiv:1002.4374 [math.AG]}}.

\bibitem{Stoppa:2011BSMF}
J.~Stoppa and R.~P. Thomas, ``{Hilbert schemes and stable pairs: GIT and derived category wall crossings},'' \href{https://dx.doi.org/10.24033/bsmf.2610}{{\em Bull. Soc. Math. Fr.} {\bfseries 139} no.~3, (2011) 297–339}, \href{https://arxiv.org/abs/0903.1444}{{\ttfamily arXiv:0903.1444 [math.AG]}}.

\bibitem{Toda2016HallAI}
Y.~Toda, ``{Hall algebras in the derived category and higher-rank DT invariants},'' \href{https://dx.doi.org/10.14231/ag-2020-008}{{\em Algebraic Geometry} (2020) 240–262}, \href{https://arxiv.org/abs/1601.07519}{{\ttfamily arXiv:1601.07519 [math.AG]}}.

\bibitem{Kononov:2019fni}
Y.~Kononov, A.~Okounkov, and A.~Osinenko, ``{The 2-Leg Vertex in K-theoretic DT Theory},'' \href{https://dx.doi.org/10.1007/s00220-021-03936-z}{{\em Commun. Math. Phys.} {\bfseries 382} no.~3, (2021) 1579--1599}, \href{https://arxiv.org/abs/1905.01523}{{\ttfamily arXiv:1905.01523 [math-ph]}}.

\bibitem{Jenne2020TheCP}
H.~Jenne and G.~Webb, ``{The Combinatorial PT-DT Correspondence},'' {\em Sémin. Lothar. Comb.} {\bfseries 85B} (2021) 89, \href{https://arxiv.org/abs/2012.08484}{{\ttfamily arXiv:2012.08484 [math.CO]}}. \url{https://hdl.handle.net/1794/26871}.

\bibitem{Jenne:2021irh}
H.~Jenne, G.~Webb, and B.~Young, ``{Double-dimer condensation and the PT-DT correspondence},'' \href{https://arxiv.org/abs/2109.11773}{{\ttfamily arXiv:2109.11773 [math.CO]}}.

\bibitem{Kuhn:2023koa}
N.~Kuhn, H.~Liu, and F.~Thimm, ``{The 3-fold K-theoretic DT/PT vertex correspondence holds},'' \href{https://arxiv.org/abs/2311.15697}{{\ttfamily arXiv:2311.15697 [math.AG]}}.

\bibitem{Nekrasov:2012xe}
N.~Nekrasov and V.~Pestun, ``{Seiberg-Witten Geometry of Four-Dimensional $\mathcal N=2$ Quiver Gauge Theories},'' \href{https://dx.doi.org/10.3842/SIGMA.2023.047}{{\em SIGMA} {\bfseries 19} (2023) 047}, \href{https://arxiv.org/abs/1211.2240}{{\ttfamily arXiv:1211.2240 [hep-th]}}.

\bibitem{Nekrasov:2013xda}
N.~Nekrasov, V.~Pestun, and S.~Shatashvili, ``{Quantum geometry and quiver gauge theories},'' \href{https://dx.doi.org/10.1007/s00220-017-3071-y}{{\em Commun. Math. Phys.} {\bfseries 357} no.~2, (2018) 519--567}, \href{https://arxiv.org/abs/1312.6689}{{\ttfamily arXiv:1312.6689 [hep-th]}}.

\bibitem{Kim:2016qqs}
H.-C. Kim, ``{Line defects and 5d instanton partition functions},'' \href{https://dx.doi.org/10.1007/JHEP03(2016)199}{{\em JHEP} {\bfseries 03} (2016) 199}, \href{https://arxiv.org/abs/1601.06841}{{\ttfamily arXiv:1601.06841 [hep-th]}}.

\bibitem{Kanno:2012hk}
S.~Kanno, Y.~Matsuo, and H.~Zhang, ``{Virasoro constraint for Nekrasov instanton partition function},'' \href{https://dx.doi.org/10.1007/JHEP10(2012)097}{{\em JHEP} {\bfseries 10} (2012) 097}, \href{https://arxiv.org/abs/1207.5658}{{\ttfamily arXiv:1207.5658 [hep-th]}}.

\bibitem{Kanno:2013aha}
S.~Kanno, Y.~Matsuo, and H.~Zhang, ``{Extended Conformal Symmetry and Recursion Formulae for Nekrasov Partition Function},'' \href{https://dx.doi.org/10.1007/JHEP08(2013)028}{{\em JHEP} {\bfseries 08} (2013) 028}, \href{https://arxiv.org/abs/1306.1523}{{\ttfamily arXiv:1306.1523 [hep-th]}}.

\bibitem{Bourgine:2015szm}
J.-E. Bourgine, Y.~Matsuo, and H.~Zhang, ``{Holomorphic field realization of SH$^{c}$ and quantum geometry of quiver gauge theories},'' \href{https://dx.doi.org/10.1007/JHEP04(2016)167}{{\em JHEP} {\bfseries 04} (2016) 167}, \href{https://arxiv.org/abs/1512.02492}{{\ttfamily arXiv:1512.02492 [hep-th]}}.

\bibitem{Bourgine:2016vsq}
J.-E. Bourgine, M.~Fukuda, Y.~Matsuo, H.~Zhang, and R.-D. Zhu, ``{Coherent states in quantum $\mathcal{W}_{1+\infty}$ algebra and qq-character for 5d Super Yang-Mills},'' \href{https://dx.doi.org/10.1093/ptep/ptw165}{{\em PTEP} {\bfseries 2016} no.~12, (2016) 123B05}, \href{https://arxiv.org/abs/1606.08020}{{\ttfamily arXiv:1606.08020 [hep-th]}}.

\bibitem{Kimura:2015rgi}
T.~Kimura and V.~Pestun, ``{Quiver W-algebras},'' \href{https://dx.doi.org/10.1007/s11005-018-1072-1}{{\em Lett. Math. Phys.} {\bfseries 108} no.~6, (2018) 1351--1381}, \href{https://arxiv.org/abs/1512.08533}{{\ttfamily arXiv:1512.08533 [hep-th]}}.

\bibitem{Alday:2009aq}
L.~F. Alday, D.~Gaiotto, and Y.~Tachikawa, ``{Liouville Correlation Functions from Four-dimensional Gauge Theories},'' \href{https://dx.doi.org/10.1007/s11005-010-0369-5}{{\em Lett. Math. Phys.} {\bfseries 91} (2010) 167--197}, \href{https://arxiv.org/abs/0906.3219}{{\ttfamily arXiv:0906.3219 [hep-th]}}.

\bibitem{Awata:2009ur}
H.~Awata and Y.~Yamada, ``{Five-dimensional AGT Conjecture and the Deformed Virasoro Algebra},'' \href{https://dx.doi.org/10.1007/JHEP01(2010)125}{{\em JHEP} {\bfseries 01} (2010) 125}, \href{https://arxiv.org/abs/0910.4431}{{\ttfamily arXiv:0910.4431 [hep-th]}}.

\bibitem{Awata:2010yy}
H.~Awata and Y.~Yamada, ``{Five-dimensional AGT Relation and the Deformed beta-ensemble},'' \href{https://dx.doi.org/10.1143/PTP.124.227}{{\em Prog. Theor. Phys.} {\bfseries 124} (2010) 227--262}, \href{https://arxiv.org/abs/1004.5122}{{\ttfamily arXiv:1004.5122 [hep-th]}}.

\bibitem{Kimura:2024osv}
T.~Kimura and G.~Noshita, ``{Gauge origami and quiver W-algebras III: Donaldson--Thomas $qq$-characters},'' \href{https://dx.doi.org/10.1007/JHEP03(2025)050}{{\em JHEP} {\bfseries 03} (2025) 050}, \href{https://arxiv.org/abs/2411.01987}{{\ttfamily arXiv:2411.01987 [hep-th]}}.

\bibitem{Ginzburg}
V.~Ginzburg, M.~Kapranov, and E.~Vasserot, ``Langlands reciprocity for algebraic surfaces,'' \href{https://dx.doi.org/10.4310/MRL.1995.v2.n2.a4}{{\em Math. Res. Lett.} {\bfseries 2} (1995) 147--160}, \href{https://arxiv.org/abs/q-alg/9502013}{{\ttfamily arXiv:q-alg/9502013 [math.QA]}}.

\bibitem{ding1997generalization}
J.~Ding and K.~Iohara, ``{Generalization of Drinfeld quantum affine algebras},'' \href{https://dx.doi.org/10.1023/A:1007341410987}{{\em Lett. Math. Phys.} {\bfseries 41} no.~2, (1997) 181--193}, \href{https://arxiv.org/abs/q-alg/9608002}{{\ttfamily q-alg/9608002}}.

\bibitem{miki2007q}
K.~Miki, ``{A ($q, \gamma$) analog of the $W_{1+\infty}$ algebra},'' \href{https://dx.doi.org/10.1063/1.2823979}{{\em J. Math. Phys} {\bfseries 48} no.~12, (2007) 3520}.

\bibitem{FFJMM1}
B.~Feigin, M.~Jimbo, T.~Miwa, and E.~Mukhin, ``{Quantum continuous $\mathfrak{gl}_\infty$: Semi-infinite construction of representations},'' \href{https://dx.doi.org/10.1215/21562261-1214375}{{\em Kyoto J. Math.} {\bfseries 51} no.~2, (2011) 337--364},
\href{https://arxiv.org/abs/1002.3100}{{\ttfamily arXiv:1002.3100 [math.QA]}}.

\bibitem{Feigin2011}
B.~Feigin, E.~Feigin, M.~Jimbo, T.~Miwa, and E.~Mukhin, ``{Quantum continuous $gl(\infty)$ : Tensor products of Fock modules and $W_n$ characters},'' \href{https://dx.doi.org/10.1215/21562261-1214384}{{\em Kyoto J. Math.} {\bfseries 51} no.~2, (2011) 365--392}, \href{https://arxiv.org/abs/1002.3113}{{\ttfamily arXiv:1002.3113 [math.QA]}}.

\bibitem{Feigin2011plane}
B.~Feigin, M.~Jimbo, T.~Miwa, and E.~Mukhin, ``Quantum toroidal $gl_1$-algebra: Plane partitions,'' \href{https://dx.doi.org/10.1215/21562261-1625217}{{\em Kyoto J. Math.} {\bfseries 52} no.~3, (2012) 621--659}, \href{https://arxiv.org/abs/1110.5310}{{\ttfamily arXiv:1110.5310}}.

\bibitem{Gaiotto:2020dsq}
D.~Gaiotto and M.~Rap\v{c}\'ak, ``{Miura operators, degenerate fields and the M2-M5 intersection},'' \href{https://dx.doi.org/10.1007/JHEP01(2022)086}{{\em JHEP} {\bfseries 01} (2022) 086}, \href{https://arxiv.org/abs/2012.04118}{{\ttfamily arXiv:2012.04118 [hep-th]}}.

\bibitem{Kimura:2024xpr}
T.~Kimura and G.~Noshita, ``{Gauge origami and quiver W-algebras II: Vertex function and beyond quantum $q$-Langlands correspondence},'' \href{https://arxiv.org/abs/2404.17061}{{\ttfamily arXiv:2404.17061 [hep-th]}}.

\bibitem{Ohta:2014ria}
K.~Ohta and Y.~Sasai, ``{Exact Results in Quiver Quantum Mechanics and BPS Bound State Counting},'' \href{https://dx.doi.org/10.1007/JHEP11(2014)123}{{\em JHEP} {\bfseries 11} (2014) 123}, \href{https://arxiv.org/abs/1408.0582}{{\ttfamily arXiv:1408.0582 [hep-th]}}.

\bibitem{Nawata:2023aoq}
S.~Nawata, Y.~Pan, and J.~Zheng, ``{Class $\mathcal{S}$ theories on $S^2$},'' \href{https://dx.doi.org/10.1103/PhysRevD.109.105015}{{\em Phys. Rev. D} {\bfseries 109} no.~10, (2024) 105015}, \href{https://arxiv.org/abs/2310.07965}{{\ttfamily arXiv:2310.07965 [hep-th]}}.

\bibitem{Kim:2024vci}
S.-S. Kim, X.~Li, S.~Nawata, and F.~Yagi, ``{Freezing and BPS jumping},'' \href{https://dx.doi.org/10.1007/JHEP05(2024)340}{{\em JHEP} {\bfseries 05} (2024) 340}, \href{https://arxiv.org/abs/2403.12525}{{\ttfamily arXiv:2403.12525 [hep-th]}}.

\bibitem{Witten:2000mf}
E.~Witten, ``{BPS Bound states of D0-D6 and D0-D8 systems in a B field},'' \href{https://dx.doi.org/10.1088/1126-6708/2002/04/012}{{\em JHEP} {\bfseries 04} (2002) 012}, \href{https://arxiv.org/abs/hep-th/0012054}{{\ttfamily arXiv:hep-th/0012054}}.

\bibitem{Awata:2009dd}
H.~Awata and H.~Kanno, ``{Quiver Matrix Model and Topological Partition Function in Six Dimensions},'' \href{https://dx.doi.org/10.1088/1126-6708/2009/07/076}{{\em JHEP} {\bfseries 07} (2009) 076}, \href{https://arxiv.org/abs/0905.0184}{{\ttfamily arXiv:0905.0184 [hep-th]}}.

\bibitem{Galakhov:2021xum}
D.~Galakhov, W.~Li, and M.~Yamazaki, ``{Shifted quiver Yangians and representations from BPS crystals},'' \href{https://dx.doi.org/10.1007/JHEP08(2021)146}{{\em JHEP} {\bfseries 08} (2021) 146}, \href{https://arxiv.org/abs/2106.01230}{{\ttfamily arXiv:2106.01230 [hep-th]}}.

\bibitem{Nekrasov:2009JJM}
N.~Nekrasov, ``{Instanton partition functions and M-theory},'' \href{https://dx.doi.org/10.1007/s11537-009-0853-9}{{\em Japanese J. Math.} {\bfseries 4} no.~1, (2009) 63--93}.

\bibitem{Vuletic2007AGO}
M.~Vuleti\'c, ``{A generalization of MacMahon's formula},'' \href{https://dx.doi.org/10.1090/S0002-9947-08-04753-3}{{\em Trans. Amer. Math. Soc.} {\bfseries 361} (2007) 2789--2804}, \href{https://arxiv.org/abs/0707.0532}{{\ttfamily arXiv:0707.0532 [math.CO]}}.

\bibitem{Cai2015TheVO}
L.-Q. Cai, L.~Wang, K.~Wu, and J.~Yang, ``The vertex operator for a generalization of macmahon’s formula,'' \href{https://dx.doi.org/10.1142/S0217751X15501766}{{\em Int. J. Mod. Phys. A} {\bfseries 30} (2015) 1550176}.

\bibitem{Foda:2017tnv}
O.~Foda and J.-F. Wu, ``{A Macdonald refined topological vertex},'' \href{https://dx.doi.org/10.1088/1751-8121/aa7605}{{\em J. Phys. A} {\bfseries 50} no.~29, (2017) 294003}, \href{https://arxiv.org/abs/1701.08541}{{\ttfamily arXiv:1701.08541 [hep-th]}}.

\bibitem{Foda:2018jwz}
O.~Foda and M.~Manabe, ``{Macdonald topological vertices and brane condensates},'' \href{https://dx.doi.org/10.1016/j.nuclphysb.2018.10.001}{{\em Nucl. Phys. B} {\bfseries 936} (2018) 448--471}, \href{https://arxiv.org/abs/1801.04943}{{\ttfamily arXiv:1801.04943 [hep-th]}}.

\bibitem{Hayashi:2017jze}
H.~Hayashi and K.~Ohmori, ``{5d/6d DE instantons from trivalent gluing of web diagrams},'' \href{https://dx.doi.org/10.1007/JHEP06(2017)078}{{\em JHEP} {\bfseries 06} (2017) 078}, \href{https://arxiv.org/abs/1702.07263}{{\ttfamily arXiv:1702.07263 [hep-th]}}.

\bibitem{Kool:2016opl}
M.~Kool and A.~Gholampour, ``{Rank 2 wall-crossing and the Serre correspondence},'' \href{https://dx.doi.org/10.1007/s00029-016-0293-3}{{\em Selecta Math.} {\bfseries 23} (2017) 1599--1617}, \href{https://arxiv.org/abs/1602.03113}{{\ttfamily arXiv:1602.03113 [math.AG]}}.

\bibitem{Gholampour:2017lkc}
A.~Gholampour and M.~Kool, ``{Higher rank sheaves on threefolds and functional equations},'' \href{https://dx.doi.org/10.46298/epiga.2019.volume3.4375}{{\em EPIGA} {\bfseries 3} (2019) 5931}, \href{https://arxiv.org/abs/1706.05246}{{\ttfamily arXiv:1706.05246 [math.AG]}}.

\bibitem{Jardim2025HigherRD}
M.~Jardim, J.~Lo, A.~Maciocia, and C.~Martinez, ``{Higher rank DT/PT wall-crossing in Bridgeland stability},'' \href{https://arxiv.org/abs/2503.20008}{{\ttfamily arXiv:2503.20008 [math.AG]}}.

\bibitem{Kimura:2019msw}
T.~Kimura and V.~Pestun, ``{Super instanton counting and localization},'' \href{https://arxiv.org/abs/1905.01513}{{\ttfamily arXiv:1905.01513 [hep-th]}}.

\bibitem{Noshita:2022dxv}
G.~Noshita, ``{5d AGT correspondence of supergroup gauge theories from quantum toroidal $ \mathfrak{gl} _{1}$},'' \href{https://dx.doi.org/10.1007/JHEP12(2022)157}{{\em JHEP} {\bfseries 12} (2022) 157}, \href{https://arxiv.org/abs/2209.08313}{{\ttfamily arXiv:2209.08313 [hep-th]}}.

\bibitem{Kimura:2016dys}
T.~Kimura and V.~Pestun, ``{Quiver elliptic W-algebras},'' \href{https://dx.doi.org/10.1007/s11005-018-1073-0}{{\em Lett. Math. Phys.} {\bfseries 108} no.~6, (2018) 1383--1405}, \href{https://arxiv.org/abs/1608.04651}{{\ttfamily arXiv:1608.04651 [hep-th]}}.

\bibitem{Kimura:2017hez}
T.~Kimura and V.~Pestun, ``{Fractional quiver W-algebras},'' \href{https://dx.doi.org/10.1007/s11005-018-1087-7}{{\em Lett. Math. Phys.} {\bfseries 108} no.~11, (2018) 2425--2451}, \href{https://arxiv.org/abs/1705.04410}{{\ttfamily arXiv:1705.04410 [hep-th]}}.

\bibitem{Shiraishi:1995rp}
J.~Shiraishi, H.~Kubo, H.~Awata, and S.~Odake, ``{A Quantum deformation of the Virasoro algebra and the Macdonald symmetric functions},'' \href{https://dx.doi.org/10.1007/BF00398297}{{\em Lett. Math. Phys.} {\bfseries 38} (1996) 33--51},
\href{https://arxiv.org/abs/q-alg/9507034}{{\ttfamily arXiv:q-alg/9507034 [q-alg]}}.

\bibitem{Awata:1996dx}
H.~Awata, H.~Kubo, S.~Odake, and J.~Shiraishi, ``{Quantum deformation of the W(N) algebra},'' in {\em {Extended and Quantum Algebras and their Applications to Physics Tianjin, China, August 19-24, 1996}}.
\newblock 1996.
\newblock
\href{https://arxiv.org/abs/q-alg/9612001}{{\ttfamily arXiv:q-alg/9612001 [q-alg]}}.
\newblock

\bibitem{Awata:1995zk}
H.~Awata, H.~Kubo, S.~Odake, and J.~Shiraishi, ``{Quantum W(N) algebras and Macdonald polynomials},'' \href{https://dx.doi.org/10.1007/BF02102595}{{\em Commun. Math. Phys.} {\bfseries 179} (1996) 401--416}, \href{https://arxiv.org/abs/q-alg/9508011}{{\ttfamily arXiv:q-alg/9508011}}.

\bibitem{Frenkel:1998ojj}
E.~Frenkel and N.~Reshetikhin, \href{https://dx.doi.org/10.1090/conm/248/03823}{``{The {$q$}-characters of representations of quantum affine algebras and deformations of {$\mathcal{W}$}-algebras},''} in {\em {Recent Developments in Quantum Affine Algebras and Related Topics}}, vol.~248 of {\em Contemp. Math.}, pp.~163--205.
\newblock Amer. Math. Soc., 1999.
\newblock \href{https://arxiv.org/abs/math/9810055}{{\ttfamily math/9810055 [math.QA]}}.

\bibitem{Frenkel:1997CMP}
E.~Frenkel and N.~Reshetikhin, ``{Deformations of {$\mathcal{W}$}-algebras associated to simple {L}ie algebras},'' {\em Comm. Math. Phys.} {\bfseries 197} (1998) 1--32, \href{https://arxiv.org/abs/q-alg/9708006}{{\ttfamily q-alg/9708006 [math.QA]}}.

\bibitem{Kimura:2019hnw}
T.~Kimura, ``{Integrating over quiver variety and BPS/CFT correspondence},'' \href{https://dx.doi.org/10.1007/s11005-020-01261-5}{{\em Lett. Math. Phys.} {\bfseries 110} no.~6, (2020) 1237--1255}, \href{https://arxiv.org/abs/1910.03247}{{\ttfamily arXiv:1910.03247 [hep-th]}}.

\bibitem{Kimura:2022zsm}
T.~Kimura, ``{Double Quiver Gauge Theory and BPS/CFT Correspondence},'' \href{https://dx.doi.org/10.3842/SIGMA.2023.039}{{\em SIGMA} {\bfseries 19} (2023) 039}, \href{https://arxiv.org/abs/2212.03870}{{\ttfamily arXiv:2212.03870 [hep-th]}}.

\bibitem{DIMreview}
Y.~Matsuo, S.~Nawata, G.~Noshita, and R.-D. Zhu, ``{Quantum toroidal algebras and solvable structures in gauge/string theory},'' \href{https://dx.doi.org/10.1016/j.physrep.2023.12.003}{{\em Phys. Rept.} {\bfseries 1055} (2024) 1--144}, \href{https://arxiv.org/abs/2309.07596}{{\ttfamily arXiv:2309.07596 [hep-th]}}.

\bibitem{Gaberdiel:2017hcn}
M.~R. Gaberdiel, W.~Li, C.~Peng, and H.~Zhang, ``{The supersymmetric affine Yangian},'' \href{https://dx.doi.org/10.1007/JHEP05(2018)200}{{\em JHEP} {\bfseries 05} (2018) 200}, \href{https://arxiv.org/abs/1711.07449}{{\ttfamily arXiv:1711.07449 [hep-th]}}.

\bibitem{Gaberdiel:2018nbs}
M.~R. Gaberdiel, W.~Li, and C.~Peng, ``{Twin-plane-partitions and $\mathcal{N}=2$ affine Yangian},'' \href{https://dx.doi.org/10.1007/JHEP11(2018)192}{{\em JHEP} {\bfseries 11} (2018) 192}, \href{https://arxiv.org/abs/1807.11304}{{\ttfamily arXiv:1807.11304 [hep-th]}}.

\bibitem{Hernandez:2020xih}
D.~Hernandez, ``{Representations of Shifted Quantum Affine Algebras},'' \href{https://dx.doi.org/10.1093/imrn/rnac149}{{\em Int. Math. Res. Not.} {\bfseries 2023} no.~13, (2023) 11035--11126}, \href{https://arxiv.org/abs/2010.06996}{{\ttfamily arXiv:2010.06996 [math.RT]}}.

\bibitem{Noshita:2021dgj}
G.~Noshita and A.~Watanabe, ``{Shifted quiver quantum toroidal algebra and subcrystal representations},'' \href{https://dx.doi.org/10.1007/JHEP05(2022)122}{{\em JHEP} {\bfseries 05} (2022) 122}, \href{https://arxiv.org/abs/2109.02045}{{\ttfamily arXiv:2109.02045 [hep-th]}}.

\bibitem{Li:2020rij}
W.~Li and M.~Yamazaki, ``{Quiver Yangian from Crystal Melting},'' \href{https://dx.doi.org/10.1007/JHEP11(2020)035}{{\em JHEP} {\bfseries 11} (2020) 035}, \href{https://arxiv.org/abs/2003.08909}{{\ttfamily arXiv:2003.08909 [hep-th]}}.

\bibitem{Galakhov:2021vbo}
D.~Galakhov, W.~Li, and M.~Yamazaki, ``{Toroidal and elliptic quiver BPS algebras and beyond},'' \href{https://dx.doi.org/10.1007/JHEP02(2022)024}{{\em JHEP} {\bfseries 02} (2022) 024}, \href{https://arxiv.org/abs/2108.10286}{{\ttfamily arXiv:2108.10286 [hep-th]}}.

\bibitem{Knight:1995JA}
H.~Knight, ``{Spectra of Tensor Products of Finite Dimensional Representations of Yangians},'' \href{https://dx.doi.org/10.1006/jabr.1995.1123}{{\em J. Algebra} {\bfseries 174} (1995) 187--196}.

\bibitem{Crew:2020psc}
S.~Crew, N.~Dorey, and D.~Zhang, ``{Blocks and Vortices in the 3d ADHM Quiver Gauge Theory},'' \href{https://dx.doi.org/10.1007/JHEP03(2021)234}{{\em JHEP} {\bfseries 03} (2021) 234}, \href{https://arxiv.org/abs/2010.09732}{{\ttfamily arXiv:2010.09732 [hep-th]}}.

\bibitem{Bae:2022pif}
Y.~Bae, M.~Kool, and H.~Park, ``{Counting surfaces on Calabi-Yau 4-folds I: Foundations},'' \href{https://arxiv.org/abs/2208.09474}{{\ttfamily arXiv:2208.09474 [math.AG]}}.

\bibitem{Bae:2024bpx}
Y.~Bae, M.~Kool, and H.~Park, ``{Counting surfaces on Calabi-Yau 4-folds II: $\mathrm{DT}$-$\mathrm{PT}_0$ correspondence},'' \href{https://arxiv.org/abs/2402.06526}{{\ttfamily arXiv:2402.06526 [math.AG]}}.

\bibitem{Kimura-Noshita-PT4}
T.~Kimura and G.~Noshita, ``{The 4-fold Pandhariphande--Thomas vertex and Jeffrey--Kirwan residue},'' {\em In preparation.} (2025) .

\bibitem{Awata:2005fa}
H.~Awata and H.~Kanno, ``{Instanton counting, Macdonald functions and the moduli space of D-branes},'' \href{https://dx.doi.org/10.1088/1126-6708/2005/05/039}{{\em JHEP} {\bfseries 05} (2005) 039}, \href{https://arxiv.org/abs/hep-th/0502061}{{\ttfamily arXiv:hep-th/0502061}}.

\bibitem{Awata:2011ce}
H.~Awata, B.~Feigin, and J.~Shiraishi, ``{Quantum Algebraic Approach to Refined Topological Vertex},'' \href{https://dx.doi.org/10.1007/JHEP03(2012)041}{{\em JHEP} {\bfseries 03} (2012) 041}, \href{https://arxiv.org/abs/1112.6074}{{\ttfamily arXiv:1112.6074 [hep-th]}}.

\bibitem{Kimura:2020jxl}
T.~Kimura, \href{https://dx.doi.org/10.1007/978-3-030-76190-5}{{\em {Instanton Counting, Quantum Geometry and Algebra}}}.
\newblock Mathematical Physics Studies. Springer, 2021.
\newblock \href{https://arxiv.org/abs/2012.11711}{{\ttfamily arXiv:2012.11711 [hep-th]}}.

\bibitem{Drinfeld:1986in}
V.~G. Drinfeld, ``{Quantum groups},'' \href{https://dx.doi.org/10.1007/BF01247086}{{\em Zap. Nauchn. Semin.} {\bfseries 155} (1986) 18--49}.

\bibitem{Ding-Frenkel-quantumaffine}
J.~Ding and I.~Frenkel, ``Isomorphism of two realizations of quantum affine algebra {Mathematical expression},'' \href{https://dx.doi.org/10.1007/BF02098484}{{\em Commun. Math. Phys.} {\bfseries 156} no.~2, (1993) 277--300}.

\end{thebibliography}\endgroup
\end{document}